\documentclass[11pt,a4paper]{amsart}
\usepackage[a4paper,margin=1in]{geometry}
\usepackage{color,array}
\usepackage{cite}
\usepackage{hyperref}
\usepackage{amsmath,amssymb,amsthm,color}
\usepackage[normalem]{ulem}
\hypersetup{linktocpage,colorlinks, citecolor={magenta},bookmarksdepth=2}
\usepackage[utf8]{inputenc}
\usepackage{nicefrac}
\usepackage{dsfont}
\usepackage{graphicx}
\graphicspath{{figures/}}
\usepackage{float}
\usepackage{subcaption} 
\usepackage{tcolorbox}
\usepackage{algorithm}
\usepackage{algorithmic}
\usepackage{mathtools}

\usepackage{accents}

\usepackage{xspace}
\usepackage{bm}				
\usepackage{bbm,upgreek}		
\usepackage[e]{esvect}		
\usepackage{esint}			
\usepackage{etoolbox}			
\usepackage{calc}				
\usepackage{eso-pic}			
\usepackage{datetime}			
\usepackage{tikz}
\usepackage{tkz-euclide}

\usepackage{lmodern}
\numberwithin{equation}{section}

\usepackage{enumitem}
\setlist{nosep}
\setlist{noitemsep}

\def\XXint#1#2#3{{\setbox0=\hbox{$#1{#2#3}{\int}$}
		\vcenter{\hbox{$#2#3$}}\kern-.5\wd0}}

\newcommand{\N}{\mathbb{N}}
\newcommand{\Z}{\mathbb{Z}}

\newcommand{\R}{\mathbb{R}}

\newtheorem{theorem}{Theorem}
\newtheorem{prop}{Proposition}[section]
\newtheorem{lemma}[prop]{Lemma}
\newtheorem{coro}[prop]{Corollary}
\newtheorem{remark}[prop]{Remark}

\theoremstyle{plain}

\newtheorem{definition}[prop]{Definition}

\def \t0{\rightarrow 0} 

\def \hal{\frac{1}{2}}

\def \div{\mathrm{div} \,} 
\def \1{\mathbf{1}} 
\def \mc{\mathcal}

\def \dist{\mathrm{dist}}

\def\nab{\nabla}

\newcommand{\dip}{\mathrm{dip}}

\newcommand{\dd}{\mathrm{d}}

\newcommand{\dE}{\mathbb{E}}
\newcommand{\bad}{\mathrm{bad}}

\newcommand{\dR}{\mathbb{R}}
\newcommand{\dZ}{\mathbb{Z}}
\newcommand{\dN}{\mathbb{N}}

\newcommand{\mult}{\mathrm{mult}}

\newcommand{\Eul}{\mathrm{Eul}}
\newcommand{\Eulc}{\mathrm{Eul}_{\mathrm{c}}}

\newcommand{\odd}{\mathrm{odd}}

\newcommand{\Msf}{\mathsf{M}}
\newcommand{\Anch}{\mathrm{Anch}}

\newcommand{\SG}{\mathrm{SG}}
\newcommand{\nord}[1]{{:}\mspace{1mu}#1\mspace{1mu}{:}}

\newcommand{\Leaves}{\mathrm{Leaves}}
\newcommand{\Ksf}{\mathsf{K}}
\newcommand{\ksf}{\mathsf{k}}
\newcommand{\msf}{\mathsf{m}}
\newcommand{\Root}{\mathrm{Root}}
\newcommand{\Peeled}{\mathrm{Peeled}}
\newcommand{\Zsf}{\mathsf{Z}}
\newcommand{\Clus}{\mathrm{Clus}}
\newcommand{\Res}{\mathrm{Res}}
\newcommand{\abs}{\mathrm{abs}}
\newcommand{\Coarse}{\mathrm{Coarse}}
\newcommand{\inter}{\mathrm{inter}}
\newcommand{\intra}{\mathrm{intra}}
\newcommand{\Htrees}{\mathrm{Htrees}}

\newcommand{\br}{\mathrm{br}}

\newcommand{\loc}{\mathrm{loc}}
\newcommand{\Cut}{\bar{R}_{\beta,\lambda}}
\newcommand{\Good}{\mathrm{GoodP}}
\newcommand{\even}{\mathrm{even}}
\newcommand{\PF}{\mathrm{PF}}
\newcommand{\KA}{\mathrm{KA}}
\newcommand{\Child}{\mathrm{Child}}

\newcommand{\ed}[1]{{\color{black}{#1}}}
\definecolor{darkgreen}{rgb}{0,0.45,0}
\newcommand{\edgreen}[1]{{\color{black}{#1}}}

\def\XXint#1#2#3{{\setbox0=\hbox{$#1{#2#3}{\int}$}
		\vcenter{\hbox{$#2#3$}}\kern-.5\wd0}}

\def \XN{\vec{X}_N}
\def \YN{\vec{Y}_N}

\def\Esp{\mathbb{E}} 
\def \ZNbeta{Z_{N,\beta}}

\def\g{\mathsf{g}}

\def\Psf{\mathsf{P}}
\def\d{\mathsf{d}}

\def\indic{\mathbf{1}}
\def \E{\mathcal{E}}

\makeatletter
\def\namedlabel#1#2{\begingroup
	#2%
	\def\@currentlabel{#2}%
	\phantomsection\label{#1}\endgroup
}
\makeatother

\def \epsilon{\varepsilon}

\def\F{\mathsf{F}}
\def\loc{\mathrm{loc}}

\def \diam{\mathrm{diam}}

\def\F{\mathsf{F}}
\def\K{\mathsf{K}}

\def\Esp{\mathbb{E}}

\def\P{\mathbb{P}_{N,\beta}^{\lambda}}
\def\Z{Z_{N,\beta}^{\lambda}}

\newcommand{\vr}{\vec{r}}
\newcommand{\vro}{\vec{\rho}}

\newcommand{\ve}{\varepsilon}

\newcommand{\Id}{\mathrm{Id}}

\newcommand{\good}{\mathrm{good}}

\newcommand{\sub}{\mathrm{sub}}

\begin{document}
	\title[Multipole and BKT Transitions for the 2CP]{Multipole and Berezinskii--Kosterlitz--Thouless Transitions in the Two-component Plasma}
	\author{Jeanne Boursier}
	\address{Massachusetts Institute of Technology, Department of Mathematics, 77 Massachusetts Avenue, Cambridge, MA 02139, USA.}
	\email{boursier@mit.edu}
	
	\author{Sylvia Serfaty}
	\address{Sorbonne Universit\'e,  CNRS,  Laboratoire Jacques-Louis Lions (LJLL), F-75005 Paris,
		France \\ \& Institut Universitaire de France \&
		Courant Institute of Mathematical Sciences, 251 Mercer Street, New York, NY 10012.}
	\thanks{S.S. was supported by NSF grant DMS-2000205 and by the Simons Foundation through the Simons Investigator program.}
	\email{serfaty@cims.nyu.edu}
	\maketitle

	\begin{abstract}
		We study the two-dimensional two-component  Coulomb gas in the canonical ensemble and at inverse temperature $\beta>2$. In this regime, the partition function diverges and the interaction needs to be cut off at a length scale $\lambda\in (0,1)$. Particles of opposite charges tend to pair into dipoles of length scale comparable to $\lambda$, which themselves can aggregate into multipoles. Despite the slow decay of dipole--dipole interactions, we construct a convergent cluster expansion around a hierarchical reference model that retains only intra-multipole interactions. This yields a large-deviation result for the number of $2p$-poles as well as a \textcolor{black}{free-energy expansion} as $N\to\infty$ and $\lambda\to0$, with three contributions: (i) the free energy of $N$ independent dipoles, (ii) a perturbative correction, and (iii) the contribution of a non-dilute subsystem.

		The perturbative term has two equivalent characterizations: (a) a convergent Mayer series built from non-divergent elementary dipole activities; and (b) a variational formula as the minimum of a large-deviation rate function for the empirical counts of $2p$-poles. The Mayer coefficients \textcolor{black}{should undergo} transitions at $\beta_p=4-\tfrac{2}{p}$\textcolor{black}{, and these thresholds} accumulate at $\beta=4$, which corresponds to the \textcolor{black}{Berezinskii--Kosterlitz--Thouless} transition in the low-dipole-density limit. \textcolor{black}{We prove that the limiting $p$-dipole activities are finite for $\beta>\beta_p$.}
		
		The non-dilute system corresponds to the contribution of large dipoles. For $2<\beta<4$, the independent-dipole dilution estimate predicts the critical length scale $R_{\beta,\lambda}=\lambda^{-(\beta-2)/(4-\beta)}$. For $\beta\geq4$, we set $R_{\beta,\lambda}=+\infty$, as this estimate gives no algebraic crossover scale in $\lambda$.
	\end{abstract}
	\tableofcontents

	\section{Introduction}
	
	\subsection{Setting of the problem and objectives}
	We consider the canonical continuum Coulomb gas with formal Boltzmann weight
	\begin{equation}\label{ZGP}\exp\left(- \beta \F(\vec{X}_N, \vec{Y}_N) \right) \dd \vec{X}_N \dd \vec{Y}_N ,\end{equation}
	for  configurations $(\XN, \YN)$, with $\vec{X}_N=(x_1, \dots, x_N) \in \Lambda^N $ and $\vec{Y}_N=(y_1, \dots, y_N)\in \Lambda^N$  of $N$ positive and $N$ negative particles  in the square $\Lambda = [0,\sqrt{N}]^2$ in $\R^2$, with logarithmic interaction energy
	\begin{equation}\label{1.1} \F(\vec{X}_N, \vec{Y}_N)= \hal \left( \sum_{i\neq j} - \log |x_i-x_j|-\log |y_i-y_j|+ 2 \sum_{i,j} \log |x_i-y_j|\right),\end{equation}
	where $\dd \vec{X}_N$ and $\dd \vec{Y}_N$ denote the Lebesgue measure on $\Lambda^N$. Here, the parameter $\beta>0$ is the inverse temperature.
	
	The integral of the expression in \eqref{ZGP} (called the partition function) diverges as soon as $\beta \ge 2$, so we introduce a small-scale cutoff $\lambda\in (0,1)$. This divergence corresponds to the dipole transition that occurs at $\beta=2$, between a situation with ``free charges'' for $\beta<2$ and a situation with  pairing  of opposite-sign particles (dipoles)  of length scale $\lambda$ for \textcolor{black}{$\beta>2$}.
	For the regularized model, the theory of Berezinskii, Kosterlitz and Thouless \cite{Berezinsky1970fr,Kosterlitz1973,Kosterlitz1974} predicts a phase transition whose critical inverse temperature tends to $4$ in the dilute limit. This is the \textcolor{black}{Berezinskii--Kosterlitz--Thouless} (BKT) transition.
	It was predicted in   \cite{BenfattoGallavottiNicolo,GallavottiNicolo} that such a system also has an infinite  sequence of transitions at
	\begin{equation}\label{defbetap}\beta_p \coloneqq  4 - \frac{2}{p},\quad \text{where $p\in \mathbb{N}^*$}, \end{equation}
	each of which corresponds to the divergence (at small scale) of the free energy of $2p$-poles (i.e., sets of $p$ dipoles clustered together).
	\textcolor{black}{Indeed, if all the relative distances of a $2p$-pole are comparable to $r$, then its formal weight per logarithmic scale is $r^{p(2-\beta)+2(p-1)}$. Integrating over the scales, we obtain $\int_\lambda^1 r^{p(2-\beta)+2(p-1)}\,\dd r/r$, which diverges as $\lambda\to0$ if and only if $\beta\geq\beta_p$, and the divergence is logarithmic when $\beta=\beta_p$. For $p=1$, we recover the dipole transition at $\beta=2$.}

	\textcolor{black}{Much of the rigorous low-temperature literature has} analyzed the \textcolor{black}{grand-canonical} lattice Coulomb gas, via the \textcolor{black}{sine--Gordon} transformation. Fr\"ohlich and Spencer \cite{frohlichspencer} proved the existence of the BKT phase transition in lattice models.
	In an earlier paper, they write \cite[p.~621]{frohlichspencerstatmech}:
	\begin{quote}
		``We believe
		that the techniques of Section 5 will eventually permit us to prove convergence of an expansion of the two-dimensional Coulomb gas in terms of
		neutral multipole configurations, at low density and low temperature,
		designed to imply the existence of the Kosterlitz-Thouless transition. But
		the required combinatorial and refined electrostatic estimates are still
		missing.''
	\end{quote}
	Fisher, Li and Levin \cite[Section~1]{FisherLiLevin} argued that the thresholds $\beta_p$ need not correspond to intermediate thermodynamic phases at positive density. This question is distinct from the changes in cutoff asymptotics studied here.

	In this paper, we \textcolor{black}{prove a new free-energy expansion and large-deviation estimates as $\lambda\to 0$, which are consistent with} the multipole transitions at $\beta_p$ and with aspects of the BKT transition at $\beta=4$. We combine new electrostatic estimates with diagrammatic expansions.
	
	We also introduce, for $\beta\in (2,\infty)$, a {\it critical length scale}
	\begin{equation}
		\label{def:Rlambda}
		R_{\beta,\lambda}= \begin{cases} \lambda^{-\frac{\beta-2}{4-\beta}} & \text{ if $\beta<4$ } \\ +\infty &\text{ if $\beta \ge 4$}\end{cases}.\end{equation}
	For $2<\beta<4$, the independent-dipole estimate predicts this crossover scale: below it, dipoles are typically far from other dipoles of comparable size. For $\beta\geq4$, the convention $R_{\beta,\lambda}=+\infty$ records that this estimate gives no algebraic crossover scale in $\lambda$. At $\beta=4$, it does not rule out an exponentially large crossover scale.
	
	Our first main result is  an expansion of the partition function $Z_{N,\beta}^\lambda$ (of the model suitably truncated at $\lambda$): roughly, we show that, for $\beta\in(\beta_p,\beta_{p+1}]$,
	\begin{equation}
		\label{freeexpintro}
		\log Z_{N,\beta}^\lambda=\log N! + \text{free energy of $N$ i.i.d.~dipoles} + \text{Mayer series} + \text{non-dilute contribution}.\end{equation}
	This formula is a perturbative expansion around an i.i.d.~dipole model. The exact statement we prove is given in Theorem~\ref{theorem:expansion} below. Here $\lambda$ is sufficiently small and $N$ is sufficiently large depending on $\lambda$. The Mayer series is built from non-divergent elementary cluster activities involving at most $p$ dipoles. The non-dilute contribution is an error term which we show is \textcolor{black}{$O(NR_{\beta,\lambda}^{-2})$} (with logarithmic corrections at each $\beta_p$), \textcolor{black}{and we expect this order to be sharp}. In the regime $\beta \ge 4$, we can retain elementary activities involving at most $p_0$ dipoles for every fixed $p_0$, and obtain an error \textcolor{black}{$O(N\lambda^{2p_0})$ whose constant depends on $p_0$. In this regime, we do not prove that the full Mayer series converges}.
	\textcolor{black}{For $\beta>\beta_p$, we also prove that the limiting activity of clusters of $p$ dipoles is finite, and we expect it to diverge for $2<\beta\leq\beta_p$.}

	We also provide a probabilistic interpretation of the multipole transitions: we show a large-deviation result on 
	the fraction of $2p$-poles, \textcolor{black}{which is of order $\lambda^{2(p-1)}$ for $\beta>\beta_p$. Below $\beta_p$, we expect this fraction to be of order $R_{\beta,\lambda}^{-2}$.} The Mayer series term in \eqref{freeexpintro} is interpreted as the minimum of the rate function associated with this large-deviation result. The precise statements are given in Theorem~\ref{theorem:LDP} below. \medskip

	The paper develops a cluster expansion for long-range dipole interactions in the canonical continuum model, without iterating a renormalization map. The main ingredients are:
	\begin{itemize}
		\item\textit{Matching compatible with cluster expansion.} We use a matching between positive and negative charges (forming dipoles) whose constraint factors almost surely into pairwise conditions, as in \eqref{introAij}.
		\item \textit{Hierarchical multipole model.} \textcolor{black}{Our estimates do not control a} direct perturbative expansion around an i.i.d.~dipole model. Instead,  we introduce a hierarchical multipole reference model that retains only intra-multipole interactions and we perform a convergent perturbative (cluster) expansion around it. The non-divergent activities from this multipole Mayer series are identified, \textcolor{black}{after summing over multipole partitions and grouping connected components}, with the corresponding activities of the dipole Mayer series.
		\item \textit{Removing divergences by energy truncation.} For $\beta\in(\beta_p,\beta_{p+1})$, \textcolor{black}{our integrability estimates do not apply to clusters of more than $p$ dipoles}. The idea is to ``sandwich'' the system between two auxiliary models whose interaction ranges lie strictly below the critical scale $R_{\beta,\lambda}$, \textcolor{black}{in order to control the large-scale terms} before running the expansion. We show that the two auxiliary models have identical Mayer series up to \textcolor{black}{an error $O(NR_{\beta,\lambda}^{-2})$}. These approximations are obtained via new electrostatic estimates which provide effective upper and lower bounds 
		for the energy.

		\item \textit{Extracting an ideal dilute subsystem.} We isolate an ``ideal'' subsystem (dipoles of length smaller than the critical scale and multipoles of bounded cardinality, subject to additional constraints) and perform the cluster expansion solely on this part. We then show that the complement has negligible density with high probability.
		\item \textit{Exploiting cancellations.} Within the retained interaction range, the dipole--dipole interaction decays as the inverse square of the distance. However, it exhibits a cancellation upon averaging over the angle between the dipoles. Thus, the Mayer bond splits into a non-integrable odd contribution and an even integrable contribution. In the lower-bound expansion, symmetry restricts the odd bonds to graphs that are Eulerian relative to the multipole partition: the sum of the odd-bond degrees over each block is even. In the upper-bound expansion, we recover the necessary estimates using interactions with dipoles outside the ideal subsystem. Together with the summable even-bond contributions, these arguments yield convergent integrals and a summable cluster series, with \textcolor{black}{bounds} on the contribution of large diagrams.
	\end{itemize}

	\subsection{The BKT picture}

	\subsubsection{The BKT transition}
	
	The Berezinskii--Kosterlitz--Thouless (BKT) transition (see, e.g.,
	\cite{bietenholz2016berezinskii}, the Nobel lecture \cite{RevModPhys.89.040501},
	and the text \cite{alastueybook}) occurs in two-dimensional systems including
	thin superfluid \(^4\)He films, superconducting films, and defect-mediated
	melting of two-dimensional crystals.  Its distinctive feature is that it is not
	driven by the spontaneous breaking of a continuous symmetry, but by the
	(un)binding of topological defects.
	
	A discrete toy model for a two-dimensional superfluid is the planar rotor
	(or XY) model: to each site one assigns a unit spin \(e^{i\theta_x}\in\mathbb S^1\),
	with a nearest-neighbor energy favoring alignment.  The model has a
	global \(\mathrm U(1)\) symmetry.  By the Mermin--Wagner theorem, the
	magnetization vanishes at every positive temperature, so there is no
	conventional long-range ordered phase. Berezinskii and Kosterlitz and Thouless showed that two-dimensional systems can
	nevertheless exhibit a low-temperature phase with \emph{quasi-long-range order} \cite{Berezinsky1970fr,Kosterlitz1973}.
	
	Vortices are point defects of the \(\mathbb S^1\)-valued phase field.  On the
	lattice they correspond to plaquettes around which \(\theta\) winds by
	\(\pm 2\pi\), i.e., defects of degree \(+1\) (vortex) or \(-1\) (antivortex).
	\textcolor{black}{In the low-temperature vortex approximation, one keeps the charges \(n=\pm1\), which interact logarithmically and form a neutral two-component Coulomb gas \cite[eqs.~(50), (54)--(55)]{Kosterlitz1973}.} A single
	vortex in a system of linear size \(D\) has dimensionless energy
	\(\beta \Delta E\approx \pi K_0\log(D/a)+\beta E_c\), where \(K_0\) is the spin-wave stiffness and \(E_c\) is a core energy. Its entropic
	gain from the \(\sim(D/a)^2\) possible locations is
	\(\Delta S\approx 2\log(D/a)\). The competition
	\[
	\beta\Delta F=(\pi K_0-2)\log(D/a)+\beta E_c
	\]
	predicts that isolated vortices are suppressed at low temperature
	(\(\pi K_0>2\)), while at higher temperature the entropic gain wins, neutral
	pairs unbind, and free vortices proliferate. Below the BKT transition,
	correlations satisfy polynomial upper and lower bounds, whereas above it
	they decay exponentially. See \cite[Theorem~1]{mcbryanspencer} and
	\cite[Section~1.4, Theorem (BKT transition)]{lammers2023bijecting}.
	
	Let us emphasize that, in the lattice XY model, vortex charges live on the dual
	lattice (vorticity on plaquettes), whereas in continuum superfluid models, the
	vortex core has a finite microscopic size, which provides the short-distance
	cutoff \cite{RevModPhys.89.040501}.

	\subsubsection{The renormalization-group picture}\label{sub:renorm}
	
	Kosterlitz and Thouless \cite{Kosterlitz1973,Kosterlitz1974} made the above heuristic quantitative by deriving RG flow equations for two scale-dependent couplings: the stiffness~\(K(l)\) and the vortex fugacity~\(z(l)\), where \(l=\log(r/a)\) is the logarithmic length scale. The stiffness \(K(l)\) measures the free-energy cost of a long-wavelength twist of the phase once all vortex--antivortex dipoles up to size \(\sim ae^l\) have been integrated out. In our normalization, its bare value is \(K(0)=K_0=\beta/(2\pi)\). The fugacity is initially \(z(0)=e^{-\beta E_c}\), and it evolves as the definition of the core is enlarged by the RG.

	In the dilute-vortex regime,
	integrating out vortex--antivortex dipoles shell by shell gives, to leading
	order in~\(z\),
	\begin{equation}\label{eq:ODE}
		\frac{\dd K^{-1}}{\dd l}=4\pi^3 z^2 + O(z^4),\qquad
		\frac{\dd z}{\dd l}=(2-\pi K)\,z + O(z^3).
	\end{equation}
	The first equation expresses dielectric screening: neutral dipoles reduce the
	effective stiffness. The second encodes the energy--entropy competition
	discussed above: vortices are irrelevant when \(\pi K>2\) and proliferate
	when~\(\pi K<2\). The value $K=\frac{2}{\pi}$ with small $z$ corresponds to $\beta=4$ with small $\lambda$ in our normalization. Retaining the leading quadratic terms near the critical fixed point gives a flow in which \(x^2-\tilde y^2\) is constant, where \(x=2-\pi K\) and \(\tilde y=4\pi z\). The corresponding leading-order trajectories are separated by a critical separatrix flowing into
	\((K,z)=(2/\pi,\,0)\).

	\begin{figure}[H]
		\centering
		\includegraphics[width=0.4\textwidth]{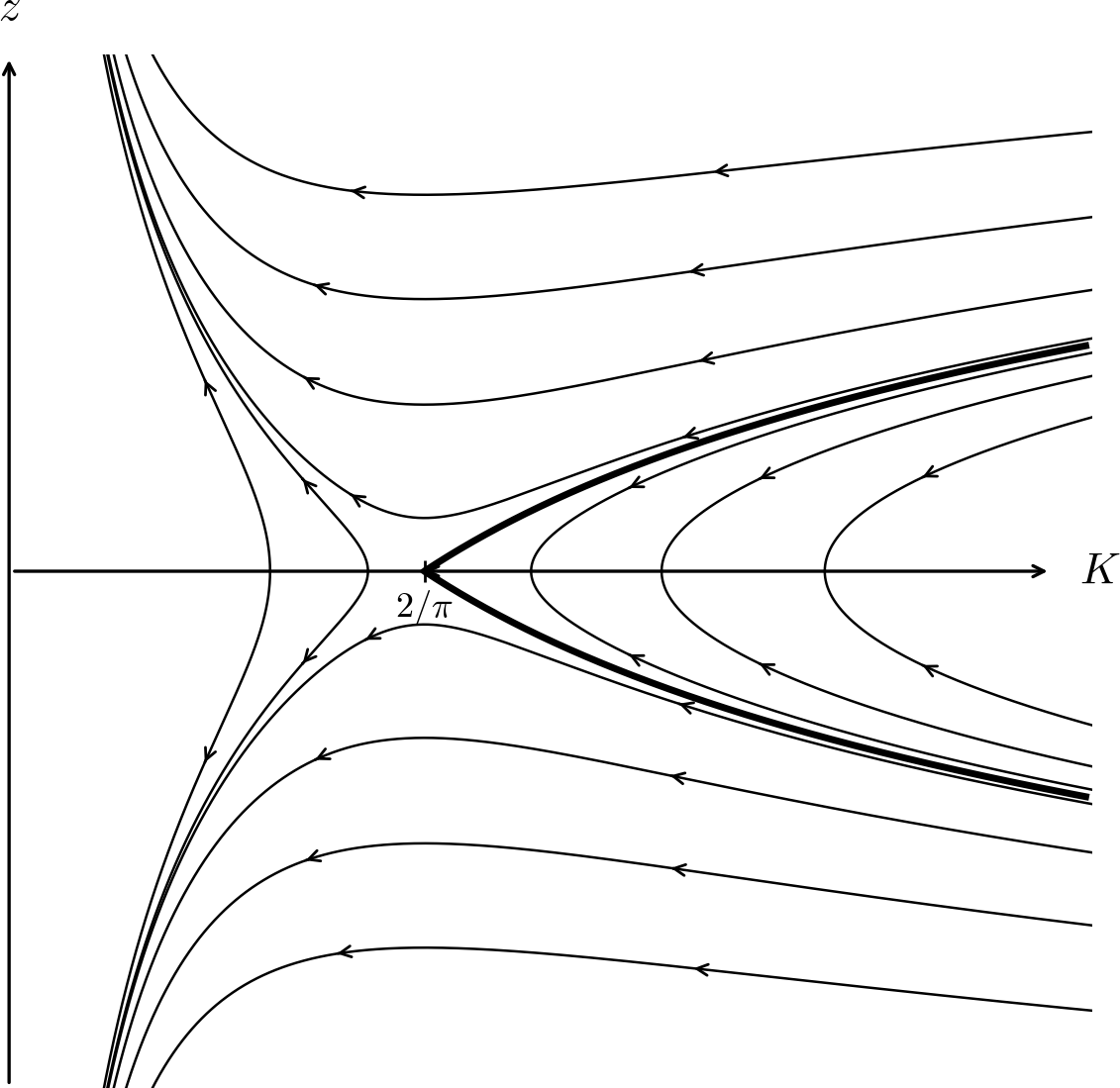}
		\caption{Kosterlitz--Thouless approximate RG flow in the \((K,z)\)-plane. The bold curve is the separatrix of the leading-order truncation of \eqref{eq:ODE}.}
	\end{figure}

	Already at leading order, these equations capture the main signatures of the
	transition. On the separatrix the fugacity flows to zero, so the infrared
	theory is a free Gaussian boson with renormalized stiffness
	\(K(\infty)=2/\pi\)\textcolor{black}{. This predicts that
	\(\langle e^{i\theta(r)}e^{-i\theta(0)}\rangle=r^{-\eta+o(1)}\), with}
	\(\eta=1/(2\pi K(\infty))=1/4\). On the disordered side, the time spent
	drifting along the nearly critical separatrix before the flow runs away gives a
	correlation length with an essential singularity,
	\(\log(\xi/a)\sim b/\sqrt{|\beta-\beta_c|}\), with $b>0$. Finally, the
	stiffness (which is proportional to the superfluid density) jumps
	discontinuously from the universal value \(K=2/\pi\) to zero at the
	transition. The predicted jump in superfluid density was supported by the measurements of Bishop and Reppy on thin helium films~\cite[Section~V and Figure~11]{PhysRevB.22.5171}.

	\subsection{Cumulants and multipole transitions}

	\ed{The \textcolor{black}{sine--Gordon} cumulants correspond to the coefficients in the cluster expansion (or Mayer diagrams) of the grand-canonical Coulomb gas, as observed in \cite{brydgeskennedy}. \textcolor{black}{Benfatto, Gallavotti and Nicol\`o \cite[eqs.~(1.3)--(1.5)]{BenfattoGallavottiNicolo} related the thresholds $\beta_p=4-\frac{2}{p}$ to ultraviolet collapse in the massive sine--Gordon model. Gallavotti and Nicol\`o \cite[eqs.~(4.1)--(4.3)]{GallavottiNicolo} related this ultraviolet problem to the infrared problem for the massless Coulomb gas by rescaling.} \textcolor{black}{These estimates on the individual coefficients do not by themselves prove that the Mayer series converges.} \textcolor{black}{Alastuey and Cornu} \cite{AlastueyCornu1992_KT_CoulombGas,AlastueyCornu1997_CriticalLine_ZeroDensity} studied the two-dimensional Coulomb gas directly through cluster calculations and resummations. Their analysis includes controlled low-order calculations, but the resummation near criticality leaves higher-order remainders uncontrolled \cite[Section~3.3]{AlastueyCornu1997_2D_CG_PartI}.}

	\subsection{Sine--Gordon and rigorous renormalization-group flows}

	The flow \eqref{eq:ODE} is a perturbative description of the evolution of two effective parameters, the stiffness and the vortex fugacity. A standard route to its rigorous analysis is the sine--Gordon representation of the Coulomb gas. The exact correspondence is recalled in Appendix~\ref{app:sine-gordon}. The massless sine--Gordon field is relevant for the logarithmic Coulomb interaction, whereas a massive covariance gives the Yukawa interaction. The rigorous literature treats the ultraviolet construction of the sine--Gordon model, the infrared stability of its dipole phase, and the marginal trajectory at the BKT transition.

	The parameter used in the sine--Gordon RG papers below is $\beta_{\rm SG}=2\pi\beta$ in our normalization.

	\subsubsection{The Fr\"ohlich--Spencer expansion}

	Fr\"ohlich and Spencer~\cite{frohlichspencer} established the dipole phase at low temperature by an inductive expansion over increasing distance scales in the sine--Gordon representation. Charge densities are successively organized into well-separated neutral ensembles, and complex spin-wave translations produce small scale-dependent activities. This reorganizes the measure as a convex superposition of dilute gases of neutral molecules. It gives algebraic decay estimates for fractional-charge correlations at low temperature. At high temperature, Yang~\cite{yang1987debye} proved Debye screening for the two-dimensional Coulomb gas. A modern account of the \textcolor{black}{Fr\"ohlich--Spencer} argument is given in~\cite{KharashPeled}. Refinements by Marchetti, Klein and P\'erez~\cite{mkp1}, followed by Marchetti and Klein~\cite{mkp2}, extended the small-activity dipole regime down to $\beta_{\rm SG}>8\pi$, equivalently $\beta>4$. These works rigorously control the low-temperature side of the transition.

	\subsubsection{Discrete polymer renormalization group}

	A systematic polymer renormalization-group formalism for perturbations of massless Gaussian fields was developed by Brydges and Yau \cite{brydgesyau}. Dimock and Hurd subsequently applied it to the sine--Gordon model. Their article~\cite{dimock2000sine} gives a corrected, self-contained treatment. The Gaussian covariance is decomposed into scales. At each step, one integrates one fluctuation field, extracts the relevant or marginal local terms, and then reblocks and rescales. The effective interaction is described by extracted local couplings, such as the vacuum energy and stiffness, together with an infinite-dimensional polymer activity $K_j(X,\phi)$ containing the remaining terms. In refinements aimed at the critical problem, the fugacity is also retained as an explicit coordinate. At sufficiently small coupling, Dimock and Hurd used this framework for the ultraviolet problem when $\beta_{\rm SG}<8\pi$ with a fixed infrared cutoff, and for the infrared dipole regime when $\beta_{\rm SG}>8\pi$ with a fixed ultraviolet cutoff~\cite[Theorems~1--2]{dimock2000sine}.

	The polymer activity is controlled by scale-dependent norms combining analyticity bounds with large-field and large-set regulators. In the infrared problem, the charged Fourier sectors contract under Gaussian integration, while the constant and gradient terms of the neutral sector are extracted. The balance between this contraction and the growth under reblocking produces the threshold $\beta_{\rm SG}=8\pi$.

	Falco's work~\cite{falco2012kosterlitz,falco2013critical} treats the marginal case within this polymer RG framework. The two principal couplings follow a discrete version of the Kosterlitz flow \eqref{eq:ODE}, while the polymer remainder contracts in an appropriate norm. A stable-manifold choice of the initial stiffness selects a small piece of the critical BKT curve. This yields control of the pressure along that curve and the predicted power-law fractional-charge correlations, including multiplicative logarithmic corrections.

	\subsubsection{Continuous renormalization group}

	A different implementation evolves the effective interaction continuously under Gaussian convolution. If
	\[
		e^{-V_t(\phi)}=\mathbb E_{C_t}\bigl[e^{-V_0(\phi+\zeta)}\bigr],
		\qquad C_t=\int_0^t\dot C_s\,\dd s\textcolor{black}{,}
	\]
	where $\zeta$ is a centered Gaussian field with covariance $C_t$, then $V_t$ satisfies the exact Polchinski equation
	\[
		\partial_tV_t=\frac12\Delta_{\dot C_t}V_t
		-\frac12(\nabla V_t,\nabla V_t)_{\dot C_t}\textcolor{black}{.}
	\]
	Following Brydges and Kennedy~\cite{brydgeskennedy}, Bauerschmidt and Bodineau~\textcolor{black}{\cite[eq.~(1.3) and Lemma~3.4]{bauerschmidt2021log}} expand $V_t$ in cumulants. In the sine--Gordon representation, these cumulants correspond to the coefficients in the Mayer expansion of the Coulomb gas \textcolor{black}{with covariance $C_t$, integrated up to time $t$}. They control this expansion using heat-kernel decay and cancellations in the neutral sectors.

	This continuous method proves logarithmic Sobolev inequalities for the continuum sine--Gordon model on a fixed torus with fixed positive mass when $\beta_{\rm SG}<6\pi$, or $\beta<3$ in our normalization~\cite[Theorem~1.6]{bauerschmidt2021log}. Its direct tree estimate applies below $\beta_{\rm SG}=4\pi$. The extension from $4\pi$ to $6\pi$ already requires neutral-sector cancellations and heat-kernel smoothing. The cited argument does not cover the full ultraviolet range $\beta_{\rm SG}<8\pi$. The resulting bounds on the cumulants yield Hessian estimates for the effective potential.

	Another approach uses stochastic quantization on the torus in two dimensions. Hairer and Shen~\cite[Theorem~1.1]{hairershen2016dynamical} constructed local solutions to the renormalized dynamical sine--Gordon equation in part of the subcritical regime. Chandra, Hairer and Shen~\cite[Theorem~1.1]{chandrahairershen2018} extended this local construction to the full subcritical regime using cancellations between charged components. The sine--Gordon measure itself was constructed in infinite volume by Bauerschmidt and Webb, in the massless case at $\beta_{\rm SG}=4\pi$, that is $\beta=2$ \cite[Theorem~1.3]{bauerschmidtwebb2024coleman}. In a recent preprint, Gubinelli and Meyer constructed the massive sine--Gordon measure on the plane for $\beta_{\rm SG}<6\pi$ \cite[Theorem~1.1]{gubinellimeyer2024fbsde}.

	\subsubsection{Comparison of the approaches}

	The works on stochastic quantization concern the ultraviolet construction of the dynamics rather than the infrared RG flow underlying the BKT transition.

	The sine--Gordon representation gives access to results that are difficult to obtain directly in the Coulomb model, for instance on the decay of correlations. It \textcolor{black}{allows one to renormalize the model} by integrating out the small \textcolor{black}{scales} of the GFF. However, \textcolor{black}{the cited methods do not control, at each scale, the number of dipoles defined by our matching}.
	
	Our methods analyze the canonical particle model directly. We believe that they may also yield, at sufficiently small activity and fixed positive mass, convergence of the renormalized field-dependent cumulant expansion of the massive sine--Gordon model throughout the subcritical range $\beta_{\rm SG}<8\pi$, equivalently $\beta<4$ in our normalization, after subtracting the necessary field-independent counterterms.
	
	\subsection{Height functions and the BKT transition}\label{sec:duality}

	The planar XY and Villain models admit exact dual descriptions in terms of integer-valued height functions. Expanding the spin interaction in Fourier modes and integrating over the spins produces a divergence-free integer current. On a simply connected planar graph, this current is the rotated gradient of a height function on the faces. For the classical XY model the resulting interaction is expressed through modified Bessel functions, whereas the Villain model is dual to the discrete Gaussian model, also called the integer-valued GFF. Notice that the temperature is reversed by this correspondence: the low-temperature spin regime corresponds to the high-temperature height regime. The height and vortex representations are related but should not be identified: the height gradients represent Fourier currents, while vortices are the charges of the Coulomb-gas representation.

	On the square lattice, the height model has a localization--delocalization, or roughening, transition. In the localized phase the height fluctuations remain bounded as the domain grows, while in the delocalized phase they diverge, typically logarithmically. For these two model pairs, the localized height phase corresponds to the massive spin phase, while the rough phase corresponds to the BKT phase. Fr\"ohlich and Spencer already studied these two phenomena together, proving both the existence of the low-temperature BKT regime and the high-temperature rough phase of the dual height model~\cite{frohlichspencer}.

	More recent works have made this correspondence precise. Aizenman, Harel, Peled and Shapiro showed that delocalization of the discrete Gaussian height function forces the dual Villain model into its BKT regime~\cite[Theorem~1.4]{aizenman2021depinning}. Van Engelenburg and Lis proved the corresponding implication for the XY model and its dual height function~\cite[Corollary~4.3]{vanengelenburglis2023}. Lammers subsequently proved that, for the classical XY model, the mass of the dual height function is exactly twice the spin mass~\cite[Theorem~1]{lammers2023bijecting}. Consequently, the BKT transition coincides with the localization--delocalization transition. He also established the corresponding qualitative equivalence between the Villain and discrete Gaussian models~\cite[Theorem~2]{lammers2023bijecting}.

	Lammers further developed a general geometric theory of these height functions. For a broad class of potentials, localization implies exponential decay of the height covariance, whereas delocalization implies at least logarithmic growth of the variance. Moreover, the transition point itself belongs to the delocalized phase. His proof uses level lines, a percolation representation and RSW-type coarse-graining, rather than a perturbative RG expansion~\cite{Lammers2022ADT}. For the XY model and its dual height function, Durand and Lammers gave another characterization: the massive phase corresponds exactly to a corner of the height surface tension at zero slope, and this corner disappears in the BKT phase~\cite[Theorem~1.4]{durandlammers2026}.

	Finally, Bauerschmidt, Park and Rodriguez proved convergence of the discrete Gaussian model to a multiple of the GFF sufficiently far inside the rough phase~\cite{bauerschmidt2022discrete,bauerschmidt2022discrete2}.

	\subsection{Prior results on our model}
	Prior to our work, the ensemble \eqref{ZGP} was  studied  in the regime $\beta <2$ in  \cite{GunPan,DeutschLavaud,LSZ}.
	In the neutral grand-canonical ensemble, Fr\"ohlich \cite[Theorem~4.5]{Frohlich} proved the existence of the thermodynamic limit of the pressure. Gunson and Panta \cite[Sections~4 and 5.1]{GunPan} proved the thermodynamic limit, which gives an expansion of $\log \ZNbeta^0$ up to $o(N)$, including an order-$N$ constant. Lebl\'e, Serfaty and Zeitouni \cite[Theorem~1 and Corollary~1.3]{LSZ} identified this constant through a variational principle and proved a large deviation principle on point processes. Their analysis builds on \cite{GunPan} and on techniques developed for the one-component Coulomb gas in \cite{SS2d,RougSer,LS1}. The variational principle describes the competition between the attraction of opposite charges and entropy. This corresponds to the situation of ``free charges'' for $\beta<2$.
	
	In  our prior paper \textcolor{black}{\cite[Theorem~1]{boursier2024dipole}}, using large-deviation techniques, a description of dipoles based on nearest-neighbor graphs inspired by \cite{GunPan} as well as electrostatic estimates, we showed that, for $\beta \ge 2$, the system concentrates on configurations with a large proportion of dipoles \textcolor{black}{(of size $O(\lambda)$ when $\beta>2$)}, and proved the  expansion 
	$$\log Z_{N,\beta}^\lambda=\log N!+ \text{free energy of $N$ i.i.d.~dipoles} +N o_\lambda(1),$$
	where $o_\lambda(1)$ is bounded by a positive power of $\lambda$ for $\beta>2$ and by a positive power of $|\log \lambda|^{-1}$ for $\beta=2$.
	
	\subsection{Lean formalization}\label{sub:lean-formalization}
	The free-energy expansion of Theorem~\ref{theorem:expansion} and the multipole distribution result of Theorem~\ref{theorem:LDP} have been formalized in Lean~4 using Mathlib. The formalization is available at~\url{https://github.com/jboursier/multipole-bkt-lean}. We thank Scott Armstrong for sharing the setup used for this formalization.
	
	\subsection*{Statement of AI use}
	All the ideas of this paper are our own, and the first version of the paper was produced in September 2025 with minimal AI use to polish the paper. For the present version, we decided to Lean formalize the paper with AI agents (see Subsection~\ref{sub:lean-formalization}), which led to finding some mistakes in the prior version, which were then corrected with AI help. This was done using Codex and Claude Opus~5.5. We take full responsibility for the content of the paper.
	
	\medskip
	
	{\bf Acknowledgments:} The authors are particularly indebted to  Ofer Zeitouni for many invaluable discussions that greatly benefited this work.
	They also thank Christophe Garban, Tom Spencer, Angel Alastuey, Eyal Lubetzky and Ahmed Bou-Rabee for very helpful discussions.

	\section{Definitions, method of proof and main results}
	
	\subsection{Model}
	We now describe the precise model we study by providing the  definition of the interaction truncation introduced in \cite{boursier2024dipole} and adopted here.
	
	For convenience, we always write 
	\begin{equation} 
		\g(x)=-\log |x|,
	\end{equation}
	and we will abuse notation by considering $\g$ as either a function on $\R^2$ or on $\R$ depending on the context.
	
	Truncating the interaction involves introducing a small length scale $\lambda\in (0,1)$ and {\it renormalizing} the divergent part of the energy. A natural  way is to truncate the logarithmic kernel at a distance $\lambda$ and consider
	\begin{equation}\label{eq:F2}   \hal \sum_{i,j}\left(\min( \g(x_i-x_j) , \g(\lambda))+ \min(\g(y_i-y_j), \g(\lambda)) -2 \min(  \g(x_i-y_j) , \g(\lambda))\right).
	\end{equation}
	
	Instead, it is much more convenient and natural to consider charges smeared on disks of radius $\lambda\in (0,1)$, with $\lambda$ small, interacting otherwise in the normal Coulomb fashion. Let $\chi$ denote a  radial nonnegative function, supported in  the unit ball $B(0,1)$ and such that $\int \chi =1$. We will only assume that $\chi$ is bounded. For instance, $\chi$ could be the normalized indicator function of $B(0,1)$.
	We let $\delta_{z}^{(\lambda)}= \frac{1}{\lambda^2} \chi ( \frac{\cdot -z}{\lambda})$, which is a measure of mass $1$ supported on $B(z, \lambda)$. We then  let 
	\begin{equation} \label{defkappa}
		\kappa\coloneqq  \iint \g(x-y) \delta_0^{(1)}(x)\delta_0^{(1)}(y),\end{equation}
	and observe, by scaling, that 
	\begin{equation}\label{geta0}\iint \g(x-y) \delta_0^{(\lambda)} (x) \delta_0^{(\lambda)}(y)= \g(\lambda)+ \kappa.\end{equation} 
	
	The energy we consider is
	\begin{equation}\label{eq:F3}
		\F_\lambda(\XN, \YN)= \hal \iint \g(x-y) \dd\left( \sum_{i=1}^N \delta_{x_i}^{(\lambda)}-\delta_{y_i}^{(\lambda)}\right)(x) 
		\dd\left( \sum_{i=1}^N \delta_{x_i}^{(\lambda)}-\delta_{y_i}^{(\lambda)}\right)(y) -  N   ( \g( \lambda)+\kappa).
	\end{equation}
	Here, compared to \eqref{1.1}, we have reinserted the self-interaction terms, which are no longer infinite  but equal to $\g(\lambda) + \kappa$,
	and then subtracted them off.

	We will denote by 
	\begin{equation} \label{eq:defgeta} \g_\lambda(z)= \iint \g(x-y) \delta_0^{(\lambda)}(x)\delta_z^{(\lambda)} (y) = \g* \delta_0^{(\lambda)} * \delta_0^{(\lambda)} (z),
	\end{equation} the effective interaction between two points at distance $|z|$. The function $ \delta_0^{(\lambda)} * \delta_0^{(\lambda)}$ is radial and supported in $B(0, 2\lambda)$, hence  from the mean-value theorem or Newton's theorem, 
	\begin{equation}  \g_\lambda(z)=  \g(z) \quad \text{for } |z|\ge 2\lambda.\end{equation} 
	\textcolor{black}{Thus $\F_\lambda$ is obtained from \eqref{eq:F2} by replacing the truncated kernel by $\g_\lambda$ and subtracting $N(\g(\lambda)+\kappa)$. It coincides with $\F$ when all pairwise distances are at least $2\lambda$, so that $\F_\lambda\to\F$ as $\lambda\downarrow0$ for every collision-free configuration.}
	
	Let us emphasize that this renormalization of the interaction is also more physically motivated. As stressed in \cite{RevModPhys.89.040501}, vortices in a two-dimensional superfluid are disks rather than point particles: the superfluid model itself therefore includes a short-distance cutoff parameter $\lambda$. Moreover, this choice is consistent with the regularization on the dual sine--Gordon side \cite{bauerschmidt2021log}.

	\smallskip
	
	We will thus work with \eqref{eq:F3} and study 
	\begin{equation} \label{defP}
		\dd\P= \frac{1}{\ZNbeta^\lambda} \exp\left(- \beta \F_\lambda(\vec{X}_N, \vec{Y}_N)\right) \dd \vec{X}_N \dd \vec{Y}_N
	\end{equation}
	for $\lambda\in (0,1)$ small and fixed, where 
	\begin{equation} \label{eq:ZRV}\Z= \int_{(\Lambda^2)^N} \exp\left(-\beta \F_\lambda(\XN, \YN)\right)  \dd \XN \dd\YN.
	\end{equation}
	When $\beta <2$, one can immediately set $\lambda=0$ and recover the model studied in \cite{LSZ}, but when $\beta>2$, $\log \ZNbeta^\lambda$  diverges as  $\lambda \to 0$.

	\subsection{Dipoles and multipoles}\label{sub:dipoles} 
	Throughout the paper, for every integer $n$, we write $[n]\coloneqq \{1,\dots,n\}$.
	
	The first question that arises is how to properly define dipoles, i.e., how one pairs the positive charges with the negative ones. 
	We  present a notion of dipoles based on a   stable matching algorithm. This is a different definition from that in \cite{boursier2024dipole}, in which dipoles were defined based on the nearest-neighbor graph of the particles (irrespective of their sign), as in \cite{GunPan}.
	
	We introduce a matching between the positive charges $(x_i)_{i\in [N]}$ and the negative charges $(y_i)_{i\in [N]}$. We use the stable matching $\sigma_N$  defined as follows: starting from $r=0$, grow disks of radius $r$ around each $x_i$. If $B(x_i,r)$ is the first disk to touch a negative charge, say $y_j$, then we set $\sigma_N(i)=j$. The particles $x_i$ and $y_j$ are then removed from the family of particles, and we repeat the operation. This procedure defines a random permutation of $[N]$.
	\begin{definition}\label{def:stable intro}
		We call {\it dipoles} the pairs $(x_i, y_{\sigma_N(i)})$, where $\sigma_N$ is the stable matching as above. 
	\end{definition}

	\begin{figure}
		\centering
		\fbox{\includegraphics[width=0.5\textwidth]{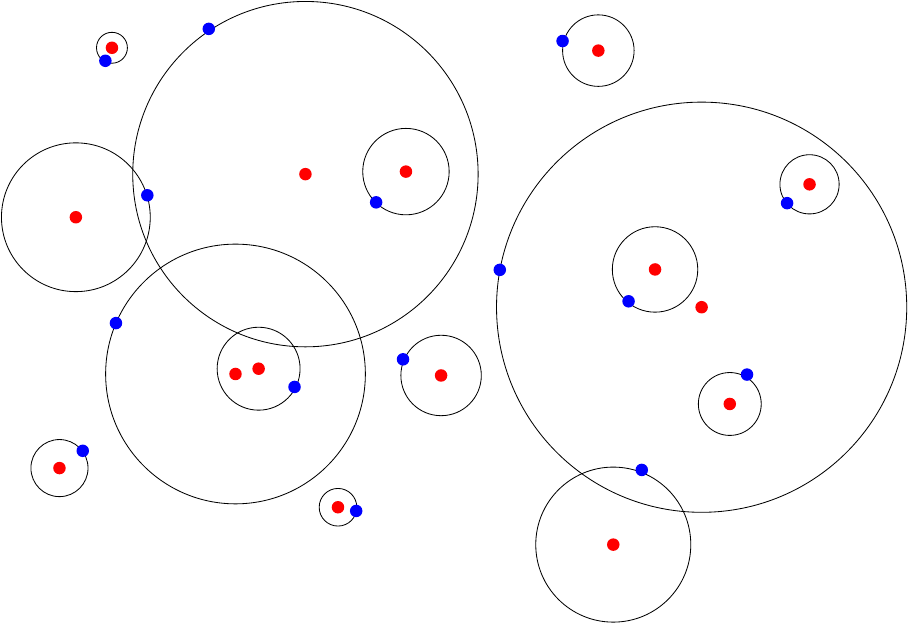}} 
		\caption{A stable matching. The positive charges are represented as red dots and negative charges as blue dots.}
	\end{figure}

	Without loss of generality, up to the multiplicative factor $N!$, one can reduce to the situation where $\sigma_N=\Id$. 
	The stable matching has the nice property (see Subsection~\ref{sub:stable})  that
	\begin{equation}\label{introAij}
		\{\sigma_N=\Id\}=\bigcap_{1\leq i,j\leq N:i\neq j}\mc{A}_{ij},
	\end{equation}
	where $\mc{A}_{ij}$ is an event depending only on $x_i, y_i, x_j, y_j$: it is the geometric condition that says that $x_i$ is matched to $y_i$ and $x_j$ to $y_j$ when performing the stable matching with $x_i , x_j$ and $y_i, y_j$ only. This decomposition of the geometric constraint of the matching $\sigma_N$ into the pair conditions $\mc{A}_{ij}$ is crucial.

	We now define the notion of $2p$-poles or multipoles, which are groups of dipoles. Two dipoles belong to the same multipole if the distance between these two sets of points is, roughly speaking, of the same order as the smallest dipole size. This construction involves an arbitrary choice of a large parameter.

	\begin{definition}[Multipoles]\label{def:multipoles}
		Let $\sigma_N$ be the stable matching introduced in Definition~\ref{def:stable intro}. Let $M>20$ be a  fixed number. For each $i=1,\ldots,N$, let $\mathbf{r}_i\coloneqq |x_i-y_{\sigma_N(i)}|$. 
		
		\begin{enumerate} 
			\item For every $i, j \in [N]$ with $i\neq j$, define the event
			\begin{equation}\label{eq:defBij}
				\mathsf{B}_{ij}\coloneqq \{i\leftrightarrow_{\sigma_N} j\},
			\end{equation}
			where $i\leftrightarrow_{\sigma_N} j$ if
			\begin{equation*}
				\mathrm{dist}(\{x_i,y_{\sigma_N(i)}\},\{x_j,y_{\sigma_N(j)}\})\leq M \min(\max(\mathbf{r}_i,\lambda), \max(\mathbf{r}_j,\lambda)).
			\end{equation*}
			\item The relation $\leftrightarrow_{\sigma_N}$ defines an undirected graph on $[N]$.  
			The connected components of that graph are called multipoles. Those of size $k$ are called $2k$-poles or multipoles of size $k$ (or pure dipoles for $k=1$).
			\item Let $S$ be a subset of $[N]$. Define the event
			\begin{equation}\label{eq:defBC}
				\mathsf{B}_S\coloneqq \textcolor{black}{\{\text{the graph induced on }S\text{ by }\leftrightarrow_{\sigma_N}\text{ is connected}\}}\,.
			\end{equation}
		\end{enumerate}
	\end{definition}

	Note that to have a convergent cluster expansion, we will take $M$ large enough with respect to $\beta$.

	\begin{figure}[H]
		\centering
		\fbox{\includegraphics[width=0.4\textwidth]{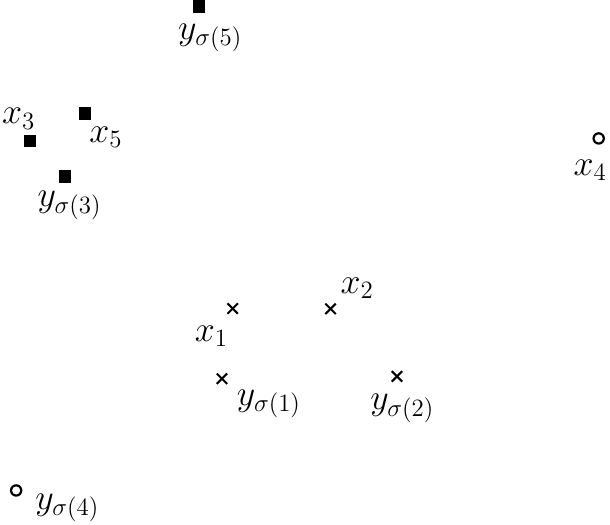}}
		\caption{Three multipoles: $\{1,2\}$, $\{4\}$, $\{3,5\}$.}
		\label{fig:mult}
	\end{figure}
	These definitions are motivated by the fact that the free energy of a $\textcolor{black}{2k}$-pole is a true $2k$-body integral and cannot be Taylor-expanded. Note that a large dipole (say of size $\gg 1$) which is not a multipole is still considered to be in a ``dipole state'', exactly like a small dipole that is well separated from other dipoles. Such a large dipole may be surrounded by small dipoles that are very close to it, but that is still not considered a multipole.

	We also define the analogues of $\mathsf{B}_{ij}$ and $\mathsf{B}_S$ when $y_{\sigma_N(i)}$ is replaced by $y_i$.

	\begin{definition}\label{def:multipoles mc} Let $M>20$ be as in Definition~\ref{def:multipoles}. For each $i=1,\ldots,N$, let $r_i\coloneqq |x_i-y_i|$. 
		
		\begin{enumerate} 
			\item For every $i, j \in [N]$ with $i\neq j$, define the event
			\begin{equation}\label{eq:defBij mc}
				\mc{B}_{ij}\coloneqq \{i\leftrightarrow j\},
			\end{equation}
			where $i\leftrightarrow j$ if
			\begin{equation*}
				\mathrm{dist}(\{x_i,y_i\},\{x_j,\textcolor{black}{y_j}\})\leq M \min(\max(r_i,\lambda), \max(r_j,\lambda)).
			\end{equation*}
			\item Let $S$ be a subset of $[N]$. Define the event
			\begin{equation}\label{eq:defBC mc}
				\mc{B}_S\coloneqq \textcolor{black}{\{\text{the graph induced on }S\text{ by }\leftrightarrow\text{ is connected}\}}.
			\end{equation}
		\end{enumerate}
	\end{definition}

	\textcolor{black}{On the event $\{\sigma_N=\Id\}$, the events $\mathsf{B}_{ij}$ and $\mathsf{B}_S$ coincide with $\mc{B}_{ij}$ and $\mc{B}_S$.}

	\subsection{Starting point of the method}\label{sub:idea}
	Let us give a simplified sketch of the strategy. Details will be given over the course of the paper.

	\subsubsection{Cluster expansion}
	
	We begin by introducing a dipole expansion (which we in fact do not perform), which will allow us to introduce the main quantities.

	The starting point is to separate out the leading 
	part of the interactions, which is the interaction within dipoles, and consider the other pair interactions, which are the dipolar interactions (and are weaker), as small perturbations by writing:
	\begin{equation}\label{premdevZ}
		Z_{N, \beta}^\lambda = \int_{(\Lambda^2)^N}  e^{-\beta \F_\lambda}= N! \int_{(\Lambda^2)^N}  e^{-\beta \F_\lambda} \indic_{\sigma_N=\Id}=  N! \int_{(\Lambda^2)^N} \prod_{i< j}e^{-\beta v_{ij}}\indic_{\mc{A}_{ij}}\prod_{i=1}^N e^{\beta \g_\lambda(x_i-y_i) }\dd x_i \dd y_i,
	\end{equation}
	where $v_{ij}$ is the interaction between the dipole $\{x_i, y_i\}$ and the dipole $\{x_j,y_j\}$, and $\mathcal A_{ij}$ are the sets appearing in \eqref{introAij}. 
	
	The principle of a cluster expansion is to write
	\begin{equation}\label{vijfij}
		e^{-\beta v_{ij}}\indic_{\mc{A}_{ij}}=1+f_{ij}^v,
	\end{equation} where $f_{ij}^v$ are the so-called Mayer bonds, 
	and expand the product over $i,j\in [N]$, $i< j$. This yields 
	\begin{equation}\label{eq:try1}
		Z_{N,\beta}^\lambda=N!\sum_{E}\int_{(\Lambda^2)^N} \prod_{ij\in E} f_{ij}^v\prod_{i=1}^N e^{\beta \g_\lambda(x_i-y_i) }\dd x_i \dd y_i,
	\end{equation}
	where the sum is over sets of pairs $E$ in $[N]$.
	
	Cluster expansions allow us to use combinatorial graph expansions to rewrite the \emph{logarithm} of expressions like \eqref{eq:try1} as a sum over subsets of $[N]$, involving certain integral terms called \emph{activities}. We refer for instance to \cite{faris,FriedliVelenik} for an introduction.  It is important to note that the activity is not an energetic quantity.

\begingroup\color{black}
For a connected graph $(V,E)$ with $p$ vertices, we average the integrand of the term of this graph in the activity over the independent reversals $x_i\leftrightarrow y_i$, and then fix $x_1=0$. The resulting integral is absolutely convergent provided $\beta>\beta_p$. When $\beta<4$, this condition holds only for finitely many $p$, and the argument controls only the corresponding clusters.
\endgroup

	\subsubsection{Critical scale}
	With
	\begin{equation*}
		Z^\dip=\int_{(\Lambda^2)^N}\prod_{i=1}^N e^{\beta \g_\lambda(x_i-y_i)}\dd x_i \dd y_i,
	\end{equation*}
	the expression \eqref{eq:try1} can be read as 
	\begin{equation}\label{eq:try2}\frac{Z_{N, \beta}^\lambda}{Z^\dip } =N!\sum_{E} \Esp_{\mathrm{dip}}\left[  \prod_{ij\in E} f_{ij}^v\right],\end{equation}
	where the expectation is taken with respect to the i.i.d.~dipole model, i.e., the measure with density proportional to 
	\begin{equation*}
		\prod_{i=1}^N e^{\beta \g_\lambda(x_i - y_i)}  \dd x_i \dd y_i.
	\end{equation*}
	We will denote by 
	\begin{equation}\label{def:Zbeta}
		\mc{Z}_\beta\coloneqq 2\pi\int_0^{+\infty} e^{\beta \g_1(r)}r\dd r,
	\end{equation} the rescaled (limiting) normalization constant for this model.

	The cluster expansion corresponding to \eqref{eq:try2} leads to a series \textcolor{black}{whose convergence we do not prove}. To obtain a convergent expansion, \textcolor{black}{we first truncate the interactions}. Let us see at what scale the i.i.d.~dipole approximation breaks down. Under the i.i.d.~dipole model, the number of dipoles with size $R > \lambda$ is typically of order $N (\frac{\lambda}{R})^{\beta - 2}.$ Therefore, under the i.i.d.~dipole measure, the typical distance between dipoles of size $R$ is of the order $(\frac{R}{\lambda})^{\frac{\beta}{2} - 1}.$ Observe that the independent-dipole approximation is valid only if the inter-dipole distance is at least of the same order as the dipole length, i.e., if
	\begin{equation}\label{echelle optimale}
		\left(\frac{R}{\lambda}\right)^{\frac{\beta}{2} - 1} \geq R \quad \Longleftrightarrow \quad \textcolor{black}{R \leq R_{\beta,\lambda}}.
	\end{equation}
	Indeed, if the inter-dipole spacing is much smaller than the dipole length, the matching forces an overcrowding of same-signed particles, which is heavily energetically penalized.

	\subsubsection{Expansion around the hierarchical model}
	
	To have a convergent expansion, we should expand only the interactions that are small enough, i.e., the interactions between dipoles in distinct multipoles. In other words, we will perform a perturbative expansion around a hierarchical multipole model in which only intra-multipole interactions (as defined in Definition~\ref{def:multipoles}) are retained.

	For every subpartition $X$ of $[N]$, the hierarchical model on $X$, denoted by $\Psf_{\beta,\lambda}^X$, is defined by
	\begin{equation}\label{eq:PXfirst}
		\dd \Psf_{\beta,\lambda}^X \propto	\prod_{ S\in X}\indic_{\mc{B}_S} \prod_{i,j \in S:i<j} e^{-\beta v_{ij}} \indic_{\mc{A}_{ij}} \prod_{i\in S}e^{\beta \g_\lambda(x_i - y_i)} \dd x_i \dd y_i,
	\end{equation}
	which can be compared to \eqref{premdevZ}.
	
	Denote by $\Pi_\mult$ the (random) partition into multipoles and let $\pi$ be a partition of $[N]$. Expanding around the hierarchical model by using \eqref{vijfij} allows one to expand $\P(\sigma_N=\Id,\Pi_\mult=\pi)$ as the multipole partition function multiplied by the perturbative term
	\begin{equation}\label{sumEpi} \sum_{E \subset \mc{E}^\inter(\pi)} \dE_{\textcolor{black}{\Psf_{\beta,\lambda}^{\pi}}} \left[\prod_{ij \in E} (e^{-\beta v_{ij}}\indic_{\mc{B}_{ij}^c}-1) \right],
	\end{equation}
	where  $\mc{E}^\inter(\pi)$ represents the  set of edges joining pairs of dipoles in {\it distinct} multipoles, defined by 
	\begin{equation}\label{defmcE} 
		\mc{E}^\inter(\pi)\coloneqq \bigcup_{S\neq S'\in \pi} \bigcup_{i\in S, j\in S'} \{ij\}.\end{equation}
	Let us also define $\mc{E}^\intra(\pi)$ to be the set of edges between dipoles in the same multipole:
	\begin{equation*}
		\mc{E}^\intra(\pi)\coloneqq \bigcup_{S\in \pi} \bigcup_{i,j\in S:i< j} \{ij\}.
	\end{equation*}

	Following the cluster-expansion roadmap, we resum \eqref{sumEpi} over the connected components of the augmented graph 
	\begin{equation}\label{def:augmented}
		([N], E \cup \mc{E}^\intra(\pi)).
	\end{equation}
	Let $X$ be a subpartition of $[N]$. We set 
	\begin{equation}\label{eq:VX first}
		V_X\coloneqq \bigcup_{S\in X}S,
	\end{equation}
	and we say that a set of edges $E\subset \mc{E}^\inter(X)$ is connected relative to $X$ and write $E\in \mathsf{E}^X$ if the augmented graph $(V_X,E\cup \mc{E}^\intra(X))$ is connected and if {$E$ contains at least one edge}.

	This allows one to rewrite \eqref{sumEpi} as 
	\begin{equation}\label{eq:take my log}
		\sum_{E \subset \mc{E}^\inter(\pi)} \dE_{\textcolor{black}{\Psf_{\beta,\lambda}^{\pi}}} \left[\prod_{ij \in E} (e^{-\beta v_{ij}}\indic_{\mc{B}_{ij}^c}-1) \right]=\sum_{n=0}^\infty \frac{1}{n!}\sum_{\substack{X_1,\ldots,X_n\in\mc{P}(\pi) \\ \mathrm{disjoint} } }K(X_1)\cdots K(X_n),
	\end{equation}
	where $K(X)$ is the {\it activity} of the subpartition $X\in \mc{P}(\pi)$ defined by 
	\begin{equation}\label{2.11}
		K(X)=\sum_{E \in \mathsf{E}^X} \dE_{\textcolor{black}{\Psf_{\beta,\lambda}^{\pi}}} \left[\prod_{ij \in E} (e^{-\beta v_{ij}}\indic_{\mc{B}_{ij}^c}-1) \right].
	\end{equation}
	
	One of the main advantages of such a representation is that, \textcolor{black}{if the connected series on the right-hand side of \eqref{eq:log} is absolutely convergent}, one has a formula for the logarithm of \eqref{eq:take my log} as another series. A simple combinatorial argument gives the following identity of formal series:
	\begin{equation}\label{eq:log}
		\log \sum_{n=0}^\infty \frac{1}{n!} \sum_{\substack{X_1, \dots, X_n \in \mc{P}(\pi)\\ \mathrm{disjoint}}} K(X_1) \dots K(X_n)=\sum_{n=1}^{+\infty}\frac{1}{n!}\sum_{\substack{X_1,\ldots,X_n\in \mc{P}(\pi) \\ \text{ connected}}}K(X_1)\cdots K(X_n)\mathrm{I}(G(X_1,\ldots,X_n)),
	\end{equation}
	where $G(X_1,\ldots,X_n)$ denotes the connection graph of $X_1,\ldots,X_n$ (the graph on $[n]$ with an edge $ij$ if $X_i$ and $X_j$ intersect) and where $\mathrm{I}$ is the   Ursell function whose definition we now recall.	
	\begin{definition}\label{def:graphU}
		The Ursell function of  a connected graph $G$ is defined by
		\begin{equation}\label{def:UrsellI}
			\mathrm{I}(G)\coloneqq \sum_{H\subset G}(-1)^{|E(H)|},
		\end{equation}
		where the sum over $H$ runs over all spanning subgraphs of $G$ ($H$ connected and with edges included in those of $G$) and where $|E(H)|$ is the number of edges of $H$.
	\end{definition}

	Note that we will in fact expand the restricted partition function in a different way, which will eventually give the same activity by M\"obius inversion. The main challenge is to establish the absolute convergence of the series \eqref{eq:log}. Because of the long-range nature of the interaction, the usual cluster-expansion criteria do
	not apply in their standard form.

	\subsection{Statement of the main results}
	
	\textcolor{black}{All asymptotic estimates below hold for $\lambda$ sufficiently small and $N$ sufficiently large depending on $\lambda$.}
	
	We begin by defining the dipole activity for families of $k$ dipoles in $\Lambda$.

	\begin{definition}[Dipole activity]\label{def:dipole activity}The dipole measure is the probability measure on $\Lambda^2$ defined by
		\begin{equation}\label{defmubetalambda}
			\dd \mu_{\beta,\lambda}(x,y)=\frac{1}{NC_{\beta,\lambda}}e^{\beta \g_\lambda(x-y)}\dd x\dd y,
		\end{equation}
		where 
		\begin{equation}\label{defCbetalambda}
			C_{\beta,\lambda}=\frac{1}{N}\int_{\Lambda^2}e^{\beta \g_\lambda(x-y)}\dd x\dd y.
		\end{equation}
		For every $V\subset [N]$, we let $\mc{G}_c(V)$ be the set of collections of edges $E$ on $V$ such that $(V,E)$ is connected, with at least one edge. Moreover, we let
		\begin{equation*}
			\Ksf_{\beta,\lambda}^{\dip}(V)=\sum_{E\in \mc{G}_c(V)} \dE_{\mu_{\beta,\lambda}^{\otimes |V|} }\left[\prod_{ij\in E}(e^{-\beta v_{ij}}\indic_{\mc{A}_{ij}}-1) \right].
		\end{equation*}
		Notice that $\Ksf_{\beta,\lambda}^{\dip}$ still depends on $N$, and that, for $|V|\in \{0,1\}$, we have $\Ksf_{\beta,\lambda}^{\dip}(V)=0$.
	\end{definition}

	\textcolor{black}{Since $\g_1$ is bounded on $B(0,2)$ and equals $-\log|x|$ outside, the constant $\mc Z_\beta$ is positive and finite for $\beta>2$. By the scaling relation $\g_\lambda(x)=-\log\lambda+\g_1(x/\lambda)$, the constant $C_{\beta,\lambda}$ is also positive and finite, and}
	\begin{equation}\label{limCbl}
		\lim_{N\to \infty}C_{\beta,\lambda}=\int_{\dR^2}e^{\beta \g_\lambda(x)}\dd x=\lambda^{2-\beta}\mathcal{Z}_\beta,
	\end{equation}
	where $\mc{Z}_\beta$ is as in \eqref{def:Zbeta}.

	Next, we introduce the limiting activity associated with $\Ksf_{\beta,\lambda}^\dip$. For $\beta>2$, set
	\begin{equation}\label{defpstar}
		p^*(\beta)\coloneqq \sup\{\textcolor{black}{q\in\mathbb N^*}:\ \beta>\beta_q\}\in \mathbb{N}^*\cup\{+\infty\}.
	\end{equation}
	Note that $p^*(\beta)=1$ for $\beta\in(2,3]$ and $p^*(\beta)=+\infty$ for $\beta\ge 4$. We will show that whenever $|V|=k\in \{2,\ldots, p^*(\beta)\}$, the limit
	\begin{equation}\label{def:little k dip}
		\lim_{N\to\infty}N^{k-1}\Ksf_{\beta,\lambda}^\dip(V)\eqqcolon\ksf_{\beta,\lambda}^\dip(k)
	\end{equation}
	exists. For $k\in \{0,1\}$, we set $\ksf_{\beta,\lambda}^\dip(k)=0$.

	We next introduce the truncated Mayer series that arises in \eqref{freeexpintro}.

\begin{definition}[Hypertrees]
\begingroup\color{black}
For every finite set $A$ and $n\in\mathbb N^*$, define
\begin{multline}\label{def:Htrees}
\Htrees_n(A)\coloneqq\Bigl\{\mathbf X=(X_1,\ldots,X_n)\in\mc P(A)^n:\ |X_i|\geq2\text{ for every }i,\\
G(X_1,\ldots,X_n)\text{ is connected},\quad
\Bigl|\bigcup_{i=1}^nX_i\Bigr|+n-\sum_{i=1}^n|X_i|=1\Bigr\}.
\end{multline}
For $\mathbf X\in\Htrees_n(A)$, set
\[
\omega(\mathbf X)\coloneqq\prod_{a\in\cup_iX_i}\bigl(|\{i\in[n]:a\in X_i\}|-1\bigr)!.
\]
\endgroup
\end{definition}

	\begin{definition}[Truncated Mayer series]\label{def:free energy}
		Let $\beta\in (2,\infty)$. Let $p\in\mathbb N^*$ with $p\leq p^*(\beta)$, where $p^*(\beta)$ is as in \eqref{defpstar}. 
		
		\textcolor{black}{Throughout, a product of limiting activities and indicator functions vanishes as soon as one of the indicators does, even if some of the activities are not defined.}
		Define the Mayer series truncated at $p$ to be \textcolor{black}{(with $\mathbf V=(V_1,\ldots,V_n)$)} 
		\begin{equation}\label{def:fbeta}
			\mathrm{Mayer}_{\beta,p,\lambda}\coloneqq \sum_{n=1}^{\infty} \frac{(-1)^n}{n!}\sum_{k=1}^\infty \frac{1}{k! }\sum_{\substack{(V_1,\ldots,V_n)\in \Htrees_n([k]):\\ V_1\cup \cdots \cup V_n=[k] }}\textcolor{black}{\omega(\mathbf V)}\mathsf{k}^{\dip}_{\beta,\lambda}(|V_1|)\indic_{|V_1|\leq p }\cdots \mathsf{k}^{\dip}_{\beta,\lambda}(|V_n|)\indic_{|V_n|\leq p }.
		\end{equation}
		Note that when $\beta \in (2,3]$, $\mathrm{Mayer}_{\beta, p,\lambda}=0$.
	\end{definition}
	
	\textcolor{black}{We prove in Subsection~\ref{sub:free energy proof} that this series converges absolutely for $\lambda$ small enough, depending on $\beta$ and $p$.}
	When $\beta\in (2,4)$, the Mayer series will be truncated at $p=p^*(\beta)$, while for $\beta\geq 4$ at an arbitrary number that will be fixed throughout the paper.

	\begin{definition}[Series truncation parameter]\label{def:pbeta}
		Fix $p_0\in \mathbb{N}^*$ for the rest of the paper. Set 
		\begin{equation*}
			p(\beta)=\begin{cases}
				p^*(\beta) & \text{if $\beta\in (2,4)$}\\
				p_0 & \text{if $\beta\geq 4$}.
			\end{cases}
		\end{equation*}
	\end{definition}

	\textcolor{black}{In particular, $p(\beta)\leq p^*(\beta)$.}
	
	We next quantify the remainder in \eqref{freeexpintro}, which corresponds to the contribution of the non-dilute subsystem for $\beta\in(2,4)$ and to large diagrams for $\beta\ge 4$ in the expansion \eqref{freeexpintro}.

	\begin{definition}[\textcolor{black}{Error rate}]\label{def:optimal error}
		We let 
		\begin{equation}\label{defdelta}
			\delta_{\beta,\lambda}\coloneqq R_{\beta,\lambda}^{-2}\indic_{\beta\in (2,4)} \left(1+\sum_{p= 2}^\infty |\log \lambda|\indic_{\beta=\beta_{p}}\right) +\lambda^{2p_0}\indic_{\beta\geq 4}.
		\end{equation}
	\end{definition}
	\textcolor{black}{Writing $p=p(\beta)$, we have $\delta_{\beta,\lambda}=o(\lambda^{2(p-1)})$. Indeed, for $\beta_p<\beta<\beta_{p+1}$, the exponent of $\lambda$ in $\delta_{\beta,\lambda}$ is $2(\beta-2)/(4-\beta)>2(p-1)$, while $\delta_{\beta,\lambda}=\lambda^{2p}(1+|\log\lambda|)$ at $\beta=\beta_{p+1}$ and $\delta_{\beta,\lambda}=\lambda^{2p}$ for $\beta\geq4$.}

	The following result decomposes the free energy of the Coulomb gas as the sum of the free energy of $N$ i.i.d.~dipoles, the truncated Mayer series of non-divergent clusters of Definition~\ref{def:free energy} and the \textcolor{black}{error rate} from Definition~\ref{def:optimal error}.

	\begin{theorem}[Free-energy expansion]\label{theorem:expansion}
		Let $\beta\in (2,\infty)$ and let $p(\beta)$ be as in Definition~\ref{def:pbeta}. \edgreen{Let $\lambda\in(0,1)$. In items (1) and (3), each term $O_{\beta,p(\beta)}(N\delta_{\beta,\lambda})$ means that there exist $C,\lambda_0>0$, depending only on $\beta$, $p_0$ and $\chi$, such that, for every $\lambda\in(0,\lambda_0)$, there is $N_0(\lambda)$ for which this term is bounded in absolute value by $CN\delta_{\beta,\lambda}$ when $N\geq N_0(\lambda)$. In item (1), the series is absolutely convergent for such $\lambda$ and $N$. In item (2), the bound holds for $\lambda$ smaller than a constant depending on $\beta$ and $k$.} Let $\mc{Z}_\beta$ be as in \eqref{def:Zbeta} and $\Ksf_{\beta,\lambda}^\dip$ be as in Definition~\ref{def:dipole activity}. Recall the Ursell function $\mathrm{I}$ from Definition~\ref{def:graphU}. \textcolor{black}{We have $0<Z_{N,\beta}^\lambda<\infty$.}
		\begin{enumerate}
			\item We have
			\begin{multline}\label{eq:th1 exp1}
				\log Z_{N,\beta}^\lambda=2N\log N  +N((2-\beta)\log \lambda+\log \mc{Z}_\beta-1)\\+\sum_{n=1}^\infty \frac{1}{n!}\sum_{\substack{V_1,\ldots, V_n\subset [N]\\ \textcolor{black}{\mathrm{connected}},\\ \forall i,2\leq |V_i|\leq p(\beta)}}\Ksf_{\beta,\lambda}^\dip(V_1)\cdots \Ksf_{\beta,\lambda}^\dip(V_n)\mathrm{I}(G(V_1,\ldots,V_n))+O_{\beta, p(\beta)}(N\delta_{\beta,\lambda}).
			\end{multline}
			\item The limit $\ksf_{\beta,\lambda}^\dip(k)$ in \eqref{def:little k dip} exists and is finite if $k\leq p^*(\beta)$, i.e., if $\beta>\beta_k$. Moreover, for every $k\leq p^*(\beta)$, there exists a constant $C>0$ depending on $\beta$ and $k$ such that
			\begin{equation}\label{eq:th1 ksf}
				|\ksf_{\beta,\lambda}^\dip(k)|\leq  C\lambda^{2(k-1)}.
			\end{equation}	
			\item Recalling $\mathrm{Mayer}_{\beta,p,\lambda}$ from Definition~\ref{def:free energy}, we have
			\begin{equation}\label{eq:th1 exp2}
				\sum_{n=1}^\infty \frac{1}{n!}\sum_{\substack{V_1,\ldots, V_n\subset [N]\\ \textcolor{black}{\mathrm{connected}},\\ \forall i,|V_i|\leq p(\beta)}}\Ksf_{\beta,\lambda}^\dip(V_1)\cdots \Ksf_{\beta,\lambda}^\dip(V_n)\mathrm{I}(G(V_1,\ldots,V_n))=-N\mathrm{Mayer}_{\beta,p(\beta),\lambda}
				+O_{\beta,p(\beta)}(N\delta_{\beta,\lambda}).
			\end{equation}
		\end{enumerate}
	\end{theorem}

	\medskip

	\begin{remark}[Link with the multipole transitions]
		\textcolor{black}{For $p\geq2$ and $\beta>\beta_p$, the limiting activity of clusters of $p$ dipoles, or $p$-clusters, is finite and $O(\lambda^{2(p-1)})$}, while
		\begin{equation*}
			R_{\beta,\lambda}^{-2}< \lambda^{\frac{2(\beta_p-2)}{4-\beta_p}}=\lambda^{2(p-1)},
		\end{equation*}with equality for $\beta=\beta_p$.
		\textcolor{black}{We expect this activity to diverge for $2<\beta\leq\beta_p$.} Note that here $p$-clusters are not the same as $2p$-poles. 
	\end{remark}
	
	\begin{remark}[On the optimality of the error rate]
		Recall from \eqref{echelle optimale} that, for $2<\beta<4$, the independent-dipole dilution estimate predicts the crossover scale $R_{\beta,\lambda}$.
		Hence, we expect there are about $NR_{\beta, \lambda}^{-2}$ dipoles of size $R_{\beta, \lambda}$. The free energy of that subsystem should then be proportional to the number of points, as in the regime $\beta\in (0,2)$ in \cite{LSZ}. Hence, the error that these dipoles contribute should be of order $NR_{\beta, \lambda}^{-2}$. This error rate is therefore expected to be sharp when $\beta\in (2,4)$. \textcolor{black}{When $\beta\geq4$, we also expect the rate $\lambda^{2p_0}$ to be sharp.} An important remaining question is the dependence on $p_0$ of the error term, equivalent to the convergence of the cluster-expansion series in that regime.
	\end{remark}

	We now state our large-deviation result for the number of $2k$-poles. Let us recall the hierarchical multipole model, already introduced in \eqref{eq:PXfirst}.

	\begin{definition}[Hierarchical multipole model and multipole partition function]
		For every \textcolor{black}{nonempty} $S\subset [N]$, define
		\begin{equation}\label{defM}
			\Msf_{\beta,\lambda}(S)\coloneqq \dE_{\mu_{\beta,\lambda}^{\otimes |S|}}\left[\indic_{\mc{B}_S}\prod_{i,j\in S:i<j}e^{-\beta v_{ij}}\indic_{\mc{A}_{ij}}\right].
		\end{equation}
		We also define the probability measure $\Psf_{\beta,\lambda}^S$ on $(\Lambda^{2})^{|S|}$ by
		\begin{equation}\label{defPS}
			\dd \Psf_{\beta,\lambda}^S\coloneqq \frac{1}{\Msf_{\beta,\lambda}(S)}\indic_{\mc{B}_S}\prod_{i,j\in S:i<j}e^{-\beta v_{ij}}\indic_{\mc{A}_{ij}}\prod_{i\in S}\dd \mu_{\beta,\lambda}(x_i,y_i).
		\end{equation}
		When $X$ is a subpartition of $[N]$, we define 
		\begin{equation*}
			\Psf_{\beta,\lambda}^X=\bigotimes_{S\in X} \Psf_{\beta,\lambda}^S\end{equation*}
		to be the hierarchical multipole model on  $X$.
	\end{definition}

	\begin{definition}[Multipole activity]\label{def:inter activity} For every 
		subpartition $X$ of $[N]$, we set
		\begin{equation}V_X\coloneqq \bigcup_{S\in X} S.\end{equation} 
		Recall \eqref{defmcE}, and recall that $\mathsf{E}^X$ is the set of $E\subset \mc{E}^\inter(X)$ with at least one edge such that $(V_X,E\cup \mc{E}^\intra(X))$ is connected. We define
		\begin{equation}\label{defK}
			\Ksf_{\beta,\lambda}^{\mult}(X)\coloneqq \sum_{E\in \mathsf{E}^X }\dE_{\Psf_{\beta,\lambda}^X}\left[\prod_{ij\in E}(e^{-\beta v_{ij}}\indic_{\mc{B}_{ij}^c}-1)\right].
		\end{equation}
		Notice that, for $|X|\in \{0,1\}$, we have $\Ksf_{\beta,\lambda}^{\mult}(X)=0$.
	\end{definition}

	We will show that, for $\textcolor{black}{1\leq}|S|\leq p^*(\beta)$, there exists $\msf_{\beta,\lambda}(|S|)$ such that
	\begin{equation}\label{def:little m}
		\lim_{N\to \infty}N^{|S|-1}\Msf_{\beta,\lambda}(S)=\msf_{\beta,\lambda}(|S|).
	\end{equation}
	\textcolor{black}{Notice that, for $|S|=1$, we have $\Msf_{\beta,\lambda}(S)=1$, hence $\msf_{\beta,\lambda}(|S|)=1$.} 
	
	For every subpartition $X$ of $[N]$ such that $k\coloneqq \sum_{S\in X}|S|\leq p^*(\beta)$, we will show that there exists $\ksf_{\beta,\lambda}^\mult$ such that 
	\begin{equation}\label{def:little k}
		\lim_{N\to \infty}N^{|X|-1}\Ksf_{\beta,\lambda}^{\mult}(X)=\ksf_{\beta,\lambda}^{\mult}\textcolor{black}{((n_i)_{i\geq1})},
	\end{equation}
	where $n_i$ stands for the number of elements of $X$ of cardinality $i$. \textcolor{black}{We identify a finite tuple $(n_1,\ldots,n_q)$ with the sequence $(n_1,\ldots,n_q,0,0,\ldots)$.} Notice that, for $|X|\in \{0,1\}$, we have
\[
\textcolor{black}{\Ksf_{\beta,\lambda}^{\mult}(X)=\ksf_{\beta,\lambda}^{\mult}((\#_iX)_i)=0.}
\]

	We now define the multipole rate function.
	
	\begin{definition}[Canonical partition with given block profile]\label{def:canonical}
		Let $p\in\mathbb{N}$ and $m_1,\ldots,m_p\in\mathbb{N}$, and set $k=m_1+2m_2+\cdots+pm_p$. For every set partition $Y$, let
		\[
		\#_i Y \coloneqq \bigl|\{S\in Y : |S|=i\}\bigr|
		\]
		be the number of blocks of size $i$ in $Y$.

		We define $Y(m_1,\ldots,m_p)$ to be the partition of $[k]$ obtained by taking the ordered list
		$1,2,\ldots,k$ and cutting it, in order, into blocks as follows: first form $m_1$ singletons from
		the first $m_1$ elements; next form $m_2$ consecutive pairs from the following $2m_2$ elements; and
		continue in this way until $m_p$ consecutive $p$-tuples have been formed from the last $pm_p$
		elements. By construction,
		\[
		\#_i Y(m_1,\ldots,m_p)=m_i \quad\text{for every } i=1,\ldots,p,
		\quad
		\#_i Y(m_1,\ldots,m_p)=0 \quad\text{for every } i>p.
		\]
	\end{definition}

	\begin{definition}[Multipole rate function]\label{def:rate function}
Let $\beta\in (2,\infty)$ and $p\in\mathbb N^*$ with $p\leq p^*(\beta)$. Recall $Y(m_1,\ldots,m_p)$ from Definition~\ref{def:canonical}.
		
		For every $x\in [0,1]$, let $\triangle_{p,x}\coloneqq \{(\gamma_1,\ldots,\gamma_p)\in [0,1]^p:\gamma_1+2\gamma_2+\cdots+p\gamma_p=1-x\}$. Define 
		\begin{multline*}
			\mathcal{I}_{\beta,p,\lambda}:(\gamma_1,\ldots,\gamma_p)\in \cup_{x\in [0,1]}\triangle_{p,x}\mapsto\sum_{i=1}^p \gamma_i(\log i!-1+\log \gamma_i-\log \msf_{\beta,\lambda}(i))+1
			\\+\sum_{n=1}^{\infty}\frac{(-1)^n}{n!}\sum_{m_1,\ldots,m_p\in \mathbb{N} }\prod_{i=1}^p \frac{\gamma_i^{m_i}}{m_i!}\sum_{\substack{(X_1,\ldots,X_n)\\\in \Htrees_n(Y(m_1,\ldots,m_p))\\ X_1\cup \cdots \cup X_n=Y(m_1,\ldots,m_p)}}\textcolor{black}{\omega(\mathbf X)}\prod_{i=1}^n \ksf_{\beta,\lambda}^\mult((\#_kX_i)_{k=1}^{p^*(\beta)})\indic_{|V_{X_i}|\leq p},
		\end{multline*}
		with the convention that $\gamma \log \gamma=0$ for $\gamma=0$.
	\end{definition}

	The rate function $\mc{I}_{\beta,p,\lambda}$ includes an entropy term (which counts the number of partitions), the free energy of the hierarchical model, and the correction to the hierarchical model (Mayer series).

	We can now state our large-deviation result.

	\begin{theorem}[Multipole distribution]\label{theorem:LDP}
		Let $\beta\in (2,\infty)$ and let \textcolor{black}{$p=p(\beta)$} be as in Definition~\ref{def:pbeta}. For every $k\geq 1$, denote by $\mc{N}_k$ the number of multipoles (as defined in Definition~\ref{def:multipoles}) of cardinality $k$. Then, for $M$ large enough with respect to $\beta$ and $p(\beta)$, the following hold:
		\begin{enumerate}
			\item \emph{(Coefficients of the rate function)} The limits \eqref{def:little m} and \eqref{def:little k} exist and are finite. Moreover, for every $k\in \{2,\ldots,p(\beta)\}$, there exists $C>0$ depending on $\beta,k$ and $M$ such that
			\begin{equation}\label{eq:msf thm}
				\frac{1}{C}\lambda^{2(k-1)}\leq  \msf_{\beta,\lambda}(k)\leq C\lambda^{2(k-1)}.
			\end{equation}
			For every subpartition $X$ of $\mathbb{N}^*$ such that $|V_X|\leq p(\beta)$, there exists $C>0$ depending on $\beta,|V_X|$ and $M$ such that 
			\begin{equation}\label{eq:ksf thm}
				|\ksf_{\beta,\lambda}^{\mult}( (\#_m \textcolor{black}{X})_m)|\leq C\lambda^{2(|X|-1)}.
			\end{equation}  
			\item \emph{(Free-energy expansion in terms of the rate function)}
			We  have 
			\begin{equation}\label{eq:expansion Z}
				\log Z_{N,\beta}^\lambda= 2N\log N+N\left((2-\beta)\log\lambda + \log \mc{Z}_\beta-1\right)-N \inf_{\triangle_{p,0}}\mathcal{I}_{\beta,p,\lambda}+O_{\beta,p(\beta),M}(N\delta_{\beta,\lambda}).
			\end{equation}

			\item \emph{(Control of large multipoles and large dipoles)} For every $n\geq 1$,
			\begin{equation}\label{eq:thm bad}
				\log \P\left( |\mc{N}_1+2\mc{N}_2+\cdots+p(\beta)\mc{N}_{p(\beta) }-N|\geq n\right)\leq -\tfrac{1}{2}n+O_{\beta,p(\beta),M}(N\delta_{\beta,\lambda}).
			\end{equation}
			Moreover, for every $n\geq 1$, 
			\begin{equation}\label{eq:thm bad2}
				\log \P\left(\left|\left\{i\in [N]: |x_i-y_{\sigma_N(i)}|\geq R_{\beta,\lambda}\right\}\right|\geq n\right)\leq -\tfrac{1}{2}n+O_{\beta,p(\beta),M}(N\delta_{\beta,\lambda}) .
			\end{equation}
			\item \emph{(Large deviations for the multipole counts)} Let $\gamma=(\gamma_1,\ldots,\gamma_{p(\beta)})\in \triangle_{p(\beta),0}$ and suppose that there exists $C_0>0$ such that, for every $k\in\{2,\ldots,p(\beta)\}$,
			\[
			\gamma_k\le C_0\,\lambda^{2(k-1)}.
			\] Then, 
			\begin{multline}\label{eq:LD bounds}
				\log \P\left(\max_{i\in [p(\beta)] } |\mc{N}_i-N\gamma_i|\leq C_0N\delta_{\beta,\lambda} \right)\\=-N\left(\textcolor{black}{\inf_{\substack{\gamma'\in\triangle_{p(\beta),0}\\ \|\gamma'-\gamma\|_\infty\leq C_0\delta_{\beta,\lambda}}}\mc{I}_{\beta,p(\beta),\lambda}(\gamma')}-\inf_{\triangle_{p(\beta),0}} \mc{I}_{\beta,p(\beta),\lambda}\right)+O_{\beta,p(\beta),M,C_0}(N\delta_{\beta,\lambda}),
			\end{multline}
			\textcolor{black}{for $\lambda$ small enough and $N$ large enough depending on $\lambda$.}
			\item \emph{(Minimizer of the rate function)} The infimum of $\mc{I}_{\beta,p(\beta),\lambda}|_{\triangle_{p,0}}$ is attained at a unique $\gamma^*\in \triangle_{p(\beta),0}$, and, for every $k\in\{2,\ldots,p(\beta)\}$, there exists a constant $C>0$ depending on $\beta$, $M$ and $k$ such that
			\begin{equation}\label{eq:minimizer}
				\frac{1}{C}\lambda^{2(k-1)}\leq  \gamma_k^*\leq C\lambda^{2(k-1)}.
			\end{equation}
			\item \emph{(Typical number of multipoles)} Let $k\in \{2,\ldots,p(\beta)\}$. There exists $C>0$ depending on $\beta,p(\beta)$ and $M$ such that, for $C_0$ large enough,
			\begin{equation}\label{eq:multipoles number}
				\log	\P\left(\mc{N}_k\notin \Bigl(\frac{1}{C_0}N\lambda^{2(k-1)},C_0N\lambda^{2(k-1)}\Bigr)  \right)\leq -CN\lambda^{2(k-1)}+O_{\beta,p(\beta),M}(N\delta_{\beta,\lambda}).
			\end{equation}
		\end{enumerate}
	\end{theorem}

	\begin{remark}[On the large-deviation assumption]
		In the grand-canonical setting, an expansion of the partition function conditional on the multipole fractions around the i.i.d.~dipole model 
		would yield a simpler rate function and make the assumption $\gamma_k\le C_0\,\lambda^{2(k-1)}$ in (4) unnecessary.
	\end{remark}

	\begin{remark}[On the number of multipoles of given size and cardinality]
		Let \textcolor{black}{$\beta\in (2,4)$} and $k\in \{2,\ldots, p(\beta)\}$. For every \textcolor{black}{$R\geq\lambda$}, denote by $\mc{N}_{k,R}$ the number of multipoles of cardinality $k$ with the largest dipole of size between $R$ and $2R$, and set $n_{k,R}\coloneqq NR^{2(k-1)}(\frac{\lambda}{R})^{k(\beta-2)}$. Then we expect that there exists $C>0$ depending on $\beta, p(\beta)$ and $M$ such that, for $C_0$ large enough, \textcolor{black}{$\lambda$ small enough and $N$ large enough depending on $\lambda$ and $R$,}
		\begin{equation*}
			\log \P \Bigl( \mc{N}_{k,R} \notin \Bigl(\frac{1}{C_0}n_{k,R},C_0 n_{k,R}\Bigr)\Bigr)\leq -Cn_{k,R}+O_{\beta,p(\beta),M}(N\delta_{\beta,\lambda}). 
		\end{equation*}
		Notice that since $k\leq p(\beta)\leq p^*(\beta)$, the quantity $n_{k,R}$ is decreasing in $R$. This is at the heart of the multipole transition: conditionally on seeing a multipole of cardinality $k$, one should bet that it is of small size whenever $k\leq p^*(\beta)$ and of large size whenever $k>p^*(\beta)$.
	\end{remark}

	\subsection{More details about the proof and plan of the paper}\label{sub:more}
	We assume here that we have restricted to the event where the stable matching introduced in Definition~\ref{def:stable intro} is given by $\sigma_N=\Id$.

	\subsubsection{Truncation scale}
	
	In order to get a convergent Mayer series, we need to restrict the cluster expansion to dipoles of length smaller than the critical scale. However, for $\beta\geq 4$, this critical scale is infinite. Hence, we define the following truncation scale:
	\begin{definition}\label{def:Rlambda p0}
		We let 
		\begin{equation*}
			\Cut\coloneqq \begin{cases}
				R_{\beta,\lambda} & \text{if $\beta\in (2,4)$}\\
				\lambda^{-2p_0} & \text{if $\beta\geq 4$},
			\end{cases}
		\end{equation*}
		where $R_{\beta, \lambda}$ is as in \eqref{def:Rlambda}.
	\end{definition}
	We will restrict the cluster expansion to dipoles of length smaller than $\ve_0 \Cut$, where $\ve_0\in (0,1)$ is a small number.
	
	The next task is to approximate the dipole--dipole interaction $v_{ij}$ (and hence $f_{ij}^v$) by two {\it short-range} interactions, one denoted by $\tilde v_{ij}$ which corresponds to a bound from above on the energy, and one denoted by $v_{ij}'$ which corresponds to a bound from below. These approximations will have ranges \textcolor{black}{$O(\ve_0\Cut)$} and will yield errors of size $N\delta_{\beta,\lambda}$ (see Definition~\ref{def:optimal error}). This relies on new electrostatic estimates that are given in Section~\ref{section:energy}. These two approximate models allow us to discard delicate interactions, corresponding to dipoles at large distances from each other or to highly packed dipoles, while still providing a good approximation of the true interaction $v_{ij}$.

	\subsubsection{Electric representation and electrostatic estimates}
	
	Our electrostatic estimates are based on an electric formulation of the energy, as developed for the one-component Coulomb gas in the second author's prior works \cite{SS2d,RougSer,serfaty2024lectures} and in the two-component case in \cite{LSZ,boursier2024dipole}. They provide effective upper and lower bounds 
	for the energy  based on a method of increasing the radii of the smearing balls for each charge.

	First, the {\it electric potential} $h_\lambda$ generated by the configuration $(\vec{X}_N, \vec{Y}_N)$ is defined as a function {\it on all of} $\R^2$  by 
	\begin{equation}\label{eq:heta}h_\lambda[\vec{X}_N, \vec{Y}_N]\coloneqq  \g* \left( \sum_{i=1}^N \delta_{x_i}^{(\lambda)}-\delta_{y_i}^{(\lambda)}\right),\end{equation} where $*$ denotes the convolution.
	In the sequel, we will most often drop the $[\vec{X}_N, \vec{Y}_N]$ dependence in the notation.
	
	Note that, by definition of $\g$, $h_\lambda$ satisfies  the Poisson equation
	\begin{equation}\label{eq:Deltaheta}
		-\Delta h_\lambda[\vec{X}_N, \vec{Y}_N]= 2\pi \left( \sum_{i=1}^N \delta_{x_i}^{(\lambda)}-\delta_{y_i}^{(\lambda)}\right).\end{equation}
	\textcolor{black}{Since the total charge is zero, we have $h_\lambda(x)=O(|x|^{-1})$ and $\nabla h_\lambda(x)=O(|x|^{-2})$ at infinity. Hence the boundary term vanishes when we integrate \eqref{eq:F3} by parts with \eqref{eq:Deltaheta}, and we obtain}
	\begin{equation}\label{eq:rewrF}\F_\lambda(\XN, \vec{Y}_N)= \frac{1}{4\pi} \int_{\R^2} \left|\nab h_\lambda[\vec{X}_N,\vec{ Y}_N]\right|^2 -  N (\g(\lambda)+\kappa).\end{equation}
	
	The lower bound on the energy follows from the ball-growth method 
	(see \cite[Section~4.2]{serfaty2024lectures} and references therein), 
	which consists in expanding the spherical charges $\delta_{x_i}^{(\lambda)}, \delta_{y_i}^{(\lambda)}$ into charges $\delta_{x_i}^{(\tau_i)}, \delta_{y_i}^{(\tau_i)}$ of the same mass but supported in the disks $B(x_i,\tau_i)$, $B(y_i,\tau_i)$. By Newton's theorem, if $B(x_i, \tau_i)$ and $B(x_j, \tau_j)$ are {\it disjoint}, then the interaction between the charge distributions $\delta_{x_i}^{(\tau_i)}$ and $\delta_{x_j}^{(\tau_j)}$ is equal to the interaction between $\delta_{x_i}^{(\lambda)}$ and $\delta_{x_j}^{(\lambda)}$, and the same holds with $x_i$ replaced by $y_i$ or $x_j$ replaced by $y_j$. Moreover, if they are not disjoint, one can bound from below $\int |\nab h_\lambda|^2$ by 
	the difference between the energies before ball growth and after ball growth, which, by Newton's theorem, {\it only counts interactions between  non-disjoint larger balls}.
	This way, we obtain a lower bound for the energy of the form  
	\begin{equation}\label{eqbinf}
		\begin{split}
			\F_\lambda(\vec{X}_N, \vec{Y}_N) \geq -\sum_{i=1}^N \g_\lambda(x_i - y_i) + \sum_{i < j} v_{ij}' + \text{Error},
		\end{split}
	\end{equation}
	where $v_{ij}'$ is a new {\it short-range interaction} between $\{x_i, y_i\}$ and $\{x_j, y_j\}$, which vanishes if the balls of radius $\tau_i$ and $\tau_j$
	centered at dipoles $i$ and $j$ do not intersect. 
	
	The delicate task accomplished in Subsection~\ref{sub:bad points} is to choose the $\tau_i$'s appropriately. These radii should be small enough to suppress diverging interactions (from long dipoles or highly overcrowded multipoles) but large enough so that \textcolor{black}{$v'_{ij}$} remains a good approximation of $v_{ij}$.

	To bound the partition function from below, we restrict the phase space to the event where all dipoles are smaller than the critical scale. To each pair $(x_i, y_i)$, we associate a {\it screened dipole electric field} obtained by solving 
	\begin{equation}\label{eq:neuman_corrected}
		\left\{
		\begin{array}{ll}
			-\Delta u_i = 2\pi (\delta_{x_i}^{(\lambda)} - \delta_{y_i}^{(\lambda)}) & \text{in} \ B(z_i,\tau_i), \\
			\frac{\partial u_i}{\partial \nu} = 0 & \text{on} \ \partial B(z_i, \tau_i),
		\end{array}
		\right.
	\end{equation}
	where $\nu$ denotes the outer unit normal, $z_i$ is the barycenter of $\{x_i, y_i\}$ and the $\tau_i$'s  are this time chosen equal to $\textcolor{black}{8}\ve_0 \Cut$.

	Adding the vector fields $\nab u_i$ generated by all dipoles, one obtains a vector field $E= \sum_i \nab u_i \indic_{B(z_i, \tau_i)}$, which is not necessarily a gradient, but which, thanks to the vanishing Neumann boundary condition, satisfies
	\begin{equation}\label{eq:divE_corrected}
		-\div E = 2\pi \left( \sum_{i=1}^N \delta_{x_i}^{(\lambda)} - \sum_{i=1}^N \delta_{y_i}^{(\lambda)} \right)\quad \text{in } \ \R^2.
	\end{equation}
	Using Green's theorem (which reveals the $L^2$-minimality of gradients), we may then obtain an energy upper bound  
	\begin{equation}\int_{\R^2}|\nab h_\lambda|^2 \le \int_{\R^2}|E|^2 = \int_{\R^2} \left|\sum_{i=1}^N (\nab u_i )\indic_{B(z_i, \tau_i)}\right|^2.\end{equation}
	Expanding the square on the right-hand side gives pair interactions which {\it only concern pairs of  dipoles for which the balls $B(z_i, \tau_i)$ and $B(z_j, \tau_j)$ intersect}.

	After some PDE estimates, we are led to  an upper bound of the form
	\begin{equation}\label{eq:intro u}
		\F_\lambda(\vec{X}_N,\vec{ Y}_N) \leq -\sum_{i=1}^N \g_\lambda(x_i - y_i) +  \sum_{i < j} \tilde{v}_{ij} +\text{Error},
	\end{equation}
	where $\tilde{v}_{ij}$ represents the new interaction between the screened dipoles, which vanishes at distances larger than ${\textcolor{black}{17}\ve_0 \Cut}$.

	We will work with these two reduced models, the one with dipole interactions $v_{ij}'$ as in \eqref{eqbinf}, and the one with dipole interactions $\tilde v_{ij}$ as in \eqref{eq:intro u}, and perform cluster expansions on each of  them.

	\subsubsection{Cancellation effects and integration}
	
	The dipole--dipole interaction $v_{ij}$ between dipoles with vectors $\vr_i\coloneqq y_i-x_i$ and $\vr_j\coloneqq y_j-x_j$, of lengths $r_i$ and $r_j$ respectively, is well known to be, at leading order, 
	\begin{equation}\label{ridrj}  v_{ij}\approx \textcolor{black}{\frac{\vr_i\cdot \vr_j}{\rho_{ij}^2} -2 \frac{(\vr_i\cdot \vro_{ij})(\vr_j\cdot \vro_{ij})}{\rho_{ij}^4}}\end{equation}
	\textcolor{black}{where $\vro_{ij}=(x_j+y_j-x_i-y_i)/2$ and $\rho_{ij}=|\vro_{ij}|$, when $\rho_{ij}\gg\max(r_i,r_j,\lambda)$.} In general, one has the crude bound
	\begin{equation}\label{eq:crudevintro}
		|v_{ij}|\lesssim \frac{r_i r_j}{d_{ij}\max(d_{ij},r_i,r_j)},
	\end{equation}
	which is long-range in two dimensions since the spatial integral of $d_{ij}^{-2}$ diverges.
	
	Applied to the Mayer bond $f_{ij}^v,$ this crude control produces divergent contributions in the series \eqref{sumEpi} when $\beta\geq 3$. To obtain a convergent expansion, we must exploit the angular cancellation. We split $f_{ij}^v$ into its odd part (essentially $v_{ij}$) and its even correction. In the lower-bound expansion, a parity argument shows that products of odd bonds integrate to zero unless the sum of the odd-bond degrees over each multipole is even. Thus the odd graphs are Eulerian relative to the multipole partition, as in Definition~\ref{def:Eulerian}. When all multipoles consist of a single dipole, this is the even-degree condition at every vertex. We illustrate the estimates in this case.

	The key step is that, after integrating over inter-dipole positions/angles, cluster integrals associated with an Eulerian graph \textcolor{black}{have a quadratic dependence on the dipole lengths, up to logarithms of their ratios}:
	\begin{equation}\label{eq:quad intro}
		\int r_2^2 r_3^2 \dots r_k^2 \prod_{i=1}^k e^{\beta \g_\lambda(r_i)}\indic_{r_i\leq \ve_0 \Cut } \textcolor{black}{r_i}\,\dd r_i,
	\end{equation}
	where $r_i=| \vr_i|$ is the length of the dipole \textcolor{black}{and $r_1=\max_i r_i$}. 
	
	For intuition, consider the triangle $(1,2,3)$ and assume for simplicity that $d_{ij}\geq \max(r_i,r_j)$ for each edge. If $r_3 = \max(r_1,r_2,r_3)$ and $d_{13} = \max(d_{12},d_{23},d_{13})$, then
	\begin{equation*}
		\prod_{ij} \frac{r_i r_j}{d_{ij}\max(d_{ij},r_i,r_j)}=r_1^2r_2^2 \cdot d_{12}^{-2}d_{23}^{-2}\cdot \Bigl(\frac{r_3^2}{d_{13}^2}\Bigr)\leq r_1^2r_2^2 \cdot d_{12}^{-2}d_{23}^{-2} \cdot \Bigl(\frac{r_3^2}{\max(d_{12},d_{23})^2}\Bigr).
	\end{equation*}
	Integrating over $d_{12}$ and $d_{23}$ yields the desired quadratic factor $r_1^2 r_2^2$\textcolor{black}{, up to logarithms of ratios of dipole lengths}. \textcolor{black}{We control these logarithms in Lemma~\ref{lemma:integral small -}.} Notice that we have opened the triangle by removing the variable $d_{13}$. For a general Eulerian graph, we extract a minimally 2-edge-connected spanning subgraph and perform a strict ear decomposition. Opening each ear repeats the triangle mechanism and produces the same quadratic structure.

	{\color{black}
	We then integrate $r_2,\ldots,r_k$ from $0$ to $r_1$ in \eqref{eq:quad intro}. For $R\geq2\lambda$,
	\[
	\int_0^R t^3e^{\beta\g_\lambda(t)}\,\dd t
	\asymp_\beta
	\begin{cases}
	R^{4-\beta},&2<\beta<4,\\
	1+\log(R/\lambda),&\beta=4,\\
	\lambda^{4-\beta},&\beta>4.
	\end{cases}
	\]
	Thus the radial tail is integrable for every fixed $k$ if $\beta\geq4$. For $2<\beta<4$, the integrability condition is
	\[
	(4-\beta)(k-1)+1-\beta<-1
	\quad\Longleftrightarrow\quad \beta>\beta_k.
	\]
	Since this inequality is strict, the logarithms of ratios of dipole lengths do not change it.
	}
	
	\textcolor{black}{In particular, when $\beta\in(\beta_p,\beta_{p+1})$, this sufficient condition holds for every $k\leq p$.}

	Finally, the even part of the interaction is controlled via a Penrose resummation, reducing it to a convergent tree sum. Summing over partitions of $[N]$ with prescribed multipole counts then yields the desired lower bound on the partition function (see Section~\ref{section:lower}).
	
	\subsubsection{Extracting a set of good points}

	The proof of the upper bound on the partition function, occupying Section~\ref{section:upper}, is significantly more difficult than that of the lower bound. Indeed, the absolute convergence of the cluster expansion in \eqref{eq:log} crucially relies on the fact that dipoles are of length smaller than $\ve_0\Cut$ and multipoles of bounded cardinality. The idea is to extract a set of good points and to perform a cluster expansion on them, where the frozen bad points act as an external force. We then control the number of bad points and show that their contribution to the partition function is negligible.
	
	Two aspects are particularly delicate: (i) the clustering procedure that defines the radii $\tau_i$, where we must ensure that interactions among bad points are well controlled and that each bad point has attractive interactions with only a bounded number of other charges per scale; and (ii) the fact that the good points are not an isolated system, since interactions with the bad points disrupt the structure underlying the lower bound. We discuss the second issue in the next paragraph.

	\subsubsection{Recovering the quadratic estimate}
	
	Freezing the bad points breaks the perfect odd-bond cancellations: on the set of good dipoles, the subgraph formed by odd Mayer bonds is not necessarily Eulerian or 2-edge-connected. To repair this, we analyze each connected component of the odd subgraph. It is made of a 2-core consisting of a tree of 2-edge-connected blocks, with pendant trees attached to it. Crucially, each vertex of odd degree in the odd graph interacts with bad points (in particular, the leaves of the pendant trees). We will restore the quadratic estimate \eqref{eq:quad intro} by taking advantage of the interactions with bad points.

	With our choice of radii and the good-point definition, for every good $i$,
	$$\mathrm{Interaction}(i,\text{bad points})\lesssim \frac{r_i}{\dist(\{x_i,y_i\},\text{bad points})}.$$ 
	Using a key (sub)geometric growth property of $\dist(\{x_i,y_i\},\text{bad points})$ along the graph, which is forced by our choice of enlargement radii, we replace the missing odd edges by these interactions with bad points and thereby recover the quadratic estimate \eqref{eq:quad intro}.

	\subsubsection{Identifying the infimum of the rate function}
	
	Combining the matching large-deviation upper and lower bounds identifies the limiting free energy as the infimum of the multipole rate function of Definition~\ref{def:rate function}. The remaining step in Section~\ref{section:upper} is to rewrite this infimum as the Mayer series of Theorem~\ref{theorem:expansion}. The idea is to use Theorem~\ref{theorem:LDP} to reduce to clusters and multipoles of cardinality \textcolor{black}{at most} $p(\beta)$ and then sum over the multipole partition, in order to approximate \textcolor{black}{$\log Z_{N,\beta}^\lambda-\log N!-N\log(NC_{\beta,\lambda})$} by \textcolor{black}{the logarithm of} a convergent series
	\begin{equation*}
		\textcolor{black}{\log}\sum_{n=0}^\infty\frac{1}{n!}\sum_{\substack{V_1,\ldots,V_n\subset [N]\\ \mathrm{disjoint}}}K(V_1)\cdots K(V_n),
	\end{equation*}
	where $K$ denotes an intricate activity. Since a cluster expansion of the above form (with a variable total number of points) is essentially unique by M\"obius inversion, we can identify this activity $K$ with the simple dipole activity of Definition~\ref{def:dipole activity} (for the relevant clusters). This way, we complete the proof of Theorem~\ref{theorem:expansion}.
	
	\subsection{Notation}
	
	Since the paper is quite long, we gather here some notation that has already been introduced in the introduction and/or that will be introduced later.
	
	\begin{itemize}
		\item Recalling that $N$ is a fixed integer, we set 
		$	\Lambda=[0,\sqrt{N}]^2.$
		\item The Euclidean norm in $\dR^2$ is denoted by $|\cdot|$.
		\item \textcolor{black}{We denote by $|A|$ the cardinality of a finite set $A$, or the area of a planar region $A$.}
		\item For every set $A$, we let $\mc{P}(A)$ be the set of subsets of $A$.
		\item For every $n\in \mathbb{N}$, we denote by $[n]$ the set $\{1,\ldots,n\}$.
		\item For a set $V$, we let $\mc{G}_c(V)$ be the set of collections of edges $E$ on $V$ with at least one edge and such that $(V,E)$ is connected. \textcolor{black}{We let $\mc{T}_c(V)$ be the set of spanning trees on $V$, so that $\mc{T}_c(\{i\})=\{\emptyset\}$.}
		\item For a graph $G$, we let $V(G)$ be the set of vertices and $E(G)$ be the set of edges. For every $v\in V(G)$, we write $\deg_G(v)$ for the degree of $v$ in $G$ (we will sometimes write this with an abuse of notation as $\deg_{E(G)}(v)$).
		\item Let $A$ be a set and $X_1,\ldots,X_n\subset A$. We denote by $G(X_1,\ldots,X_n)$ the connection graph of $X_1,\ldots,X_n$, which is the graph on $[n]$ with an edge between $i$ and $j$ if $X_i$ and $X_j$ intersect.
		\item We denote by $\mathrm{I}$ the Ursell function (see Definition~\ref{def:graphU}).
		\item \textcolor{black}{For every finite set $A$ and $n\in\mathbb N^*$, we recall the set $\Htrees_n(A)$ from \eqref{def:Htrees} and the weight $\omega$ of its elements.}
		\item Let $A$ be a set. The set $\mathbf{\Pi}(A)$ will denote the set of partitions of $A$ and $\mathbf{\Pi}_\sub(A)$ the set of subpartitions of \textcolor{black}{$A$}.
		\item For every subpartition $X$ of $[N]$, let
		\begin{equation*}
			V_X=\bigcup_{S\in X}S.
		\end{equation*}
		Also set
		\begin{equation*}
			\mc{E}^\inter(X)=\bigcup_{S,S'\in X:S\neq S'}\bigcup_{i\in S,j\in S'}\{ij\}\quad \text{and}\quad  \mc{E}^\intra(X)=\bigcup_{S\in X}\bigcup_{i\in S,j\in S:i\neq j}\{ij\}.
		\end{equation*}
		\item For every subpartition $X$ of $[N]$, we let $\mathsf{E}^X$ be the set of collections of edges $E\subset \mc{E}^\inter(X)$ with at least one edge and such that $(V_X,E\cup \mc{E}^\intra(X))$ is connected. \textcolor{black}{The relative trees $\mathsf T^X$ and the classes $\mathsf E_0^X$ and $\mathsf T_0^X$ are defined in Definition~\ref{def:quotient}, with $\mathsf E_0^{\{S\}}=\mathsf T_0^{\{S\}}=\{\emptyset\}$.}
		\item For every subpartition $X$ of $[N]$ and every integer $k\in \mathbb{N}$, we let $\#_k X$ be the number of parts of $X$ of cardinality $k$.
		\item Let $X$ be a subpartition of $[N]$ and let $E\subset \mc{E}^\inter(X)$. For every $S\in X$, we set 
		\begin{equation*}
			\deg_E(S)=\sum_{v\in S}\deg_{(V_X,E)}(v).
		\end{equation*}
		\item We denote by $\mc{Z}_\beta$ the integral defined in \eqref{def:Zbeta}:
		\begin{equation*}
			\mc{Z}_\beta=2\pi\int_0^{+\infty}e^{\beta \g_1(r)}r\dd r.
		\end{equation*}
		\item We denote by $\sigma_N$ the stable matching and by $\mc{A}_{ij}$ the event from \eqref{introAij}.
		\item We set $\vec{r}_i\coloneqq y_i-x_i$ and $r_i\coloneqq |y_i-x_i|$. \textcolor{black}{Set $\mathbf r_i=|x_i-y_{\sigma_N(i)}|$.} Moreover, for every \textcolor{black}{nonempty} $S\subset [N]$, we let
		\begin{equation*}
			r_S\coloneqq \max_{i\in S}r_i.
		\end{equation*}
		\item Let $M>20$ be the multipole parameter from Definition~\ref{def:multipoles}. For every $i, j \in [N]$ with $i\neq j$, define the event
		\begin{equation*}
			\mathsf{B}_{ij}\coloneqq \{i\leftrightarrow_{\sigma_N} j\},
		\end{equation*}
		where $i\leftrightarrow_{\sigma_N} j$ if
		\begin{equation*}
			\mathrm{dist}(\{x_i,y_{\sigma_N(i)}\},\{x_j,y_{\sigma_N(j)}\})\leq M \min(\max(\mathbf{r}_i,\lambda), \max(\mathbf{r}_j,\lambda)).
		\end{equation*}
		Let $S$ be a subset of $[N]$. Define the event
		\begin{equation*}
			\mathsf{B}_S\coloneqq \textcolor{black}{\{\text{the graph induced on }S\text{ by }\leftrightarrow_{\sigma_N}\text{ is connected}\}}.
		\end{equation*}
		\item 
		For every $i, j \in [N]$ with $i\neq j$, define the event
		\begin{equation*}
			\mc{B}_{ij}\coloneqq \{i\leftrightarrow j\},
		\end{equation*}
		where $i\leftrightarrow j$ if
		\begin{equation*}
			\mathrm{dist}(\{x_i,y_i\},\{x_j,y_j\})\leq M \min(\max(r_i,\lambda), \max(r_j,\lambda)).
		\end{equation*}
		Let $S$ be a subset of $[N]$. Define the event
		\begin{equation*}
			\mc{B}_S\coloneqq \textcolor{black}{\{\text{the graph induced on }S\text{ by }\leftrightarrow\text{ is connected}\}}.
		\end{equation*}
		\item For every $i,j\in [N]$ with $i\neq j$, let
		\begin{gather*}
		d_{ij} \coloneqq \dist(\{x_i,y_i\},\{x_j,y_j\}),\qquad d_{i,j}^+ \coloneqq \dist(\{x_i,y_i\},\{x_j\}),\\
		d_{i,j}^- \coloneqq \dist(\{x_i,y_i\},\{y_j\}),\qquad \textcolor{black}{\vro_{ij}} \textcolor{black}{\coloneqq (x_j+y_j-x_i-y_i)/2,\qquad\rho_{ij}\coloneqq|\vro_{ij}|.}
	\end{gather*}
		\item Given $i,j\in [N]$ with $i\neq j$, let 
		\begin{gather*}
		v_{ij} \coloneqq \g_\lambda(x_i-x_j)+\g_\lambda(y_{i}-y_{j})-\g_\lambda(x_i-y_{j})-\g_\lambda(x_j-y_{i}),\\
		v_{i,j,+} \coloneqq \g_\lambda(x_i-x_j)-\g_\lambda(y_{i}-x_j),\qquad v_{i,j,-} \coloneqq \g_\lambda(x_i-y_j)-\g_\lambda(\textcolor{black}{y_i}-y_j).
	\end{gather*}
		\item As in \eqref{vijfij}, we denote by $f_{ij}^v$ the Mayer bond 
		\begin{equation*}
			f_{ij}^v=e^{-\beta v_{ij}}\indic_{\mc{A}_{ij}}-1.
		\end{equation*}
		\item We let $R_{\beta,\lambda}$ be the critical scale introduced in \eqref{def:Rlambda}, i.e., \begin{equation*}
			R_{\beta,\lambda}= \begin{cases} \lambda^{-\frac{\beta-2}{4-\beta}} & \text{ if $\beta<4$ } \\ +\infty &\text{ if $\beta \ge 4$}\end{cases}\,.\end{equation*}
		Recall that $p_0\in \mathbb{N}^*$ is fixed throughout the paper. Moreover,  in Definition~\ref{def:Rlambda p0}, we set
		\begin{equation*}
			\Cut\coloneqq \begin{cases}
				R_{\beta,\lambda} & \text{if $\beta\in (2,4)$}\\
				\lambda^{-2p_0} & \text{if $\beta\geq 4$}.
			\end{cases}
		\end{equation*}
		\item As in Definition~\ref{def:optimal error}, we let
		\begin{equation*}
			\delta_{\beta,\lambda}\coloneqq R_{\beta,\lambda}^{-2}\indic_{\beta\in (2,4)} \left(1+\sum_{p= 2}^\infty |\log \lambda|\indic_{\beta=\beta_{p}}\right) +\lambda^{2p_0}\indic_{\beta\geq 4}.
		\end{equation*}
		Notice that, for $\beta\geq 4$, $\delta_{\beta,\lambda}\neq \Cut^{-2}$.
		\item Recall from \eqref{defpstar} 
		\begin{equation*}
			p^*(\beta)\coloneqq \sup\{\textcolor{black}{q\in\mathbb N^*}:\ \beta>\beta_q\}\in \mathbb{N}^*\cup\{+\infty\}
		\end{equation*}
		and from Definition~\ref{def:pbeta},
		\begin{equation*}
			p(\beta)=\begin{cases}
				p^*(\beta) & \text{if $\beta\in (2,4)$}\\
				p_0 & \text{if $\beta\geq 4$}.
			\end{cases}
		\end{equation*}
		\item As in \eqref{def:alphabeta}, we set
		\begin{equation*}
			\alpha(\beta)\coloneqq \begin{cases}
				\frac{4-\beta}{2} & \text{if $\beta\in (2,4)$}\\
				\frac{1}{2} & \text{if $\beta\geq 4$}.
			\end{cases}
		\end{equation*}
		\item As in \eqref{def:Clambda}, we set
		\begin{equation*}
			C_{\beta,\lambda,\ve_0}=\frac{1}{N}\int_{\Lambda^2}e^{\beta \g_\lambda(x-y)}\indic_{|x-y|\leq \ve_0 \Cut} \dd x\dd y
		\end{equation*}
		and $C_{\beta,\lambda}=C_{\beta,\lambda,\infty}$.		
		\item As in \eqref{def:qbeta}, we set
		\begin{equation*}
			q(\beta)\coloneqq 
			\begin{cases} \lfloor \frac{3}{2-\beta/2}\rfloor  & \text{if $\beta\in (2,4)$}\\
				\textcolor{black}{2p_0+1} & \text{if $\beta\geq 4$}.
			\end{cases}
		\end{equation*}
		\item As in Definition~\ref{def:gamma beta}, for every \textcolor{black}{$k\in\mathbb N^*$}, we set
		\begin{equation*}
			\gamma_{\beta,\lambda,k} =
			\begin{cases}
				\lambda^{2(k-1)} 
				& \text{if } k \leq p^*(\beta), \\[4pt]
				\Cut^{-2} 
				& \begin{aligned}
					&\text{if } k > p^*(\beta) \ \text{and} \ 
					\beta \in \big(\beta_{p^*(\beta)},\beta_{p^*(\beta)+1}\big) \\
					&\text{or } \beta = \beta_{p^*(\beta)+1} \ \text{and} \ k > p^*(\beta) + 1,
				\end{aligned} \\[8pt]
				\Cut^{-2} \lvert\log \lambda\rvert
				& \text{if } k = p^*(\beta) + 1 \ \text{and} \ 
				\beta = \beta_{p^*(\beta)+1},
			\end{cases}
		\end{equation*}
		with the convention that $\beta_\infty=4$.

		\item Throughout the paper, $\ve_0\in(0,1)$ is a small parameter used to restrict to dipoles of length at most $\ve_0\,\Cut$ and to interactions of range at most $C\ve_0\,\Cut$. The fixed parameter $p_0$ determines the largest number of dipoles in the retained elementary activities when $\beta\geq4$.
		\item By convention, we do not track the dependence of constants on $p_0$ and on the mollifier $\chi$ used to define $\delta_0^{(1)}$. All implicit constants may depend on these quantities.
	\end{itemize}

	\section{Electric formulation and energy bounds}\label{section:energy}
	
	This section collects the energy bounds used later in the paper. Subsection~\ref{sub:electric} introduces the electric formulation of the energy. From it we derive lower and upper bounds in Subsections~\ref{sub:lower} and~\ref{sub:upper}. After establishing decay estimates for the exact dipole--dipole interaction in Subsection~\ref{sub:decayv}, we quantify in Subsection~\ref{sub:errors} the error between the exact interaction and the approximations employed in these bounds.

	\subsection{Electric formulation}\label{sub:electric}
	
	This formulation is useful for obtaining both lower and upper bounds on the energy. Writing $z_1,\ldots,z_N$ for $x_1,\ldots,x_N$ and $z_{N+1},\ldots,z_{2N}$ for $y_1,\ldots,y_N$, and $d_i$ for the sign of $z_i$, one can rewrite \eqref{eq:heta} as 
	\begin{equation*} 
		h_\lambda=\g* \left( \sum_{i=1}^{2N} d_i \delta_{z_i}^{(\lambda)}\right)
	\end{equation*}
	and thus 
	\begin{equation} \label{Deltah} -\Delta h_\lambda= 2\pi \sum_{i=1}^{2N} d_i \delta_{z_i}^{(\lambda)}.\end{equation}
	
	When enlarging the disks, we similarly write,
	for every vector $\vec{\alpha}=(\alpha_1, \dots , \alpha_{2N})$ in $\textcolor{black}{(0,\infty)^{2N}}$,
	\begin{equation}\label{defhalpha} h_{\vec{\alpha}} = \g* \left( \sum_{i=1}^{2N} d_i \delta_{z_i}^{(\alpha_i)} \right).\end{equation}
	\begingroup\color{black}
	Since the total charge is zero, $h_{\vec\alpha}$ decays at infinity like $h_\lambda$, and an integration by parts gives
	\[
	\int_{\R^2}|\nabla h_{\vec\alpha}|^2=2\pi\sum_{i,j=1}^{2N}d_id_j\iint\g(x-y)\dd\delta_{z_i}^{(\alpha_i)}(x)\dd\delta_{z_j}^{(\alpha_j)}(y).
	\]
	\endgroup

	As a preliminary step, we  will need the following.
	\begin{lemma}\label{lemma:g1W2} If the density of $\chi=\delta_0^{(1)}$ is bounded, the
		functions $\g*\delta_0^{(1)}$ and  $\g_1= \g* \delta_0^{(1)}*\delta_0^{(1)}$ \textcolor{black}{have bounded and globally Lipschitz gradients, and essentially bounded weak Hessians.}
	\end{lemma}
	\begin{proof} Let $\tilde \g_1\coloneqq  \g*\delta_0^{(1)}$. 
		It suffices to use that $-\Delta \tilde \g_1= 2\pi \chi$, where the right-hand side is bounded and supported in $B(0,1)$, and the fact that $\tilde \g_1$ is radial.
		\textcolor{black}{Since $\nabla\g\in L^1_{\mathrm{loc}}$, the gradient of $\tilde\g_1$ is continuous.} Set $f(t)= \int_{B(0,t)} \chi$. \textcolor{black}{Then $f$ is absolutely continuous, $0\leq f(t)\leq\min(\pi\|\chi\|_{L^\infty}t^2,1)$, and $|f'(t)|\leq2\pi t\|\chi\|_{L^\infty}$ for almost every $t>0$.} On the other hand, Stokes's formula gives 
		$$\int_{\partial B(0,t)} \frac{\partial \tilde \g_1}{\partial \nu}=-2\pi f(t),$$ where $\nu$ denotes the outer unit normal to $B(0,t)$. Viewing $\tilde \g_1 $ as a function of one variable, we then find that 
		$-\tilde \g_1'(t)= \frac{f(t)}{t}$ \textcolor{black}{for $t>0$. By the bounds on $f$, we deduce that $|\tilde\g_1'(t)|\leq C\min(t,t^{-1})$ and $|\tilde\g_1''(t)|+|\tilde\g_1'(t)|/t\leq C$ almost everywhere. Since $\nabla\tilde\g_1(x)=O(|x|)$ at the origin, we may integrate by parts outside $B(0,\varepsilon)$ and let $\varepsilon\downarrow0$, which shows that the weak Hessian of $\tilde\g_1$ is given by the formula for radial functions. Hence $D^2\tilde\g_1\in L^\infty(\R^2)$, and $\nabla\tilde\g_1$ is globally Lipschitz. Since $\g_1=\tilde\g_1*\chi$, the function $\g_1$ satisfies the same bounds, and its Hessian is moreover continuous.}
	\end{proof}
	We next state another preliminary result.
	
	\begin{lemma}[Monotonicity]\label{lemma:mono}
		If $0<\alpha_1\leq \eta_1$ and $0<\alpha_2\leq \eta_2$, then
		\begin{equation}\label{eq:alphaeta}
			\g*\delta_0^{(\alpha_1)}*\delta_0^{(\alpha_2)}\geq  \g*\delta_0^{(\eta_1)}*\delta_0^{(\eta_2)}.
		\end{equation}
		In particular, in view of \eqref{eq:defgeta}, if $0<\alpha\leq \eta$, then 
		\begin{equation}\label{eq:alphaeta 2}
			\g_\alpha\geq \g_\eta.
		\end{equation}
	\end{lemma}
	\begin{proof}
		We first show that if $\alpha\leq \eta$, then 
		\begin{equation}\label{eq:claimgg}
			\g*\delta_0^{(\alpha)}\ge \g*\delta_0^{(\eta)}.
		\end{equation}
		
		Let $u_{\alpha,\eta}\coloneqq  \g*\delta_0^{(\alpha)}- \g*\delta_0^{(\eta)}$.
		By definition of $\delta_0^{(\eta)}$ as $\frac{1}{\eta^2}\chi(\frac{\cdot}{\eta})$ with $\chi\ge 0$ supported in  $B(0, 1)$, we have
		\begin{equation*}
			-\Delta u_{\alpha,\eta}=2\pi\Bigl(\frac{1}{ \alpha^2 }\chi\Bigl(\frac{x}{\alpha} \Bigr) - \frac{1}{\eta^2} \chi \Bigl(\frac{x}{\eta} \Bigr) \Bigr).
		\end{equation*}
		Thus, for every $r>0$, 
		\begin{align*}
			\int_{B(0, r)} \Delta u_{\alpha, \eta} &= 2\pi \int_{B(0,r)} \frac{1}{\eta^2} \chi \Bigl(\frac{x}{\eta} \Bigr)\dd x- 2\pi \int_{B(0,r)} \frac{1}{ \alpha^2 }\chi\Bigl(\frac{x}{\alpha}\Bigr) \dd x  \\ &= 2\pi \int_{B(0, \frac{r}{\eta}) } \chi(x) \dd x- 2\pi \int_{B(0,\frac{r}{\alpha}) }\chi(x) \dd x,\end{align*}
		where we have used a change of variables. Since $\alpha\le \eta$ and $\chi\ge 0$, this is nonpositive.
		
		By Stokes's theorem, we deduce that, for every $r>0$, 
		$$\int_{\partial B(0,r)} \frac{\partial u_{\alpha,\eta}}{\partial \nu} \le 0,$$ where $\nu$ is the outer unit normal to $B(0,r)$.
		But $u_{\alpha,\eta}$ is radial, so $\frac{\partial u_{\alpha,\eta}}{\partial \nu}$ is constant on $\partial B(0,r)$ and also equal to $
		\frac{\partial u_{\alpha,\eta}}{\partial r}$, so we have found that $\frac{\partial u_{\alpha,\eta}}{\partial r} \le 0$, that is, $u_{\alpha,\eta}$ is nonincreasing in $r$. On the other hand, by Newton's theorem, $\g*\delta_0^{(\alpha)}$ and $ \g*\delta_0^{(\eta)}$
		both coincide with $\g$ in $B(0,\eta)^c$, so $u_{\alpha,\eta}=0$ there. We deduce that $u_{\alpha, \eta}\ge 0$ everywhere and \eqref{eq:claimgg} holds.

		Next, we write 
		\begin{equation*}
			\g*\delta_0^{(\alpha_1)}*\delta_0^{(\alpha_2)}- \g*\delta_0^{(\eta_1)}*\delta_0^{(\eta_2)}=  \Bigl(\g*\delta_0^{(\alpha_1)}*\delta_0^{(\alpha_2)}- \g*\delta_0^{(\alpha_1)}*\delta_0^{(\eta_2)}\Bigr)+\Bigl(  \g*\delta_0^{(\alpha_1)}*\delta_0^{(\eta_2)}- \g*\delta_0^{(\eta_1)}*\delta_0^{(\eta_2)}\Bigr).
		\end{equation*}

		By \eqref{eq:claimgg},
		\begin{equation*}
			\g*\delta_0^{(\alpha_1)}*\delta_0^{(\alpha_2)}- \g*\delta_0^{(\alpha_1)}*\delta_0^{(\eta_2)}=\Bigl( \g*\delta_0^{(\alpha_2)}- \g*\delta_0^{(\eta_2)}\Bigr)*\delta_0^{(\alpha_1)} \geq 0.  
		\end{equation*}
		Similarly,
		\begin{equation*}
			\g*\delta_0^{(\alpha_1)}*\delta_0^{(\eta_2)}- \g*\delta_0^{(\eta_1)}*\delta_0^{(\eta_2)}\geq 0. 
		\end{equation*}
		Combining the above two displays completes the proof of \eqref{eq:alphaeta}.

	\end{proof}

	\subsection{Lower bound on the energy}\label{sub:lower}
	Let us now describe how to deduce easy energy lower bounds via the method of radii growth. 
	
	\begin{prop}\label{prop:mino}
		Let $(\vec{X}_N, \vec{Y}_N)$ be a configuration in $(\Lambda^2)^N$ such that $\sigma_N[\vec{X}_N, \vec{Y}_N]=\Id$, 
		where $\sigma_N		:((\Lambda^2)^N,\mc{B}((\Lambda^2)^N))\to\Sigma_N$ is  the stable matching. Let $\tau_i^{+}, \tau_i^{-}\geq \lambda$, for $1\leq i\leq N$, be a set of positive variables.  		For each $i,j\in [N]$ with $i\neq j$, let
		\begin{multline}\label{defvijp}
			v_{ij}'\coloneqq \Bigl(\g_\lambda-\g * \delta_0^{(\tau_i^{+})}*\delta_0^{(\tau_j^{+})}\Bigr) (x_i-x_j) +\Bigl(\g_\lambda-\g * \delta_0^{(\tau_i^{-})}*\delta_0^{(\tau_j^{-})}\Bigr) (y_{i}-y_{j})\\
			-\Bigl(\g_\lambda-\g * \delta_0^{(\tau_i^{-})}*\delta_0^{(\tau_j^{+})}\Bigr) (y_{i}-x_j)  -\Bigl(\g_\lambda-\g * \delta_0^{(\tau_i^{+})}*\delta_0^{(\tau_j^{-})}\Bigr) (x_i-y_{j}).
		\end{multline} 
		Assume that  there are two sets $I_1$ and $I_2$ such that $ [N]=I_1\sqcup I_2$ and 
		\begin{align}\label{assumpI1} & \text{for every} \ i \in I_1 \text{ we have } |x_i-y_{i} |\le 2(\tau_i^++\tau_i^-)  \text{ and }   \tau_i^+=\tau_i^-,\\
			\label{assumpI2} & \text{for every}\  i\  \in I_2 \text{ we have } |x_i-y_{i} |\textcolor{black}{\geq}2(\tau_i^++\tau_i^-).\end{align}
		Then, there exists a constant $C>0$ depending only on $\g_1$ such that
		\begin{multline}\label{eq:minoF}
			\F_\lambda(\vec{X}_N,\vec{Y}_N)\geq -\sum_{i\in I_1} \g_\lambda(x_i-y_{i})-\frac{1}{2}\sum_{i\in I_2} (\g(\tau_i^{+})+ \g(\tau_i^{-})+2\kappa)
			\\  +\hal \sum_{1\leq i, j\leq N: i\neq j    } v_{ij}'- C\sum_{i\in I_1} \Bigl(\frac{|x_i-y_{i}|}{\tau_i^{+}}\Bigr)^2.
		\end{multline}  
		
	\end{prop}

		\begingroup\color{black}
		Notice that each term in $v_{ij}'$ vanishes when the corresponding enlarged balls are disjoint. For instance, if $B(x_i,\tau_i^+)\cap B(x_j,\tau_j^+)=\emptyset$, then by Newton's theorem,
		\begin{equation*}
			\bigl(\g_\lambda-\g*\delta_0^{(\tau_i^+)}*\delta_0^{(\tau_j^+)}\bigr)(x_i-x_j)=0.
		\end{equation*}
		The other three terms vanish in the same way, with $(y_i,\tau_i^-)$ or $(y_j,\tau_j^-)$ in place of the corresponding positive charge and radius.
		\endgroup

	The above proposition thus allows us to bound from below the energy by the energy within dipoles plus dipole pair interactions concerning a reduced set of  pairs (that will depend on the choices of $\tau_i$), up to an error.

	\begin{proof}
		Here we introduce the notation $z_i$, for $1\leq i\leq 2N$, to denote the points in the collection, whether  positively or negatively charged, and $d_i=\pm 1$ for its charge (positive if it belongs to $\XN$, negative if it belongs to $\YN$).  
		
		For $i=1,\ldots,N$, we temporarily write $\tau_i=\tau_i^{+}$ \textcolor{black}{and $\tau_{N+i}=\tau_i^-$}. In view of \eqref{geta0} and \eqref{defhalpha}, one can express the variation of energy as
		\begin{align}\label{peq}
			\int_{\R^2} |\nab h_{\lambda}|^2-
			\int_{\R^2}  |\nab h_{\vec{\tau}}|^2 &=
			2\pi \sum_{i,j} d_id_j \left(\iint \g(x-y) \delta_{z_i}^{(\lambda)} (x)
			\delta_{z_j}^{(\lambda)} (y)-  \iint \g(x-y) \delta_{z_i}^{(\tau_i)} (x)
			\delta_{z_j}^{(\tau_j)} (y)
			\right)
			\\ \notag
			& = 
			2\pi \sum_{i=1}^{2N} (\g(\lambda)-\g(\tau_i)) \\ \notag &
			+2\pi \sum_{1\leq i\neq j\leq 2N} d_id_j \left(\iint \g(x-y) \delta_{z_i}^{(\lambda)} (x)
			\delta_{z_j}^{(\lambda)} (y)-  \iint \g(x-y) \delta_{z_i}^{(\tau_i)} (x)
			\delta_{z_j}^{(\tau_j)} (y) \right).
		\end{align}
		By \eqref{eq:rewrF} and \eqref{geta0},
		\begin{align*}
			\F_\lambda(\vec{X}_N,\vec{Y}_N) &=\frac{1}{4\pi}\int_{\dR^2}|\nabla h_\lambda|^2-N(\g(\lambda)+\kappa)\geq \frac{1}{4\pi}\Bigl( \int_{\R^2} |\nab h_{\lambda}|^2-
			\int_{\R^2}  |\nab h_{\vec{\tau}}|^2\Bigr)-N(\g(\lambda)+\kappa)\\
			& \geq -\frac{1}{2}\sum_{i=1}^{2N}\iint \g(x-y) \delta_{z_i}^{(\tau_i)}(x) \delta_{z_i}^{(\tau_i)}(y)\\ &+\frac{1}{2}\sum_{1\leq i\neq j\leq 2N} d_id_j \left(\iint \g(x-y) \delta_{z_i}^{(\lambda)} (x)
			\delta_{z_j}^{(\lambda)} (y)-  \iint \g(x-y) \delta_{z_i}^{(\tau_i)} (x)
			\delta_{z_j}^{(\tau_j)} (y) \right).
		\end{align*}
		One can perform a resummation over dipole indices. Recalling  $v_{ij}'$ defined in \eqref{defvijp} and spelling out $z_i=x_i$ \textcolor{black}{and $z_{N+i}=y_i$}, one can write
		\begin{multline*}
			\sum_{1\leq i\neq j\leq 2N} d_id_j \left(\iint \g(x-y) \delta_{z_i}^{(\lambda)} (x)
			\delta_{z_j}^{(\lambda)} (y)-  \iint \g(x-y) \delta_{z_i}^{(\tau_i)} (x)
			\delta_{z_j}^{(\tau_j)} (y) \right)\\=\sum_{1\leq i,j\leq N:i\neq j}v_{ij}'-2\sum_{i=1}^N \left(\iint \g(x-y) \delta_{x_i}^{(\lambda)} (x)
			\delta_{y_i}^{(\lambda)} (y)-  \iint \g(x-y) \delta_{x_i}^{(\tau_i^+)} (x)
			\delta_{y_i}^{(\tau_i^-)} (y) \right).
		\end{multline*}
		It follows in view of \eqref{geta0} that
		\begin{equation}\label{eq:lowerA1}
			\F_\lambda(\XN, \YN) \geq \hal \sum_{i\neq j}v_{ij}'-\sum_{i=1}^{N}  \g_\lambda(x_i-y_i)+A,
		\end{equation}
		where
		\begin{multline*}
			A\coloneqq \sum_{i=1}^N \iint \g(x-y)\delta_{x_i}^{(\tau_i^{+})}(x)\delta_{y_i}^{(\tau_i^{-})}(y)-\frac{1}{2}\sum_{i=1}^N \iint \g(x-y)\delta_{x_i}^{(\tau_i^{+})}(x)\delta_{x_i}^{(\tau_i^{+})}(y)\\-\frac{1}{2}\sum_{i=1}^N \iint \g(x-y)\delta_{y_i}^{(\tau_i^{-})}(x)\delta_{y_i}^{(\tau_i^{-})}(y)=-\frac{1}{2} \sum_{i=1}^N \iint \g(x-y)\dd(\delta_{x_i}^{(\tau_i^{+})}-\delta_{y_i}^{(\tau_i^{-})})(x)\dd(\delta_{x_i}^{(\tau_i^{+})}-\delta_{y_i}^{(\tau_i^{-})})(y).
		\end{multline*}
		Now let $i\in I_1$. Then, by assumption, $\tau_i^{+}=\tau_i^{-}$. We claim that
		\begin{equation}\label{eq:qerror}
			\iint \g(x-y)\dd(\delta_{x_i}^{(\tau_i^{+})}-\delta_{y_i}^{(\tau_i^{+})})(x)\dd(\delta_{x_i}^{(\tau_i^{+})}-\delta_{y_i}^{(\tau_i^{+})})(y)=O\Bigl(\Bigl(\frac{|x_i-y_i|}{\tau_i^{+}}\Bigr)^2\Bigr).
		\end{equation}
		Indeed, by scaling, one can check that
		\begin{equation}\label{eq:s}
			\iint \g(x-y)\dd(\delta_{x_i}^{(\tau_i^{+})}-\delta_{y_i}^{(\tau_i^{+})})(x)\dd(\delta_{x_i}^{(\tau_i^{+})}-\delta_{y_i}^{(\tau_i^{+})})(y)= \iint \g(x-y)\dd(\delta_{z}^{(1)}-\delta_{0}^{(1)})(x)\dd(\delta_{z}^{(1)}-\delta_{0}^{(1)})(y),  
		\end{equation}
		where $z\in \mathbb{R}^2$ is such that $|z|=\frac{|x_i-y_i|}{\tau_i^{+}}$. First, note that
		\begin{multline}\label{eq:break1}
			\iint \g(x-y)\dd(\delta_{z}^{(1)}-\delta_{0}^{(1)})(x)\dd(\delta_{z}^{(1)}-\delta_{0}^{(1)})(y)=-\iint \g(x-y)\dd \delta_0^{(1)}(x)\dd (\delta_{-z}^{(1)}-\delta_{0}^{(1)})(y)\\-\iint \g(x-y)\dd \delta_0^{(1)}(x)\dd(\delta_z^{(1)}-\delta_0^{(1)})(y).
		\end{multline}
		Set $\tilde \g_1=\g*\delta_0^{(1)}$. By Taylor expansion and Lemma~\ref{lemma:g1W2}, we have
		\begin{equation*}
			\iint \g(x-y)\dd\delta_{0}^{(1)}(x)\dd(\delta_{z}^{(1)}-\delta_{0}^{(1)})(y)=\int_{B(0,1)}(\tilde \g_1(x+z)-\tilde \g_1(x)){\chi(x)}\dd x=\int_{B(0,1)}(\nabla \tilde \g_1(x)\cdot z)  {\chi(x)}\dd x+O(|z|^2).
		\end{equation*}
		Since $\chi$ is radial and hence even, 
		\begin{equation*}
			\int \nabla \tilde \g_1(x) {\chi(x)}\dd x=0,
		\end{equation*}
		which gives 
		\begin{equation*}
			\iint \g(x-y)\dd\delta_{0}^{(1)}(x)\dd(\delta_{z}^{(1)}-\delta_{0}^{(1)})(y)=O(|z|^2).  
		\end{equation*}
		Replacing $z$ by $-z$, we deduce from \eqref{eq:break1} that 
		\begin{equation*}
			\iint \g(x-y)\dd (\delta_z^{(1)}-\delta_0^{(1)})(x)\dd (\delta_z^{(1)}-\delta_0^{(1)})(y)=O(|z|^2),
		\end{equation*}
		which proves \eqref{eq:qerror} in view of  \eqref{eq:s}.
		
		Now suppose that $i\in I_2$ so that $\textcolor{black}{|x_i-y_{i}|\geq2(\tau_i^{+}+\tau_i^{-})\geq4\lambda>2\lambda}$. Then,  $B(x_i,\tau_i^{+})$ and $B(y_{i},\tau_i^{-})$ are disjoint, so Newton's theorem gives
		\begin{equation}\label{eq:qerror2}
			\begin{split}
				\iint \g(x-y)\dd(\delta_{x_i}^{(\tau_i^{+})}-\delta_{y_i}^{(\tau_i^{-})})(x)\dd(\delta_{x_i}^{(\tau_i^{+})}-\delta_{y_{i}}^{(\tau_i^{-})})(y)&=(\g(\tau_i^{+})+\g(\tau_i^{-})+2\kappa)- 2\g(x_i-y_{i})\\
				&=\g(\tau_i^+) + \g(\tau_i^-) +2\kappa-2\g_\lambda(x_i-y_i),
			\end{split}
		\end{equation}in view of \eqref{geta0} and the fact that $|x_i-y_i|> 2\lambda$.
		
		Inserting \eqref{eq:qerror} and \eqref{eq:qerror2} into \eqref{eq:lowerA1} shows the result.
		
	\end{proof}

	\subsection{Upper bound on the energy}\label{sub:upper}
	The electric formulation is also useful in obtaining energy upper bounds by constructing electric potentials over a union of balls of the domain only. To do so, we consider electric fields $E=\nab h_\lambda$ and relax the condition that they be gradient fields. The Neumann boundary condition is crucial.

	\begin{lemma}\label{lemma:bsupF} 
		Let $(\vec{X}_N, \vec{Y}_N)$ be a configuration in $(\Lambda^2)^N$ such that $\sigma_N[\vec{X}_N, \vec{Y}_N]=\Id$.
		For each $i\in [N]$, let $z_i$ be the barycenter of $\{x_i,y_{i}\}$. Letting $\tau $ be such that, for every $i \in [N]$, we have  $\tau \ge 8\max(r_i,\lambda)$ (recall that $r_i=|y_i-x_i|$),  we define, for each $i$, $u_i$ to be \textcolor{black}{the weak solution with zero mean} to 
		\begin{equation}\label{eqsurui}
			\left\{ \begin{array}{ll}
				-\Delta u_i= 2\pi (\delta_{x_i}^{(\lambda)} -\delta_{y_{i}}^{(\lambda)} )& \text{in} \ B(z_i, \tau)\\
				\frac{\partial u_i}{\partial \nu}= 0 & \text{on} \ \partial B(z_i, \tau).\end{array}\right.
		\end{equation}
		
		Then 
		\begin{equation}\label{eq:upF}
			\F_\lambda(\vec{X}_N,\vec{Y}_N)\leq -\sum_{i=1}^N \g_\lambda(x_i-y_{i})+\hal\sum_{i\neq j}\tilde{v}_{ij}+O\left(\sum_{i=1}^N \frac{r_i^2}{\tau^2}\right),
		\end{equation} 
		where 
		\begin{equation}
			\label{deftvij}\tilde{v}_{ij}\coloneqq \frac{1}{2\pi}\int_{B(z_i,\tau)\cap B(z_j,\tau) } \nabla u_i\cdot \nabla u_j.
		\end{equation}
	\end{lemma}
	
	\textcolor{black}{A solution} to \eqref{eqsurui} can be computed explicitly using the classical method of image charges, which is closely related to the Schwarz reflection principle from complex analysis. To this end, consider the Neumann Green's function $G_1(x,y)$ on the unit ball $B_1$, defined as the solution (unique up to an additive constant) to
	\begin{equation*}
		\begin{cases}
			-\Delta G_1(x,y)=2\pi(\delta_y - \frac{1}{\pi}) & \text{in } B_1, \\
			\frac{\partial G_1}{\partial \nu}(x,y)=0 & \text{on } \partial B_1,\quad \text{where $y \in B_1$.}
		\end{cases}
	\end{equation*}
	
	The method of image charges introduces a fictitious charge placed outside the domain to explicitly satisfy the Neumann boundary condition. For the unit disk $B_1$, the appropriate image charge for a point charge at $y \in \textcolor{black}{B_1\setminus\{0\}}$ is obtained by reflection across the boundary circle $\partial B_1$, specifically at the point
	\begin{equation*}
		y^* = \frac{y}{|y|^2}.
	\end{equation*}
	\textcolor{black}{Identifying $\R^2$ with $\mathbb C$, we choose}
	\begin{equation}\label{formulagreen}
		G_1(x,y)= -\log|x-y| \textcolor{black}{-\log|1-x\overline y|} +\frac{|x|^2}{2},  \text{ for } |x|\le 1.
	\end{equation}
	\textcolor{black}{Up to a constant, $u_i$ is obtained by integrating $G_1((x-z_i)/\tau,(y-z_i)/\tau)$ against $\delta_{x_i}^{(\lambda)}-\delta_{y_i}^{(\lambda)}$. Its gradient is unique, since the difference of two solutions is harmonic with zero Neumann data. Since the sources lie in $B(z_i,\tau/4)$, the function $u_i$ is smooth near the boundary, where the Neumann condition follows from the identity}
	\begin{equation*}
		\textcolor{black}{\partial_\nu\bigl(-\log|x-y|-\log|1-x\overline y|\bigr)=-1\quad\text{for $|x|=1$.}}
	\end{equation*}

	We will need the following lemma.
	\begin{lemma} \label{lem35}
		Let $u_i$ be as in \eqref{eqsurui} and let  $\tilde u_i\coloneqq  u_i -  \g*(\delta_{x_i}^{(\lambda)}- \delta_{y_{i}}^{(\lambda)})$.
		We have
		\begin{align}\label{eqtu1} &
			\|\nab \tilde u_i\|_{L^\infty(B(z_i, \tau))} \le  C \frac{r_i}{\tau^2},\\ 
			\label{eqtu3} &
			|\nab u_i(x) |\le C \frac{ r_i}{ |x-z_i|^2}\quad \text{if} \ x \in B(z_i,\tau)\backslash B(z_i, 2\max(r_i, \lambda)),\\
			\label{eqtu2} & |\nab u_i(x)|\le \frac{C}{\dist(x, \{x_i\}\cup \{y_{i}\}) } + C\frac{r_i}{\tau^2} \quad \text{for all } \ x\in B(z_i, \tau),
		\end{align}
		where $C$ is universal.
	\end{lemma}
	\begin{proof} First of all, by scaling and translation invariance, it suffices to prove the result for $\tau=1$ and $z_i=0$, noting that $\nab \tilde u_i(x)$ equals $\frac{1}{\tau}$ times the gradient of its rescaled function that solves the equation on $B(0,1)$.
		The function $\tilde u_i$ solves
		\begin{equation}\label{eqsurtui}
			\left\{ \begin{array}{ll}
				-\Delta \tilde u_i= 0& \text{in} \ B(0, 1)\\
				\frac{\partial \tilde u_i}{\partial \nu}= \textcolor{black}{\left(\frac{x-x_i}{|x-x_i|^2} -\frac{x-y_i}{|x-y_i|^2} \right) \cdot \nu} & \text{on} \ \partial B(0, 1).\end{array}\right.
		\end{equation}
		\textcolor{black}{By \eqref{formulagreen}, for $x\in B(0,1)$ and a constant $c_i\in\R$, we have}
		\begin{equation}u_i(x) =  \int_{\R^2} \g(x -y)  \dd \left( \delta_{x_i}^{(\lambda)}  -\delta_{ y_i}^{(\lambda)}\right)(y)    \textcolor{black}{-\int_{\R^2}\log|1-x\overline y|\dd\left(\delta_{x_i}^{(\lambda)}-\delta_{y_i}^{(\lambda)}\right)(y)+c_i},
		\end{equation} so that 
		$$\tilde u_i(x)= \textcolor{black}{-\int_{\R^2}\log|1-x\overline y|\dd\left(\delta_{x_i}^{(\lambda)}-\delta_{y_i}^{(\lambda)}\right)(y)+c_i}.$$
		\textcolor{black}{Since $|x_i|,|y_i|\leq1/16$ and $\lambda\leq1/8$, the sources lie in $B(0,1/4)$, where the mixed derivatives $D_xD_y\log|1-x\overline y|$ are bounded, uniformly for $|x|\leq1$. Since $\delta_{y_i}^{(\lambda)}$ is the translate of $\delta_{x_i}^{(\lambda)}$ by $y_i-x_i$, we deduce that $\|\nabla\tilde u_i\|_\infty\leq Cr_i$, and \eqref{eqtu1} follows by scaling.} The relations \eqref{eqtu3} and \eqref{eqtu2} then follow directly from the properties of $\g$.

	\end{proof}
	
	\begin{proof}[Proof of Lemma~\ref{lemma:bsupF}]  Following \cite{SS2d} and subsequent works,
		we  define a global ``electric field'' $E$ by pasting together the electric fields defined over these non-disjoint balls:
		$$E\coloneqq  \sum_{i=1}^N \indic_{B(z_i, \tau)} \nab u_i.$$
		\textcolor{black}{Testing the weak formulation of \eqref{eqsurui} on each ball $B(z_i,\tau)$ against a smooth compactly supported function on $\R^2$, and summing over $i$, we obtain}
		\begin{equation}\label{divE}
			-\div E= 2\pi \left(\sum_{i=1}^N \delta_{x_i}^{(\lambda)} -\delta_{y_{i}}^{(\lambda)}\right) = -\Delta h_\lambda, \end{equation}
		where $h_\lambda$ is the electric potential of the configuration as in \eqref{eq:heta}.
		The trick is then to take advantage of the $L^2$ projection property onto gradients to show that the energy can be estimated from above by the $L^2$ norm of $E$:
		indeed $$\int_{\R^2} |E|^2= \int_{\R^2} |\nab h_\lambda|^2 + \int_{\R^2}|E- \nab h_\lambda|^2 + 2 \int_{\R^2} (E-\nab h_\lambda) \cdot \nab h_\lambda$$
		and the last term vanishes after integration by parts, in view of \eqref{divE} (noting that the boundary term vanishes). Hence 
		$\int_{\R^2} |\nab h_\lambda|^2\le \int_{\R^2} |E|^2$.
		It follows from \eqref{eq:rewrF} that 
		\begin{equation}\label{hui}\F_\lambda (\XN, \YN) \le    \frac{1}{4\pi} \sum_{i=1}^N \int_{B(z_i, \tau) } |\nab u_i|^2+\frac{1}{4\pi}\sum_{i\neq j}\int_{B(z_i,\tau)\cap B(z_j,\tau) } \nabla u_i\cdot \nabla u_j - N (\g(\lambda)+\kappa) .
		\end{equation}
		\textcolor{black}{Testing the weak equation~\eqref{eqsurui} by $u_i$, we have}
		\begin{equation*}
			\int_{B(z_i, \tau)} |\nab u_i|^2 = 2\pi \int_{B(z_i, \tau)} u_i \, \dd (\delta_{x_i}^{(\lambda)}-\delta_{y_i}^{(\lambda)} ) \end{equation*}
		and, decomposing $u_i$ as $  \g*(\delta_{x_i}^{(\lambda)}- \delta_{y_i}^{(\lambda)})+\tilde u_i$, we obtain 
		\begin{align*}
			\int_{B(z_i, \tau)} |\nab u_i|^2& = 4\pi(\g(\lambda)+ \kappa)-4\pi \g_\lambda(x_i-y_i)+ 2\pi \int_{\R^2}  \tilde u_i \,\dd (\delta_{x_i}^{(\lambda)}-\delta_{y_i}^{(\lambda)} )\\
			& = 4\pi(\g(\lambda)+ \kappa)-4\pi \g_\lambda(x_i-y_i)+ O\Bigl(\frac{r_i^2}{\tau^2}\Bigr), \end{align*}
		in view of \eqref{eqtu1}. The result follows.
	\end{proof}

	\subsection{Decay of the true dipole interaction}\label{sub:decayv}
	We denote by $v_{ij}$ the true dipole interaction.
	
	\begin{definition}[Non-truncated interaction]\label{def:vij}
		For every $i,j\in [N]$ with $i\neq j$, let
		\begin{equation*}
			v_{ij}\coloneqq \iint \g(x-y)\dd (\delta_{x_i}^{(\lambda)}-\delta_{y_{i}}^{(\lambda)})(x)\dd (\delta_{x_j}^{(\lambda)}-\delta_{y_{j}}^{(\lambda)})(y).
		\end{equation*} 
		Also let	
		\begin{align}\label{defvij+} &
			v_{i,j,+}\coloneqq \g_\lambda (x_i-x_j)-\g_\lambda(y_{i}-x_j),
			\\ \label{defvij-}&	v_{i,j,-}\coloneqq  \g_\lambda (x_i-y_{j})-\g_\lambda(y_{i}-y_{j}),\end{align}
		so that 
		$$v_{ij}= v_{i,j,+}-v_{i,j,-}.$$
		
	\end{definition}
	
	We next estimate the dipole--dipole interaction precisely, and for reference we also record in item (1) of the lemma below a more precise result than we will need, under a stronger regularity assumption on the truncation.
	\begin{lemma}[Estimating the dipole--dipole interaction]\label{lemma:vij}
		Let $x_i,x_j, y_{i}, y_{j} \in \Lambda$ be such that $\sigma_2[(x_i, x_j), (y_i,y_j)] =\Id_{\{i,j\}} $. Let $z_i, z_j$ be the barycenters of $\{x_i,y_{i}\}$ and $\{x_j,y_{j}\}$. Let $\vro_{ij}\coloneqq  z_j-z_i$,  $\rho_{ij}\coloneqq |z_i-z_j|$,  $\vr_i\coloneqq y_{i}-x_i$, $\vr_j=y_{j}-x_j$ and $d_{ij}= \dist ( \{x_i, y_{i}\}, \{ x_j, y_{j}\})$.  Let $\vec{R}_{ij}= z_j-x_i$ and $\vec{R}'_{ij}= z_j-y_{i}$.
		Let $v_{ij}$ be as in Definition~\ref{def:vij}. The following holds.
		\begin{enumerate}
			\item 
			If we assume that \textcolor{black}{the Hessian of $\g_1=\g*\delta_0^{(1)}*\delta_0^{(1)}$ is globally Lipschitz}, then, if $d_{ij} \ge  \min (r_i, r_j)$,  \textcolor{black}{after relabeling so that $r_j\leq r_i$,} we have 
			\begin{align}	\label{vijexp}
				v_{ij}
				= \textcolor{black}{-\nabla\g_\lambda(\vec R_{ij})\cdot\vr_j+\nabla\g_\lambda(\vec R'_{ij})\cdot\vr_j} + O\left(\frac{\min(r_i , r_j)^3}{  \max(|\vec{R}_{ij}|,\lambda) ^3}+\frac{\min(r_i , r_j)^3}{\max(|\vec{R}'_{ij}|,\lambda)^3}\right) ,\end{align}
			and if, in addition, $d_{ij} \ge 4 \max(r_i, r_j)$, we have 
			\begin{align}\label{vijexp2}
				v_{ij}= \textcolor{black}{-\vr_i\cdot D^2\g_\lambda(\vro_{ij})\vr_j} + O\left( \frac{\min(r_i , r_j) (r_i^2+r_j^2)}{\max(\rho_{ij},\lambda)^3}\right),\end{align}
			with 
			\begin{equation}
				\label{controledgl}
				\begin{cases}
					\g_\lambda'(t)= - \frac{1}{t}\ \text{and } \g_\lambda''(t) =  \frac{1}{t^2} &\ \text{if} \ t\ge 2\lambda\\
					|\g_\lambda'(t)|\le \frac{C}{\max(t,\lambda)} \quad |\g_\lambda'' (t)|\le \frac{C}{(\max(t,\lambda))^2} & \text{otherwise}.
				\end{cases}
			\end{equation}
			
			\item 
			If we assume \textcolor{black}{$D^2\g_1\in L^\infty$} (which follows from our assumption here), then, if $d_{ij}\ge \min(r_i,r_j)$, we have
			\begin{equation}\label{eqcorobv}
				|v_{ij}|\le C \frac{ r_ir_j}{\max(d_{ij},\lambda)\max(d_{ij}, \max(r_i,r_j))}.
			\end{equation}
			
			\item Let
			\begin{equation}\label{defdijpm}
				d_{i,j}^+\coloneqq  \dist(\{x_i,y_i\}, x_j), \quad d_{i,j}^-\coloneqq \dist(\{x_i,y_i\}, y_j).\end{equation}
			If $d_{i,j}^\pm \ge r_i$, then, $v_{i,j,\pm}$ being as in Definition~\ref{def:vij}, under the assumption that \textcolor{black}{$D^2\g_1\in L^\infty$}, we have 
			\begin{equation}\label{eqcorobv2}
				|v_{i,j,\pm}|\le C \frac{r_i}{\max(d_{i,j}^\pm, \lambda)}.
			\end{equation}
		\end{enumerate}
		\textcolor{black}{Here the constants depend only on $\|\nabla\g_1\|_{L^\infty}+\|D^2\g_1\|_{L^\infty}$, and additionally on the Lipschitz constant of $D^2\g_1$ in item~(1).}
	\end{lemma}
	
\begin{proof}
\begingroup\color{black}
By scaling,
\begin{equation}\label{425}
\g_\lambda(z)=\g(\lambda)+\g_1(z/\lambda).
\end{equation}
Since $\g_1=\g$ outside $B(0,2)$, Lemma~\ref{lemma:g1W2} gives
\begin{equation}\label{controlvarphi}
|D^q\g_\lambda(z)|\leq\frac{C}{\max(|z|,\lambda)^q}
\end{equation}
for $q=1,2$, and for $q=3$ almost everywhere under the hypothesis of item~(1).
Write $d=d_{ij}$ and relabel so that $r_j\leq r_i$.
We use repeatedly that a segment of length $\ell\leq d$ whose endpoints have distance at least $d$ from a point stays at distance at least $\sqrt{d^2-\ell^2/4}\geq\sqrt3d/2$ from that point.

For item~(1), set $\vec R=\vec R_{ij}$. Since $r_j\leq d$, the segment $\vec R+[-\tfrac12,\tfrac12]\vr_j$ stays at distance at least $|\vec R|/\sqrt3$ from zero. A Taylor expansion at $\vec R$ with the symmetric increments $\pm\vr_j/2$, in which the second-order terms cancel, therefore gives
\begin{equation}\label{eq:alter2}
\g_\lambda(x_i-x_j)-\g_\lambda(x_i-y_j)
=-\nabla\g_\lambda(\vec R)\cdot\vr_j
+O\left(\frac{r_j^3}{\max(|\vec R|,\lambda)^3}\right).
\end{equation}
Replacing $\vec R$ by $\vec R'_{ij}$ and subtracting proves \eqref{vijexp}.
Integrating along both dipoles, we also have
\begin{equation}\label{eq:mixed dipole}
v_{ij}=-\int_{-1/2}^{1/2}\int_{-1/2}^{1/2}
\vr_i\cdot D^2\g_\lambda(\vro_{ij}+s\vr_i+t\vr_j)\vr_j\dd s\dd t.
\end{equation}
If $d\geq4r_i$, the points $\vro_{ij}+s\vr_i+t\vr_j$ in \eqref{eq:mixed dipole} have modulus comparable to $\rho_{ij}$, and a first-order expansion of $D^2\g_\lambda$ at $\vro_{ij}$ gives
\[
v_{ij}=-\vr_i\cdot D^2\g_\lambda(\vro_{ij})\vr_j
+O\left(\frac{r_ir_j(r_i+r_j)}{\max(\rho_{ij},\lambda)^3}\right).
\]
Since $r_j\leq r_i$, this proves \eqref{vijexp2}.

For item~(2), suppose first that $r_j\leq d\leq4r_i$. We integrate $\nabla\g_\lambda$ along the two segments of length $r_j$ that join $x_i-x_j$ to $x_i-y_j$ and $y_i-x_j$ to $y_i-y_j$. Since they stay at distance at least $\sqrt3d/2$ from zero, we obtain $|v_{ij}|\leq Cr_j/\max(d,\lambda)$. If $d\geq4r_i$, the points $\vro_{ij}+s\vr_i+t\vr_j$ in \eqref{eq:mixed dipole} stay at distance at least $d/2$ from zero, so that $|v_{ij}|\leq Cr_ir_j/\max(d,\lambda)^2$. In both cases, \eqref{eqcorobv} follows.
Finally, $v_{i,j,\pm}$ is the integral of $\nabla\g_\lambda$ along a segment of length $r_i$ whose endpoints have distance at least $d_{i,j}^{\pm}\geq r_i$ from zero, and the same argument proves \eqref{eqcorobv2}.
\endgroup
\end{proof}
	
	\begin{remark}
		\textcolor{black}{If $d_{ij}\geq4\lambda$, the derivatives in the proof of item~(1) are evaluated where $\g_\lambda=\g$, so that item~(1) holds as soon as $\chi$ has a bounded density, as items~(2) and~(3) do.}
	\end{remark}

	\subsection{Reduced models and control of error terms}\label{sub:errors}

	\begin{lemma}[Errors from the lower bound]\label{lemma:error'}Let $x_i,x_j, y_{i}, y_{j} \in \Lambda$ be such that $\sigma_2[(x_i, x_j), (y_i,y_j)] =\Id_{\{i,j\}} $.
		Recall that    $r_i= |x_i-y_{i}|$ and $d_{ij}= \dist(\{x_i,y_{i}\},\{x_j,y_{j}\})$. Let $\tau_i^{+}, \tau_i^{-} \ge \lambda$, $\tau_j^{+}, \tau_j^{-} \ge \lambda$ be a set of positive variables.
		
		Let $v_{ij}, v_{i,j,\pm}, d_{i,j}^\pm$ be as in Definition~\ref{def:vij}  and \eqref{defdijpm}, and $v_{ij}'$ as in  \eqref{defvijp}.
		We also let 	\begin{align}\label{defvijp+} &
			v_{i,j,+}'\coloneqq \Bigl(\g_\lambda-\g * \delta_0^{(\tau_i^{+})}*\delta_0^{(\tau_j^{+})}\Bigr) (x_i-x_j) 
			-\Bigl(\g_\lambda-\g * \delta_0^{(\tau_i^{-})}*\delta_0^{(\tau_j^{+})}\Bigr) (y_{i}-x_j),
			\\ \label{defvijp-}&	v_{i,j,-}'\coloneqq  \Bigl(\g_\lambda-\g * \delta_0^{(\tau_i^{+})}*\delta_0^{(\tau_j^{-})}\Bigr) (x_i-y_{j})-\Bigl(\g_\lambda-\g * \delta_0^{(\tau_i^{-})}*\delta_0^{(\tau_j^{-})}\Bigr) (y_{i}-y_{j}).
		\end{align} 
		Then, the following holds.
		\begin{enumerate}
			\item If $d_{i,j}^{\pm}\geq r_i$, then
			\begin{equation}\label{diffvijcassuppl2}
				|v_{i,j,\pm}'|\le C \left(\frac{r_i}{\max(d_{i,j}^\pm,\lambda)} + \frac{r_i}{\tau_j^\pm}+\frac{\max(\tau_i^+, \tau_i^-)^2}{(\tau_j^\pm)^2}\right)\indic_{d_{i,j}^\pm\le \max(\tau_i^+,\tau_i^-)+\tau_j^\pm}
			\end{equation}
			and
			\begin{equation}\label{diffvijcassuppl}
				|v_{i,j,\pm}-v_{i,j,\pm}'|\le C \left( \frac{r_i}{\tau_j^\pm}+\frac{\max(\tau_i^+, \tau_i^-)^2}{(\tau_j^\pm)^2}\right)\indic_{d_{i,j}^{\pm}\le \max(\tau_i^+,\tau_i^-)+\tau_j^\pm} + C \frac{r_i}{\max(d_{i,j}^\pm,\lambda)} \indic_{d_{i,j}^\pm>\max(\tau_i^+,\tau_i^-)+\tau_j^\pm}.
			\end{equation}
			
			\item 
			If $d_{i,j}^\pm \ge r_i$ and $\tau_i^+=\tau_i^-=\tau_i$,  we have 
			\begin{equation}\label{eqvijpp}
				|v_{i,j,\pm}'|\le C\frac{r_i}{\max(d_{i,j}^{\pm},\lambda)} \indic_{d_{i,j}^\pm\le \tau_i +\tau_j^\pm}
			\end{equation}
			and 
			\begin{equation}\label{eqdiffdesvivvj}
				|v_{i,j,\pm}- v_{i,j,\pm}'|\le C \frac{r_i}{\max(\tau_i,\tau_j^\pm)} \indic_{d_{i,j}^\pm\le \tau_i+\tau_j^\pm} +C\frac{r_i}{\max(d_{i,j}^\pm,\lambda)}\indic_{d_{i,j}^\pm> \tau_i+\tau_j^\pm}  . \end{equation}

			\item Assume that  $r_i\le 2(\tau_i^++\tau_i^-)$, $\tau_i^+=\tau_i^-=\tau_i$,  the same holds for $j$,  and $d_{ij} \ge \min(r_i,r_j)$. Then we have 
			\begin{equation}\label{eq:v'}
				|v_{ij}'|\leq C \frac{r_ir_j}{\max(d_{ij},\lambda)\max(r_i,r_j,d_{ij})}\indic_{d_{ij}\le \tau_i+\tau_j} 
			\end{equation}
			and
			\begin{equation}\label{eq:diff v}
				|v_{ij}'-v_{ij}|\leq  			C	\frac{r_ir_j}{\max(\tau_i,\tau_j)^2} \indic_{d_{ij}\le \tau_i+\tau_j}+C\frac{r_ir_j}{\max(d_{ij},\lambda)\max(r_i,r_j,d_{ij})} \indic_{d_{ij} >\tau_i+\tau_j}. \end{equation}
            \textcolor{black}{The bound $|v_{ij}'-v_{ij}|\leq Cr_ir_j/\max(\tau_i,\tau_j)^2$ also holds when $d_{ij}<\min(r_i,r_j)$.}
		\end{enumerate}
	\end{lemma}
	
	\medskip	
	\begin{proof} 
		Let us first prove item (1) and consider the $+$ case, the other one being analogous.
		For the first item, we note that if $d_{i,j}^+>\max(\tau_i^+,\tau_i^-)+\tau_j^+$, then 
		$B(x_i, \tau_i^+)\cap B(x_j,\tau_j^+)=\emptyset$ and $B(y_i ,\tau_i^-)\cap B(x_j, \tau_j^+)=\emptyset$, which implies, by Newton's theorem, that $v_{i,j,+}'=0$. The results then follow in view of Lemma~\ref{lemma:vij}.
		
		Let us turn to the case where $d_{i,j}^+\le \max(\tau_i^+,\tau_i^-)+\tau_j^+$.
		By the definitions, we have
		\begin{equation}\label{eqproofle1}
			v_{i,j,+}'-v_{i,j,+}= - \g*\delta_0^{(\tau_i^+)}* \delta_0^{(\tau_j^+)} (x_i-x_j)
			+\g*\delta_0^{(\tau_i^-)}*\delta_0^{(\tau_j^+)}(y_i-x_j).
		\end{equation}

		Setting $f= \g *  \delta_{x_j}^{(\tau_j^+)}$, we may rewrite 
		this as 
\begin{equation*}
v_{i,j,+}'-v_{i,j,+} ={-} \int f ( \delta_{x_i}^{(\tau_i^+)}- \delta_{y_i}^{(\tau_i^-)}) = {-} \int f ( \delta_{x_i}^{(\tau_i^+)}- \delta_{y_i}^{(\tau_i^+) }){-} \int f ( \delta_{y_i}^{(\tau_i^+) }- \delta_{y_i}^{(\tau_i^-)}).
\end{equation*}
		The first term on the right-hand side is bounded by $|x_i-y_i|\|\nab f\|_{L^\infty}$. Moreover, one can check in view of Lemma~\ref{lemma:g1W2} and scaling that $f$ is Lipschitz with constant $ C(\tau_j^+)^{-1}$. Thus the first term is bounded by $ C\frac{r_i}{\tau_j^+}$. For the second term, Taylor-expanding $f$ to order $2$ near $y_i$, we obtain 
		$$\Bigr|\int f(\delta_{y_i}^{(\tau_i^+) }- \delta_{y_i}^{(\tau_i^-)})\Bigr| \le C\|D^2 f\|_{L^\infty}\max(\tau_i^+, \tau_i^-)^2  \le C'\frac{\max(\tau_i^+, \tau_i^-)^2}{(\tau_j^+)^2}.
		$$	
		This proves \eqref{diffvijcassuppl}. Moreover, \eqref{diffvijcassuppl2} follows from Lemma~\ref{lemma:vij}.
		
		Let us turn to the second item,  considering again the $+$ case.
		If $d_{i,j}^+ >\tau_i+ \tau_j^+ $, then $B(x_i, \tau_i) \cap B(x_j, \tau_j^+)=\emptyset $ and $B(y_i, \tau_i) \cap B(x_j, \tau_j^+)=\emptyset $, so by Newton's theorem 
		$v_{i,j,+}'=0$.	The results then follow from Lemma~\ref{lemma:vij}.

		Let us thus consider the case $d_{i,j}^+\le \tau_i+ \tau_j^+$. 
		Letting $$f\coloneqq   \g*\delta_{0}^{(\tau_i^+)} * \delta_{x_j}^{(\tau_j^+)}= \g*\delta_{0}^{(\tau_i^-)} * \delta_{x_j}^{(\tau_j^+)},$$
		we may thus write in view of \eqref{eqproofle1} that 
		$$|v_{i,j,+}'-v_{i,j,+}|=|f (x_i)-f(y_i)|\le \|\nab f\|_{L^\infty} r_i.$$
		\textcolor{black}{By commutativity of convolution, $\nabla f$ is the convolution of $\nabla\g*\delta_0^{(\max(\tau_i,\tau_j^+))}$ with a probability measure. By Lemma~\ref{lemma:g1W2} and scaling, we deduce that $\|\nabla f\|_\infty\leq C/\max(\tau_i,\tau_j^+)$.}
		This proves \eqref{eqdiffdesvivvj}, and  from \eqref{eqcorobv2}, \eqref{eqvijpp} follows by using $d_{i,j}^+\le\tau_i+\tau_j^+\le 2\max(\tau_i, \tau_j^+)$ to absorb terms.

		Let us turn to the third item.	  If $d_{ij}> \tau_i+\tau_j$, then 
		$B(x_i, \tau_i)\cap B(x_j,\tau_j)=\emptyset$ and the same with $(y_i, y_j)$, $(x_i, y_j)$,  and $(y_i, x_j)$. By Newton's theorem, this implies that $v_{ij}'=0$, and the results follow in view of \eqref{eqcorobv}. Let us now turn to the case  $d_{ij}\le \tau_i+\tau_j\le 2 \max(\tau_i, \tau_j)$. 
		We have 
		\begin{multline}\label{lbtermestau}
			v_{ij}'-v_{ij}=- \g*\delta_0^{(\tau_i)}* \delta_0^{(\tau_j)} (x_i-x_j)
			+\g*\delta_0^{(\tau_i)}*\delta_0^{(\tau_j)}(x_i-y_j)
			\\+\g*\delta_0^{(\tau_i)}* \delta_0^{(\tau_j)}(x_j-y_i)- \g*\delta_0^{(\tau_i)}* \delta_0^{(\tau_j)}(y_i-y_j).		\end{multline}    
		Letting $f\coloneqq \g*\delta_0^{(\tau_i)}*\delta_0^{(\tau_j)}$ and $h= f(y_i-\cdot) -f(x_i-\cdot) $, we can rewrite the above as 
		$$v_{ij}'-v_{ij}= h(x_j)-h(y_j),$$
		hence 
		\begin{equation}\label{avhessien}|v_{ij}'-v_{ij}|\le |x_j-y_j| \|\nab h\|_{L^\infty} \le |x_j-y_j||x_i-y_i| \|D^2 f\|_{L^\infty}.\end{equation}
		But $f$ is at least as regular as $\g*\delta_0^{(\max(\tau_i, \tau_j))}$,  while by Lemma~\ref{lemma:g1W2} and scaling, we have
		$$|D^2 \g*\delta_0^{(\max(\tau_i, \tau_j))}|\le \frac{C}{\max(\tau_i, \tau_j)^2}.$$
		Inserting this bound into \eqref{avhessien} gives \eqref{eq:diff v}. \textcolor{black}{Neither of these two bounds uses $d_{ij}\geq\min(r_i,r_j)$, so they also give the last assertion of item~(3).} \textcolor{black}{Since $d_{ij}\leq\tau_i+\tau_j$, $r_i\leq4\tau_i$ and $r_j\leq4\tau_j$, we have}
		\[
		\textcolor{black}{\max(d_{ij},\lambda)\max(d_{ij},r_i,r_j)\leq8\max(\tau_i,\tau_j)^2.}
		\]
		\textcolor{black}{Combining \eqref{eq:diff v} and \eqref{eqcorobv} now proves \eqref{eq:v'}.}
	\end{proof}

	\begin{lemma}[Errors from the upper bound]\label{lemma:errortilde}Let $x_i,x_j, y_{i}, y_{j} \in \Lambda$ be such that $\sigma_2[(x_i, x_j), (y_i,y_j)] =\Id_{\{i,j\}} $.
		Let  $\tau>0$  be such that $\tau \ge 8\max(r_i, \lambda)$ for all $i$. Assume that $d_{ij}\ge \min (r_i,r_j)$.
		Let $v_{ij}$ be as in Definition~\ref{def:vij} and $\tilde{v}_{ij}$ be as in \eqref{deftvij}.
		There exists $C>0$ (depending only on the regularization $\g_1$)  such that
		\begin{equation}\label{estisurtvij}
			|\tilde{v}_{ij}|\leq C
			\frac{r_ir_j}{\max(d_{ij},\lambda)\max(r_i,r_j,d_{ij})}\indic_{d_{ij}\le \textcolor{black}{\tfrac{17}{8}\tau}}.
		\end{equation}
		Moreover,
		\begin{equation}
			\label{comptvijvij}	|\tilde{v}_{ij}-v_{ij}|\leq C
			\frac{r_ir_j}{\max(d_{ij},\lambda)\max(r_i,r_j,d_{ij})}\indic_{d_{ij}>\textcolor{black}{\tfrac{17}{8}\tau}}+ 				\textcolor{black}{C}\frac{r_ir_j }{\tau^2} \indic_{d_{ij}\le \textcolor{black}{\tfrac{17}{8}\tau}}.
		\end{equation}
        \textcolor{black}{The bound $|\tilde v_{ij}-v_{ij}|\leq Cr_ir_j/\tau^2$ also holds when $d_{ij}<\min(r_i,r_j)$.}
		
	\end{lemma}

\begin{proof}
\begingroup\color{black}
For $l\in\{i,j\}$, let
\[
h_l=\g*(\delta_{x_l}^{(\lambda)}-\delta_{y_l}^{(\lambda)}),
\qquad D_l=\indic_{B(z_l,\tau)}\nabla u_l-\nabla h_l.
\]
By polarization of \eqref{eq:rewrF}, $v_{ij}=(2\pi)^{-1}\int\nabla h_i\cdot\nabla h_j$. Outside $B(z_l,\tau)$, Newton's theorem shows that $h_l=\g(\cdot-x_l)-\g(\cdot-y_l)$, and the mean value theorem then gives $|\nabla h_l(x)|\leq Cr_l/|x-z_l|^2$. Inside the ball, $D_l=\nabla\tilde u_l$ is bounded by \eqref{eqtu1}. Integrating both bounds, we obtain
\[
\|D_l\|_{L^2(\mathbb R^2)}^2\leq C r_l^2/\tau^2.
\]
By the zero Neumann condition in \eqref{eqsurui}, we have $\operatorname{div}D_l=0$ in the sense of distributions on $\R^2$. Since $h_m(x)=O(|x|^{-1})$ and $\nabla h_m(x)=O(|x|^{-2})$ at infinity, an integration by parts with a cutoff then shows that $\int D_l\cdot\nabla h_m=0$ for $l,m\in\{i,j\}$. Therefore,
\[
|\tilde v_{ij}-v_{ij}|=\frac1{2\pi}\left|\int_{\mathbb R^2}D_i\cdot D_j\right|
\leq C\frac{r_ir_j}{\tau^2}.
\]
If the balls $B(z_i,\tau)$ and $B(z_j,\tau)$ intersect, then
\[
d_{ij}\leq |z_i-z_j|+\frac{r_i+r_j}{2}\leq\frac{17}{8}\tau.
\]
Hence $\tilde v_{ij}=0$ when $d_{ij}>17\tau/8$. Combining this with the preceding estimate and \eqref{eqcorobv} proves \eqref{estisurtvij} and \eqref{comptvijvij}, since $\max(d_{ij},\lambda)\max(d_{ij},r_i,r_j)\leq C\tau^2$ when $d_{ij}\leq17\tau/8$.
\endgroup
\end{proof}
	\begin{definition}[Interaction of the energy upper bound model]\label{def:vijtilde}
		From now on, we let $\tilde v_{ij}$ be as in \eqref{deftvij} with the choice $\tau=\tau_0$, where
		\begin{equation}\label{deftauub}
			\tau_0= 8\ve_0 \Cut.\end{equation}
		Here $\Cut$ is defined in Definition~\ref{def:Rlambda p0}, \textcolor{black}{and $\ve_0\in(0,1)$ is a free parameter, whose smallness will be fixed in Proposition~\ref{prop:absolute lower}.}
	\end{definition}
	With that definition, \textcolor{black}{let $L=\beta C/64$, where $C>0$ is the constant implicit in \eqref{eq:upF}. If $\lambda\leq\ve_0\Cut$, then Lemma~\ref{lemma:bsupF} shows that,} on the event $\{\sigma_N=\Id\}\cap \cap_{i\in [N]}\{|x_i-y_i|\leq \ve_0 \Cut\}$, we have
	\begin{equation}\label{eq:b2}
		\F_\lambda(\vec{X}_N,\vec{Y}_N)\leq -\sum_{i=1}^N \g_\lambda(x_i-y_{i})+\hal\sum_{i\neq j}\tilde{v}_{ij}+\frac{L}{\beta}\left(\sum_{i=1}^N \frac{|x_i-y_{i}|^2}{(\ve_0 \Cut)^2}\right).
	\end{equation}

	\section{Large-deviation lower bound}\label{section:lower}

	In this section, we prove the large-deviation lower bound of Theorem~\ref{theorem:LDP}. We begin by providing more details on the stable matching in Subsection~\ref{sub:stable}, then state a general cluster-expansion result that represents the logarithm of a formal series as another series. In Subsection~\ref{sub:graphs and partitions}, we collect additional graph-theoretic and partition terminology. A detailed strategy for the proof of the lower bound is given in Subsection~\ref{sub:strat lower}.

	\subsection{Stable matching}\label{sub:stable}
	Recall that a matching is a random permutation of $[N]$, i.e., a way to assign to each $x_i$ a match $y_{\sigma_N(i)}$. A matching is stable with respect to an ordering of preferences if there does not exist any pair of points for which it would be preferable to switch partners. 
	
	We let $\Sigma_N$ denote the set of permutations on $[N]$.
	\begin{definition}\label{def:stable match}
		Let us introduce a random matching to pair particles of opposite sign. We define a random variable $\sigma_N$ from $((\Lambda^2)^N,\mc{B}((\Lambda^2)^N))$ to $\Sigma_N$ using the following procedure:
		\begin{itemize}
			\item Enlarge all points $x_i$, for $1\leq i\leq N$, into $B(x_i,\ve)$ for the same parameter $\ve>0$.
			\item Increase $\ve$ until some $B(x_i,\ve)$ touches some negative charge $y_j$. Set $\sigma_N(i)=j$. \textcolor{black}{For simultaneous contacts, choose the lexicographically smallest pair $(i,j)$.}
			\item Remove $x_i$ and $y_j$ from the configuration and repeat.
		\end{itemize}
	\end{definition}
	\textcolor{black}{We denote by $\sigma_N[\vec{X}_N,\vec{Y}_N]$ the permutation of $[N]$ given by this procedure. It is measurable, since each of its finitely many choices is determined by comparisons of distances.}
	\begingroup\color{black}
	We may exclude the null set on which two opposite-sign distances indexed by distinct pairs coincide, or two same-sign distances indexed by distinct unordered pairs coincide. For the clustering of Definition~\ref{def:continuous clustering}, we may also exclude the configurations in which $|x_i-x_j|/2$ or $|y_i-y_j|/2$ equals $r_k/4$ or $\ve_0\Cut$. Indeed, for a fixed matching permutation, each of these equalities becomes, after squaring, a nonzero polynomial equation.

	For every measurable event $\mc A$ invariant under relabeling the $y_j$, we have
	\begin{equation}\label{eq:relabel}
	\int_{\mc A}e^{-\beta\F_\lambda}=N!\int_{\mc A\cap\{\sigma_N=\Id\}}e^{-\beta\F_\lambda}.
	\end{equation}
	Indeed, the events $\{\sigma_N=\sigma\}$, for $\sigma\in\Sigma_N$, partition the configuration space. Outside the null set above, relabeling the $y_j$ maps each of them onto $\{\sigma_N=\Id\}$, and it preserves the Lebesgue measure and $\F_\lambda$.
	\endgroup

	\begin{remark}[Uniqueness of the stable matching]
		When preferences are ordered by (minimal) distances and are strict, there is a unique stable matching, and the algorithm of Definition~\ref{def:stable match} produces it. This algorithm is a special case of the Gale--Shapley algorithm. For two independent Poisson point processes of intensity $1$ in the plane, Holroyd, Pemantle, Peres and Schramm showed that the distance from a typical point to its partner in the stable matching has an infinite second moment, and that the probability that it exceeds $r$ is at most $Cr^{-s}$ for every $r>0$, where $C,s>0$ are constants \cite[Theorem~6]{holroydpemantleperesschramm2009poisson}.
	\end{remark}

	\begin{lemma}\label{lemma:stable}
		Let $(\vec{X}_N,\vec{Y}_N)$ be a point configuration in $(\Lambda^2)^N$ such that the distances $|x_i-y_j|$, for $1\leq i,j\leq N$, are all distinct. Then the following statements are equivalent:
		\begin{enumerate}
			\item $\sigma_N[\vec{X}_N,\vec{Y}_N]=\Id$,
			\item For every $1\leq i\neq j \leq N$, $\sigma_2[(x_i,x_j),(y_i,y_j)]=\Id_{\{i,j\}}$.
		\end{enumerate}
	\end{lemma}

	The above lemma means that $\sigma_N$ is equal to the identity if and only if there is no ``blocking pair'', i.e., if and only if for every pair of couples $1\leq i, j\leq N$ with $i\neq j$, $x_i$ is assigned to $y_i$ and $x_j$ is assigned to $y_j$ when performing the stable matching between $x_i,x_j$ and $y_i,y_j$ only.

	\begin{definition}\label{def:Aij}
		For every $i,j\in [N]$ with $i\neq j$, we set
		\begin{equation}\label{defAij}
			\mc{A}_{ij}\coloneqq \{(\vec{X}_N,\vec{Y}_N)\in (\Lambda^2)^N:\sigma_2[(x_i,x_j),(y_i,y_j)]=\Id_{\{i,j\}}\}.
		\end{equation}
	\end{definition}
	
	\medskip

	\begin{proof}[Proof of Lemma~\ref{lemma:stable}]
		We first prove that (1) implies (2). Let $(\vec{X}_N,\vec{Y}_N)$ satisfy $\sigma_N[\vec{X}_N,\vec{Y}_N]=\Id$. Let $i_1$ be the index of the first ball $B(x_i,\ve)$ that touches a negative charge. Then 
		\begin{equation*}
			(\vec{X}_N,\vec{Y}_N)\in \bigcap_{j:j\neq i_1}\mc{A}_{i_1j}.
		\end{equation*}
		Now let $i_2$ be the index of the second ball to touch. We have similarly
		\begin{equation*}
			(\vec{X}_N,\vec{Y}_N)\in \bigcap_{j:j\neq i_2}\mc{A}_{i_2j}.
		\end{equation*}
		Iterating until the last ball touches a negative charge shows that 
		\begin{equation*}
			(\vec{X}_N,\vec{Y}_N)\in \bigcap_{i\neq j}\mc{A}_{ij}.
		\end{equation*}
		
		We now prove that (2) implies (1). Let $(\vec{X}_N,\vec{Y}_N)\in \cap_{i\neq j}\mc{A}_{ij}$. Let $i_1$ be the index of the first ball which touches. Then $\sigma_N[\vec{X}_N,\vec{Y}_N](i_1)=i_1$ because $(\vec{X}_N,\vec{Y}_N)\in \cap_{j:j\neq i_1}\mc{A}_{i_1,j}$. Let $i_2$ be the index of the second ball which touches. Then $\sigma_N[\vec{X}_N,\vec{Y}_N](i_2)=i_2$ because $(\vec{X}_N,\vec{Y}_N)\in \cap_{j:j\neq i_1,i_2}\mc{A}_{i_2,j}$. We iterate the argument until the last ball touches a negative point.
	\end{proof}
	
	\begin{lemma}\label{lemma:Aij}
		Recall that $r_i=|x_i-y_i|$ and $r_j=|x_j-y_j|$. Suppose that $r_i\le r_j$. \textcolor{black}{Assume that the four opposite-sign distances are distinct.} Then the following statements are equivalent:
		\begin{enumerate}
			\item \begin{equation*}
				\textcolor{black}{\sigma_2[(x_i,x_j),(y_i,y_j)]}=\Id_{\{i,j\}},
			\end{equation*}
			\item 
			\begin{equation*}
				|y_j-x_i|\ge r_i \quad \text{and}\quad |y_i-x_j|\ge r_i.
			\end{equation*}
		\end{enumerate}
	\end{lemma}
	\begingroup\color{black}
	\begin{proof}
	Since the four distances are distinct and $r_i\leq r_j$, the first match is the pair at the smallest of the distances $r_i$, $|y_j-x_i|$ and $|y_i-x_j|$. If both cross distances exceed $r_i$, this pair is $(x_i,y_i)$, and the second match is $(x_j,y_j)$. Otherwise, the first match is a cross pair, so the matching is not the identity.
	\end{proof}
	We will also use the uniform bound
	\begin{equation}\label{eq:bond-bound}
	e^{-\beta v_{ij}}\indic_{\mc A_{ij}}\leq C_\beta.
	\end{equation}
	Indeed, by scaling and Newton's theorem, $\g_\lambda(w)=-\log\max(|w|,\lambda)+O_\chi(1)$ uniformly in $w$ and $\lambda>0$. By Lemma~\ref{lemma:Aij}, both cross distances are at least $\min(r_i,r_j)$ on $\mc A_{ij}$. If $r_i\leq r_j$, the triangle inequality then gives
	\[
	|x_i-x_j|\leq2|x_j-y_i|,\qquad |y_i-y_j|\leq2|x_i-y_j|.
	\]
	Inserting these bounds into $v_{ij}=v_{i,j,+}-v_{i,j,-}$, with \eqref{defvij+} and \eqref{defvij-}, we obtain $v_{ij}\geq-C$. The case $r_j\leq r_i$ is symmetric. Since $\tau_0=8\ve_0\Cut$, the last assertion of Lemma~\ref{lemma:errortilde} gives the same bound for $\tilde v_{ij}$ when $r_i,r_j,\lambda\leq\ve_0\Cut$.
	\endgroup
	
	\begin{remark}\label{remark:the inclusion}
		Recall $\mc{B}_{ij}$ from \eqref{eq:defBij mc}. Observe that 
		\begin{equation*}
			\mc{B}_{ij}^c\subset \mc{A}_{ij}.
		\end{equation*}
		Indeed, suppose that $r_i\leq r_j$ and $d_{ij}>M\max(\min(r_i,r_j),\lambda)=M\max(r_i,\lambda)$. Then, since $M\geq 2$, we have 
		\begin{equation*}
			|y_j-x_i|\textcolor{black}{>}r_i \quad \text{and}\quad |y_i-x_j|\textcolor{black}{>}r_i,
		\end{equation*}
		\textcolor{black}{so that the shortest of the four opposite-sign distances is $r_i$. The stable matching of the two dipoles is therefore the identity.} Proceeding similarly when $r_i\geq r_j$, we deduce that $\mc{B}_{ij}^c\subset \mc{A}_{ij}$.
	\end{remark}

	\subsection{Connected-cluster resummation}\label{sub:resum}
	The following lemma gives a (formal) identity for the logarithm of a cluster-expansion series.

	\begin{lemma}\label{lemma:resum2}
		Let $C$ be a finite set with a symmetric, reflexive relation expressing the intersection of its elements. Let $K:C\to \dR$. We say that $X_1,\ldots,X_n\in C$ are connected if their connection graph $G(X_1,\ldots,X_n)$ is connected. Let $\mathrm{I}$ be the Ursell function defined in \eqref{def:UrsellI}. We have
		\begin{equation}\label{eq:resum}
			\log \sum_{n=0}^{+\infty}\frac{1}{n!}\sum_{\substack{X_1,\ldots,X_n\in C \\ \mathrm{disjoint}}}K(X_1)\cdots K(X_n)=\sum_{n=1}^{+\infty}\frac{1}{n!}\sum_{\substack{X_1,\ldots,X_n\in C\\ \mathrm{connected}}}K(X_1)\cdots K(X_n)\mathrm{I}(G(X_1,\ldots,X_n)),
		\end{equation}
		as an identity between formal series in $K$.
	\end{lemma}
	
	\textcolor{black}{If the right-hand side of \eqref{eq:resum} converges absolutely, then the equality holds in $\dR$. Indeed, for $t\in[0,1]$, let $Z(t)$ be the sum inside the logarithm in \eqref{eq:resum}, which is a polynomial in $t$, and let $F(t)$ be the right-hand side, both with $tK$ in place of $K$. Since the series $F(1)$ converges absolutely, $F$ is continuous on $[0,1]$ and differentiable for $0\leq t<1$. The formal identity gives $Z'=F'Z$ there, so that $(Ze^{-F})'=0$. Since $Z(0)=1$ and $F(0)=0$, we obtain by continuity that $Z(1)=e^{F(1)}>0$.}

	The statement and the proof of Lemma~\ref{lemma:resum2} can be found in \cite[Theorem 4.4]{Bauerschmidt2016FerromagneticSS}.

	\begin{remark}
		In the grand-canonical setting, the identity \eqref{eq:resum} assumes a more concise form.  
		As a concrete illustration, let
		\[
		K:\mathcal{P}(\mathbb{N})\longrightarrow\mathbb{R}, 
		\qquad 
		K(\{i\})=0 \quad\text{for every } i\in\mathbb{N},\quad K(\emptyset)=0
		\]
		so that 
		\begin{equation*}
			\sum_{n=0}^\infty \frac{1}{n!}\sum_{\substack{ V_1,\ldots,V_n\subset [N]\\ \mathrm{disjoint} }}K(V_1)\cdots K(V_n)=\sum_{\substack{X\text{ partition}\\\text{of }[N]}}\prod_{S\in X}\tilde{K}(S),
		\end{equation*}
		where 
		\begin{equation*}
			\tilde{K}(X)=\begin{cases}
				K(X) & \text{if $|X|>1$}\\
				1 & \text{if $|X|\in \{0,1\}$}.
			\end{cases}
		\end{equation*}
		\textcolor{black}{Assume that $K(S)$ depends only on $|S|$ for finite $S$.} Then, regarding all sums below as formal power series, we have
		\[
		\log\!\Biggl[\;
		\sum_{N=0}^{\infty}\frac{z^{N}}{N!}
		\sum_{\substack{X\text{ partition}\\\text{of }[N]}}
		\prod_{S\in X}\tilde{K}(S)
		\Biggr]
		=\sum_{N=1}^{\infty}\frac{z^{N}}{N!}\tilde{K}\bigl([N]\bigr).
		\]
		The identity above is an instance of the \emph{exponential (or connected-components) formula} for exponential generating functions: whenever a class of labeled structures is built as a disjoint union of ``components,'' the exponential generating function (EGF) of all structures is the exponential of the EGF of the connected ones (see for instance \cite[Theorem~7.2.1]{Joyal1981}).
	\end{remark}

	\subsection{Graph and partition notions}\label{sub:graphs and partitions}
	
	Recall that we denote by $\mc{G}_c(V)$ the set of \textcolor{black}{nonempty} collections of edges $E$ on $V$ such that $(V,E)$ is connected. We also denote by $\mc{T}_c(V)$ the set of collections of edges $E$ on $V$ such that $(V,E)$ is a tree.

	\begin{definition}[More partition notions]\label{def:part notions}
		\begin{enumerate}
			\item For every set $A$, we denote by $\mathbf{\Pi}(A)$ the set of partitions of $A$
			and  by $\mathbf{\Pi}_\sub(A)$ the set of subpartitions of $A$. Recall that $B\subset \mc{P}(A)$ is a subpartition of $A$ if there exists $C\subset A$ such that $B$ is a partition of $C$. 
			\item (Set of points). Let $X$ be a subpartition of $[N]$. We set
			\begin{equation}\label{def:VX}
				V_X=\bigcup_{S\in X}S.\end{equation}
			\item Let $X$ be a subpartition of $[N]$ and $i\in V_X$. We denote by $[i]^X$ the part $S\in X$ such that $i\in S$.
			\item (Set of edges between dipoles in distinct multipoles and in the same multipole). Let $X$ be a subpartition of $[N]$. We let
			\begin{equation}\label{def:E inter}
				\mc{E}^\inter(X)\coloneqq \bigcup_{S\neq S'\in X}\bigcup_{i\in S, j\in S'}\{ij\}.
			\end{equation}
			Moreover, we let  
			\begin{equation}\label{def:E intra}
				\mc{E}^\intra(X)\coloneqq \bigcup_{S\in X} \bigcup_{i\in S, j\in S:i\neq j}\{ij\}.
			\end{equation}
			\item (Restriction of a subpartition). Let $X$ be a subpartition of $[N]$. For each  $E\subset \mc{E}^\inter(X)$, we let $\Res(X,E)$ be the set of blocks $S\in X$ such that there exists a vertex in $S$ incident to an edge in $E$.
			\item (Coarsening of a partition). Let $X$ be a subpartition of $[N]$. Let $X_1,\ldots,X_n$ be \textcolor{black}{nonempty} disjoint subsets of $X$. We define the coarsening of $X$ along $X_1,\ldots,X_n$ to be 
			\begin{equation}\label{def:merge}
				\Coarse_X(X_1,\ldots,X_n)\coloneqq \left(\bigsqcup_{i=1}^n \{V_{X_i}\}\right)\sqcup \bigsqcup_{S\in X, S \notin X_1\cup \cdots \cup X_n}\{S\}.
			\end{equation}
			Note that $\Coarse_X(X_1,\ldots,X_n)$ is a partition of $V_X$. Each group $X_i$ of blocks is replaced by a single block $\textcolor{black}{V_{X_i}}$ and every other block of $X$ is untouched. See Figure~\ref{fig:Y} below.
	\end{enumerate}\end{definition}

	\begin{figure}[H]
		\centering
		\includegraphics[width=0.8\textwidth]{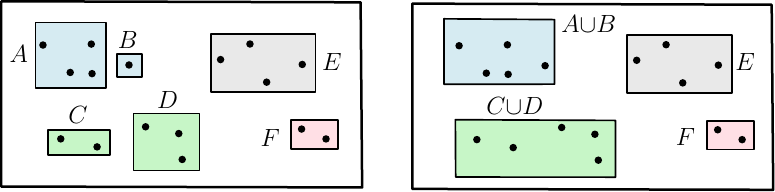} 
		\caption{Left: the partition $X$ with blocks $A,B,C,D,E$ and $F$. 
			Let $X_1=\{A,B\}$ and $X_2=\{C,D\}$. 
			Right: the coarsening $\Coarse_X(X_1,X_2)$ with blocks $A\cup B$, $C\cup D$, $E$, and $F$.}
		\label{fig:Y}
	\end{figure}

	\begin{definition}[More graph notions]\label{def:quotient}
		Let $X$ be a subpartition of $[N]$.
		\begin{enumerate}
			\item (Quotient graph). Let $E\subset \mc{E}^\inter(X)$. We let $G=(V_X,E)/X $ be the graph given by $V(G)=X$ and
			\begin{equation*}
				SS'\in E(G)\quad \text{if there exist $i\in S$ and $j\in S'$ such that $ij\in E$}.
			\end{equation*}
			\item (Multigraph quotient). The \emph{multigraph quotient}
			$G^{\mathrm{multi}}$ is the undirected multigraph with vertex set
			$V(G^{\mathrm{multi}})=X$ and edge multiplicities
			\[
			m_E(S,S') \coloneqq  |\{ij\in E:\ i\in S,\ j\in S'\}|, \qquad S\neq S'\in X.
			\]
			\item ($X$-connected-components). Let $E\subset \mc{E}^\inter(X)$. We call $X$-components the connected components of the augmented graph 
			\begin{equation*}
				(V_X,E\cup \mc{E}^\intra(X)).
			\end{equation*}
			\item (Connected graph relative to $X$). We define $\mathsf{E}^X$ to be the set of $E\subset \mc{E}^\inter(X)$ with at least one edge and such that the augmented graph
			\begin{equation*}
				\Bigl(V_X,E\cup \mc{E}^\intra(X) \Bigr)
			\end{equation*}
			is connected. If $E\in\mathsf{E}^X$, we say that $(V_X,E)$ is connected relative to $X$. 
			\item (Tree relative to $X$). We define $\mathsf{T}^X$ to be the set of $E\in \mathsf{E}^X$ such that $G\coloneqq (V_X,E)/X$ is a tree and such that, for every edge $SS'\in E(G)$, there exists a unique $i\in S$ and a unique $j\in S'$ such that $ij\in E$.
			\item \textcolor{black}{We set $\mathsf E_0^X=\mathsf E^X$ and $\mathsf T_0^X=\mathsf T^X$, except that $\mathsf E_0^{\{S\}}=\mathsf T_0^{\{S\}}=\{\emptyset\}$.}
			\item Let $E\subset \mc{E}^\inter(X)$ and let $S\in X$. We set 
			\begin{equation*}
				\deg_E(S)\coloneqq \sum_{i\in S}\deg_{(V_X,E)}(i).
			\end{equation*}
		\end{enumerate}
	\end{definition}

	\begin{figure}[H]
		\centering
		\fbox{\includegraphics[width=0.4\textwidth]{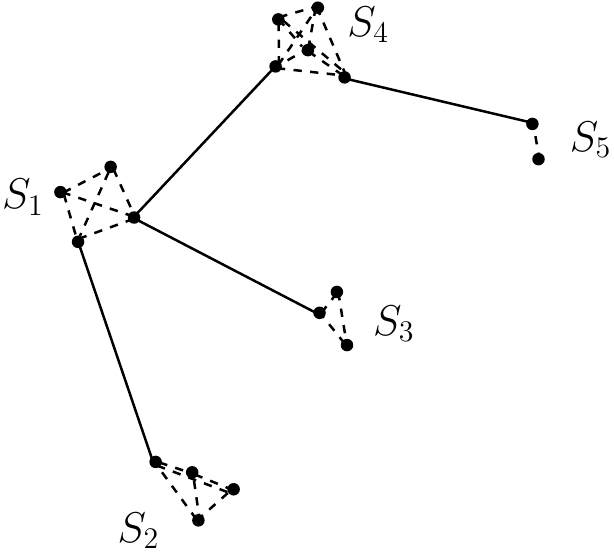}} 
		\caption{A tree relative to $X\coloneqq \{S_1,S_2,S_3,S_4,S_5\}$. Notice that $T\in \mathsf{T}^X$ implies that the quotient graph $(V_X,T)/X$ is a tree on $X$, but the converse is not true.}
	\end{figure}

	\subsection{Strategy for the lower bound}\label{sub:strat lower}
	Our aim is to provide a lower bound on the partition function by performing a cluster expansion. We start by bounding the partition function from below by that of a simplified model, the lower-bound model obtained from the energy upper bound of Lemma~\ref{lemma:bsupF}.
	
	We wish to rewrite the partition function of the lower-bound model in terms of a cluster-expansion series, as on the left-hand side of \eqref{eq:resum}. In order to obtain a formula for the logarithm as in \eqref{eq:resum}, we need the cluster-expansion series, i.e., the right-hand side of \eqref{eq:resum}, to be absolutely convergent. For that, we expand only the interactions that correspond to distinct multipoles which are therefore ``small enough''.

	Restricting interactions to only those within each multipole defines what we call a hierarchical model. We then compute the error between our lower-bound model and the hierarchical model by considering the ratio of their partition functions, and expanding the interactions between dipoles in distinct multipoles around $1$, 
	which amounts to performing a perturbative expansion around the hierarchical model.  This is done after restricting to the event on which the stable matching $\sigma_N$ is the identity,  the cardinality of each multipole is \textcolor{black}{at most $p(\beta)$}, the number of multipoles of cardinality $k \leq p(\beta)$ is fixed, and every dipole has length smaller than $\varepsilon_0 \Cut$, where $\ve_0\in (0,1)$ is a small parameter. In Lemma~\ref{lemma:start low}, we bound from below the integral of the Boltzmann--Gibbs weight on our event by the product of some combinatorial factors, the partition function of the hierarchical model, and a cluster-expansion series.

	The goal of the rest of the section is then to expand the logarithm of the cluster-expansion series appearing in Lemma~\ref{lemma:start low}. We define a \emph{cluster} as a subpartition $X$ of $[N]$, and the \emph{size} of a cluster as the total number of points in the union of the elements of $X$.  First, we aim to control the activity of clusters of bounded size, which is addressed in Subsection~\ref{sub:bounded}. Suppose that $\beta \in (\beta_p, \beta_{p+1}]$, where we recall that  $\beta_p = 4 - \frac{2}{p}$. \textcolor{black}{Our bounds on the activity of clusters of size $k$ then differ for $k\leq p$ and for $k>p$, with a logarithmic correction when $\beta=\beta_{p+1}$ and $k=p+1$.}   We also control in Subsection~\ref{sub:bounded} the error between the activities of the lower-bound model and the activities of the true model. The second main task of the section is to control the activity of large clusters in order to  prove that the cluster-expansion series in Lemma~\ref{lemma:start low} is absolutely convergent, allowing us to apply the formula of Lemma~\ref{lemma:resum2} to its logarithm.

	The second task is the more delicate. As we will see in Subsection~\ref{sub:pert}, the activity of a cluster $X$ is defined as a sum over certain sets of edges $E$ of an expectation under a certain measure of $\prod_{ij \in E} f_{ij}^{\tilde{v}}$, where the weights $f_{ij}^{\tilde{v}}$ are the Mayer bonds and $\tilde v=(\tilde v_{ij})_{i<j}$ is the interaction in the lower-bound model.
	The ingredients that will be used were alluded to in the introduction:
	\begin{itemize}
		\item {\it cancellation} by parity that allows us to reduce to Eulerian graphs,
		\item {\it Penrose resummation} to rewrite an activity defined as a sum over connected graphs in terms of a sum over trees,
		\item a {\it peeling procedure} on Eulerian graphs to construct a spanning tree on which we can control the contributions of all edges.
	\end{itemize}

	With these ingredients, we complete the proof of the absolute convergence of the cluster-expansion series in Subsection~\ref{sub:abs lower}. In Subsection~\ref{sub:proof low}, we give a lower bound on the partition function of the true model in terms of the infimum of a multipole free energy. This will complete the proof of the lower bound in \eqref{eq:expansion Z} and \eqref{eq:LD bounds} in Theorem~\ref{theorem:LDP}.

	\subsection{Definitions: hierarchical model and  activities}\label{sub:multi}
	
	For every $V\subset [N]$, we let $\Pi_\mult^V$ be the partition of $V$ into multipoles (computed among points in $V$ only). When $V=[N]$, we let $\Pi_\mult\coloneqq \Pi_\mult^V$.
	Let $X$ be the value of $\Pi_\mult$. Recall that we say that $i$ belongs to a $2k$-pole if $|[i]^X|=k$ with $k\geq 2$ or to a pure dipole if $|[i]^X|=1.$

	\begin{definition}[Dipole measure]\label{def:dipole measure}
		\textcolor{black}{Let $\ve_0\in(0,1)$ be the parameter of Definition~\ref{def:vijtilde}.}

		The dipole measure is  defined as  the probability measure over $(\R^2)^2$ with density
		\begin{equation}\label{def:muve}
			\dd \mu_{\beta,\lambda,\ve_0}=\frac{1}{C_{\beta,\lambda,\ve_0} N}e^{\beta \g_\lambda(x-y)}\indic_{x,y\in \Lambda}\indic_{|x-y|\leq \ve_0 \Cut}\dd x\dd y,
		\end{equation}
		where $C_{\beta,\lambda,\ve_0}$ is the normalization constant 
		\begin{equation}\label{def:Clambda}
			C_{\beta,\lambda,\ve_0}\coloneqq \frac{1}{N}\int_{\Lambda^2} e^{\beta \g_\lambda(x-y)}\indic_{|x-y|\leq \ve_0 \Cut}\dd x \dd y.
		\end{equation}
		\textcolor{black}{For $\ve_0=\infty$, we omit the length cutoff and set $\mu_{\beta,\lambda,\infty}=\mu_{\beta,\lambda}$ and $C_{\beta,\lambda,\infty}=C_{\beta,\lambda}$.}
	\end{definition}
	
	\medskip

	\begin{lemma}\label{lem411}
		Let $\beta\in (2,\infty)$ and $\ve_0\in (0,1)$. There exist $C>0$ and $O$ depending only on $\beta$ such that 
		\begin{equation}\label{devtCbeta}
			C_{\beta,\lambda,\ve_0}= \lambda^{2-\beta} (\mathcal Z_\beta + O((\ve_0 \Cut)^{-2})) \quad \text{as} \ N\to \infty,
		\end{equation}
		where $\mathcal Z_\beta$ is as in \eqref{def:Zbeta}. 
		Moreover, given $\ve_0>0$, there exists $\lambda_0$  depending on $\ve_0$ and $\beta$ such that, for $\lambda<\lambda_0$ \textcolor{black}{and $N$ large enough depending on $\lambda,\ve_0$}, the following holds: 
		\begin{equation}\label{eq:Clambda bound}
			C_{\beta,\lambda,\ve_0}\geq \frac{1}{2}\mathcal Z_\beta \lambda^{2-\beta}
		\end{equation}
		and 
		\begin{equation}\label{eq:comparisons Cve0}
			\frac{|C_{\beta,\lambda,\ve_0}- C_{\beta,\lambda}|}{C_{\beta,\lambda}}\leq C (\ve_0 \Cut)^{-2},
		\end{equation}
		where $C_{\beta, \lambda}$ is as in \eqref{defCbetalambda}.
		\textcolor{black}{For every $N\geq1$ and $\lambda>0$, we also have $0<C_{\beta,\lambda,\ve_0}\leq C_{\beta,\lambda}\leq\lambda^{2-\beta}\mc Z_\beta$.}
	\end{lemma}
	\begin{proof}\textcolor{black}{If $\lambda$ is small enough that $\ve_0\Cut\geq2\lambda$, then $\g_\lambda=\g$ outside $B(0,\ve_0\Cut)$ by Newton's theorem, and an integration of the tail in polar coordinates gives}
		\begin{equation}\label{Cbe1}
			|C_{\beta,\lambda, \ve_0}-C_{\beta, \lambda}|\le C (\ve_0\Cut)^{2-\beta},\end{equation}
		for $C>0$ depending only on $\beta$.
		Next, observe that
		\begin{equation*}
			C_{\beta,\lambda}=\int_{\dR^2}e^{\beta \g_\lambda(y)} V_N(y)\dd y,
		\end{equation*}
		where 
		\begin{equation*}
			V_N(x)=\frac{1}{N}|\{y\in \Lambda:x+y\in \Lambda\}|,
		\end{equation*}
		where $|\cdot|$ here denotes the area. Since $0\leq V_N\leq 1$, we thus get by dominated convergence that 
		\begin{equation}\label{limCbeta}
			\lim_{N\to \infty}C_{\beta,\lambda}=\int_{\dR^2}e^{\beta \g_\lambda(x)} \dd x= \lambda^{2-\beta}\mc{Z}_\beta.
		\end{equation}
		It follows in view of \eqref{Cbe1} that, as $N \to \infty$, we have
		\begin{equation*}
			C_{\beta,\lambda,\ve_0} =\lambda^{2-\beta}\left(\mathcal Z_\beta+ O\left( \left(\frac{\lambda}{\ve_0\Cut}\right)^{\beta-2}\right)\right),
		\end{equation*}
		where $O$ depends only on $\beta$.
		Recall that, for $R>\lambda$, 
		\begin{equation*}
			\Bigl(\frac{\lambda}{R}\Bigr)^{\beta-2}\leq R^{-2}\quad\Longleftrightarrow \quad R\leq R_{\beta,\lambda}.
		\end{equation*}
		(Note that, for $\beta\geq 4$, $R_{\beta,\lambda}=\infty$ so that the above is indeed true.) Thus, since $\ve_0 \Cut\leq R_{\beta,\lambda}$, we deduce \eqref{devtCbeta}, and \eqref{eq:Clambda bound}, provided $\lambda$ is small enough, is then a direct consequence. The relation \eqref{eq:comparisons Cve0} is then also a direct consequence of \eqref{devtCbeta}, \eqref{eq:Clambda bound} and \eqref{limCbeta}.
	\end{proof}

	We now introduce the multipole measure and partition function for our approximate model.

	\begin{definition}[Multipole measure] \label{def:multipolemeasure}
		Let $\ve_0\in (0,1)\cup \{\infty\}$. For every subpartition $X$ of $[N]$, define the probability measure on $(\R^2)^{2|V_X|}$ 
		\begin{equation}\label{def:P0 ve0}
			\dd \Psf_X^{0,\ve_0}= \frac{1}{\Msf^0_{\ve_0} (X)}  \prod_{S\in X}\left(\indic_{\mc{B}_S}\prod_{ i,j\in S:i<j}e^{-\beta v_{ij}}\indic_{\mc{A}_{ij}}\right)\prod_{\textcolor{black}{i\in V_X}} \dd \mu_{\beta,\lambda,\ve_0}(x_i, y_i),\end{equation}
		where we recall \textcolor{black}{\eqref{eq:defBC mc}} and \eqref{defAij}, and  where $\Msf^0_{\ve_0} (X)$ is the normalization constant
		\begin{equation}\label{def:M0 ve_0}
			\Msf^0_{\ve_0} (X)\coloneqq \prod_{S\in X}\dE_{(\mu_{\beta, \lambda, \ve_0})^{\otimes |S| }}\left[\indic_{\mc{B}_S}\prod_{ i,j\in S:i<j}e^{-\beta v_{ij}}\indic_{\mc{A}_{ij}}\right].
		\end{equation}

		\textcolor{black}{For $\ve_0\in(0,1)$,} let $L$ be the positive constant in \eqref{eq:b2}. For every subpartition $X$ of $[N]$,
		we let $\Psf_{X}^{-,\ve_0}$ be the probability measure
		\begin{equation}\label{def:PwL}
			\dd \Psf_{X}^{-,\ve_0}= \frac{1}{\Msf_{\ve_0}^-(X)}  \prod_{S\in X}\left(\indic_{\mc{B}_S}\prod_{ i,j\in S:i<j}e^{-\beta \tilde{v}_{ij}}\indic_{\mc{A}_{ij}}\right)\prod_{i\in V_X }e^{-L\frac{|x_i-y_i|^2}{(\ve_0\Cut)^2}}\prod_{i\in V_X}  \dd \mu_{\beta, \lambda, \ve_0}(x_i, y_i),\end{equation}
		where $\Msf_{\ve_0}^-(X)$ is the normalization constant
		\begin{equation}\label{def:M-}
			\Msf_{\ve_0}^-(X)\coloneqq \prod_{S\in X}\dE_{(\mu_{\beta, \lambda, \ve_0})^{\otimes |S| }}\left[\indic_{\mc{B}_S}\prod_{ i,j\in S:i<j}e^{-\beta \tilde{v}_{ij}}\indic_{\mc{A}_{ij}}\prod_{i\in S }e^{-L\frac{|x_i-y_i|^2}{(\ve_0\Cut)^2}} \right].
		\end{equation}
		With a slight abuse of notation, when $X=\{S\}$, we write $\Msf_{\ve_0}^{0}(S)=\Msf_{\ve_0}^{0}(\{S\})$ and $\Msf_{\ve_0}^{-}(S)=\Msf_{\ve_0}^{-}(\{S\})$.
	\end{definition}

	\begin{figure}[H]
		\centering
		\fbox{\includegraphics[width=0.3\textwidth]{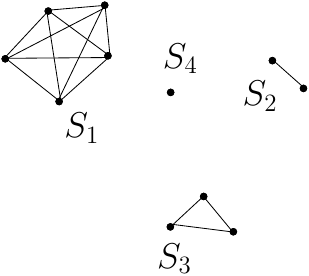}} 
		\caption{Hierarchical model. Four multipoles are featured: $X=\{S_1,S_2,S_3,S_4\}$. The hierarchical model is $P_X^- \coloneqq  P_{S_1}^- \otimes P_{S_2}^- \otimes P_{S_3}^- \otimes P_{S_4}^-$. Only the interactions within each multipole are taken into account.}
		\label{fig:hier}
	\end{figure}

	We introduce the corrections to the hierarchical multipole model approximation. First, recall the notion of Mayer bond.

	\begin{definition}[Mayer bond]\label{def:mayer} 
		Let $w=(w_{ij})$ be a collection of interactions. We let
		\begin{equation*}f_{ij}^{w}\coloneqq  e^{-\beta w_{ij}}\indic_{\mc{A}_{ij}}-1.
		\end{equation*}   
	\end{definition}

	\begin{definition}\label{def:activity lower}
		Let $\mathbf{\Pi}_\sub([N])$ be as in Definition~\ref{def:part notions}. For every $X\in \mathbf{\Pi}_\sub([N])$, recall $\mc{E}^\inter(X)$ and $\Coarse_X(X_1,\ldots,X_n)$ from Definition~\ref{def:part notions}, and $\mathsf{E}^X$ from Definition~\ref{def:quotient}. Let $\ve_0\in (0,1)\cup\{\infty\}$.
		
		Let $\Ksf_{\ve_0}^0:\mathbf{\Pi}_\sub([N])\to \dR$ be \textcolor{black}{zero for $|X|\leq1$ and} defined for every $X\in \mathbf{\Pi}_\sub([N])$ \textcolor{black}{with $|X|\geq2$} by 
		\begin{multline}
			\Ksf_{\ve_0}^0(X)\coloneqq \sum_{n=0}^\infty \frac{1}{n!}\sum_{\substack{X_1,\ldots,X_n\subset X\\ \mathrm{disjoint} }} \sum_{E_1\in \mathsf{E}^{X_1}}\cdots\sum_{E_n\in \mathsf{E}^{X_n}} \sum_{
				F\in \textcolor{black}{\mathsf{E}_0}^{\Coarse_X(X_1,\ldots,X_n)}}\\ \dE_{\Psf_{X}^{0,\ve_0}}\left[\prod_{ij\in E_1\cup \cdots \cup E_n}f^{v}_{ij}\prod_{ij\in \mc{E}^\inter(X_1)\cup \cdots \cup \mc{E}^\inter(X_n)}\indic_{\mc{B}_{ij}^c}\prod_{ij\in F}(-\indic_{\mc{B}_{ij}})\right].
		\end{multline}
		\textcolor{black}{For $\ve_0\in(0,1)$,} let $\Ksf_{\ve_0}^-:\mathbf{\Pi}_\sub([N])\to \dR$ be \textcolor{black}{zero for $|X|\leq1$ and} defined for every $X\in \mathbf{\Pi}_\sub([N])$ \textcolor{black}{with $|X|\geq2$} by 
		\begin{multline}\label{eq:activity lower}
			\Ksf_{\ve_0}^-(X)\coloneqq \sum_{n=0}^\infty \frac{1}{n!}\sum_{\substack{X_1,\ldots,X_n\subset X\\ \mathrm{disjoint} }} \sum_{E_1\in \mathsf{E}^{X_1}}\cdots\sum_{E_n\in \mathsf{E}^{X_n}} \sum_{
				F\in \textcolor{black}{\mathsf{E}_0}^{\Coarse_X(X_1,\ldots,X_n)}}\\ \dE_{\Psf_{X}^{-,\ve_0}}\left[\prod_{ij\in E_1\cup \cdots \cup E_n}f^{\tilde{v}}_{ij}\prod_{ij\in \mc{E}^\inter(X_1)\cup \cdots \cup \mc{E}^\inter(X_n)}\indic_{\mc{B}_{ij}^c}\prod_{ij\in F}(-\indic_{\mc{B}_{ij}})\right].
		\end{multline}
	\end{definition}
	
	\begin{remark}\label{remark:K0}
		\textcolor{black}{The term $F=\emptyset$ occurs when $n=1$ and $X_1=X$, since $\Coarse_X(X)$ then has a single block.}
	\end{remark}

	\begin{remark}
		By the M\"obius inversion formula (see Lemma~\ref{lemma:mobius}), one can obtain a more compact expression for $\Ksf_{\ve_0}^{-}(X)$. However, because the proof of the absolute convergence of the cluster-expansion series relies on the representation given in \eqref{eq:activity lower}, we keep that expression for now.
	\end{remark}

	\subsection{Perturbative expansion around a hierarchical model}\label{sub:pert}
	We now perform the series of reductions announced in Subsection~\ref{sub:strat lower}.

	For each $k\geq 1$, let us denote by $\mc{N}_k$ the number of $2k$-poles as defined in Definition~\ref{def:multipoles}, item (2). Let $n_1,\ldots,n_{p(\beta)}$ be such that $n_1+2n_2+\cdots+p(\beta)n_{p(\beta)}=N$ and such that 	
	\begin{equation}\label{eq:assnk}
		n_k\leq \ve_0^{-\alpha(\beta)}\lambda^{2(k-1)}N\qquad\text{for every }2\leq k\leq p(\beta),
	\end{equation}
	where
	\begin{equation}\label{def:alphabeta}
		\alpha(\beta)\coloneqq \begin{cases}
			\frac{4-\beta}{2} & \text{if $\beta\in (2,4)$}\\
			\frac{1}{2} & \text{if $\beta\geq 4$}.
		\end{cases}
	\end{equation}
	Set
	\begin{equation}\label{defeventA}
		\mc{A}\coloneqq \{\mc{N}_1=n_1,\ldots,\mc{N}_{p(\beta)}=n_{p(\beta)}\}\bigcap_{i=1}^N \{|x_i-y_{\sigma_N(i)}|\leq \ve_0\Cut\}.  
	\end{equation}
	
	Our aim is to give a lower bound on $\int_{\mc{A}}e^{-\beta \F_\lambda}$ based on the energy upper bound \eqref{eq:b2}, restricting to events where there are no $2k$-poles for $k$ beyond $p(\beta)$, the number of $\textcolor{black}{2k}$-poles for $k\le p(\beta)$ is well controlled, and the size of the dipoles does not exceed $\ve_0 \Cut$. We start with the following cluster-expansion form.
	
	\begin{lemma}\label{lemma:start low}
		Let $\beta\in (2,\infty)$ and $p(\beta)$ be as in  Definition~\ref{def:pbeta}. \textcolor{black}{Assume $\lambda\leq\ve_0\Cut$.} Let $n_1,\ldots,n_{p(\beta)}$ be such that $n_1+2n_2+\cdots+p(\beta)n_{p(\beta)}=N$ and such that \eqref{eq:assnk} holds, and let $\mc{A}$ be as in \eqref{defeventA}.
		Let $\pi$ be a partition of $[N]$ such that, for every \textcolor{black}{integer $k\geq1$},
		\begin{equation}\label{eqSk}
			|\{S \in \pi:|S|=k\}|=\begin{cases} n_k & \text{if} \ k \le p(\beta)\\
				0 & \text{if} \ k\ge p(\beta)+1.\end{cases}
		\end{equation}
		Recalling  $\Msf_{\ve_0}^-$ from \eqref{def:M-} and $\Ksf_{\ve_0}^-$ from \eqref{eq:activity lower}, 		we have
		\begin{multline*}
			\int_{\mc{A}}\indic_{\sigma_N=\Id}e^{-\beta \F_\lambda(\XN,\YN)}\dd \XN \dd \YN\\ \geq \frac{N!}{1^{n_1}(2!)^{n_2}\cdots (p(\beta)!)^{n_{p(\beta)}}n_1!\cdots n_{p(\beta)}!}(C_{\beta, \lambda,\ve_0}N)^N\Msf_{\ve_0}^-(\pi)\sum_{n=0}^{\infty} \frac{1}{n!}\sum_{\substack{X_1,\ldots,X_n\in \mc{P}(\pi) \\ \mathrm{disjoint}}}\Ksf_{\ve_0}^-(X_1)\cdots \Ksf_{\ve_0}^-(X_n),
		\end{multline*}
		with the convention that, for $n=0$, the sum over $X_1,\ldots,X_n\in \mc{P}(\pi)$ disjoint is 1.

	\end{lemma}

	\begin{remark}
		Once the partition $\pi$ is fixed, we will often abuse notation and call each part $S\in \pi$ a multipole.
	\end{remark}

	\begin{proof}
		For simplicity, let us write $p$ for $p(\beta)$.
		
		\paragraph{\bf{Step 1: fixing multipoles}} Let $\pi$ be a partition of $[N]$ satisfying \eqref{eqSk}.
		Recall that the number of ways to partition $[N]$ into $m$ sets of cardinality $k_1,\ldots,k_{m}$ is given by
		\begin{equation}\label{eq:n partition}
			\frac{N!}{k_1!\cdots k_m! \prod_{i\geq 1} m_i !},
		\end{equation}
		where, for every $i\geq 1$, $m_i$ stands for the number of sets of size $i$, and $m=\sum_{i\ge 1} m_i$. The number of ways to choose the partition $\pi$ is therefore given by 
		\begin{equation}\label{eq:choicespi}
			\frac{N!}{1^{n_1}(2!)^{n_2}\cdots (p!)^{n_p}n_1!\cdots n_p!}.
		\end{equation}
		Therefore, recalling from the beginning of Subsection~\ref{sub:multi} that 
		$\Pi_\mult$ is the partition of $[N]$ into multipoles (see Definition~\ref{def:multipoles} for multipoles), we get
		\begin{multline}\label{eq:e1}
			\int_{\mc{A}\cap\{\sigma_N=\Id\}}e^{-\beta \F_\lambda}\\=  \frac{N!}{1^{n_1}(2!)^{n_2}\cdots (p!)^{n_p}n_1!\cdots n_p!}\int_{\Pi_\mult=\pi,\sigma_N=\Id}e^{-\beta \F_\lambda(\vec{X}_N,\vec{Y}_N)}\prod_{i=1}^N \indic_{|x_i-y_i|\leq \ve_0 \Cut} \dd \vec{X}_N \dd \vec{Y}_N.
		\end{multline}
		
		\paragraph{\bf{Step 2: rewriting the event $\{\Pi_\mult=\pi,\sigma_N=\Id\}$ as an intersection of simpler events}}

		By Lemma~\ref{lemma:stable}, one can write
		\begin{equation}\label{intersaij}
			\{\sigma_N=\Id\}=\bigcap_{i<j}\mc{A}_{ij}, 
		\end{equation}
		where we recall from Definition~\ref{def:Aij} that 
		$\mc{A}_{ij}$ is the event where $x_i$ is assigned to $y_i$ and $x_j$ to $y_j$ when performing the stable matching
		between $x_i,y_i$ and $x_j,y_j$ only.
		
		Moreover, recalling $\mc{B}_{ij}^c$ and $\mc{B}_S$  from \eqref{eq:defBij mc} and \eqref{eq:defBC mc},
		\begin{equation*}
			\{\sigma_N=\Id\}\cap\{\Pi_\mult=\pi\}=\{\sigma_N=\Id\}\cap \left\{ \bigcap_{S\in \pi}\mc{B}_S \right\} \cap \left\{ \bigcap_{S,S'\in \pi : S\neq S'}\bigcap_{i\in S, j\in S'}\mc{B}_{ij}^c \right\} .
		\end{equation*}
		Indeed, every connected component of $\Pi_\mult$ is connected and there are no edges between distinct connected components. Therefore, in view of \eqref{intersaij},
		\begin{equation}\label{eq:event eq}
			\begin{split}
				\{\sigma_N=\Id\}\cap\{\Pi_\mult=\pi\}&=\left\{\bigcap_{S\in \pi}\Bigl(\mc{B}_S\bigcap_{i\neq j\in S}\mc{A}_{ij}\Bigr)\right\}\cap\left\{ \bigcap_{S\neq S'\in \pi}\bigcap_{i\in S, j\in S'}\mc{B}_{ij}^c\cap \mc{A}_{ij}\right\}.
			\end{split}
		\end{equation}
		
		\paragraph{\bf{Step 3: incorporating the energy upper bound}}
		Combining \eqref{eq:event eq} and \eqref{eq:b2}, and recalling the notation $\mc{E}^\inter(\pi)$ from Definition~\ref{def:part notions}, we get
		\begin{align}\label{eq:e2}
			&\int_{\Pi_\mult=\pi,\sigma_N=\Id}e^{-\beta \F_\lambda(\vec{X}_N,\vec{Y}_N)}\prod_{i=1}^N \indic_{|x_i-y_i|\leq \ve_0 \Cut}\dd \vec{X}_N \dd \vec{Y}_N\\ & \hspace{1cm} \notag \geq \int \Bigl(\prod_{ij\in \mc{E}^\inter(\pi)} e^{-\beta \tilde{v}_{ij}}\indic_{\mc{A}_{ij}} \Bigr)\prod_{ij\in \mc{E}^\inter(\pi)}\indic_{\mc{B}_{ij}^c} \prod_{S\in \pi}\Bigl( \indic_{\mathcal B_S}\prod_{ i,j\in S:i<j} e^{-\beta \tilde{v}_{ij}}\indic_{\mc{A}_{ij}}\Bigr)\\ & \hspace{1cm}\notag \times \prod_{i=1}^N e^{\beta \g_\lambda(x_i-y_i)-L(\frac{|x_i-y_i|}{\ve_0\Cut})^2}\indic_{|x_i-y_i|\leq \ve_0 \Cut}\dd x_i \dd y_i.   
		\end{align}

		\paragraph{\bf{Step 4: expansion of interactions between dipoles in distinct multipoles}}
		As explained in the introduction, our starting point is to expand the interaction between dipoles which do not belong to the same multipole. Expanding the product of the $1+f_{ij}^{\tilde{v}}$ and recalling \eqref{def:E inter} and Definition~\ref{def:mayer}, we have 
		\begin{equation*}
			\prod_{ij\in \mc{E}^\inter(\pi)}e^{-\beta \tilde{v}_{ij}}\indic_{\mc{A}_{ij}}=\prod_{ij\in \mc{E}^\inter(\pi)}(1+f_{ij}^{\tilde{v}})=\sum_{E\subset \mc{E}^\inter(\pi)}\prod_{ij\in E}f_{ij}^{\tilde{v}}.
		\end{equation*}
		Inserting this into \eqref{eq:e2} gives
		\begin{align}\label{defIw}
			&\int_{\Pi_\mult=\pi,\sigma_N=\Id}e^{-\beta \F_\lambda(\vec{X}_N,\vec{Y}_N)}\prod_{i=1}^N \indic_{|x_i-y_i|\leq \ve_0 \Cut}\dd \vec{X}_N \dd \vec{Y}_N\\ &\hspace{1cm} \notag \geq 
			\sum_{E\subset\mc{E}^\inter(\pi)}\int \prod_{ij\in \mc{E}^\inter(\pi)}\indic_{\mc{B}_{ij}^c} \prod_{ij\in E}f_{ij}^{\tilde{v}} \prod_{S\in \pi}\Bigl(\indic_{\mc{B}_S}\prod_{ i,j\in S:i<j}e^{-\beta \tilde{v}_{ij}}\indic_{\mc{A}_{ij}}\Bigr)\\ &\hspace{1cm}\notag\times \prod_{i=1}^N e^{\beta \g_\lambda(x_i-y_i)-L(\frac{|x_i-y_i|}{\ve_0\Cut})^2 }\indic_{|x_i-y_i|\leq \ve_0 \Cut}\dd x_i \dd y_i.
		\end{align}

		\paragraph{\bf{Step 5: summing according to connected components relative to $\pi$ of $([N],E)$}}   
		Next, we rewrite the above sum according to the connected components of the augmented graph 
		\begin{equation*}
			\Bigl([N],E\cup \mc{E}^\intra(\pi)\Bigr)
		\end{equation*}
		that contain at least two multipoles. 
		This yields 
		\begin{equation}\label{eq:res2}
			\sum_{E\subset \mc{E}^\inter(\pi)}\prod_{ij\in E}f_{ij}^{\tilde{v}} =\sum_{n=0}^\infty \frac{1}{n!}\sum_{\substack{X_1,\ldots,X_n\subset \pi\\ \mathrm{disjoint}}} \sum_{E_1\in \mathsf{E}^{X_1},\ldots,E_n\in \mathsf{E}^{X_n}}\prod_{l=1}^n \prod_{ij\in E_l}f_{ij}^{\tilde{v}}.
		\end{equation}
		Notice that, for $n=0$, the product over $l\in [n]$ of the sum over $X_l\subset \pi$ is empty, so the contribution equals $1$. Notice that if $|X_i|=1$, the sum over $E_i\in \mathsf{E}^{X_i}$ is zero. Hence, one can keep in mind that the sum in \eqref{eq:res2} can be restricted to $|X_i|\geq 2$.
		
		Thus,
		\begin{align}\label{eq:ws}
			&\int_{\Pi_\mult=\pi,\sigma_N=\Id}e^{-\beta \F_\lambda(\vec{X}_N,\vec{Y}_N)}\prod_{i=1}^N \indic_{|x_i-y_i|\leq \ve_0 \Cut}\dd \vec{X}_N \dd \vec{Y}_N  \geq \sum_{n=0}^\infty \frac{1}{n!}\sum_{\substack{X_1,\ldots,X_n\subset \pi\\ \mathrm{disjoint}}} \sum_{E_1\in \mathsf{E}^{X_1},\ldots,E_n\in \mathsf{E}^{X_n}}\\ & \hspace{1cm} \notag
			\int  \prod_{ij\in \mc{E}^\inter(\pi)}\indic_{\mc{B}_{ij}^c} \prod_{ij\in E_1\cup \cdots \cup E_n}f_{ij}^{\tilde{v}} \prod_{S\in \pi}\Bigl(\indic_{\mc{B}_S}\prod_{ i,j\in S:i<j}e^{-\beta \tilde{v}_{ij}}\indic_{\mc{A}_{ij}}\Bigr)\\ &\hspace{1cm}\notag \times \prod_{i=1}^N e^{\beta \g_\lambda(x_i-y_i)-L(\frac{|x_i-y_i|}{\ve_0\Cut})^2 }\indic_{|x_i-y_i|\leq \ve_0 \Cut}\dd x_i \dd y_i.
		\end{align}

		\paragraph{\bf{Step 6: expanding the weights $\indic_{\mc{B}_{ij}^c}$ over disjoint connected components}}
		
		Fix $X_1,\ldots,X_n\subset \pi$ that contain at least two multipoles and suppose that $X_1,\ldots,X_n$ are disjoint. In order to have a multiplicative activity, we need to expand the product of the $\indic_{\mc{B}_{ij}^c}$. To be able to control the products of the Mayer bonds $f_{ij}^{\tilde{v}}$ inside each connected component $X_l$, we keep the weights $\indic_{\mc{B}_{ij}^c}$ for every edge $ij\in \mc{E}^\inter(X_l)$ and expand the rest:
		\begin{equation*}
			\prod_{ij\in \mc{E}^\inter(\pi)}\indic_{\mc{B}_{ij}^c}=\left(\prod_{l=1}^n\prod_{ij\in \mc{E}^\inter(X_l)}\indic_{\mc{B}_{ij}^c}\right)\prod_{ij\in \mc{E}^\inter(\Coarse_\pi(X_1,\ldots,X_n)) }\indic_{\mc{B}_{ij}^c}.
		\end{equation*}
		We now expand the second product in the above display by writing $\indic_{\mc{B}_{ij}^c}=1-\indic_{\mc{B}_{ij}}$. This gives 
		\begin{equation*}
			\prod_{ij\in \mc{E}^\inter(\Coarse_\pi(X_1,\ldots,X_n))}\indic_{\mc{B}_{ij}^c}= \sum_{F\subset   \mc{E}^\inter(\Coarse_\pi(X_1,\ldots,X_n))}\prod_{ij\in F}(-\indic_{\mc{B}_{ij}}).
		\end{equation*}
\textcolor{black}{Let $\mathbf X=(X_l)_{l=1}^n$ and $\mathbf E=(E_l)_{l=1}^n$.}
		Inserting this into \eqref{eq:ws} gives
		\begin{align}\label{eq:wss}
			&	\int_{\Pi_\mult=\pi,\sigma_N=\Id}e^{-\beta \F_\lambda(\vec{X}_N,\vec{Y}_N)}\prod_{i=1}^N \indic_{|x_i-y_i|\leq \ve_0 \Cut}\dd \vec{X}_N \dd \vec{Y}_N   \\\notag & \qquad \geq \sum_{n=0}^\infty \frac{1}{n!}\sum_{\substack{X_1,\ldots,X_n\subset \pi\\   \mathrm{disjoint}}} \sum_{E_1\in \mathsf{E}^{X_1},\ldots,E_n\in \mathsf{E}^{X_n}}\sum_{F\subset \mc{E}^\inter(\Coarse_\pi(X_1,\ldots,X_n)) }\textcolor{black}{I_{\mathrm{cl},\ve_0}^{-}(\pi;\mathbf X,\mathbf E,F)},
		\end{align}
		where 
		\begin{multline*}
			\textcolor{black}{I_{\mathrm{cl},\ve_0}^{-}(\pi;\mathbf X,\mathbf E,F)}\\\coloneqq \int  \prod_{ij\in \cup_{l=1}^n\mc{E}^\inter(X_l)} \indic_{\mc{B}_{ij}^c}\prod_{ij\in F}(-\indic_{\mc{B}_{ij}}) \prod_{ij\in E_1\cup \cdots \cup E_n}f_{ij}^{\tilde{v}} \prod_{S\in \pi}\Bigl(\indic_{\mc{B}_S}\prod_{ i,j\in S:i<j}e^{-\beta \tilde{v}_{ij}}\indic_{\mc{A}_{ij}}\Bigr) \\  \times \prod_{i=1}^N e^{\beta \g_\lambda(x_i-y_i)-L(\frac{|x_i-y_i|}{\ve_0\Cut})^2 }\indic_{|x_i-y_i|\leq \ve_0 \Cut}\dd x_i \dd y_i.
		\end{multline*}

		\paragraph{\bf{Step 7: resumming over the connected components relative to $\pi$ of $([N],E\cup F)$}}
		Resumming \eqref{eq:wss} according to the connected components of the graph 
		\begin{equation*}
			\Bigl([N],E_1\cup \cdots \cup E_n\cup F\cup \mc{E}^\intra(\pi) \Bigr)
		\end{equation*}
		that contain at least two multipoles, we get
		\begin{align}\label{eq:proFF}
			&\int_{\Pi_\mult=\pi,\sigma_N=\Id}e^{-\beta \F_\lambda(\vec{X}_N,\vec{Y}_N)}\prod_{i=1}^N \indic_{|x_i-y_i|\leq \ve_0 \Cut}\dd \vec{X}_N \dd \vec{Y}_N\\ \notag
			&\hspace{1cm}\geq \sum_{m=0}^\infty \frac{1}{m!}\sum_{\substack{X_1',\ldots,X_m'\subset \pi\\ \mathrm{disjoint} }}\int F(X_1')\cdots F(X_m')\prod_{S\in \pi}\Bigl(\indic_{\mc{B}_S}\prod_{ i,j\in S:i<j}e^{-\beta \tilde{v}_{ij}}\indic_{\mc{A}_{ij}}\Bigr)\\\notag &\hspace{1cm}\times \prod_{i=1}^N e^{\beta \g_\lambda(x_i-y_i)-L(\frac{|x_i-y_i|}{\ve_0\Cut})^2 }\indic_{|x_i-y_i|\leq \ve_0 \Cut}\dd x_i \dd y_i,
		\end{align}
		where \textcolor{black}{$F(X)=0$ for $|X|\leq1$, and, for $|X|\geq2$,}
		\begin{multline*}
			F(X)\coloneqq \sum_{n=0}^\infty \frac{1}{n!}\sum_{\substack{X_1,\ldots,X_n\subset X\\ \mathrm{disjoint}}}\sum_{E_1\in \mathsf{E}^{X_1}}\cdots \sum_{E_n\in \mathsf{E}^{X_n}}\sum_{F\in \textcolor{black}{\mathsf{E}_0}^{\Coarse_X(X_1,\ldots,X_n) } }\\ \prod_{ij\in \cup_{l=1}^n E_l}f_{ij}^{\tilde{v}}\prod_{ij\in \cup_{l=1}^n\mc{E}^\inter(X_l)}\indic_{\mc{B}_{ij}^c}\prod_{ij\in F}(-\indic_{\mc{B}_{ij}}).
		\end{multline*}
		
		\paragraph{\bf{Step 8: conclusion}}
		Inserting the last display into \textcolor{black}{\eqref{eq:proFF}}, dividing by the normalization constant $\textcolor{black}{(NC_{\beta,\lambda,\ve_0})^N}\Msf_{\ve_0}^-(\pi)$ introduced in Definition~\ref{def:multipolemeasure}, and using independence over disjoint components to obtain the multiplicativity of $\Msf_{\ve_0}^-$ over disjoint connected components, we obtain from \eqref{eq:proFF}
		\begin{multline*}
			\frac{1}{(NC_{\beta,\lambda,\ve_0})^N\Msf_{\ve_0}^-(\pi)} \int_{\Pi_\mult=\pi,\sigma_N=\Id}e^{-\beta \F_\lambda(\vec{X}_N,\vec{Y}_N)}\prod_{i=1}^N \indic_{|x_i-y_i|\leq \ve_0 \Cut}\dd \vec{X}_N \dd \vec{Y}_N\\ \geq \sum_{n=0}^\infty \frac{1}{n!}\sum_{\substack{X_{1},\ldots,X_{n}\subset \pi \\ \mathrm{disjoint}}}\Ksf_{\ve_0}^-(X_1)\cdots \Ksf_{\ve_0}^-(X_n),
		\end{multline*}
		where $\Ksf_{\ve_0}^-$ is as in Definition~\ref{def:activity lower}.  In view of \eqref{eq:e1}, we have proved the result.
		
	\end{proof}

	\subsection{Control of activities for bounded-size clusters and limiting activities, statements}\label{sub:bounded}
	
	Let us introduce a truncated version of the dipole activity of Definition~\ref{def:dipole activity} that will be used later in the proof.
	
	\begin{definition}[Truncated dipole activity]\label{def:dipole activity trunc}
		For all $\ve_0\in (0,1)\cup\{\infty\}$ and $V\subset [N]$, recalling \eqref{def:muve}, we let
		\begin{equation}\label{def:Kdip ve0}
			\Ksf_{\beta,\lambda,\ve_0}^\dip(V)\coloneqq \sum_{E\in \mc{G}_c(V)}\dE_{\mu_{\beta, \lambda, \ve_0}^{\otimes |V|}}\left[\prod_{ij\in E}f_{ij}^v\right].
		\end{equation}
	\end{definition}

	\begin{remark}
		Notice that, as in \textcolor{black}{Definition~\ref{def:activity lower}}, for $|V|\in \{0,1\}$, one has $\Ksf_{\beta,\lambda,\ve_0}^\dip(V)=0$.
	\end{remark}

	\begin{definition}\label{def:gamma beta}
		Let $\beta\in (2,+\infty)$. Recalling  \eqref{defpstar},
		\textcolor{black}{for every $k\in\mathbb N^*$,} we define 
		\begin{equation*}
			\gamma_{\beta,\lambda,k} =
			\begin{cases}
				\lambda^{2(k-1)} 
				& \text{if } k \leq p^*(\beta), \\[4pt]
				\Cut^{-2} 
				& \begin{aligned}
					&\text{if } k > p^*(\beta) \ \text{and} \ 
					\beta \in \big(\beta_{p^*(\beta)},\beta_{p^*(\beta)+1}\big) \\
					&\text{or } \beta = \beta_{p^*(\beta)+1} \ \text{and} \ k > p^*(\beta) + 1,
				\end{aligned} \\[8pt]
				\Cut^{-2} \lvert\log \lambda\rvert
				& \text{if } k = p^*(\beta) + 1 \ \text{and} \ 
				\beta = \beta_{p^*(\beta)+1},
			\end{cases}
		\end{equation*}
		with the convention that $\beta_\infty=4$.
	\end{definition}

	\textcolor{black}{The following proposition gives bounds on the activities that are consistent with the predicted multipole transitions.}

	\begin{prop}[Control of the activity of bounded-size clusters]\label{prop:bounded lower}
		Let $\beta\in (2,+\infty)$. \textcolor{black}{Let $\ve_0\in(0,1)$.} Let $\gamma_{\beta,\lambda,n}$ be as in Definition~\ref{def:gamma beta}. Let $M$ be the constant used in the definition of multipoles (Definition~\ref{def:multipoles}). Recall the activities $\Msf_{\ve_0}^-$, $\Msf_{\ve_0}^0$, $\Ksf_{\ve_0}^-$ and $\Ksf_{\ve_0}^0$ from Definitions~\ref{def:multipolemeasure} and~\ref{def:activity lower}, and $\Ksf_{\beta,\lambda}^\dip$ from Definition~\ref{def:dipole activity}.

		\textcolor{black}{For every integer $k$ with $2\leq k\leq p^*(\beta)$, there exists $M_0>\max(20,3^{p(\beta)})$, depending only on $\beta$ and $k$, such that, for every $M\geq M_0$, the following holds.} There exists $C>0$ depending only on $\beta$, $M$ and $k$ such that, for $\lambda$ sufficiently small and $N$ sufficiently large depending on $\lambda$, every $S\subset[N]$ with $|S|=k$ and every $\Msf\in\{\Msf_{\ve_0}^-,\Msf_{\ve_0}^0\textcolor{black}{,\Msf_\infty^0}\}$ satisfy
		\begin{equation}\label{eq:MX -low}
			\Msf(S)\geq \frac{\lambda^{2(|S|-1)}}{CN^{|S|-1}},
		\end{equation}
		\begin{equation}\label{eq:MX -up}
			\Msf(S)\leq \frac{C\lambda^{2(|S|-1)}}{N^{|S|-1}}.
		\end{equation}

		\textcolor{black}{For every integer $k\geq2$, there exists $M_0>\max(20,3^{p(\beta)})$, depending only on $\beta$ and $k$, such that, for every $M\geq M_0$, the following holds.} There exists $C>0$ depending only on $\beta$, $M$, $\ve_0$ and $k$ such that, for $\lambda$ sufficiently small and $N$ sufficiently large depending on $\lambda$, every subpartition $X$ of $[N]$ with $|V_X|=k$ and $|X|>1$ and every $\Ksf\in\{\Ksf_{\ve_0}^-,\Ksf_{\ve_0}^0\}$ satisfy
		\begin{equation}\label{eq:K- bound}
			|\Ksf(X)|\leq C\frac{\gamma_{\beta,\lambda,|V_X|}}{\Msf^0_{\ve_0}(X)N^{|V_X|-1}}.
		\end{equation}
		\textcolor{black}{For $k\leq p^*(\beta)$, \eqref{eq:K- bound} also holds for $\Ksf_\infty^0$, with denominator $\Msf_\infty^0(X)N^{|V_X|-1}$ and $C$ depending only on $\beta$, $M$ and $k$. In both clauses, the thresholds for $\lambda$ and $N$ are uniform over sets and subpartitions of the prescribed cardinalities.}

		If $V\subset[N]$ is such that $1<|V|\leq p^*(\beta)$, then there exists $C>0$ depending only on $\beta$\edgreen{ and $|V|$} such that, for $\lambda$ sufficiently small and $N$ sufficiently large depending on $\lambda$,
		\begin{equation}\label{eq:Kdip bound}
			|\Ksf_{\beta,\lambda}^\dip(V)|\leq C\frac{\lambda^{2(|V|-1)}}{N^{|V|-1}}.
		\end{equation}
		\edgreen{The activity $\Ksf_{\beta,\lambda}^\dip$ does not involve $M$, so \eqref{eq:Kdip bound} holds without the restriction on $M$.}
	\end{prop}
	\medskip

	\begin{prop}[Control of activity errors]\label{prop:expansion -}
		Let $\beta\in (2,+\infty)$. Let  $p^*(\beta)$ be as in \eqref{defpstar}. \textcolor{black}{Let $k$ be an integer with $2\leq k\leq p^*(\beta)$. There exists $M_0>\max(20,3^{p(\beta)})$, depending only on $\beta$ and $k$, such that the following holds for every $M\geq M_0$ in Definition~\ref{def:multipoles}, every $S\subset[N]$ with $|S|=k$ and every subpartition $X$ of $[N]$ with $|V_X|=k$ and $|X|>1$.}
		
		There exists $C>0$ depending on $\beta$, $M$, $\ve_0$ and $|S|$ such that 
		\begin{equation}\label{eq:M-diff}
			|\Msf_{\ve_0}^-(S)-\Msf^0_{\infty} (S)|\leq \frac{C}{N^{|S|-1}}\Cut^{-2}.
		\end{equation} 
		Moreover, there exists $C>0$ depending on $\beta$, $M$, $\ve_0$ and $|S|$ such that 
		\begin{equation}\label{eq:Kdip diff}
			|\Ksf_{\beta,\lambda,\ve_0}^\dip(S)-\Ksf_{\beta,\lambda}^\dip(S)|\leq  \frac{C}{N^{|S|-1}}\Cut^{-2}.
		\end{equation}
		
		Besides, there exists $C>0$ depending on $\beta$, $M$, $\ve_0$ and $|V_X|$ such that 
		\begin{equation}\label{eq:K-diff}
			|\Ksf_{\ve_0}^-(X)-\Ksf_{\infty}^0(X)|\leq \frac{C}{\Msf_{\infty}^0(X)N^{|V_X|-1}}\Cut^{-2}
		\end{equation}
		and 
		\begin{equation}\label{eq:K-diff0}
			|\Ksf_{\ve_0}^0(X)-\Ksf_{\infty}^0(X)|\leq \frac{C}{\Msf_{\infty}^0(X)N^{|V_X|-1}}\Cut^{-2}.
		\end{equation}

	\end{prop}

	Propositions~\ref{prop:bounded lower} and~\ref{prop:expansion -} are proved in Subsections~\ref{sub:bounded lower proof} and~\ref{sub:exp -}.

	\begin{remark}[On activity sizes and assumption \eqref{eq:assnk}]
		\textcolor{black}{Suppose $\beta\in(2,4)$.} Let $n_1,\ldots,n_{p(\beta)}$ be such that $n_1+2n_2+\cdots+p(\beta)n_{p(\beta)}> p(\beta)+1$ and set $k_0\coloneqq n_1+\cdots +n_{p(\beta)}$. Suppose that $X$ has $n_i$ multipoles of cardinality $i$ for every $i\in [p(\beta)]$. Then, as indicated by \eqref{eq:K- bound} and \eqref{eq:MX -low}, 
		\begin{equation*}
			|\Ksf_{\ve_0}^-(X)| \textcolor{black}{\leq C} \frac{\Cut^{-2} }{N^{k_0-1}} \prod_{k=1}^{p(\beta)} \frac{1}{\lambda^{2(k-1)n_k}}.\end{equation*}    
		If we do not impose \eqref{eq:assnk}, then the number of choices of $X$ with $\#_i X=n_i$ for all $i$ is at most $\binom{|\pi|}{k_0}\leq N^{k_0}/k_0!$, hence
		\begin{equation*}
			\sum_{\substack{X\in \mc{P}(\pi):\\ \forall i, \#_i X=n_i}}|\Ksf_{\ve_0}^-(X)|\lesssim  N \frac{\Cut^{-2}}{\prod_{k=1}^{p(\beta)}\lambda^{2(k-1)n_k}}.
		\end{equation*}
		\textcolor{black}{This upper bound is too large for our purpose,} since in Theorems~\ref{theorem:expansion} and~\ref{theorem:LDP} the contribution of clusters strictly larger than $p(\beta)$ must be $O(N\delta_{\beta,\lambda})$.
	\end{remark}

	\textcolor{black}{In the definitions of $G_S$, $G^\prime_{S,E}$ and $G^{\prime\prime}_{X,E}$ below and in their uses, we first relabel the elements of $S$ or $V_X$ in increasing order as $1,\ldots,k$, together with $X$ and $E$, so that $x_1$ is the anchor. This does not change the integrals, since $\mc A_{ij}=\mc A_{ji}$ outside the null set of Subsection~\ref{sub:stable}.}
	We now study limiting activities. Let us define, as in \eqref{def:little m}, the candidate for the limit of $\Msf_{\beta,\lambda}(S)$ for $|S|\leq p^*(\beta)$. Recall from \eqref{defM} that 
	\begin{equation*}
		\Msf_{\beta,\lambda}(S)=\dE_{\mu_{\beta,\lambda}^{\otimes |S|}}\left[\indic_{\mc{B}_S}\prod_{i,j\in S:i<j}e^{-\beta v_{ij}}\indic_{\mc{A}_{ij}}\right].
	\end{equation*}
	\textcolor{black}{For $S\neq\emptyset$,} there exists a function $G_{S}:(\dR^{2})^{2|S|-1}\to \dR $ such that 
	\begin{equation}\label{eq:introduceGS}
		\indic_{\mc{B}_S}\prod_{i,j\in S:i<j}e^{-\beta v_{ij}}\indic_{\mc{A}_{ij}}\prod_{i\in S}e^{\beta \g_\lambda(x_i-y_i)}=  G_{S}(x_2-x_1,\ldots,x_{|S|}-x_1,y_1-x_1,\ldots,y_{|S|}-x_1).
	\end{equation}
	We thus set 
	\begin{equation}\label{eq:msf exp}
		\msf_{\beta,\lambda}(|S|)\coloneqq \frac{\displaystyle{\int_{(\dR^{2})^{(2|S|-1)}}G_{S}  }}{\displaystyle{\left(\int_{\dR^2 } e^{\beta \g_\lambda(x)}\dd x \right)^{|S|}}}.
	\end{equation}
	
\textcolor{black}{For a partition $Y$, let $\operatorname{Sym}_Y$ be the average of the $2^{|Y|}$ maps that, for each subfamily of parts of $Y$, exchange $x_i$ and $y_i$ for every $i$ in these parts.}

	Let \textcolor{black}{$S\neq\emptyset$ and} $E\in \mc{G}_c(S)$. Observe that there exists a map $G'_{S,E}:(\dR^2)^{2|S|-1 }\to \dR$ such that 
	\begin{multline}\label{eq:GS'}
		\textcolor{black}{\operatorname{Sym}_{\{\{i\}:i\in S\}}\left[\prod_{ij\in E}f_{ij}^v\right]} \prod_{i\in S}e^{\beta \g_\lambda(x_i-y_i)}\\=G'_{S,E}(x_2-x_1,\ldots,x_{|S|}-x_1,y_1-x_1,\ldots,y_{|S|}-x_1).
	\end{multline}
	Therefore, we set
	\begin{equation}\label{eq:ksf exp again}
		\ksf_{\beta,\lambda}^\dip(|S|)\coloneqq  \frac{\displaystyle{\sum_{E\in \mc{G}_c(S)}\int_{(\dR^{2})^{(2|S|-1)}}G'_{S,E}  }}{\displaystyle{\left(\int_{\dR^2} e^{\beta \g_\lambda(x)}\dd x \right)^{|S|}}}.
	\end{equation}
	
	Let $X$ be a \textcolor{black}{nonempty} subpartition of $[N]$ \textcolor{black}{with $n_i=\#_iX$ for every $i\geq1$.} Let $E\in \mathsf{E}^X$. Again, there exists a map $G''_{X,E}:(\dR^2)^{2|V_X|-1 }\to \dR$ such that 
	\begin{multline}\label{eq:GX''}
		\textcolor{black}{\operatorname{Sym}_X\left[\prod_{ij\in E}(e^{-\beta v_{ij}}\indic_{\mc B_{ij}^c}-1)\right]}\prod_{S\in X}\indic_{\mc{B}_S} \times\prod_{S\in X}\prod_{i,j\in S:i<j}e^{-\beta v_{ij}}\indic_{\mc{A}_{ij}}  \prod_{i\in V_X}e^{\beta \g_\lambda(x_i-y_i)}\\
		=G''_{X,E}(x_2-x_1,\ldots,x_{|V_X|}-x_1,y_1-x_1,\ldots,y_{|V_X|}-x_1).
	\end{multline}
	Hence, whenever $|V_X|\leq p^*(\beta)$, we set 
	\begin{equation}\label{eq:ksf mult again}
		\ksf_{\beta,\lambda}^\mult\textcolor{black}{((n_i)_{i\geq1})}= \frac{\displaystyle{\sum_{E\in \mathsf{E}^X}\int_{(\dR^{2})^{(2|V_X|-1)}}G''_{X,E}  }}{\displaystyle{\prod_{S\in X}\int_{(\dR^{2})^{(2|S|-1)}}G_{S} } }.
	\end{equation}

	In the following lemma, we assert that the quantities $\ksf_{\beta,\lambda}^\dip$, $\msf_{\beta,\lambda}$ and $\ksf_{\beta,\lambda}^{\mult}$ are well defined and are indeed the limits of the normalized activities.

	\begin{lemma}[Limiting activities]\label{lemma:limiting}
		Let $\beta\in (2,+\infty)$ and let $p^*(\beta)$ be as in \eqref{defpstar}.

		Let $S\subset [N]$ with $\textcolor{black}{1\leq}|S|\leq p^*(\beta)$. The quantity \eqref{eq:msf exp} is well defined. Moreover,
		\begin{equation}\label{eq:proof little kdip}
			\lim_{N\to\infty}N^{|S|-1}\Msf_{\beta,\lambda} (S)=\msf_{\beta,\lambda}(|S|).
		\end{equation}
		The quantity \eqref{eq:ksf exp again} is well defined. Moreover,
		\begin{equation}\label{eq:proof little mdip}
			\lim_{N\to\infty}N^{|S|-1}\Ksf^\dip_{\beta,\lambda} (S)=\ksf^\dip_{\beta,\lambda}(|S|).
		\end{equation}
		
		Let $X$ be a \textcolor{black}{nonempty} subpartition of $[N]$ such that $|V_X|\leq p^*(\beta)$. \textcolor{black}{For every $i\geq1$, let $n_i=\#_iX$. Then the quantity $\ksf_{\beta,\lambda}^{\mult}((n_i)_{i\geq1})$} \eqref{eq:ksf mult again} is well defined. Moreover, 
		\begin{equation}\label{def:Kbeta un}
			\lim_{N\to\infty} N^{|X|-1}\Ksf_{\beta,\lambda}^{\mult}(X)=\ksf_{\beta,\lambda}^{\mult}\textcolor{black}{((n_i)_{i\geq1})}.
		\end{equation}
	\end{lemma}
	
	The proof of Lemma~\ref{lemma:limiting} is postponed to Subsection~\ref{sub:limiting}.

	\subsection{Control of activities for unbounded clusters, statement}
	
	We next state the main result of this section, which is the control of activities for possibly unbounded clusters.
	\begin{prop}[Absolute convergence of the cluster-expansion series]\label{prop:absolute lower}
		Let $\beta\in (2,\infty)$, and let $p(\beta)$ be as in Definition~\ref{def:pbeta}. Let $\ve_0\in (0,1)$.
		Let $\pi$ be a partition of $[N]$ such that, for every $S\in \pi$, one has $|S|\leq p(\beta)$. Assume that \edgreen{the numbers $n_k\coloneqq \#_k \pi$ satisfy \eqref{eq:assnk}, that is,} for every $k\in \{2,\ldots,p(\beta)\}$,
		\begin{equation*}
			n_k\leq \ve_0^{-\alpha(\beta)}N\lambda^{2(k-1)},
		\end{equation*}
		where $\alpha(\beta)$ is as in \eqref{def:alphabeta}.
		
		For $M$ large enough with respect to $\beta$ and $p(\beta)$, $\ve_0$ small enough with respect to $\beta$, $p(\beta)$ and $M$, and $\lambda$ small enough, there exists $C>0$ depending only on $\beta, M, p(\beta)$ and $\ve_0$ such that
		\begin{equation}\label{eq:bb1}
			\sum_{X\in \mc{P}(\pi)}|\Ksf_{\ve_0}^-(X)|\indic_{|V_X|>p(\beta)}\leq CN\delta_{\beta,\lambda},
		\end{equation}
		with $\delta_{\beta,\lambda}$ as in \eqref{defdelta}. \textcolor{black}{For $K=\Ksf_{\ve_0}^-$, the whole series on the right-hand side of \eqref{eq:resum} converges absolutely.}
		
		Moreover, there exists $C>0$ depending only on $\beta, M, p(\beta)$ and $\ve_0$ such that
		\begin{equation}\label{eq:bb2}
			\sum_{n=1}^\infty\frac{1}{n!}\sum_{\substack{X_1,\ldots,X_n\in \mc{P}(\pi)\\ \mathrm{connected},\, \exists i, |V_{X_i}|>p(\beta) } }|\Ksf_{\ve_0}^-(X_1)\cdots \Ksf_{\ve_0}^-(X_n) \mathrm{I}(G(X_1,\ldots,X_n))| \leq CN\delta_{\beta,\lambda}.	\end{equation}
	\end{prop}

	The first step consists in splitting the Mayer bond $f_{ij}^{\tilde{v}}$ into its odd and even parts \textcolor{black}{on $\mc B_{ij}^c$}, and using cancellations of the odd parts to reduce to Eulerian graphs, and then using the Penrose resummation argument.

	\subsection{Parity and cancellations}\label{sub:parity}
	
	This subsection contains the cancellation arguments that form the starting point of the proof of the absolute convergence of the cluster-expansion series appearing in Lemma~\ref{lemma:start low}. Recall that, by Lemma~\ref{lemma:errortilde}, the weight $\tilde{v}_{ij}$ typically decays as the inverse of the squared distance between dipoles $i$ and $j$, provided this distance is larger than $r_i$ and $r_j$. The key issue is that $\dist^{-2}$ is not integrable at infinity in dimension $2$. Therefore, the absolute convergence of the cluster-expansion series, even with truncated weights at distance $\ve_0 \Cut$, is nontrivial. 
	
	A crucial ingredient in the argument is the observation that if one of the dipoles is flipped, then the interaction is changed  into its opposite. In other words, the dipole--dipole interaction is an odd function of both dipole vectors. 
	
	Recall from Definition~\ref{def:mayer} that the Mayer bond $f_{ij}^{\tilde{v}}$ is given by
	\begin{equation*}
		f_{ij}^{\tilde{v}} = e^{-\beta \tilde{v}_{ij}}\indic_{\mc{A}_{ij}} - 1.
	\end{equation*}
	One can decompose $f_{ij}^{\tilde{v}}$ as $a_{ij}^{\tilde{v}} + b_{ij}^{\tilde{v}}$, where $a_{ij}^{\tilde{v}}$ and $b_{ij}^{\tilde{v}}$ are given as follows:

	\begin{definition}[\textcolor{black}{Decomposition of the Mayer bond}]\label{def:awbw}
		Let $w=(w_{ij})_{i<j}$ be a collection of weights. We set 
		\begin{equation*}
			a^w_{ij}\coloneqq -\sum_{k\,  \mathrm{odd}}\frac{\beta^k}{k!}w_{ij}^k\indic_{\mc{A}_{ij}}\quad \mathrm{and}\, \quad b_{ij}^w\coloneqq \left(\sum_{k \, \mathrm{even},\, k\neq 0}\frac{\beta^k}{k!}w_{ij}^k \indic_{\mc{A}_{ij}}\right)- \indic_{\mc{A}_{ij}^c}.
		\end{equation*}
	\end{definition}
	
	For later use, we introduce the following weights:
	
	\begin{definition}\label{def:abs ab}
		We let 
		\begin{equation}\label{defaabs}
			a^{\abs}_{ij}\coloneqq \frac{r_ir_j}{d_{ij}\max(r_i,r_j,d_{ij})}\indic_{\mc{B}_{ij}^c}\indic_{d_{ij}\leq \textcolor{black}{17} \ve_0 \Cut}
		\end{equation}
		and
		\begin{equation}\label{defbabs}
			b^{\abs}_{ij}\coloneqq (a_{ij}^{\abs})^2+\indic_{\mc{B}_{ij}}.
		\end{equation}
	\end{definition}
	
	\begin{remark}\label{remarkabbi}
		\textcolor{black}{For $\lambda\leq\ve_0\Cut$ and $\max(r_i,r_j)\leq\ve_0\Cut$, Lemma~\ref{lemma:errortilde} with $\tau_0=8\ve_0\Cut$ gives}
		\begin{equation}\label{trueboundaij}
			|a_{ij}^{\tilde v} \indic_{\mc{B}_{ij}^c}| \le C a_{ij}^\abs,
		\end{equation}
		with $C$ depending on $\beta$.
		Indeed, on $\mc{B}_{ij}^c$, we have $d_{ij}\ge M \min(r_i,r_j)$ with $M>1$, so the assumptions of the lemma are verified. \textcolor{black}{Moreover, $a_{ij}^{\abs}\leq1$, hence $|\beta\tilde v_{ij}|\leq C_\beta$ there, and on this bounded range we have $|\sinh t|\leq C|t|$ and $\cosh t-1\leq Ct^2$.}
		
		Since, by Remark~\ref{remark:the inclusion}, $\mc{B}_{ij}^c \subset \mc{A}_{ij}$, we have $\indic_{\mc{A}_{ij}^c}\le \indic_{\mc{B}_{ij}}$. Hence, in view of Lemma~\ref{lemma:errortilde} again, if $ \max(r_i,r_j)\le \ve_0\Cut $, we have
		\begin{equation}
			|b_{ij}^{\tilde v}\indic_{\mc{B}_{ij}^c}|\le C b_{ij}^\abs,\end{equation}
		with $C$ depending on $\beta$.
	\end{remark}

	\begin{remark}[Parity]\label{remark:parity}
		
		One can observe that, for the weight $w= \tilde v$ we are interested in, $\textcolor{black}{a_{ij}^w\indic_{\mc B_{ij}^c}}$ is odd in $\vr_i$ and $\vr_j$ while $\textcolor{black}{b_{ij}^w\indic_{\mc B_{ij}^c}}$ is even in $\vr_i$ and $\vr_j$.
		Indeed, changing $\vr_i$ into $-\vr_i$ while keeping the midpoint $z_i$ fixed (see Lemma~\ref{lemma:bsupF}) corresponds to swapping the positions of $x_i$ and $y_i$. Thus, the quantity $\textcolor{black}{\nabla u_i}$ defined in \eqref{eqsurui} is changed into $\textcolor{black}{-\nabla u_i}$ and, by definition \eqref{deftvij}, $\tilde v_{ij}$ is changed into its opposite. Thus $\tilde v_{ij}$ is odd in $\vr_i$ and in $\vr_j$, keeping $z_i$ and $z_j$ fixed. \textcolor{black}{The event $\mc B_{ij}^c$ is invariant under these reversals and implies $\mc A_{ij}$.} 
	\end{remark}
	
	\begin{lemma}\label{lemma:cancellation odd}
		Let $X$ be a subpartition of $[N]$. Let $\mathbb{S}^1$ be the circle in $\dR^2$.  
		
		Let $E_1,E_2$ be disjoint sets of edges included in $\mc{E}^\inter(X)$. For each $ij\in E_1\cup E_2$, let $a_{ij}:(\mathbb{S}^1)^2\to \dR$ be a \textcolor{black}{measurable} function with the property that $a_{ij}$ is an odd function in both variables, i.e.,
		\begin{equation*}
			a_{ij}(-x,y)=-a_{ij}(x,y)=a_{ij}(x,-y)\quad \text{for every $x, y\in \mathbb{S}^1$}.
		\end{equation*}
		For each $ij\in E_1\cup E_2$, let $b_{ij}:(\mathbb{S}^1)^2\to \dR$ be a \textcolor{black}{measurable} function with the property that $b_{ij}$ is an even function in both variables, i.e.,
		\begin{equation*}
			b_{ij}(-x,y)=b_{ij}(x,y)=b_{ij}(x,-y)\quad \text{for every $x, y\in \mathbb{S}^1$}.
		\end{equation*}
		For each $S\in X$, let $F_S:(\mathbb{S}^1)^{|S|}\to \dR$ be a \textcolor{black}{measurable} function with the property that 
		\begin{equation*}
			F_S(-x_1,\ldots,-x_{|S|})= F_S(x_1,\ldots,x_{|S|}) \quad \text{for every $x_1,\ldots, x_{|S|}\in \mathbb{S}^1$}.
		\end{equation*}
		\textcolor{black}{Assume that the product below is integrable.} Suppose that there exists $S_0\in X$ such that  $\deg_{E_1}(S_0)$ is odd. Then 
		\begin{equation*}
			\int_{(\mathbb{S}^1)^{|V_X|}} \prod_{ij\in E_1}a_{ij}(\vr_i,\vr_j)\prod_{ij\in E_2}b_{ij}(\vr_i,\vr_j)\prod_{S\in X}F_{S}((\vr_i)_{i\in S})\prod_{i\in V_X}\dd \vr_i=0.
		\end{equation*}
	\end{lemma}
	
	\medskip

	\begin{proof}
		Performing the change of variables $\vr_i\mapsto \vr_i$ for every $i\in \cup_{S\in X,S\neq S_0} S$ and $\vr_{i}\mapsto -\vr_{i}$ for every $i\in S_0$, 
		we find
		\begin{multline*}
			\int_{(\mathbb{S}^1)^{|V_X|}} \prod_{ij\in E_1}a_{ij}(\vr_i,\vr_j)\prod_{ij\in E_2}b_{ij}(\vr_i,\vr_j)\prod_{S\in X}F_{S}((\vr_i)_{i\in S})\prod_{i\in V_X}\dd \vr_i\\=-\int_{(\mathbb{S}^1)^{|V_X|}} \prod_{ij\in E_1}a_{ij}(\vr_i,\vr_j)\prod_{ij\in E_2}b_{ij}(\vr_i,\vr_j)\prod_{S\in X}F_{S}((\vr_i)_{i\in S})\prod_{i\in V_X}\dd \vr_i,
		\end{multline*}
		hence 
		\begin{equation*}
			\int_{(\mathbb{S}^1)^{|V_X|}} \prod_{ij\in E_1}a_{ij}(\vr_i,\vr_j)\prod_{ij\in E_2}b_{ij}(\vr_i,\vr_j)\prod_{S\in X}F_{S}((\vr_i)_{i\in S})\prod_{i\in V_X}\dd \vr_i=0.
		\end{equation*}
	\end{proof}
	
	Below, we introduce the notion of Eulerian graphs relative to a partition.
	
	\begin{definition}[Eulerian graphs relative to $X$]\label{def:Eulerian}
		Let $X$ be a subpartition of $[N]$ and $E\subset \mc{E}^\inter(X)$. We say that $E\in \Eul^X$ if 
		\begin{equation*}
			\sum_{i \in S} \deg_E(i) \equiv 0 \pmod{2}
			\quad \text{for every } S \in X,
		\end{equation*}
		where, for every $i\in V_X$, $\deg_E(i)$ denotes the degree of $i$ in the graph $(V_X,E)$.

		Moreover, we say that $E\in \Eulc^X$ if, in addition, $E$ is connected relative to $X$, i.e., we set 
		\begin{equation*}
			\Eulc^X=\mathsf{E}^X\cap \Eul^X.
		\end{equation*}
	\end{definition}

	\begin{figure}[H]
		\centering
		\begin{subfigure}{0.3\textwidth}
			\centering
			\fbox{\includegraphics[width=\linewidth]{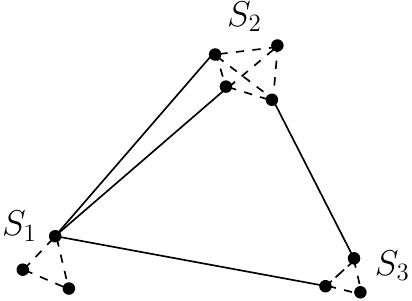}}
			\caption{}
		\end{subfigure}
		\hspace{1cm}
		\begin{subfigure}{0.3\textwidth}
			\centering
			\fbox{\includegraphics[width=\linewidth]{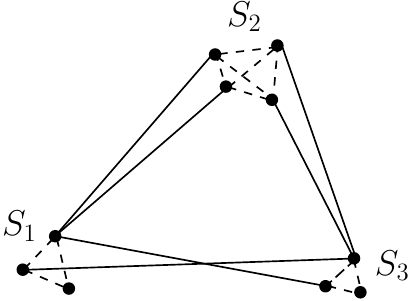}}
			\caption{}
		\end{subfigure}
		\caption{Unlike the graph in (b), the graph in (a) is not Eulerian relative to the partition $X\coloneqq \{S_1,S_2,S_3\}$, although its quotient graph is Eulerian.}
		\label{fig:two_images}
	\end{figure}
	
	In other words, $E$ is Eulerian relative to $X$ if for every $S\in X$, $S$ is incident to an even number of edges in $E$. 
	
	\begin{coro}\label{coro:reduction Euler}
		Let $X$ be a subpartition of $[N]$ \textcolor{black}{with $|X|\geq2$}. We have 
		\begin{multline}\label{eq:expE11}
			\Ksf_{\ve_0}^-(X)=\sum_{n=0}^\infty \frac{1}{n!}\sum_{\substack{X_1,\ldots,X_n\subset X\\ \mathrm{disjoint} }}\prod_{l=1}^n
			\Biggl(\sum_{ E_{l,1}\in \Eul^{X_l} }  \sum_{\substack{E_{l,2}:E_{l,1}\cup E_{l,2}\in \mathsf{E}^{X_l}\\
					E_{l,1}\cap E_{l,2}=\emptyset  }} \Biggr)\sum_{F\in \textcolor{black}{\mathsf{E}_0}^{\Coarse_X(X_1,\ldots,X_n)}}\\ \dE_{\Psf_{X}^{-,\ve_0}}\left[\prod_{l=1}^n\left(\prod_{ij\in E_{l,1}}a^{\tilde{v}}_{ij}\prod_{ij\in E_{l,2} }b^{\tilde{v}}_{ij}\right)\prod_{ij\in \cup_l\mc{E}^\inter(X_l)}\indic_{\mc{B}_{ij}^c}\prod_{ij\in F}(-\indic_{\mc{B}_{ij}})\right].
		\end{multline}
	\end{coro}

	\begin{proof}
		Starting from \eqref{eq:activity lower}, we obtain the result from Lemma~\ref{lemma:cancellation odd} and the fact that both $\mc{B}_{ij}$ and $\mc{B}_{ij}^c$ are even in $\vr_i$ and $\vr_j$ keeping the midpoints $z_i$ and $z_j$ fixed, and the fact that $\indic_{z_i + \textcolor{black}{\vr_i/2} \in \Lambda} \indic_{z_i-\textcolor{black}{\vr_i/2}\in \Lambda}$ is even. \textcolor{black}{Moreover, the factor of the density of $\Psf_X^{-,\ve_0}$ associated with a block $S\in X$ is invariant under the reversal of all the dipoles of $S$, since $\tilde v_{ij}$ is odd in each dipole and, outside the null set of Subsection~\ref{sub:stable}, $\mc A_{ij}$ is invariant under the reversal of both dipoles.}
	\end{proof}

	\subsection{Reduction to 2-edge-connected graphs}
	
	We recall here some standard graph notions that we next adapt to our partition setting.

	\begin{definition}[2-edge-connected graphs]
		Let $G=(V,E)$ be a finite graph. 
		\begin{enumerate}
			\item One says that $G$ is 2-edge-connected if \textcolor{black}{$G$ is connected and} for every $e\in E$, the graph $(V,E\setminus \{e\})$ is connected.
			\item One says that $G$ is minimally 2-edge-connected if \textcolor{black}{$G$ is 2-edge-connected and} for every $e\in E$, the graph $(V,E\setminus \{e\})$ is not 2-edge-connected.
			\item One says that an edge $e$ is a bridge in $G$ if \textcolor{black}{its deletion increases the number of connected components}. 
		\end{enumerate}
	\end{definition}
	
	\begin{definition}[2-edge-connected graph relative to a partition]\label{def:2egde relative}
		Let $X$ be a subpartition of $[N]$, let $E\in \mathsf{E}^X$ and set $G\coloneqq (V_X,E)$. 
		\begin{enumerate}
			\item One says that $G$ is 2-edge-connected relative to $X$ if for every $e\in E$, the graph $(V_X,E\setminus \{e\})$ is connected relative to $X$ (see Definition~\ref{def:quotient}). This means that $(V_X,E\cup \mc{E}^\intra(X))$ has no bridge in $E$.
			\item One  says that $G$ is minimally 2-edge-connected relative to $X$ if \textcolor{black}{$G$ is 2-edge-connected relative to $X$ and} for every $e\in E$, the graph $(V_X,E\setminus \{e\})$ is not 2-edge-connected relative to $X$.
		\end{enumerate}
	\end{definition}

	\begin{lemma}\label{lemma:euler implies}
		Let $X$ be a subpartition of $[N]$ and $E\in \Eulc^X$. Then, $(V_X,E)$ is 2-edge-connected relative to $X$.
	\end{lemma}
	
	\medskip

	\begin{proof}
		Suppose, for a contradiction, that $(V_X,E\cup \mc{E}^\intra(X))$ admits a bridge $ab\in E$. Let $X'$ be the connected component of $[a]^X$ in $(V_X,E\setminus\{ab\})/X$. Let $E'$ be the set of edges in $E\setminus \{ab\}$ incident to some vertex in $\textcolor{black}{V_{X'}}$.
		
		Then, since $E\in \Eulc^X$, the degree of $[a]^X$ in $E$ is even. Hence, the degree of $[a]^X$ in $E'$ is odd. By the handshaking lemma,
		\begin{equation*}
			\sum_{S\in X'}\deg_{E'}(S)=2|E'|.
		\end{equation*}
		Therefore, there is an even number of $S\in X'$ such that $\deg_{E'}(S)$ is odd. Hence, there exists $S\in X'$ with $S\neq [a]^X$ such that $\deg_{E'}(S)$ is odd. Since $ab$ is a bridge, $b$ is not in $\textcolor{black}{V_{X'}}$, hence  $\deg_{E'}(S)=\deg_E(S)$, which is odd, a contradiction to the fact that $E\in \Eulc^X$.\end{proof}

	\begin{definition}[Peeling into a minimal 2-edge-connected graph]\label{def:peeling into minimal}
		Let $X$ be a subpartition of $[N]$ and $E\in \mathsf{E}^X$ be such that $(V_X,E)$ is 2-edge-connected relative to $X$. \textcolor{black}{Let $E'\subset E$ be an inclusion-minimal subset of edges such that the multigraph quotient of $(V_X,E')$ by $X$ is connected and has no bridge, chosen first in the lexicographical order of increasing edge lists when there are several.}
		
		Define $\Peeled_X(E) \coloneqq E'$. 
	\end{definition}

	\begin{definition}[Strict ear decomposition]\label{def:strict ear relative}
		Let $X$ be a subpartition of $[N]$ and let $E\in \mathsf{E}^X$. \textcolor{black}{Each ear $P_i$ has vertex set $V_{\Res(X,E(P_i))}$, and the quotients by $X$ below have as vertices only the parts incident to their edges.} A sequence of \textcolor{black}{pairwise edge-disjoint} graphs $(P_1,\ldots,P_K)$ is called a strict ear decomposition of $(V_X,E)$ if $P_1\cup \cdots \cup P_K=(V_X,E)$ and:
		\begin{enumerate}
			\item The multigraph quotient (see Definition~\ref{def:quotient}) of $P_1$ by $X$ is a simple cycle. 
			\item \textcolor{black}{For every $i=2,\ldots,K$, the multigraph quotient of $P_i$ by $X$ has at least two edges and is either a simple path with endpoints in $(P_1\cup\cdots \cup P_{i-1})/X$ and all other vertices outside it, or a simple cycle meeting $(P_1\cup\cdots \cup P_{i-1})/X$ at exactly one vertex.}
		\end{enumerate}
		\textcolor{black}{Here a simple cycle may consist of two parallel edges.}
		
	\end{definition}

	\textcolor{black}{See Figure~\ref{fig:ear} for an example of a strict ear decomposition.}

	\begin{figure}[H]
		\centering
		\begin{subfigure}{0.3\textwidth}
			\centering
			\fbox{\includegraphics[width=\linewidth]{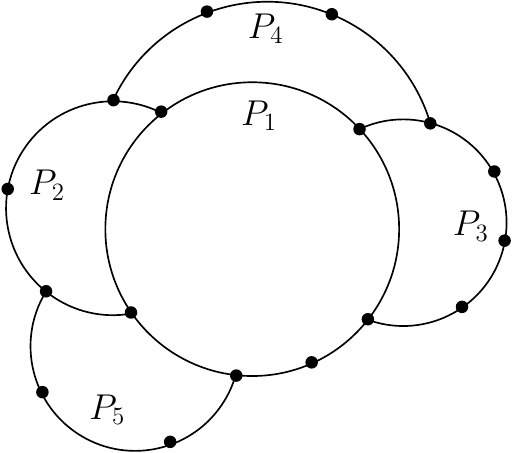}}
			\caption{}
		\end{subfigure}
		\hspace{2cm}
		\begin{subfigure}{0.3\textwidth}
			\centering
			\fbox{\includegraphics[width=\linewidth]{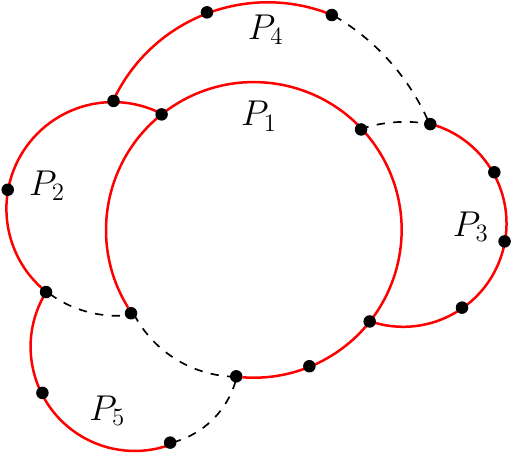}}
			\caption{}
		\end{subfigure}
		\caption{In (a), a strict ear decomposition of a minimal $2$-edge-connected graph. In (b), a spanning tree obtained by opening each ear and the base cycle.}
		\label{fig:ear}
	\end{figure}

	\begin{lemma}\label{lemma:strict ear}
		Let $X$ be a subpartition of $[N]$, $E\in \mathsf{E}^X$, and $(V_X,E)$ be a minimally 2-edge-connected graph relative to $X$. Then $(V_X,E)$ admits a strict ear decomposition relative to $X$. \textcolor{black}{Moreover, the first cycle $P_1$ may be prescribed.}
	\end{lemma}

	\begingroup\color{black}
	\begin{proof}
		Since the multigraph quotient has no bridge, the endpoints of each of its edges remain connected after the deletion of this edge, so that every edge lies on a simple cycle. In particular, there is a simple cycle through every prescribed part. We start from the prescribed cycle. As long as some vertices are not covered, each component of the uncovered vertices is joined to the current union of ears by at least two edges, since the quotient has no bridge. A simple path in this component that joins two such edges gives an ear with new internal vertices, which is closed if the two edges meet the union at the same vertex. The procedure ends with a spanning connected subgraph without bridges, which is the whole graph by minimality.
	\end{proof}
	\endgroup

	\subsection{Peeling lemma}\label{sub:peeling}

	Our next goal is to handle the contribution of the $a_{ij}^{\tilde{v}}$ terms appearing in \eqref{eq:expE11}. Notice that the cancellation in the interaction has already allowed us to reduce to Eulerian graphs.

	\begin{remark}[Why a naive strategy fails]\label{remark:crude}
\begingroup\color{black}
Suppose that $\beta\in(3,4)$ and that $X$ consists of singletons. Fix $E_1\in\Eulc^X$, set $k=|V_X|\geq3$, and write $h=\ve_0R_{\beta,\lambda}\geq\lambda$. We work on $\bigcap_{ij\in E_1}\mc B_{ij}^c$, where each factor of the products below is at most one. On the event that $e_0$ maximizes $d_e$ over $E_1$, let $T$ be a spanning tree of $(V_X,E_1)$ that avoids $e_0$. Then
\begin{equation}\label{eq:naive T}
\prod_{ij\in E_1}\frac{r_ir_j}{d_{ij}\max(d_{ij},r_i,r_j)}
\leq\prod_{ij\in T}\frac{r_ir_j}{d_{ij}\max(d_{ij},r_i,r_j)}
\min\Bigl(\frac{\max_{i\in V_X}r_i}{\max_{e\in T}d_e},1\Bigr)^2.
\end{equation}
The naive strategy drops the last factor, and keeps only the product over the tree $T$. For a fixed spanning tree $T\subset E_1$, let $\ell=|\Leaves(T)|$. Since $\tilde v_{ij}=0$ when $d_{ij}>17h$, by Lemma~\ref{lemma:errortilde}, integrating along $T$ with respect to $(NC_{\beta,\lambda,\ve_0}\dd\mu_{\beta,\lambda,\ve_0})^{\otimes k}$ gives
\begin{multline*}
J(E_1)\coloneqq\int_{(\Lambda^2)^k}
\prod_{ij\in E_1}\bigl(|a_{ij}^{\tilde v}|\indic_{\mc B_{ij}^c}\bigr)
\prod_{i\in V_X}e^{\beta\g_\lambda(x_i-y_i)}
\indic_{r_i\leq h}\dd x_i\dd y_i\\
\leq C_{\beta,k}N\int_{[0,h]^k}
\prod_{ij\in T}r_ir_j\left(1+\log\frac{17h}{\max(r_i,r_j,\lambda)}\right)
\prod_{i\in V_X}e^{\beta\g_\lambda(r_i)}r_i\dd r_i.
\end{multline*}
Since $k\geq3$, every edge of $T$ has an internal endpoint, and we assign the logarithm of each edge to such an endpoint. At an internal vertex of degree $d\geq2$ that receives $q$ logarithms, we have
\[
\int_0^h r^{d+1}(r\vee\lambda)^{-\beta}
\left(1+\log\frac{17h}{r\vee\lambda}\right)^q\dd r
\leq C_{\beta,d,q}h^{d+2-\beta}.
\]
Since $\beta>3$, each leaf contributes at most $C_\beta\lambda^{3-\beta}$. Hence
\[
J(E_1)\leq C_{\beta,k}N\lambda^{(3-\beta)\ell}
h^{(4-\beta)k-2+(\beta-3)\ell}.
\]
By \eqref{eq:Clambda bound} and $R_{\beta,\lambda}^{4-\beta}=\lambda^{2-\beta}$,
\[
\frac{J(E_1)}{(NC_{\beta,\lambda,\ve_0})^k}
\leq C_{\beta,k}N^{1-k}R_{\beta,\lambda}^{-2}
\left(\frac{R_{\beta,\lambda}}{\lambda}\right)^{(\beta-3)\ell}.
\]
This bound exceeds the order $N^{1-k}R_{\beta,\lambda}^{-2}$ required in Proposition~\ref{prop:bounded lower} for $k>p(\beta)$, away from the critical temperatures, by the factor $(R_{\beta,\lambda}/\lambda)^{(\beta-3)\ell}$. This loss occurs for clusters of every size, since if $E_1$ is a cycle, then $T$ is a path and $\ell=2$.
\endgroup
\end{remark}

	In fact, we need to incorporate some geometric constraints on the loops and some properties of Eulerian graphs in order to replace the crude peeling procedure and the control \eqref{eq:naive T}. The simplest nontrivial case of an Eulerian graph is that of a triangle with edges $12, 23, 13$, between multipoles reduced to singletons (i.e., pure dipoles).
	Let us describe our procedure in that case to give an idea (as already done in Subsection~\ref{sub:more} of the introduction).

	For the sake of exposition, we assume that $d_{ij} \ge \max(r_i, r_j)$ for every edge. By Lemma~\ref{lemma:errortilde},
	\begin{equation}\label{454} |a_{12}^{\tilde{v}}||a_{23}^{\tilde{v}}||a_{31}^{\tilde{v}}|\le C\frac{r_1r_2}{d_{12}^2} \frac{r_2r_3}{d_{23}^2}\frac{r_3r_1}{d_{13}^2}.\end{equation}
	We assume that $r_1$ is the largest of $r_1, r_2, r_3$, and, up to changing the labeling, that $d_{12}\ge d_{31}$.
	We may then define 
	$$g_{12}\coloneqq  \frac{r_1^2}{d_{12}^2}\le 1.$$
	Inserting this into \eqref{454} gives
	$$\frac{r_1r_2}{d_{12}^2} \frac{r_2r_3}{d_{23}^2}\frac{r_3r_1}{d_{13}^2}= \frac{r_1^2 r_2^2r_3^2}{ d_{12}^2 d_{23}^2 d_{13}^2}= g_{12} \frac{r_2^2 r_3^2}{d_{13}^2 d_{23}^2}.$$
	Thus, instead of \eqref{454}, we have bounded the product of the $|a_{ij}^{\tilde{v}}|$ by a product where the largest $r_i^2$ has been removed, the edge $12$ has been removed from the product in the denominator, and replaced by a multiplicative ``error'' $g_{12}$, which is smaller than $1$. This means that, in effect, we have ``opened'' the triangle and replaced it by the tree $13, 23$, while removing from \eqref{454} the largest $r_i^2$.

	We next explain how to generalize this to arbitrary Eulerian graphs. The key point is that if a graph $(V_X,E)$ is \textcolor{black}{connected and} Eulerian relative to $X$, then \textcolor{black}{$(V_X,E)$ is 2-edge-connected relative to $X$ by Lemma~\ref{lemma:euler implies}}. We will consider a minimally 2-edge-connected subgraph of $(V_X,E)$ relative to $X$, which, as such, admits a strict ear decomposition (see Definition~\ref{def:strict ear relative} and Lemma~\ref{lemma:strict ear}).  To construct the desired spanning tree of $(V_X,E)$, we will then open each ear by removing a carefully chosen edge similarly
	to what we did for the triangle case.

	We first isolate the largest dipole size in each multipole.
	\begin{definition}\label{def:rS}For every \textcolor{black}{nonempty} $S\subset [N]$, we let  $r_S\coloneqq \max_{i\in S} r_i$.
	\end{definition}

	\begin{definition}[Peeling of a minimally \(2\)-edge-connected graph]\label{def:peeling mini}
		Let \(X\) be a subpartition of \([N]\) and let
		\((x_i,y_i)_{i\in V_X}\in(\Lambda^2)^{|V_X|}\).
		For an edge \(ij\), recall that
		$d_{ij}=\mathrm{dist}\bigl(\{x_i,y_i\},\{x_j,y_j\}\bigr)$.
		Let $E\in \mathsf{E}^X$ be a minimally 2-edge-connected graph relative to $X$. \textcolor{black}{When a choice below has several candidates, we take the first one in a fixed order.}
		
		Let $v_1\in V_X$ be such that $r_{[v_1]^X}$ is maximal among the $r_S$ with $S\in X$. Let $(P_1,\ldots,P_K)$ be a strict ear decomposition of $E$ relative to $X$ such that $[v_1]^X$ is incident to some edge in $P_1$.
		Let 
		\[
		e_0\in E
		\quad\text{with}\quad
		d_{e_0}\coloneqq \max_{ij\in E} d_{ij},
		\]
		and let \(P_{\ell_0}\) be the (unique) ear containing \(e_0\).
		
		We now define a deletion rule for each ear $P_\ell$ together with an orientation on the deleted edge.

		\begin{enumerate}
			\item
			If \(\ell=\ell_0\) and \(d_{e_0}\ge\max_{i\in V_X}r_i\),  
			\textcolor{black}{we orient \(e_0\) from an endpoint whose block first appears in \(P_\ell\), and}
			set \(\vec{\mathcal S}(P_\ell)\coloneqq \{\vec e_0\}\).
			
			\item
			\textcolor{black}{Suppose $\ell\ge2$ and either}
			\(\ell\neq \ell_0\) or \(\bigl(\ell=\ell_0\text{ and }d_{e_0}<
			\max r_i\bigr)\).
			\textcolor{black}{We write \(P_\ell\) in order as}
			\[
			v_1v_2',\,v_2v_3',\dots,\,v_{n-1}v_n' \quad \text{with} \ \textcolor{black}{[v_i']^X=[v_i]^X\ (2\le i\le n-1)}
			\qquad (n\ge 3).
			\]
			\textcolor{black}{For a closed ear, the class $[v_1]^X=[v_n']^X$ is called its attachment vertex.}
			\begingroup\color{black}
			Let \(e_\ell\) maximize \(\min(d_{ij},\max(r_i,r_j))\) over \(ij\in\{v_1v_2',v_{n-1}v_n'\}\). Define
			\begin{equation}\label{eq:SPl}
				\vec{\mathcal S}(P_\ell)\coloneqq
				\begin{cases}
					\{v_2'\!\to\! v_1\} &\text{if }e_\ell=v_1v_2',\\[2pt]
					\{v_{n-1}\!\to\! v_n'\} &\text{otherwise}.
				\end{cases}
			\end{equation}
			\endgroup
			
			\item
			For the base cycle \(P_1\) (if either \(\ell_0\ne 1\) or
			\(d_{e_0}<\max r_i\)):
			write it as
			\((v_1v_2',\,v_2v_3',\dots,\,v_nv_1')\)
			with \textcolor{black}{$n\ge2$ and} \([v_1]=[v_1']\) the part of maximal radius in the sense of Definition~\ref{def:rS}.
			Use exactly the same deletion and orientation rule as in (2),
			replacing \(v_{n-1}v_n'\) by \(v_nv_1'\).
		\end{enumerate}
		
		Define the oriented discarded-edge set
		\[
		\vec{\mathcal F}^X\bigl((x_i,y_i)_{i\in V_X},E\bigr)
		\coloneqq 
		\bigcup_{\ell=1}^K \vec{\mathcal S}(P_\ell).
		\]
		We also let \(\mathcal F^X((x_i,y_i),E)\) be its undirected version, and set
		\[
		\mathcal T^X((x_i,y_i)_{i\in V_X},E)
		\coloneqq 
		E\setminus
		\mathcal F^X((x_i,y_i)_{i\in V_X},E).
		\]
		The edges in \(\mathcal F^X\) are ``peeled off'', while the set $\mc{T}^X$ is a tree relative to $X$, i.e., \(\mathcal T^X \in \mathsf{T}^X\).
	\end{definition}
	\textcolor{black}{Indeed, removing one edge from the base cycle leaves a tree, and removing one edge from each later ear attaches its new vertices without creating a cycle.}

	\begin{definition}[Peeling of an Eulerian graph]\label{def:peeling lower bound}
		Let $X$ be a subpartition of $[N]$ and let $E\in \Eulc^X$. Recall $\Peeled_X(E)$ from Definition~\ref{def:peeling into minimal}. \textcolor{black}{It is minimally $2$-edge-connected relative to $X$ by construction.} We extend the maps $\mc{T}^X$, $\vec{\mc{F}}^X$ and $\mc{F}^X$ from Definition~\ref{def:peeling mini} by setting 
		\begin{gather*}
		\mc{T}^X(\cdot,E) \coloneqq  \mc{T}^X(\cdot,\Peeled_X(E)),\qquad \vec{\mc{F}}^X(\cdot,E) \coloneqq  \vec{\mc{F}}^X(\cdot,\Peeled_X(E)),\\
		\mc{F}^X(\cdot,E) \coloneqq  \mc{F}^X(\cdot,\Peeled_X(E)).
	\end{gather*}   
	\end{definition}

	The next lemma allows us to control the product of the weights $|a_{ij}^{\tilde{v}}|$ on the edges of each ear (of the strict ear decomposition) in terms of the product of the $r_i^2$ of the internal vertices of that ear.

	\begin{lemma}[Opening of the paths and cycles]\label{lemma:along cycles} 
		\begingroup\color{black}
		Let $X$ be a subpartition of $[N]$ and let $E\in\mathsf{E}^X$ be minimally 2-edge-connected relative to $X$. Let the strict ear decomposition and the edge $e_0$ be as in Definition~\ref{def:peeling mini}. Let $a_{ij}^{\abs}$ be as in Definition~\ref{def:abs ab}.

		Let $P=(v_1v_2',v_2v_3',\ldots,v_{n-1}v_n')$ be an ear, with $n\ge3$ and $[v_i']^X=[v_i]^X$ for $2\le i\le n-1$, ordered as in that definition. Let $uv$ be its removed edge. Then, there exists a constant $C>0$ depending on $\beta$ such that
		\begin{equation}\label{mul436}
			\prod_{ij\in P}a_{ij}^{\abs}\leq C^{|P|}\prod_{i=2}^{n-1}(r_{v_i}r_{v_i'})\prod_{ij\in P\setminus\{uv\}}\Bigl(\frac{1}{d_{ij}^2}\indic_{d_{ij}\leq\textcolor{black}{17}\ve_0\Cut}\indic_{\mc{B}_{ij}^c}\Bigr)g_{uv},
		\end{equation}
		\endgroup
		where, for every $i,j\in [N], i\neq j$,
		\begin{equation}\label{def:gij}
			g_{ij}\coloneqq \indic_{\mc{B}_{ij}^c} \begin{cases}
				\indic_{d_{ij} \le \textcolor{black}{17} \ve_0 \Cut} & \text{if $d_{ij}< \max(r_i,r_j)$}\\
				\frac{\max_{v\in V_X}r_v^2}{d_{e_0}^2 }\indic_{d_{e_0}\leq \textcolor{black}{17} \ve_0\Cut} &\text{if $d_{ij}\geq \max_{v\in V_X}r_v$ \text{and $ij=e_0$}} \\
				\frac{\max(r_i,r_j)^2}{d_{ij}^2}\indic_{d_{ij}\leq \textcolor{black}{17} \ve_0\Cut } & \text{otherwise}.
			\end{cases}
		\end{equation}
		
	\end{lemma}

	\medskip

	\begin{figure}[H]
		\centering
		\fbox{\includegraphics[width=0.5\textwidth]{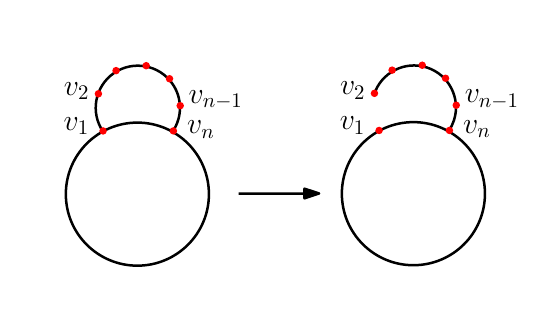}} 
		\caption{Opening of the path $(v_1v_2,v_2v_3,\ldots,v_{n-1}v_n)$ in the pure dipole case.}
	\end{figure}

	\begingroup\color{black}
	\begin{proof}[Proof of Lemma~\ref{lemma:along cycles}]
		We work on the support of the left-hand side. Unless $uv=e_0$ and $d_{e_0}\ge\max_{i\in V_X}r_i$, we write $e=v_1v_2'$, $f=v_{n-1}v_n'$, $s=\min(d_e,\max_{i\in e}r_i)$ and $t=\min(d_f,\max_{i\in f}r_i)$. Then
		\[
			\frac{r_{v_1}r_{v_n'}d_ed_f}{\max(d_e,\max_{i\in e}r_i)\max(d_f,\max_{i\in f}r_i)}
			\le st\le\max(s^2,t^2)=d_{uv}^2g_{uv}.
		\]
		For the other edges of $P$, we use $d_{ij}\max(d_{ij},r_i,r_j)\ge d_{ij}^2$, and \eqref{mul436} follows. If $uv=e_0$ and $d_{e_0}\ge\max_{i\in V_X}r_i$, we use $r_{v_1}r_{v_n'}\le\max_{i\in V_X}r_i^2=d_{e_0}^2g_{e_0}$ instead.
	\end{proof}
	\endgroup

	Consider a minimally $2$-edge-connected graph, which by Lemma~\ref{lemma:strict ear} admits a strict ear decomposition. We are going to  apply Lemma~\ref{lemma:along cycles} successively to each ear, while replacing each $r_i$ by the maximal one over its multipole.
	This allows us to bound the product of the 
	$a_{ij}^{\abs}$ over the graph by the product of the $r_{S}^2$, {\it after removing the contribution of the largest one}, keeping the product of the $d_{ij}^{-2}$ only on the tree left after peeling, with the multiplicative terms $g_{ij}$ appearing for the peeled-out edges. Notice that the terms $g_{ij}$ make the cluster-expansion series summable (and the integral convergent).

	\begin{definition}\label{def:B(T)}
		Let $X$ be a subpartition of $[N]$ and $T\in \mathsf{T}^X$. Let $(x_i,y_i)_{i\in V_X}\in (\Lambda^2)^{|V_X|}$. We let $B((x_i,y_i)_{i\in V_X},T)$ be the set of $T'\in \mathsf{T}^X$ such that 
		\begin{equation*}
			T'\cap \{ ij:d_{ij}\geq \max(r_i,r_j)\}= T\cap \{ ij:d_{ij}\geq \max(r_i,r_j)\}.
		\end{equation*}
	\end{definition}

	It means that $T$ and $T'$ coincide up to some edges $ij$ which are such that $d_{ij}<\max(r_i,r_j)$.

	\begin{coro}[Quadratic estimate for minimal 2-edge-connected graphs]\label{coro:prod a}
		Let $X$ be a subpartition of $[N]$. Let $a_{ij}^{\abs}$ be as in Definition~\ref{def:abs ab}. Let $E\in \mathsf{E}^X$ be such that $(V_X,E)$ is minimally 2-edge-connected relative to $X$.
		
		Let $(x_i,y_i)_{i\in V_X}\in (\Lambda^2)^{|V_X|}$. Let $T=\mc{T}^{X}(\cdot,E)$ and $F=\mc{F}^X(\cdot,E)$ be as in Definition \textcolor{black}{\ref{def:peeling mini}}. Let $B((x_i,y_i)_{i\in V_X},T)$ be as in Definition~\ref{def:B(T)}. Let $v_1$ be the index $i\in V_X$ such that $r_i$ is maximal.

		Then, there exists a constant $C>0$ depending on $\beta$ such that 	
		\begin{equation}\label{eq:peeling}
			\prod_{ij\in E}a_{ij}^{\abs}\leq C^{|V_X|}\prod_{S\in X:S\neq [v_1]^X}r_{S}^2\sum_{T'\in B((x_i,y_i)_{i\in V_X},T) }\prod_{ij\in T'}\Bigl(\frac{1}{d_{ij}^2}\indic_{d_{ij}\leq \textcolor{black}{17} \ve_0 \Cut}\indic_{\mc{B}_{ij}^c}\Bigr) \prod_{ij\in F}g_{ij},
		\end{equation}
		where $g_{ij}$ is as in \eqref{def:gij}.
	\end{coro}
	
	\begin{proof}
		Applying Lemma~\ref{lemma:along cycles} iteratively to every ear of the strict ear decomposition and to the base cycle, \textcolor{black}{and using $T\in B((x_i,y_i)_{i\in V_X},T)$}, we get
		\begin{equation}\label{eq:peeling first}
			\prod_{ij\in E}a_{ij}^{\abs}\leq C^{|E|}\prod_{S\in X:S\neq [v_1]^X}r_{S}^2\sum_{T'\in B((x_i,y_i)_{i\in V_X},T) }\prod_{ij\in T'}\Bigl(\frac{1}{d_{ij}^2}\indic_{d_{ij}\leq \textcolor{black}{17} \ve_0 \Cut}\indic_{\mc{B}_{ij}^c}\Bigr)\prod_{ij\in F}g_{ij}.
		\end{equation}
		
		\textcolor{black}{Since $T$ is a tree relative to $X$ and one edge is removed from each of the $K$ ears, we have $|E|=|X|-1+K$. Each ear has at least two edges, so that the base cycle covers at least two parts of $X$ and each later ear at least one new part. Hence $K\le|X|-1$, which yields}
		\begin{equation*}
			|E|\leq 2(|V_X|-1).
		\end{equation*}
		Inserting this into \eqref{eq:peeling first} concludes the proof.
	\end{proof}

	We now state a technical lemma which will be used later in the activity controls. We begin with a useful definition \cite{MR1154585}.

	\begin{definition}[Pseudoforest]\label{def:pseudo forest}
		Let $V\subset [N]$. 
		We denote by $\PF(V)$ the set of edges $F$ on $V$ such that $(V,F)$ is a pseudoforest, i.e., each connected component of $(V,F)$ has at most one cycle. Equivalently, $F\in \PF(V)$ if and only if $(V,F)$ admits an orientation such that the out-degree of every vertex is at most $1$.
	\end{definition}

	\begin{lemma}[Properties of the peeling]\label{lemma:technical peeling}
		Let $X$ be a subpartition of $[N]$ such that, for each $S\in X$, $|S|\le p(\beta)$. Let $(x_i,y_i)_{i\in V_X}\in(\Lambda^2)^{|V_X|}$. Let $E\in \Eulc^X$. Set
		\begin{equation*}
			T\coloneqq \mc{T}^X( (x_i,y_i)_{i\in V_X},E)\quad \text{and}\quad F\coloneqq \mc{F}^X((x_i,y_i)_{i\in V_X},E).
		\end{equation*}
		\begin{enumerate}
			\item 
			Then, $F\in \PF(V_X)$. Moreover,
			\begin{equation}\label{eq:sum PF}
				\sum_{F\in \PF(V_X) }\prod_{ij\in F}g_{ij}\leq \prod_{i\in V_X}\left( 1+\sum_{j\in V_X:j\neq i}g_{ij}\right).
			\end{equation}
			\item We have
			\begin{equation}\label{eq:hate F}
				\prod_{ij\in F}g_{ij}\leq \min\Bigl(\frac{\max_{i\in V_X}r_i}{\max_{e\in T}d_{e}},1\Bigr)^2.
			\end{equation}
			\item Let $T_0'\in \mathsf{T}^X$. \textcolor{black}{There exists a constant $C$ depending only on $p(\beta)$ such that, on $\bigcap_{i\in V_X}\{0<r_i\leq\ve_0\Cut\}\cap\bigcap_{ij\in\mc{E}^{\inter}(X)}\mc{B}_{ij}^c$,}
			\begin{equation}\label{eq:claimBT}
				|\{T_0\in \mathsf{T}^X:T_0'\in B((x_i,y_i)_{i\in V_X},T_0)\}|\leq C^{|V_X|}\textcolor{black}{\prod_{i\in V_X}\left(1+\log_M\frac{\textcolor{black}{17}\ve_0\Cut}{r_i}\right)^2}.
			\end{equation}
		\end{enumerate}
		
	\end{lemma}
	
	\begin{proof}
		\textcolor{black}{Each discarded edge is oriented from a vertex whose block first appears in its ear. Since each block first appears in only one ear, from which only one edge is discarded, every vertex has out-degree at most $1$. Hence $F\in\PF(V_X)$.}
		
		For every $F\in \PF(V_X)$, let $O(F)$ be an orientation of $F$ such that the out-degree of every vertex is at most $1$. Then,
		\begin{equation*}
			\sum_{F\in \PF(V_X) }\prod_{ij\in F}g_{ij}=\sum_{F\in \PF(V_X)}\prod_{i\to j\in O(F)} g_{ij}\leq  \prod_{i\in V_X}\left( 1+\sum_{j\in V_X:j\neq i}g_{ij}\right),
		\end{equation*}
		where the $1$ serves to include the possibility of an empty product.

		We turn to the proof of \eqref{eq:hate F}. Suppose that $\max_{e\in T}d_e\geq \max_{i\in V_X}r_i$. Otherwise the result is clear. Let $e_0$ be the index of the largest $d_e$ for $e\in T\cup F$. Then, $d_{e_0}\geq \max_{i\in V_X}r_i$. Therefore, by  definition of the peeling algorithm (see Definition~\ref{def:peeling mini}, case $\ell=\ell_0$), we have that $e_0\in F$. Therefore, by \eqref{def:gij}, 
		\begin{equation*}
			g_{e_0}\leq \Bigl(\frac{\max_{i\in V_X}r_i}{d_{e_0} }\Bigr)^2\leq \Bigl(\frac{\max_{i\in V_X}r_i}{\max_{e\in T}d_{e} }\Bigr)^2=\min\Bigl(\frac{\max_{i\in V_X}r_i}{\max_{e\in T}d_{e}},1\Bigr)^2.
		\end{equation*}
		Using that $g_{ij}\leq 1$ for every $ij\in F\setminus\{e_0\}$, we obtain \eqref{eq:hate F}.

		\begingroup\color{black}
		Finally, since the relation $T_0'\in B((x_i,y_i)_{i\in V_X},T_0)$ is symmetric in $T_0$ and $T_0'$, the left-hand side of \eqref{eq:claimBT} equals $|B((x_i,y_i)_{i\in V_X},T_0')|$. Set $n=|V_X|$, $R=\ve_0\Cut$, $q=M/3$ and $h_i=\lfloor\log_q(R/r_i)\rfloor$. Let $G$ have vertex set $V_X$ and edge set
		\[
			T_0'\cup\{ij\in\mc{E}^{\inter}(X):d_{ij}<\max(r_i,r_j)\}.
		\]
		We orient each edge $ij$ with $d_{ij}<\max(r_i,r_j)$ from the smaller to the larger radius. For such an edge $i\to j$, we have $r_j>d_{ij}>Mr_i$ on $\mc B_{ij}^c$, hence $h_j<h_i$. If two out-neighbors $j,k$ of $i$ satisfy $h_j=h_k$, then their radii lie in an interval $(u,qu]$, and
		\[
			d_{jk}< (2q+q/M)u<M\min(r_j,r_k).
		\]
		Since $d_{jk}>M\min(r_j,r_k)$ when $j$ and $k$ lie in distinct blocks, they lie in the same block. Thus, for each value of $h<h_i$, the vertex $i$ has at most $p(\beta)$ out-neighbors $j$ with $h_j=h$. Let $L_G$ be the Laplacian of $G$ and $D=\mathrm{diag}_{i\in V_X}(1+h_i)^2$. Since $T_0'$ has at most $n-1$ edges, we obtain
		\[
			\operatorname{tr}\bigl(D^{-1}(I+L_G)\bigr)
			\le 3n+p(\beta)\sum_{i\in V_X}\left(\frac{h_i}{(1+h_i)^2}+\sum_{h=0}^{h_i-1}\frac1{(1+h)^2}\right)\le Cn.
		\]
		Every $T_0\in B((x_i,y_i)_{i\in V_X},T_0')$ is a forest in $G$. By the Cauchy--Binet formula, $\det(I+L_G)$ is a sum of positive integers over the forests of $G$, so that it bounds their number. The arithmetic-geometric mean inequality, applied to the eigenvalues of $D^{-1/2}(I+L_G)D^{-1/2}$, then gives
		\[
			|B((x_i,y_i)_{i\in V_X},T_0')|
			\le\det(I+L_G)
			\le\det D\left(\frac{\operatorname{tr}(D^{-1}(I+L_G))}{n}\right)^n
			\le C^n\prod_{i\in V_X}(1+h_i)^2.
		\]
		Since $\log M/\log(M/3)<2$, this proves \eqref{eq:claimBT}.
		\endgroup
	\end{proof}

	\subsection{Control of bounded clusters}\label{sub:bounded lower proof}
	
	We now proceed to the proof of Proposition~\ref{prop:bounded lower}.

	\begin{proof}[Proof of Proposition~\ref{prop:bounded lower}]
		Let us define
		\begin{equation}\label{defI1}
			\textcolor{black}{I_{\mathrm{mult},\ve_0}^{-}}(S)=\int_{(\Lambda^2)^{|S|} } \indic_{\mc{B}_S}\prod_{ i,j\in S:i<j}e^{-\beta \tilde{v}_{ij}}\indic_{\mc{A}_{ij}}\prod_{i\in S }e^{-L\frac{|x_i-y_i|^2}{(\ve_0\Cut )^2}} \prod_{i\in S} e^{\beta \g_\lambda(x_i-y_i)}\indic_{|x_i-y_i|\leq \ve_0\Cut}\dd x_i \dd y_i.
		\end{equation}
		Observe that
		\begin{equation*}
			\Msf_{\ve_0}^{-}(S)=\frac{1}{(NC_{\beta,\lambda,\ve_0})^{|S|}}\textcolor{black}{I_{\mathrm{mult},\ve_0}^{-}}(S).
		\end{equation*}
		Taking $Z=\emptyset$, $\tau= (\tau_0)$ and $w=\tilde{v}$ in Lemma~\ref{lemma:lowerM+}, \textcolor{black}{with $-64L$ in place of $L$,} we get that there exists $C>0$ depending only on  $\beta$, $M$ and $|S|$  such that
		\begin{equation*}
			\frac{1}{C}N\lambda^{(2-\beta)|S|}\lambda^{2(|S|-1)} \leq \textcolor{black}{I_{\mathrm{mult},\ve_0}^{-}}(S)\leq CN \lambda^{(2-\beta)|S|}\lambda^{2(|S|-1)}.
		\end{equation*}
		Moreover, by \eqref{eq:Clambda bound} \textcolor{black}{and $C_{\beta,\lambda,\ve_0}\leq\mathcal Z_\beta\lambda^{2-\beta}$}, there exists $C>0$ depending on $\beta$ such that $\textcolor{black}{C^{-1}\lambda^{2-\beta}\leq C_{\beta,\lambda,\ve_0}\leq C\lambda^{2-\beta}}$. This proves \eqref{eq:MX -low} and \eqref{eq:MX -up} \textcolor{black}{for $\Msf_{\ve_0}^-$}.
\begingroup\color{black}
Let $h=\ve_0\Cut$. The logarithm of the ratio of the integrands in \eqref{def:M-} and \eqref{def:M0 ve_0} is
\[
-\beta\sum_{S\in X}\sum_{i<j\in S}(\tilde v_{ij}-v_{ij})-L\sum_{i\in V_X}r_i^2/h^2,
\]
whose absolute value is at most $C_\beta|V_X|^2$ when $r_i\leq h$ for every $i$, by Lemma~\ref{lemma:errortilde}. Hence, for every subpartition $X$ of $[N]$,
\[
C_{\beta,|V_X|}^{-1}\Msf_{\ve_0}^0(X)
\leq\Msf_{\ve_0}^-(X)\leq C_{\beta,|V_X|}\Msf_{\ve_0}^0(X).
\]
Together with the bounds on $\Msf_{\ve_0}^-$, this proves \eqref{eq:MX -low} and \eqref{eq:MX -up} for $\Msf_{\ve_0}^0$. For $\Msf_\infty^0$, we obtain the lower bound by restricting the integral to $r_i\leq h$ and using \eqref{eq:comparisons Cve0}. For the upper bound, we use that $e^{-\beta v_{ij}}\indic_{\mc A_{ij}}\leq C_\beta$\edgreen{ by \eqref{eq:bond-bound}}, and that the right-hand side of \eqref{st2} is bounded in the same way with $\infty$ in place of $\ve_0\Cut$ when $|S|\leq p^*(\beta)$. Dividing by $(NC_{\beta,\lambda})^{|S|}$, we obtain \eqref{eq:MX -up}.
\endgroup
		
		Let us now prove \eqref{eq:Kdip bound}. \textcolor{black}{Let $k=|V|$.} Set 
		\begin{equation*}
			\textcolor{black}{I_{\mathrm{dip},\infty}(V;E)}=\int_{(\Lambda^2)^{|V|}}\prod_{ij\in E}f_{ij}^v \prod_{i\in V}e^{\beta \g_\lambda(x_i-y_i)}\dd x_i \dd y_i.
		\end{equation*}
		In view of Definition~\ref{def:dipole activity}, the definition of $f_{ij}$ and \eqref{def:Clambda}, we have that 
		\begin{equation*}
			\Ksf_{\beta,\lambda}^\dip(V)=\sum_{E\in \mc{G}_c(V)} \frac{\textcolor{black}{I_{\mathrm{dip},\infty}(V;E)}}{(NC_{\beta,\lambda})^{|V|}},
		\end{equation*}
		where $\mc{G}_c(V)$ stands for the set of \textcolor{black}{nonempty} collections of edges $E$ on $V$ such that $(V,E)$ is connected.
		
		In Lemma~\ref{lemma:int dipole}, we show that there exists $C>0$ depending on $\beta$ and $|V|$ such that, for every $E\in \mc{G}_c(V)$, 
		\begin{equation*}
			| \textcolor{black}{I_{\mathrm{dip},\infty}(V;E)}|\leq CN\lambda^{k(2-\beta)+2(k-1)}.
		\end{equation*}
		Besides, recall  from Lemma~\ref{lem411} that there exists $C>0$ depending on $\beta$ such that $C_{\beta,\lambda}\geq \frac{1}{C}\lambda^{2-\beta}$. Hence, combining the above relations and using that the number of graphs is bounded by a constant depending only on $|V|$ establishes \eqref{eq:Kdip bound}.

		It therefore remains to prove \eqref{eq:K- bound}. Let $k=|V_X|$. Recall \eqref{eq:activity lower}.
		\textcolor{black}{The unordered families $(X_l,E_l)$ are determined by $\cup_lE_l$, through its connected components relative to $X$. Since the factor $1/n!$ compensates for the orderings of these families, the sum of the weights $1/n!$ of the terms of \eqref{eq:activity lower} is at most the number $2^{2\binom{k}{2}}$ of pairs $(\cup_lE_l,F)$. Hence}
		\begingroup\color{black}
		\begin{equation}\label{eq:stt}
			|\Ksf_{\ve_0}^-(X)|\leq 2^{2\binom{k}{2}}\max_{\substack{n\geq 0, X_1,\ldots,X_n\text{ disjoint} \\ E_1\in \mathsf{E}^{X_1},\ldots,E_n\in \mathsf{E}^{X_n}\\ F\in \textcolor{black}{\mathsf{E}_0}^{\Coarse_X(X_1,\ldots,X_n)} }} \left|\dE_{\Psf_{X}^{-,\ve_0}}\left[\prod_{ij\in E_1\cup\cdots \cup E_n}f_{ij}^{\tilde{v}}\prod_{ij\in\cup_l\mc{E}^\inter(X_l)}\indic_{\mc{B}_{ij}^c}\prod_{ij\in F}(-\indic_{\mc{B}_{ij}})\right]\right|.
		\end{equation}
		\endgroup
		Moreover, by the parity argument of Lemma~\ref{lemma:cancellation odd}, there exists $C>0$ depending on $k$ such that
		\begin{multline*}
			|\Ksf_{\ve_0}^-(X)|\leq C\max_{\substack{n\geq 0, X_1,\ldots,X_n\text{ disjoint}}} \Bigl(\prod_{l=1}^n\max_{\substack{E_{l,1}\in \Eul^{X_l},\\ E_{l,2}\subset \mc{E}^{\inter}(X_l):\\ 
					E_{l,1}\cap E_{l,2}=\emptyset \\
					E_{l,1 }\cup E_{l,2}\in \mathsf{E}^{X_l} } } \Bigr)\max_{F\in \textcolor{black}{\mathsf{E}_0}^{\Coarse_X(X_1,\ldots,X_n)}}  \\ \left|\dE_{\Psf_{X}^{-,\ve_0}}\left[\prod_{ij\in \cup_l E_{l,1}}a_{ij}^{\tilde{v}} \prod_{ij\in \cup_l E_{l,2}}b_{ij}^{\tilde{v}}\prod_{ij\in\cup_l\mc{E}^\inter(X_l)}\indic_{\mc{B}_{ij}^c}\prod_{ij\in F}(-\indic_{\mc{B}_{ij}})\right]\right| .
		\end{multline*}
		Recall from \textcolor{black}{Remark~\ref{remarkabbi}} that there exists $C>0$ depending on $\beta$ such that $|a_{ij}^{\tilde{v}}|\indic_{\mc{B}_{ij}^c}\leq Ca_{ij}^\abs$ and $|b_{ij}^{\tilde{v}}|\indic_{\mc{B}_{ij}^c}+\indic_{\mc{B}_{ij}}\leq Cb_{ij}^\abs$. Thus, taking the maximum according to the connected components of the graph $(V_X,\cup_l \textcolor{black}{E_{l,1}}\cup \mc{E}^\intra(X))$ first, we get that there exists $C>0$ depending on $\beta$ and $k$ such that
		\begin{multline*}
			|\Ksf_{\ve_0}^-(X)|\leq C\max_{\substack{n\geq 0, X_1,\ldots,X_n\text{ disjoint}}}\Bigl(\prod_{l=1}^n \max_{E_l\in \Eulc^{X_l}}\Bigr)\max_{T^b\in \textcolor{black}{\mathsf{T}_0}^{\Coarse_X(X_1,\ldots,X_n)}}
			\dE_{\Psf_{X}^{-,\ve_0}}\left[\prod_{ij\in \cup_l E_{l}}a_{ij}^{\abs} \prod_{ij\in T^b}b_{ij}^{\abs}\right].
		\end{multline*}

		For every $l\in [n]$, let $\hat{S}_l(X_1,\ldots,X_n)$ be the $S\in X_l$ such that $r_S$ is maximal. Applying Corollary~\ref{coro:prod a} and the estimate \eqref{eq:hate F} of Lemma~\ref{lemma:technical peeling}, \textcolor{black}{whose proof applies to blocks of every size,} we get that there exists $C>0$ depending on $\beta, M$ and $k$ such that 
		\begin{multline}\label{eq:int rhs}
			|\Ksf_{\ve_0}^-(X)|\leq C\max_{\substack{n\geq 0, X_1,\ldots,X_n\text{ disjoint}}}\Bigl(\prod_{l=1}^n \max_{E_l\in \Eulc^{X_l}}\max_{T_l^a\in \mathsf{T}^{X_l} }\Bigr) \max_{T^b\in \textcolor{black}{\mathsf{T}_0}^{\Coarse_X(X_1,\ldots,X_n)}}\\
			\frac{1}{(NC_{\beta,\lambda,\ve_0})^k \Msf_{\ve_0}^-(X)}\int_{(\Lambda^2)^k } \prod_{l=1}^n \prod_{\substack{S\in X_l:\\ S\neq \hat{S}_l(X_1,\ldots,X_n) }}r_S^2 \prod_{ij\in T_l^a} \frac{1}{d_{ij}^2}\indic_{d_{ij}\leq \textcolor{black}{17}\ve_0\Cut}\indic_{\mc{B}_{ij}^c}\prod_{ij\in T^b}b_{ij}^\abs \min\Bigl(\frac{\max_{i\in V_{X}}r_i}{\max_{ij\in  \cup_l T_l^a}d_{ij}},1\Bigr)^2\\ \times \prod_{S\in X}\indic_{\mc{B}_S}\prod_{i\in V_X}e^{\beta \g_\lambda(x_i-y_i)}\indic_{|x_i-y_i|\leq \ve_0 \Cut}\dd x_i\dd y_i.
		\end{multline}

		In Lemma~\ref{lemma:integral small -}, we control the integral on the right-hand side of \eqref{eq:int rhs}. Inserting \eqref{eq:bound J1} into \eqref{eq:int rhs}, we find that there exists $C>0$ depending on $\beta,M$ and $k$ such that
		\begin{equation*}
			|\Ksf_{\ve_0}^-(X)|\leq C\frac{N\lambda^{(2-\beta)k}}{(NC_{\beta,\lambda,\ve_0})^k \Msf_{\ve_0}^-(X)} \gamma_{\beta,\lambda,k},
		\end{equation*}
		where $\gamma_{\beta,\lambda,k}$ is as in Definition~\ref{def:gamma beta}. Recall that $C_{\beta,\lambda,\ve_0}\geq \frac{1}{C}\lambda^{2-\beta}$. \textcolor{black}{Moreover, $\Msf_{\ve_0}^-(X)\geq C_{\beta,k}^{-1}\Msf_{\ve_0}^0(X)$ by the comparison of $\Msf_{\ve_0}^-$ and $\Msf_{\ve_0}^0$ above.} This proves \eqref{eq:K- bound} \textcolor{black}{for $\Ksf_{\ve_0}^-$}.
\begingroup\color{black}
For $\Ksf_{\ve_0}^0$, we use the weights $\widehat a,\widehat b$ defined after Lemma~\ref{lemma:integral small -}. By Lemma~\ref{lemma:vij}, the bonds $a_{ij}^v$ and $b_{ij}^v$ are bounded by $C_\beta\widehat a_{ij}$ and $C_\beta\widehat b_{ij}$ on $\mc B_{ij}^c$. The parity argument above and the proof of \eqref{eq:prodE1l'}, applied directly with the radii $r_i$, which are positive outside a null set, then give \eqref{eq:int rhs} without the spatial cutoffs, with $\widehat b$ in place of $b^\abs$ and with $\Msf_{\ve_0}^0(X)$ in the denominator. The factors $e^{-\beta v_{ij}}\indic_{\mc A_{ij}}$ with $i$ and $j$ in the same block are bounded\edgreen{, by \eqref{eq:bond-bound}}. Since \eqref{eq:bound J1} holds for $\widehat J$, this proves \eqref{eq:K- bound} for $\Ksf_{\ve_0}^0$. If $|V_X|\leq p^*(\beta)$, then \eqref{eq:bound J1} also holds for $\widehat J_\infty$, which proves the bound on $\Ksf_\infty^0$.
\endgroup
	\end{proof}

	\subsection{Control of expansion errors}\label{sub:exp -}

	We can now complete the proof of Proposition~\ref{prop:expansion -}.

	\begin{proof}[Proof of Proposition~\ref{prop:expansion -}]
		We first prove \eqref{eq:M-diff}. Let $\textcolor{black}{I_{\mathrm{mult},\ve_0}^{-}}(S)$ be as in \eqref{defI1} and
		\begin{equation*}
			\textcolor{black}{I_{\mathrm{mult},\infty}^{0}}(S)\coloneqq \int_{(\Lambda^2)^{|S|} } \indic_{\mc{B}_S}\prod_{ i,j\in S:i<j}e^{-\beta v_{ij}}\indic_{\mc{A}_{ij}}\prod_{i\in S} e^{\beta \g_\lambda(x_i-y_i)}\dd x_i \dd y_i.
		\end{equation*}
		By Definition~\ref{def:multipolemeasure}, we have
		\begin{equation}\label{eq:M-M0}
			\begin{aligned}\Msf_{\ve_0}^-(S)-  \Msf^0_{\infty} (S)&=\frac{1}{(N C_{\beta,\lambda,\ve_0})^{|S|}}(\textcolor{black}{I_{\mathrm{mult},\ve_0}^{-}}(S)-\textcolor{black}{I_{\mathrm{mult},\infty}^{0}}(S))\\ &\quad+\textcolor{black}{I_{\mathrm{mult},\infty}^{0}}(S)\left(\frac{1}{(N C_{\beta,\lambda,\ve_0})^{|S|}}-\frac{1}{(N C_{\beta,\lambda})^{|S|}}\right).
		\end{aligned}\end{equation}
		Taking \textcolor{black}{$Z=\emptyset$, $(w,\tau)=(\tilde v,\tilde\tau)$ and replacing the parameter $L$ by $-64L$} in Lemma~\ref{lemma:general exp M}, we obtain that there exists $C>0$ depending on $\beta,M,|S|$ and $\ve_0$ such that
		\begin{equation*}
			| \textcolor{black}{I_{\mathrm{mult},\ve_0}^{-}}(S)-\textcolor{black}{I_{\mathrm{mult},\infty}^{0}}(S)|\leq CN\Cut^{-2}\lambda^{(2-\beta)|S|}.
		\end{equation*}
		Therefore, by \eqref{eq:Clambda bound}, there exists $C>0$ depending on $\beta,M,|S|$ and $\ve_0$ such that
		\begin{equation}\label{eq:hii1}
			\frac{1}{(N C_{\beta,\lambda,\ve_0})^{|S|}}|\textcolor{black}{I_{\mathrm{mult},\ve_0}^{-}}(S)-\textcolor{black}{I_{\mathrm{mult},\infty}^{0}}(S)| \leq CN^{1-|S|}\Cut^{-2}.
		\end{equation}
		On the other hand, by \eqref{eq:comparisons Cve0}, there exists $C>0$ depending on $\beta$, $|S|$, and $\ve_0$ such that
		\begin{equation*}
			\left|\frac{1}{C_{\beta,\lambda,\ve_0}^{|S|}}-\frac{1}{C_{\beta,\lambda}^{|S|}}\right|\leq C\frac{\Cut^{-2}}{\lambda^{(2-\beta)|S|}}.
		\end{equation*}
		\textcolor{black}{If $|S|=1$, then $\textcolor{black}{I_{\mathrm{mult},\infty}^{0}}(S)=NC_{\beta,\lambda}$. If $|S|>1$, then $\textcolor{black}{I_{\mathrm{mult},\infty}^{0}}(S)=(NC_{\beta,\lambda})^{|S|}\Msf_\infty^0(S)$, and we use \eqref{eq:MX -up}. In both cases, since $C_{\beta,\lambda}\leq\mathcal Z_\beta\lambda^{2-\beta}$, we obtain, with $C$ depending on $\beta,M,|S|$,}
		\begin{equation*}
			\textcolor{black}{I_{\mathrm{mult},\infty}^{0}}(S)\leq CN\lambda^{(2-\beta)|S|+2(|S|-1)}.
		\end{equation*}
		Hence, there exists $C>0$ depending on $\beta,M,|S|$ and $\ve_0$ such that
		\begin{equation}\label{eq:hii2}
			\textcolor{black}{I_{\mathrm{mult},\infty}^{0}}(S)\left|\frac{1}{(N C_{\beta,\lambda,\ve_0})^{|S|}}-\frac{1}{(N C_{\beta,\lambda})^{|S|}}\right|\leq CN^{1-|S|}\lambda^{2(|S|-1)} \Cut^{-2}\leq CN^{1-|S|}\Cut^{-2}.
		\end{equation}
		Combining \eqref{eq:hii1} and \eqref{eq:hii2} proves \eqref{eq:M-diff}.

		Let us now prove \eqref{eq:K-diff}. Let $X$ be a subpartition of $[N]$ and set $k\coloneqq |V_X|$. Recall that $\Ksf_{\ve_0}^-(X)$ and $\Ksf_\infty^0$ are defined in Definition~\ref{def:activity lower}. 	We set
		\begin{multline*}
			\textcolor{black}{I_{\mathrm{cl},\ve_0}^{-}(X;\mathbf X,\mathbf E,F)}\\= \int_{(\Lambda^2)^{k} } \prod_{ij\in \cup_l E_l}f_{ij}^{\tilde{v}}\prod_{ij\in F}(-\indic_{\mc{B}_{ij}})\prod_{ij\in \cup_{l=1}^n \mc{E}^\inter(X_l)}\indic_{\mc{B}_{ij}^c}  \left(\prod_{S\in X}\indic_{\mc{B}_S} \textcolor{black}{\prod_{\substack{i,j\in S\\i<j}}}e^{-\beta \tilde{v}_{ij}}\indic_{\mc{A}_{ij}}\right) \\ \times\prod_{i\in V_X}e^{-L\frac{|x_i-y_i|^2}{(\ve_0\Cut)^2}}\prod_{i\in V_X}e^{\beta \g_\lambda(x_i-y_i)}\indic_{|x_i-y_i|\leq \ve_0\Cut}\dd x_i \dd y_i
		\end{multline*}
		and
		\begin{multline*}
			\textcolor{black}{I_{\mathrm{cl},\infty}^{0}(X;\mathbf X,\mathbf E,F)}\\= \int_{(\Lambda^2)^{k} } \prod_{ij\in \cup_l E_l}f_{ij}^{v}\prod_{ij\in F}(-\indic_{\mc{B}_{ij}})\prod_{ij\in \cup_{l=1}^n \mc{E}^\inter(X_l)}\indic_{\mc{B}_{ij}^c}  \left(\prod_{S\in X}\indic_{\mc{B}_S} \textcolor{black}{\prod_{\substack{i,j\in S\\i<j}}}e^{-\beta v_{ij}}\indic_{\mc{A}_{ij}}\right) \\ \times\prod_{i\in V_X}e^{\beta \g_\lambda(x_i-y_i)}\dd x_i \dd y_i.
		\end{multline*}
\textcolor{black}{Here $\mathbf X=(X_l)$ and $\mathbf E=(E_l)$, and we omit the arguments $(X;\mathbf X,\mathbf E,F)$ below.}
		Fix $X_1,\ldots,X_n\subset X$ disjoint, $E_1\in \mathsf{E}^{X_1},\ldots,E_n\in \mathsf{E}^{X_n}$ and $F\in \textcolor{black}{\mathsf{E}_0}^{\Coarse_X(X_1,\ldots,X_n)}$. We may write
		\begin{equation*}
			\dE_{\Psf_{X}^{-,\ve_0}}\left[\prod_{ij\in E_1\cup \cdots \cup E_n}f^{\tilde{v}}_{ij}\prod_{ij\in \mc{E}^\inter(X_1)\cup \cdots \cup \mc{E}^\inter(X_n)}\indic_{\mc{B}_{ij}^c}\prod_{ij\in F}(-\indic_{\mc{B}_{ij}})\right]=\frac{\textcolor{black}{I_{\mathrm{cl},\ve_0}^{-}}}{\Msf_{\ve_0}^-(X)(NC_{\beta,\lambda,\ve_0})^{k}}
		\end{equation*}
		and
		\begin{equation*}
			\dE_{\Psf_{X}^{0,\infty}}\left[\prod_{ij\in E_1\cup \cdots \cup E_n}f^{v}_{ij}\prod_{ij\in \mc{E}^\inter(X_1)\cup \cdots \cup \mc{E}^\inter(X_n)}\indic_{\mc{B}_{ij}^c}\prod_{ij\in F}(-\indic_{\mc{B}_{ij}})\right]=\frac{\textcolor{black}{I_{\mathrm{cl},\infty}^{0}}}{\Msf^0_{\infty} (X)(NC_{\beta,\lambda})^{k}} .
		\end{equation*}
		Next, we use 
		\begin{align}\label{eq:rr3 -}
			& \left| \dE_{\Psf_{X}^{-,\ve_0}}\left[\prod_{ij\in E_1\cup \cdots \cup E_n}f^{\tilde{v}}_{ij}\prod_{ij\in \mc{E}^\inter(X_1)\cup \cdots \cup \mc{E}^\inter(X_n)}\indic_{\mc{B}_{ij}^c}\prod_{ij\in F}(-\indic_{\mc{B}_{ij}})\right]\right.\\ \notag &\qquad -\left. \dE_{\Psf_{X}^{0,\infty}}\left[\prod_{ij\in E_1\cup \cdots \cup E_n}f^{v}_{ij}\prod_{ij\in \mc{E}^\inter(X_1)\cup \cdots \cup \mc{E}^\inter(X_n)}\indic_{\mc{B}_{ij}^c}\prod_{ij\in F}(-\indic_{\mc{B}_{ij}})\right] \right| \\ \notag & \qquad\leq \frac{1}{\Msf_{\ve_0}^-(X)(NC_{\beta,\lambda,\ve_0})^k}\Bigr| \textcolor{black}{I_{\mathrm{cl},\ve_0}^{-}}- \textcolor{black}{I_{\mathrm{cl},\infty}^{0}}\Bigr|\\ \notag &\qquad + |\textcolor{black}{I_{\mathrm{cl},\infty}^{0}}|\times \left|\frac{1}{(NC_{\beta,\lambda,\ve_0})^k\Msf_{\ve_0}^-(X)}-\frac{1}{(NC_{\beta,\lambda})^k\Msf^0_{\infty} (X)}\right|.
		\end{align}
		\textcolor{black}{With the same parameters, Lemma {\ref{lemma:general exp K}} gives} $C>0$ depending on $\beta,M, k$ and $\ve_0$ such that
		\begin{equation*}
			\Bigr|\textcolor{black}{I_{\mathrm{cl},\ve_0}^{-}}-\textcolor{black}{I_{\mathrm{cl},\infty}^{0}}\Bigr|\leq CN\Cut^{-2}\lambda^{(2-\beta)k}.
		\end{equation*}
		Therefore, by Proposition~\ref{prop:bounded lower} and \eqref{eq:Clambda bound}, there exists $C>0$ depending on $\beta,M, k$ and $\ve_0$ such that
		\begin{equation}\label{eq:rr1 -}
			\frac{1}{\Msf_{\ve_0}^-(X)(NC_{\beta,\lambda,\ve_0})^k}\Bigr| 
			\textcolor{black}{I_{\mathrm{cl},\ve_0}^{-}}- \textcolor{black}{I_{\mathrm{cl},\infty}^{0}}\Bigr|\leq \frac{CN^{1-k}}{\Msf^0_{\ve_0} (X)}\Cut^{-2}.
		\end{equation}
		Notice that 
		\begin{equation*}
			\left|\frac{1}{(NC_{\beta,\lambda,\ve_0})^k\Msf_{\ve_0}^-(X)}-\frac{1}{(NC_{\beta,\lambda})^k\Msf^0_{\infty} (X)}\right|=\left| \frac{1}{\prod_{S\in X}\textcolor{black}{I_{\mathrm{mult},\ve_0}^{-}}(S) }-  \frac{1}{\prod_{S\in X}\textcolor{black}{I_{\mathrm{mult},\infty}^{0}}(S)}\right|.
		\end{equation*}

\begingroup\color{black}
By \eqref{eq:MX -low}, \eqref{eq:MX -up} and Lemma~\ref{lem411},
$\textcolor{black}{I_{\mathrm{mult},\ve_0}^{-}}(S)\asymp \textcolor{black}{I_{\mathrm{mult},\infty}^{0}}(S)\asymp N\lambda^{(4-\beta)|S|-2}$ for $1<|S|\leq p^*(\beta)$.
For a singleton,
\[
e^{-L}NC_{\beta,\lambda,\ve_0}\leq \textcolor{black}{I_{\mathrm{mult},\ve_0}^{-}}(\{i\})\leq NC_{\beta,\lambda,\ve_0},
\qquad \textcolor{black}{I_{\mathrm{mult},\infty}^{0}}(\{i\})=NC_{\beta,\lambda}.
\]
Hence these bounds also hold for $|S|=1$. Writing the difference of the two products as a telescoping sum and using Lemma~\ref{lemma:general exp M}, we obtain, with $C$ depending on $\beta,M,k,\ve_0$,
\endgroup
		\begin{equation*}
			\frac{\left|\prod_{S\in X}\textcolor{black}{I_{\mathrm{mult},\ve_0}^{-}}(S)-\prod_{S\in X}\textcolor{black}{I_{\mathrm{mult},\infty}^{0}}(S)\right|}{\prod_{S\in X} \textcolor{black}{I_{\mathrm{mult},\ve_0}^{-}}(S)}\leq C\Cut^{-2}\lambda^{-2(\max_{S\in X}|S|-1)}.
		\end{equation*}
		Hence, by \textcolor{black}{these lower bounds}, there exists $C>0$ depending on $\beta,M,k$ and $\ve_0$ such that
		\begin{equation*}
			\left| \frac{1}{\prod_{S\in X}\textcolor{black}{I_{\mathrm{mult},\ve_0}^{-}}(S) }-  \frac{1}{\prod_{S\in X}\textcolor{black}{I_{\mathrm{mult},\infty}^{0}}(S)}\right|\leq  C\frac{\Cut^{-2}\lambda^{-2(\max_{S\in X}|S|-1)} }{N^{|X|}\lambda^{(2-\beta)k}\prod_{S\in X}\lambda^{2(|S|-1)}}.
		\end{equation*}
		\textcolor{black}{We average over the simultaneous reversals of the dipoles of each block, and argue as in the proof of Lemma~\ref{lemma:int dipole}, with its split into near and far edges and its bounds on the odd and even bonds. This gives, with $C$ depending on $\beta,M,k$,}
		\begin{equation}\label{bornI2}
			|\textcolor{black}{I_{\mathrm{cl},\infty}^{0}}|\leq CN\lambda^{(2-\beta)k+2(k-1)}.
		\end{equation}
		Assembling the two above displays yields the existence of $C>0$ depending on $\beta,M,k$ and $\ve_0$ such that
		\begin{multline*}
			|\textcolor{black}{I_{\mathrm{cl},\infty}^{0}}|\times  \left| \frac{1}{\prod_{S\in X}\textcolor{black}{I_{\mathrm{mult},\ve_0}^{-}}(S) }-  \frac{1}{\prod_{S\in X}\textcolor{black}{I_{\mathrm{mult},\infty}^{0}}(S)}\right|\\ \leq CN^{1-|X|}\Cut^{-2}\frac{1}{\prod_{S\in X}\lambda^{2(|S|-1)}}\lambda^{2(k-1)}\lambda^{-2(\max_{S\in X}|S|-1)} \leq CN^{1-|X|}\Cut^{-2}\frac{1}{\prod_{S\in X}\lambda^{2(|S|-1)}}.
		\end{multline*}
		Therefore, by \textcolor{black}{\eqref{eq:MX -up}}, there exists $C>0$ depending on $\beta,M,k$ and $\ve_0$ such that
		\begin{equation}\label{eq:rr2 -}
			|\textcolor{black}{I_{\mathrm{cl},\infty}^{0}}|\times  \left| \frac{1}{\prod_{S\in X}\textcolor{black}{I_{\mathrm{mult},\ve_0}^{-}}(S) }-  \frac{1}{\prod_{S\in X}\textcolor{black}{I_{\mathrm{mult},\infty}^{0}}(S)}\right|\leq \frac{CN^{1-k}}{\Msf^0_{\infty} (X)}\Cut^{-2}.  
		\end{equation}
		Inserting \eqref{eq:rr1 -} and \eqref{eq:rr2 -} into \eqref{eq:rr3 -} gives 
		\begin{multline*}
			\left| \dE_{\Psf_{X}^{-,\ve_0}}\left[\prod_{ij\in E_1\cup \cdots \cup E_n}f^{\tilde{v}}_{ij}\prod_{ij\in \mc{E}^\inter(X_1)\cup \cdots \cup \mc{E}^\inter(X_n)}\indic_{\mc{B}_{ij}^c}\prod_{ij\in F}(-\indic_{\mc{B}_{ij}})\right]\right.\\ -\left. \dE_{\Psf_{X}^{0,\infty}}\left[\prod_{ij\in E_1\cup \cdots \cup E_n}f^{v}_{ij}\prod_{ij\in \mc{E}^\inter(X_1)\cup \cdots \cup \mc{E}^\inter(X_n)}\indic_{\mc{B}_{ij}^c}\prod_{ij\in F}(-\indic_{\mc{B}_{ij}})\right] \right| \leq \frac{CN^{1-k}}{\Msf^0_{\infty} (X)}\Cut^{-2},
		\end{multline*}
		for some constant $C>0$ depending on $\beta,M,k$ and $\ve_0$. This concludes the proof of \eqref{eq:K-diff}. \textcolor{black}{We prove \eqref{eq:K-diff0} in the same way, with \eqref{eq:CC2} for each multipole and \eqref{def:estimateI2 bis} for the numerator.}

		\textcolor{black}{Finally, for \eqref{eq:Kdip diff}, Lemma~\ref{lemma:I3} bounds the difference of the numerators, and Lemma~\ref{lemma:int dipole} and \eqref{eq:comparisons Cve0} bound the error due to the normalizations.}
	\end{proof}

	\subsection{Study of limiting activities}\label{sub:limiting}

	\begin{proof}[Proof of Lemma~\ref{lemma:limiting}] Let $G:(\dR^2)^{n-1}\to \dR$ be integrable. Define 
		\begin{equation*}
			I_N\coloneqq \int_{\Lambda^{n}}G(z_2-z_1,z_3-z_1,\ldots,z_n-z_1)\dd z_1\ldots \dd z_n,
		\end{equation*}
		where we recall that $\Lambda=\Lambda(N)=[0,\sqrt{N}]^2$. Then, we claim that 
		\begin{equation}\label{eq:claim int G}
			\lim_{N\to \infty}\frac{I_N}{N}=\int_{(\dR^2)^{n-1}}G(y_1,\ldots,y_{n-1})\dd y_1\ldots \dd y_{n-1}.
		\end{equation}
		Indeed,
		\begin{equation*}
			\frac{I_N}{N}=\int_{(\dR^2)^{n-1}}G(y_1,\ldots,y_{n-1})V_N(y_1,\ldots,y_{n-1})\dd y_1\ldots \dd y_{n-1},
		\end{equation*}
		where 
		\begin{equation*}
			V_N(y_1,\ldots,y_{n-1})=\frac{1}{N} |\{x\in \Lambda:x+y_1\in \Lambda,\ldots, x+y_{n-1}\in \Lambda\}|,
		\end{equation*}
		which is the normalized measure of the intersection $\Lambda\cap (\Lambda-y_1)\cap \cdots \cap (\Lambda-y_{n-1})$. We can notice that $0\leq V_N\leq 1$. Moreover, $V_N$ converges pointwise to $1$. Therefore, by dominated convergence, we obtain \eqref{eq:claim int G}.

		Let $S\subset [N]$ be such that $|S|\leq p^*(\beta)$. Recall that 
		\begin{equation*}
			\Msf^0_{\infty} (S)=\frac{1}{\displaystyle{\left(\int_{\Lambda^2} e^{\beta \g_\lambda(x-y)}\dd x\dd y\right)^{|S|} }}\int_{(\Lambda^2)^{|S|} } \indic_{\mc{B}_S}\prod_{i,j\in S:i<j}e^{-\beta v_{ij}}\indic_{\mc{A}_{ij}}\prod_{i\in S}e^{\beta \g_\lambda(x_i-y_i)} \dd x_i \dd y_i.
		\end{equation*}
		Since $\beta>2$, the map $e^{\beta \g_\lambda}:\dR^2\to \dR$ is integrable. Therefore, by \eqref{eq:claim int G},
		\begin{equation}\label{eq:conv1}
			\int_{\Lambda^2} e^{\beta \g_\lambda(x-y)}\dd x\dd y=N \int_{\dR^2}e^{\beta \g_\lambda(y)}\dd y+o(N).
		\end{equation}
		Recall from \eqref{eq:introduceGS} that there exists a map $G_S:(\dR^2)^{2|S|-1}\to \dR$ such that
		\begin{equation*}
			\indic_{\mc{B}_S}\prod_{i,j\in S:i<j}e^{-\beta v_{ij}}\indic_{\mc{A}_{ij}}\prod_{i\in S}e^{\beta \g_\lambda(x_i-y_i)}=  G_{S}(x_2-x_1,\ldots,x_{|S|}-x_1,y_1-x_1,\ldots,y_{|S|}-x_1).
		\end{equation*}
		\begingroup\color{black}
Let $k=|S|$, $\rho_i=r_i\vee\lambda$ and $R=\max_i\rho_i$. We fix $x_1=0$, so that the change of variables \eqref{eq:g5}, with $i_0=1$, has Jacobian $1$ and there is no integral over $x_{i_0}$. After integration over the variables $z_{ij}$, each near edge contributes at most $C\min(\rho_i,\rho_j)^2$. By Lemma~\ref{lemma:ordering tree}, and since $\prod_i\rho_i^2/R^2\leq\prod_i\rho_i^{2-2/k}$ and $4-\beta-2/k<0$ for $k\leq p^*(\beta)$, we obtain
\[
\int |G_S|\leq C\int_{(\dR^2)^k}\frac{\prod_i\rho_i^2}{R^2}\prod_i\rho_i^{-\beta}\dd\vr_i
\leq C\left(\int_0^\infty r(r\vee\lambda)^{2-2/k-\beta}\dd r\right)^k
\leq C\lambda^{(4-\beta)k-2}.
\]
The configurations used in the proof of the lower bound of Lemma~\ref{lemma:lowerM+} also show that $\int G_S>0$ for every $\lambda>0$.
\endgroup Thus, by \eqref{eq:claim int G},
		\begin{multline}\label{eq:conv2}
			\int_{(\Lambda^2)^{|S|} } \indic_{\mc{B}_S}\prod_{i,j\in S:i<j}e^{-\beta v_{ij}}\indic_{\mc{A}_{ij}}\prod_{i\in S}e^{\beta \g_\lambda(x_i-y_i)} \dd x_i \dd y_i\\=N\int_{(\dR^2)^{2|S|-1}} G_S(y_1,\ldots,y_{2|S|-1})\dd y_1\ldots \dd y_{2|S|-1}+o(N).
		\end{multline}
		Combining \eqref{eq:conv1} and \eqref{eq:conv2} shows that 
		\begin{equation*}
			\lim_{N\to \infty} N^{|S|-1}\Msf_\infty^0(S)=\frac{\int_{(\dR^2)^{2|S|-1}} G_S(y_1,\ldots,y_{2|S|-1})\dd y_1\ldots \dd y_{2|S|-1}}{\Bigl(\int_{\dR^2}e^{\beta \g_\lambda(y)}\dd y\Bigr)^{|S|}}=\msf_{\beta,\lambda}(|S|).
		\end{equation*}
		
\begingroup\color{black}
Since the Lebesgue measure is invariant under the exchange of $x_i$ and $y_i$, the averages $\operatorname{Sym}$ do not change the numerators of the activities in $\Lambda$. As in the proof of Lemma~\ref{lemma:int dipole}, we split the edges into near and far edges, and we merge the parts joined by near edges, starting from the singletons for $G'_{S,E}$ and from the parts of $X$ for $G''_{X,E}$. The average over the reversals of the merged parts removes pointwise the odd graphs that are not Eulerian relative to these parts. For each merged part $S_1$, the factors internal to $S_1$ are bounded by $C\indic_{\mc B_{S_1}}$\edgreen{, by \eqref{eq:bond-bound}}. The average over the reversals of the original parts is an average of averages over the reversals of the merged parts, and the majorants below depend only on the radii and on the distances $d_{ij}$, which are invariant under every reversal.

We may therefore bound the absolute value of the averaged integrand as in the proof of Lemma~\ref{lemma:int dipole}. We fix $x_1=0$, so that the tree coordinates have Jacobian $1$. Then near edges contribute $C\min(\rho_i,\rho_j)^2$, even edges contribute at most $C\min(\rho_i,\rho_j)^2(1+\log(R/m))$, and the integrations along the odd trees contribute at most $C(1+\log(R/m))^{|T^a|}$, where $m=\min_i\rho_i$\edgreen{. A single odd edge $r_ir_j/d_{ij}^2$ is not integrable at infinity. The convergence comes from the tail factor $\min(1,R^2/\max_{e\in T^a}d_e^2)$ in \eqref{eq:prodE1l'}, which the ear decomposition of the Eulerian odd quotient produces}. By Lemma~\ref{lemma:ordering tree}, we obtain, for $G=G'_{S,E}$ or $G''_{X,E}$ and respectively $k=|S|$ or $|V_X|$,
\[
\int |G|\leq C\int_{(\dR^2)^k}\frac{\prod_i\rho_i^2}{R^2}
(1+\log(R/m))^{k-1}\prod_i\rho_i^{-\beta}\dd\vr_i.
\]
Let $0<\eta<\beta-4+2/k$. Since
\[
(1+\log(R/m))^{k-1}\leq C_{\eta,k}\prod_i(\rho_i/\lambda)^\eta,
\]
this integral factorizes as in the bound on $\int|G_S|$, and it is at most
\[
C\lambda^{-k\eta}\left(\int_0^\infty r(r\vee\lambda)^{2-2/k+\eta-\beta}\dd r\right)^k
\leq C\lambda^{(4-\beta)k-2}.
\]
Thus \eqref{eq:claim int G}, \eqref{eq:conv1} and the positivity of $\int G_S$ prove \eqref{eq:proof little mdip} and \eqref{def:Kbeta un}.
\endgroup
	\end{proof}
	
	\subsection{Kruskal's algorithm and Penrose resummation}\label{sub:Penrose}
	
	We have already established the bounded-cluster control of Proposition~\ref{prop:bounded lower}, and it remains only to prove the absolute convergence asserted in Proposition~\ref{prop:absolute lower}, which concerns large clusters. In \eqref{eq:stt}, we appealed to the crude bound
	\[
	|\mathsf{E}^X| \le 2^{\binom{|V_X|}{2}},
	\]
	which grows like $\exp(O(|V_X|^2))$. Since clusters can now be arbitrarily large, this estimate is too coarse to close the argument. We must therefore control the sum over graphs more delicately.

	We recall Kruskal's well-known algorithm used to produce a spanning tree on a general connected graph.

	\begin{definition}[Kruskal's algorithm]\label{def:Kruskal}
		Let $X$ be a subpartition of $[N]$. Consider the lexicographic order on the edges in $\mc{E}^\inter(X)$. For every $E\subset\mc{E}^\inter(X)$, we say that $(V_X,E)$ has a cycle relative to $X$ if the augmented graph $(V_X,E\cup \mc{E}^\intra(X))$ has a cycle that is not included in a single component of $X$.
		\begin{enumerate}
			\item Let $E\in \textcolor{black}{\mathsf{E}_0}^X$. We construct a spanning tree $T\in \textcolor{black}{\mathsf{T}_0}^X$ with $T\subset E$ by selecting the edges in $E$ in increasing order, discarding the edges that form a cycle relative to $X$ among the prior edges. This defines a map
			\begin{equation}\label{def:T}
				\mc{T}^\KA_X:\textcolor{black}{\mathsf{E}_0}^X\to \textcolor{black}{\mathsf{T}_0}^X.
			\end{equation}
			
			\item Let $T\in \textcolor{black}{\mathsf{T}_0}^X$. There exists a maximal $\mc{E}^\KA_{\max}(X,T)\in \textcolor{black}{\mathsf{E}_0}^X$ such that $\mc{T}^\KA_X(\mc{E}_{\max}^\KA(X,T))=T$. The set of edges $\mc{E}^\KA_{\max}(X,T)$ can be constructed as follows: select all edges in $T$. Then select the edges $e\in \mc{E}^\inter(X)\setminus T$ such that $e$ is larger than every edge in \textcolor{black}{the unique path in the quotient tree $(V_X,T)/X$ between the parts containing its endpoints, identifying each quotient edge with its original edge in $T$.}
		\end{enumerate}
	\end{definition}
	
	Contracting each part of $X$ to a vertex turns item~(1) into Kruskal's algorithm on the quotient multigraph, in which an edge is discarded exactly when the previously selected edges already join the parts containing its endpoints. Hence the output of item~(1) belongs to ${\mathsf{T}_0}^X$. Let $T\in{\mathsf{T}_0}^X$ and $H\in{\mathsf{E}_0}^X$. If $\mc{T}^\KA_X(H)=T$ and $e\in H\setminus T$, then the edges selected before $e$ belong to $T$ and join the parts containing the endpoints of $e$, so they contain the unique path in the quotient tree $(V_X,T)/X$ between these parts, and $e$ is larger than every edge of this path. Thus $H\subset\mc{E}^\KA_{\max}(X,T)$. Conversely, if $T\subset H\subset\mc{E}^\KA_{\max}(X,T)$, then induction along the order shows that the algorithm applied to~$H$ selects exactly the edges of~$T$. This proves item~(2) and the equivalence used in the proof of Lemma~\ref{lemma:penrose}.
	
	Below is the well-known Penrose lemma. See for instance \cite[Section 4]{Bauerschmidt2016FerromagneticSS}. It allows us to factor out the contribution of trees from a cluster-expansion series.
	
	\begin{lemma}[Penrose resummation]\label{lemma:penrose}
		Let $X$ be a subpartition of $[N]$. For every $ij\in \mc{E}^\inter(X)$, let $c_{ij}\in \dR$. We have
		\begin{equation}\label{eq:Kbis}
			\sum_{H\in \textcolor{black}{\mathsf{E}_0}^X}\prod_{ij\in H}c_{ij}=\sum_{T\in \textcolor{black}{\mathsf{T}_0}^X} \prod_{ij\in T}c_{ij}\prod_{ij\in \mc{E}_{\max}^\KA(X,T)\setminus T}(1+c_{ij}).
		\end{equation}
	\end{lemma}

	\begin{proof}
		One has
\begin{equation*}
\sum_{H\in \textcolor{black}{\mathsf{E}_0}^X}\prod_{ij\in H}c_{ij} =\sum_{T\in \textcolor{black}{\mathsf{T}_0}^X} \sum_{H:\mc{T}^\KA_X(H)=T} \prod_{ij\in H}c_{ij} =\sum_{T\in \textcolor{black}{\mathsf{T}_0}^X} \prod_{ij\in T}c_{ij}\sum_{H:\mc{T}^\KA_X(H)=T} \prod_{ij\in H\setminus T}c_{ij}.
\end{equation*}
		Note that $\mc{T}^\KA_X(H)=T$ is equivalent to $T\subset H$ and $H\setminus T\subset \mc{E}_{\max}^\KA(X,T)\setminus T$. Hence
\begin{equation*}
\sum_{H\in \textcolor{black}{\mathsf{E}_0}^X}\prod_{ij\in H}c_{ij} =\sum_{T\in \textcolor{black}{\mathsf{T}_0}^X} \prod_{ij\in T}c_{ij}\sum_{H'\subset \mc{E}_{\max}^\KA(X,T)\setminus T} \prod_{ij\in H'}c_{ij} =\sum_{T\in \textcolor{black}{\mathsf{T}_0}^X} \prod_{ij\in T}c_{ij}\prod_{ij\in \mc{E}^\KA_{\max}(X,T)\setminus T}(1+c_{ij}).
\end{equation*}
	\end{proof}
	
	A consequence of the Penrose resummation lemma is the following bound on the Ursell function recalled in Definition~\ref{def:UrsellI}. This bound is known as Rota's theorem and will be used in the proof of Proposition~\ref{prop:absolute lower} in Subsection~\ref{sub:abs lower}.
	
	\begin{remark}[Rota's theorem] Recall that, for every connected graph $G$, the Ursell function $\mathrm{I}(G)$ is defined by 
		\begin{equation*}
			\mathrm{I}(G)=\sum_{H\subset G}(-1)^{|E(H)|},
		\end{equation*}
		where the sum runs over all spanning graphs $H$ of $G$ (i.e., $H$ connected and with edges included in those of $G$). Recalling the notation from Definition~\ref{def:Kruskal}, \textcolor{black}{with $X=\{\{v\}:v\in V(G)\}$,} one can rewrite this as
		\begin{equation*}
			\mathrm{I}(G)=\sum_{T\subset G}\sum_{\textcolor{black}{H\subset(\mc{E}^\KA_{\max}(X,T)\setminus T)\cap E(G)}}(-1)^{|E(T)|+|E(H)|}=\sum_{T\subset G} (-1)^{ |E(T)|} \prod_{ij\in (\mc{E}^\KA_{\max}(X,T)\setminus T)\cap E(G)}(1-1),
		\end{equation*}
		where the sum runs over spanning trees $T$ of $G$. The above product is either $0$ if nonempty or $1$ if empty. Therefore
		\begin{equation}\label{eq:rota}
			|\mathrm{I}(G)|\leq \sum_{T\subset G}1.
		\end{equation}
	\end{remark}

\begingroup\color{black}
\begin{lemma}\label{lemma:bounded cluster tail}
Let $A$ be a nonempty finite set with $|A|\leq N$, and let $p\geq2$. Suppose that $K:\mc P(A)\to\dR$ vanishes unless $2\leq|X|\leq p$ and satisfies $|K(X)|\leq C_0(\lambda^2/N)^{|X|-1}$ on this range. There exists $C=C(C_0,p)$ such that, for $0<\lambda\leq1$ and every $n\geq1$,
\begin{equation}\label{eq:bounded cluster tail}
\frac{1}{Nn!}\sum_{\substack{X_1,\ldots,X_n\subset A\\\mathrm{connected}}}
\bigl|\mathrm I(G(X_1,\ldots,X_n))\bigr|\prod_{i=1}^n|K(X_i)|\leq(C\lambda^2)^n.
\end{equation}
\end{lemma}
\begin{proof}
Let $q=\max_{a\in A}\sum_{X\ni a}|K(X)|$. Since $|A|\leq N$ and $\lambda\leq1$, we have $q\leq C\lambda^2$. By \eqref{eq:rota}, we may bound $|\mathrm I(G(X_1,\ldots,X_n))|$ by the number of trees $T$ on $[n]$ with $X_i\cap X_j\neq\emptyset$ for every $ij\in T$, and then sum over the $X_i$ leaf by leaf. Since the parent of a leaf has at most $p$ elements, each leaf costs at most $pq$, and the root costs at most $Nq$. By Cayley's formula, the left-hand side of \eqref{eq:bounded cluster tail} is therefore at most
\[
\frac{n^{n-2}}{n!}q(pq)^{n-1}\leq(C\lambda^2)^n
\]
for $n\geq2$. For $n=1$, it is at most $q$.
\end{proof}
\endgroup


	We next apply the Penrose resummation lemma to replace the sum over graphs with integrable weights $-\indic_{\mc{B}_{ij}}$ and $b_{ij}^{\tilde{v}}$ by a sum over trees with weight $b_{ij}^\abs$.

	\begin{lemma}[Penrose resummation for integrable weights]\label{lemma:LXi}
		Let $X$ be a subpartition of $[N]$. Let $a_{ij}^{\abs}$ and $b_{ij}^{\abs}$ be as in Definition~\ref{def:abs ab}. 
		Then, there exists $C_1>0$ depending \textcolor{black}{only on $\beta$} such that
		\begin{equation}\label{eq:Kstep2''}
			|\Ksf_{\ve_0}^-(X)|\leq C_1^{|V_X|}\sum_{n=0}^\infty \frac{1}{n!}\sum_{\substack{X_1,\ldots,X_n\subset X\\ \mathrm{disjoint} }} \sum_{T^{b}\in \textcolor{black}{\mathsf{T}_0}^{\Coarse_X(X_1,\ldots,X_n)}}\dE_{\Psf_X^{-,\ve_0}}\left[\prod_{l=1}^n \mc{L}_{C_1}(X_l)\prod_{ij\in T^{b}}b_{ij}^{\abs} \right],
		\end{equation}
		where \textcolor{black}{$\mc L_{C_1}(X')=0$ for $|X'|\leq1$, and, for $|X'|\geq2$,}
		\begin{multline}\label{def:mcL}
			\mc{L}_{C_1}(X')\coloneqq \sum_{k=0}^{\infty}\frac{1}{k!}\sum_{\substack{X_1',\ldots,X_k'\subset X'\\ \mathrm{disjoint} } }\sum_{E_1\in \Eulc^{X_1'}}\ldots \sum_{E_k\in \Eulc^{X_k'}} \sum_{\tilde{T}^{b}\in \textcolor{black}{\mathsf{T}_0}^{\Coarse_{X'}(X_1',\ldots,X_k') } }\\ \prod_{ij\in E_1\cup \cdots \cup E_k}(C_1a_{ij}^{\abs})\prod_{ij\in \tilde{T}^{b}}b_{ij}^{\abs}\prod_{ij\in \mc{E}^\inter(X')} e^{C_1a_{ij}^{\abs}} \prod_{ij\in \mc{E}^\inter(X')}\indic_{\mc{B}_{ij}^c}.
		\end{multline}
		
	\end{lemma}

	\medskip
	
	\begin{proof}
\textcolor{black}{We may assume that $|X|\geq2$. Recall \eqref{eq:expE11}, and fix $n\geq0$ and disjoint $X_1,\ldots,X_n\subset X$.}

\paragraph{\bf{Step 1: summing over graphs for a given component}}
\begingroup\color{black}
For $X'\subset X$, set
\[
\mathcal U(X')\coloneqq
\sum_{E_1\in\Eul^{X'}}
\sum_{\substack{E_2:E_1\cup E_2\in\mathsf E^{X'}\\E_1\cap E_2=\emptyset}}
\prod_{ij\in E_1}a_{ij}^{\tilde v}\prod_{ij\in E_2}b_{ij}^{\tilde v}
\prod_{ij\in\mc E^\inter(X')}\indic_{\mc B_{ij}^c}.
\]
This vanishes for $|X'|\leq1$, so we suppose that $|X'|\geq2$. Resumming according to the nontrivial connected components of the odd graph gives
\begin{multline}\label{eq:FX}
\mathcal U(X')=\sum_{k=0}^\infty\frac1{k!}
\sum_{\substack{X_1',\ldots,X_k'\subset X'\\\mathrm{disjoint}}}
\sum_{E_1\in\Eulc^{X_1'},\ldots,E_k\in\Eulc^{X_k'}}
\prod_{ij\in\cup_l E_l}a_{ij}^{\tilde v}\\
\times\sum_{\substack{E'\subset\mc E^\inter(X')\\E'\cap(\cup_l E_l)=\emptyset\\E'\cup(\cup_l E_l)\in\mathsf E^{X'}}}
\prod_{ij\in E'}b_{ij}^{\tilde v}
\prod_{ij\in\mc E^\inter(X')}\indic_{\mc B_{ij}^c}.
\end{multline}
For fixed $X_1',\ldots,X_k'$ and $E_1,\ldots,E_k$, write $Y_0=\Coarse_{X'}(X_1',\ldots,X_k')$. The sum over $E'$ in \eqref{eq:FX} equals
\[
\prod_{ij\in\cup_l(\mc E^\inter(X_l')\setminus E_l)}(1+b_{ij}^{\tilde v})
\sum_{H\in\mathsf E_0^{Y_0}}\prod_{ij\in H}b_{ij}^{\tilde v}.
\]
On $\bigcap_{ij\in\mc E^\inter(X')}\mc B_{ij}^c$, we have $b_{ij}^{\tilde v}\geq0$, so Lemma~\ref{lemma:penrose} bounds the last display by
\[
\sum_{T^b\in\mathsf T_0^{Y_0}}\prod_{ij\in T^b}b_{ij}^{\tilde v}
\prod_{ij\in\mc E^\inter(X')}(1+b_{ij}^{\tilde v}).
\]
By Lemma~\ref{lemma:errortilde}, we have $|a_{ij}^{\tilde v}|\leq Ca_{ij}^{\abs}$ and $0\leq b_{ij}^{\tilde v}\leq Cb_{ij}^{\abs}\leq Ca_{ij}^{\abs}$ on the same event. Taking absolute values in \eqref{eq:FX}, we therefore obtain
\begin{equation}\label{eq:bound FX'}
|\mathcal U(X')|\leq C^{|V_{X'}|}\mc L_C(X')\,.
\end{equation}
\endgroup

\paragraph{\bf{Step 2: summing over the edges in $F$}}
	
	We now rewrite the term
	\begin{equation*}
		\sum_{F\in \textcolor{black}{\mathsf{E}_0}^{\Coarse_X(X_1,\ldots,X_n)}}\prod_{ij\in F}(-\indic_{\mc{B}_{ij}})
	\end{equation*} from \eqref{eq:expE11}.
	For brevity, let $Y_1\coloneqq \Coarse_X(X_1,\ldots,X_n)$.
	Arguing by resummation as in Step 1, we can write 
	\begin{equation*}
		\sum_{F\in \textcolor{black}{\mathsf{E}_0}^{Y_1}}\prod_{ij\in F}(-\indic_{\mc{B}_{ij}}) =\sum_{\tilde{T}^b\in \textcolor{black}{\mathsf{T}_0}^{Y_1}}\sum_{E'\subset \mc{E}^\KA_{\max}(Y_1,\tilde{T}^b)\setminus \tilde{T}^b } \prod_{ij\in \tilde{T}^b\cup E' }(-\indic_{\mc{B}_{ij}}).
	\end{equation*}
	Fix $\tilde{T}^b\in \textcolor{black}{\mathsf{T}_0}^{Y_1}$. We have 
	\begin{equation*}
		\sum_{E'\subset \mc{E}_{\max}^\KA(Y_1,\tilde{T}^b)\setminus \tilde{T}^b } \prod_{ij\in E' }(-\indic_{\mc{B}_{ij}})=\prod_{ij\in \mc{E}_{\max}^\KA(Y_1,\tilde{T}^b)\setminus \tilde{T}^b }(1-\indic_{\mc{B}_{ij}})=\prod_{ij\in \mc{E}_{\max}^\KA(Y_1,\tilde{T}^b)\setminus \tilde{T}^b}\indic_{\mc{B}_{ij}^c}.
	\end{equation*}
	Summing over $\tilde T^b$ and using the above displays,  we find
	\begin{equation}\label{eq:486}
		\left| \sum_{F\in \textcolor{black}{\mathsf{E}_0}^{\Coarse_X(X_1,\ldots,X_n)}}\prod_{ij\in F}(-\indic_{\mc{B}_{ij}})\right|\leq   \sum_{\tilde{T}^b\in \textcolor{black}{\mathsf{T}_0}^{Y_1}}\prod_{ij\in \tilde{T}^b}\indic_{\mc{B}_{ij}}.
	\end{equation}
	By Definition~\ref{def:abs ab}, $\indic_{\mc{B}_{ij}}\leq b_{ij}^{\abs}$. Hence,
	\begin{equation}\label{eq:bound mergeF}
		\left| \sum_{F\in \textcolor{black}{\mathsf{E}_0}^{\Coarse_X(X_1,\ldots,X_n)}}\prod_{ij\in F}(-\indic_{\mc{B}_{ij}})\right|\leq   \sum_{\tilde{T}^b\in \textcolor{black}{\mathsf{T}_0}^{Y_1}}\prod_{ij\in \tilde{T}^b}b_{ij}^{\abs}.
	\end{equation}
	
	\paragraph{\bf{Step 3: conclusion}}
	
	With the notation of \eqref{eq:FX}, we have 
	\begin{equation*}
		\Ksf_{\ve_0}^-(X)= \sum_{n=0}^\infty \frac{1}{n!}\sum_{\substack{X_1,\ldots,X_n\subset X\\ \mathrm{disjoint} }}\sum_{F\in \textcolor{black}{\mathsf{E}_0}^{\Coarse_X(X_1,\ldots,X_n)}} \dE_{\Psf_X^{-,\ve_0} }\left[ \mathcal U(X_1)\cdots \mathcal U(X_n)\prod_{ij\in F}(-\indic_{\mc{B}_{ij}})\right]. 
	\end{equation*}
	Therefore, inserting the bounds \eqref{eq:bound FX'} and \eqref{eq:bound mergeF} proves the result. 
\end{proof}

\subsection{Rewriting the sum over Eulerian graphs as a sum over trees}

Our aim is now to bound the term $\mc{L}_{C_1}(X)$ defined in \eqref{def:mcL}. We rewrite the sum over Eulerian graphs by incorporating the result of the peeling procedure of Definition~\ref{def:peeling lower bound}. Using the key inequality \eqref{eq:peeling} of Corollary~\ref{coro:prod a}, we bound the product of the weights by a quantity involving the product over the set of edges $F$ of the $g_{ ij}$'s defined in \eqref{def:gij}, and a product of squares of dipole lengths.

\begingroup\color{black}
For $i\neq j$, we set
\[
\bar g_{ij}\coloneqq\indic_{\mc B_{ij}^c}\indic_{d_{ij}\leq\textcolor{black}{17}\ve_0\Cut}
\frac{\max(r_i,r_j)^2}{\max(d_{ij},r_i,r_j,\lambda)^2}.
\]
\endgroup

\begin{lemma}[Summing over Eulerian graphs]\label{lemma:sum Euler}	
	Let $X$ be a subpartition of $[N]$ such that, for every $S\in X$, one has $|S|\leq p(\beta)$. \textcolor{black}{Suppose $0<r_i\leq\ve_0\Cut$ for every $i\in V_X$.} Let $X_1,\ldots,X_n\subset X$ be disjoint. Let $\mc{L}_{C_1}(X)$ be as in \eqref{def:mcL}. For every $l\in [n]$, let 
	\begin{equation}\label{def:hatSl}
		\hat{S}_l(X_1,\ldots,X_n)\coloneqq \mathrm{argmax}\{r_{S}: S \in X_l\},
	\end{equation}
	where $r_S$ is as in Definition~\ref{def:rS}.
	
	Then, there exists a constant $C>0$ depending on $\beta$, $p(\beta)$ and $M$ such that
	\begin{equation}\label{eq:Eul0}
		\mc{L}_{C_1}(X) \le e^{C|V_{X}|}\sum_{n=0}^{\infty}\frac{1}{n!}\sum_{\substack{X_1,\ldots,X_n\subset X\\ \mathrm{disjoint} } } H_{X}(X_1,\ldots,X_n),
	\end{equation}
	where
	\begin{multline}\label{def:H}
		H_{X}(X_1,\ldots,X_n)\coloneqq \prod_{l=1}^n\left(\prod_{\substack{S\in X_l:\\ S\neq \hat{S}_l(X_1,\ldots,X_n) }}r_{S}^2\sum_{T_l^a\in \mathsf{T}^{X_l}}\prod_{ij\in T_l^a}\frac{1}{d_{ij}^2}\indic_{d_{ij}\leq \textcolor{black}{17}\ve_0\Cut }\indic_{\mc{B}_{ij}^c} \right)\\ \times\left(\sum_{\substack{T^b\in \textcolor{black}{\mathsf{T}_0}^{\Coarse_{X}(X_1,\ldots,X_n)}}}\prod_{ij\in T^b}b_{ij}^{\abs}\right) \textcolor{black}{\left(\prod_{i\in V_X}\left(1+\log_M\frac{\textcolor{black}{17}\ve_0\Cut}{r_i}\right)^2\right)} \\ \times\exp\left(\textcolor{black}{2C_1}\sum_{ij\in \mc{E}^\inter(X)} a_{ij}^{\abs}\right)\left(\prod_{i\in V_{X}}\Big( 1+\sum_{j\in V_{X}:j\neq i}\textcolor{black}{\bar g_{ij}}\Big)\right).
	\end{multline}
\end{lemma}

\medskip

\begin{proof}
	\textcolor{black}{We work on $\bigcap_{ij\in\mc{E}^{\inter}(X)}\mc{B}_{ij}^c$, since otherwise $\mc{L}_{C_1}(X)=0$.}
	Fix $l\in [n]$. Recall $\mc{T}^{X_l}$ and $\mc{F}^{X_l}$ from Definition~\ref{def:peeling lower bound}. By Lemma~\ref{lemma:technical peeling}, item (1), 
	\begin{equation}\label{sommedesa}
		\sum_{E\in \Eulc^{X_l}}\prod_{ij\in E}(C_1a_{ij}^{\abs}) =\sum_{T_l^a\in \mathsf{T}^{X_l}} \sum_{F_l\in \PF(V_{X_l}) } \sum_{\substack{E\in \Eulc^{X_l}:\\ \mc{T}^{X_l}(\cdot,E)=T_l^a\\ \mc{F}^{X_l}(\cdot,E)={F}_l}}\prod_{ij\in E}(C_1a_{ij}^{\abs}),
	\end{equation}
	where we recall from Definition~\ref{def:pseudo forest} that $\PF(V_{X_l})$ stands for the set of pseudoforests on $V_{X_l}$. Fix $T_l^a\in \mathsf{T}^{X_l}$ and $F_l\in \PF(V_{X_l})$. Let $E\in\Eulc^{X_l}$ be such that $\mc{T}^{X_l}(\cdot,E)=T_l^a$ and $\mc{F}^{X_l}(\cdot,E)=F_l$. We have
	\begin{equation*}
		\prod_{ij\in E}(C_1a_{ij}^{\abs})=\prod_{ij\in T_l^a\cup F_l}(C_1a_{ij}^{\abs}) \prod_{ij\in E\setminus(T_l^a\cup F_l)}(C_1a_{ij}^{\abs}).
	\end{equation*}
	Since $T_l^a\cup F_l$ is minimally 2-edge-connected relative to $X_l$ (since $T_l^a\cup F_l=\Peeled_{X_l}(E)$, where $\Peeled_X(E)$ is as in Definition~\ref{def:peeling into minimal}), we have $|T_l^a\cup F_l|\leq 2(|V_{X_l}|-1)$. Hence, 
	\begin{equation*}
		\prod_{ij\in E}(C_1a_{ij}^{\abs})\leq C_1^{2(|V_{X_l}|-1)}\prod_{ij\in T_l^a\cup F_l}a_{ij}^{\abs} \prod_{ij\in E\setminus(T_l^a\cup F_l)}(C_1a_{ij}^{\abs}).
	\end{equation*}
	By  Corollary~\ref{coro:prod a}, there exists a constant $C>0$ depending on $\beta$ such that
	\begin{equation}\label{eq:paij}
		\prod_{ij\in T_l^a\cup F_l }a_{ij}^{\abs}\leq C^{|V_{X_l}|}\prod_{S\in X_l, S\neq \hat{S}_l}r_{S}^2 \sum_{T \in B(\cdot,T_l^a) } \prod_{ij\in T}\Bigl(\frac{1}{d_{ij}^2}\indic_{d_{ij}\leq \textcolor{black}{17}\ve_0 \Cut}\indic_{\mc{B}_{ij}^c} \Bigr)\prod_{ij\in F_l}g_{ ij},
	\end{equation}
	where $g_{ ij}$ is as in \eqref{def:gij} and $\hat{S}_l$ is as in \eqref{def:hatSl}. 
	\textcolor{black}{The weights $g_{ij}$ equal $\bar g_{ij}$, except possibly at $e_0$, where $g_{e_0}\leq1$. Removing $e_0$ from $F_l$ leaves a pseudoforest, obtained in this way from at most $1+\binom{|V_{X_l}|}{2}\leq2^{|V_{X_l}|}$ graphs $F_l$. We may thus replace $g_{ij}$ by $\bar g_{ij}$ in the sum over $F_l$, at the cost of a factor absorbed into $C^{|V_{X_l}|}$.}
	Therefore,
	\begin{multline*}
		\sum_{F_l\in \PF(V_{X_l}) } \sum_{\substack{E\in \Eulc^{X_l}:\\ \mc{T}^{X_l}(\cdot,E)=T_l^a \\\\ {\mc{F}}^{X_l}(\cdot,E)={F}_l }}\prod_{ij\in E}(C_1 a_{ij}^{\abs})\leq C^{|V_{X_l}|}\prod_{S\in X_l, S\neq \hat{S}_l}r_{S}^2 \sum_{T \in B(\cdot, T_l^a) } \prod_{ij\in T}\Bigl(\frac{1}{d_{ij}^2}\indic_{d_{ij}\leq \textcolor{black}{17}\ve_0 \Cut}\indic_{\mc{B}_{ij}^c} \Bigr) \\ \times \left(\sum_{F_l\in \PF(V_{X_l}) }\prod_{ij\in F_l}\textcolor{black}{\bar g_{ij}}\right)
		\sum_{E'\subset \mc{E}^\inter(X_l) }\prod_{ij\in E'}C_1a_{ij}^\abs.
	\end{multline*}
	Using $1+C_1a_{ij}^\abs\leq e^{C_1a_{ij}^\abs}$, we get 
	\begin{multline}\label{eq:sumg}
		\sum_{F_l\in \PF(V_{X_l}) } \sum_{\substack{E\in \Eulc^{X_l}:\\ \mc{T}^{X_l}(\cdot,E)=T_l^a \\\\ {\mc{F}}^{X_l}(\cdot,E)={F}_l }}\prod_{ij\in E}(C_1a_{ij}^{\abs})\leq C^{|V_{X_l}|}\prod_{S\in X_l, S\neq \hat{S}_l}r_{S}^2 \sum_{T \in B(\cdot, T_l^a) } \prod_{ij\in T}\Bigl(\frac{1}{d_{ij}^2}\indic_{d_{ij}\leq \textcolor{black}{17}\ve_0 \Cut}\indic_{\mc{B}_{ij}^c} \Bigr) \\ \times \left(\sum_{F_l\in \PF(V_{X_l}) }\prod_{ij\in F_l}\textcolor{black}{\bar g_{ij}}\right)
		\prod_{ij\in \mc{E}^\inter(X_l)}e^{C_1a_{ij}^\abs}.
	\end{multline}
	\textcolor{black}{As in the proof of \eqref{eq:sum PF}, with $\bar g_{ij}$ in place of $g_{ij}$, we have}
	\begin{equation}\label{eq:sumPF1}
		\sum_{F_l\in \PF(V_{X_l}) }\prod_{ij\in F_l}\textcolor{black}{\bar g_{ij}}\leq \prod_{i\in V_{X_l}}\left(1+\sum_{j\in V_{X_l}:j\neq i}\textcolor{black}{\bar g_{ij}}\right).
	\end{equation}
	Finally, using \eqref{eq:claimBT},
	\begin{align}\notag &
		\sum_{T_l^a\in \mathsf{T}^{X_l}} \sum_{T \in B(\cdot, T_l^a) } \prod_{ij\in T}\Bigl(\frac{1}{d_{ij}^2}\indic_{d_{ij}\leq \textcolor{black}{17}\ve_0 \Cut}\indic_{\mc{B}_{ij}^c} \Bigr)  \leq \sum_{T\in \mathsf{T}^{X_l}}|\{T'\in \mathsf{T}^{X_l}:T\in B(\cdot,T') \}|\prod_{ij\in T}\Bigl(\frac{1}{d_{ij}^2}\indic_{d_{ij}\leq \textcolor{black}{17}\ve_0 \Cut}\indic_{\mc{B}_{ij}^c} \Bigr)
		\\   \label{eq:change sum T} & \qquad \qquad \qquad  \leq C^{|V_{X_l}|}\textcolor{black}{\prod_{i\in V_{X_l}}\left(1+\log_M\frac{\textcolor{black}{17}\ve_0\Cut}{r_i}\right)^2}  \sum_{T\in \mathsf{T}^{X_l}} \prod_{ij\in T}\Bigl(\frac{1}{d_{ij}^2}\indic_{d_{ij}\leq \textcolor{black}{17}\ve_0 \Cut}\indic_{\mc{B}_{ij}^c} \Bigr).
	\end{align}
	Combining \eqref{eq:sumg}, \eqref{eq:sumPF1}, \eqref{eq:change sum T} and summing over $T_l^a$, we find that there exists $C>0$ depending only on  $\beta$ such that
	\begin{multline}\label{eq:sumEul}
		\sum_{E\in \Eulc^{X_l}}\prod_{ij\in E}(C_1a_{ij}^{\abs})  \leq C^{|V_{X_l}|}\prod_{S\in X_l:S\neq \hat{S}_l}r_{S}^2 \sum_{T_l^a\in \mathsf{T}^{X_l}} \Bigl(\prod_{ij\in T_l^a}\frac{1}{d_{ij}^2}\indic_{d_{ij}\leq \textcolor{black}{17}\ve_0\Cut}\indic_{\mc{B}_{ij}^c} \Bigr)\\ \times \prod_{ij\in \mc{E}^\inter(X_l)} e^{C_1a_{ij}^\abs }\prod_{i\in V_{X_l}}\Big(1+\sum_{j\in V_{X_l}:j\neq i}\textcolor{black}{\bar g_{ij}}\indic_{\mc{B}_{ij}^c}\Big)\textcolor{black}{\prod_{i\in V_{X_l}}\left(1+\log_M\frac{\textcolor{black}{17}\ve_0\Cut}{r_i}\right)^2}.
	\end{multline}
	\textcolor{black}{Taking the product of \eqref{eq:sumEul} over $l$, together with the exponential factor of \eqref{def:mcL}, gives the result.}
\end{proof}   

We now address the last two types of terms appearing in the function $H_X$ from \eqref{def:H}. This analysis will rely on the geometric properties of the configurations.

\begin{lemma}[Control of the interactions]\label{lemma:control sum intera} 
	Let $\beta\in (2,\infty)$ and $p(\beta)$ be as in Definition~\ref{def:pbeta}. Let $X$ be a \textcolor{black}{subpartition} of $[N]$ such that, for every $S\in X$, one has $|S|\leq p(\beta)$. Define the event
	\begin{equation*}
		\mc{D}=\bigcap_{i\in V_{X}}\{\textcolor{black}{0<}r_i\leq \ve_0\Cut\}\cap \bigcap_{ij\in \mc{E}^\inter(X)}\mc{B}_{ij}^c.
	\end{equation*}
	Then, there exists a constant $C>0$ depending on $p(\beta)$ such that on the event $\mc{D}$, for every $i\in V_X$,
	\begin{equation}\label{eq:claim int}
		\sum_{j\in V_{X}:ij\in \mc{E}^\inter(X)}a_{ij}^{\abs}\indic_{r_j\geq r_i} \leq \frac{C}{M}\left( \log \Bigl(\frac{\textcolor{black}{17} \ve_0\Cut}{r_i}\Bigr)+1\right).
	\end{equation}
	\begingroup\color{black}
	Moreover, on $\mc D$, after increasing $C$,
	\begin{equation}\label{eq:claim int2}
		\prod_{i\in V_X}\left(1+\sum_{j\in V_X:j\neq i}\bar g_{ij}\right)
		\leq(CM)^{|V_X|}\prod_{i\in V_X}\left(\frac{\ve_0\Cut}{r_i}\right)^{1/M}.
	\end{equation}
	\endgroup
\end{lemma}

\medskip

\begin{proof} 
	Let $p\coloneqq p(\beta)$.
	
	Fix $i\in V_{X}$. \textcolor{black}{To prove \eqref{eq:claim int}, we may assume $Mr_i<\textcolor{black}{17}\ve_0\Cut$, since otherwise its left-hand side vanishes.} Let us write
	\begin{equation*}
		\sum_{j\in V_{X}: ij\in \mc{E}^\inter(X) }a_{ij}^\abs\textcolor{black}{\indic_{r_j\geq r_i}}\leq C(A_1+A_2),
	\end{equation*}
	where
	\begin{equation}\label{B1}
		A_1\coloneqq \sum_{ j\in V_{X}, j\neq i}\frac{r_ir_j}{d_{ij}^2}\indic_{ r_j<d_{ij} \textcolor{black}{\leq}\textcolor{black}{17}\ve_0\Cut}\indic_{\mc{B}_{ij}^c}\indic_{r_j\geq r_i},
	\end{equation}
	\begin{equation}\label{B2}
		A_2\coloneqq \sum_{j\in V_{X}, j\neq i}\frac{r_i}{d_{ij}}\indic_{d_{ij}\leq  r_j}\indic_{\mc{B}_{ij}^c}\indic_{r_j\geq r_i}.
	\end{equation}

	Let us first consider the sum in \eqref{B1}. \textcolor{black}{On $\mc{B}_{ij}^c$, we have $d_{ij}>M\max(r_i,\lambda)$ when $r_j\geq r_i$.} Therefore,
	\begin{equation}
		A_1\leq \sum_{ j\in V_{X}, j\neq i}\frac{r_ir_j}{d_{ij}^2}\indic_{\max(r_j,Mr_i)<d_{ij} \textcolor{black}{\leq}\textcolor{black}{17}\ve_0\Cut}\indic_{r_j\geq r_i}.
	\end{equation}

	Let us introduce 
	\begin{equation}\label{cardN}\mc{E}_1(s,t)\coloneqq \{ j \in V_{X}:\textcolor{black}{j\neq i},\ r_i \leq r_j ,  s< r_j <2 s,\max(Mr_i,r_j)\leq d_{ij}\leq \textcolor{black}{17}\ve_0\Cut, t<d_{ij} < 2t\},\end{equation}
	and
	\begin{equation*}
		\mathcal{N}(s,t)\coloneqq |\mc{E}_1(s,t)|.
	\end{equation*}

	Note that, in view of the constraints $d_{ij} \in \textcolor{black}{[Mr_i,\textcolor{black}{17}\ve_0\Cut]}$ and   $ r_j \in [r_i, \ve_0\Cut]$, the function $\mathcal N(s,t)$ vanishes identically for $2t\le \max(M r_i,s)$, for $t \ge \textcolor{black}{17}\ve_0\Cut$, and for $s \leq \frac{r_i}{2}$.

	\textcolor{black}{Let $j_1,\ldots,j_q\in\mc E_1(s,t)$ belong to distinct blocks of $X$. In each dipole $j_k$, we pick an endpoint within distance $2t+r_i$ of $x_i$. On $\mc D$, these endpoints are separated by more than $Ms$. Comparing the areas of the disjoint balls of radius $Ms/2$ around them with that of $B(x_i,2t+r_i+Ms/2)$ gives $q(Ms/2)^2\leq(2t+r_i+Ms/2)^2$.}

	\textcolor{black}{Since every block of $X$ has cardinality at most $p$}, summing over the \textcolor{black}{blocks} gives
	\begin{equation} \label{cardN2}
		\mathcal{N}(s,t)\le C\left(\frac{t^2}{M^2s^2}+ 1\right) \indic_{ \frac{1}{2} \max(Mr_i,s)\le t \leq \textcolor{black}{17}\ve_0\Cut}\indic_{s\geq \tfrac{r_i}{2}}.
	\end{equation}
	We may now write 
	\begin{multline*}
		\sum_{\substack{j\in V_{X}\\\textcolor{black}{j\neq i}}}\frac{r_j}{d_{ij}^2}\indic_{\max(r_j,Mr_i)<d_{ij} \textcolor{black}{\leq}\textcolor{black}{17}\ve_0\Cut}\indic_{r_j\geq r_i}\\ \leq C\sum_{\substack{j\in V_{X}\\\textcolor{black}{j\neq i}}}  \indic_{\max(r_j,Mr_i)<d_{ij} \textcolor{black}{\leq}\textcolor{black}{17} \ve_0\Cut}\indic_{r_j\geq r_i}\iint \frac{s}{t^2}\frac{1}{ts}\indic_{(t,2t)}(d_{ij})\indic_{(s,2s)}(r_j)\indic_{ \frac{1}{2} \max(Mr_i,s)\le t \leq \textcolor{black}{17}\ve_0\Cut}\indic_{s\geq \tfrac{r_i}{2}}\dd s \dd t.
	\end{multline*}
	Therefore, inverting the order of summation and integration gives 
	\begin{equation*}
		\sum_{j\in V_{X}\textcolor{black}{:j\neq i}}\frac{r_j}{d_{ij}^2}\indic_{\max(r_j,Mr_i)<d_{ij} \textcolor{black}{\leq}\textcolor{black}{17}\ve_0\Cut}\indic_{r_j\geq r_i}\leq C\iint \frac{1}{t^3}\mc{N}(s,t)\indic_{ \frac{1}{2} \max(Mr_i,s)\le t \leq \textcolor{black}{17}\ve_0\Cut}\indic_{s\geq \tfrac{r_i}{2}}\dd s\dd t. 
	\end{equation*}
	Inserting the estimate \eqref{cardN2}, we obtain
	\begin{align*}
		\sum_{j\in V_{X}\textcolor{black}{:j\neq i}}\frac{r_j}{d_{ij}^2}\indic_{\max(r_j,Mr_i)<d_{ij} \textcolor{black}{\leq}\textcolor{black}{17}\ve_0\Cut}\indic_{r_j\geq r_i}& \le C  \int_{t=\tfrac{M}{2} r_i}^{\textcolor{black}{17} \ve_0\Cut}\int_{s= \tfrac{r_i}{2}}^{2t}\frac{1}{t^3} \left(\frac{t^2}{M^2s^2}+ 1\right)
		\dd s \dd t\\
		& \le C'\left( \frac{1}{M^2 r_i} \log \frac{\textcolor{black}{17}\ve_0\Cut}{r_i}+ \frac{1}{M r_i}\right) .\end{align*}
	Inserting this into \eqref{B1} yields
	\begin{equation}\label{bornB1}
		|A_1|\le  \frac{C}{M}\left( \log \frac{\textcolor{black}{17}\ve_0\Cut}{r_i}+1\right).\end{equation}

	We next turn to the right-hand side of \eqref{B2}. Notice that 
	\begin{equation*}
		A_2\textcolor{black}{\leq} \sum_{j\in V_{X}:j\neq i}\frac{r_i}{d_{ij}}\indic_{Mr_i\le d_{ij}\leq r_j }\indic_{r_j\geq r_i}.
	\end{equation*}
	Define 
	\begin{equation*}
		\textcolor{black}{\mc{E}_2(t)=\{j\in V_X:j\neq i,\ Mr_i\leq d_{ij}\leq r_j,\ t<d_{ij}<2t\}}\quad \text{and}\quad \mc{N}(t)\coloneqq |\mc{E}_2(t)|.
	\end{equation*}
	We observe that if $j_1,j_2\in \mc{E}_2(t)$, then, using $\dist(A,C)\leq \dist(A,B)+\diam(B)+\dist(B,C)$ and the fact that $M>20$, 
	\begin{equation*}
		\textcolor{black}{d_{j_1j_2}<4t+r_i\leq(4+2/M)t<M\min(r_{j_1},r_{j_2}).}
	\end{equation*}
	Therefore, $j_1$ and $j_2$ are in the same \textcolor{black}{block of $X$}. \textcolor{black}{Since every block of $X$ has cardinality at most $p$}, we deduce that $|\mc{E}_2(t)|=:\mc{N}(t)\leq p$.

	Applying the same reasoning as for $A_1$, we obtain
	\begin{equation}\label{eq:B2i}
		\begin{split}
			\sum_{j\in V_{X}\textcolor{black}{:j\neq i}}\frac{1}{d_{ij}}\indic_{Mr_i\leq d_{ij} \le r_j}\indic_{r_j\geq r_i}
			&\leq C\sum_{j\in V_{X}\textcolor{black}{:j\neq i}}\indic_{Mr_i\leq d_{ij} \le r_j}\indic_{r_j\ge r_i} \int_{\tfrac{M}{2}r_i}^{\textcolor{black}{17} \ve_0\Cut} \frac{1}{t^2}\indic_{d_{ij}\in (t,2t)} \dd t \\
			&\leq C\int_{\tfrac{M}{2}r_i}^{\textcolor{black}{17} \ve_0\Cut} \frac{1}{t^2}\mc{N}(t)\dd t\leq \frac{Cp}{Mr_i}.
		\end{split}
	\end{equation}
	Inserting this into \eqref{B2}, we obtain 
	$|A_2|\le \frac{C}M .$
	Combined with \eqref{bornB1}, this proves \eqref{eq:claim int}.

	\begingroup\color{black}
	We finally prove \eqref{eq:claim int2}. Set $R=\ve_0\Cut$ and $w_i=(r_i/R)^{1/M}$. Fix $i\in V_X$ and $0<s\leq R/2$. Among the indices $j\neq i$ with $r_j\geq r_i$ and $s<r_j\leq2s$, those with $d_{ij}<2s$ lie in one block, since every two of them, $j$ and $k$, satisfy $d_{jk}<6s<M\min(r_j,r_k)$. In each remaining annulus $t\leq d_{ij}<2t$, the argument that gives \eqref{cardN2} bounds their number by $C(1+t^2/(M^2s^2))$. Summing over dyadic annuli gives
	\[
		\sum_{\substack{j\in V_X:j\neq i\\r_j\geq r_i,\ s<r_j\leq2s}}\bar g_{ij}
		\leq C\left(1+\frac{\log(\textcolor{black}{17}R/s)}{M^2}\right).
	\]
	Taking $s=2^{-h-1}R$ and summing over $h\geq0$, we get
	\[
		\sum_{j\in V_X:j\neq i,\ r_j\geq r_i}w_j\bar g_{ij}
		\leq C\sum_{h=0}^\infty2^{-h/M}\left(1+\frac{h+1}{M^2}\right)\leq CM.
	\]
	By the symmetry of $\bar g_{ij}$, and since $w_i\leq w_j$ when $r_i\leq r_j$,
	\[
		\sum_{i\in V_X}w_i\left(1+\sum_{j\in V_X:j\neq i}\bar g_{ij}\right)
		\leq |V_X|+2\sum_{i\in V_X}\sum_{j\in V_X:j\neq i,\ r_j\geq r_i}w_j\bar g_{ij}
		\leq CM|V_X|.
	\]
	The arithmetic-geometric mean inequality proves \eqref{eq:claim int2}.
	\endgroup
	
\end{proof}

\subsection{Simplified bound on the activity}

Observe that both the tree $T^b$ in \eqref{eq:Kstep2''} and the trees $\tilde{T}^{b}$ in each $\mc{L}_{C_1}(X_l)$ carry the weight $b_{ij}^{\abs}$. Thus, using the bound on the interaction provided by Lemma~\ref{lemma:control sum intera} and the result of Lemma~\ref{lemma:sum Euler} gives a much simpler bound on the activity, by summing first over the connected components of the Eulerian graph.

\begin{lemma}\label{lemma:simplification}
	Let $\beta\in (2,\infty)$.
	Let $X$ be a subpartition of $[N]$ such that, for every $S\in X$, $|S|\leq p(\beta)$, with $p(\beta)$ as in Definition~\ref{def:pbeta}. There exist $C>0$, depending only on $\beta$ and $M$, and $C_0>0$, depending only on $\beta$, such that 
	\begin{multline*}
		|\Ksf_{\ve_0}^-(X)|\leq \frac{e^{C|V_X|}}{(N\lambda^{(2-\beta)})^{|V_X|}\Msf_{\ve_0}^0(X)}\\ \times\sum_{n=0}^\infty \frac{1}{n!}\sum_{\substack{X_1,\ldots,X_n\subset X\\ \mathrm{disjoint} }} \prod_{l=1}^n\sum_{T_l^a\in \mathsf{T}^{X_l}}\sum_{T^b \in \textcolor{black}{\mathsf{T}_0}^{\Coarse_X(X_1,\ldots,X_n)}} \textcolor{black}{J_{\mathrm{large},C_0}}(\cup_{l=1}^n T_l^a,T^b),
	\end{multline*}
	where, for every $T_1^a\in \mathsf{T}^{X_1},\ldots,T_n^a\in \mathsf{T}^{X_n}$ and $T^b\in \textcolor{black}{\mathsf{T}_0}^{\Coarse_X(X_1,\ldots,X_n)}$,
	\begin{multline}\label{def:IT'}
		\textcolor{black}{J_{\mathrm{large},C_0}}(\cup_{l=1}^n T_l^a,T^b)\coloneqq \int_{(\Lambda^2)^{|V_X|}} \prod_{ij\in T^b}b_{ij}^{\abs}\prod_{l=1}^n\Bigl(\prod_{S\in X_l, S\neq \hat{S}_l}r_{S}^2\Bigr)
		\prod_{ij\in \cup_{l=1}^nT_l^a }\Bigl(\frac{1}{d_{ij}^2}\indic_{ d_{ij}\leq \textcolor{black}{17}\ve_0\Cut}\indic_{\mc{B}_{ij}^c} \Bigr) \\ \times \prod_{S\in X}\indic_{\mc{B}_S} \prod_{i\in V_X}\Bigl(\Bigl(\frac{\ve_0\Cut}{r_i}\Bigr)^{\frac{C_0}{M}}e^{\beta \g_\lambda(x_i-y_i)}\indic_{|x_i-y_i|\leq \ve_0\Cut} \dd x_i \dd y_i\Bigr).
	\end{multline}
\end{lemma}

\begin{proof}
	By combining Lemmas~\ref{lemma:LXi}, \ref{lemma:sum Euler} and \ref{lemma:control sum intera}, and by \textcolor{black}{using \eqref{eq:claim int2} and $(1+\log_M(\textcolor{black}{17}t))^2\leq C_Mt^{1/M}$ for $t\geq1$, we find that there exist $C>0$, depending on $\beta$ and $M$, and $C_0>0$, depending only on $\beta$, such that}
	\begin{multline}\label{eq:K*}
		|\Ksf_{\ve_0}^-(X)|\leq C^{|V_X|}\sum_{n=0}^\infty \frac{1}{n!}\sum_{\substack{X_1,\ldots,X_n\subset X\\ \mathrm{disjoint} }} \sum_{T^{b}\in \textcolor{black}{\mathsf{T}_0}^{\Coarse_X(X_1,\ldots,X_n)}}\\ \dE_{\Psf_X^{-,\ve_0}}\left[\prod_{l=1}^n \mc{L}'(X_l)\prod_{ij\in T^{b}}b_{ij}^{\abs} \prod_{i\in V_X}\Bigl(\frac{\ve_0\Cut}{r_i}\Bigr)^{\textcolor{black}{\frac{C_0}{M}}} \right],
	\end{multline}
	where \textcolor{black}{$\mc L'(X')=0$ for $|X'|\leq1$, and, for $|X'|\geq2$,}
	\begin{multline*}
		\mc{L}'(X')\coloneqq \sum_{k=0}^{\infty}\frac{1}{k!}\sum_{\substack{X_1',\ldots,X_k'\subset X'\\ \mathrm{disjoint} } }\sum_{T_1^a\in \mathsf{T}^{X_1'}}\ldots \sum_{T_k^a\in \mathsf{T}^{X_k'}} \sum_{\tilde T^b\in \textcolor{black}{\mathsf{T}_0}^{\Coarse_{X'}(X_1',\ldots,X_k') } }\\ \prod_{l=1}^k \prod_{S\in X_l':S\neq \hat{S}_l}r_S^2 \prod_{ij\in T_1^a\cup\cdots\cup T_k^a }\frac{1}{d_{ij}^2}\indic_{d_{ij}\leq \textcolor{black}{17}\ve_0\Cut}\indic_{\mc{B}_{ij}^c} \prod_{ij\in \tilde T^b}b_{ij}^{\abs}.
	\end{multline*}
	\textcolor{black}{Fix the odd components and their tree $T^b$, obtained by merging the outer tree with the trees $\tilde T^b$. Since the outer edges of $T^b$ determine the groups, there are at most $2^{|T^b|}\leq2^{|X|-1}$ choices of groups. The factorials remove their orderings, and $\mc L^\prime$ vanishes on singleton groups. After absorbing $2^{|X|-1}$ into $C^{|V_X|}$, we obtain from \eqref{eq:K*}}
	\begin{multline*}
		|\Ksf_{\ve_0}^-(X)|\leq C^{|V_X|}\sum_{n=0}^\infty \frac{1}{n!}\sum_{\substack{X_1,\ldots,X_n\subset X\\ \mathrm{disjoint} }} \sum_{T_1^a\in \mathsf{T}^{X_1}}\ldots \sum_{T_n^a\in \mathsf{T}^{X_n}}\sum_{T^b\in \textcolor{black}{\mathsf{T}_0}^{\Coarse_{X}(X_1,\ldots,X_n)}}\\ \dE_{\Psf_X^{-,\ve_0}} \left[\prod_{ij\in T^b}b_{ij}^{\abs} \prod_{l=1}^n\prod_{S\in X_l, S\neq \hat{S}_l}r_{S}^2
		\prod_{ij\in \cup_{l=1}^nT_l^a }\Bigl(\frac{1}{d_{ij}^2}\indic_{d_{ij}\leq \textcolor{black}{17}\ve_0\Cut}\indic_{\mc{B}_{ij}^c} \Bigr) \prod_{i\in V_X}\Bigl(\frac{\ve_0\Cut}{r_i}\Bigr)^{\textcolor{black}{\frac{C_0}{M}}}\right].
	\end{multline*}
	Using the definition of $\Psf_X^{-,\ve_0}$, we find that there exist $C>0$, depending only on $\beta$ and $M$, and $C_0>0$, depending only on $\beta$, such that 
	\begin{multline*}
		|\Ksf_{\ve_0}^-(X)|\leq \frac{e^{C|V_X|}}{(N\lambda^{(2-\beta)})^{|V_X|}\Msf_{\ve_0}^-(X)}\\ \times\sum_{n=0}^\infty \frac{1}{n!}\sum_{\substack{X_1,\ldots,X_n\subset X\\ \mathrm{disjoint} }} \prod_{l=1}^n\sum_{T_l^a\in \mathsf{T}^{X_l}}\sum_{T^b \in \textcolor{black}{\mathsf{T}_0}^{\Coarse_X(X_1,\ldots,X_n)}} \textcolor{black}{J_{\mathrm{large},C_0}}(\cup_{l=1}^n T_l^a,T^b).
	\end{multline*}
	Finally, by Proposition~\ref{prop:bounded lower} applied to $\Msf_{\ve_0}^-$ and $\Msf_{\ve_0}^0$, there exists $C>0$ depending on $\beta$ and $M$ such that 
	\begin{equation*}
		\Msf_{\ve_0}^-(X)\geq C^{-|V_X|}\Msf_{\ve_0}^0(X).
	\end{equation*}
	Inserting this into the last display concludes the proof.
\end{proof}

In Appendix~\ref{sub:app unbounded}, we prove the following lemma.

\begin{lemma}[Integration and summation over trees]\label{lemma:integration summation}
	
	Let $\beta\in (2,\infty)$ and $p(\beta)$ be as in Definition~\ref{def:pbeta}. Let $X$ be a subpartition of $[N]$ such that, for every $S\in X$, $|S|\leq p(\beta)$. Let $X_1,\ldots,X_n\subset X$ be disjoint. For every $T_1^a\in \mathsf{T}^{X_1},\ldots,T_n^a\in \mathsf{T}^{X_n}$ and $T^b\in \textcolor{black}{\mathsf{T}_0}^{\Coarse_X(X_1,\ldots,X_n)}$, let $\textcolor{black}{J_{\mathrm{large},C_0}}(\cup_{l=1}^n T_l^a,T^b)$ be as in \eqref{def:IT'}. Suppose that 
	\begin{equation*}
		|V_X|> \begin{cases}
			\frac{4}{4-\beta}& \text{if $\beta\in (2,4)$}\\
			2p_0 & \text{if $\beta\geq 4$}.
		\end{cases}
	\end{equation*}
	Then, for $M$ large enough with respect to $\beta$ and $p(\beta)$, and $\lambda$ small enough with respect to $\beta, p(\beta)$, $M$ \textcolor{black}{and $\ve_0$}, there exists a constant $C>0$ depending on $\beta$, $M$ and $p(\beta)$ such that
	\begin{multline}\label{eq:sum IT s beta}
		\sum_{\substack{T_1^a\in \mathsf{T}^{X_1},\ldots,T_n^a\in \mathsf{T}^{X_n} \\ T^b\in \textcolor{black}{\mathsf{T}_0}^{\Coarse_X(X_1,\ldots,X_n)}}} \textcolor{black}{J_{\mathrm{large},C_0}}(\cup_{l=1}^n T_l^a,T^b)\\ \leq C^{|V_X|}N |\Coarse_X(X_1,\ldots,X_n)|^{|\Coarse_X(X_1,\ldots,X_n)|}\prod_{l=1}^n|X_l|^{|X_l|} \ve_0^{2\alpha(\beta)|V_X|-2 }\lambda^{(2-\beta)|V_X|}\delta_{\beta,\lambda},
	\end{multline}
	where $\alpha(\beta)$ is as in \eqref{def:alphabeta}.
\end{lemma}

\medskip

Let us briefly comment on the proof of Lemma~\ref{lemma:integration summation}. 
Suppose for simplicity that $\beta\in (2,4)$ and that $X$ is made of pure dipoles only. Integrating the distances $d_{ij}$ for $ij\in \cup_l T_l^a\cup T^b$ reduces the problem to controlling
\begin{equation*}
	\int \frac{\prod_{i\in V_X} r_i^2 }{\max_{j\in V_X}r_j^2}\textcolor{black}{\prod_{i\in V_X}}\Bigl(\frac{2\ve_0R_{\beta,\lambda}}{r_i}\Bigr)^{\frac{C_0}{M}} \prod_{i\in V_X}e^{\beta\g_\lambda(r_i)}\indic_{r_i\leq \ve_0R_{\beta,\lambda}}\dd \vr_i.
\end{equation*}
Switching to polar coordinates and denoting by $i_0$ the index of the largest $r_i$, we are led  to integrating
\begin{equation*}
	\prod_{i\in V_X:i\neq i_0}\Bigl(r_i^{3-\beta}\Bigl(\frac{\ve_0R_{\beta,\lambda}}{r_i}\Bigr)^{\frac{C_0}{M}}\indic_{r_i\leq r_{i_0}} \Bigr)r_{i_0}^{1-\beta}\Bigl(\frac{2\ve_0R_{\beta,\lambda}}{r_{i_0}}\Bigr)^{\frac{C}{M}}.
\end{equation*}
The point is that, since $\beta\in (2,4)$, we can take $M$ large enough so that $3-\beta-\frac{C_0}{M}>-1$, forcing each $r_i$ to concentrate at its upper limit, i.e., $\ve_0\Cut$. In turn, the ``many-body interaction term''
\begin{equation*}
	\prod_{i\in V_X}\frac{2\ve_0R_{\beta,\lambda}}{r_i}
\end{equation*}
disappears upon integration. This argument depends critically on the  fact that we expand around the multipole model. Otherwise, we could not make the exponent in front of that error arbitrarily small.

\subsection{Absolute convergence of the cluster series} \label{sub:abs lower}

We may now assemble all the above steps and sum over the number of multipoles of given cardinality. Let $\pi$ be a partition of $[N]$ and $X\subset \pi$ be a subpartition.
If $X$ contains only pure dipoles, i.e., is a subpartition into singletons, then, 
\begin{equation}\label{eq:denom}
	\Msf_{\ve_0}^0(X)=1.
\end{equation}
Combining Lemma~\ref{lemma:simplification}, \eqref{eq:sum IT s beta} and the above display yields \textcolor{black}{for $\beta\in(2,4)$ and $|V_X|>4/(4-\beta)$} 
\begin{equation*}
	|\Ksf_{\ve_0}^-(X)|\leq C^{|V_X|} |V_X|^{|V_X|} \frac{1}{N^{|V_X|-1}} \ve_0^{(4-\beta)|V_X|-2} \textcolor{black}{\delta_{\beta,\lambda}}.
\end{equation*}
Therefore, supposing for simplicity that $\beta \in (2,4)$, 
\begin{multline}\label{eq:sumXX}
	\sum_{X:|X|=|V_X|>\frac{4}{4-\beta} }|\Ksf_{\ve_0}^-(X)|=\sum_{k>\frac4{4-\beta}}\sum_{X\subset \pi:|X|=|V_X|=k}|\Ksf_{\ve_0}^-(X)|\\
	\leq \sum_{k>\frac4{4-\beta}}\Bigl(C^k k^k \frac{1}{N^{k-1} } \textcolor{black}{\ve_0^{(4-\beta)k-2}\delta_{\beta,\lambda}}\Bigr)\sum_{X\subset \pi:|X|=|V_X|=k}1.
\end{multline}
Moreover, 
\begin{equation*}
	\sum_{X\subset \pi:|X|=|V_X|=k}1\leq \binom{|\pi|}{k}\leq \frac{|\pi|^k}{k!}\leq \frac{N^k}{k!}. 
\end{equation*}
Inserting this into \eqref{eq:sumXX} and using Stirling's formula, we deduce that, for $\ve_0\in (0,1)$ small enough, the series on the right-hand side of \eqref{eq:sumXX} is absolutely convergent.

When $X$ is a general partition with $n_i'$ multipoles of cardinality $i$ for every $i=1,\ldots,p$, the lower bound \eqref{eq:denom} is replaced, in view of the estimate \eqref{eq:MX -low} of Proposition~\ref{prop:bounded lower}, by
\begin{equation*}
	\Msf_{\ve_0}^0(X)\geq \frac{1}{C^{|V_X|}}\prod_{i=1}^p \lambda^{2(i-1)n_i'}N^{|X|-|V_X|}.
\end{equation*}
\textcolor{black}{For $\beta\in(2,4)$ and $k=|V_X|>4/(4-\beta)$, the proof of Lemma~\ref{lem:cvgentseries} below gives}
\begin{equation*}
	\textcolor{black}{|\Ksf_{\ve_0}^-(X)|\leq C^k |X|^{|X|} \frac{\delta_{\beta,\lambda}}{N^{|X|-1}} \ve_0^{(4-\beta)k-2}\prod_{i=1}^p \lambda^{-2(i-1)n_i'}.}
\end{equation*}

By computing the number of ways to choose a subpartition $X$ of $\pi$ having $n_1'$ pure dipoles, $\ldots$, $n_p'$ $2p$-poles, and using the crucial assumption \textcolor{black}{\eqref{eq:assnk}}, we deduce from the last display that, for $\ve_0\in (0,1)$ small enough, \textcolor{black}{the sum over these large activities converges absolutely}. 

Let us now prove this completely.

\begin{lemma}\label{lem:cvgentseries}
	Let $\beta\in (2,\infty)$.  Let $\pi$ be a partition of $[N]$ such that, for every $S\in \pi$, one has $|S|\leq p(\beta)$, where $p(\beta)$ is as in Definition~\ref{def:pbeta}. For every $k\in [p(\beta)]$, let $n_k$ be the number of parts of $\pi$ of cardinality $k$. Assume that these $n_k$'s satisfy \eqref{eq:assnk}.
	
	Then, for $M$ large enough with respect to $\beta$ and $p(\beta)$, $\ve_0$ small enough with respect to $\beta,M$ and $p(\beta)$, and $\lambda$ small enough, there exists $C>0$ depending only on $\beta, M$, $p(\beta)$, and $\ve_0$ such that
	\begin{equation}\label{eq:bb12}
		\sum_{X\in \mc{P}(\pi):|X|>1}|\Ksf_{\ve_0}^-(X)|20^{|V_X|} \indic_{|V_X|>p(\beta)}\leq CN \delta_{\beta,\lambda},
	\end{equation}
	where $\delta_{\beta,\lambda}$ is as in \eqref{defdelta}.
	In particular, the series $\sum_{X\in \mc{P}(\pi)} \Ksf_{\ve_0}^-(X)$ is convergent.

	Moreover, for $\ve_0$ small enough, there exists $C>0$ depending on $\beta,M$, $p(\beta)$ and $\ve_0$ such that
	\begin{equation}\label{eq:anch lem all}
		\max_{S_0\in \pi}\sum_{X\in \mc{P}(\pi):S_0\in X,|X|>1 }|\Ksf_{\ve_0}^-(X)| 20^{|V_X|}\leq \textcolor{black}{C\left(\lambda^2+\lambda^{-2(p(\beta)-1)}\delta_{\beta,\lambda}\right).}
	\end{equation}
\end{lemma}

\begin{proof}Let $p\coloneqq p(\beta)$ and 
	\begin{equation}\label{eq:beta proof}
		q_0\coloneqq \begin{cases}
			\lfloor\frac{4}{4-\beta}\rfloor +1& \text{if $\beta\in (2,4)$}\\
			2p_0 & \text{if $\beta\geq 4$}.
		\end{cases}
	\end{equation}
	\paragraph{\bf{Step 1: summing over multipole sizes}}
	Recall that, for every subpartition $X$ of $[N]$ and every $i\geq 1$, $\#_i X$ stands for the number of blocks of cardinality $i$ in $X$.
	We begin by summing the cluster series \eqref{eq:bb12} according to the number of multipoles of cardinality $i$, for every $i\in [p]$. We have
	\begin{equation*}\begin{split}\sum_{X\in \mc{P}(\pi):|V_X|>p}|\Ksf_{\ve_0}^-(X)|&=\sum_{\substack{n_1',\ldots,n_p':\\ n_1'+2n_2'+\cdots+pn_p'>p}}\sum_{X\in \mc{P}(\pi):\forall i\in [p], \#_i X=n_i'}|\Ksf_{\ve_0}^-(X)|\\ &\leq \sum_{\substack{n_1',\ldots,n_p':\\ n_1'+2n_2'+\cdots+pn_p'>p}}|\{X\in \mc{P}(\pi):\forall i, \#_i X=n_i'\}|\max_{X\in \mc{P}(\pi):\forall i\in [p], \#_i X=n_i'}|\Ksf_{\ve_0}^-(X)|.
	\end{split}\end{equation*}
	Fix $n_1',\ldots,n_p'\geq 0$ such that $$k\coloneqq n_1'+2n_2'+\cdots+pn_p'>p.$$ Set $$k_0\coloneqq n_1'+\cdots+n_p'\leq k.$$  Since  there are $n_i$ multipoles of cardinality $i$ in $\pi$ and since $X$ contains $n_i'$ multipoles of cardinality $i$, we have $\binom{n_i}{n_i'}$ choices for each $i\in [p]$. In other words,
	\begin{equation*}
		|\{X\in \mc{P}(\pi):\forall i\in [p], \#_i X=n_i'\}|=\prod_{i=1}^p \binom{n_i}{n_i'}.
	\end{equation*}
	For every $i\in [p]$, set $\gamma_i\coloneqq \frac{n_i}{N}$ and $\gamma_i'\coloneqq \frac{n_i'}{k_0}$. 
	By the  standard inequality $ \binom{n}{m}\le n^m /m!$ and  Stirling's formula, we have 
	\begin{equation}\label{eq:binom} \binom{n}{m}\le e^{Cm}\frac{n^m}{m^m}.\end{equation}
	Hence,
	\begin{equation*}
		\prod_{i=1}^p \binom{n_i}{n_i'}  \leq e^{Ck}\prod_{i=1}^p \frac{n_i^{n_i'}}{(n_i')^{n_i'}}\leq e^{Ck}\Bigl(\frac{N}{k_0}\Bigr)^{k_0}\prod_{i=1}^p \Bigl(\frac{\gamma_i}{\gamma_i'}\Bigr)^{\gamma_i'k_0},
	\end{equation*}
	with the convention that $(\frac{\gamma_i}{\gamma_i'})^{\gamma_i'}=1$ if $\gamma_i'=0$. Therefore, 
	\begin{equation}\label{eq:sumoverX}
		\sum_{X\in \mc{P}(\pi):\forall i\in [p], \#_i X=n_i'}|\Ksf_{\ve_0}^-(X)|\leq e^{Ck} \Bigl(\frac{N}{k_0}\Bigr)^{k_0} \prod_{i=1}^p \Bigl(\frac{\gamma_i}{\gamma_i'}\Bigr)^{\gamma_i'k_0} \max_{X\in \mc{P}(\pi):\forall i\in [p], \#_i X=n_i'}|\Ksf_{\ve_0}^-(X)|.
	\end{equation}
	Fix $X\in \mc{P}(\pi)$ such that, for every $i\in [p]$, $\#_iX=n_i'$. Recall that in this setting  $k=|V_X|$ and $k_0=|X|$.

	\medskip
	
	\paragraph{\bf{Step 2: summing over $X_1,\ldots,X_n$}}  
	Suppose first that $k>q_0$, where $q_0$ is as in \eqref{eq:beta proof}.
	
	Recall from Lemma~\ref{lemma:simplification} that there exist $C_0>0$ depending on $\beta$ and $C>0$ depending only on $\beta$ and $M$ such that
	\begin{multline}\label{eq:all tog 2}
		|\Ksf_{\ve_0}^-(X)|\leq \frac{e^{C|V_X|}}{(N\lambda^{(2-\beta)})^{|V_X|}\Msf_{\ve_0}^0(X)}\\ \times \sum_{n=0}^\infty \frac{1}{n!}\sum_{\substack{X_1,\ldots,X_n\subset X\\ \mathrm{disjoint} } } \sum_{T_1^a\in \mathsf{T}^{X_1},\ldots,T_n^a\in \mathsf{T}^{X_n}}\sum_{T^b\in \textcolor{black}{\mathsf{T}_0}^{\Coarse_X(X_1,\ldots,X_n)}}\textcolor{black}{J_{\mathrm{large},C_0}}(\cup_{l=1}^n T_l^a,T^b),
	\end{multline}
	where $\textcolor{black}{J_{\mathrm{large},C_0}}$ is controlled in Lemma~\ref{lemma:integration summation}. 
	
	Observe that  $|\Coarse_X(X_1,\ldots,X_n)|=k_0-k_0'+n $. 
	We now sum \eqref{eq:all tog 2} according to the cardinality of $X_1 \cup \dots \cup X_n\subset X$, which we denote by $k_0'$ (notice $ n\le k_0'\le k_0$). Choosing the elements of $\pi$ that belong to $X_1 \cup \dots \cup X_n$ gives $\binom{k_0}{k_0'}$ choices.
	Then, we partition these $k_0'$ elements into $X_1 \cup \dots \cup X_n$ and sum according to the cardinalities $m_1, \dots, m_n$ of $X_1, \dots, X_n$. There are $\frac{(k_0')!}{m_1 !\dots m_n!}$ ways to choose the $X_i$'s accordingly. 
	
	Thus, in view of \eqref{eq:sum IT s beta}, there exists a constant $C>0$ depending only on $\beta$, $p(\beta)$ and $M$ such that
	\begin{flalign}
		&& \sum_{n=0}^\infty \frac{1}{n!} \sum_{\substack{X_1,\ldots,X_n\subset X\\ \mathrm{disjoint} }} \sum_{\substack{T_1^a\in \mathsf{T}^{X_1},\ldots,T_n^a\in \mathsf{T}^{X_n} \\ T^b\in \textcolor{black}{\mathsf{T}_0}^{\Coarse_X(X_1,\ldots,X_n)}}} \textcolor{black}{J_{\mathrm{large},C_0}}(\cup_{l=1}^n T_l^a,T^b) \leq C^k N \ve_0^{2\alpha(\beta)|V_X|-2 }\lambda^{(2-\beta)|V_X|}\delta_{\beta,\lambda} \notag\\
		&& \times \sum_{k_0'=0}^{k_0} \binom{k_0}{k_0'}\sum_{n=0}^{k_0'} \frac{1}{n!}\sum_{m_1+\cdots+ m_n=k_0',m_i\geq 1}\frac{(k_0')!}{m_1!\ldots m_n!}m_1^{m_1} \ldots m_n^{m_n} (k_0-k_0'+n)^{ k_0-k_0'+n}.\label{eq:i0''}
	\end{flalign}
	By Stirling's formula and since $m_1+\dots +m_n=k_0'\le k$, the factor
	$\frac{m_1^{m_1} \ldots m_n^{m_n} } {m_1!\ldots m_n!}$ can be absorbed into $e^{Ck}$.
	\textcolor{black}{The term $n=0$ contributes $k_0^{k_0}$. For $1\leq n\leq k_0'$,} by the stars and bars theorem and \eqref{eq:binom}, we have 
	\begin{equation*}
		\sum_{m_1+\cdots +m_n=k_0',m_i\geq 1}1=\binom{k_0'-1}{n-1}\le e^{Ck_0} \frac{(k_0')^n}{n^n} .
	\end{equation*}
	Therefore,
	\begin{multline*}
		\sum_{k_0'=\textcolor{black}{1}}^{k_0}  \binom{k_0}{k_0'}\sum_{n=\textcolor{black}{1}}^{k_0'}\frac{1}{n!} \sum_{m_1+\cdots +m_n=k_0',m_i\geq 1}\frac{(k_0')!}{m_1!\ldots m_n!}m_1^{m_1}\ldots m_n^{m_n}(k_0-k_0'+n)^{ k_0-k_0'+n}\\
		\le e^{Ck}  \sum_{k_0'=\textcolor{black}{1}}^{k_0}  \binom{k_0}{k_0'}\sum_{n=\textcolor{black}{1}}^{k_0'} \frac{1}{n!}(k_0')! (k_0-k_0'+n)^{ k_0-k_0'+n} \frac{(k_0')^n}{n^n}.
	\end{multline*}
	By Stirling's formula and since $k_0'\le k_0$ and $k_0-k_0'+n\le k_0$, 
	\begin{equation*}
		\frac{1}{n!}(k_0')! (k_0-k_0'+n)^{ k_0-k_0'+n}\frac{(k_0')^n}{n^n}\leq C^{k_0}\frac{1}{n^{2n}}k_0^{k_0+2n}.
	\end{equation*}
	It follows that 
	\begin{equation*}
		\sum_{n=\textcolor{black}{1}}^{k_0'}\frac{1}{n!}(k_0')! (k_0-k_0'+n)^{ k_0-k_0'+n}
		\frac{(k_0')^n}{n^n}
		\leq C^{k_0}k_0^{k_0}.
	\end{equation*}
	
	Thus
	\begin{multline}\label{eq:final sum}
		\sum_{k_0'=0}^{k_0}  \binom{k_0}{k_0'}\sum_{n=0}^{k_0'}\frac{1}{n!} \sum_{m_1+\cdots +m_n=k_0',m_i\geq 1}\frac{(k_0')!}{m_1!\ldots m_n!}m_1^{m_1}\ldots m_n^{m_n}(k_0-k_0'+n)^{ k_0-k_0'+n}\\ \leq  C^{k_0} k_0^{k_0}\sum_{k_0'=0}^{k_0}  \binom{k_0}{k_0'}=(2C)^{k_0} k_0^{k_0}.
	\end{multline}
	Inserting this into \eqref{eq:i0''}, and recalling that $k_0=|X|$, yields
	\begin{multline*}
		\sum_{n=0}^\infty \frac{1}{n!} \sum_{\substack{X_1,\ldots,X_n\subset X\\ \mathrm{disjoint} }} \sum_{\substack{T_1^a\in \mathsf{T}^{X_1},\ldots,T_n^a\in \mathsf{T}^{X_n} \\ T^b\in \textcolor{black}{\mathsf{T}_0}^{\Coarse_X(X_1,\ldots,X_n)}}} \textcolor{black}{J_{\mathrm{large},C_0}}(\cup_{l=1}^n T_l^a,T^b)\\ \leq C^k N \ve_0^{2\alpha(\beta)|V_X|-2 }\lambda^{(2-\beta)|V_X|}\delta_{\beta,\lambda}|X|^{|X|}.
	\end{multline*}
	Hence, by \eqref{eq:all tog 2}, if $k>q_0$, then there exists $C>0$ depending on $\beta$, $M$ and $p(\beta)$ such that
	\begin{equation}\label{eq:KwL b}
		|\Ksf_{\ve_0}^-(X)|\leq e^{Ck}\frac{N}{N^k\mathsf{M}^0_{\ve_0}(X) }\ve_0^{2\alpha(\beta) k -2} \delta_{\beta,\lambda}|X|^{|X|}.
	\end{equation}
	On the other hand, by the estimate \eqref{eq:MX -low} of Proposition~\ref{prop:bounded lower}, there exists $C>0$ depending on $\beta$, $M$ and $p(\beta)$ such that, for every $S\in X$, 
	\begin{equation*}
		\mathsf{M}_{\ve_0}^{0}(S)\geq \frac{1}{CN^{|S|-1}}\lambda^{2(|S|-1)}.
	\end{equation*}
	Therefore, there exists $C>0$ depending on $\beta$, $M$ and $p(\beta)$ such that
	\begin{equation*}
		\mathsf{M}_{\ve_0}^{0}(X)=\prod_{S\in X}\mathsf{M}_{\ve_0}^{0}(S)\geq \frac{1}{C^{|X|} N^{k-|X|}}\prod_{i=1}^p \lambda^{2(i-1)n_i'}.
	\end{equation*}
	Inserting this into \eqref{eq:KwL b} and using $k_0=|X|$, we find that there exists $C>0$ depending only on $\beta$, $p(\beta)$ and $M$ such that
	\begin{equation}\label{eq:KF large}
		|\Ksf_{\ve_0}^-(X)|\leq e^{Ck}\frac{N}{N^{k_0} }\ve_0^{2\alpha(\beta) k-2} \delta_{\beta,\lambda} k_0^{k_0} \prod_{i=1}^p \lambda^{-2(i-1)n_i'}.
	\end{equation}

	Next, let us treat the case of small clusters $k\in \{p+1,\ldots,q_0\}$. Then, by Proposition~\ref{prop:bounded lower}, there exists $C>0$ depending on $\beta,p(\beta),M$ and $\ve_0$ such that 
	\begin{equation*}
		|\Ksf_{\ve_0}^-(X)|\leq e^{Ck}\frac{N}{N^k\mathsf{M}_{\ve_0}^{0}(X) }
		\delta_{\beta,\lambda}.
	\end{equation*}	Therefore, there exists $C>0$ depending on $\beta$, $p(\beta)$, $M$ and $\ve_0$ such that
	\begin{equation}\label{eq:KF small}
		|\Ksf_{\ve_0}^-(X)|\leq e^{Ck}\frac{N}{N^{k_0} }\delta_{\beta,\lambda} k_0^{k_0} \prod_{i=1}^p \lambda^{-2(i-1)n_i'}.
	\end{equation}

	\paragraph{\bf{Step 3: proof of \eqref{eq:bb12}}}
	Inserting \eqref{eq:KF large} into \eqref{eq:sumoverX},  and using $n_i'= k_0 \gamma_i'$, we obtain that if $k>q_0$, then there exists $C>0$ depending on $\beta, M$ and $p(\beta)$ such that
	\begin{equation*}
		\sum_{X\in \mc{P}(\pi):\forall i, \#_i X=n_i'}|\Ksf_{\ve_0}^-(X)|20^{|V_X|}\leq e^{Ck} \ve_0^{2\alpha(\beta) k-2} N\delta_{\beta,\lambda}
		\prod_{i=1}^p \Bigl(\frac{\gamma_i}{\gamma_i'}\lambda^{-2(i-1)}\Bigr)^{\gamma_i'k_0}.
	\end{equation*}
	By assumption \eqref{eq:assnk}, for every $i=2,\ldots,p$, we have $\gamma_i\leq \ve_0^{-\alpha(\beta)} \lambda^{2(i-1)}$. Therefore,
	\begin{equation*}
		\ve_0^{2\alpha(\beta)k} \prod_{i=1}^p \Bigl(\frac{\gamma_i}{\gamma_i'}\lambda^{-2(i-1)}\Bigr)^{\gamma_i'k_0}\leq \ve_0^{\alpha(\beta)k}\prod_{i=1}^p (\gamma_i')^{-\gamma_i'k_0}.
	\end{equation*}
	Moreover, by the convexity of $x\mapsto x\log x$ and $\sum_i\gamma_i'=1$, we check that 
	$\prod_i(\gamma_i')^{-\gamma_i'k_0}\le C^{k_0}$.

	Thus, there exists $C>0$ depending on $\beta$, $p(\beta)$ and $M$ such that, for every $k>q_0$,
	\begin{equation}\label{eq:series lemma}
		\sum_{X\in \mc{P}(\pi):|V_X|=k}|\Ksf_{\ve_0}^-(X)|20^{|V_X|}\leq e^{Ck} \ve_0^{\alpha(\beta) k-2}N
		\delta_{\beta,\lambda}.
	\end{equation}
	Similarly, there exists $C>0$ depending on $\beta$, $p(\beta)$, $M$ and $\ve_0$ such that, for every $k\in \{p+1,\ldots,q_0\}$,
	\begin{equation}\label{eq:smalli}
		\sum_{X\in \mc{P}(\pi):|V_X|=k}|\Ksf_{\ve_0}^-(X)|20^{|V_X|}\leq e^{Ck} N
		\delta_{\beta,\lambda}.
	\end{equation}
	
	Choosing $\ve_0\in (0,1)$ small enough, we see that the series \eqref{eq:series lemma} in $k$ is convergent. Using this along with \eqref{eq:smalli} concludes the proof of \eqref{eq:bb12}.\\

	\paragraph{\bf{Step 4: proof of the anchored estimates}}
	Fix $S_0\in \textcolor{black}{\pi}$ and let $r\coloneqq |S_0|$.

	Let $n_1',\ldots,n_{p}'\geq 0$ with $n_r'\geq 1$ and set $k_0\coloneqq n_1'+\cdots+n_p'$. Suppose that $k\coloneqq n_1'+2n_2'+\cdots+pn_{p}'>q_0$. 
	The number of ways to choose $X$ with $n_i'$ multipoles of cardinality $i$ for every $i\in [p]$, and such that $S_0\in X$, is given by
	\begin{equation*}
		\binom{n_r-1}{n_r'-1}\prod_{i\in [p]:i\neq r }\binom{n_i}{n_i'}=\frac{n_r'}{n_r}\prod_{i=1}^p \binom{n_i}{n_i'}.
	\end{equation*}
	By Stirling's formula, there exists $C>0$ such that 
	\begin{equation*}
		\frac{n_r'}{n_r}\prod_{i=1}^p \binom{n_i}{n_i'}\leq C^{k_0}\Bigl(\frac{n_r}{n_r'}\Bigr)^{n_r'-1}\prod_{i\in [p]:i\neq r}\Bigl(\frac{n_i}{n_i'}\Bigr)^{n_i'}.
	\end{equation*}
	Combining \eqref{eq:KF large} and the above, we find that there exists $C>0$ depending on $\beta,p$ and $M$ such that 
	\begin{multline}\label{eq:ldis}
		\sum_{X\in \mc{P}(\pi):\forall i, \#_iX=n_i',S_0\in X}|\Ksf_{\ve_0}^-(X)|20^{|V_X|}\\\leq C^k\ve_0^{2\alpha(\beta)k-2} N^{1-k_0} k_0^{k_0}\Bigl(\frac{n_r}{n_r'\lambda^{2(r-1)} }\Bigr)^{n_r'-1}\frac{1}{\lambda^{2(r-1)}}\prod_{i\in [p]:i\neq r}\Bigl(\frac{n_i}{n_i' \lambda^{2(i-1)}}\Bigr)^{n_i'}
		\delta_{\beta,\lambda}.
	\end{multline}
	By \eqref{eq:assnk}, for every $i\in \{2,\ldots,p\}$,
	$n_i\leq N\ve_0^{-\alpha(\beta)}\lambda^{2(i-1)}.$
	Inserting this into \eqref{eq:ldis} and using that $n_1'+\cdots+n_p'=k_0$, we deduce that there exists $C>0$ depending on $\beta,p$ and $M$ such that 
	\begin{equation*}
		\sum_{X\in \mc{P}(\pi):\forall i, \#_iX=n_i',S_0\in X}|\Ksf_{\ve_0}^-(X)|20^{|V_X|}  \leq C^k \ve_0^{\alpha(\beta)k-2}N^{1-k_0} N^{k_0-1} \frac{k_0^{k_0}}{\prod_{i\in [p]}(n_i')^{n_i'}} \frac{1}{\lambda^{2(r-1)}} 
		\delta_{\beta,\lambda}.
	\end{equation*}
	Summing over the $n_i'$, we get 
	\begin{equation*}
		\sum_{X\in \mc{P}(\pi):|V_{X}|=k,S_0\in X}|\Ksf_{\ve_0}^-(X)|20^{|V_X|} \leq C^k \ve_0^{\alpha(\beta)k-2}\frac{1}{\lambda^{2(r-1)}} \delta_{\beta,\lambda}.
	\end{equation*}
	For $\ve_0$ small enough, the above series in $k$ is convergent. Therefore, there exists $C>0$ depending on $\beta,p$, $M$ and $\ve_0$ such that 
	\begin{equation}\label{eq:Large}
		\sum_{X\in \mc{P}(\pi):|V_{X}|>q_0,S_0\in X}|\Ksf_{\ve_0}^-(X)|20^{|V_X|} \leq C \frac{1}{\lambda^{2(r-1)}} \delta_{\beta,\lambda}.
	\end{equation}
	Using the estimate \eqref{eq:KF small} and proceeding similarly, we obtain that there exists $C>0$ depending on $\beta,p$, $M$ and $\ve_0$ such that 
	\begin{equation}\label{eq:Med}
		\sum_{X\in \mc{P}(\pi):|V_{X}|>p,S_0\in X}|\Ksf_{\ve_0}^-(X)|20^{|V_X|}\leq C \frac{1}{\lambda^{2(r-1)}} \delta_{\beta,\lambda}.
	\end{equation}
	Similarly, using the bounded cluster estimate of Proposition~\ref{prop:bounded lower} for clusters smaller than $p^*(\beta)$, we get that there exists $C>0$ depending on $\beta,p$, $M$ and $\ve_0$ such that 
	\begin{equation}\label{eq:Small}
		\sum_{X\in \mc{P}(\pi):|V_{X}|\leq p,S_0\in X,|X|>1}|\Ksf_{\ve_0}^-(X)|20^{|V_X|}\leq C \lambda^2.
	\end{equation}
\begingroup\color{black}
Combining \eqref{eq:Small}, \eqref{eq:Med} and \eqref{eq:Large}, and using $r\leq p$, we obtain
\[
\max_{S_0\in\pi}\sum_{X\in\mc P(\pi):S_0\in X,\,|X|>1}
|\Ksf_{\ve_0}^-(X)|20^{|V_X|}
\leq C\left(\lambda^2+\lambda^{-2(p-1)}\delta_{\beta,\lambda}\right).
\]
This proves \eqref{eq:anch lem all}. For $\beta\in(2,4)$, the second term of its right-hand side equals
\[
\lambda^{2\frac{\beta-2}{4-\beta}-2(p-1)}
\left(1+|\log\lambda|\indic_{\beta=\beta_{p+1}}\right),
\]
whose power of $\lambda$ is positive by \eqref{defpstar}. For $\beta\geq4$, it equals $\lambda^2$. Hence this right-hand side tends to zero as $\lambda\to0$.
\endgroup
\end{proof}

We next  prove the absolute convergence of the series \eqref{eq:bb2} using the Kotecký--Preiss \cite{KoteckyPreiss1986} anchored-norm criterion, following roughly the presentation of the proof of Theorem~4.9 in \cite{Bauerschmidt2016FerromagneticSS}.

\begin{proof}[Proof of Proposition~\ref{prop:absolute lower}]
	Let $p\coloneqq p(\beta)$.
	In view of the previous lemma, it remains to prove \eqref{eq:bb2}. \textcolor{black}{The term $n=1$ of this series is bounded by \eqref{eq:bb12}, and we now fix $n\geq2$.}
	
	Since $\Ksf_{\ve_0}^-(X)=0$ if $|X|\in \{0,1\}$, we have
	\begin{multline}\label{eq:rhs2}
		\sum_{\substack{X_1,\ldots,X_n\subset \pi \\\mathrm{connected}\\ \exists i, |V_{X_i}|>p }}|\Ksf_{\ve_0}^-(X_1)\cdots \Ksf_{\ve_0}^-(X_n) \mathrm{I}(G(X_1,\ldots,X_n))|\\ \leq n\sum_{\substack{X_1,\ldots,X_n\subset \pi \\ \forall i, |X_i|>1\\\mathrm{connected},  |V_{X_1}|>p }}|\Ksf_{\ve_0}^-(X_1)\cdots \Ksf_{\ve_0}^-(X_n) \mathrm{I}(G(X_1,\ldots,X_n))|.
	\end{multline}
	Then, by Rota's theorem \eqref{eq:rota}, for every connected graph $G$ on $[n]$, we have
	\begin{equation*}
		|\mathrm{I}(G)|\leq \sum_{T\subset G}1,
	\end{equation*}
	where the sum runs over every spanning tree of $G$. Fix a spanning tree $T$ of $G$. We have
	\begin{align*}
		&\sum_{\substack{X_1,\ldots,X_n\subset \pi \\\mathrm{connected}, |V_{X_1}|>p }}|\Ksf_{\ve_0}^-(X_1)\cdots \Ksf_{\ve_0}^-(X_n) \mathrm{I}(G(X_1,\ldots,X_n))|\\ \notag &\hspace{2cm} \leq \sum_{\substack{T \, \mathrm{ tree}\\ \mathrm{on} \, [n]}}\sum_{\substack{X_1,\ldots,X_n\subset \pi \\\mathrm{connected}\\ \forall i, |X_i|>1, |V_{X_1}|>p }} \indic_{T\subset G(X_1,\ldots,X_n)}|\Ksf_{\ve_0}^-(X_1)\cdots \Ksf_{\ve_0}^-(X_n)|\\ \notag
		&\hspace{2cm} =\sum_{\substack{T \, \mathrm{ tree}\\ \mathrm{on} \,  [n]}}\sum_{\substack{X_1,\ldots,X_n\subset \pi \\\mathrm{connected}\\ \forall i, |X_i|>1, |V_{X_1}|>p }}\prod_{ij\in T}\indic_{X_i\cap X_j\neq \emptyset } |\Ksf_{\ve_0}^-(X_1)\cdots \Ksf_{\ve_0}^-(X_n)|.
	\end{align*}
	We then sum according to the cardinality of the $X_i$'s:
	\begin{multline}\label{eq:rhs}
		\sum_{\substack{X_1,\ldots,X_n \\\mathrm{connected}\\|V_{X_1}|>p}}|\Ksf_{\ve_0}^-(X_1)\cdots \Ksf_{\ve_0}^-(X_n) \mathrm{I}(G(X_1,\ldots,X_n))|\\\leq  \sum_{\substack{T \ \mathrm{ tree}\\ \mathrm{on}\, [n] }}\sum_{k_1,\ldots,k_n\geq 2} \sum_{\substack{X_1,\ldots,X_n:\\ \forall i, |X_i|=k_i\\|V_{X_1}|>p }}\prod_{ij\in T}\indic_{X_i\cap X_j\neq \emptyset }|\Ksf_{\ve_0}^-(X_1)\cdots \Ksf_{\ve_0}^-(X_n)|.
	\end{multline}
	
	Let $T$ be a rooted tree on $[n]$ with root equal to $1$. Let $k_1,\ldots,k_n\geq 2$. Let $i_1\in \textcolor{black}{[n]\setminus\{1\}}$ be a leaf of $T$ and let $p(i_1)$ be the unique neighbor of $i_1$ in $T$. Since $i_1p(i_1)\in T$, we have $X_{i_1}\cap X_{p(i_1)}\neq \emptyset$. The only term on the right-hand side of \eqref{eq:rhs} that depends on $X_{i_1}$ is $\indic_{X_{i_1}\cap X_{p(i_1)}\neq \emptyset} |\Ksf_{\ve_0}^-(X_{i_1})|$. Summing it, we get 
	\begin{equation*}
		\sum_{X_{i_1}\in \mc{P}(\pi): |X_{i_1}|=k_{i_1}}\indic_{X_{i_1}\cap X_{p(i_1)}\neq \emptyset} |\Ksf_{\ve_0}^-(X_{i_1})|\leq \left(\sum_{S_{i_1}\in X_{p(i_1)}}1\right) \sup_{S_0\in \pi}\sum_{X\in \mc{P}(\pi):S_0\in X,|X|=k_{i_1}} |\Ksf_{\ve_0}^-(X)|.
	\end{equation*}
	Let us denote by $\Anch$ the anchored quantity introduced in Lemma~\ref{lem:cvgentseries}:
	\begin{align*}
		\Anch\coloneqq \sup_{S_0\in \pi}\sum_{X\in\mc{P}(\pi):S_0\in X,|X|>1}|\Ksf_{\ve_0}^-(X)|20^{|V_X|}.
	\end{align*}
	By the definition of $\Anch$, we have 
	\begin{equation*}
		\sup_{S_0\in \pi}\sum_{X:S_0\in X,|X|=k_{i_1}} |\Ksf_{\ve_0}^-(X)|\leq \Anch \times 20^{-k_{i_1}}.
	\end{equation*}
	Since there are at most $k_{p(i_1)}$ choices for an element in $X_{i_1}\cap X_{p(i_1)}$ (called the anchor at the vertex $i_1$), we get by combining the above displays that
	\begin{equation*}
		\sum_{X_{i_1}: |X_{i_1}|=k_{i_1}}\indic_{X_{i_1}\cap X_{p(i_1)}\neq \emptyset} |\Ksf_{\ve_0}^-(X_{i_1})|\leq \Anch \times k_{p(i_1)} 20^{-k_{i_1}}.  
	\end{equation*}
	
	We deduce that 
	\begin{multline*}
		\sum_{\substack{X_1,\ldots,X_n:\\ \forall i, |X_i|=k_i\\ |V_{X_1}|>p }}\prod_{ij\in T}\indic_{X_i\cap X_j\neq \emptyset }|\Ksf_{\ve_0}^-(X_1)\cdots \Ksf_{\ve_0}^-(X_n)|\\ \leq \Anch \times k_{p(i_1)} 20^{-k_{i_1}} \sum_{\substack{(X_i)_{i\in [n]\setminus \{i_1\}}:\\  \forall i, |X_i|=k_i\\ \textcolor{black}{|V_{X_1}|>p} }}\prod_{ij\in T\setminus\{i_1p(i_1)\} }\indic_{X_i\cap X_j\neq \emptyset }\prod_{i\in [n]\setminus\{i_1\}} |\Ksf_{\ve_0}^-(X_i)|.
	\end{multline*}
	
	\textcolor{black}{We then choose} a leaf \textcolor{black}{$i_2\neq1$} of the residual tree $T\setminus \{i_1p(i_1)\}$ and repeat the previous step. Iterating leaf removal, we eventually reduce the tree to its root (the vertex $1$). We obtain 
	\begin{multline*}
		\sum_{\substack{X_1,\ldots,X_n:\\  \forall i, |X_i|=k_i\\ |V_{X_1}|>p }}\prod_{ij\in T}\indic_{X_i\cap X_j\neq \emptyset }|\Ksf_{\ve_0}^-(X_1)\cdots \Ksf_{\ve_0}^-(X_n)|\\ \leq \Anch^{n-1} \prod_{i\in [n]} k_i^{\deg_T(i)}\prod_{i\in [n]}20^{-k_i} \sum_{X\in \mc{P}(\pi):|V_X|>p }|\Ksf_{\ve_0}^-(X)|20^{|V_X|}.
	\end{multline*}
	Inserting this into \eqref{eq:rhs2} gives 
	\begin{multline*}
		\sum_{\substack{X_1,\ldots,X_n \\\mathrm{connected}\\ \exists i, |V_{X_i}|>p } }|\Ksf_{\ve_0}^-(X_1)\cdots \Ksf_{\ve_0}^-(X_n) \mathrm{I}(G(X_1,\ldots,X_n))|\\ \leq   \textcolor{black}{n}\Anch^{n-1} \sum_{k_1,\ldots,k_n\geq 2}\prod_{i=1}^n 20^{-k_i} \sum_{\substack{T\ \mathrm{ tree}\\ \mathrm{on}\, [n] }} \prod_{i=1}^n k_i^{\deg_T(i)}\sum_{X\in \mc{P}(\pi):|V_X|>p }|\Ksf_{\ve_0}^-(X)|20^{|V_X|}.
	\end{multline*}
	By Cayley's formula, the number of trees with degree sequence $(d_1,\ldots,d_n)$ on a set of $n$ points is given by $\frac{(n-2)!}{\prod_i (d_i-1)!}$. Therefore, 
	\begin{equation*}
		\sum_{\substack{T\ \mathrm{ tree}\\ \mathrm{on}\,  [n] }} \prod_{i=1}^n k_i^{\deg_T(i)}\leq \sum_{d_1,\ldots,d_n}\frac{(n-2)!}{\prod_{i=1}^n(d_i-1)!}  \prod_{i=1}^n k_i^{d_i}\leq (n-2)! \prod_{i=1}^n k_i e^{k_i}.
	\end{equation*}
	Combining the last two displays yields
	\begin{multline*}
		\sum_{\substack{X_1,\ldots,X_n \\\mathrm{connected}\\ \exists i, |V_{X_i}|>p } }|\Ksf_{\ve_0}^-(X_1)\cdots \Ksf_{\ve_0}^-(X_n) \mathrm{I}(G(X_1,\ldots,X_n))|\\  \leq    nn!\Anch^{n-1} \sum_{X\in \mc{P}(\pi):|V_X|>p }|\Ksf_{\ve_0}^-(X)|20^{|V_X|} \left(\sum_{k=2}^\infty  e^{2k}20^{-k}\right)^{n}. 
	\end{multline*}
	Hence,
	\begin{multline}\label{eq:concXXn}
		\sum_{\substack{X_1,\ldots,X_n \\\mathrm{connected}\\ \exists i, |V_{X_i}|>p } }|\Ksf_{\ve_0}^-(X_1)\cdots \Ksf_{\ve_0}^-(X_n) \mathrm{I}(G(X_1,\ldots,X_n))|\\ \leq    n n!\Bigl(\frac{1}{20e^{-2}-1}\Bigr)^n\Anch^{n-1}\sum_{X\in \mc{P}(\pi):|V_X|>p }|\Ksf_{\ve_0}^-(X)|20^{|V_X|}.
	\end{multline}
	By \eqref{eq:anch lem all},
	\begin{equation}\label{eq:normK}
		\textcolor{black}{\Anch=O_{\beta,M,p,\ve_0}\left(\lambda^2+\lambda^{-2(p-1)}\delta_{\beta,\lambda}\right)=o_{\lambda\to0}(1).} 
	\end{equation}
	\textcolor{black}{Combining} \eqref{eq:concXXn} and \eqref{eq:normK}, we conclude that, for $\lambda$ small enough, there exists $C>0$ depending on $\beta, M,p$ and $\ve_0$ such that
	\begin{equation*}
		\sum_{n=1}^\infty \frac{1}{n!}	\sum_{\substack{X_1,\ldots,X_n \\\mathrm{connected}\\ \exists i, |V_{X_i}|>p } }|\Ksf_{\ve_0}^-(X_1)\cdots \Ksf_{\ve_0}^-(X_n) \mathrm{I}(G(X_1,\ldots,X_n))|\leq C\sum_{X\in \mc{P}(\pi):|V_X|>p }|\Ksf_{\ve_0}^-(X)|20^{|V_X|}.
	\end{equation*}
	Inserting \eqref{eq:bb12} concludes the proof of \eqref{eq:bb2}.
\textcolor{black}{Without the restriction $|V_{X_1}|>p$, the same argument proves that the full connected series converges absolutely, since
\[
\sum_{X\subset\pi}|\Ksf_{\ve_0}^-(X)|20^{|V_X|}
\leq|\pi|\Anch\leq N\Anch.
\]}
\end{proof}

\begingroup\color{black}
The proofs of Lemma~\ref{lem:cvgentseries} and Proposition~\ref{prop:absolute lower} apply to every activity that vanishes for $|X|\leq1$ and satisfies the majorant of Lemma~\ref{lemma:simplification} and the bound \eqref{eq:K- bound}, with constants that depend on the same parameters.
\endgroup

\subsection{\texorpdfstring{\edgreen{Equality of the multipole activities}}{Equality of the multipole activities}}\label{sub:simp mult}

\edgreen{The equality of the two multipole activities, Corollary~\ref{coro:equality mult}, is proved by induction and does not use M\"obius inversion. We nevertheless record M\"obius inversion here, since it is used in the proof of Lemma~\ref{lemma:auxiliary}.} Let us recall the standard M\"obius inversion on the partition lattice. See \cite[Proposition~3.7.1 and (3.37)]{StanleyEC1} for a modern exposition and \cite{Rota1964} for the original source.

\begin{lemma}[M\"obius inversion on the partition lattice]\label{lemma:mobius}
	Let $E$ be a finite set. Let $g:\mathbf{\Pi}(E) \to \dR$, where we recall that $\mathbf{\Pi}(E)$ stands for the set of partitions of $E$.
	\begin{enumerate}
		\item Define
		\begin{equation*}
			f:\pi \in \mathbf{\Pi}(E)\mapsto \sum_{\sigma\leq \pi}g(\sigma),
		\end{equation*}
		where the sum is over refinements $\sigma$ of $\pi$.
		Then, for every $\sigma\in \mathbf{\Pi}(E)$, we have
		\begin{equation*}
			g(\sigma)=\sum_{\pi\leq \sigma}\mu(\pi,\sigma)f(\pi),
		\end{equation*}   
		where $\mu$ is the M\"obius function of the partition lattice defined as follows: for each block $B\in \sigma$, let $k_B$ be the number of blocks of $\pi$ contained in $B$. If $\pi\leq \sigma$, then
		\begin{equation}\label{eq:Mobiuslattice}
			\mu(\pi,\sigma)=\prod_{B\in \sigma }(-1)^{k_B-1}(k_B-1)!.
		\end{equation}
		If $\pi \nleq \sigma$, then $\mu(\pi,\sigma)=0$.
		\item Define
		\begin{equation*}
			f:\pi \in \mathbf{\Pi}(E)\mapsto \sum_{\sigma\geq  \pi}g(\sigma),
		\end{equation*}
		where the sum is over coarsenings $\sigma$ of $\pi$. Then, for every $\sigma\in \mathbf{\Pi}(E)$, we have
		\begin{equation}\label{eq:otherdirection}
			g(\sigma)=\sum_{\pi\geq \sigma}\mu(\sigma,\pi)f(\pi),
		\end{equation}
		where $\mu$ is as in \eqref{eq:Mobiuslattice}. 
	\end{enumerate}
\end{lemma}

\textcolor{black}{We equate the activity} $\Ksf_{\beta,\lambda}^{\mult}$ from Definition~\ref{def:inter activity} with the activity $\Ksf^{0}_\infty(X)$ of Definition~\ref{def:activity lower}.

\begin{coro}\label{coro:equality mult}
	For every subpartition $X$ of $[N]$, we have  
	\begin{equation*}
		\Ksf_{\beta,\lambda}^{\mult}(X)=\Ksf_{\infty}^0(X).
	\end{equation*}
\end{coro}

\medskip

\begin{proof}
\begingroup\color{black}
Fix a subpartition $X$. Since $\mc B_{ij}^c\subset\mc A_{ij}$, the expansion in connected components of Lemma~\ref{lemma:start low}, with $\tilde v$ replaced by $v$ and $\ve_0=\infty$, gives
\begin{equation*}
\dE_{\Psf_X^{0,\infty}}\left[\prod_{ij\in\mc E^\inter(X)}e^{-\beta v_{ij}}\indic_{\mc B_{ij}^c}\right] =\sum_{n=0}^\infty\frac1{n!}\sum_{\substack{X_1,\ldots,X_n\subset X\\\mathrm{disjoint}}}\prod_{l=1}^n\Ksf_\infty^0(X_l) =\sum_{n=0}^\infty\frac1{n!}\sum_{\substack{X_1,\ldots,X_n\subset X\\\mathrm{disjoint}}}\prod_{l=1}^n\Ksf_{\beta,\lambda}^{\mult}(X_l).
\end{equation*}
The expansion in $\Ksf^{\mult}_{\beta,\lambda}$ on the right follows by grouping the edges with weight $e^{-\beta v_{ij}}\indic_{\mc B_{ij}^c}-1$ into connected components. Since both activities vanish for $|X|\leq1$, an induction on $|X|$ then gives $\Ksf^{\mult}_{\beta,\lambda}(X)=\Ksf^0_\infty(X)$.
\endgroup
\end{proof}

\subsection{Proof of the LDP lower bound}\label{sub:proof low}
We conclude the proof of the lower bound.

\medskip

\begin{prop}[Large-deviation lower bound]\label{prop:lower bound}
	Let $\beta\in (2,\infty)$ and let $p(\beta)$ be as in Definition~\ref{def:pbeta}. Let $\mc{Z}_\beta$ be as in \eqref{def:Zbeta} and $\mc{I}_{\beta,p(\beta),\lambda}$ be as in Definition~\ref{def:rate function}.
	There exists $M_0>20$ such that, for every $M\geq M_0$, there exists $\ve_1\in(0,1]$ with the following property. For every $\ve_0\in(0,\ve_1)$, there exist $C>0$ and $\lambda_0\in(0,1]$ such that, for every $\lambda\in(0,\lambda_0)$, there exists $N_0(\lambda)\in\mathbb N^*$ for which the estimate below holds for every $N\geq N_0(\lambda)$ and every profile $n_1,\ldots,n_{p(\beta)}\in\mathbb N$ satisfying
	\[
	n_1+2n_2+\cdots+p(\beta)n_{p(\beta)}=N
	\]
	and \eqref{eq:assnk}. The constants and $N_0(\lambda)$ are uniform over these profiles.
	
	\textcolor{black}{For $i\in[p(\beta)]$, let $\gamma_i=n_i/N$.} For every $k\in \{1,\ldots,p(\beta)\},$ let $\mc{N}_k$ be the number of $2k$-poles in $[N]$. The event $\mc{A}$ being as in \eqref{defeventA}, we have
	\begin{multline}\label{eq:lowerA}
		\log \int_{\mc{A}}e^{-\beta \F_\lambda}\dd \XN \dd \YN\geq N\log N+ \log (N!)+(2-\beta)N\log \lambda+N\log \mc{Z}_\beta-N \mc{I}_{\beta,p(\beta),\lambda}(\gamma_1,\ldots,\gamma_{p(\beta)})\\-CN\delta_{\beta,\lambda}.
	\end{multline}
\end{prop}

\medskip

\begin{proof}
	Let $p\coloneqq p(\beta)$.
	
	\paragraph{\bf{Step 1: expanding the multipole terms}}
	First, since labels do not matter,
	\begin{equation*}
		\int_{\mc{A}}e^{-\beta \F_\lambda}\dd \XN \dd \YN=N! \int_{\mc{A}}\indic_{\sigma_N=\Id}e^{-\beta \F_\lambda}\dd \XN \dd \YN.
	\end{equation*}
	By Lemma~\ref{lemma:start low}, we have
	\begin{equation}\label{eq:comb}
		\int_{\mc{A}} \indic_{\sigma_N=\Id}e^{-\beta \F_\lambda}\geq\frac{N!}{1^{n_1}(2!)^{n_2}\cdots (p!)^{n_p}n_1!\cdots n_p!} (C_{\beta,\lambda, \ve_0}N)^N \Msf_{\ve_0}^-(\pi)\sum_{n=0}^{\infty} \frac{1}{n!}\sum_{\substack{X_1,\ldots,X_n\\ \mathrm{disjoint}}}\Ksf_{\ve_0}^-(X_1)\cdots \Ksf_{\ve_0}^-(X_n).
	\end{equation}
	\textcolor{black}{Let $h=\ve_0\Cut$. By the inequality $1-e^{-t}\leq t$ and \eqref{eq:Clambda bound}, we have
\[
0\leq1-\Msf_{\ve_0}^-(\{i\})
\leq C\lambda^{\beta-2}h^{-2}\int_0^h r^3(r\vee\lambda)^{-\beta}\dd r
\leq C_{\beta,p,\ve_0}\Cut^{-2}.
\]
Thus $\log\Msf_{\ve_0}^-(\{i\})=O_{\beta,p,\ve_0}(\Cut^{-2})$.}
By Proposition~\ref{prop:expansion -} and Lemma~\ref{lemma:limiting}, \textcolor{black}{together with the last display, for every $S\in\pi$ and for $N$ sufficiently large at fixed $\lambda$,}
	\begin{equation*}
		\Msf_{\ve_0}^{-}(S)=\frac{1}{N^{|S|-1}}\Bigl(\msf_{\beta,\lambda}(|S|)+O_{\beta,M,p,\ve_0}(\Cut^{-2})\Bigr).
	\end{equation*}
	Using that $\msf_{\beta,\lambda}(1)=1$ \textcolor{black}{and keeping the errors of the singletons}, we find that
	\begin{equation*}
		\prod_{S\in \pi}\Msf_{\ve_0}^{-}(S)=\frac{N^{n_1+\cdots+n_{p}}}{N^{N}}\prod_{k=2}^{p}\msf_{\beta,\lambda}(k)^{n_k}\textcolor{black}{\prod_{k=1}^p}e^{\frac{n_k}{\msf_{\beta,\lambda}(k) }O_{\beta,M,p,\ve_0}(\Cut^{-2}) } .
	\end{equation*}
	By the lower bound on the multipole partition function of \textcolor{black}{Proposition~\ref{prop:bounded lower}} and the fact that $C_{\beta,\lambda}\leq C\lambda^{2-\beta}$, we get that, for every $k\in \{1,\ldots,p\}$, there exists a constant $C>0$ depending on $\beta$, $M$ and $p$ such that 
	\begin{equation*}
		\msf_{\beta,\lambda}(k)\geq \frac{1}{C}\lambda^{2(k-1)}.
	\end{equation*}
	Therefore, by the assumption \eqref{eq:assnk}, 
	\begin{equation}\label{eq:tZ}
		\prod_{S\in \pi}\Msf_{\ve_0}^{-}(S)=\frac{N^{n_1+\cdots+n_{p}}}{N^{N}}\left(\prod_{k=2}^{p}\msf_{\beta,\lambda}(k)^{n_k}\right)e^{O_{\beta,M,p,\ve_0}(N\Cut^{-2}) } .
	\end{equation}
	Combining \eqref{eq:comb} and \eqref{eq:tZ}, we obtain
	\begin{multline*}
		\frac{1}{(NC_{\beta,\lambda,\ve_0})^N}\int_{\mc{A}}\indic_{\sigma_N=\Id} e^{-\beta \F_\lambda }\geq\frac{N!}{1^{n_1}(2!)^{n_2}\cdots (p!)^{n_{p}}n_1!\cdots n_{p}!} \frac{N^{n_1+\cdots+n_{p}}}{N^{N}}\prod_{k=2}^{p}\msf_{\beta,\lambda}(k)^{n_k}\\ \times\sum_{n=0}^\infty \frac1{n!}\sum_{\substack{X_1,\ldots,X_n\in \mc{P}(\pi)\\ \mathrm{disjoint}}}\Ksf_{\ve_0}^-(X_1)\cdots \Ksf_{\ve_0}^-(X_n) e^{O_{\beta,M,p,\ve_0}(N\Cut^{-2})}.
	\end{multline*}
	Therefore, by Stirling's formula,
	\begin{multline}\label{eq:asym}
		\frac{1}{(NC_{\beta,\lambda,\ve_0})^N}\int_{\mc{A}}\indic_{\sigma_N=\Id} e^{-\beta \F_\lambda }\geq
		e^{-N}\prod_{k=1}^p \Bigl(\frac{Ne}{k!n_k}\Bigr)^{n_k}\prod_{k=2}^{p}\msf_{\beta,\lambda}(k)^{n_k}\\ \times  \sum_{n=0}^\infty \frac1{n!}\sum_{\substack{X_1,\ldots,X_n\in \mc{P}(\pi)\\ \mathrm{disjoint}}}\Ksf_{\ve_0}^-(X_1)\cdots \Ksf_{\ve_0}^-(X_n) e^{O_{\beta,M,p,\ve_0}(N\Cut^{-2})}.
	\end{multline}
	In view of \eqref{devtCbeta} and of Definition~\ref{def:rate function}, we will have proved the result once we show that 
	\begin{multline}\label{claimprop446}
		\log \sum_{n=0}^{+\infty}\frac{1}{n!}\sum_{\substack{X_1,\ldots,X_n\in \mc{P}(\pi)\\ \mathrm{disjoint}}}\Ksf_{\ve_0}^-(X_1)\cdots \Ksf_{\ve_0}^-(X_n)
		\\
		=
		-N\sum_{n=1}^{\infty}\frac{(-1)^n}{n!}\sum_{m_1,\ldots,m_p\in \mathbb{N} }\prod_{i=1}^p \frac{\gamma_i^{m_i}}{m_i!}\sum_{\substack{(X_1,\ldots,X_n)\in \Htrees_n(Y(m_1,\ldots,m_p))\\ X_1\cup \cdots \cup X_n=Y(m_1,\ldots,m_p)}}\textcolor{black}{\omega(\mathbf X)}\textcolor{black}{\prod_{i=1}^n\ksf_{\beta,\lambda}^{\mult}((\#_kX_i)_{k=1}^p)\indic_{|V_{X_i}|\leq p}} 
		\\+O_{\beta,M,p,\ve_0}\left(N\delta_{\beta,\lambda}\right),
	\end{multline} with $\Htrees_n$ defined in  \eqref{def:Htrees} and $Y(m_1,\ldots,m_p)$ as in Definition~\ref{def:canonical}.
	
	\paragraph{\bf{Step 2: expansion of the perturbative term}}

	By Proposition~\ref{prop:absolute lower}, the series 
	\begin{equation*}
		\sum_{n\geq 1}\frac{1}{n!}\sum_{\substack{X_1,\ldots,X_n\in \mc{P}(\pi)\\ \mathrm{connected}}}\Ksf_{\ve_0}^-(X_1)\cdots \Ksf_{\ve_0}^-(X_n)\mathrm{I}(G(X_1,\ldots,X_n))
	\end{equation*}
	is absolutely convergent, which allows one to write, by Lemma~\ref{lemma:resum2},
	\begin{multline*}
		\log \sum_{n=0}^{+\infty}\frac{1}{n!}\sum_{\substack{X_1,\ldots,X_n\in \mc{P}(\pi)\\ \mathrm{disjoint}}}\Ksf_{\ve_0}^-(X_1)\cdots \Ksf_{\ve_0}^-(X_n)\\=\sum_{n=1}^{+\infty}\frac{1}{n!}\sum_{\substack{X_1,\ldots,X_n\in \mc{P}(\pi)\\ \mathrm{connected}}}\Ksf_{\ve_0}^-(X_1)\cdots \Ksf_{\ve_0}^-(X_n)\mathrm{I}(G(X_1,\ldots,X_n)).
	\end{multline*}
	Moreover, by \eqref{eq:bb2}, 
	\begin{multline*}
		\log \sum_{n=0}^{+\infty}\frac{1}{n!}\sum_{\substack{X_1,\ldots,X_n\in \mc{P}(\pi)\\ \mathrm{disjoint}}}\Ksf_{\ve_0}^-(X_1)\cdots \Ksf_{\ve_0}^-(X_n)\\=\sum_{n=1}^{+\infty}\frac{1}{n!}\sum_{\substack{X_1,\ldots,X_n\in \mc{P}(\pi)\\ \mathrm{connected}}} \indic_{|V_{X_1}|\leq p}\Ksf_{\ve_0}^-(X_1)\cdots \indic_{|V_{X_n}|\leq p}\Ksf_{\ve_0}^-(X_n) \mathrm{I}(G(X_1,\ldots,X_n))+O_{\beta,M,p,\ve_0}(N\delta_{\beta,\lambda}).
	\end{multline*}
	\begingroup\color{black}
Set
\[
d(X)=|\Ksf_{\ve_0}^-(X)-\Ksf_\infty^0(X)|\indic_{2\leq|X|,\ |V_X|\leq p},
\qquad
w(X)=(|\Ksf_{\ve_0}^-(X)|+|\Ksf_\infty^0(X)|)\indic_{2\leq|X|,\ |V_X|\leq p}.
\]
For $k=|V_X|\leq p$ and $m=|X|$, \eqref{eq:MX -low} and \eqref{eq:K-diff} give
\[
\Msf_\infty^0(X)\geq cN^{m-k}\lambda^{2(k-m)},
\qquad d(X)\leq C\Cut^{-2}N^{1-m}\lambda^{-2(k-m)}.
\]
By \eqref{eq:assnk}, the partition $\pi$ has at most $C_{\ve_0}N\lambda^{2(i-1)}$ blocks of cardinality $i$, which cancels the negative powers of $\lambda$ in $d(X)$. Moreover, \eqref{eq:K- bound} gives $w(X)\leq C(\lambda^2/N)^{m-1}$. Counting the blocks of each cardinality in $X$, and using $\Cut^{-2}\leq\lambda^{2(p-1)}$, we obtain
\[
D\coloneqq\sum_{X\subset\pi}d(X)20^{|V_X|}\leq CN\Cut^{-2},
\qquad
q\coloneqq\sup_{S\in\pi}\sum_{X\ni S}w(X)20^{|V_X|}\leq C\lambda^2.
\]
In the telescoping sum for the difference of the products of activities, one factor is bounded by $d$ and the others by $w$. Summing over the trees rooted at this factor, as in the proof of Lemma~\ref{lemma:bounded cluster tail}, we bound the error in the connected series by
\[
D\sum_{n\geq1}\frac{n^{n-1}}{n!}(pq)^{n-1}\leq CD\leq CN\Cut^{-2}.
\]
Thus
\endgroup
	\begin{multline*}
		\log \sum_{n=0}^{+\infty}\frac{1}{n!}\sum_{\substack{X_1,\ldots,X_n\in \mc{P}(\pi)\\ \mathrm{disjoint}}}\Ksf_{\ve_0}^-(X_1)\cdots \Ksf_{\ve_0}^-(X_n)\\=\sum_{n=1}^{+\infty}\frac{1}{n!}\sum_{\substack{X_1,\ldots,X_n\in \mc{P}(\pi)\\ \mathrm{connected}}} \indic_{|V_{X_1}|\leq p} \Ksf_{\infty}^0(X_1 )\dots \indic_{|V_{X_n}|\leq p}\Ksf_{\infty}^0(X_n) \mathrm{I}(G(X_1,\ldots,X_n))+O_{\beta,M,p,\ve_0}\left(N\delta_{\beta,\lambda}\right).
	\end{multline*}	
	Using Corollary~\ref{coro:equality mult} to equate $\Ksf_{\infty}^0$ with $\Ksf_{\beta,\lambda}^{\mult}$, we thus get 
	\begin{multline}\label{eq:clS1}
		\log \sum_{n=0}^{+\infty}\frac{1}{n!}\sum_{\substack{X_1,\ldots,X_n\in \mc{P}(\pi)\\ \mathrm{disjoint}}}\Ksf_{\ve_0}^-(X_1)\cdots \Ksf_{\ve_0}^-(X_n)\\=\sum_{n=1}^{+\infty}\frac{1}{n!}\sum_{\substack{X_1,\ldots,X_n\in \mc{P}(\pi)\\ \mathrm{connected}}} \indic_{|V_{X_1}|\leq p} \Ksf_{\beta,\lambda}^{\mult}(X_1 )\dots \indic_{|V_{X_n}|\leq p}\Ksf_{\beta,\lambda}^{\mult}(X_n) \mathrm{I}(G(X_1,\ldots,X_n))+O_{\beta,M,p,\ve_0}\left(N\delta_{\beta,\lambda}\right).
	\end{multline}	
	
	\paragraph{\bf{Step 3: summing over unlabeled quantities}} 
	
	\leavevmode\begingroup\color{black}
Below, limits are taken as $N\to\infty$ at fixed $\lambda$, uniformly over the admissible block profiles. By Proposition~\ref{prop:bounded lower} and Lemma~\ref{lemma:limiting}, the activities with $|V_X|\leq p$ satisfy $|\Ksf_{\beta,\lambda}^{\mult}(X)|\leq C(\lambda^2/N)^{|X|-1}$ for $N$ large enough, and they vanish for $|X|\leq1$. For $\lambda$ small enough, Lemma~\ref{lemma:bounded cluster tail} therefore allows us to pass to the limit term by term in \eqref{eq:clS1}.
\endgroup
First, we show that only hypertrees contribute an order-$N$ term in \eqref{eq:clS1}. Let $n\geq 1$ and recall $\Htrees_n$ from \eqref{def:Htrees}.
	Let us prove that the contribution of non-hypertrees is negligible. Summing according to the cardinality of $X_1\cup \cdots \cup X_n$ gives
	\begin{multline}\label{eq:no survive}
		\sum_{\substack{X_1,\ldots,X_n\in \mc{P}(\pi)\\ \mathrm{connected}\\ (X_1,\ldots,X_n)\notin \Htrees_n(\pi) }}\Bigr|\Ksf_{\beta,\lambda}^{\mult}(X_1)\indic_{|V_{X_1}|\leq p} \cdots \Ksf_{\beta,\lambda}^{\mult}(X_n) \indic_{|V_{X_n}|\leq p}\mathrm{I}(G(X_1,\ldots,X_n))\Bigr|\\\leq \sum_{k=1}^{\textcolor{black}{|\pi|}} \frac{N^k}{k!}\max_{\substack{Y\in \mc{P}(\pi)\\|Y|=k}} \sum_{\substack{X_1,\ldots,X_n\in \mc{P}(\pi)\\ \mathrm{connected}\\ (X_1,\ldots,X_n)\notin \Htrees_n(\pi) \\X_1\cup \cdots \cup X_n=Y  }}\Bigr|\Ksf_{\beta,\lambda}^{\mult}(X_1)\indic_{|V_{X_1}|\leq p} \cdots \Ksf_{\beta,\lambda}^{\mult}(X_n) \indic_{|V_{X_n}|\leq p}\mathrm{I}(G(X_1,\ldots,X_n))\Bigr|.
	\end{multline}
	If $|V_{X_i}|\leq p$, then by the scaling of $\Ksf_{\beta,\lambda}^{\mult}$ in Lemma~\ref{lemma:limiting},
	\begin{equation}\label{eq:asKT}
		\Ksf_{\beta,\lambda}^{\mult}(X_i)=N^{1-|X_i|}(\ksf_{\beta,\lambda}^{\mult}(\#_1X_i,\ldots,\#_pX_i)+o_N(1)).
	\end{equation}
	\textcolor{black}{Therefore, after division by $N$, a term can have a nonzero limit only if}
	\begin{equation}\label{eq:cond non vanish}
		\Bigr|\bigcup_{i=1}^n X_i\Bigr|+n-\sum_{i=1}^n |X_i|=1.
	\end{equation}
	\begingroup\color{black}
Since $X_1,\ldots,X_n$ are connected, the left-hand side is at most $1$, and, by \eqref{def:Htrees}, equality holds if and only if $\mathbf X\in\Htrees_n(\pi)$. Thus,
\endgroup 
	\begin{equation*}
		\sum_{\substack{X_1,\ldots,X_n\in \mc{P}(\pi)\\ \mathrm{connected}\\ (X_1,\ldots,X_n)\notin \Htrees_n(\pi) }}\Bigr|\Ksf_{\beta,\lambda}^{\mult}(X_1)\indic_{|V_{X_1}|\leq p} \cdots \Ksf_{\beta,\lambda}^{\mult}(X_n) \indic_{|V_{X_n}|\leq p}\mathrm{I}(G(X_1,\ldots,X_n))\Bigr|=o(N).
	\end{equation*}
	\begingroup\color{black}
For $\mathbf X\in\Htrees_n(\pi)$, the incidence graph on $[n]\sqcup\bigcup_iX_i$, with edges $(i,a)$ for $a\in X_i$, is connected and, by \eqref{def:Htrees}, has one fewer edge than vertices, so it is a tree. Write $d_a=|\{i:a\in X_i\}|$. We color $G(X_1,\ldots,X_n)$ with $t$ colors along this tree, from the root $1\in[n]$. Since the $d_a-1$ children of an element $a$ must receive distinct colors, different from the color of its parent, the chromatic polynomial of $G(X_1,\ldots,X_n)$ is
\[
t\prod_{a\in\cup_iX_i}\prod_{j=1}^{d_a-1}(t-j).
\]
By inclusion--exclusion, its coefficient of $t$ is $\mathrm I(G(X_1,\ldots,X_n))$. Since $\sum_a(d_a-1)=n-1$,
\begin{equation}\label{eq:hypertree weight}
\mathrm I(G(X_1,\ldots,X_n))=(-1)^{n-1}\omega(\mathbf X).
\end{equation}

We will also use a bound on the total weight of hypertrees. Let $\mathbf V$ be a hypertree covering $[k]$, and let us root its incidence tree at $1\in[k]$. Let $f:[n]\to[k]$ and $\varepsilon:[k]\to[n]$ map each vertex to its parent, with $\varepsilon(1)=\min\{i:1\in V_i\}$ at the root. Then $a\in V_i$ if and only if $f(i)=a$ or $\varepsilon(a)=i$. These maps determine the hypertree, and $\omega(\mathbf V)\leq\prod_{a=1}^k|f^{-1}(a)|!$. Summing over these maps, we obtain, for every $n,k\geq1$,
\begin{equation}\label{eq:hypertree-count}
\sum_{\substack{\mathbf V\in\Htrees_n([k])\\\cup_iV_i=[k]}}\omega(\mathbf V)
\leq n^k k(k+1)\cdots(k+n-1)\leq n^k(k+n)^n.
\end{equation}
Using \eqref{eq:hypertree weight}, we obtain
\endgroup
	\begin{multline}\label{eq:bf1}
		\sum_{\substack{X_1,\ldots,X_n\in \mc{P}(\pi)\\ \mathrm{connected}}}\Ksf_{\beta,\lambda}^{\mult}(X_1)\indic_{|V_{X_1}|\leq p} \cdots \Ksf_{\beta,\lambda}^{\mult}(X_n) \indic_{|V_{X_n}|\leq p}\mathrm{I}(G(X_1,\ldots,X_n))\\=(-1)^{n-1}\sum_{ (X_1,\ldots,X_n)\in \Htrees_n(\pi) }\textcolor{black}{\omega(\mathbf X)}\Ksf_{\beta,\lambda}^{\mult}(X_1)\indic_{|V_{X_1}|\leq p} \cdots \Ksf_{\beta,\lambda}^{\mult}(X_n) \indic_{|V_{X_n}|\leq p}\textcolor{black}{{}+o(N)}.
	\end{multline}
	We now sum the above display according to the number of multipoles of cardinality $i$ in $X_1\cup \cdots \cup X_n$:
	\begin{multline}\label{eq:bf2}
		\sum_{(X_1,\ldots,X_n)\in \Htrees_n(\pi) }\textcolor{black}{\omega(\mathbf X)}\Ksf_{\beta,\lambda}^{\mult}(X_1)\indic_{|V_{X_1}|\leq p} \cdots \Ksf_{\beta,\lambda}^{\mult}(X_n) \indic_{|V_{X_n}|\leq p}\\=\sum_{m_1,\ldots,m_p\geq 0}\sum_{\substack{Y\in \mc{P}(\pi):\\ \forall i\in [p], \#_iY=m_i}}  \sum_{\substack{(X_1,\ldots,X_n)\in \Htrees_n(\pi) \\ X_1\cup \cdots \cup X_n=Y}}\textcolor{black}{\omega(\mathbf X)}\Ksf_{\beta,\lambda}^{\mult}(X_1)\indic_{|V_{X_1}|\leq p} \cdots \Ksf_{\beta,\lambda}^{\mult}(X_n)\indic_{|V_{X_n}|\leq p}.
	\end{multline}
	Given $m_1,\ldots,m_p$, notice that the number of ways to choose $Y\subset \pi$ with $m_i$ parts of cardinality $i$ for every $i\in [p]$ is given by 
	\begin{equation*}
		\prod_{i=1}^p \binom{n_i}{m_i}.
	\end{equation*}
\begingroup\color{black}
For fixed $m_1,\ldots,m_p$, uniformly in $0\leq n_i\leq N$,
\begin{equation*}
N^{-\sum_i m_i}\prod_{i=1}^p\binom{n_i}{m_i}
=\prod_{i=1}^p\frac{\gamma_i^{m_i}}{m_i!}+O_{p,m_1,\ldots,m_p}(N^{-1}).
\end{equation*}
\endgroup
	Recall from Definition~\ref{def:canonical} that $Y(m_1,\ldots,m_p)$ is the canonical partition of $[m_1+2m_2+\cdots+pm_p]$ with $m_i$ blocks of cardinality $i$ for every $i\in [p]$. Then, combining the two above displays with \eqref{eq:asKT}, we obtain 
	\begin{multline}\label{eq:bf3}
		\sum_{\substack{Y\in \mc{P}(\pi):\\ \forall i\in [p], \#_iY=m_i}}  \sum_{\substack{ (X_1,\ldots,X_n)\in \Htrees_n(\pi)\\ X_1\cup \cdots \cup X_n=Y}}\textcolor{black}{\omega(\mathbf X)}\Ksf_{\beta,\lambda}^{\mult}(X_1)\indic_{|V_{X_1}|\leq p} \cdots \Ksf_{\beta,\lambda}^{\mult}(X_n)\indic_{|V_{X_n}|\leq p}\\
		=N \prod_{i=1}^p \frac{\gamma_i^{m_i}}{m_i!}\sum_{\substack{(X_1,\ldots,X_n)\in \Htrees_n(Y(m_1,\ldots,m_p))\\ X_1\cup \cdots \cup X_n=Y(m_1,\ldots,m_p)}}\textcolor{black}{\omega(\mathbf X)}\textcolor{black}{\prod_{i=1}^n\ksf_{\beta,\lambda}^{\mult}((\#_kX_i)_{k=1}^p)\indic_{|V_{X_i}|\leq p}}+o(N).
	\end{multline}

	\textcolor{black}{Combining \eqref{eq:bf1}, \eqref{eq:bf2} and \eqref{eq:bf3}, we obtain, uniformly in $0\leq n_i\leq N$,}
	\begin{multline}\label{eq:asin}
		\sum_{\substack{X_1,\ldots,X_n\in \mc{P}(\pi)\\ \mathrm{connected}}}\Ksf_{\beta,\lambda}^{\mult}(X_1)\indic_{|V_{X_1}|\leq p} \cdots \Ksf_{\beta,\lambda}^{\mult}(X_n) \indic_{|V_{X_n}|\leq p}\mathrm{I}(G(X_1,\ldots,X_n)) \\=N(-1)^{n-1} \sum_{m_1,\ldots,m_p}\prod_{i=1}^p \frac{\gamma_i^{m_i}}{m_i!}\sum_{\substack{(X_1,\ldots,X_n)\in \Htrees_n(Y(m_1,\ldots,m_p))\\ X_1\cup \cdots \cup X_n=Y(m_1,\ldots,m_p)}}\textcolor{black}{\omega(\mathbf X)}\textcolor{black}{\prod_{i=1}^n\ksf_{\beta,\lambda}^{\mult}((\#_kX_i)_{k=1}^p)\indic_{|V_{X_i}|\leq p}}\\+o(N).
	\end{multline}
	\begingroup\color{black}
By \eqref{eq:bounded cluster tail}, we may sum \eqref{eq:asin} over $n$ with weight $1/n!$. For fixed $\lambda$ and $N$ large enough, the resulting error $o(N)$ is at most $N\delta_{\beta,\lambda}$. Combining this with \eqref{eq:clS1} proves \eqref{claimprop446}.
\endgroup
\end{proof}

\section{Large-deviation upper bound}\label{section:upper}

In this section, we establish an upper bound on the partition function. Unlike in the case of the lower bound, we cannot assume that all dipoles have size smaller than $\varepsilon_0 \Cut$ and that all multipoles have cardinality \textcolor{black}{at most $p(\beta)$}. We must consider all configurations. However, the absolute convergence of the cluster-expansion series in Proposition~\ref{prop:absolute lower} crucially relied on the bounded (by $\ve_0 \Cut$) size of dipoles and the bounded cardinality of multipoles.

\begingroup\color{black}
Throughout this section, $M>\max(20,3^{p(\beta)})$ is large enough for Propositions~\ref{prop:bounded lower} and~\ref{prop:expansion +}, $\lambda$ is small enough, and $N$ is large enough at fixed $\lambda$. We use $\inf\emptyset=+\infty$. On a configuration with $\sigma_N\neq\Id$, we apply the definitions below to the matched dipoles $(x_i,y_{\sigma_N(i)})$.
\endgroup

\subsection{Strategy for the upper bound}

The key idea of the proof is to isolate a family of ``good dipoles'' and to perform a cluster expansion on the partition function of this subsystem, following the approach of Section~\ref{section:lower}.  
A major difficulty is that good dipoles do not form an isolated system: they still interact with bad dipoles.  
A dipole is declared good if it satisfies the following three conditions:

\begin{enumerate}
	\item its length is smaller than $\ve_0\Cut $;  
	\item it belongs to a multipole of cardinality \textcolor{black}{at most $p(\beta)$}, and the number of such multipoles is not too large;
	\item it contains no ``internal'' bad dipole of either of the two previously defined types whose size is smaller than its own.
\end{enumerate}

In Subsection~\ref{sub:bad points} we introduce a reduced model by applying the energy lower bound of Proposition~\ref{prop:mino} with carefully chosen radii.  
For a good point, the radius is taken to be a fixed fraction of its distance to the closest bad point, with a cutoff at $\varepsilon_{0}\Cut$.  
Furthermore, the radii of bad points are chosen small enough to prevent overcrowding of too many same-sign particles.

In Subsection~\ref{sub:isolating} we isolate the good dipoles and write their partition function as a cluster-expansion series, within the framework of Subsection~\ref{sub:pert}.  
The corresponding activity estimates are stated in Subsection~\ref{sub:stat upper}.

A key fact in Section~\ref{section:lower} was that the graphs weighted by $a_{ij}^{\tilde v}$ contributing to the activity are Eulerian \textcolor{black}{relative to the multipole partition, so that their nontrivial relative components are 2-edge-connected relative to their induced partitions}.  
Here, because of the interaction with bad dipoles, the analogous graphs with odd weights are not always Eulerian or 2-edge-connected.  
We decompose these graphs into a 2-core consisting of a tree of 2-edge-connected blocks, to which pendant trees are attached.  
By a parity argument, every \textcolor{black}{nonroot leaf multipole} of a pendant tree interacts with bad points.  
Using fine geometric properties of the configuration and our choice of radii, we show that the smallness of the good-bad interactions compensates for the loss of 2-edge-connectivity.  
This is done in Subsections~\ref{sub:pendant} and~\ref{sub:quadra}.  
Since the contributing graphs obey exactly the same bound as in Corollary~\ref{coro:prod a}, we can import the analysis of Section~\ref{section:lower} to establish the activity estimates.

Thanks to the choice of radii for bad points, we may bound their energy from below, allowing us to reduce to a system of noninteracting dipoles.  
Subsections~\ref{sub:type 12}, \ref{sub:type1} and \ref{sub:type2} control the total number of bad points.  
Together with the cluster-expansion estimate for good dipoles, this yields the desired LDP upper bound.

Combining this upper bound with the LDP lower bound of Proposition~\ref{prop:lower bound} proves the full LDP of Theorem~\ref{theorem:LDP}.  
As a corollary, we obtain an expansion of the partition function expressed through the infimum of the rate function $\mathcal{I}_{\beta,p(\beta),\lambda}$ of Definition~\ref{def:rate function}.  
Our final aim is to simplify that expression and derive the expansion of Theorem~\ref{theorem:expansion}. \textcolor{black}{For this, we keep the multipole activities of clusters of at most $p(\beta)$ dipoles and sum over the multipole partitions. M\"obius inversion then identifies the resulting activities with the dipole activities of Definition~\ref{def:dipole activity}.}

\subsection{Definition of bad points and of a reduced model}\label{sub:bad points}

We first define clusters of positive charges and clusters of negative charges. \textcolor{black}{We assume $\lambda\leq\ve_0\Cut$.}

Roughly speaking, we grow balls of radius $\tau$ around all the positive charges. As soon as  $\tau$ equals $\frac{1}{4}r_i$ (where $r_i$ is the distance from $x_i$ to its match $y_{\sigma_N(i)}$), we remove the point and its associated ball. Moreover, as soon as a ball intersects with  $q(\beta)$ (a fixed parameter) other balls, we remove $B(x_i,\tau)$ and all the balls touching it. The erased points are declared to be a cluster. Stop the process when $\tau$ reaches $\ve_0\Cut$ and declare each remaining charge to be a (singleton) cluster.


\begin{definition}[$q$-clustering of same-sign particles]\label{def:continuous clustering}
	Let $q\geq 2$ be a fixed \textcolor{black}{integer}.
	Let $(\vec{X}_N, \vec{Y}_N)$ be a configuration in $(\Lambda^2)^N$ such that $\sigma_N[\vec{X}_N, \vec{Y}_N]=\Id$, where~$\sigma_N$ is the stable matching from Definition~\ref{def:stable match}.  Recall $r_i = |y_i-x_i|$. \textcolor{black}{All balls in this definition are closed.} A~$q$-clustering of~$\{x_i\}$ is a partition of $[N]$ and a list of radii given by the following iteration: 
	\begin{itemize}
		\item Define~$I_0 \coloneqq  [N]$
		and $\Clus_0^+\coloneqq \emptyset.$
		\item Having defined~$I_{n-1}$ and~$\Clus_{n-1}^+$,
		\begin{enumerate}
			\item If $I_{n-1}=\emptyset$, stop and let~$\Clus_{n-1}^+$ be the final cluster.
			\item Otherwise, let $\tau$ be given by the minimum of \textcolor{black}{$\ve_0\Cut$,} $\frac{1}{4}\min\{r_i : i\in I_{n-1} \}$ and the minimal $t$\textcolor{black}{, if it exists,} such that some ball $B(x_i,t)$ for $i\in I_{n-1}$ intersects at least $q$ other balls $B(x_j,t)$ for $j\in I_{n-1}$.
			\begin{itemize}
				\item If $\tau\textcolor{black}{=} \ve_0\Cut$, then set $\Clus_{n}^+ \coloneqq \Clus_{n-1}^+\sqcup \{\{i\}:i\in I_{n-1}\}$ and $I_n\coloneqq \emptyset$.
				\item \textcolor{black}{Otherwise, if} $\tau=\frac{1}{4}r_i$ for some $i\in I_{n-1}$, set $\Clus_n^+\coloneqq \Clus_{n-1}^+\sqcup \{\{i\}\}$, $\tau_i^+\coloneqq \frac{1}{4}r_i$ and $I_n\coloneqq I_{n-1}\setminus\{i\}$.
				\item \textcolor{black}{Otherwise, among the centers whose balls intersect at least $q$ others, choose $i$ with the smallest $r_i$, breaking ties by index, and set} $\Clus_n^+=\Clus_{n-1}^+\sqcup \textcolor{black}{\{J\}}$ and $I_n=I_{n-1}\setminus J$,
				where $J$ is the set of $j\in I_{n-1}$ such that $B(x_j,\tau)$ intersects $B(x_i,\tau)$.  (Notice that $i\in J$.) 
			\end{itemize}
		\end{enumerate}
		\item Denote by $\Clus_N^+$ the set of clusters obtained and, for every $i\in [N]$, by $\tau_i^+$ the maximum of $\lambda$ and the value of $\tau$ at which $B(x_i,\tau)$ has been removed. By exchanging the roles of the positive and negative charges above, define a $q$-clustering of negative charges, denoted $\Clus_N^-$, with an associated list of radii $\tau_i^-$.
	\end{itemize}
\end{definition}

Up to the removal of a set of measure $0$, we will assume that \textcolor{black}{at each overcrowding step, the selected ball} $B(x_i, \tau)$ intersects exactly $q$ other balls.

\begingroup\color{black}
Since removing labels can only increase the next crowding time, the removal values are nondecreasing, and singletons with the same removal value may be removed in any order. The output is a partition. The removal value of $i$ is at most $\min(\ve_0\Cut,r_i/4)$, and, before this value, the ball of $i$ meets fewer than $q$ other balls still present.

Moreover, deleting a union of output clusters preserves the other clusters and their removal values. Indeed, we may skip the deleted clusters in the recursion, since a retained overcrowded center keeps all its neighbors and deletions cannot create an earlier overcrowding. Finally, the recursion commutes with order-preserving relabelings, and with every relabeling outside the null set of Subsection~\ref{sub:stable}.
\endgroup

\textcolor{black}{We first declare bad all long dipoles and all dipoles containing an overcrowded charge.}

\begin{definition}[Bad points of type 1]\label{def:type1}
	Let us define
	\begin{equation}\label{def:qbeta}q(\beta)\coloneqq 
		\begin{cases} \lfloor \frac{3}{2-\beta/2} \rfloor & \text{if $\beta\in (2,4)$}\\
			\textcolor{black}{2p_0+1} & \text{if $\beta\geq 4$}.
		\end{cases}
	\end{equation}
	Let $\tau_i^+$ and $\tau_i^-$ be the radii given by the $q(\beta)$-clustering algorithm of Definition~\ref{def:continuous clustering}. 
	
{\color{black}
On the event $\{\sigma_N=\Id\}$, we define
\begin{equation*}
    J_\bad^{1,\pm}\coloneqq
    \{i\in[N]:r_i>\ve_0\Cut\}
    \cup\bigcup_{\substack{S\in\Clus_N^\pm\\|S|=q(\beta)+1}}S,
\end{equation*}
}
	and also let
	\begin{equation*}
		I_\bad^1\coloneqq \{i\in [N]:i\in J_\bad^{1,+} \text{ or }i\in J_\bad^{1,-}\}.
	\end{equation*}
\end{definition}

Next, we include every dipole that belongs to a multipole whose size exceeds $p(\beta)$ and all small-cardinality multipoles whenever their total count surpasses the permitted threshold.

\begin{definition}[Bad points of type $2$]\label{def:type2}
	Recall $p(\beta)$ from Definition~\ref{def:pbeta}. \textcolor{black}{Let $U=[N]\setminus I_\bad^1$, and let $\Pi_\mult^U$ be the partition of $U$ into multipoles.} We let $I_\bad^{2,1}$ be the set of $i\notin I_\bad^1$ such that $|[i]^{\textcolor{black}{\Pi_\mult^U}}|>p(\beta)$. 
	
	Moreover, for every $k\in \{1,\ldots,p(\beta)\}$, we let 
	\begin{equation*}
		I_\bad^{2,2,k}\coloneqq \begin{cases}
			\textcolor{black}{\{i\in U:|[i]^{\Pi_\mult^U}|=k\}} & \text{if $\textcolor{black}{|\{i\in U:|[i]^{\Pi_\mult^U}|=k\}|}>\ve_0^{-\alpha(\beta)} N\lambda^{2(k-1)}$}\\
			\emptyset & \text{otherwise}.
		\end{cases} 
	\end{equation*}
	Set
	\begin{equation*}
		I_\bad^{2,2}\coloneqq \bigcup_{k\in [p(\beta)]}I_{\bad}^{2,2,k}
	\end{equation*}
	and 
	\begin{equation*}
		I_\bad^2\coloneqq I_\bad^{2,1}\cup I_\bad^{2,2}.
	\end{equation*}
\end{definition}

Finally, we discard the dipoles having an ``internal'' bad dipole and take the closure by multipole of this set.

\begin{definition}[Bad points of type $3$]\label{def:type3}
	We let $J_\bad^{3}$ be the set of $i\in [N]\setminus (I_\bad^1 \cup I_\bad^2)$ such that there exists $j\in I_\bad^1\cup I_\bad^2$ such that 
	\begin{equation}\label{nelledef}
		d_{ij}\leq {M\max(r_i,\lambda)}.
	\end{equation}
	We then let 
	\begin{equation*}
		I_\bad^{3}\coloneqq \Bigr\{ i\in [N]\setminus (I_\bad^1 \cup I_\bad^2):[i]^{\textcolor{black}{\Pi_\mult^U}}\cap J_\bad^{3}\neq \emptyset\Bigr\}.
	\end{equation*}
\end{definition}

\begin{definition}[Good and bad points]\label{def:bad points}We let
	\begin{equation*}
		I_\bad\coloneqq I_\bad^1\cup I_\bad^2\cup I_\bad^3
	\end{equation*}
	and finally set
	\begin{equation*}
		I_\good\coloneqq [N]\setminus I_\bad.
	\end{equation*}
\end{definition}

The set $I_\bad$ and, for every $k\geq1$, the number of multipoles of cardinality $k$ contained in $I_\good$ are measurable functions of the configuration.

\begin{remark}\label{remarkgoodbad}
	\begin{enumerate}
		\item
		\textcolor{black}{If $i\in U$ and $j\in I_\bad^1$ satisfy $i\leftrightarrow_{\sigma_N}j$, then $[i]^{\Pi_\mult^U}\subset I_\bad$. Since $I_\bad^2$ and $I_\bad^3$ are unions of blocks of $\Pi_\mult^U$, the set $I_\bad$ is a union of multipoles.}
		\item
		\textcolor{black}{If $i\notin I_\bad^1\cup I_\bad^2$, then $|[i]^{\Pi_\mult^U}|\leq p(\beta)$ and $r_i\leq\ve_0\Cut$. Every multipole contained in $I_\good$ is a block of $\Pi_\mult^U$, and the number of elements of $I_\good$ in multipoles of size $k$ is at most $\ve_0^{-\alpha(\beta)}N\lambda^{2(k-1)}$.}
	\end{enumerate}
	
\end{remark}

\begin{definition}\label{def:dist bad} We let $Z_\bad=\{\{x_i,y_i\}, i \in I_\bad\}$  and 
	for every $i\in I_\good$, we let 
	\begin{equation}\label{eq:dibad}
		d_{i,\bad}\coloneqq \min(\dist(\{x_i,y_{i}\},Z_\bad),4\ve_0\Cut).
	\end{equation}
	For every $S\subset I_\good$, we let 
	\begin{equation}\label{def:Sbad}
		d_{S,\bad}=\inf_{i\in S}d_{i,\bad}.
	\end{equation}
\end{definition}

We next apply the result of the ball-growth method of Proposition~\ref{prop:mino} by choosing the enlargement radii as follows:

\begin{definition}[Interaction for the energy lower bound]\label{def:vijZ} 
	Let $Z=Z_\bad$ and let $\ve_0\in (0,1)$.  
	\begin{itemize}
		\item For every $i\in I_\good$, on the event $\{\sigma_N=\Id\}$, we let
		\begin{equation}\label{deftauz}
			\tau_i^{+,Z}=\tau_i^{-,Z}=
			\max(\tfrac{1}{4}d_{i,\bad},\lambda).
		\end{equation}
		We will also sometimes denote by $\tau_i^Z$ the common value of $\tau_i^{+,Z}$ and $\tau_i^{-,Z}$.
		\item For every $i\in I_\bad$, on the event $\{\sigma_N=\Id\}$, we let 
		\begin{equation*}
			\tau_i^{+,Z}=\max(\tau_i^+,\lambda)\  \text{and}\quad   \tau_i^{-,Z}=\max(\tau_i^-,\lambda),
		\end{equation*}
		where $\tau_i^+$ and $\tau_i^-$ are as in Definition~\ref{def:continuous clustering}.  Note that $\lambda\leq \tau_i^{\pm,Z}\le \ve_0 \Cut$.
	\end{itemize}
	For every $i,j\in [N]$ with $i\neq j$, on the event $\{\sigma_N=\Id\}$, define
	\begin{equation}\label{def:v+}
		v_{i,j,+}^Z\coloneqq \Bigl(\g_\lambda-\g * \delta_0^{(\tau_i^{+,Z})}*\delta_0^{(\tau_j^{+,Z})}\Bigr) (x_i-x_j) -\Bigl(\g_\lambda-\g * \delta_0^{(\tau_i^{-,Z})}*\delta_0^{(\tau_j^{+,Z})}\Bigr) (y_{i}-x_j) 
	\end{equation}
	and
	\begin{equation}\label{def:v-}
		v_{i,j,-}^Z\coloneqq \Bigl(\g_\lambda-\g * \delta_0^{(\tau_i^{+,Z})}*\delta_0^{(\tau_j^{-,Z})}\Bigr) (x_i-y_{j}) -\Bigl(\g_\lambda-\g * \delta_0^{(\tau_i^{-,Z})}*\delta_0^{(\tau_j^{-,Z})}\Bigr) (y_{i}-y_{j}),
	\end{equation}
	and 
	\begin{equation}\label{defvijZ}
		v_{ij}^Z\coloneqq  v_{i,j,+}^Z-v_{i,j,-}^Z.\end{equation}
\end{definition}

Recalling that $v_{i,j,+}^Z=v_{i,j,+}'$, we see that $v_{i,j,+}^Z$ is nonzero only if $B(x_i, \tau_i^{+,Z})\cup B(y_{i}, \tau_i^{-,Z})$ intersects $B(x_j, \tau_j^{+,Z})$. The definition of the radii will ensure that the good points do not interact too much with bad points and that each bad point does not have an attractive interaction with more than a bounded number of charges per scale.

\begin{lemma}[Geometric properties]\label{lemma:geo} 
	Recall the constant $M$ in Definition~\ref{def:multipoles}. \textcolor{black}{Suppose that $M>\max(20,3^{p(\beta)})$.} Let $\Pi_\mult$ be the partition of $[N]$ into multipoles. Recall $d_{i,\bad}$ from Definition~\ref{def:dist bad}. Let $d_{i,j}^+=\dist(\{x_i,y_i\},x_j)$ and $d_{i,j}^-=\dist(\{x_i,y_i\},y_j)$. Let $(\vec{X}_N,\vec{Y}_N)\in (\Lambda^2)^N$ be such that $\sigma_N[\vec{X}_N,\vec{Y}_N]=\Id$.
	\begin{enumerate}
		\item Let $i\in I_\good$. We have
		\begin{equation}\label{eq:geo1}
			d_{i,\bad}\geq \max(r_i,\lambda).
		\end{equation}
		\item Let $i\in I_\good$ and $j\in I_\good$ be such that $i$ and $j$ belong to the same multipole. Then, there exists $C>0$ depending on $M$ and $p(\beta)$ such that
		\begin{equation}\label{eq:geo1'}
			d_{i,\bad}\leq Cd_{j,\bad}.
		\end{equation}
		\item Let $i \in I_{\good}$ and $j\in I_\bad$ be such that $v_{i,j,+}^Z\neq 0$. Then, 
		\begin{equation}\label{eq:geo2}
			d_{i,j}^+\leq 2\tau_j^{+,Z}\quad\text{and}\quad \max(r_i, \lambda)\le \max(r_j,\lambda). 
		\end{equation}
		\item Let $i \in I_{\good}$ and $j\in I_\bad$ be such that $v_{i,j,-}^Z\neq 0$. Then, 
		\begin{equation}\label{eq:geo3}
			d_{i,j}^-\leq 2\tau_j^{-,Z}\quad\text{and}\quad \max(r_i, \lambda)\le \max(r_j,\lambda). 
		\end{equation}
		\item Let $i_1,\ldots,i_K \in I_{\good}$. Suppose that 
		\begin{equation*}
			v_{i_1i_2}^Z\cdots  v_{i_{K-1}i_K}^Z\neq 0.
		\end{equation*}
		Then,
		\begin{equation}\label{eq:geo4}
			d_{i_K,\bad}\leq 6^{K-1} d_{i_1,\bad}.
		\end{equation}
		\item Let $i_1,\ldots,i_K\in I_\good$ and $i_1',\ldots,i_K'\in I_\good$ \textcolor{black}{with $i_l\neq i_l'$ for every $l\in\{1,\ldots,K\}$.} Suppose that, for every $l\in \{1,\ldots,K-1\}$, $[i_l']^{\Pi_\mult}=[i_{l+1}]^{\Pi_\mult}$. Suppose that 
		\begin{equation*}
			v_{i_1 i_1'}^Z \ldots v_{i_Ki_K'}^Z\neq 0.
		\end{equation*}
		Then, there exists a constant $C>0$ depending on $M$ and $p(\beta)$ such that 
		\begin{equation}\label{eq:geo4'}
			d_{i_K',\bad}\leq C^{K}d_{i_1,\bad}.
		\end{equation}
		\item Let $t\geq 0$ and $A\geq \frac{1}{2}$. Then, there exists a constant $C>0$ depending on $q(\beta)$ and $A$ such that, for every $z\in \Lambda$,
		\begin{equation}\label{eq:geo5}
			|\{i\in I_\bad: x_i\in B(z,At),\tau_i^{+,Z}\geq t  \text{ and }\tau_i^{+,Z}>\lambda \}|\leq C,
		\end{equation}
		\begin{equation}\label{eq:geo6}
			|\{i\in I_\bad: y_i\in B(z,At),\tau_i^{-,Z}\geq t\text{ and }\tau_i^{-,Z}>\lambda \}|\leq C.
		\end{equation}
	\end{enumerate}
\end{lemma}

\medskip

\begin{remark}\label{remcrucialetau}
	In particular, \eqref{eq:geo1}  and \eqref{deftauz} imply that, for every $i\in I_\good$,
	\begin{equation*}
		\tau_i^{+,Z}=\tau_i^{-,Z}\geq \frac{r_i}{4}, 
	\end{equation*}
	and thus  $r_i\le 2( \tau_i^{+,Z}+ \tau_i^{-,Z})$.  
	
	If   $i \in I_\bad$, then by construction, we always have $\tau_i^{\pm,Z} \le \max( \frac14 r_i, \lambda)$. Thus, if $i\in I_\bad$, 
	either $\frac{1}{4}r_i<\lambda$, in which case, by construction, $\tau_i^{+,Z}=\tau_i^{-,Z}=\lambda$
	and $r_i\le 2( \tau_i^{+,Z}+ \tau_i^{-,Z})$, or $\frac{1}{4}r_i \ge \lambda$, in which case 
	$2(\tau_i^{+,Z} +\tau_i^{-,Z}) \le  4\max( \frac14 r_i, \lambda) \le  r_i$.
	
\end{remark}
\medskip

\begin{proof} 
{\color{black}
We begin by proving \eqref{eq:geo1}. We use
\begin{equation}\label{eq:distAC}
    \dist(A,C)\leq\dist(A,B)+\diam(B)+\dist(B,C).
\end{equation}
Let $i\in I_\good$ and write $A=\max(r_i,\lambda)$. Suppose that $d_{i,\bad}<A$. The nearest bad dipole has type $3$, since otherwise $i$ would belong to $J_\bad^3$. From this dipole, a path of at most $p(\beta)-1$ steps in its block of $\Pi_\mult^U$ leads to an element of $J_\bad^3$, and one final step given by \eqref{nelledef} leads to a dipole of type $1$ or $2$.
Let $j$ be a type-$3$ vertex of this path with $D=d_{ij}<MA$. Since $i$ and $j$ belong to distinct multipoles by item (1) of Remark~\ref{remarkgoodbad}, we have $\max(r_j,\lambda)<D/M$. By \eqref{eq:distAC}, the distance from $i$ to the next vertex is therefore less than
\[
    D+(M+1)\max(r_j,\lambda)<3D.
\]
By induction, the distance from $i$ to the last vertex is less than $3^{p(\beta)}A<MA$, so that $i\in J_\bad^3$, a contradiction. This proves \eqref{eq:geo1}.
}

	We turn to the proof of \eqref{eq:geo1'}. Suppose that $i\leftrightarrow j$, i.e., that $d_{ij}\leq M\max(\min(r_i,r_j),\lambda)$. Then, by \eqref{eq:distAC},
	\begin{equation}
		d_{j,\bad}\leq d_{i,\bad}+r_i+d_{ij}\leq d_{i,\bad}+(M+1)\max(r_i,\lambda).
	\end{equation}
	By \eqref{eq:geo1}, $\max(r_i,\lambda)\leq d_{i,\bad}$, which implies that 
	\begin{equation}\label{eq:compa1}
		d_{j,\bad}\leq (M+2)d_{i,\bad}.
	\end{equation}
	Similarly, 
	\begin{equation}\label{eq:compa2}
		d_{i,\bad}\leq (M+2)d_{j,\bad}.
	\end{equation}
	Now let $i$ and $j\in I_\good$ be in the same multipole. Recall that, by definition of $I_\good$, one has $|[i]^{\Pi_\mult}|\leq p(\beta)$. Hence, taking a path in the multipole $[i]^{\Pi_\mult}$ (of length \textcolor{black}{at most $p(\beta)-1$}), we deduce \eqref{eq:geo1'} from \eqref{eq:compa1} and \eqref{eq:compa2}.

	We now prove \eqref{eq:geo2}. Recall from Definition~\ref{def:vijZ} that, for $i\in I_\good$, we have $\tau_i^{+,Z}=\tau_i^{-,Z}=:\tau_i^Z$. By \eqref{eqvijpp} (or Newton's theorem), $v^Z_{i,j,+}\neq 0$ implies that 
	\begin{equation}\label{dijtitj} d_{i,j}^+\le \tau_i^Z+\tau_j^{+,Z}.\end{equation}
	By definition, we have
	$\tau_i^Z \leq \max(\frac{1}{4}d_{ij},\lambda)\leq \max(\frac{1}{4}d_{i,j}^+,\lambda)$. Therefore, by \eqref{dijtitj}, 
	\begin{equation*}
		d_{i,j}^+\le \max(\tfrac{1}{4}d_{i,j}^+,\lambda)+ \tau_j^{+,Z}.
	\end{equation*}
	Hence, since $\tau_j^{+,Z}\geq \lambda$,
	\begin{equation*}
		\max(d_{i,j}^+,\lambda)\leq 2\tau_j^{+,Z}.
	\end{equation*}
	This proves the first part of the statement. In particular,  since $j\in I_\bad$, by Remark~\ref{remcrucialetau}, $\tau_j^{+,Z}\leq \max(\tfrac{r_j}{4},\lambda)$, so $d_{i,j}^+\leq 2\max(r_j,\lambda)$.  If $\max(r_i,\lambda)\geq \max(r_j,\lambda)$, then this would imply that $i$ and $j$ are in the same multipole, contradicting item (1) of Remark~\ref{remarkgoodbad}. This proves \eqref{eq:geo2}. The proof of \eqref{eq:geo3} is similar.

	Let us prove \eqref{eq:geo4}. Let $k\in \{1,\ldots,K-1\}$. Using \eqref{eq:distAC}, we get 
	\begin{equation}\label{eq:tri}
		d_{i_{k+1},\bad}\leq d_{i_k i_{k+1}}+d_{i_k,\bad}+r_{i_k}.
	\end{equation}

	Since $v^Z_{i_ki_{k+1}}\neq 0$, \textcolor{black}{by Newton's theorem, as in the proof of Lemma~\ref{lemma:error'}, we have} 
	\begin{equation}\label{eq:yyyy}
		d_{i_ki_{k+1}}\leq \tau_{i_k}^{Z}+\tau_{i_{k+1}}^{Z}.
	\end{equation}
	Recall that 
	\begin{equation*}
		\tau_{i_{k+1}}^{Z}\leq \max(\tfrac{1}{4}d_{i_{k+1},\bad},\lambda).
	\end{equation*}
	Hence, by \eqref{eq:distAC} and \eqref{eq:tri},
	\begin{equation*}
		\tau_{i_{k+1}}^{Z}\leq \max(\tfrac{1}{4}d_{i_{k},\bad},\lambda)+\frac{1}{4}d_{i_ki_{k+1}}+\frac{1}{4}r_{i_k}.
	\end{equation*}
	Using that 
	\begin{equation*}
		\tau_{i_k}^{Z}\leq \max(\tfrac{1}{4}d_{i_{k},\bad},\lambda),
	\end{equation*}
	we therefore get from \eqref{eq:yyyy} that
	\begin{equation*}
		d_{i_ki_{k+1}}\leq \frac{8}{3}\max(\tfrac{1}{4}d_{i_k,\bad},\lambda)+\frac{1}{3}r_{i_k} \leq 3\max(d_{i_k,\bad},\lambda)+\frac{1}{3}r_{i_k}.
	\end{equation*}
	Thus, inserting \eqref{eq:geo1}, 
	\begin{equation*}
		d_{i_ki_{k+1}}\leq 4\max(d_{i_k,\bad},\lambda).
	\end{equation*}
	Hence, using \eqref{eq:tri}, we get
	\begin{equation}\label{eq:sim}
		d_{i_{k+1},\bad}\leq 6\max(d_{i_k,\bad},\lambda),
	\end{equation}
	which completes the proof of \eqref{eq:geo4} in view of \eqref{eq:geo1}.

	Let us prove \eqref{eq:geo4'}. By \eqref{eq:geo1'}, there exists $C_0>0$ depending on $M$ and $p(\beta)$ such that, for every $i,j$ in the same multipole,
	\begin{equation}\label{eq:same mult}
		d_{i,\bad}\leq C_0d_{j,\bad}.
	\end{equation}
	
	Let $k\in \{1,\ldots,K-1\}$. By \eqref{eq:sim}, 
	\begin{equation*}
		d_{i_k',\bad}\leq 6\max(d_{i_k,\bad},\lambda)=6d_{i_k,\bad}.
	\end{equation*}
	Therefore, by \eqref{eq:same mult},
	\begin{equation*}
		d_{i_{k+1},\bad}\leq 6C_0 d_{i_k,\bad}.
	\end{equation*}
	We deduce that 
	\begin{equation*}
		d_{i_K,\bad}\leq (6C_0)^{K-1}d_{i_1,\bad}
	\end{equation*}
	and that 
	\begin{equation*}
		d_{i_K',\bad}\leq 6(6C_0)^{K-1} d_{i_1,\bad},
	\end{equation*}
	which proves the result.

	Finally, we turn to the proof of \eqref{eq:geo5}. Since the ball $B(z,At)$ can be covered by $O(A^2)$ balls of radius $\frac{t}{2}$, it is enough to prove the result for $A=\frac{1}{2}$ (with a constant that may depend on $A$).
	
	Fix $z\in\Lambda$ and $t\ge 0$, and set
	\[
	I\coloneqq \{ i\in I_\bad: x_i\in B(z,\tfrac{t}{2}),\ \tau_i^{+,Z}\ge t, \tau_i^{+,Z}>\lambda \}.
	\]
	Assume for contradiction that $|I|\ge q(\beta)+2$. Let
	\[
	\tau\coloneqq \min_{i\in I}\tau_i^{+,Z},
	\]
	and choose $i_0\in I$ such that $\tau_{i_0}^{+,Z}=\tau$ (possible since $I$ is finite).
	Since $\tau>\lambda$, define
	\[
	s\coloneqq \frac{\tau+\lambda}{2},
	\]
	so that $\lambda<s<\tau$.
	
	For every $j\in I\setminus\{i_0\}$, we have $\tau_j^{+,Z}\ge \tau> s$, so at scale $s$ the balls $B(x_{i_0},s)$ and $B(x_j,s)$ are both still present in the clustering dynamics.
	Moreover, since $x_{i_0},x_j\in B(z,\tfrac{t}{2})$, we have $|x_{i_0}-x_j|\le t$. On the other hand,
	\[
	2s=\tau+\lambda>\tau\ge t,
	\]
	so $|x_{i_0}-x_j|<2s$ and therefore $B(x_{i_0},s)\cap B(x_j,s)\neq\emptyset$.
	Thus, at scale $s$, the ball $B(x_{i_0},s)$ intersects all the balls $\{B(x_j,s): j\in I\setminus\{i_0\}\}$, in particular at least $q(\beta)+1$ of them.
	
	By the definition of $\tau_{i_0}^{+,Z}$ as the disappearance time of $B(x_{i_0},\cdot)$ in the clustering algorithm, this is impossible: the algorithm would have removed $B(x_{i_0},\cdot)$ at or before scale $s<\tau$, contradicting $\tau_{i_0}^{+,Z}=\tau$.
	Hence $|I|\le q(\beta)+1$, proving \eqref{eq:geo5} for $A=\frac{1}{2}$, and therefore for all $A\ge \frac{1}{2}$.
	
	The proof of \eqref{eq:geo6} is identical, replacing $x_i,\tau_i^{+,Z}$ by $y_i,\tau_i^{-,Z}$.
	
\end{proof}

\begin{lemma}[Control of good-bad interactions]\label{lemma:energy goodbad}
	Let $(\vec{X}_N,\vec{Y}_N)\in (\Lambda^2)^N$ be such that $\sigma_N[\vec{X}_N,\vec{Y}_N]=\Id$. Suppose that the $r_i$ are all distinct. Let $d_{i,\bad}$ be as in \eqref{eq:dibad}. 
	
	Then, there exists a constant $C>0$ depending on $q(\beta)$ such that, for every $i\in I_\good$,
	\begin{equation}\label{eq:interaction gb}
		\sum_{j\in I_\bad}|v^Z_{ij}|\leq C\frac{r_i}{d_{i,\bad}}.
	\end{equation}
\end{lemma}

\medskip

\begin{proof}
	First, we write 
	\begin{equation}\label{eq:v+-}
		\sum_{j\in I_\bad}v^Z_{ij}=\sum_{j\in I_\bad}v^Z_{i,j,+}-\sum_{j\in I_\bad}v^Z_{i,j,-},
	\end{equation}
	where $v_{i,j,+}^Z$ and $v_{i,j,-}^Z$ are as in \eqref{def:v+} and \eqref{def:v-}. 
	
	Let $j\in I_\bad$. By Lemma~\ref{lemma:geo}, if $v_{i,j,+}^Z\neq 0$, then $d_{i,j}^+\leq 2\textcolor{black}{\tau_{j}^{+,Z}}$ and $\max(r_i,\lambda)\leq  \max(r_j,\lambda)$. \textcolor{black}{Since good and bad dipoles belong to distinct multipoles, $d_{i,j}^+\geq d_{ij}>M\lambda$.}
	Besides, by \eqref{eq:geo2}, we have $d_{i,j}^+\leq 2\tau_j^{+,Z}$. Moreover, $d_{i,j}^+\geq d_{i,\bad}\ge \max(r_i,\lambda)$ by \eqref{eq:geo1}, and since $i\in I_\good$, we have $\tau_i^{+,Z}=\tau_i^{-,Z}$.  We may thus apply Lemma~\ref{lemma:error'} and, in particular, \eqref{eqvijpp},  to find  
	\begin{equation}\label{eq:B+ini}
		\sum_{j\in I_\bad}|v_{i,j,+}^Z|\leq C\sum_{j\in I_\bad}\frac{r_i}{d_{i,j}^+}\textcolor{black}{\indic_{d_{i,\bad}\leq d_{i,j}^+\leq 2\tau_j^{+,Z}}}. 	\end{equation}
	As in the proof of Lemma~\ref{lemma:control sum intera}, let us define for all $t\geq 0$,
	\begin{equation}\label{def:Nit}
		\mc{E}_i(t)\coloneqq \Bigr\{ j\in I_\bad:d_{i,\bad}\leq d_{i,j}^+\leq 2\tau_j^{+,Z}, d_{ij}^+\in (t,2t)\Bigr\}\quad \text{and}\quad \mc{N}_i(t)\coloneqq | \mc{E}_i(t)|.
	\end{equation}
	By \eqref{eq:B+ini},
	\begin{equation}\label{eq522}
		\sum_{j\in I_\bad}|v_{i,j,+}^Z|\leq Cr_i\sum_{\substack{ j\in I_\bad:\\ d_{i,\bad}\leq d^+_{ij}\leq 2\tau_j^{+,Z} }} \int_{\tfrac{1}{2}d_{i,\bad} }^{2\ve_0\Cut}\frac{1}{t^2}\indic_{d_{i,j}^+\in (t,2t)}\dd t\leq Cr_i \int_{\tfrac{1}{2}d_{i,\bad}}^{2\ve_0\Cut}\frac{1}{t^2}\mc{N}_i(t)\dd t.
	\end{equation}
	Notice that
	\begin{equation*}
		\mc{E}_i(t)\subset \Bigr\{j\in I_\bad: x_j\in B(x_i,2t), \tau_{j}^{+,Z}\geq  \max(\tfrac{t}{2},\lambda)\Bigr\}\cup \Bigr\{j\in I_\bad: x_j\in B(y_i,2t), \tau_{j}^{+,Z}\geq  \max(\tfrac{t}{2},\lambda)\Bigr\}.
	\end{equation*}\textcolor{black}{For every $j\in\mc E_i(t)$, we have $\tau_j^{+,Z}\geq d_{i,j}^+/2>M\lambda/2>\lambda$.} Therefore, by \eqref{eq:geo5}, there exists a constant $C>0$ depending on $q(\beta)$ such that $\mc{N}_i(t)\leq C$. Inserting this bound into \eqref{eq522} gives, for a constant $C>0$ depending on $q(\beta)$,
	\begin{equation}\label{eq:ibb0}
		\sum_{j\in I_\bad}|v_{i,j,+}^Z|\leq C\frac{r_i}{d_{i,\bad}}.
	\end{equation}
	\textcolor{black}{Since exchanging all the $x_j$ and $y_j$ exchanges $\tau_j^{+,Z}$ and $\tau_j^{-,Z}$ and preserves $\{\sigma_N=\Id\}$, the same bound holds for $v_{i,j,-}^Z$, which proves \eqref{eq:interaction gb}.}
\end{proof}

\begin{lemma}[Control of bad-bad interactions]\label{lemma:badbad}
	Let $(\vec{X}_N,\vec{Y}_N)\in (\Lambda^2)^N$ be such that $\sigma_N[\vec{X}_N,\vec{Y}_N]=\Id$. Suppose that the $r_i$ are all distinct. Then, there exists a constant $C>0$ depending on $q(\beta)$ such that
	\begin{equation}\label{eq:interaction bb}
		\sum_{i,j\in I_\bad:i<j}v_{ij}^Z\geq -C|I_\bad|.
	\end{equation}   
\end{lemma}

\begin{proof}
	We begin by writing 
	\begin{equation*}
		\sum_{i,j\in I_\bad:i<j}v_{ij}^Z= \sum_{i\neq j\in I_\bad}v_{ij}^Z\indic_{r_j\geq r_i}.
	\end{equation*}
	Fix $i\in I_\bad$. First, decompose the sum of the interactions as follows:
	\begin{equation}\label{eq:deco}
		\sum_{j\in I_\bad:j\neq i}v_{ij}^Z\indic_{r_j\geq r_i}= \sum_{j\in I_\bad:j\neq i}v_{i,j,+}^Z\indic_{r_j\geq r_i}+\sum_{j\in I_\bad:j\neq i}(-v_{i,j,-}^Z)\indic_{r_j\geq r_i}.
	\end{equation}  
	We bound the first term. The second one is analogous. First, decompose the sum of the $v_{i,j,+}^Z$ into
	\begin{multline}\label{eq:vZZZ}
		\sum_{j\in I_\bad:j\neq i}v_{i,j,+}^Z\indic_{r_j\geq r_i}=\sum_{j\in I_\bad:j\neq i}v_{i,j,+}^Z\indic_{r_j\geq r_i}\indic_{d_{i,j}^+\geq r_i}\indic_{\tau_j^{+,Z}=\lambda} +   \sum_{j\in I_\bad:i\neq j}v_{i,j,+}^Z\indic_{d_{i,j}^+<r_i} \indic_{r_j\geq r_i} \\+ \sum_{j\in I_\bad:j\neq i}v_{i,j,+}^Z\indic_{d_{i,j}^+\geq \max(r_i,\lambda)}\indic_{\tau_j^{+,Z}>\lambda}\indic_{r_j\geq r_i} + \sum_{j\in I_\bad:j\neq i}v_{i,j,+}^Z\indic_{ r_i\le d_{i,j}^+<\lambda}\indic_{\tau_j^{+,Z}>\lambda}\indic_{r_j\geq r_i} .
	\end{multline}

	\paragraph{\bf{Step 1: the case $\tau_j^{+,Z}=\lambda$ and $d_{i,j}^+\geq r_i$}}
	
	By Newton's theorem, if \begin{equation}\label{eq:nonzero}\Bigl(\g_\lambda-\g * \delta_0^{(\tau_i^{+,Z})}*\delta_0^{(\lambda)}\Bigr) (x_i-x_j)\neq 0,
	\end{equation}then
	\begin{equation*}
		|x_i-x_j|\leq \tau_i^{+,Z}+\lambda.
	\end{equation*}
	Moreover, since $\g_\lambda=\g*\delta_0^{(\lambda)}*\delta_0^{(\lambda)}$, if $\Bigl(\g_\lambda-\g * \delta_0^{(\tau_i^{+,Z})}*\delta_0^{(\lambda)}\Bigr) (x_i-x_j)\neq 0$, then $\tau_i^{+,Z}>\lambda$, which implies $\tau_i^{+,Z}\leq \frac{r_i}{4}$ (since by Remark~\ref{remcrucialetau}, $\tau_i^{+,Z}\leq \max(\frac{1}{4}r_i,\lambda)$), so $r_i\geq 4\lambda$. Thus, if \eqref{eq:nonzero} holds, and if, in addition, $\tau_j^{+,Z}=\lambda$ and $d_{i,j}^+\geq r_i$, then $r_i\geq 4\lambda$ and 
	\begin{equation*}
		r_i\le d_{i,j}^+\leq |x_i-x_j|\leq \frac{r_i}{4}+\lambda,
	\end{equation*}
	which implies $r_i\leq \frac{4}{3}\lambda$ and $r_i\geq 4\lambda$. Clearly, this is impossible. Thus,
	\begin{equation*}
		\sum_{j\in I_\bad:j\neq i}\Bigl(\g_\lambda-\g * \delta_0^{(\tau_i^{+,Z})}*\delta_0^{(\lambda)}\Bigr) (x_i-x_j)\indic_{d_{i,j}^+\geq r_i}\indic_{\tau_j^{+,Z}=\lambda}=0.
	\end{equation*}
	Arguing similarly for the interaction of $y_i$ with the $x_j$'s, we obtain 
	\begin{equation}\label{eq:vr1}
		\sum_{j\in I_\bad:j\neq i}v_{i,j,+}^Z \indic_{d_{i,j}^+\geq r_i}\indic_{\tau_j^{+,Z}=\lambda} =0.
	\end{equation}

	\paragraph{\bf{Step 2: control for $d_{i,j}^+<r_i$}}
	\textcolor{black}{We keep $r_j\geq r_i$ throughout this step.} First, by Newton's theorem, if $|x_j-y_i|>\tau_i^{-,Z}+\tau_j^{+,Z}$, then
	\begin{equation*}
		\Bigl(\g_\lambda-\g * \delta_0^{(\tau_i^{-,Z})}*\delta_0^{(\tau_j^{+,Z})}\Bigr) (x_j-y_i)=0.
	\end{equation*}
	Therefore, 
	\begin{equation*}
		v_{i,j,+}^Z\indic_{|x_j-y_i|>\tau_i^{-,Z}+\tau_j^{+,Z}}= \Bigl(\g_\lambda-\g * \delta_0^{(\tau_i^{+,Z})}*\delta_0^{(\tau_j^{+,Z})}\Bigr) (x_i-x_j)\indic_{|x_j-y_i|>\tau_i^{-,Z}+\tau_j^{+,Z}}.
	\end{equation*}
	Since $\tau_i^{+,Z}\geq \lambda$ and $\tau_j^{+,Z}\geq \lambda$, we deduce from the monotonicity property of Lemma~\ref{lemma:mono} that
	\begin{equation}\label{eq:nonneg}
		\Bigl(\g_\lambda-\g * \delta_0^{(\tau_i^{+,Z})}*\delta_0^{(\tau_j^{+,Z})}\Bigr) (x_i-x_j)\geq 0.  
	\end{equation}
	Hence, combining the above two displays yields
	\begin{equation}\label{eq:v+case1}
		v_{i,j,+}^Z\indic_{|x_j-y_i|>\tau_i^{-,Z}+\tau_j^{+,Z}}\geq 0.  
	\end{equation}

	Now suppose that $|x_j-y_i|\leq \tau_i^{-,Z}+\tau_j^{+,Z}$. Recall that, by Lemma~\ref{lemma:stable}, $\{\sigma_N=\Id\}=\cap_{ij}\mc{A}_{ij}$. On the event $\mc{A}_{ij}$, by compatibility of the matching (see Lemma~\ref{lemma:Aij}), we have 
	$|x_j-y_i|\ge r_i$.
	It follows that on the event $\mc{A}_{ij}$,
	\begin{equation}\label{eq:key}
		r_i\leq |x_j-y_i|\leq \tau_i^{-,Z}+\tau_j^{+,Z}\leq \max(\tfrac{1}{4}r_i,\lambda)+\tau_j^{+,Z}.
	\end{equation}
	Thus, on the event $\mc{A}_{ij}$, if $v_{i,j,+}^Z\textcolor{black}{<}0$ and $d_{i,j}^+\leq r_i$, then 
	\begin{equation}\label{eq:v+case2}
		\tau_j^{+,Z}\geq \max(\tfrac{3r_i}{4}\indic_{r_i\geq 4\lambda},\lambda)\geq \max(\tfrac{1}{4}r_i,\lambda),
	\end{equation}
	which gives by \eqref{eq:key}
	\begin{equation}\label{eq:bd2tau}
		d_{i,j}^+\leq r_i\leq  2\tau_j^{+,Z}.
	\end{equation}
	$\bullet$ Suppose that $\tau_j^{+,Z}=\lambda$. Then, since the interaction between $x_i$ and $x_j$ is always nonnegative by \eqref{eq:nonneg}, if $v_{i,j,+}^Z<0$, then 
	\begin{equation}\label{eq:dis}
		\Bigl(\g_\lambda-\g * \delta_0^{(\tau_i^{-,Z})}*\delta_0^{(\lambda)}\Bigr) (y_i-x_j)\neq 0,
	\end{equation}
	which implies that $\tau_i^{-,Z}\neq \lambda$, so $r_i>4\lambda$. But \eqref{eq:dis} also implies that $r_i\leq |y_i-x_j|\leq \frac{r_i}{4}+\lambda$, which gives $r_i\leq \frac{4}{3}\lambda$. These two conditions are, therefore, incompatible. Thus,
	\begin{equation}\label{eq:vr21}
		\sum_{j\in I_\bad:j\neq i}v_{i,j,+}^Z\indic_{d_{i,j}^+\leq r_i}\indic_{\tau_j^{+,Z}=\lambda}\textcolor{black}{\indic_{r_j\geq r_i}}\geq 0.
	\end{equation}

	It remains to treat the case where $\tau_j^{+,Z}>\lambda$, assuming still that $d_{i,j}^+<r_i\le r_j$. Our goal is to show that there exists a constant $C>0$ such that 
	\begin{equation}\label{eq:minoC}
		v_{i,j,+}^Z\geq -C.
	\end{equation}
	{By the matching, we must have $|x_j-y_i|\ge r_i$,
		hence 
		$|x_j-x_i|\leq   r_i+ |x_j-y_i|\le  2|x_j-y_i|$.\\
		$\bullet$ Suppose that $\tau_j^{+,Z}\geq \frac{1}{8}r_i$. By definition
		\begin{equation*}
			v_{i,j,+}^Z= \g_\lambda(x_i-x_j) - \g_\lambda (y_i-x_j)+ f*\delta_0^{(\tau_i^{-,Z})}(y_i-x_j)
			-f*\delta_0^{(\tau_i^{+,Z})} (x_i-x_j),
		\end{equation*}
		where $f:= \g*\delta_0^{(\tau_j^{+,Z})}$.
		Using $|x_i-x_j|\le 2 |y_i-x_j|$ and the properties of $\g_\lambda$, we easily find that 
		$ \g_\lambda(x_i-x_j) - \g_\lambda (y_i-x_j)\ge - C$.
		Then we may write 
		\begin{equation*}
			f*\delta_0^{(\tau_i^{-,Z})}(y_i-x_j)
			-f*\delta_0^{(\tau_i^{+,Z})} (x_i-x_j)
			= \int \textcolor{black}{\bigl(}f(y_i-x_j-w) \delta_{0}^{(\tau_i^{-,Z})}(w)-f(x_i-x_j-w)  \delta_{0}^{(\tau_i^{+,Z})}(w)\textcolor{black}{\bigr)\dd w}.\end{equation*}
		Using that $f$ is radially decreasing and $ |y_i-x_j|\le |x_i-x_j|+r_i$, we find that 
		\begin{equation*}	
			\int \textcolor{black}{\bigl(}f(y_i-x_j-w) \delta_{0}^{(\tau_i^{-,Z})}(w)-f(x_i-x_j-w)  \delta_{0}^{(\tau_i^{+,Z})}(w)\textcolor{black}{\bigr)\dd w}
			\ge - \|\nab f\|_{L^\infty} (r_i+ 2\tau_i^{-,Z}+2\tau_i^{+,Z}).
		\end{equation*}Since by definition $ \|\nab f\|_{L^\infty} \le C (\tau_j^{+,Z})^{-1}$, 
		$\tau_j^{+,Z}\geq \max(\frac{1}{8}r_i,\lambda)$ and  $\tau_i^{\pm,Z}\leq \max(\tfrac{1}{4}r_i,\lambda)$,  we conclude that the right-hand side is also bounded below, and consequently \eqref{eq:minoC} holds.\\
		$\bullet$ Suppose that $\tau_j^{+,Z}\leq \frac{1}{8}r_i$ (which implies $r_i>8\lambda$ since $\tau_j^{+,Z}>\lambda$). 
		Since $\tau_i^{+,Z}\geq \textcolor{black}{\lambda}$, we always have, by Lemma~\ref{lemma:mono}, that 
		\begin{equation*}
			\Bigl(\g_\lambda-\g*\delta_0^{(\tau_i^{+,Z})}*\delta_0^{(\tau_j^{+,Z})}\Bigr)(x_j-x_i)\geq 0.
		\end{equation*}
		Moreover, we have $|x_j-y_i|\ge r_i>\tau_j^{+,Z}+\tau_i^{-,Z}$
		since $\tau_i^{-,Z}\le \max(\tfrac14 r_i,\lambda)$  for $i \in I_\bad$,
		hence by Newton's theorem
		\begin{equation*}
			\Bigl(\g_\lambda-\g*\delta_0^{(\tau_i^{-,Z})}*\delta_0^{(\tau_j^{+,Z})}\Bigr)(x_j-y_i)=0.
		\end{equation*}
		Thus, \eqref{eq:minoC} holds.}

	By \eqref{eq:v+case1}, \eqref{eq:v+case2}, \eqref{eq:bd2tau} and \eqref{eq:minoC},
	\begin{equation*}
		\sum_{j\in I_\bad:j\neq i}v_{i,j,+}^Z\indic_{d_{i,j}^+< r_i}\indic_{\tau_j^{+,Z}>\lambda}\indic_{r_j\geq r_i}\geq -C\sum_{j\in I_\bad:j\neq i}\indic_{\tau_j^{+,Z}\geq \tfrac{r_i}{4}}\indic_{\tau_j^{+,Z}>\lambda}\indic_{d_{i,j}^+\leq \min(2\tau_j^{+,Z},r_i)}\indic_{r_j\geq r_i}.  
	\end{equation*}
	By Lemma~\ref{lemma:geo} (see \eqref{eq:geo5}), there exists a constant $C>0$ depending on $q(\beta)$ such that 
	\begin{equation*}
		\sum_{j\in I_\bad:j\neq i}\indic_{\tau_j^{+,Z}\geq \tfrac{r_i}{4}}\indic_{\tau_j^{+,Z}>\lambda}\indic_{d_{i,j}^+\leq \min(2\tau_j^{+,Z},r_i)}\indic_{r_j\geq r_i}\leq C.
	\end{equation*}
	Thus, there exists a constant $C>0$ depending on $q(\beta)$ such that 
	\begin{equation}\label{eq:vr22}
		\sum_{j\in I_\bad:i\neq j}v_{i,j,+}^Z\indic_{d_{i,j}^+< r_i}\indic_{\tau_j^{+,Z}>\lambda}\indic_{r_j\geq r_i} \geq -C. 
	\end{equation}
	Combining this with \eqref{eq:vr21}, we conclude that there exists a constant $C>0$ such that
	\begin{equation}\label{eq:vr2}
		\sum_{j\in I_\bad:i\neq j}v_{i,j,+}^Z\indic_{d_{i,j}^+< r_i}\indic_{r_j\geq r_i} \geq -C. 
	\end{equation}

	\paragraph{\bf{Step 3: control for $d_{i,j}^+\geq \max(r_i,\lambda)$ and $\tau_j^{+,Z}>\lambda$}}
	
	By Newton's theorem, if $v_{i,j,+}^Z\neq 0$, then $d_{i,j}^+\leq \tau_j^{+,Z}+\max(\tau_i^{+,Z},\tau_i^{-,Z})$. Moreover, by the construction of the radii, 
	\begin{equation*}
		\max(\tau_i^{+,Z},\tau_i^{-,Z})\leq \max(\tfrac{1}{4}r_i,\lambda).
	\end{equation*}
	Hence, if $v_{i,j,+}^Z\neq 0$ and $d_{i,j}^+\geq r_i$, then  $r_i\le d_{i,j}^+\le \tau_j^{+,Z}+\max(\frac14r_i, \lambda)$ and 
	\begin{equation*}
		\tau_j^{+,Z}\geq \frac{r_i}{2}\indic_{r_i\geq 4\lambda}+\lambda\indic_{r_i\leq 4\lambda}.
	\end{equation*}
	In particular, $\tau_j^{+,Z}\geq  \max(\frac{r_i}{4},\lambda)\ge \max(\tau_i^{+,Z}, \tau_i^{-,Z})$ and 
	\begin{equation}\label{eq:HHH}
		d_{i,j}^+\leq 2\tau_j^{+,Z}.
	\end{equation}
	Moreover, by \eqref{diffvijcassuppl2} in Lemma~\ref{lemma:error'}, we have
	\begin{equation*}
		|v_{i,j,+}^Z| \le C \left( \frac{r_i}{\max(d_{i,j}^+,\lambda)} + \frac{r_i}{\tau_j^{+,Z}} + \frac{\max(\tau_i^{+,Z}, \tau_i^{-,Z})^2}{(\tau_j^{+,Z})^2} \right) \indic_{d_{i,j}^+\leq   \max(\tau_i^{+,Z},\tau_i^{-,Z}) + \tau_j^{+,Z}}.
	\end{equation*}
	Using that $\tau_j^{+,Z}\ge\hal d_{i,j}^+$ and that $\tau_i^{\pm,Z}\le \max(\frac{r_i}{4},\lambda)$, we deduce that 
	\begin{equation*}
		|v_{i,j,+}^Z|\leq C\frac{\max(r_i,\lambda)}{d_{i,j}^+}.
	\end{equation*}
	Combining this with \eqref{eq:HHH}, we find that 
	\begin{equation*}
		\sum_{j\in I_\bad:j\neq i}|v_{i,j,+}^Z|\indic_{d_{i,j}^+\geq \max(r_i,\lambda),\tau_j^{+,Z}>\lambda }\indic_{r_j\geq r_i }\leq C\sum_{j\in I_\bad:j\neq i}\frac{\max(r_i,\lambda)}{d_{i,j}^+}\indic_{\max(r_i,\lambda)\leq d_{i,j}^+\leq 2\tau_j^{+,Z},\tau_j^{+,Z}>\lambda }.
	\end{equation*}
	Define for all $t\geq 0$,
	\begin{equation*}
		\mc{E}_i(t)\coloneqq \Bigr\{ j\in I_\bad:\max(r_i,\lambda) \leq d_{i,j}^+\leq 2\tau_j^{+,Z},\tau_j^{+,Z}>\lambda, d_{ij}^+\in (t,2t)\Bigr\}\quad \text{and}\quad \mc{N}_i(t)\coloneqq | \mc{E}_i(t)|.
	\end{equation*}
	We get
	\begin{equation*}
		\sum_{j\in I_\bad:j\neq i}\frac{\max(r_i,\lambda)}{d_{ij}^+}\indic_{\max(r_i,\lambda)\leq d_{ij}^+\leq 2\tau_j^{+,Z},\tau_j^{+,Z}>\lambda }\leq C\max(r_i,\lambda)\int_{\tfrac{1}{2}\max(r_i,\lambda)}^{2\ve_0\Cut}\frac{1}{t^2}\mc{N}_i(t)\dd t.
	\end{equation*}
	By Lemma~\ref{lemma:geo}, estimate \eqref{eq:geo5}, there exists a constant $C>0$ depending on $q(\beta)$ such that $\mc{N}_i(t)\leq C$. Hence, there exists $C>0$ depending on $q(\beta)$ such that
	\begin{equation}\label{eq:vr3}
		\sum_{j\in I_\bad:j\neq i}|v_{i,j,+}^Z|\indic_{d_{i,j}^+\geq \max( r_i,\lambda)}\indic_{\tau_j^{+,Z}>\lambda } \indic_{r_j\geq r_i }\leq C.
	\end{equation}
	
	\paragraph{\bf{Step 4: control for $r_i\le d_{i,j}^+<\lambda$ and $\tau_j^{+,Z}>\lambda$}}
	In that case, since for $i \in I_\bad$ we have  $\lambda\le \tau_i^{\pm, Z}\le \max(\frac14r_i,\lambda)$ (see Remark~\ref{remcrucialetau}), we have $\tau_i^{\pm, Z}=\lambda$. Inserting this, $r_i\le d_{i,j}^+<\lambda$ and $\tau_j^{+,Z}>\lambda$ into \eqref{diffvijcassuppl2}, we find that  $|v_{i,j,+}^Z|\le C$ and thus
	$$\sum_{j\in I_\bad:j\neq i}|v_{i,j,+}^Z|\indic_{ r_i\le d_{i,j}^+<\lambda}\indic_{\tau_j^{+,Z}>\lambda } \indic_{r_j\geq r_i }\leq C \sum_{j\in I_\bad:j\neq i}\indic_{ r_i\le d_{i,j}^+<\lambda}\indic_{\tau_j^{+,Z}>\lambda } \indic_{r_j\geq r_i }.$$
	Then, applying the estimate \eqref{eq:geo5} of Lemma~\ref{lemma:geo} shows that the sum on the right-hand side is bounded by a constant $C>0$ depending on $q(\beta)$, hence
	\begin{equation}\label{eq:vr4}
		\sum_{j\in I_\bad:j\neq i}|v_{i,j,+}^Z|\indic_{r_i\leq d_{i,j}^+<\lambda}\indic_{\tau_j^{+,Z}>\lambda } \indic_{r_j\geq r_i }\leq C.
	\end{equation}
	
	\paragraph{\bf{Step 5: conclusion}}
	Combining \eqref{eq:vr1}, \eqref{eq:vr2}, \eqref{eq:vr3}  and \eqref{eq:vr4} gives the existence of a constant $C>0$ depending on $q(\beta)$ such that
	\begin{equation*}
		\sum_{j\in I_\bad:j\neq i}v_{i,j,+}^Z\indic_{r_j\geq r_i}\geq -C.  
	\end{equation*}
	\textcolor{black}{Since exchanging all positive and negative charges sends $v_{i,j,+}^Z$ to $-v_{i,j,-}^Z$ and preserves $\{\sigma_N=\Id\}$, the same bound holds for the second sum in \eqref{eq:deco}, which proves the result.}
	
\end{proof}

Let us finally record here  the main energy lower bound that we will use.
Let $I_1\coloneqq I_\good \cup \{i\in  I_\bad, \frac14r_i<\lambda\}$, and $I_2\coloneqq  I_\bad\cap \{ i\in  I_\bad, \frac14r_i\ge \lambda\}$. Then $[N]=I_1\sqcup I_2$ and, by Remark~\ref{remcrucialetau}, the sets $I_1$ and $I_2$ satisfy the assumptions of  the energy lower bound of Proposition~\ref{prop:mino}. Thus, we know that  there exists $L>0$, depending on $\beta$, such that, if $\sigma_N[\vec{X}_N,\vec{Y}_N]=\Id$, 
\begin{equation*}		\F_\lambda(\vec{X}_N,\vec{Y}_N)\geq -\sum_{i\in I_1}\g_\lambda(x_i-y_i)-\frac{1}{2}\sum_{i\in  I_2}(\g(\tau_i^{+,Z})+\g(\tau_i^{-,Z}))+\sum_{i<j}v_{ij}^Z-\frac{L}{\beta}\sum_{i\in I_1} \Bigl(\frac{|x_i-y_i|}{\tau_i^{+,Z}}\Bigr)^2-\frac{L}{\beta}|I_2|,
\end{equation*}
where the interactions $v_{ij}^Z$ and the radii $\tau_i^{\pm,Z}$ are given by Definition~\ref{def:vijZ}. 
\textcolor{black}{Here, $L$ also absorbs the term $\kappa|I_2|$ of \eqref{eq:minoF}. For $i\in I_\bad\cap I_1$, we have $\tau_i^{\pm,Z}=\lambda$, $r_i^2/\lambda^2<16$ and $|\g_\lambda(x_i-y_i)-\g(\lambda)|\leq C$, so that, after increasing $L$ again, we obtain} 
\begin{multline}\label{eq:b1}
	\F_\lambda(\vec{X}_N,\vec{Y}_N)\geq -\sum_{i\in I_\good}\g_\lambda(x_i-y_i)-\frac{1}{2}\sum_{i\in  I_\bad}(\g(\tau_i^{+,Z})+\g(\tau_i^{-,Z}))+\sum_{i<j}v_{ij}^Z\\-\frac{L}{\beta}\sum_{i\in I_\good} \Bigl(\frac{|x_i-y_i|}{\tau_i^{+,Z}}\Bigr)^2-\frac{L}{\beta}\textcolor{black}{|I_\bad|}.
\end{multline}

\subsection{Definitions: hierarchical model and  activities with frozen bad points}\label{sub:multi u}

We now introduce the hierarchical multipole model that we will use in this section.

\textcolor{black}{In the rest of this section, frozen configurations $Z$ satisfy $I_\bad^{V_\bad}=V_\bad$ and belong to $\mc A_{ij}$ for all $i,j\in V_\bad$ with $i<j$. We compute the bad sets and the radii from $Z$, keeping $N$ in the thresholds of Definition~\ref{def:type2}.}

\begin{definition}\label{def:mult upper}
	Let $V_\bad\subset [N]$ and $Z\in (\Lambda^{2})^{|V_{\bad}|}$. Let $X$ be a subpartition of $[N]$ such that $V_X\cap V_\bad=\emptyset$. Let us define
	\begin{equation}\label{def:IgoodZ}
		I_\good^Z\coloneqq \bigcap_{j\in I_\bad^1\cup I_\bad^2}\{i\in [N]:d_{ij}> M\max(r_i,\lambda)\}\cap \bigcap_{j\in I_\bad^3 }\{i\in [N]:d_{ij}> M\max(\min(r_i,r_j),\lambda)\}.
	\end{equation}
	Note that $I_\good\subset I_\good^Z$.	We let $\Psf_X^{+,\ve_0,Z}$ be the probability measure
	\begin{multline}\label{eq:defPXw u}
		\dd \Psf^{+,\ve_0,Z}_X \propto \prod_{S\in X}\Bigl(\indic_{\mc{B}_S}\prod_{ i,j\in S:i<j}e^{-\beta v^Z_{ij}}\indic_{\mc{A}_{ij}}\Bigr) \\ \times \prod_{i\in V_X}\Bigl(e^{\beta \g_\lambda(x_i-y_i)+\textcolor{black}{Lr_i^2/(\tau_i^Z)^2}}\indic_{|x_i-y_i|\leq \ve_0\Cut}\indic_{i\in I_\good^Z}  \Bigr) \prod_{i\in V_X}\dd x_i \dd y_i,
	\end{multline}
	where $L$ is the constant in \eqref{eq:b1},
	and we let $\Msf_{\ve_0}^{+,Z}(X)$ be the normalization constant
	\begin{equation}\label{def:Mw u}
		\Msf_{\ve_0}^{+,Z}(X)\coloneqq \prod_{S\in X}\dE_{(\mu_{\beta, \lambda, \ve_0})^{\otimes |S| }}\left[\indic_{\mc{B}_S}\prod_{ i,j\in S:i<j}e^{-\beta v^Z_{ij}}\indic_{\mc{A}_{ij}}\prod_{i\in S}\textcolor{black}{e^{Lr_i^2/(\tau_i^Z)^2}}\indic_{i\in I_\good^Z} \right].
	\end{equation}
\end{definition}

\begin{definition}\label{def:iinfty}
	Let $V_\bad\subset [N]$ and $Z\in (\Lambda^2)^{|V_\bad|}$. For every $i\in [N]\setminus V_\bad$, we let
	\begin{equation*}
		f_{i\infty}^{v^Z}\coloneqq \prod_{j\in V_\bad}e^{-\beta v_{ij}^Z}-1.
	\end{equation*}
\end{definition}

\begin{definition}\label{def:activity upper}
	Let $V_\bad\subset [N]$ and set $V_\good\coloneqq [N]\setminus V_\bad$. Let $Z\in (\Lambda^2)^{|V_\bad|}$. 
	
	We define $\Ksf_{\ve_0}^{+,Z}:\mathbf{\Pi}_\sub(V_\good)\to \dR$ to be given for every $X\in \mathbf{\Pi}_\sub(V_\good)$ \textcolor{black}{with $|X|\geq2$} by
	\begin{multline*}
		\Ksf_{\ve_0}^{+,Z}(X)\coloneqq \sum_{V^\infty\subset V_X} \sum_{n=0}^\infty \frac{1}{n!}\sum_{\substack{X_1,\ldots,X_n\subset X\\ \mathrm{disjoint} }} \sum_{E_1\in \mathsf{E}^{X_1}}\cdots\sum_{E_n\in \mathsf{E}^{X_n}} \sum_{F\in \textcolor{black}{\mathsf{E}_0}^{\Coarse_X(X_1,\ldots,X_n)} }\\ \dE_{\Psf_{X}^{+,\ve_0,Z}}\left[\prod_{ij\in E_1\cup \cdots \cup E_n}f^{v^Z}_{ij}\prod_{ij\in \cup_l\mc{E}^\inter(X_l)}\indic_{\mc{B}_{ij}^c}\prod_{ij\in F}(-\indic_{\mc{B}_{ij}})\prod_{i\in V^\infty }f_{i\infty}^{v^Z}\right],
	\end{multline*}
	where $\mc{E}^\inter(X)$ is as in Definition~\ref{def:part notions}, $\mathsf{E}^X$ as in Definition~\ref{def:quotient}, $\Coarse_X(X_1,\ldots,X_n)$ as in \eqref{def:merge}, $f_{ij}^{v^Z}$ as in Definition~\ref{def:mayer}, and $f_{i\infty}^{v^Z}$ as in Definition~\ref{def:iinfty}.

	\textcolor{black}{For $|X|\leq1$, we set $\Ksf_{\ve_0}^{+,Z}(X)=0$.}

    {\color{black}
    For every nonempty $S\subset V_\good$, let
    \begin{equation}\label{def:external normalization}
        g_{\ve_0}^Z(S)\coloneqq \dE_{\Psf_{\{S\}}^{+,\ve_0,Z}}\left[e^{-\beta\sum_{i\in S,j\in V_\bad}v_{ij}^Z}\right],
        \qquad
        \widehat{\Msf}_{\ve_0}^{+,Z}(S)\coloneqq \Msf_{\ve_0}^{+,Z}(S)g_{\ve_0}^Z(S).
    \end{equation}
    For every subpartition $X$ of $V_\good$, let
    \begin{equation}\label{def:normalized upper}
        \widehat{\Msf}_{\ve_0}^{+,Z}(X)\coloneqq\prod_{S\in X}\widehat{\Msf}_{\ve_0}^{+,Z}(S),
        \qquad
        \widehat{\Ksf}_{\ve_0}^{+,Z}(X)\coloneqq\frac{\Ksf_{\ve_0}^{+,Z}(X)}{\prod_{S\in X}g_{\ve_0}^Z(S)}.
    \end{equation}
    }
\end{definition}

\subsection{Isolating good points and perturbative expansion around the hierarchical model}\label{sub:isolating}

In this subsection, we rewrite the partition function of good points, defined in Definition~\ref{def:bad points}, as a cluster-expansion series. This process largely follows the steps of Subsection~\ref{sub:pert}.

\begin{lemma}\label{lemma:start upper}
	Let $\beta\in (2,\infty)$ and $p(\beta)$ be as in Definition~\ref{def:pbeta}. Let $\mc{N}_1,\ldots,\mc{N}_{p(\beta)}$ be the numbers of multipoles of sizes $1,\ldots,p(\beta)$ in $I_\good$. Let \textcolor{black}{$n_1,\ldots,n_{p(\beta)}\in\mathbb N$ with $\sum_{k=1}^{p(\beta)}kn_k\leq N$}
	satisfy, for every $k\in \{2,\dots, p(\beta)\}$,
	\begin{equation}\label{bornesnkbs}
		n_k\le \ve_0^{-\alpha(\beta)}\lambda^{2(k-1)}N,\end{equation}
	where we recall that $\alpha(\beta)$ is as in \eqref{def:alphabeta}. 
	Set $N'=n_1+2n_2+\cdots+p(\beta)n_{p(\beta)}$.	
	Define
	\begin{equation}\label{eq:defevA}
		\mc{A}\coloneqq \{\mc{N}_1=n_1,\ldots,\mc{N}_{p(\beta)}=n_{p(\beta)},|I_\good|=N'\}.
	\end{equation}
	Let $V_\bad\subset [N]$ have cardinality $N-N'$ and set $V_\good\coloneqq [N]\setminus V_\bad$. 
	Let $\pi$ be a partition of $V_\good$ such that, for every $k=1,\ldots,p(\beta)$,
	\begin{equation*}
		|\{S\in \pi:|S|=k\}|=n_k
	\end{equation*}
	and such that, for every $k>p(\beta)$,
	\begin{equation*}
		|\{S\in \pi:|S|=k\}|=0.
	\end{equation*}
	\textcolor{black}{Recall $C_{\beta,\lambda,\ve_0}$ from \eqref{def:Clambda} and $\widehat{\Msf}_{\ve_0}^{+,Z},\widehat{\Ksf}_{\ve_0}^{+,Z}$ from Definition~\ref{def:activity upper}.}
	
	Then, \textcolor{black}{there exists $C>0$ depending only on $\beta$ and $p(\beta)$ such that}
	\begin{multline}\label{numero539}
		(NC_{\beta,\lambda,\ve_0})^{-N} \int_{\mc{A}\cap\{\sigma_N=\Id\}}e^{-\beta \F_\lambda}\leq \textcolor{black}{e^{C|V_\bad|}}\binom{N}{N'}\frac{(N')!}{1^{n_1}(2!)^{n_2}\cdots (p(\beta)!)^{n_{p(\beta)}}n_1!\cdots n_{p(\beta)}!} \\ \times \frac{1}{(NC_{\beta,\lambda,\ve_0})^{|V_\bad|}}\int_{(\Lambda^{2})^{|V_{\bad}|}} \indic_{\{ I_\bad= V_\bad\}}W(Z) \prod_{i,j\in V_\bad:i<j}(e^{-\beta v^Z_{ij}}\indic_{\mc{A}_{ij}})\\ \times\prod_{i\in V_\bad}\frac{1}{(\tau_i^{+,Z})^{\frac{\beta}{2}}}\frac{1}{(\tau_i^{-,Z})^{\frac{\beta}{2}}}  \prod_{S\in \pi}\textcolor{black}{\widehat{\Msf}_{\ve_0}^{+,Z}}(S)\dd Z,
	\end{multline}
	where
	\begin{equation}\label{def:H upper}
		W:Z\in (\Lambda^{2})^{|V_{\bad}|}\mapsto \sum_{n=0}^{\infty}\frac{1}{n!}\sum_{\substack{X_1,\ldots,X_n\in \mc{P}(\pi)\\ \mathrm{disjoint} }} \textcolor{black}{\widehat{\Ksf}_{\ve_0}^{+,Z}}(X_1)\cdots \textcolor{black}{\widehat{\Ksf}_{\ve_0}^{+,Z}}(X_n).
	\end{equation}
\end{lemma}

\begin{remark}\label{rem:bad-locality}
\begingroup\color{black}
Almost everywhere on $\{I_\bad=V_\bad\}$, the bad sets and the radii computed from $Z$ alone coincide with the original ones, if we keep the ambient $N$ in the thresholds of $I_\bad^{2,2,k}$. Indeed, every overcrowded cluster has $q(\beta)+1$ labels, so it lies in $V_\bad$, and the labels outside $I_\bad^1$ are removed as singletons at $r_i/4$. As observed after Definition~\ref{def:continuous clustering}, deleting the good singleton clusters therefore preserves the other clusters and their radii. Hence every block of $\Pi_\mult^U$ contained in $V_\bad$ is unchanged, and if $I_\bad^{2,2,k}\neq\emptyset$, then all the blocks of size $k$ are contained in $V_\bad$. The definitions of types $2$ and $3$ then give $I_\bad^{V_\bad}=V_\bad$.
\endgroup
\end{remark}

\begin{proof}Let $p\coloneqq p(\beta)$.
	\paragraph{\bf{Step 1: starting point}}
	Recall $I_\good$ and $I_\bad$ defined in Definition~\ref{def:bad points}. For every $V\subset [N]$, we let $I_\bad^V$ be the set of bad points when computed only among $(x_i,y_i)_{i\in V}$.

	Let $\Pi_\mult^{I_\good}$ be the partition of $I_\good$ into multipoles (see Definition~\ref{def:multipoles mc}). Recall that, by definition, $I_\good\subset [N]\setminus I_\bad^2$ and therefore every $S\in \Pi_\mult^{V_\good}$ such that $S\cap I_\good\neq \emptyset$ satisfies $S\subset I_\good$ and $|S|\leq p$. Hence $|I_\good|=\mc{N}_1+2\mc{N}_2+\cdots +p\mc{N}_{p}$ and, for $\mc{A}$ to be nonempty, we need $N'=n_1+2n_2+\cdots +p n_{p}$. \medskip
	
	\paragraph{\bf{Step 2: fixing dipoles, good dipoles and good multipoles}}
	We wish to give an upper bound on $\int_{\mc{A}} e^{-\beta \F_\lambda}$. First, since particles are indistinguishable, one may write 
	\begin{equation*}
		\int_{\mc{A}} e^{-\beta \F_\lambda}=N!\int_{\mc{A}\cap \{\sigma_N=\Id\}} e^{-\beta \F_\lambda}.
	\end{equation*}
	\textcolor{black}{\edgreen{Off the null set of Subsection~\ref{sub:stable} and of Remark~\ref{rem:bad-locality}}, simultaneous relabeling of the dipoles commutes with the bad-set construction. Thus,}
	\begin{equation}\label{eqFcond1}
		\int_{\mc{A}\cap \{\sigma_N=\Id\}} e^{-\beta \F_\lambda}=\binom{N}{N'}\int_{\mc{A}\cap \{\sigma_N=\Id, I_\bad=V_\bad\} } e^{-\beta \F_\lambda}.
	\end{equation}
	By \eqref{eq:choicespi},
	\begin{equation}\label{eqFcond2}
		\int_{\mc{A}\cap \{\sigma_N=\Id, I_\bad=V_\bad\}} e^{-\beta \F_\lambda}= \frac{(N')!}{1^{n_1}(2!)^{n_2}\cdots (p!)^{n_{p}}n_1!\cdots n_{p}!} \int_{\{\sigma_N=\Id,I_\bad=V_\bad,\Pi_\mult^{V_\good}=\pi\} } e^{-\beta \F_\lambda}.
	\end{equation}
	
	\paragraph{\bf{Step 3: rewriting the event $\{\sigma_N=\Id,I_\bad=V_\bad,\Pi_\mult^{V_\good}=\pi\}$ as an intersection of simpler events }}
	Recalling the notation $I_\good^Z$ from \eqref{def:IgoodZ}, in view of the definition of $I_\bad$ and of \eqref{bornesnkbs}, by Remark~\ref{remarkgoodbad}, item (2)\edgreen{, and by Remark~\ref{rem:bad-locality}},
	we have
	\begin{multline}\label{eq:decevent}
		\{\sigma_N=\Id\}\cap \{I_\bad=V_\bad\}\textcolor{black}{\subseteq}\{\sigma_N=\Id\}\cap \{I_\bad^{V_\bad}=V_\bad\}\\ \cap \bigcap_{i\in V_\good}\Bigl(\{|x_i-y_i|\leq \ve_0\Cut\}\cap \{|[i]^{\Pi_\mult^{V_\good} }|\leq p\}\Bigr) \cap \bigcap_{i\in V_\good}\{i\in I_\good^Z\} .
	\end{multline}
	
	Moreover, 
	\begin{equation}\label{eq:remplacevent}
		\{\sigma_N=\Id\}\cap  \{\Pi_\mult^{V_\good}=\pi\}=  \{\sigma_N=\Id\}\cap \bigcap_{S\in \pi}\mc{B}_S\bigcap_{ij\in \mc{E}^\inter(\pi)}\mc{B}_{ij}^c.
	\end{equation}
	Also recall that 
	\begin{equation}\label{eq:siAA}
		\{\sigma_N=\Id\}=\bigcap_{ij:i<j}\mc{A}_{ij}.   
	\end{equation}

	\paragraph{\bf{Step 4: isolating good points}}
	
	Combining \eqref{eq:decevent}--\eqref{eq:siAA}, \eqref{eqFcond1}, \eqref{eqFcond2} and \eqref{eq:b1}, we obtain that, for some $C>0$ depending only on \textcolor{black}{$\beta$ and $p(\beta)$},
	\begin{multline*}
		\int_{\mc{A}\cap \{\sigma_N=\Id\}}e^{-\beta \F_\lambda}\leq e^{C(N-N')}\binom{N}{N'}\frac{(N')!}{1^{n_1}(2!)^{n_2}\cdots (p!)^{n_{p}}n_1!\cdots n_{p}!}\\ \times \int_{(\Lambda^{2})^{|V_{\bad}|} } \indic_{I_\bad^{V_\bad}=V_\bad}\prod_{i,j\in V_\bad:i<j}(e^{-\beta v_{ij}^Z}\indic_{\mc{A}_{ij}}) W'(Z) \prod_{i\in V_\bad}\frac{1}{(\tau_i^{+,Z})^{\frac{\beta}{2}}}\frac{1}{(\tau_i^{-,Z})^{\frac{\beta}{2}}} \dd Z,
	\end{multline*}
	where, for every bad point configuration $Z\in (\Lambda^2)^{|V_\bad|}$,
	\begin{multline*}
		W'(Z)\coloneqq  \int_{(\Lambda^2)^{(N-|V_\bad|)}}\left( \prod_{i\in V_\bad,j\in V_\good } e^{-\beta v_{ij}^Z}\indic_{\mc{B}_{ij}^c\cap\mc{A}_{ij} } \right)\prod_{ij\in \mc{E}^\inter(\pi)}e^{-\beta v_{ij}^Z}\indic_{\mc{B}_{ij}^c\cap\mc{A}_{ij} }\\ \times \prod_{S\in \pi}\left(\indic_{\mc{B}_S}\prod_{ i,j\in S:i<j}e^{-\beta v_{ij}^Z}\indic_{\mc{A}_{ij} }\right) \prod_{i\in V_\good}\left(\indic_{|x_i-y_i|\leq \ve_0\Cut}\indic_{|[i]^{\Pi_\mult^{V_\good}}|\leq p}\indic_{i\in I_\good^Z }  \right) \\ \times \prod_{i\in V_\good}\left(e^{\beta \g_\lambda(x_i-y_i)+L(\frac{|x_i-y_i|}{\tau_i^Z})^2}\right)\prod_{i\in V_\good}\dd x_i \dd y_i.
	\end{multline*}
	Here we recall the notation $\mc{E}^\inter(\pi)$ from Definition \textcolor{black}{\ref{def:part notions}}. By Remark~\ref{remark:the inclusion}, we have $\mc{B}_{ij}^c\subset \mc{A}_{ij}.$ Therefore, 
	\begin{multline*}
		W'(Z)= \int_{(\Lambda^2)^{(N-|V_\bad|)}}\left( \prod_{i\in V_\bad,j\in V_\good } e^{-\beta v_{ij}^Z}\indic_{\mc{B}_{ij}^c} \right)\prod_{ij\in \mc{E}^\inter(\pi)}e^{-\beta v_{ij}^Z}\indic_{\mc{B}_{ij}^c\cap \mc{A}_{ij} }\\ \times\prod_{S\in \pi}\left(\indic_{\mc{B}_S}\prod_{ i,j\in S:i<j}e^{-\beta v_{ij}^Z}\indic_{\mc{A}_{ij} }\right) \prod_{i\in V_\good}\left(\indic_{|x_i-y_i|\leq \ve_0\Cut}\indic_{|[i]^{\Pi_\mult^{V_\good}}|\leq p}\indic_{i\in I_\good^Z }  \right)\\ \times\prod_{i\in V_\good}\left(e^{\beta \g_\lambda(x_i-y_i)+L(\frac{|x_i-y_i|}{\tau_i^Z})^2}\right)\prod_{i\in V_\good}\dd x_i \dd y_i.
	\end{multline*}
	
	Since we are restricted to the event where $\Pi_\mult^{V_\good}=\pi$ and since $\pi$ has no elements of cardinality strictly larger than $p$, one may simplify $W'(Z)$ into
	\begin{multline*}
		W'(Z)\coloneqq  \int_{(\Lambda^2)^{(N-|V_\bad|)}}\left( \prod_{i\in V_\bad,j\in V_\good } e^{-\beta v_{ij}^Z} \right)\prod_{ij\in \mc{E}^\inter(\pi)}e^{-\beta v_{ij}^Z}\indic_{\mc{B}_{ij}^c\cap\mc{A}_{ij} }\prod_{S\in \pi}\left(\indic_{\mc{B}_S}\prod_{ i,j\in S:i<j}e^{-\beta v_{ij}^Z}\indic_{\mc{A}_{ij} }\right)\\ \times  \prod_{i\in V_\good}\left(e^{\beta \g_\lambda(x_i-y_i)+L(\frac{|x_i-y_i|}{\tau_i^Z})^2}\indic_{|x_i-y_i|\leq \ve_0\Cut} \indic_{i\in I_\good^Z }\right)\prod_{i\in V_\good}\dd x_i \dd y_i.
	\end{multline*}

	\paragraph{\bf{Step 5: rewriting $W'(Z)$ as a cluster-expansion series}}

	As in Subsection~\ref{sub:pert}, we expand the interaction between dipoles in distinct multipoles, the terms $\indic_{\mc{B}_{ij}^c}$, and the interaction between good and bad points.

	Recalling $f_{ij}^{v^Z}$ from Definition~\ref{def:mayer} and $f_{i\infty}^{v^Z}$ from Definition~\ref{def:iinfty}, one can write
	\begin{equation*}
		W'(Z)=\sum_{E\subset \mc{E}^\inter(\pi)}\sum_{V^\infty\subset V_\good} I(E,V^\infty),
	\end{equation*}
	where 
	\begin{multline*}
		I(E,V^\infty)\coloneqq  \int_{(\Lambda^2)^{(N-|V_\bad|)}} \prod_{ij\in E}f_{ij}^{v^Z}\prod_{i\in V^\infty}f_{i\infty}^{v^Z} \prod_{S\in \pi}\left(\indic_{\mc{B}_S}\prod_{ i,j\in S:i<j}e^{-\beta v_{ij}^Z}\indic_{\mc{A}_{ij}}\right)\\ \times \prod_{ij\in \mc{E}^\inter(\pi)}\indic_{\mc{B}_{ij}^c} \prod_{i\in V_\good}\left(e^{\beta \g_\lambda(x_i-y_i)+L(\frac{|x_i-y_i|}{\tau_i^Z})^2}\indic_{|x_i-y_i|\leq \ve_0\Cut }\indic_{i\in I_\good^Z } \right)\prod_{i\in V_\good}\dd x_i \dd y_i.
	\end{multline*}
	
	Resumming according to the connected components relative to $\pi$ of the graph $(V_\good,E)$ gives
	\begin{equation*}
		W'(Z)=\sum_{n=0}^\infty \frac{1}{n!}\sum_{\substack{X_1,\ldots,X_n\subset \pi\\ \mathrm{disjoint}}}\sum_{E_1\in \mathsf{E}^{X_1}}\cdots \sum_{E_n\in \mathsf{E}^{X_n}}\textcolor{black}{\sum_{V^\infty\subset V_\good}}I(\cup_{l=1}^nE_l,\textcolor{black}{V^\infty}).
	\end{equation*}
	Next, we use
	\begin{equation*}
		\prod_{ij\in \mc{E}^\inter(\pi)}\indic_{\mc{B}_{ij}^c} =\prod_{ij\in \cup_l \mc{E}^\inter(X_l)} \indic_{\mc{B}_{ij}^c} \prod_{ij\in \mc{E}^\inter(\Coarse_{\pi}(X_1,\ldots,X_n))}\indic_{\mc{B}_{ij}^c}.
	\end{equation*}
	For every $ij\in \mc{E}^\inter(\Coarse_\pi(X_1,\ldots,X_n))$, we write $\indic_{\mc{B}_{ij}^c}=1-\indic_{\mc{B}_{ij}}$ to obtain
	\begin{equation*}
		\prod_{ij\in \mc{E}^\inter(\Coarse_{\pi}(X_1,\ldots,X_n))}\indic_{\mc{B}_{ij}^c}=\sum_{F\subset \mc{E}^\inter(\Coarse_{\pi}(X_1,\ldots,X_n))}\prod_{ij\in F}(-\indic_{\mc{B}_{ij}}).   
	\end{equation*}
	This gives 
	\begin{multline*}
		W'(Z)\\
		=\sum_{n=0}^\infty \frac{1}{n!}\sum_{\substack{X_1,\ldots,X_n\subset \pi\\ \mathrm{disjoint}}}\sum_{E_1\in \mathsf{E}^{X_1}}\cdots \sum_{E_n\in \mathsf{E}^{X_n}}\textcolor{black}{\sum_{V^\infty\subset V_\good}} \sum_{F\subset \mc{E}^\inter(\Coarse_{\pi}(X_1,\ldots,X_n))}\textcolor{black}{I_{\mathrm{cl},\ve_0}^{+,Z}(\pi;\mathbf X,\mathbf E,F,V^\infty)}.
	\end{multline*}
	Here, denoting by $X_1,\ldots,X_n$ the connected components relative to $\pi$ of $(V_\good,E)$ with at least two multipoles, \textcolor{black}{we write $\mathbf X=(X_l)$, $\mathbf E=(E_l)$ and $E=\bigcup_lE_l$. The last argument is the set $V^\infty$ of the indices $i$ that carry a factor $f_{i\infty}^{v^Z}$, and we omit it when it is empty.} We set
	\begin{multline*}
		\textcolor{black}{I_{\mathrm{cl},\ve_0}^{+,Z}(\pi;\mathbf X,\mathbf E,F,V^\infty)}\\\coloneqq  \int_{(\Lambda^2)^{(N-|V_\bad|)}} \prod_{ij\in E}f_{ij}^{v^Z}\prod_{i\in V^\infty}f_{i\infty}^{v^Z} \prod_{S\in \pi}\Bigl(\indic_{\mc{B}_S}\prod_{ i,j\in S:i<j}e^{-\beta v_{ij}^Z}\indic_{\mc{A}_{ij}}\Bigr) \prod_{ij\in \cup_{l=1}^n\mc{E}^\inter(X_l)}\indic_{\mc{B}_{ij}^c}\\ \times\prod_{ij\in F}(-\indic_{\mc{B}_{ij}}) \prod_{i\in V_\good}\left(e^{\beta \g_\lambda(x_i-y_i)+L(\frac{|x_i-y_i|}{\tau_i^Z})^2}\indic_{|x_i-y_i|\leq \ve_0\Cut}\indic_{i\in I_\good^Z }\right)\prod_{i\in V_\good}\dd x_i \dd y_i.
	\end{multline*}
	Next, resumming according to the connected components relative to $\pi$ of the graph $(V_\good,E\cup F)$, \textcolor{black}{and summing over the elements of $V^\infty$ outside $\cup_lV_{X_l}$, which produces the factors $1+f_{i\infty}^{v^Z}$ for $i\in V_\good\setminus\cup_lV_{X_l}$,} we obtain
	\begin{multline*}
		W'(Z)=\sum_{n=0}^\infty \frac{1}{n!}\sum_{\substack{X_1,\ldots,X_n\subset \pi\\ \mathrm{disjoint}}} \int_{(\Lambda^2)^{(N-|V_\bad|)}} \mathcal U(X_1)\cdots \mathcal U(X_n)\\ \times\textcolor{black}{\prod_{i\in V_\good\setminus\cup_l V_{X_l}}(1+f_{i\infty}^{v^Z})}\prod_{S\in \pi}\Bigl(\indic_{\mc{B}_S}\prod_{ i,j\in S:i<j}e^{-\beta v_{ij}^Z}\indic_{\mc{A}_{ij}}\Bigr)\\ \times \prod_{i\in V_\good}\Bigl(e^{\beta \g_\lambda(x_i-y_i)+L(\frac{|x_i-y_i|}{\tau_i^Z})^2}\indic_{|x_i-y_i|\leq \ve_0\Cut }\indic_{i\in I_\good^Z } \Bigr)\prod_{i\in V_\good}\dd x_i \dd y_i,
	\end{multline*}
	where \textcolor{black}{$\mathcal U(X)=0$ for $|X|\leq1$, and, for $|X|\geq2$,}
	\begin{multline*}
		\mathcal U(X)=\sum_{V^\infty\subset V_X}\sum_{k=0}^\infty \frac{1}{k!}\sum_{\substack{Y_1,\ldots,Y_k\subset X\\ \mathrm{disjoint}}} \sum_{E_1\in \mathsf{E}^{Y_1}}\cdots \sum_{E_k\in \mathsf{E}^{Y_k}}\sum_{F\in \textcolor{black}{\mathsf{E}_0}^{\Coarse_X(Y_1,\ldots,Y_k)}}
		\prod_{ij\in\cup_{l=1}^k E_l}f_{ij}^{v^Z}\prod_{i\in V^\infty}f_{i\infty}^{v^Z}\prod_{ij\in F}(-\indic_{\mc{B}_{ij}})\\ \times \prod_{ij\in \cup_{l=1}^k\mc{E}^\inter(Y_l) } \indic_{\mc{B}_{ij}^c}.
	\end{multline*}
	{\color{black}
    By \eqref{def:external normalization}, each multipole $S$ of $\pi$ that belongs to none of the $X_l$ contributes the factor $g_{\ve_0}^Z(S)$. Since the integral factorizes over the multipoles, \eqref{def:normalized upper} then gives
    \begin{equation*}
        W(Z)=\frac{W'(Z)}{(NC_{\beta,\lambda,\ve_0})^{N'}\widehat{\Msf}_{\ve_0}^{+,Z}(\pi)}
        =\sum_{n=0}^{\infty}\frac{1}{n!}\sum_{\substack{X_1,\ldots,X_n\in\mc P(\pi)\\ \mathrm{disjoint}}}\widehat{\Ksf}_{\ve_0}^{+,Z}(X_1)\cdots\widehat{\Ksf}_{\ve_0}^{+,Z}(X_n).
    \end{equation*}
    }
\end{proof}

\subsection{Statement of the main activity controls}\label{sub:stat upper}

We now state the analogues of Propositions~\ref{prop:bounded lower} and~\ref{prop:expansion -}.

\begin{prop}[Control of the activity of bounded-size clusters with frozen bad points]\label{prop:bounded upper}
	Let $\beta\in (2,+\infty)$ and $p^*(\beta)$ be as in \eqref{defpstar}. Recall $\gamma_{\beta,\lambda,k}$ from Definition~\ref{def:gamma beta}. Let $\ve_0\in (0,1)$. 
	Let $V_\bad \subset [N]$ and set $V_\good=[N]\setminus V_\bad$. Let $Z\in (\Lambda^{2})^{|V_{\bad}|}$. Let $M$ be as in Definition~\ref{def:multipoles}. Let \textcolor{black}{$S\subset V_\good$} be such that \textcolor{black}{$1\leq |S|\leq p^*(\beta)$} and $X$ be a \textcolor{black}{nonempty} subpartition of $V_\good$ \textcolor{black}{such that $|S'|\leq p(\beta)$ for every $S'\in X$.}  
	
	There exists $C>0$ depending only on $\beta$, $M$, $\ve_0$ and $|S|$ such that
	\begin{equation}\label{eq:MX +low}
		\Msf_{\ve_0}^{+,Z}(S)\geq \frac{\lambda^{2(|S|-1)}}{CN^{|S|-1}},
	\end{equation}
	\begin{equation}\label{eq:MX +up}
		\Msf_{\ve_0}^{+,Z}(S)\leq \frac{C\lambda^{2(|S|-1)}}{N^{|S|-1}}.
	\end{equation}
	Moreover, there exists $C>0$ depending only on  $\beta$, $M$, $\ve_0$ and $|V_X|$ such that
	\begin{equation}\label{eq:K+ bound}
		|\Ksf_{\ve_0}^{+,Z}(X)|\leq C\frac{\gamma_{\beta,\lambda,|V_X|}}{\Msf^0_{\ve_0} (X)N^{|V_X|-1} }.
	\end{equation}
\textcolor{black}{The same bounds hold for $\widehat{\Msf}_{\ve_0}^{+,Z}$ and $\widehat{\Ksf}_{\ve_0}^{+,Z}$.}
\end{prop}

\begin{prop}[Expansion errors with frozen bad points]\label{prop:expansion +}
	Let $\beta\in (2,+\infty)$, $p(\beta)$ be as in Definition~\ref{def:pbeta} and $\gamma_{\beta,\lambda,n}$ be as in Definition~\ref{def:gamma beta}. Let $\ve_0\in (0,1)$.
	Let $V_\bad \subset [N]$ and set $V_\good=[N]\setminus V_\bad$. Let $Z\in (\Lambda^{2})^{|V_{\bad}|}$. Let \textcolor{black}{$S\subset V_\good$} be such that \textcolor{black}{$1\leq |S|\leq p(\beta)$} and $X$ be a subpartition of $V_\good$ such that $|V_X|\leq p(\beta)$.
	
	\textcolor{black}{Then there exists $M_0>\max(20,3^{p(\beta)})$, depending only on $\beta$, such that, for every $M\geq M_0$ in Definition~\ref{def:multipoles} and for $\lambda$ small enough, there exist $C>0$ depending only on $\beta,M,|S|$ and $C_{\ve_0}>0$ also depending on $\ve_0$ such that}
	\begin{equation}\label{eq:M+diff}
		|\Msf_{\ve_0}^{+,Z}(S)-\Msf^0_{\ve_0} (S)|\leq \textcolor{black}{\frac{1}{N^{|S|-1}}\Bigl(C_{\ve_0}\Cut^{-2}+C\ve_0^{\alpha(\beta)|S|}\frac{|V_\bad|}{N}\Bigr)}.
	\end{equation} 
	\textcolor{black}{Moreover, there exist $C>0$ depending only on $\beta,M,|V_X|$ and $C_{\ve_0}>0$ also depending on $\ve_0$ such that} 
	\begin{equation}\label{eq:K+diff}
		|\Ksf_{\ve_0}^{+,Z}(X)-\Ksf_{\ve_0}^0(X)|\leq \textcolor{black}{\frac{C_{\ve_0}\Cut^{-2}+C\ve_0^{\alpha(\beta)|V_X|}|V_\bad|/N}{N^{|V_X|-1}\Msf^0_{\ve_0}(X)}}.
	\end{equation}
\textcolor{black}{The same estimates hold for $\widehat{\Msf}_{\ve_0}^{+,Z}$ and $\widehat{\Ksf}_{\ve_0}^{+,Z}$.}
\end{prop}

Propositions~\ref{prop:bounded upper} and  \ref{prop:expansion +} are proved in Subsection~\ref{sub:abs upper}.

In the next proposition, we expand the logarithm of the function $W(Z)$ of Lemma~\ref{lemma:start upper} (see \eqref{def:H upper}). \textcolor{black}{We fix a configuration of bad points and show that the series of the $\Ksf_{\ve_0}^{+,Z}(X)$ and $\widehat{\Ksf}_{\ve_0}^{+,Z}(X)$ are absolutely convergent.}

\begin{prop}[Absolute convergence of the cluster-expansion series for frozen bad points]\label{prop:absolute upper}
	Let $\beta\in (2,+\infty)$ and $p(\beta)$ be as in Definition~\ref{def:pbeta}. Let $\ve_0\in (0,1)$.
	Let $V_\bad \subset [N]$ and set $V_\good=[N]\setminus V_\bad$. Let $Z\in (\Lambda^{2})^{|V_{\bad}|}$. Let $M$ be as in Definition~\ref{def:multipoles}. Let $\pi$ be a partition of $V_\good$ such that, for every $S\in \pi$, one has $|S|\leq p(\beta)$. Assume that, for every $k\in \{2,\ldots,p(\beta)\}$,
	$	n_k\coloneqq |\{S\in \pi:|S|=k\}| $ satisfies \eqref{bornesnkbs}.

	For $M$ large enough with respect to $\beta$ and $p(\beta)$, $\ve_0$ small enough with respect to $\beta$, $p(\beta)$ and $M$, and $\lambda$ small enough, there exists $C>0$ depending only on $\beta, M, p(\beta)$ and $\ve_0$ such that
	\begin{equation} \label{eq:rest u}
		\textcolor{black}{\sum_{X\in\mc P(\pi)}|\Ksf_{\ve_0}^{+,Z}(X)|\indic_{|V_X|>p(\beta)}\leq CN\delta_{\beta,\lambda}}.
	\end{equation}
	\textcolor{black}{The full connected series in \eqref{eq:resum} is absolutely convergent for both $\Ksf_{\ve_0}^{+,Z}$ and $\widehat{\Ksf}_{\ve_0}^{+,Z}$.}
	
	Moreover, for $M$ large enough with respect to $\beta$ and $p(\beta)$, $\ve_0$ small enough with respect to $\beta$ and $M$, and $\lambda$ small enough, there exists $C>0$ depending only on $\beta, M, p(\beta)$ and $\ve_0$ such that
	\begin{equation}\label{eq:rest b u}
		\sum_{n=1}^\infty\frac{1}{n!}\sum_{\substack{X_1,\ldots,X_n\in \mc{P}(\pi):\\ \mathrm{connected}} }|\Ksf_{\ve_0}^{+,Z}(X_1)\cdots \Ksf_{\ve_0}^{+,Z}(X_n) \mathrm{I}(G(X_1,\ldots,X_n))|\textcolor{black}{\indic_{\exists i:\,|V_{X_i}|>p(\beta)}\leq CN\delta_{\beta,\lambda}}.
	\end{equation}
\textcolor{black}{Both estimates also hold for $\widehat{\Ksf}_{\ve_0}^{+,Z}$.}
\end{prop}

The proof is a variation on the proof of Proposition~\ref{prop:absolute lower} and is provided in Subsection~\ref{sub:abs upper}.

\subsection{Core-forest decomposition} \label{sub:core forest}

In this subsection, we bound the activity $\Ksf_{\ve_0}^{+,Z}(X)$ of a subpartition
$X$ of $[N]$. Following Definition~\ref{def:awbw}, decompose each Mayer bond into
\[
f_{ij}^{v^Z}=a_{ij}^{v^Z}+b_{ij}^{v^Z}.
\]

We begin by splitting $f_{i\infty}^{v^Z}$ (Definition~\ref{def:iinfty}), the Mayer bond between good and bad points, into its odd and even parts.

\begin{definition}\label{def:ainfty}
	Let $V_\bad\subset [N]$ and $Z\in (\Lambda^2)^{|V_\bad|}$. For every $i\in [N]\setminus V_\bad$, we split $f_{i\infty}^{v^Z}$ into $a_{i\infty}^{v^Z}+b_{i\infty}^{v^Z}$, where
	\begin{equation*}
		a_{i\infty}^{v^Z}\coloneqq -\sum_{k\, \mathrm{ odd}\, }\frac{\beta^k}{k!}\left(\sum_{j\in V_\bad}v_{ij}^Z\right)^k\quad \mathrm{and}\, \quad b_{i\infty}^{v^Z}\coloneqq \sum_{k\, \mathrm{ even},\,  k\neq 0}\frac{\beta^k}{k!}\left(\sum_{j\in V_\bad}v_{ij}^Z\right)^k.
	\end{equation*}
\end{definition}
\begin{remark}
	\label{remarkbornesab}
	In view of Lemma~\ref{lemma:energy goodbad}, on the event $I_\bad=V_\bad$, using \eqref{eq:geo1}, we may bound
	\begin{equation}\label{bornesabinfty}
		|a_{i\infty}^{v^Z}|\le C\frac{r_i}{d_{i,\bad}},\qquad |b_{i\infty}^{v^Z}|\le C\frac{r_i^2}{d_{i,\bad}^2}
	\end{equation}
	for some constant $C>0$ depending on $\beta$ and $q(\beta)$.
	
	Also, if $i,j\in I_\good$, in view of \eqref{eq:v'} in Lemma~\ref{lemma:error'}, and using that $\tau_i^{+,Z}=\tau_i^{-,Z}\in[\max(\frac14r_i,\lambda), \ve_0\Cut]$ for $i\in I_\good$,  we have
	\begin{equation}\label{bornesabZ}
		|a_{ij}^{v^Z}\indic_{\mc{B}_{ij}^c}|\le C a_{ij}^\abs,\quad |b_{ij}^{v^Z}\indic_{\mc{B}_{ij}^c}|\le C b_{ij}^\abs,
	\end{equation}
	where $a_{ij}^\abs, b_{ij}^\abs$ are as in Definition~\ref{def:abs ab}, and $C>0$ depends only on $\beta$.
\end{remark}

Recall that, for the lower bound, Corollary~\ref{coro:reduction Euler} shows that when the odd weights $a_{ij}^{\tilde v}$ are used, only Eulerian graphs relative to $X$ contribute. In the present setting, some vertices can interact with bad points through $a_{i\infty}^{v^Z}$ or $b_{i\infty}^{v^Z}$. Let $V_1^\infty$ and $V_2^{\infty}$ denote these (disjoint) sets. We will show that only the following graphs with odd internal weights contribute.

\begin{definition}[Eulerian graphs with exceptional vertices]\label{def:Euler off}
	Let $X$ be a subpartition of $[N]$. Let $V_1^\infty\subset V_X$. We say that $E\in \Eul^X(V_1^\infty)$ if
	\begin{enumerate}
		\item $E\subset \mc{E}^\inter(X)$,
		\item for every $S\in X$, 
		\begin{equation*}
			\deg_{E}(S)+\sum_{i\in S}\indic_{i\in V_1^\infty}\quad \mathrm{ is} \,  \mathrm{even}.
		\end{equation*}
	\end{enumerate}
	
	Moreover, we say that $E\in \Eulc^X(V_1^\infty)$ if $E\in \Eul^X(V_1^\infty)$ and if $E$ is connected relative to $X$, i.e., we set 
	\begin{equation*}
		\Eulc^X(V_1^\infty)\coloneqq \Eul^X(V_1^\infty)\cap \mathsf{E}^X. 
	\end{equation*}
\end{definition}

\textcolor{black}{Reversing all dipoles in one $S\in X$ preserves the density of $\Psf_X^{+,\ve_0,Z}$, the even bonds and the separation indicators, but it changes the sign of the odd bonds incident to $S$ and of the $a_{i\infty}^{v^Z}$ with $i\in S$. Lemma~\ref{lemma:cancellation odd} therefore gives the following.}

\begin{coro}\label{coro:reduction Euler u}
	Let $X$ be a subpartition of $[N]$ \textcolor{black}{with $|X|\geq2$}.
	
	We have 
	\begin{multline*}
		\Ksf_{\ve_0}^{+,Z}(X)=\sum_{\substack{V_1^\infty, V_2^\infty \subset V_X:\\ V_1^\infty\cap V_2^\infty=\emptyset }}  \sum_{n=0}^\infty \frac{1}{n!}\sum_{\substack{X_1,\ldots,X_n\subset X\\ \mathrm{disjoint} }}\prod_{l=1}^n
		\Biggl(\sum_{ E_{l,1}\in \Eul^{X_l}(V_{X_l}\cap V_1^\infty) }  \sum_{\substack{E_{l,2}:E_{l,1}\cup E_{l,2}\in \mathsf{E}^{X_l}\\
				E_{l,1}\cap E_{l,2}=\emptyset  }} \Biggr)\\ \sum_{F\in \textcolor{black}{\mathsf{E}_0}^{\Coarse_X(X_1,\ldots,X_n)}}\dE_{\Psf_{X}^{+,\ve_0,Z}}\biggl[\prod_{l=1}^n\left(\prod_{ij\in E_{l,1}}a^{v^Z}_{ij}\prod_{ij\in E_{l,2} }b^{v^Z}_{ij}\right)\prod_{ij\in \cup_l\mc{E}^\inter(X_l)}\indic_{\mc{B}_{ij}^c}\\ \times\prod_{ij\in F}(-\indic_{\mc{B}_{ij}})\prod_{i\in V_1^\infty}a_{i\infty}^{v^Z}\prod_{i\in V_2^\infty}b_{i\infty}^{v^Z} \biggr].
	\end{multline*}
\end{coro}

We now weaken the condition $E\in \Eul^X(V_1^\infty)$. Note that $E\in \Eulc^X(V_1^\infty)$ does not guarantee that $E$ is 2-edge-connected relative to $X$, in contrast to the case $V_1^\infty\cup V_2^\infty=\emptyset$ discussed in Lemma~\ref{lemma:euler implies}.

Recall the definition of the $2$-core of a graph, introduced in \cite{SEIDMAN1983269}, and of related notions.

\begin{definition}[Core decomposition of a graph]
	\label{def:2core simple}
	Let $G=(V,E)$ be a finite connected graph.  
	
	\begin{enumerate}
		\item \textbf{$\boldsymbol{2}$-core.}
		The \emph{$2$-core} of $G$, written $C_2(G)$,
		is the unique maximal \emph{induced} subgraph
		whose minimum degree among all vertices (denoted $\delta$) is at least~$2$:
		\[
		C_2(G)=\max\bigl\{H\subseteq G \text{ induced} :
		\delta(H)\ge 2\bigr\}.
		\]
		It can be found by repeatedly deleting every vertex whose
		\emph{current} degree is $<2$ (that is, $0$ or $1$)
		together with all incident edges, until no such vertex remains.
		The resulting subgraph is either empty or connected.
		
		\item \textbf{Pendant trees.}
		\textcolor{black}{If $C_2(G)\neq\emptyset$, each connected component of the graph obtained from $G$ by deleting the vertices of $C_2(G)$, together with its edge to the core and the endpoint of this edge, is a \emph{pendant tree}. Otherwise, $G$ is the only pendant tree.}
		Thus
		\[
		G = C_2(G)\cup
		\bigl(\text{\textcolor{black}{edge-disjoint} union of pendant trees}\bigr).
		\]
		
		\item \textbf{Bridge-blocks.}
		A \emph{bridge-block} is a maximal $2$-edge-connected
		subgraph of $C_2(G)$.
		Equivalently, it is a maximal \textcolor{black}{connected} subgraph of $C_2(G)$
		containing no \emph{bridge} (an edge whose deletion
		disconnects $C_2(G)$).
		(Note that a singleton can be a bridge-block.)
		\item \textbf{Bridge-tree.}
		Let $\mathcal{B}$ denote the family of bridge-blocks of $C_{2}(G)$ and write $\mathcal{E}_{\br}$ for the set of bridges of $C_{2}(G)$. Then
		\[
		V\bigl(C_{2}(G)\bigr)=\bigcup_{B\in\mathcal{B}}V(B), 
		\qquad
		E\bigl(C_{2}(G)\bigr)=
		\Bigl(\bigcup_{B\in\mathcal{B}}E(B)\Bigr)\cup
		\mathcal{E}_{\br},
		\]
		so the vertex sets of distinct blocks are disjoint, the edges of the blocks are pairwise disjoint, and every edge that is not in a block is a bridge. \textcolor{black}{If $C_2(G)\neq\emptyset$, contracting} each block $B\in\mathcal{B}$ to a single vertex and
		retaining every edge in $\mathcal{E}_{\br}$ produces a tree, called the \emph{bridge-tree} of $G$.  Its vertices correspond one-to-one with the bridge-blocks, and its edges correspond to the bridges that join them. 
		
		\textcolor{black}{Let $Y\coloneqq\{V(B):B\in\mc B\}$. If $C_2(G)\neq\emptyset$, then}
		\begin{equation*}
			\textcolor{black}{\mc E_\br\in\mathsf T_0^Y.}
		\end{equation*}
		\textcolor{black}{For an empty core, $\mc B=\mc E_\br=\emptyset$ and the bridge-tree is empty.}
	\end{enumerate}
\end{definition}

\begin{figure}[H]
	\centering
	\fbox{\includegraphics[width=0.5\textwidth]{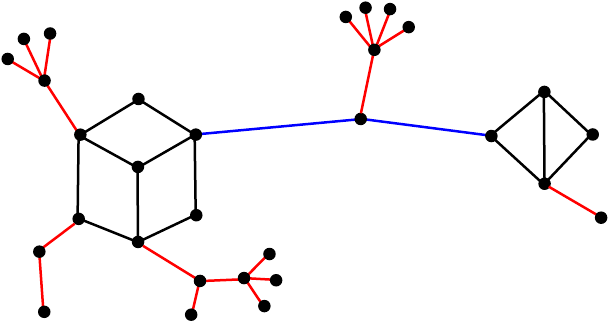}} 
	\caption{In red, the edges of the pendant trees. In blue and black, the edges of the $2$-core. In blue, the edges of the bridge-tree linking the three 2-edge-connected blocks.}
\end{figure}

We now extend Definition~\ref{def:2core simple} to the same notion but ``relative to a partition''.

\begin{definition}[Core decomposition of a graph relative to a partition]
	Let $X$ be a subpartition of $[N]$, let $E\in \mathsf{E}^X$ and set $G\coloneqq (V_X,E)$. For every $S\in X$, recall $\deg_E(S)\coloneqq \sum_{i\in S}\deg_E(i)$. 
	
	\begin{enumerate}
		\item \textbf{$\boldsymbol{2}$-core.}
		The \emph{$2$-core} of $G$ relative to $X$, written $C_2(G,X)$,
		is the unique maximal subgraph \emph{induced}  by $V_{X'}$ for $X'\subset X$ such that each $S\in X'$ has degree at least 2. It can be found by repeatedly deleting every $S\in X$ whose
		\emph{current} degree is $<2$ (that is, $0$ or $1$) together with all incident edges, until no such $S$ remains. 
		
		Letting $X'$ be the set of $S\in X$ incident to some edge in $C_2(G,X)$, \textcolor{black}{we observe that $C_2(G,X)$ is connected relative to $X'$ when $X'\neq\emptyset$.}
		\item \textbf{Pendant trees.}
		
\begingroup\color{black}
The pendant trees relative to $X$ are those of the quotient multigraph $G/X$, with the original edges and all the labels of their multipoles, including the attachment multipole. They are trees relative to their restricted partitions. If the core is empty, $G$ is the only pendant tree. Thus,
\[
G=C_2(G,X)\cup\bigl(\text{edge-disjoint union of pendant trees}\bigr).
\]
\endgroup

		\item \textbf{Bridge-blocks.}
		\textcolor{black}{A \emph{bridge-block} is a maximal subgraph $B$ of $C_2(G,X)$ whose vertex set is a union of parts of $X$ and which is $2$-edge-connected relative to $\{S\in X:S\subset V(B)\}$. A single part with no edges can be a bridge-block.}
		\item \textbf{Bridge-tree.}
		Let $\mathcal{B}$ denote the family of bridge-blocks of $C_{2}(G,X)$ and write $\mathcal{E}_{\br}$ for the set of bridges of $C_{2}(G,X)$ \textcolor{black}{relative to $X'$}. Then,
		\[
		V\bigl(C_{2}(G,X)\bigr)=\bigcup_{B\in\mathcal{B}}V(B), 
		\qquad
		E\bigl(C_{2}(G,X)\bigr)=
		\Bigl(\bigcup_{B\in\mathcal{B}}E(B)\Bigr)\cup
		\mathcal{E}_{\br}.
		\]
		\textcolor{black}{Let $Y\coloneqq\{V(B):B\in\mc B\}$. If the core is nonempty, then}
		\begin{equation*}
			\textcolor{black}{\mc E_\br\in\mathsf T_0^Y,\qquad\text{on the vertex set }V_Y.}
		\end{equation*}
	\end{enumerate}
\end{definition}

We now simplify the condition that $E\in \Eulc^X(V_1^\infty)$. Recall from Corollary~\ref{coro:reduction Euler u} that $V_1^\infty$ stands for the set of points $i\in V_X$ interacting with bad points with weight $a_{i\infty}^{v^Z}$.

\begin{lemma}[Graph structure]\label{lemma:weaker}
	Let $X$ be a subpartition of $[N]$ and $V_1^\infty \subset V_X$. Suppose that $V_1^\infty\neq \emptyset$ and let
	\begin{equation*}
		X_1^\infty=\{[i]^X:i\in V_1^\infty\}.
	\end{equation*}
	Let $E\in \textcolor{black}{\Eulc^X(V_1^\infty)}$ and set $G\coloneqq (V_X,E)$. Let $G_1,\ldots,G_n$ be the bridge-blocks and $\mc{E}_\br$ be the bridges of \textcolor{black}{its $2$-core $C_2(G,X)$, relative to the parts of $X$ incident to the edges of $C_2(G,X)$}. For every $i\in [n]$, let $B_i$ be the set of multipoles \textcolor{black}{contained in $G_i$ or in a pendant tree attached to $G_i$.} Let $Y\coloneqq \{V_{G_i}:1\leq i\leq n\}$.

	Then, the following holds:
	\begin{enumerate}
		\item \textcolor{black}{Let $S\in X$ satisfy $\deg_E(S)=1$.} Then, $S\in X_1^\infty$.
		\item \textcolor{black}{Suppose that $n\geq2$. Let $l\in[n]$ be such that $V_{G_l}$ is a leaf of $(V_Y,\mc E_\br)/Y$.} Then, 
		\begin{equation*}
			B_l\cap X_1^\infty\neq \emptyset.
		\end{equation*}
	\end{enumerate}
\end{lemma}

\begin{proof}
	By the definition of $\Eulc^X(V_1^\infty)$, since the degree of $S$ in $(V_X,E)$ relative to $X$ is odd, we must have $S\in X_1^\infty$ (all multipoles not touching $V_1^\infty$ have even degree).

\begingroup\color{black}
For (2), let $E_l=E\cap\binom{V_{B_l}}2$. Since $V_{G_l}$ is a leaf, exactly one edge of $E$ joins $V_{B_l}$ to its complement. If $B_l\cap X_1^\infty=\emptyset$, then $\deg_E(S)$ is even for every $S\in B_l$, whereas
\[
\sum_{S\in B_l}\deg_E(S)=2|E_l|+1,
\]
a contradiction.
\endgroup
\end{proof}

As in Definition~\ref{def:peeling into minimal}, we peel each 2-edge-connected block down to a graph that is minimally 2-edge-connected relative to \textcolor{black}{its restricted partition}. This reduction leaves the resulting graph with only \(O\bigl(|V_{X}|\bigr)\) edges, ensuring that our activity estimates remain well controlled. By Lemma~\ref{lemma:sum Euler}, the edges removed during peeling simply migrate to the interaction term in the exponential term of \eqref{def:H}.

We extend the notation $\Peeled_X$ from Definition \textcolor{black}{\ref{def:peeling into minimal}} to general connected graphs by applying it blockwise and adding pendant trees and bridges.

\begin{definition}[Peeling into a minimal skeleton]\label{def:minimal skeleton} 
	Let $X$ be a subpartition of $V_\good$, $E\in \mathsf{E}^X$ and $G\coloneqq (V_X,E)$. Let $G_1,\ldots,G_n$ be the bridge-blocks, $P_1,\ldots,P_k$ be the pendant trees and $\mc{E}_\br$ be the bridges of \textcolor{black}{its $2$-core $C_2(G,X)$, relative to the parts of $X$ incident to the edges of $C_2(G,X)$}. For every $i\in [n]$, let $G_i'\subset G_i$ be a minimal spanning 2-edge-connected (relative to $\textcolor{black}{X_i\coloneqq\{S\in X:S\subset V(G_i)\}}$) subgraph of $G_i$ (if multiple exist, select one according to lexicographical order). We let
	\begin{equation*}
		\Peeled_X(E)\coloneqq \bigcup_{i=1}^n E(G_i')\cup \bigcup_{i=1}^k E(P_i)\cup \mc{E}_\br.
	\end{equation*}
	We say that $\Peeled_X(E)$ is the ``minimal skeleton of~$E$ relative to $X$''. If $\Peeled_X(E)=E$, we say that $E$ is a minimal skeleton relative to $X$.
\end{definition}

From every \emph{minimally 2-edge-connected} block of the reduced skeleton we now
extract a spanning tree using the algorithm described in Definition~\ref{def:peeling mini}.  Adding the pendant trees together with the
collection of bridges then produces a spanning tree of the whole graph relative to~\(X\).

\begin{definition}[Peeling of a minimal skeleton]\label{def:peeling general} 
	Let $X$ be a subpartition of $V_\good$ and let $E\in\mathsf{E}^X$ be a minimal skeleton. Let $G\coloneqq (V_X,E)$. Let $G_1,\ldots,G_n$ be the bridge-blocks, $P_1,\ldots,P_k$ be the pendant trees, and $\mc{E}_\br$ be the bridges of \textcolor{black}{its $2$-core $C_2(G,X)$, relative to the parts of $X$ incident to the edges of $C_2(G,X)$}. \textcolor{black}{For every $i\in[n]$, set $X_i=\{S\in X:S\subset V(G_i)\}$. If $|X_i|=1$, then $E(G_i)=\emptyset$, and we let $\mc T^{X_i}(\cdot,E(G_i))=\mc F^{X_i}(\cdot,E(G_i))=\emptyset$.} Recall $\mc{T}^{X}$ and $\mc{F}^{X}$ from Definition \textcolor{black}{\ref{def:peeling mini}}.
	
	We define
	\begin{equation*}
		\mc{T}_+^{X}(\cdot,E)\coloneqq \bigcup_{i=1}^n \mc{T}^{X_i}(\cdot,E(G_i) )\cup \bigcup_{i=1}^k E(P_i)\cup\mc{E}_{\br},\quad  \text{and}\quad {\mc{F}}_+^X(\cdot,E)\coloneqq \bigcup_{i=1}^n{\mc{F}}^{X_i}(\cdot,E(G_i)).
	\end{equation*}
\end{definition}

\textcolor{black}{For a general $E\in\mathsf E^X$, we define $\mc T_+^X(\cdot,E)$ and $\mc F_+^X(\cdot,E)$ as those of the minimal skeleton $\Peeled_X(E)$. We define the weights $g_{ij}$ of \eqref{def:gij} separately on each bridge-block.}

\subsection{Study of the $2$-edge-connected blocks and their pendant trees}\label{sub:pendant}

By Lemma~\ref{lemma:weaker}, the subgraphs with odd internal weights can be reduced to a $2$-core with pendant trees, whose \textcolor{black}{nonroot leaf multipoles} are linked to bad points. \edgreen{Lemmas~\ref{lemma:pendant} and~\ref{lemma:block} below estimate pendant trees and $2$-edge-connected blocks separately. They are not used in the sequel: Corollary~\ref{coro:prod a u} directly proves the bound needed for minimal skeletons. We keep them because they isolate the two mechanisms at work.}

For later use, we need a nice way to decompose a tree into edge-disjoint paths starting at the leaves and ending at branching points or the root. We use the heavy-light decomposition algorithm \cite{SleatorTarjan83,HarelTarjan84}.

\begin{definition}[Heavy-light decomposition]\label{def:heavy} Let $T=(V,E)$ be a finite rooted tree with root $r$. 
	
	For every internal vertex $u$ of $T$, choose one child $c$ whose subtree is largest (breaking ties arbitrarily) and declare the edge $(u,c)$ \emph{heavy}.
	All other child edges of~$u$ are \emph{light}.
	
\begingroup\color{black}
From a leaf $i$ distinct from the root, we move upward along heavy edges until we meet a light edge or reach $r$. In the first case, we let $\loc_T(i)$ be the parent endpoint of this light edge, and in the second case, we set $\loc_T(i)=r$. This defines the map
\[
\loc_T:\Leaves(T)\setminus\{r\}\longrightarrow V.
\]
\endgroup
\end{definition}

\begin{figure}
	\centering
	\fbox{\includegraphics[width=0.5\textwidth]{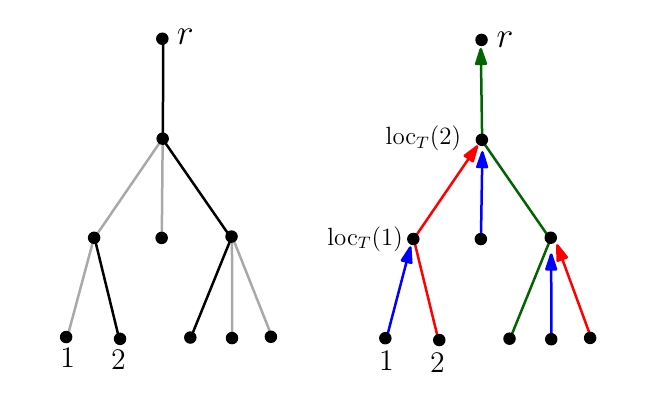}} 
	\caption{Left: rooted tree \(T\) with root \(r\). Heavy edges are black, and light edges are gray. Right: decomposition of $T$ into edge-disjoint paths. Alternating blue, red and green colors delineate different paths. The head of each colored path starting at a leaf $v$ is the vertex $\loc_T(v)$, \textcolor{black}{that is, the parent endpoint of the first light edge met, or the root if there is none.}}
\end{figure}

\begin{lemma}\label{lemma:heavy light}
	Let $T=(V,E)$ be a finite tree with root $r$. Let $\loc_T$ be as in Definition~\ref{def:heavy}. For every $i\in \Leaves(T)\setminus \{r\}$, let $P(i,\loc_T(i))$ stand for the set of edges in $T$ between $i$ and $\loc_T(i)$. Then,
	\begin{equation}\label{eq:sqcup}
		E=\bigsqcup_{i\in \Leaves(T)\setminus \{r\} }P(i,\loc_T(i)).
	\end{equation}
	In particular,
	\begin{equation}\label{eq:boundSV}
		\sum_{i\in \Leaves(T)\setminus \{r\} }\dist_T(i,\loc_T(i)) =|V|-1.
	\end{equation}
	Moreover, for every $i\in V\setminus\{r\}$, 
	\begin{equation}\label{eq:degree1}
		|\{l\in \Leaves(T)\setminus \{r\}:\loc_T(l)=i\}|=\max(\deg_{T}(i)-2,0)
	\end{equation}
	and 
	\begin{equation}\label{eq:degree2}
		|\{l\in \Leaves(T)\setminus \{r\}:\loc_T(l)=r\}|=\deg_{T}(r).
	\end{equation}
\end{lemma}

\begin{proof}
If $T$ has no edge, then $V=\{r\}$ and the four assertions are immediate, so we assume that $T$ has an edge. Every vertex with children has exactly one heavy child edge, so the heavy edges partition $V$ into vertex-disjoint paths, some reduced to a single vertex, each going down from its top vertex, which is $r$ or the child endpoint of a light edge, to a leaf distinct from $r$. For $i\in\Leaves(T)\setminus\{r\}$, the set $P(i,\loc_T(i))$ consists of the edges of the heavy path containing $i$ and, if the top vertex of this path is not $r$, of the light edge above it. Hence the edge between a vertex and one of its children $c$ lies in $P(i,\loc_T(i))$ exactly when $i$ is the bottom vertex of the heavy path containing $c$, which proves~\eqref{eq:sqcup}. Since $|E|=|V|-1$,~\eqref{eq:sqcup} gives~\eqref{eq:boundSV}. For $i\in V\setminus\{r\}$, the leaves $l$ with $\loc_T(l)=i$ correspond to the light edges between $i$ and its children; there are $\deg_T(i)-2$ of them if $i$ has children, and none otherwise, which proves~\eqref{eq:degree1}. The leaves $l$ with $\loc_T(l)=r$ correspond to the light edges between $r$ and its children and to the heavy path containing $r$, so there are $\deg_T(r)$ of them, which proves~\eqref{eq:degree2}.
\end{proof}

\textcolor{black}{We now study the odd weights on pendant trees. By Lemma~\ref{lemma:weaker}(1), every nonroot leaf multipole $S$ of a pendant tree contains a label $j_S\in V_1^\infty$, and \eqref{bornesabinfty} bounds its weight $a_{j_S\infty}^{v^Z}$ by $Cr_{j_S}/d_{j_S,\bad}$.}

\begin{lemma}\label{lemma:pendant}
	Let $X$ be a subpartition of $V_\good$ such that, for every $S\in X$, we have $|S|\leq p(\beta)$. Let $P\in \mathsf{T}^X$ be a rooted tree with root $a\in V_X$ (chosen so that $a$ is incident to some edge in $P$). \textcolor{black}{Let $Q=(V_X,P)/X$ be the quotient tree, rooted at $A=[a]^X$, and let $\mathcal L=\Leaves(Q)\setminus\{A\}$. For every $S\in\mathcal L$, let $j_S\in S$.} Recall $d_{i,\bad}$ from \eqref{eq:dibad}. There exists a constant $C>0$ depending on $\beta$, $M$ and $p(\beta)$ such that
	\begin{multline*}
		\left(\prod_{ij\in P}|a_{ij}^{v^Z}|\indic_{\mc{B}_{ij}^c}\right)\left(\textcolor{black}{\prod_{S\in\mathcal L}\frac{r_{j_S}}{d_{j_S,\bad}}}\right)\indic_{V_X\subset I_\good}\indic_{\Pi_\mult^{V_X}=X} \\  \leq C^{|P|}\left(\prod_{S\in X:S\neq [a]^X}r_S^2 \right)\left(\prod_{ij\in P}\frac{1}{d_{ij}^2}\indic_{M\min(r_i,r_j)\leq d_{ij}\leq \textcolor{black}{17}\ve_0 \Cut }\right)\Bigl(\frac{r_a}{d_{[a]^X,\bad}}\Bigr).
	\end{multline*}
\end{lemma}

\medskip

Let us explain the proof of the lemma. Start with the toy case $P=(12)$ \textcolor{black}{on $\mc B_{12}^c$} and declare $1$ as the root. One should think of $P$ as rooted at $1$ in the 2-core.

By Lemma~\ref{lemma:error'}, we have 
\begin{equation*}
	|a_{12}^{v^Z}|\leq C\frac{r_1 r_2 }{\max(d_{12},\lambda)\max(d_{12},r_1,r_2)}\indic_{d_{12}\leq \textcolor{black}{17}\ve_0\Cut} \le C a_{12}^\abs.
\end{equation*}
Therefore,
\begin{equation}\label{eq:displaya}
	|a_{12}^{v^Z}| \frac{r_2}{d_{2,\bad}} \leq C \frac{r_2^2}{d_{12}\max(d_{12},r_1,r_2)}\indic_{d_{12}\leq \textcolor{black}{17}\ve_0\Cut} \frac{r_1}{d_{2,\bad}}.
\end{equation}
Since the tree $12$ is rooted in the 2-core at $1$, the term $r_1^2$ will already appear in the peeling of the 2-core. Hence, to have a purely quadratic estimate, our bound on $|a_{12}^{v^Z}| |f_{2\infty}^{v^Z}|$ should not contain the term $r_1$. The key fact is that, if $a_{12}^{v^Z}\neq 0$, then $v_{12}^Z\neq 0$ and therefore, by the estimate \eqref{eq:geo4} of Lemma~\ref{lemma:geo}, 
$
d_{2,\bad}\geq \frac{1}{6}d_{1,\bad}.
$
Moreover, by \eqref{eq:geo1} in Lemma~\ref{lemma:geo}, since $1$ is a good point, we have $r_1\leq d_{1,\bad}$.
Inserting these two facts into \eqref{eq:displaya}, we obtain
\begin{equation*}
	|a_{12}^{v^Z}| \frac{r_2}{d_{2,\bad}}\leq  C\frac{r_2^2}{d_{12}\max(d_{12},r_1,r_2)}\indic_{d_{12}\leq \textcolor{black}{17}\ve_0\Cut}.
\end{equation*}

\begingroup\color{black}
For a general $P\in\mathsf T^X$, we decompose the quotient tree into edge-disjoint paths from its nonroot leaves. We compare the distances to the bad points by \eqref{eq:geo1'} within each multipole and by \eqref{eq:geo4'} along bonds, at a cost $C$ per edge. Since the path lengths sum to $|P|$, the total cost is $C^{|P|}$.
\endgroup
\medskip

\begingroup\color{black}
\begin{proof}[Proof of Lemma~\ref{lemma:pendant}]
We work on the event where $V_X\subset I_\good$, $\Pi_\mult^{V_X}=X$ and the product on the left is nonzero.
\paragraph{\bf{Step 1: heavy-light decomposition}}
For $S\in\mathcal L$, let $P_S$ be the edges of $P$ on the path from $S$ to $\loc_Q(S)$ in $Q$, and let $x_S$ be the endpoint in $\loc_Q(S)$ of its last edge. By Lemma~\ref{lemma:heavy light},
\[
P=\bigsqcup_{S\in\mathcal L}P_S,\qquad \sum_{S\in\mathcal L}|P_S|=|P|.
\]
Thus,
\begin{equation}\label{eq:HL decomp} \left(\prod_{ij\in P}|a_{ij}^{v^Z}|\indic_{\mc B_{ij}^c}\right) \prod_{S\in\mathcal L}\frac{r_{j_S}}{d_{j_S,\bad}} =\prod_{S\in\mathcal L}\left(\left(\prod_{ij\in P_S}|a_{ij}^{v^Z}|\indic_{\mc B_{ij}^c}\right)\frac{r_{j_S}}{d_{j_S,\bad}}\right).
\end{equation}
\paragraph{\bf{Step 2: equilibrating degrees in each path}}
Fix $S\in\mathcal L$. In the product below, the factor $r_{j_S}$ and the endpoint in $S$ of the first edge of $P_S$ give two radii in $S$, and every internal multipole of the path also gives two. Hence,
\[
r_{j_S}\prod_{ij\in P_S}r_ir_j
\leq r_{x_S}\prod_{R\in\Res(X,P_S)\setminus\{\loc_Q(S)\}}r_R^2.
\]
By \eqref{bornesabZ},
\begin{flalign}
		&& \left(\prod_{ij\in P_S}|a_{ij}^{v^Z}|\indic_{\mc B_{ij}^c}\right)\frac{r_{j_S}}{d_{j_S,\bad}} \leq C^{|P_S|}\prod_{R\in\Res(X,P_S)\setminus\{\loc_Q(S)\}}r_R^2
\prod_{ij\in P_S}\frac{\indic_{d_{ij}\leq17\ve_0\Cut}}{d_{ij}\max(d_{ij},r_i,r_j)}\frac{r_{x_S}}{d_{j_S,\bad}}.\label{eq:partial conc}
	\end{flalign}
Since the odd weights are nonzero, \eqref{eq:geo1'} and \eqref{eq:geo4'} give $d_{x_S,\bad}\leq C^{|P_S|}d_{j_S,\bad}$. Therefore,
\begin{flalign}
		&& \left(\prod_{ij\in P_S}|a_{ij}^{v^Z}|\indic_{\mc B_{ij}^c}\right)\frac{r_{j_S}}{d_{j_S,\bad}} \leq C^{|P_S|}\!\prod_{R\in\Res(X,P_S)\setminus\{\loc_Q(S)\}}r_R^2
\prod_{ij\in P_S}\frac{\indic_{d_{ij}\leq17\ve_0\Cut}}{d_{ij}\max(d_{ij},r_i,r_j)}\frac{r_{x_S}}{d_{x_S,\bad}}.\label{eq:apply}
	\end{flalign}
\paragraph{\bf{Step 3: conclusion}}
Every part other than $A$ is a nonterminal part of exactly one path. We multiply the bounds \eqref{eq:apply} over $S\in\mathcal L$. On $\mc B_{ij}^c$, we have $d_{ij}>M\max(\min(r_i,r_j),\lambda)$ by Definition~\ref{def:multipoles mc}, and we obtain
\begin{multline}\label{eq:wLeaves}
\left(\prod_{ij\in P}|a_{ij}^{v^Z}|\indic_{\mc B_{ij}^c}\right)
\prod_{S\in\mathcal L}\frac{r_{j_S}}{d_{j_S,\bad}}
\\\leq C^{|P|}\prod_{S\in X\setminus\{A\}}r_S^2
\prod_{ij\in P}\frac{\indic_{M\max(\min(r_i,r_j),\lambda)\leq d_{ij}\leq17\ve_0\Cut}}{d_{ij}^2}
\prod_{S\in\mathcal L}\frac{r_{x_S}}{d_{x_S,\bad}}.
\end{multline}
Every edge incident to the root is the last edge of one path. Since $a$ is incident to $P$, one path has $x_S=a$. By \eqref{eq:geo1}, the other ratios $r_{x_S}/d_{x_S,\bad}$ are at most $1$, so their product is at most $r_a/d_{a,\bad}\leq r_a/d_{A,\bad}$, which proves the claim.
\end{proof}
\endgroup

We now bound the product of the weights over the 2-edge-connected blocks and their pendant trees.

\begin{lemma}\label{lemma:block}
	Let $X$ be a subpartition of $V_\good$ such that, for every $S\in X$, we have $|S|\leq p(\beta)$. Let $V_1^\infty\subset V_X$ and let $E\in \Eulc^X(V_1^\infty)$ be such that $\Peeled_X(E)=E$ (i.e., a connected Eulerian graph with exceptional vertices $V_1^\infty$ by Definition~\ref{def:Euler off} and a minimal skeleton by Definition~\ref{def:minimal skeleton}). Set $G\coloneqq (V_X,E)$. 
	
	Let $G_1,\ldots,G_n$ be the bridge-blocks, $P_1,\ldots,P_k$ be the pendant trees, and $\mc{E}_\br$ be the bridges of \textcolor{black}{its $2$-core $C_2(G,X)$, relative to the parts of $X$ incident to the edges of $C_2(G,X)$}. For every $l\in [n]$, let $G_l'$ be the union of $G_l$ and the pendant trees attached to it. \textcolor{black}{For every $l\in[n]$, let $B_l=\{S\in X:S\subset V(G_l)\}$ and $B_l'=\{S\in X:S\subset V(G_l')\}$. Below, the sets $B(\cdot,T)$ over $V_{B_l}$ and over $V_{B_l'}$ are relative to $B_l$ and to $B_l'$, respectively. For a block of one part, we let $B(\cdot,\emptyset)=\{\emptyset\}$.}

	For every $l=1,\ldots,n$, let $\hat{S}_l$ be the index of the largest $r_S$, $S\in B_l$. Let $T\coloneqq \mc{T}_+^X(\cdot,E)$ and $F\coloneqq {\mc{F}}_+^X(\cdot,E)$ be as in Definition~\ref{def:peeling general}. Let $g_{ij}$ be as in Definition~\ref{def:gij}.

	If $V_{B_l'}\cap V_1^\infty=\emptyset$, then there exists $C>0$ depending on $\beta, M$ and $p(\beta)$ such that
	\begin{multline}\label{eq:noleaf}
		\left(\prod_{ij\in E(G_l')}|a_{ij}^{v^Z}|\indic_{\mc{B}_{ij}^c}\right)\left(\prod_{i\in V_{B_l'}\cap V_1^\infty }\frac{r_i}{d_{i,\bad} } \right)\indic_{V_X\subset I_\good }\indic_{\Pi_\mult^{V_X}=X}  \\ \leq C^{|V_{B_l'}|}\left(\prod_{S\in B_l':S\neq \hat{S}_l}r_S^2\right){\sum_{T'\in B((x_i,y_i)_{i\in \textcolor{black}{V_{B_l'}}},E(G_l')\cap T) }} \left(\prod_{ij\in T'}\frac{1}{d_{ij}^2}\indic_{ d_{ij}\leq \textcolor{black}{17}\ve_0\Cut}\right)\prod_{ij\in E(G_l)\cap F}g_{ij}.
	\end{multline}
	If $V_{B_l'}\cap V_1^\infty\neq \emptyset$, then there exists $C>0$ depending on $\beta, M$ and $p(\beta)$ such that
	\begin{multline}\label{eq:leaf}
		\left(\prod_{ij\in E(G_l') }|a_{ij}^{v^Z}|\indic_{\mc{B}_{ij}^c}\right)\left(\prod_{i\in V_{B_l'}\cap V_1^\infty}\frac{r_i}{d_{i,\bad}}\right)\indic_{V_X\subset I_\good }\indic_{\Pi_\mult^{V_X}=X} \\ \leq C^{|V_{B_l'}|}\left(\prod_{S\in B_l':S\neq \hat{S}_l}r_S^2\right){\sum_{T'\in B((x_i,y_i)_{i\in \textcolor{black}{V_{B_l'}}},E(G_l')\cap T) }} \\ \textcolor{black}{\times}\left(\prod_{ij\in T'}\frac{1}{d_{ij}^2}\indic_{ d_{ij}\leq \textcolor{black}{17}\ve_0\Cut}\right)\frac{r_{\hat{S}_l}}{d_{\hat{S}_l,\bad}}\prod_{ij\in E(G_l)\cap F}g_{ij}.
	\end{multline}
\end{lemma}

\medskip

\begin{proof}We work on the event where $V_X\subset I_\good$, $\Pi_\mult^{V_X}=X$ and where 
	\begin{equation*}
		\prod_{ij\in E(G_l') }|a_{ij}^{v^Z}|\neq 0.
	\end{equation*}
	\textcolor{black}{For each pendant tree $P$, let $A(P)$ be its attachment multipole, and let $\Root(P)\in A(P)$ be incident to an edge of $P$.}

	\paragraph{\bf{Step 1: the case $V_{B_l'}\cap V_1^\infty=\emptyset$}}
	By Lemma~\ref{lemma:weaker}, the parts $S\in X$ of degree $1$ in $(V_X,E)$ intersect $V_1^\infty$. Thus, we deduce that $G_l$ has no pendant tree, so $G_l'=G_l$. Since $G_l$ is minimally 2-edge-connected relative to $B_l$ (which implies $|E(G_l)|\le 2(|V_{B_l}|-1)$), we deduce from \eqref{bornesabZ} and Corollary~\ref{coro:prod a} that
	\begin{multline}\label{eq:2edge}
		\prod_{ij\in E(G_l)  }|a_{ij}^{v^Z}|\indic_{\mc{B}_{ij}^c}\\ \leq C^{|V_{B_l}|}\left(\prod_{S\in B_l:S\neq \hat{S}_l}r_S^2\right){\sum_{T'\in B((x_i,y_i)_{i\in \textcolor{black}{V_{B_l}}},E(G_l)\cap T) }}\left(\prod_{ij\in T'}\frac{1}{d_{ij}^2}\indic_{d_{ij}\leq \textcolor{black}{17} \ve_0\Cut}\right)\prod_{ij\in E(G_l)\cap F}g_{ij},
	\end{multline}
	\textcolor{black}{when $|B_l|\geq2$. When $|B_l|=1$, the graph $G_l$ has no edge, the sum in \eqref{eq:2edge} has the single term $T'=\emptyset$, and \eqref{eq:2edge} holds since its left-hand side is $1$. This proves the result in this case.}

	\paragraph{\bf{Step 2: the case $V_{B_l'}\cap V_1^\infty\neq \emptyset$}} Let us denote by $I_l$ the set of $i\in \{1,\ldots,k\}$ such that $P_i$ is attached to $G_l$, so that $G_l'=G_l\cup \cup_{i\in I_l}P_i$. 
	
	Assume first that $I_l=\emptyset$, i.e., $G_l'=G_l$. Let $i_0\in V_{B_l'}\cap V_1^\infty$. By the estimate \eqref{eq:geo1} of Lemma~\ref{lemma:geo}, for every $i\in I_\good$, we have $\max(r_i,\lambda)\leq d_{i,\bad}$. Thus, we deduce that
	\begin{equation}\label{eq:GlV}
		\prod_{i\in V_{B_l'}\cap V_1^\infty}\frac{r_i}{d_{i,\bad}}\leq  \frac{r_{i_0}}{d_{i_0,\bad}}.
	\end{equation}We next use $r_{i_0}\leq r_{\hat{S}_l}$. By the estimate \eqref{eq:geo4'} of Lemma~\ref{lemma:geo}, there exists a constant $C>0$ depending on $M$ and $p(\beta)$ such that
	\begin{equation*}
		d_{\hat{S}_l,\bad}\leq C^{|V_{B_l}|} d_{i_0,\bad}.
	\end{equation*}
	Inserting these bounds into \eqref{eq:GlV} gives, for a constant $C>0$ depending on $M$ and $p(\beta)$,
	\begin{equation*}
		\prod_{i\in V_{B_l}\cap V_1^\infty}\frac{r_i}{d_{i,\bad}}\leq C^{|V_{B_l}|} \frac{r_{\hat{S}_l}}{d_{\hat{S}_l,\bad}}.
	\end{equation*}
	Combining this with \eqref{eq:2edge} gives \eqref{eq:leaf} in the case $V_{B_l'}\cap V_1^\infty\neq \emptyset$ and $I_l=\emptyset.$

	\textcolor{black}{Now assume that $I_l\neq\emptyset$. By Lemma~\ref{lemma:weaker}, every nonroot leaf multipole $S$ of each $P_u$ contains a label $j_S\in S\cap V_1^\infty$. Let $J_u$ be the set of these labels. The $J_u$ are disjoint and avoid the core. Since $r_i/d_{i,\bad}\leq1$, we bound the other factors $r_i/d_{i,\bad}$ by $1$. Fix $r\in I_l$. By Lemma~\ref{lemma:pendant} and \eqref{eq:geo1'},} there exists $C>0$ depending on $\beta$, $M$ and $p(\beta)$ such that
	\begin{multline}\label{eq:pendant1}
		\prod_{ij\in E(P_r)} |a_{ij}^{v^Z}|\indic_{\mc{B}_{ij}^c}\textcolor{black}{\prod_{i\in J_r}}\frac{r_i}{d_{i,\bad}}\leq C^{|V(P_r)|}\frac{r_{\Root(P_r)} }{d_{\Root(P_r),\bad} } \prod_{\substack{S\in \Res(X,E(P_r)):\\ S\neq [\Root(P_r)]^X }}r_S^2\\ \times \prod_{ij\in E(P_r)\cap T}\frac{1}{d_{ij}^2}\indic_{d_{ij}\leq \textcolor{black}{17}\ve_0\Cut},
	\end{multline}
	where we recall that, for every set $E\subset \mc{E}^\inter(X)$, $\Res(X,E)$ stands for the set of $S\in X$ incident to some edge in $E$.

	For the other pendant trees, we use a rougher estimate:  using Lemma~\ref{lemma:pendant} and \eqref{eq:geo1'}, we get that there exists $C>0$ depending on $\beta, M$ and $p(\beta)$ such that, for every $u\in I_l\setminus \{r\}$,
	\begin{equation}\label{eq:pendant2}
		\prod_{ij\in E(P_u)} |a_{ij}^{v^Z}|\indic_{\mc{B}_{ij}^c}\textcolor{black}{\prod_{i\in J_u}}\frac{r_i}{d_{i,\bad}}\leq C^{|V(P_u)|} \prod_{\substack{ S\in \Res(X,E(P_u)):\\ S\neq [\Root(P_u)]^X}}r_S^2 \prod_{ij\in E(P_u)\cap T}\frac{1}{d_{ij}^2}\indic_{d_{ij}\leq \textcolor{black}{17}\ve_0\Cut},
	\end{equation}
	where we used that, by \eqref{eq:geo1}, for every $a\in V_X$, $r_a\leq d_{a,\bad}$. \textcolor{black}{Adding the pendant edges maps the trees of the sum in \eqref{eq:2edge} injectively to trees in $B(\cdot,E(G_l')\cap T)$ relative to $B_l'$.} Thus, combining \eqref{eq:2edge} for the 2-edge-connected part, \eqref{eq:pendant1} and \eqref{eq:pendant2}, we get that there exists $C>0$ depending on $\beta, M$ and $p(\beta)$ such that
	\begin{multline*}
		\left(\prod_{ij\in E(G_l') }|a_{ij}^{v^Z}|\indic_{\mc{B}_{ij}^c}\right)\left(\prod_{i\in V_{B_l'}\cap V_1^\infty}\frac{r_i}{d_{i,\bad}}\right)
		\\ \leq C^{|V_{B_l'}|}\frac{r_{\Root(P_r)}}{d_{\Root(P_r),\bad} } \prod_{S\in B_l':S\neq \hat{S}_l}r_S^2 {\sum_{T'\in B((x_i,y_i)_{i\in \textcolor{black}{V_{B_l'}}},E(G_l')\cap T) }}\prod_{ij\in  T' }\frac{1}{d_{ij}^2}\indic_{d_{ij}\leq \textcolor{black}{17}\ve_0\Cut}\prod_{ij\in E(G_l)\cap F}g_{ij}. 
	\end{multline*}
	By \eqref{eq:geo4'}, there exists $C>0$ depending on $M$ and $p(\beta)$ such that
	\begin{equation*}
		d_{\hat{S}_l,\bad}\leq C^{|V_{B_l}|}d_{\Root(P_r),\bad}.
	\end{equation*}
	Using this and $r_{\Root(P_r)}\leq r_{\hat{S}_l}$, we complete the proof of \eqref{eq:leaf}.
\end{proof}

\subsection{Proof of the quadratic estimate}\label{sub:quadra}

We can now prove the quadratic estimate on the product of the odd weights for minimal skeletons.

\begin{coro}[Quadratic estimate for minimal skeletons]\label{coro:prod a u}
	Let $X$ be a subpartition of $V_\good$ such that, for every $S\in X$, we have $|S|\leq p(\beta)$. \textcolor{black}{Let $\emptyset\neq V_1^\infty\subset V_X$, let $E_0\in\Eulc^X(V_1^\infty)$ and let $E=\Peeled_X(E_0)$.} Let $T\coloneqq \mc{T}_+^X(\cdot,E)$ and $F\coloneqq \mc{F}_+^X(\cdot,E)$ be as in Definition~\ref{def:peeling general}. Let $g_{ij}$ be as in Definition~\ref{def:gij} and $a_{i\infty}^{v^Z}$ be as in Definition~\ref{def:ainfty}. Fix an arbitrary $i_0\in V_{X}$.

	There exists $C>0$ depending on $\beta, M$ and $p(\beta)$ such that
	\begin{multline}\label{eq:squares}
		\left(\prod_{ij\in E}|a_{ij}^{v^Z}|\indic_{\mc{B}_{ij}^c}\right)\left(\prod_{i\in V_1^\infty}|a_{i\infty}^{v^Z}| \right)\indic_{V_X\subset I_\good}\indic_{\Pi_\mult^{V_X}=X} \leq C^{|V_X|}\left(\prod_{S\in X}r_{S}^2\right)\\ \times \sum_{T'\in B((x_i,y_i)_{i\in V_X}, T) }\left(\prod_{ij\in T'}\frac{1}{d_{ij}^2}\indic_{\min(r_i,r_j)\leq d_{ij}\leq \textcolor{black}{17}\ve_0\Cut} \right)\frac{1}{d_{i_0,\bad}^2} 
		\prod_{ij\in F}g_{ij}.
	\end{multline}
	In particular, letting $\hat{S}$ be the $S\in X$ such that $r_S$ is maximal, there exists $C>0$ depending on $\beta, M$ and $p(\beta)$ such that
	\begin{multline}\label{mainducorollaire}
		\left(\prod_{ij\in E}|a_{ij}^{v^Z}|\indic_{\mc{B}_{ij}^c}\right)\left(\prod_{i\in V_1^\infty}|a_{i\infty}^{v^Z}| \right)\indic_{V_X\subset I_\good}\indic_{\Pi_\mult^{V_X}=X}\leq C^{|V_X|}\left(\prod_{S\in X:S\neq \hat{S} }r_{S}^2\right)\\ \times {\sum_{T'\in B((x_i,y_i)_{i\in V_X}, T) }} \left(\prod_{ij\in T'}\frac{1}{d_{ij}^2}\indic_{\min(r_i,r_j)\leq d_{ij}\leq \textcolor{black}{17} \ve_0\Cut} \right)\prod_{ij\in F}g_{ij}.
	\end{multline}
\end{coro}

\medskip

\textcolor{black}{In the proof, we use the tree whose vertices are the blocks left by deleting all the bridges of the multigraph quotient, including the edges of the pendant trees. The weights $a_{i\infty}^{v^Z}$ at its leaves compensate for the extra radii at branching points. The bounds \eqref{eq:squares} and \eqref{mainducorollaire} also hold with $\indic_{V_X\subset I_\good}$ replaced by the indicator of the support of $\Psf_X^{+,\ve_0,Z}$. Indeed, the proof of Proposition~\ref{prop:bounded upper} checks \eqref{eq:geo1} there, and nonzero odd bonds and paths within each multipole then give \eqref{eq:geo4'} as before.}

\begin{proof}
	We work on the event where $V_X\subset I_\good$, $\Pi_\mult^{V_X}=X$ and where
	\begin{equation*}
		\prod_{ij\in E}|a_{ij}^{v^Z}|\neq 0.
	\end{equation*}
	{\color{black}
    We delete all the bridges of the multigraph quotient of $(V_X,E)$ by $X$, and we let $B_1,\ldots,B_n$ be the vertex sets of the remaining components, including the singletons. Let $G_l=(V_{B_l},E\cap\mc E^\inter(B_l))$, let $\mc E_\br$ be the set of all the deleted bridges, and let $\hat S_l\in B_l$ have maximal radius. The nontrivial blocks are those of the peeled core, and the other edges are the bridges of the core and the edges of the pendant trees. Since each peeled block is minimally $2$-edge-connected, $E$ is a minimal skeleton, and this decomposition gives the same $T$ and $F$ as Definition~\ref{def:peeling general}.
    Since peeling preserves every bridge cut of $E_0$, summing the parity condition of $E_0$ over the parts of a leaf block shows that this block contains an odd number of elements of $V_1^\infty$.
    }
	
	\paragraph{\bf{Step 1: the case $n=1$}}
	
	\textcolor{black}{Summing the parity condition over all parts shows that $|V_1^\infty|$ is even.} Hence, since $V_1^\infty\neq \emptyset$, we have $|V_1^\infty|\geq 2$. Let $i_1, i_2\in V_1^\infty$ be such that $i_1\neq i_2$. Since $E$ is minimally 2-edge-connected relative to $X$, we have $|E|\leq 2(|V_X|-1)$. Hence, by  Remark~\ref{remarkbornesab} and Corollary~\ref{coro:prod a}, there exists $C>0$ depending on $\beta$ such that
	\begin{multline*}
		\left(\prod_{ij\in E}|a_{ij}^{v^Z}|\indic_{\mc{B}_{ij}^c}\right)\left(\prod_{i\in V_1^\infty}|a_{i\infty}^{v^Z}| \right)\indic_{V_X\subset I_\good}\leq C^{|V_X|} \left(\prod_{S\in X:S\neq \hat{S} }r_{S}^2\right)\\ \times {\sum_{T'\in B((x_i,y_i)_{i\in V_X}, T) }}\left(\prod_{ij\in T'}\frac{1}{d_{ij}^2}\indic_{\min(r_i,r_j)\leq d_{ij}\leq \textcolor{black}{17}\ve_0\Cut} \right)\left(\prod_{ij\in F}g_{ij}\right)|a^{v^Z}_{i_1\infty}||a^{v^Z}_{i_2\infty}|.
	\end{multline*}
	Notice that we have used the fact that, for every $i\in V_X\setminus \{i_1,i_2\}$, by \eqref{bornesabinfty} and \eqref{eq:geo1}, there exists $C>0$ depending on $\beta$ such that $|a^{v^Z}_{i\infty}|\leq C$. Using Remark~\ref{remarkbornesab} again, we find $C>0$ depending on $\beta$ such that
	\begin{multline*}
		\left(\prod_{ij\in E}|a_{ij}^{v^Z}|\indic_{\mc{B}_{ij}^c}\right)\left(\prod_{i\in V_1^\infty}|a_{i\infty}^{v^Z}| \right)\indic_{V_X\subset I_\good}\leq C^{|V_X|} \left(\prod_{S\in X:S\neq \hat{S} }r_{S}^2\right)\\ \times {\sum_{T'\in B((x_i,y_i)_{i\in V_X}, T) }}\left(\prod_{ij\in T'}\frac{1}{d_{ij}^2}\indic_{\min(r_i,r_j)\leq d_{ij}\leq \textcolor{black}{17}\ve_0\Cut} \right)\left(\prod_{ij\in F}g_{ij}\right) \frac{r_{\hat{S}}^2}{d_{i_1,\bad}d_{i_2,\bad}}.
	\end{multline*}
	Using {\textcolor{black}{\eqref{eq:geo4'}}} in Lemma~\ref{lemma:geo}, we get that there exists $C>0$ depending on $M$ and $p(\beta)$ such that 
	\begin{equation*}
		\frac{1}{d_{[i_1]^X,\bad}d_{[i_2]^X,\bad}}\leq C^{|V_X|}\frac{1}{d_{i_0,\bad}^2}.
	\end{equation*}
	This proves \eqref{eq:squares} in the case $n=1$. In the rest of the proof, we assume that $n\geq 2$.
	
	\paragraph{\bf{Step 2: heavy-light decomposition}}
	
	\textcolor{black}{Let $Y\coloneqq\{V_{B_l}:1\leq l\leq n\}$.}
	
	\textcolor{black}{The quotient $\tilde T\coloneqq(V_X,\mc E_\br)/Y$ is connected, and it has no cycle since an edge on a cycle would not be a bridge. It is therefore a tree on $[n]$. Replacing each block by its extracted tree, we also obtain $T\in\mathsf T^X$. We root $\tilde T$ at a leaf $a\in[n]$, and we let $\loc_{\tilde T}$ be as in Definition~\ref{def:heavy}, with the path decomposition and the counts of Lemma~\ref{lemma:heavy light}. We define} 
	\begin{equation*}
		X^*\coloneqq \bigcup_{l\in \Leaves(\tilde{T})\setminus \{a\} } \{ S\in B_l:\deg_{\mc{E}_\br}(S)=1\}.
	\end{equation*}
	
	Let $l\in \Leaves(\tilde{T})\setminus\{a\}$ and let $u$ be the unique neighbor of $\loc_{\tilde{T}}(l)$ on the path from $l$ to $\loc_{\tilde{T}}(l)$ in $\tilde{T}$. Let $S$ be the unique element in $X^*\cap B_l$. We denote by $\loc(S)$ the unique $S'\in B_{\loc(l)}$ such that there exist $S''\in B_u$ and an edge in $\mc{E}_\br$ incident to $S''$ and $S'$ (see the figure below).

	\begin{figure}[H]
		\centering
		\fbox{\includegraphics[width=0.7\textwidth]{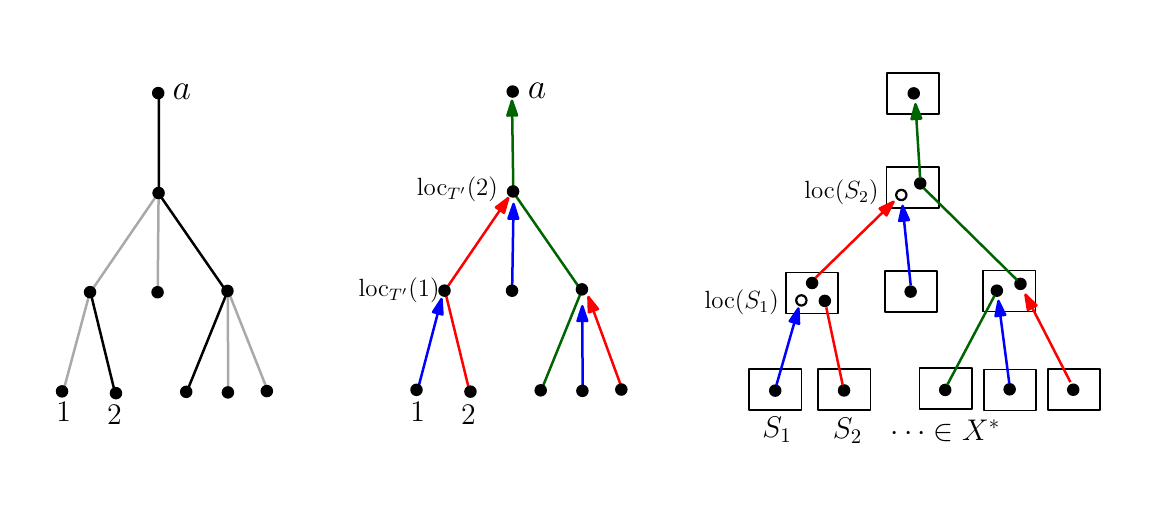}} 
		\caption{Left and center: the quotient graph $\tilde{T}$ rooted at $a$ and its heavy-light decomposition. Right: the graph $(V_X,\mc{E}_\br)/X$. Each square represents a block $B_l$. Each vertex represents a multipole $S\in B_l$ incident to some edge in $\mc{E}_\br$. The nodes $S_1,S_2,\dots$ along the bottom are elements of $X^*$, and the circles stand for $\loc(S_1)$ and $\loc(S_2)$.}\label{figure:lift}
	\end{figure}

	\paragraph{\bf{Step 3: bounding the weights on bridges}}
	By the estimate \eqref{bornesabZ} of Remark~\ref{remarkbornesab}, there exists $C>0$ depending on $\beta$ such that
	\begin{equation}\label{eq:bridges}
		\begin{split}
			\prod_{ij\in \mc{E}_\br}|a_{ij}^{v^Z}|\indic_{\mc{B}_{ij}^c}&\leq C^{|\mc{E}_\br|} \prod_{i\in V_X}r_i^{\deg_{\mc{E}_\br}(i)} \prod_{ij\in \mc{E}_\br}\frac{1}{d_{ij}\max(d_{ij},r_i,r_j)}\\
			&\leq  C^{|\mc{E}_\br|} \prod_{S\in X} r_{S}^{\deg_{\mc{E}_\br}(S)} \prod_{ij\in \mc{E}_\br}\frac{1}{d_{ij}^2}.
		\end{split}
	\end{equation}
	We next use that, by Lemma~\ref{lemma:energy goodbad}, there exists $C>0$ depending on $\beta$ such that
	\begin{equation}\label{eq:rep simp}
		\prod_{i\in V_1^\infty}|a_{i\infty}^{v^Z}|\leq C^{|V_X|}\prod_{i\in V_1^\infty}\frac{r_i}{d_{i,\bad}}.
	\end{equation}
	{\color{black}
    Corollary~\ref{coro:prod a} gives \eqref{eq:2edge} for every $G_l$ with $|B_l|\geq2$. When $|B_l|=1$, the block has no edge and its tree sum is $1$. If $V_{B_l}\cap V_1^\infty\neq\emptyset$, we keep the ratio $r_i/d_{i,\bad}$ of one $i\in V_{B_l}\cap V_1^\infty$ and bound the others by $1$. Then \eqref{eq:geo1} and \eqref{eq:geo4'} bound the product of these ratios by $C^{|V_{B_l}|}r_{\hat S_l}/d_{\hat S_l,\bad}$. The block trees together with $\mc E_\br$ belong to $B(\cdot,T)$, and distinct choices give distinct trees. Combining these bounds with \eqref{eq:bridges} and \eqref{eq:rep simp}, we obtain}
	\begin{multline*}
		\left(\prod_{ij\in E}|a_{ij}^{v^Z}|\indic_{\mc{B}_{ij}^c}\right) \left(\prod_{i\in V_1^\infty}|a_{i\infty}^{v^Z}| \right)\leq C^{|V_X|}\left(\prod_{S\in X\setminus\cup_l \{\hat{S}_l\}}r_S^2\right)\left( \prod_{S\in X} r_{S}^{\deg_{\mc{E}_\br}(S)} \right)\\ \times \prod_{l\in [n]:\textcolor{black}{V_{B_l}}\cap V_1^\infty\neq \emptyset } \frac{r_{\hat{S}_l}}{d_{\hat{S}_l,\bad}}
		{\sum_{T'\in B((x_i,y_i)_{i\in V_X}, T) }}\prod_{ij\in T'}\frac{1}{d_{ij}^2}\indic_{d_{ij}\leq \textcolor{black}{17}\ve_0\Cut}\prod_{ij\in F}g_{ij},
	\end{multline*}where we recall that, for every $l\in [n]$, $\hat{S}_l$ is the index $S$ in $B_l$ such that $r_S$ is maximal.

	\textcolor{black}{Every leaf block, including the root, meets $V_1^\infty$. Moreover, \eqref{eq:geo1} and \eqref{eq:geo1'} give $r_S\leq C d_{S,\bad}$, so we bound the other ratios $r_{\hat S_l}/d_{\hat S_l,\bad}$ by $C$, at a cost $C^n$, and obtain}
	\begin{multline*}
		\left(\prod_{ij\in E}|a_{ij}^{v^Z}|\indic_{\mc{B}_{ij}^c}\right) \left(\prod_{i\in V_1^\infty}|a_{i\infty}^{v^Z}| \right) \leq C^{|V_X|}\left(\prod_{S\in X\setminus\cup_l \{\hat{S}_l\}}r_S^2\right)\prod_{S\in X} r_{S}^{\deg_{\mc{E}_\br}(S)} \frac{r_{\hat{S}_a}}{d_{\hat{S}_a,\bad}} 	\\ \times	\prod_{l\in \Leaves(\tilde{T})\setminus \{a\}}  \frac{r_{\hat{S}_l}}{d_{\hat{S}_l,\bad}}{\sum_{T'\in B((x_i,y_i)_{i\in V_X}, T) }}\prod_{ij\in T'}\frac{1}{d_{ij}^2}\indic_{d_{ij}\leq \textcolor{black}{17}\ve_0\Cut}\prod_{ij\in F}g_{ij}.
	\end{multline*}
	Hence, using \eqref{eq:geo4'} again to replace $d_{\hat{S}_l,\bad}$ by $d_{S,\bad}$ with $S$ the unique element in $B_l\cap X^*$, we find that there exists $C>0$ depending on $\beta, M$ and $p(\beta)$ such that 
	\begin{multline*}
		\left(\prod_{ij\in E}|a_{ij}^{v^Z}|\indic_{\mc{B}_{ij}^c}\right) \left(\prod_{i\in V_1^\infty}|a_{i\infty}^{v^Z}| \right)\leq C^{|V_X|}\left(\prod_{S\in X\setminus\cup_l \{\hat{S}_l\}}r_S^2\right)\prod_{S\in X} r_{S}^{\deg_{\mc{E}_\br}(S)} \frac{r_{\hat{S}_a}}{d_{\hat{S}_a,\bad}} \\ \times \prod_{l\in \Leaves(\tilde{T})\setminus \{a\} }r_{\hat{S}_l}\prod_{S\in X^*}\frac{1}{d_{S,\bad}} 
		{\sum_{T'\in B((x_i,y_i)_{i\in V_X}, T) }}\prod_{ij\in T'}\frac{1}{d_{ij}^2}\indic_{d_{ij}\leq \textcolor{black}{17}\ve_0\Cut}\prod_{ij\in F}g_{ij}.  
	\end{multline*}
	Notice that
	\begin{equation*}
		\prod_{S\in X} r_{S}^{\deg_{\mc{E}_\br}(S)}=\left(\prod_{l\neq a:\deg_{\tilde{T}}(l)\leq 2}\prod_{S\in B_l}r_S^{\deg_{\mc{E}_\br}(S)}\right)\left(\prod_{\substack{ l:\deg_{\tilde{T}}(l)>2\\ \mathrm{or }\, l=a}}\prod_{S\in B_l} r_{S}^{\deg_{\mc{E}_\br}(S)}\right).
	\end{equation*}
	Moreover, by the definition of $\hat{S}_l$,
	\begin{equation*}
		\prod_{l\neq a:\deg_{\tilde{T}}(l)\leq 2}\prod_{S\in B_l}r_S^{\deg_{\mc{E}_\br}(S)}\leq \prod_{l\neq a:\deg_{\tilde{T}}(l)=1}r_{\hat{S}_l}\prod_{l\neq a:\deg_{\tilde{T}}(l)=2}r_{\hat{S}_l}^2.
	\end{equation*}
	Therefore, combining the  above two displays, we get
	\begin{equation*}
		\prod_{S\in X} r_{S}^{\deg_{\mc{E}_\br}(S)}\prod_{l\in \Leaves(\tilde{T})\setminus \{a\}} r_{\hat{S}_l}\leq \prod_{l\neq a:\deg_{\tilde{T}}(l)\leq 2} r_{\hat{S}_l}^2  \left(\prod_{\substack{ l:\deg_{\tilde{T}}(l)>2\\ \mathrm{or }\, l=a}}\prod_{S\in B_l} r_{S}^{\deg_{\mc{E}_\br}(S)}\right).
	\end{equation*}
	Thus, there exists $C>0$ depending on $\beta, M$ and $p(\beta)$ such that 
	\begin{multline}\label{eq:obs}
		\left(\prod_{ij\in E}|a_{ij}^{v^Z}|\indic_{\mc{B}_{ij}^c}\right) \left(\prod_{i\in V_1^\infty}|a_{i\infty}^{v^Z}| \right)\leq C^{|V_X|}\left(\prod_{S\in X\setminus\cup_l \{\hat{S}_l\}}r_S^2\right) \prod_{l\neq a:\deg_{\tilde{T}}(l)\leq 2} r_{\hat{S}_l}^2\\ \times  \left(\prod_{\substack{ l:\deg_{\tilde{T}}(l)>2\\ \mathrm{or}\, l=a}}\prod_{S\in B_l} r_{S}^{\deg_{\mc{E}_\br}(S)}\right)\frac{r_{\hat{S}_a}}{d_{\hat{S}_a,\bad}} \prod_{S'\in X^*}\frac{1}{d_{S',\bad}} 
		{\sum_{T'\in B((x_i,y_i)_{i\in V_X}, T) }}\prod_{ij\in T'}\frac{1}{d_{ij}^2}\indic_{d_{ij}\leq \textcolor{black}{17}\ve_0\Cut}\prod_{ij\in F}g_{ij}.  
	\end{multline}

	\paragraph{\bf{Step 4: compensation of high-degree multipoles by distances to bad points}}
	We now bound 
	\begin{equation*}
		\frac{r_{\hat{S}_a}}{d_{\hat{S}_a,\bad}}\left(\prod_{\substack{ l:\deg_{\tilde{T}}(l)>2\\ \mathrm{or }\, l=a}}\prod_{S\in B_l} r_{S}^{\deg_{\mc{E}_\br}(S)}\right) \prod_{S'\in X^*}\frac{1}{d_{S',\bad}}.  
	\end{equation*} Let us write 
	\begin{equation*}
		\prod_{S'\in X^*}\frac{1}{d_{S',\bad}}=\prod_{S\in X}\prod_{\substack{S'\in X^*:\\\loc(S')=S}}\frac{1}{d_{S',\bad}}.
	\end{equation*}
	By \eqref{eq:geo4'}, there exists $C>0$ depending on $M$ and $p(\beta)$ such that
	\begin{equation*}
		d_{S,\bad}\leq C^{\dist_{T}(S,S')}d_{S',\bad},
	\end{equation*}
	\textcolor{black}{where $\dist_T(S,S')$ is the number of edges of the path between $S$ and $S'$ in the quotient tree $(V_X,T)/X$.} Therefore,
	\begin{equation*}
		\prod_{S'\in X^*}\frac{1}{d_{S',\bad}}\leq C^{\sum_{S'\in X^*} \dist_T(S',\loc(S'))} \prod_{S\in X}\Bigl(\frac{1}{d_{S, \bad}}\Bigr)^{|\{S'\in X^*:\loc(S')=S\}|}.
	\end{equation*}
	\textcolor{black}{Only the path arriving through the heavy child continues inside a block, and the other paths stop where their last bridge arrives. Their lifts to $(V_X,T)/X$ are therefore edge-disjoint, and}
	\begin{equation*}
		\sum_{S'\in X^*} \dist_T(S',\loc(S')) \textcolor{black}{\leq |X|-1}\leq |V_X|-1.
	\end{equation*}
	
	Thus, there exists $C>0$ depending on $M$ and $p(\beta)$ such that
	\begin{equation}\label{eq:hl}
		\prod_{S'\in X^*}\frac{1}{d_{S',\bad}}\leq C^{|V_X|}\prod_{S\in X} \Bigl(\frac{1}{d_{S,\bad}}\Bigr)^{|\{S'\in X^*:\loc(S')=S\}|}.
	\end{equation}
	
	\textcolor{black}{Let $l\neq a$ satisfy $\deg_{\tilde T}(l)>2$. We orient the bridges towards the root block, and we let $\Child_{\mc E_\br}(S)$ be the set of child bridges incident to $S$. The block $B_l$ has a unique part $S_l$ incident to its parent bridge.} For this vertex,
	\begin{equation*}
		\deg_{\mc{E}_\br}(S_l)=|\Child_{\mc{E}_\br}(S_l)|+1,
	\end{equation*}
	whereas for every $S\in B_l\setminus \{S_l\}$,
	\begin{equation*}
		\deg_{\mc{E}_\br}(S)=|\Child_{\mc{E}_\br}(S)|.
	\end{equation*}
	Consequently,
	\begin{equation}\label{eq:jj1}
		\prod_{S\in B_l}r_{S}^{\deg_{\mc{E}_\br}(S)}\leq r_{\hat{S}_l}\prod_{S\in B_l}r_S^{|\Child_{\mc{E}_\br}(S)|}.
	\end{equation}
	By the construction of Definition~\ref{def:heavy}, exactly one edge to the children of $l$ is heavy. Let $S_l'$ be the unique $S\in B_l$ incident to the heavy edge to a child of $l$ (see Figure~\ref{figure:lift} and notice that we can have $S_l=S_l'$). Observe that, for this vertex,
	\begin{equation*}
		|\{S'\in X^*:\loc(S')=S_l'\}|=|\Child_{\mc{E}_\br}(S_l')|-1,
	\end{equation*}
	whereas for every $S\in B_l\setminus\{S_l'\}$,
	\begin{equation*}
		|\{S'\in X^*:\loc(S')=S\}|=|\Child_{\mc{E}_\br}(S)|.
	\end{equation*}
	Therefore, combining the above two displays, we obtain
	\begin{equation}\label{eq:jj2}
		\prod_{S\in B_l}r_S^{|\Child_{\mc{E}_\br}(S)|}\leq r_{\hat{S}_l}\prod_{S\in B_l}r_S^{|\{S'\in X^*:\loc(S')=S\}|}.
	\end{equation}
	
	Combining \eqref{eq:jj1}, \eqref{eq:jj2} and \eqref{eq:geo1'}, we obtain the existence of $C>0$ depending on $\beta, M$ and $p(\beta)$ such that
	\begin{equation}\label{eq:noleafa}
		\textcolor{black}{\left(\prod_{S\in B_l}r_S^{\deg_{\mc E_\br}(S)}\right)
        \prod_{S\in B_l}d_{S,\bad}^{-|\{S'\in X^*:\loc(S')=S\}|}
        \leq C^{\deg_{\tilde T}(l)}r_{\hat S_l}^2.}
	\end{equation}
	
	It remains to consider the vertices in $B_a$. For every $S\in B_a$, we have 
	\begin{equation}\label{eq:SSSS}
		|\{S'\in X^*:\loc(S')=S\}|=\deg_{\mc{E}_\br}(S). 
	\end{equation}
	\textcolor{black}{Since the root is a leaf of $\tilde T$, exactly one bridge meets $B_a$. Let $S_0\in B_a$ be the part incident to it. Then}
    \begin{equation*}
        \textcolor{black}{\left(\prod_{S\in B_a}r_S^{\deg_{\mc E_\br}(S)}\right)
        \prod_{S\in B_a}d_{S,\bad}^{-|\{S'\in X^*:\loc(S')=S\}|}
        =\frac{r_{S_0}}{d_{S_0,\bad}}.}
    \end{equation*}
    \textcolor{black}{Applying \eqref{eq:geo4'} within $G_a$, we obtain}
    \begin{equation}\label{eq:leafa}
		\textcolor{black}{\left(\prod_{S\in B_a}r_S^{\deg_{\mc E_\br}(S)}\right)
        \prod_{S\in B_a}d_{S,\bad}^{-|\{S'\in X^*:\loc(S')=S\}|}
        \leq C^{|V_{B_a}|}\frac{r_{\hat S_a}}{d_{\hat S_a,\bad}}.}
	\end{equation}

	\textcolor{black}{Since $\sum_l\deg_{\tilde T}(l)=2(n-1)$, taking} the product of \eqref{eq:noleafa} for $l $ such that $\deg_{\tilde{T}}(l)>2$ and $l\neq a$, and multiplying it by
	\eqref{eq:leafa} and then by $\frac{r_{\hat{S}_a}}{d_{\hat{S}_a,\bad}}$, gives the existence of $C>0$ depending on $M$ and $p(\beta)$ such that
	\begin{equation}\label{eq:insertme2}
		\frac{r_{\hat{S}_a}}{d_{\hat{S}_a,\bad}} \left(\prod_{\substack{ l:\deg_{\tilde{T}}(l)>2\\ \mathrm{or }\, l=a}} \prod_{S\in B_l} r_{S}^{\deg_{\mc{E}_\br}(S)}\right) \prod_{S'\in X^* }\frac{1}{d_{S',\bad}}\leq C^{|V_X|}\left(\prod_{\substack{ l:\deg_{\tilde{T}}(l)>2\\ \mathrm{or }\, l=a}} r_{\hat{S}_l}^2\right)\frac{1}{d_{\hat{S}_a,\bad}^2}.
	\end{equation}
	\paragraph{\bf{Step 5: conclusion}}
	Notice that, by using {\textcolor{black}{\eqref{eq:geo4'}}}, one can replace $d_{\hat{S}_a,\bad}$ by $d_{i_0,\bad}$ up to a multiplicative constant $C^{|V_X|}$, with $C$ depending on $M$ and $p(\beta)$. Inserting \eqref{eq:insertme2} into \eqref{eq:obs} gives the desired result.
\end{proof}

\subsection{Proof of the bounded-cluster results with frozen bad points}\label{sub:abs upper}

\begin{proof}[Proof of Proposition~\ref{prop:bounded upper}]
    {\color{black}
    We first check \eqref{eq:geo1} on the support of $\Psf_X^{+,\ve_0,Z}$. Suppose that $d_{i,\bad}<A\coloneqq\max(r_i,\lambda)$. The nearest bad dipole has type $3$, since \eqref{def:IgoodZ} gives $d_{ij}>MA$ for every $j\in I_\bad^1\cup I_\bad^2$. From this dipole, a path of at most $p(\beta)-1$ steps in its block of $\Pi_\mult^{V_\bad\setminus I_\bad^1}$ leads to an element of $J_\bad^3$, and one final step given by \eqref{nelledef} leads to a dipole of type $1$ or $2$. If a type-$3$ dipole $j$ of this path is at distance $D<MA$ from $i$, then \eqref{def:IgoodZ} gives $\max(r_j,\lambda)<D/M$, so the next distance is at most $D+(M+1)\max(r_j,\lambda)<3D$. The final distance is less than $3^{p(\beta)}A<MA$, which contradicts \eqref{def:IgoodZ}. Thus \eqref{eq:geo1} holds on this support. Since the packing argument of Lemma~\ref{lemma:energy goodbad} uses only $Z$, it gives \eqref{eq:interaction gb} there.

    Reversing all dipoles in $S$ preserves $\Psf_{\{S\}}^{+,\ve_0,Z}$ and reverses the sign of their interaction with $Z$. Hence, by \eqref{eq:interaction gb},
    \begin{equation}\label{eq:external positivity}
        1\leq g_{\ve_0}^Z(S)
        =\dE_{\Psf_{\{S\}}^{+,\ve_0,Z}}\left[\cosh\left(\beta\sum_{i\in S,j\in V_\bad}v_{ij}^Z\right)\right]\leq C,
    \end{equation}
    where $C$ depends on $\beta$, $p(\beta)$ and $|S|$. In particular, $|\widehat{\Ksf}_{\ve_0}^{+,Z}(X)|\leq|\Ksf_{\ve_0}^{+,Z}(X)|$.
    }

	Let $k\coloneqq |V_X|$. Taking $w=v^Z$ in Lemma~\ref{lemma:lowerM+} proves the estimates \eqref{eq:MX +low} and \eqref{eq:MX +up}. \textcolor{black}{When $|S|=1$, the upper bound follows from $\tau_i^Z\geq r_i/4$. For the lower bound, we restrict to $r_i\leq\lambda$ and keep $x_i$ away from the boundary and from the balls of radius $(M+1)\lambda$ around the bad endpoints, which cover an area of at most $C|V_\bad|\lambda^2\leq CN\lambda^2$. By \eqref{eq:external positivity}, the same bounds hold for $\widehat{\Msf}_{\ve_0}^{+,Z}$.}

	Let us now prove \eqref{eq:K+ bound}. \textcolor{black}{We may assume that $|X|\geq2$.} By Corollary~\ref{coro:reduction Euler u}, 	\begin{multline*}
		\Ksf_{\ve_0}^{+,Z}(X)=\sum_{\substack{V_1^\infty, V_2^\infty \subset V_X:\\ V_1^\infty\cap V_2^\infty=\emptyset }}  \sum_{K=\textcolor{black}{0}}^\infty \frac{1}{K!}\sum_{\substack{Y_1,\ldots,Y_K\subset X\\ \mathrm{disjoint} }}\prod_{l=1}^K
		\Biggl(\sum_{ E_{l,1}\in \Eul^{Y_l}(V_{Y_l}\cap V_1^\infty ) }  \sum_{\substack{E_{l,2}:E_{l,1}\cup E_{l,2}\in \mathsf{E}^{Y_l}\\
				E_{l,1}\cap E_{l,2}=\emptyset  }} \Biggr)\\ \sum_{F\in \textcolor{black}{\mathsf{E}_0}^{\Coarse_X(Y_1,\ldots,Y_K)}}\dE_{\Psf_{X}^{+,\ve_0,Z}}\biggl[\prod_{l=1}^K\left(\prod_{ij\in E_{l,1}}a^{v^Z}_{ij}\prod_{ij\in E_{l,2} }b^{v^Z}_{ij}\right)\prod_{ij\in \cup_l\mc{E}^\inter(Y_l)}\indic_{\mc{B}_{ij}^c}\\ \times\prod_{ij\in F}(-\indic_{\mc{B}_{ij}})\prod_{i\in V_1^\infty}a_{i\infty}^{v^Z}\prod_{i\in V_2^\infty}b_{i\infty}^{v^Z} \biggr].
	\end{multline*}
	We first use that, by Lemma~\ref{lemma:error'},
	\begin{equation*}
		\indic_{\mc{B}_{ij}}+|b_{ij}^{v^Z}|\indic_{\mc{B}_{ij}^c} \leq \textcolor{black}{C_\beta}b_{ij}^\abs,
	\end{equation*}
	where $b_{ij}^\abs$ is as in Definition~\ref{def:abs ab}.
	
	Let $X_1,\ldots,X_n$ be the connected components, relative to $X$, of the graph $(V_X,\cup_{l=1}^K E_{l,1})$ that contain at least two multipoles. \textcolor{black}{We keep all the separation indicators within each $X_l$, which give $\Pi_\mult^{V_{X_l}}=X_l$ on the support of the measure.} One may extract from $(V_X,\cup_{l=1}^KE_{l,2}\cup F)$ a tree $T^b\in \textcolor{black}{\mathsf{T}_0}^{\Coarse_X(X_1,\ldots,X_n)}$. Therefore, there exists $C>0$ depending on $\beta, M$ and $k$ such that
	\begin{multline*}
		|\Ksf_{\ve_0}^{+,Z}(X)|\leq  C\max_n\max_{\substack{V_1^\infty,V_2^\infty\subset V_X\\ \mathrm{disjoint} } } \max_{\substack{ X_1,\ldots,X_n\subset X\\ \mathrm{disjoint} }}\prod_{l=1}^n\max_{E_l\in \textcolor{black}{\Eulc}^{X_l}(V_1^\infty\cap V_{X_l})} \max_{T^b\in \textcolor{black}{\mathsf{T}_0}^{\Coarse_X(X_1,\ldots,X_n)}}\\
		\dE_{\Psf_{X}^{+,\ve_0,Z}}\left[\prod_{l=1}^n\left(\prod_{ij\in E_l}|a_{ij}^{v^Z}|\textcolor{black}{\prod_{ij\in\mc E^\inter(X_l)}\indic_{\mc B_{ij}^c}}\right) \prod_{ij\in T^b}b_{ij}^\abs \prod_{i\in V_1^\infty}|a_{i\infty}^{v^Z}|\prod_{i\in V_2^\infty}|b_{i\infty}^{v^Z}|\right].
	\end{multline*}
	Combining \eqref{bornesabinfty} and \eqref{eq:geo1}, we bound every term $b_{i\infty}^{v^Z}$ by a constant and get
	\begin{multline*}
		|\Ksf_{\ve_0}^{+,Z}(X)|\leq  C\max_n\max_{V_1^\infty\subset V_X} \max_{\substack{ X_1,\ldots,X_n\subset X\\ \mathrm{disjoint} }}\prod_{l=1}^n\max_{E_l\in \textcolor{black}{\Eulc}^{X_l}(V_1^\infty\cap V_{X_l})} \max_{T^b\in \textcolor{black}{\mathsf{T}_0}^{\Coarse_X(X_1,\ldots,X_n)}}\\
		\dE_{\Psf_{X}^{+,\ve_0,Z}}\left[\prod_{l=1}^n\left(\prod_{ij\in E_l}|a_{ij}^{v^Z}|\textcolor{black}{\prod_{ij\in\mc E^\inter(X_l)}\indic_{\mc B_{ij}^c}}\right) \prod_{ij\in T^b}b_{ij}^\abs \prod_{i\in V_1^\infty}|a_{i\infty}^{v^Z}|\right].
	\end{multline*}
	\textcolor{black}{We peel each $E_l$ and bound the removed odd weights by a constant, which costs a factor $C(\beta,M,k)$.} Fix $l\in \{1,\ldots,n\}$ and denote by $\hat{S}_l$ the $S\in X_l$ such that $r_S$ is maximal. Also let \textcolor{black}{$T_l\coloneqq\mc T_+^{X_l}(\cdot,E_l)$ and $F_l\coloneqq\mc F_+^{X_l}(\cdot,E_l)$} be as in Definition~\ref{def:peeling general}. By  Remark~\ref{remarkbornesab}, Corollary~\ref{coro:prod a} or Corollary~\ref{coro:prod a u} \eqref{mainducorollaire} \textcolor{black}{(according to whether $V_1^\infty\cap V_{X_l}$ is empty)}, there exists $C>0$ depending on $\beta$, $M$ and $k$ such that
	\begin{multline*}
		\prod_{ij\in E_{l}}|a_{ij}^{v^Z}|\textcolor{black}{\prod_{ij\in\mc E^\inter(X_l)}\indic_{\mc B_{ij}^c}} \prod_{i\in V_1^{\infty}\cap V_{X_l}}|a_{i\infty}^{v^Z}| \\ \leq C\prod_{S\in X_l:S\neq \hat{S}_l}r_S^2 {\sum_{T'\in B((x_i,y_i)_{i\in \textcolor{black}{V_{X_l}}}, T_l) }} \prod_{ij\in T'}\frac{1}{d_{ij}^2}\indic_{d_{ij}\leq \textcolor{black}{17}\ve_0\Cut}\indic_{\mc{B}_{ij}^c} \prod_{ij\in F_l}g_{ij}.
	\end{multline*}
	
	{\color{black}
    If $V_1^\infty\cap V_{X_l}=\emptyset$, then \eqref{eq:hate F} bounds $\prod_{ij\in F_l}g_{ij}$ by $\min(\max_{i\in V_{X_l}}r_i/\max_{e\in T_l}d_e,1)^2$. Otherwise, we combine \eqref{eq:squares} and \eqref{mainducorollaire}, with $i_0$ of maximal radius in $V_{X_l}$. On the support of the separation indicators and of the odd product, \eqref{eq:geo4'} and $d_{ij}\leq\tau_i^Z+\tau_j^Z$ bound every edge length of $T_l$ by $C^k d_{i_0,\bad}$. An edge changed in $B(\cdot,T_l)$ has length at most $r_{\hat S_l}\leq d_{i_0,\bad}$. Thus, after we bound $\prod_{ij\in F_l}g_{ij}$ by $1$, the two bounds give, for every $T'\in B(\cdot,T_l)$, the factor $C^k\min(\max_{i\in V_{X_l}}r_i/\max_{e\in T'}d_e,1)^2$. We now drop the separation indicators except those on $T^\prime$, and multiply over $l$. The minimum below is $1$ when there are no odd edges. Consequently,}
	\begin{multline*}
		|\Ksf_{\ve_0}^{+,Z}(X)|\leq \frac{C}{\Msf_{\ve_0}^{+,Z}(X)(NC_{\beta,\lambda,\ve_0})^k }\max_n\max_{\substack{ X_1,\ldots,X_n\subset X\\ \mathrm{disjoint} }}\max_{T_1^a\in \mathsf{T}^{X_1},\ldots,T_n^a\in \mathsf{T}^{X_n} } \max_{T^b\in \textcolor{black}{\mathsf{T}_0}^{\Coarse_X(X_1,\ldots,X_n)}}\\ \int_{(\Lambda^{2})^k} \left(\prod_{l=1}^n \prod_{S\in X_l:S\neq \hat{S}_l }r_S^2\right)\prod_{ij\in T^b}b_{ij}^\abs\prod_{ij\in \textcolor{black}{\cup_l T_l^a}}\frac{1}{d_{ij}^2}\indic_{\mc{B}_{ij}^c} \indic_{ d_{ij}\leq \textcolor{black}{17} \ve_0\Cut}\min\Bigl(\frac{\max_{i\in V_X}r_i}{\max_{e\in \cup_l T_l^a}d_e },1\Bigr)^2 \prod_{S\in X}\indic_{\mc{B}_S} \\ \times \prod_{i\in V_X}e^{\beta \g_\lambda(r_i)}\textcolor{black}{\indic_{r_i\leq\ve_0\Cut}}\textcolor{black}{\dd x_i\dd y_i}.
	\end{multline*}
	Thus, inserting the estimate \eqref{eq:bound J1} of Lemma~\ref{lemma:integral small -} in the appendix, we deduce that there exists a constant $C>0$ depending on $\beta$, $M$\textcolor{black}{, $\ve_0$} and $k$ such that
	\begin{equation}\label{eq:o4}
		|\Ksf_{\ve_0}^{+,Z}(X)|\leq \frac{C}{\Msf_{\ve_0}^{+,Z}(X)(NC_{\beta,\lambda,\ve_0})^k }  N\lambda^{(2-\beta)k}\gamma_{\beta,\lambda,k}.
	\end{equation}
	
\begingroup\color{black}
By \eqref{eq:MX +low}, \eqref{eq:MX -up} and the case $|S|=1$ above, we have, for every $S'\in X$ with $s=|S'|$,
\[
\Msf_{\ve_0}^{+,Z}(S')\geq cN^{1-s}\lambda^{2(s-1)},
\qquad \Msf_\infty^0(S')\leq CN^{1-s}\lambda^{2(s-1)}.
\]
Multiplying over $S'\in X$ gives
\endgroup
\begin{equation}\label{eq:o5}
		\Msf_{\ve_0}^{+,Z}(X)\geq C^{-1}\Msf^0_{\infty} (X).
	\end{equation}
	Therefore, combining \eqref{eq:o4} with \eqref{eq:o5} and \eqref{eq:Clambda bound}, we conclude that there exists $C>0$ depending on $\beta,M$\textcolor{black}{, $\ve_0$} and $k$ such that 
	\begin{equation}\label{eq:empty}
		|\Ksf_{\ve_0}^{+,Z}(X)|\leq \frac{CN^{1-k} }{\Msf^0_{\infty} (X) }\gamma_{\beta,\lambda,k}.
	\end{equation}
	\textcolor{black}{The same multipole bounds give $\Msf_{\ve_0}^0(X)\leq C_{\beta,M,k}\Msf_\infty^0(X)$, which yields \eqref{eq:K+ bound}.}
\end{proof}

\begin{proof}[Proof of Proposition~\ref{prop:expansion +}]

	{\color{black}
Write $\alpha=\alpha(\beta)$ and
    \[
        \theta_s=C_{\ve_0}\Cut^{-2}+\ve_0^{\alpha s}|V_\bad|/N,\qquad 1\leq s\leq p(\beta).
    \]
    Here $C$ depends only on $\beta,M,p(\beta)$, whereas $C_{\ve_0}$ may also depend on $\ve_0$. We take $\lambda^2\leq\ve_0^\alpha$, so that $\lambda^{2(k-s)}\theta_s\leq\theta_k$ for $s\leq k$.
}
Let us first prove \eqref{eq:M+diff}. Define
	\begin{multline*}
		\textcolor{black}{I_{\mathrm{mult},\ve_0}^{+,Z}}(S)\\=\int_{(\Lambda^2)^{|S|} }\prod_{i\in S} \indic_{i\in I_\good^Z} \indic_{\mc{B}_S}\prod_{ i,j\in S:i<j}e^{-\beta v^Z_{ij}}\indic_{\mc{A}_{ij}}\prod_{i\in S }e^{L\frac{|x_i-y_i|^2}{(\tau_i^Z)^2}}  \prod_{i\in S} e^{\beta \g_\lambda(x_i-y_i)}\indic_{|x_i-y_i|\leq \ve_0\Cut}\dd x_i \dd y_i,
	\end{multline*}
	and
	\begin{equation*}
		\textcolor{black}{I_{\mathrm{mult},\infty}^{0}}(S)=\int_{(\Lambda^2)^{|S|} }\indic_{\mc{B}_S}\prod_{ i,j\in S:i<j}e^{-\beta v_{ij}}\indic_{\mc{A}_{ij}}\prod_{i\in S} e^{\beta \g_\lambda(x_i-y_i)}\dd x_i \dd y_i.
	\end{equation*}
	Notice that
	\begin{equation*}
		\Msf_{\ve_0}^{+,Z}(S)=\frac{1}{(NC_{\beta,\lambda,\ve_0})^{|S|} }\textcolor{black}{I_{\mathrm{mult},\ve_0}^{+,Z}}(S)
	\end{equation*}
	and 
	\begin{equation*}
		\Msf^0_{\infty} (S)=\frac{1}{(NC_{\beta,\lambda})^{|S|} }\textcolor{black}{I_{\mathrm{mult},\infty}^{0}}(S).
	\end{equation*}
	\textcolor{black}{By Lemma~\ref{lemma:general exp M},}
	\begin{equation*}
		|\textcolor{black}{I_{\mathrm{mult},\ve_0}^{+,Z}}(S)-\textcolor{black}{I_{\mathrm{mult},\infty}^{0}}(S)|\leq \textcolor{black}{CN\lambda^{(2-\beta)|S|}\theta_{|S|}}.
	\end{equation*}
	\textcolor{black}{Proceeding as in \eqref{eq:M-M0}, we obtain}
	\begin{equation*}
		|\Msf_{\ve_0}^{+,Z}(S)-\Msf^0_{\infty} (S)|\leq \textcolor{black}{CN^{1-|S|}\theta_{|S|}}.
	\end{equation*}
\begingroup\color{black}
Similarly, by \eqref{eq:CC2}, \eqref{eq:comparisons Cve0} and the identity \eqref{eq:M-M0}, we have
\[
|\Msf_{\ve_0}^0(S)-\Msf_\infty^0(S)|\leq C_{\ve_0}N^{1-|S|}\Cut^{-2}.
\]
For $|S|=1$, both terms equal $1$. The triangle inequality proves \eqref{eq:M+diff}.
\endgroup

	{\color{black}
    By \eqref{eq:external positivity}, \eqref{eq:interaction gb} and $\tau_i^Z\leq d_{i,\bad}$ on the support,
    \begin{equation*}
        0\leq g_{\ve_0}^Z(S)-1
        \leq C\sum_{i\in S}\dE_{\Psf_{\{S\}}^{+,\ve_0,Z}}\left[\frac{r_i^2}{(\tau_i^Z)^2}\right].
    \end{equation*}
    We multiply by $\Msf_{\ve_0}^{+,Z}(S)$ and integrate as in the proof of \eqref{eq:qq3}, which gives
    \begin{equation*}
        0\leq\Msf_{\ve_0}^{+,Z}(S)(g_{\ve_0}^Z(S)-1)\leq\frac{C\theta_{|S|}}{N^{|S|-1}}.
    \end{equation*}
    Thus \eqref{eq:M+diff} also holds for $\widehat{\Msf}_{\ve_0}^{+,Z}$. Since the constant in \eqref{eq:int M lower} does not depend on $\ve_0$, this lower bound and the case $|S|=1$ above give \eqref{eq:MX +low} with a constant independent of $\ve_0$ for $\lambda$ small enough. Hence,
    \begin{equation}\label{eq:external factor bound}
        0\leq g_{\ve_0}^Z(S)-1\leq C\theta_{|S|}\lambda^{-2(|S|-1)}.
    \end{equation}
    }

    We now turn to the proof of \eqref{eq:K+diff}. \textcolor{black}{We may assume that $|X|\geq2$.} Let $k\coloneqq |V_X|$. Recall  $\K_{\ve_0}^{+,Z}$ from Definition~\ref{def:activity upper}.
	

	Let $V^\infty\subset V_X$ be such that $V^\infty\neq \emptyset$. \textcolor{black}{We claim that} 
	\begin{multline}\label{eq:non empty}
		\max_K\max_{\substack{ Y_1,\ldots,Y_K\subset X\\ \mathrm{disjoint} }}\prod_{l=1}^K \max_{E_l\in \mathsf{E}^{Y_l}}\max_{F\in \textcolor{black}{\mathsf{E}_0}^{\Coarse_X(Y_1,\ldots,Y_K)} } \\ \left|  \dE_{\Psf_{X}^{+,\ve_0,Z}}\left[\prod_{ij\in E_1\cup \cdots \cup E_K}f^{v^Z}_{ij}\prod_{ij\in \cup_l\mc{E}^\inter(Y_l)}\indic_{\mc{B}_{ij}^c}\prod_{ij\in F}(-\indic_{\mc{B}_{ij}})\prod_{i\in V^\infty }f_{i\infty}^{v^Z}\right]\right|\leq \textcolor{black}{\frac{C\theta_k}{\Msf_{\ve_0}^{+,Z}(X)N^{k-1}}}.
	\end{multline}
	
    \textcolor{black}{We expand the Mayer bonds as in Corollary~\ref{coro:reduction Euler u} and discard the terms that violate its parity condition. If there is $i\in V_2^\infty$, then \eqref{bornesabinfty} bounds $|b_{i\infty}^{v^Z}|$ by $Cr_i^2/d_{i,\bad}^2$. Otherwise, if an odd component meets $V_1^\infty$, we keep $r_{\hat S}^2/d_{i_0,\bad}^2$ from \eqref{eq:squares}, where $\hat S$ has maximal radius in that component and $r_{i_0}=r_{\hat S}$. In the remaining case, a multipole that meets $V_1^\infty$ has no odd edge, so parity gives two elements $i\neq j$ of $V_1^\infty$ in it. The product of their weights is at most $Cr_ir_j/(d_{i,\bad}d_{j,\bad})$, and removing them preserves parity. In each case, the factor we keep is bounded by
\[
C\min\left(\frac{\max_{i\in V_X}r_i^2}{\min_{i\in V_X}d_{i,\bad}^2},1\right).
\]
We keep all the separation indicators until we apply \eqref{eq:squares} or \eqref{mainducorollaire} to each odd component that meets $V_1^\infty$, and Corollary~\ref{coro:prod a} to the others. The other weights $a_{i\infty}^{v^Z}$ and $b_{i\infty}^{v^Z}$ are bounded by constants. Thus,}
	\begin{multline*}
		\max_K\max_{\substack{ Y_1,\ldots,Y_K\subset X\\ \mathrm{disjoint} }}\prod_{l=1}^K \max_{E_l\in \mathsf{E}^{Y_l}}\max_{F\in \textcolor{black}{\mathsf{E}_0}^{\Coarse_X(Y_1,\ldots,Y_K) }} \\ \left|  \dE_{\Psf_{X}^{+,\ve_0,Z}}\left[\prod_{ij\in E_1\cup \cdots \cup E_K}f^{v^Z}_{ij}\prod_{ij\in \cup_l\mc{E}^\inter(Y_l)}\indic_{\mc{B}_{ij}^c}\prod_{ij\in F}(-\indic_{\mc{B}_{ij}})\prod_{i\in V^\infty }f_{i\infty}^{v^Z}\right]\right|\\
		\leq  \frac{C}{\Msf_{\ve_0}^{+,Z}(X)(NC_{\beta,\lambda,\ve_0})^k }\max_n\max_{\substack{ X_1,\ldots,X_n\subset X\\ \mathrm{disjoint} }}\max_{T_1^a\in \mathsf{T}^{X_1},\ldots,T_n^a\in \mathsf{T}^{X_n} } \max_{T^b\in \textcolor{black}{\mathsf{T}_0}^{\Coarse_X(X_1,\ldots,X_n)}}\\ \int_{(\Lambda^{2})^k}\min\Bigg( \frac{\max_{i\in V_X}r_i^2}{\min_{i\in V_X}d_{i,\bad}^2} , 1\Bigg)\left(\prod_{l=1}^n \prod_{S\in X_l:S\neq \hat{S}_l }r_S^2\right)\prod_{ij\in T^b}b_{ij}^\abs\prod_{ij\in \textcolor{black}{\cup_l T_l^a}}\textcolor{black}{\frac{1}{d_{ij}^2}\indic_{\mc B_{ij}^c}\indic_{d_{ij}\leq\textcolor{black}{17}\ve_0\Cut}} \\ \textcolor{black}{\times\prod_{S\in X}\indic_{\mc B_S}}\prod_{i\in V_X}e^{\beta \g_\lambda(r_i)}\indic_{r_i\le \ve_0\Cut}\textcolor{black}{\dd x_i\dd y_i}.
	\end{multline*}
	Thus, using the estimate \eqref{eq:bound J2} from Lemma~\ref{lemma:integral small -} to bound the above integral, we deduce that \eqref{eq:non empty} holds.

	In view of \eqref{eq:non empty}, we can reduce to the situation where $V^\infty=\emptyset$ up to small errors, and obtain  
	\begin{multline}\label{eq:rr4}
		\Ksf_{\ve_0}^{+,Z}(X)-\Ksf_{\infty}^0(X)= \sum_{n=0}^\infty \frac{1}{n!}\sum_{\substack{X_1,\ldots,X_n\subset X\\ \mathrm{disjoint} }} \sum_{E_1\in \mathsf{E}^{X_1}}\cdots\sum_{E_n\in \mathsf{E}^{X_n}} \sum_{F\in \textcolor{black}{\mathsf{E}_0}^{\Coarse_X(X_1,\ldots,X_n)} }\\ \left(\dE_{\Psf_X^{+,\ve_0,Z}}\left[\prod_{ij\in E}f_{ij}^{v^Z}\prod_{ij\in \cup_l\mc{E}^\inter(X_l)}\indic_{\mc{B}_{ij}^c}\prod_{ij\in F}(-\indic_{\mc{B}_{ij}}) \right]-\dE_{\Psf_X^{0,\infty}}\left[\prod_{ij\in E}f_{ij}^{v} \prod_{ij\in \cup_l\mc{E}^\inter(X_l)}\indic_{\mc{B}_{ij}^c}\prod_{ij\in F}(-\indic_{\mc{B}_{ij}})\right] \right)\\ +\textcolor{black}{O_{\beta,M,k}\left(\frac{\theta_k}{\Msf_{\ve_0}^{+,Z}(X)N^{k-1}}\right)}.
	\end{multline}
	
	Fix disjoint $X_1,\ldots,X_n\subset X$, $E_1\in \mathsf{E}^{X_1},\ldots,E_n\in \mathsf{E}^{X_n}$ and $F\in \textcolor{black}{\mathsf{E}_0}^{\Coarse_X(X_1, \dots, X_n)}$. We set 
	\begin{multline*}
		\textcolor{black}{I_{\mathrm{cl},\ve_0}^{+,Z}(X;\mathbf X,\mathbf E,F)}= \int_{(\Lambda^2)^{k} }\prod_{i\in V_X}\indic_{i\in I_\good^Z} \prod_{ij\in \cup_l E_l}f_{ij}^{v^Z}\prod_{ij\in F}(-\indic_{\mc{B}_{ij}})\prod_{ij\in \cup_{l=1}^n \mc{E}^\inter(X_l)}\indic_{\mc{B}_{ij}^c} \\ \times \left(\prod_{S\in X}\indic_{\mc{B}_S} \textcolor{black}{\prod_{\substack{i,j\in S\\i<j}}}e^{-\beta v^Z_{ij}}\textcolor{black}{\indic_{\mc A_{ij}}}\right)\prod_{i\in V_X}e^{L\frac{|x_i-y_i|^2}{(\tau_i^Z)^2}}\prod_{i\in V_X}e^{\beta \g_\lambda(x_i-y_i)}\indic_{|x_i-y_i|\leq \ve_0\Cut}\dd x_i \dd y_i
	\end{multline*}
	and 
	\begin{multline*}
		\textcolor{black}{I_{\mathrm{cl},\infty}^{0}(X;\mathbf X,\mathbf E,F)} = \int_{(\Lambda^2)^{k} } \prod_{ij\in \cup_l E_l}f_{ij}^{v}\prod_{ij\in F}(-\indic_{\mc{B}_{ij}})\prod_{ij\in \cup_{l=1}^n \mc{E}^\inter(X_l)}\indic_{\mc{B}_{ij}^c}\\
		\times \left(\prod_{S\in X}\indic_{\mc{B}_S} \textcolor{black}{\prod_{\substack{i,j\in S\\i<j}}}e^{-\beta v_{ij}}\textcolor{black}{\indic_{\mc A_{ij}}}\right)\prod_{i\in V_X}e^{\beta \g_\lambda(x_i-y_i)}\dd x_i \dd y_i.
	\end{multline*}
\textcolor{black}{We suppress the fixed arguments $(X;\mathbf X,\mathbf E,F)$.}
	We have
	\begin{equation*}
		\dE_{\Psf_X^{+,\ve_0,Z}}\left[\prod_{ij\in E}f_{ij}^{v^Z}\prod_{ij\in \cup_l\mc{E}^\inter(X_l)}\indic_{\mc{B}_{ij}^c}\prod_{ij\in F}(-\indic_{\mc{B}_{ij}}) \right]  =\frac{1}{\Msf_{\ve_0}^{+,Z}(X)(NC_{\beta,\lambda,\ve_0})^{k}}\textcolor{black}{I_{\mathrm{cl},\ve_0}^{+,Z}}
	\end{equation*}
	and
	\begin{equation*}
		\dE_{\Psf_X^{0,\infty}}\left[\prod_{ij\in E}f_{ij}^{v}\prod_{ij\in \cup_l\mc{E}^\inter(X_l)}\indic_{\mc{B}_{ij}^c}\prod_{ij\in F}(-\indic_{\mc{B}_{ij}}) \right]  =\frac{1}{\Msf_{\infty}^{0}(X)(NC_{\beta,\lambda})^{k}}\textcolor{black}{I_{\mathrm{cl},\infty}^{0}}.
	\end{equation*}
	Hence,
	\begin{multline*}
		\left| \dE_{\Psf_X^{+,\ve_0,Z}}\left[\prod_{ij\in E}f_{ij}^{v^Z}\prod_{ij\in \cup_l\mc{E}^\inter(X_l)}\indic_{\mc{B}_{ij}^c}\prod_{ij\in F}(-\indic_{\mc{B}_{ij}}) \right]- \dE_{\Psf_X^{0,\infty}}\left[\prod_{ij\in E}f_{ij}^{v}\prod_{ij\in \cup_l\mc{E}^\inter(X_l)}\indic_{\mc{B}_{ij}^c}\prod_{ij\in F}(-\indic_{\mc{B}_{ij}}) \right]\right|\\ \leq
		\frac{1}{(NC_{\beta,\lambda,\ve_0})^k \Msf_{\ve_0}^{+,Z}(X)} \Bigr| \textcolor{black}{I_{\mathrm{cl},\ve_0}^{+,Z}}- \textcolor{black}{I_{\mathrm{cl},\infty}^{0}}\Bigr|\\+ {|\textcolor{black}{I_{\mathrm{cl},\infty}^{0}}|}\times \left|\frac{1}{(NC_{\beta,\lambda,\ve_0})^k\Msf_{\ve_0}^{+,Z}(X)}-\frac{1}{(NC_{\beta,\lambda})^k\Msf^0_{\infty} (X)}\right|.
	\end{multline*}
	\textcolor{black}{By Lemma~\ref{lemma:general exp K},}
	\begin{equation*}
		\Bigr|\textcolor{black}{I_{\mathrm{cl},\ve_0}^{+,Z}}-\textcolor{black}{I_{\mathrm{cl},\infty}^{0}}\Bigr|\leq \textcolor{black}{CN\theta_k\lambda^{(2-\beta)k}}.
	\end{equation*}
	\textcolor{black}{Proceeding as in \eqref{eq:rr2 -}, and using $\lambda^{2(k-|S|)}\theta_{|S|}\leq\theta_k$ for the normalizing factors, we obtain}
	\begin{multline*}
		\left|\dE_{\Psf_X^{+,\ve_0,Z}}\left[\prod_{ij\in E}f_{ij}^{v^Z}\prod_{ij\in \cup_l\mc{E}^\inter(X_l)}\indic_{\mc{B}_{ij}^c}\prod_{ij\in F}(-\indic_{\mc{B}_{ij}}) \right]-\dE_{\Psf_X^{0,\infty}}\left[\prod_{ij\in E}f_{ij}^{v} \prod_{ij\in \cup_l\mc{E}^\inter(X_l)}\indic_{\mc{B}_{ij}^c}\prod_{ij\in F}(-\indic_{\mc{B}_{ij}})\right] \right|\\ \leq \textcolor{black}{\frac{CN^{1-k}\theta_k}{\Msf^0_{\infty}(X)}}.
	\end{multline*}
	\textcolor{black}{By \eqref{eq:rr4}, we obtain the same bound for $|\Ksf_{\ve_0}^{+,Z}(X)-\Ksf_\infty^0(X)|$. By \eqref{eq:MX -low} and \eqref{eq:MX -up}, and since $\Msf_{\ve_0}^0(S)=\Msf_\infty^0(S)=1$ when $|S|=1$, we have
\[
C^{-1}\Msf_\infty^0(X)\leq\Msf_{\ve_0}^0(X)\leq C\Msf_\infty^0(X),
\]
with $C$ independent of $\ve_0$. Together with \eqref{eq:K-diff0}, this proves \eqref{eq:K+diff}.}

    {\color{black}
    Finally, by \eqref{eq:external positivity},
    \begin{equation*}
        |\widehat{\Ksf}_{\ve_0}^{+,Z}(X)-\Ksf_\infty^0(X)|
        \leq|\Ksf_{\ve_0}^{+,Z}(X)-\Ksf_\infty^0(X)|
        +|\Ksf_\infty^0(X)|\sum_{S\in X}(g_{\ve_0}^Z(S)-1).
    \end{equation*}
    By Proposition~\ref{prop:bounded lower} for $\Ksf_\infty^0$ and by \eqref{eq:external factor bound}, the last term is at most
    \begin{equation*}
        \frac{C}{\Msf_\infty^0(X)N^{k-1}}\sum_{S\in X}\lambda^{2(k-|S|)}\theta_{|S|}
        \leq\frac{C\theta_k}{\Msf_\infty^0(X)N^{k-1}},
    \end{equation*}
    with $C$ independent of $\ve_0$. With the comparison of $\Msf_{\ve_0}^0$ and $\Msf_\infty^0$ above and \eqref{eq:K-diff0}, this proves \eqref{eq:K+diff} for $\widehat{\Ksf}_{\ve_0}^{+,Z}$.
    }
\end{proof}

\subsection{Absolute convergence of the cluster series with frozen bad points}

We turn to the proof of Proposition~\ref{prop:absolute upper}. We first state the analogue of Lemma~\ref{lemma:LXi}.

\begin{lemma}\label{lemma:split u}
	Let $V_\bad\subset [N]$, $V_\good=[N]\setminus V_\bad$ and $X$ be a subpartition of $V_\good$ such that, for every $S\in X$, one has $|S|\leq p(\beta)$. Let $a_{ij}^{\abs}$ and $b_{ij}^{\abs}$ be as in Definition~\ref{def:abs ab}. 
	\textcolor{black}{Then, there exist $C>0$, depending on $\beta$, $p(\beta)$ and $M$, and $C_1>0$, depending only on $\beta$, such that}
	\begin{multline}\label{eq:Eul res}
		|\Ksf_{\ve_0}^{+,Z}(X)|\leq\textcolor{black}{C^{|V_X|}}\sum_{V_1^\infty\subset V_X}\sum_{n=0}^\infty \frac{1}{n!}\sum_{\substack{X_1,\ldots,X_n\subset X\\ \mathrm{disjoint} }} \sum_{\substack{T^{b}\in \textcolor{black}{\mathsf{T}_0}^{\Coarse_X(X_1,\ldots,X_n)}}} \\  \dE_{\Psf_X^{+,\ve_0,Z}}\left[\mc{L}_{C_1}^+(X_1,V_1^\infty)\cdots \mc{L}_{C_1}^+(X_n,V_1^\infty)  \prod_{ij\in T^{b}} b_{ij}^{\abs}\right],
	\end{multline}
	where \textcolor{black}{$\mc L_{C_1}^+(X',V_1^\infty)=0$ for $|X'|\leq1$, and, for $|X'|\geq2$,}
	\begin{multline}\label{def:mcL+}
		\mc{L}_{C_1}^+(X',V_1^\infty)\coloneqq \sum_{k=0}^{\infty}\frac{1}{k!}\sum_{\substack{X_1',\ldots,X_k'\subset X'\\ \mathrm{disjoint} } }\sum_{E_1\in \Eulc^{X_1'}(V_{X_1'}\cap V_1^\infty) }\ldots \sum_{E_k\in \Eulc^{X_k'}(V_{X_k'}\cap V_1^\infty ) } \\ \sum_{\tilde{T}^{b}\in \textcolor{black}{\mathsf{T}_0}^{\Coarse_{X'}(X_1',\ldots,X_k') } }\prod_{i\in V_{X'}\cap V_1^\infty}|a_{i\infty}^{v^Z}|\textcolor{black}{\prod_{ij\in E_1\cup\cdots\cup E_k}|a_{ij}^{v^Z}|}\prod_{ij\in \tilde{T}^{b}}b_{ij}^{\abs}\prod_{ij\in \mc{E}^\inter(X')} e^{C_1a_{ij}^{\abs}} \prod_{ij\in \mc{E}^\inter(X')}\indic_{\mc{B}_{ij}^c}.
	\end{multline}
\end{lemma}

\textcolor{black}{We start from Corollary~\ref{coro:reduction Euler u}. The sum over $V_2^\infty$ of the weights $b_{i\infty}^{v^Z}$ costs
\[
\prod_{i\in V_X\setminus V_1^\infty}(1+|b_{i\infty}^{v^Z}|)\leq C^{|V_X|}
\]
by \eqref{bornesabinfty} and \eqref{eq:geo1}. We keep the odd bonds $a_{ij}^{v^Z}$ and apply the tree extraction of Lemma~\ref{lemma:LXi}. We bound the even tree weights by $Cb_{ij}^{\abs}$ and the remaining even factors by $e^{Ca_{ij}^{\abs}}$.}

We now sum over Eulerian graphs as in Lemma~\ref{lemma:sum Euler} to give a bound on the quantity $\mc{L}_{C_1}^+(X,V_1^\infty)$ defined in \eqref{def:mcL+}.

\textcolor{black}{We take $C_0$ equal to the constant $C_1$ in Lemma~\ref{lemma:split u}.}

\begin{lemma}\label{lemma:sum Euler u}
	Let $V_\bad\subset [N]$, $V_\good=[N]\setminus V_\bad$ and $X$ be a subpartition of $V_\good$ such that, for every $S\in X$, one has $|S|\leq p(\beta)$. \textcolor{black}{Let $(x_i,y_i)_{i\in V_X}$ belong to the support of $\Psf_X^{+,\ve_0,Z}$, with $0<r_i\leq\ve_0\Cut$ for every $i\in V_X$.} Let $V_1^\infty\subset V_X$. Let $k\coloneqq |V_X|$. There exists a constant $C>0$ depending on $\beta$, $p(\beta)$ and $M$ such that
	\begin{equation*}
		\mc{L}_{C_0}^+(X,V_1^\infty) \le \textcolor{black}{\left(\prod_{ij\in\mc E^\inter(X)}\indic_{\mc B_{ij}^c}\right)}e^{C|V_{X}|}\sum_{n=0}^{\infty}\frac{1}{n!}\sum_{\substack{X_1,\ldots,X_n\subset X\\ \mathrm{disjoint} } } H_{X}(X_1,\ldots,X_n),
	\end{equation*}
	where $H_X(X_1,\ldots,X_n)$ is as in \eqref{def:H}, \textcolor{black}{with its exponential coefficient chosen large enough depending only on $\beta$}.
\end{lemma}

\begin{proof}
    {\color{black}
    We work on $\bigcap_{ij\in\mc E^\inter(X)}\mc B_{ij}^c$, since otherwise the left-hand side is zero. On the support of $\Psf_X^{+,\ve_0,Z}$, this gives $\Pi_\mult^{V_X}=X$ and the geometric estimates of good points. For each connected odd graph $E_0$ in \eqref{def:mcL+}, we apply Corollary~\ref{coro:prod a u} to its minimal skeleton when $V_1^\infty$ meets its vertex set, and Corollary~\ref{coro:prod a} otherwise. We bound the odd bonds removed by peeling by $C_\beta a_{ij}^{\abs}$ and sum them into $\prod_{ij}e^{C_\beta a_{ij}^{\abs}}$, which increases the exponential coefficient in $H$. The other weights $|a_{i\infty}^{v^Z}|$ cost at most $C^{|V_X|}$. We then argue as in Lemma~\ref{lemma:sum Euler}.
    }
	\begingroup\color{black}
	For a fixed tree $T$, we record the bridges of the minimal skeleton, which lie in $T$. Cutting these edges in $(V_X,T)/X$ gives components $C\subset X$. Together with the choice, in each component, of its edge $e_0$ or of no edge, this record has at most
	\[
		2^{|T|}\prod_C\left(1+\binom{|V_C|}{2}\right)\leq4^k
	\]
	possibilities. Deleting these edges $e_0$ from $F$ leaves a pseudoforest with weights $\bar g_{ij}$, and the records determine $F$. This factor is absorbed into $e^{Ck}$, which yields the result.
	\endgroup
\end{proof}

\textcolor{black}{For each group $X_l$ in \eqref{eq:Eul res}, we keep its separation indicator until we apply Lemma~\ref{lemma:control sum intera}, and then we discard it. As in the proof of Lemma~\ref{lemma:simplification}, the outer edges of the merged tree $T^b$ determine the former groups, at a cost $2^{|T^b|}\leq2^{|X|-1}$ that we absorb into $e^{C|V_X|}$.}

Combining Lemmas~\ref{lemma:split u}, \ref{lemma:sum Euler u} and \ref{lemma:control sum intera} with the multipole partition function estimates of Propositions~\ref{prop:bounded lower} and~\ref{prop:bounded upper}, we obtain the following analogue of Lemma~\ref{lemma:simplification}.

\begin{lemma}\label{lemma:simplified u}
	Let $V_\bad \subset [N]$, $V_\good\coloneqq [N]\setminus V_\bad$ and $Z\in (\Lambda^2)^{|V_\bad|}$. Let $X$ be a subpartition of $V_\good$ such that, for every $S\in X$, one has $|S|\leq p(\beta)$. Then, there exist $C>0$, depending only on $\beta$ and $M$, and $C_0>0$, depending only on $\beta$, such that 
	\begin{multline}
		|\Ksf_{\ve_0}^{+,Z}(X)|\leq \frac{e^{C|V_X|}}{(N\lambda^{(2-\beta)})^{|V_X|}\Msf_{\ve_0}^0(X)}\\ \times\sum_{n=0}^\infty \frac{1}{n!}\sum_{\substack{X_1,\ldots,X_n\subset X\\ \mathrm{disjoint} }} \prod_{l=1}^n\sum_{T_l^a\in \mathsf{T}^{X_l}}\sum_{T^b \in \textcolor{black}{\mathsf{T}_0}^{\Coarse_X(X_1,\ldots,X_n)}} \textcolor{black}{J_{\mathrm{large},C_0}}(\cup_{l=1}^n T_l^a,T^b),
	\end{multline}
	where $\textcolor{black}{J_{\mathrm{large},C_0}}(\cup_{l=1}^n T_l^a,T^b)$ is as in \eqref{def:IT'} for every $T_1^a\in \mathsf{T}^{X_1},\ldots,T_n^a\in \mathsf{T}^{X_n}$ and $T^b\in \textcolor{black}{\mathsf{T}_0}^{\Coarse_X(X_1,\ldots,X_n)}$.
\end{lemma}

We finally prove the absolute convergence result of Proposition~\ref{prop:absolute upper}.

\begin{proof}[Proof of Proposition~\ref{prop:absolute upper}]
\begingroup\color{black}
The activities $\Ksf_{\ve_0}^{+,Z}$ vanish for $|X|\leq1$. Lemma~\ref{lemma:simplified u} gives the bound of Lemma~\ref{lemma:simplification}, with the same denominators $\Msf_{\ve_0}^0$, and Proposition~\ref{prop:bounded upper} gives \eqref{eq:K+ bound}. As observed after the proof of Proposition~\ref{prop:absolute lower}, the proofs of Lemma~\ref{lem:cvgentseries} and Proposition~\ref{prop:absolute lower} then give \eqref{eq:rest u} and \eqref{eq:rest b u}, uniformly in $Z$. By \eqref{eq:external positivity}, the same bounds apply to $\widehat{\Ksf}_{\ve_0}^{+,Z}$.

For $|X|\geq2$ and $|V_X|\leq p=p(\beta)$, Propositions~\ref{prop:bounded lower} and~\ref{prop:bounded upper} give
\[
|\Ksf_{\ve_0}^{+,Z}(X)|\leq C_{\ve_0}(\lambda^2/N)^{|X|-1}.
\]
By Lemma~\ref{lemma:bounded cluster tail}, the subseries that contains only these activities converges absolutely, and \eqref{eq:rest b u} controls the other terms. The same holds for the activities $\widehat{\Ksf}_{\ve_0}^{+,Z}$. Lemma~\ref{lemma:resum2} then gives $W(Z)=\exp(S(Z))>0$, where $S(Z)$ is their full connected series.
\endgroup
\end{proof}

\subsection{Reduction to a system of noninteracting bad points}\label{sub:hier u}
Recall  Lemma~\ref{lemma:start upper}.
In Proposition~\ref{prop:absolute upper}, \textcolor{black}{we have proved that the full connected series converges absolutely}, which implies by Lemma~\ref{lemma:resum2} that, for every configuration $Z\in (\Lambda^2)^{|V_\bad|}$ of bad points, the quantity $W(Z)$ is positive and
\begin{equation*}
	\log W(Z)=\sum_{n=1}^\infty\frac{1}{n!}\sum_{\substack{X_1,\ldots,X_n\in \mc{P}(\pi):\\ \mathrm{connected}} }\textcolor{black}{\widehat{\Ksf}_{\ve_0}^{+,Z}}(X_1)\cdots \textcolor{black}{\widehat{\Ksf}_{\ve_0}^{+,Z}}(X_n) \mathrm{I}(G(X_1,\ldots,X_n)).   
\end{equation*}
\textcolor{black}{By \eqref{eq:rest b u}, we retain only activities with $|V_{X_i}|\leq p(\beta)$. Proposition~\ref{prop:expansion +} and the summation in the proof of Proposition~\ref{prop:upper bound} then give}   
\begin{multline}\label{eq:noZ}
	\log W(Z)=  \sum_{n=1}^\infty\frac{1}{n!}\sum_{\substack{X_1,\ldots,X_n\in \mc{P}(\pi):\\ \mathrm{connected}} }\Ksf_{\ve_0}^0(X_1)\cdots \Ksf_{\ve_0}^0(X_n) \mathrm{I}(G(X_1,\ldots,X_n))\textcolor{black}{\prod_{i=1}^n\indic_{|V_{X_i}|\leq p(\beta)}}\\+\textcolor{black}{O_{\beta,M,p,\ve_0}(N\delta_{\beta,\lambda})+O_{\beta,M,p}(|V_\bad|)}.
\end{multline}
Note that the right-hand side of \eqref{eq:noZ} depends on the configuration $Z$ of bad points only through the cardinality of $V_\bad$. Moreover, by Lemma~\ref{lemma:badbad}, there exists a constant $C>0$ depending on $\beta$ and $p(\beta)$ such that 
\begin{equation}\label{eq:580}
	\prod_{i,j\in V_\bad:i<j}e^{-\beta v_{ij}^Z}\indic_{\mc{A}_{ij}}\leq e^{C|V_\bad|}.
\end{equation}
It therefore remains to control the partition function of the bad points, defined by
\begin{equation}\label{def:Z bad}
	\Zsf_{\bad,N-N'}\coloneqq \int_{(\Lambda^2)^{|V_\bad|} } \indic_{\{ I_\bad=V_\bad\}}\prod_{i,j\in V_\bad:i<j}\indic_{\mc{A}_{ij}}\prod_{i\in V_\bad}\frac{1}{(\tau_i^{+,Z})^{\frac{\beta}{2}}} \frac{1}{(\tau_i^{-,Z})^{\frac{\beta}{2}}}\prod_{i\in V_\bad}\dd x_i \dd y_i.
\end{equation}

The main result is the following proposition.
\begin{prop}\label{prop:bad points}
	Let $\beta\in (2,\infty)$. Let $\Zsf_{\bad,N-N'}$ be as in \eqref{def:Z bad}. Let $t\geq 1$. 
	
	Then, if $\lambda$ and $\ve_0$ are small enough with respect to $\beta$, $M$, $p(\beta)$ \textcolor{black}{and $t$}, there exists a constant $C>0$ depending on $\beta$, $M$, $\ve_0$, $p(\beta)$  and $t$ such that
	\begin{equation*}
		\binom{N}{N-N'}\frac{\Zsf_{\bad,N-N'}}{(NC_{\beta,\lambda,\ve_0})^{N-N'}}e^{t(N-N')}\leq e^{CN\delta_{\beta,\lambda} }.
	\end{equation*}   
\end{prop}

As a first step, 
for every $N_1,N_2,N_3\geq 0$ such that $N_1+N_2+N_3=N-N'$, define 
\begin{multline}\label{def:Z bad second}
	\Zsf_{\bad,N_1,N_2,N_3}\coloneqq \int_{(\Lambda^2)^{|V_\bad|} }\prod_{i=1}^3 \indic_{\{ I_\bad^i=V_\bad^i\}}\prod_{i,j\in V_\bad:i<j}\indic_{\mc{A}_{ij}}  \prod_{i\in V_\bad}\frac{1}{(\tau_i^{+,Z})^{\frac{\beta}{2}}} \frac{1}{(\tau_i^{-,Z})^{\frac{\beta}{2}}}\prod_{i\in V_\bad}\dd x_i \dd y_i,
\end{multline}
where $V_\bad^1, V_\bad^2$, $V_\bad^3\subset [N]$ are arbitrary disjoint sets of cardinalities $N_1$, $N_2$ and $N_3$, respectively.

We have
\begin{equation}\label{eq:non interacting} \Zsf_{\bad,N-N'} \leq \textcolor{black}{3^{N-N'}}\max_{\substack{N_1,N_2,N_3:\\N_1+N_2+N_3=N-N'}}\frac{(N-N')!}{N_1!N_2!N_3!} \Zsf_{\bad,N_1,N_2,N_3},
\end{equation}
which reduces the problem to bounding $\Zsf_{\bad, N_1, N_2, N_3}$.

\subsection{Reduction to bad points of types 1 and 2}\label{sub:type 12}

We begin by controlling bad points of type 3 by bad points of types 1 and 2.

\begin{lemma}\label{lemma:type3}
	Let $\Zsf_{\bad,N_1,N_2,N_3}$ be as in \eqref{def:Z bad second}. 
	There exists a constant $C>0$ depending on $\beta$, $M$ and $p(\beta)$ such that, for all $t\geq 1$,
	\begin{equation*}
		\max_{N_3}\left(\Bigl(\frac{N}{N_3}\Bigr)^{N_3} \frac{1}{(NC_{\beta,\lambda,\ve_0})^{N_3}}  \Zsf_{\bad,N_1,N_2,N_3}e^{tN_3} \right)\leq e^{e^{Ct}(N_1+N_2)} \Zsf_{\bad,N_1}^{(1)}\Zsf_{\bad,N_2}^{(2)},
	\end{equation*}
	where 
	\begin{equation}\label{def:ZZbad1}
		\Zsf_{\bad,N_1}^{(1)}\coloneqq \int_{(\Lambda^2)^{N_1} } \indic_{I_\bad^1=V_\bad^1}\prod_{i,j\in V_\bad^1 :i<j}\indic_{\mc{A}_{ij}}\left(\prod_{i\in V_\bad^1}\frac{1}{(\tau_i^{+,Z})^{\frac{\beta}{2}}(\tau_i^{-,Z})^{\frac{\beta}{2}} }  \right)\prod_{i\in V_\bad^1}\dd x_i \dd y_i
	\end{equation}
	and
	\begin{equation}\label{def:ZZbad2}
		\Zsf_{\bad,N_2}^{(2)}\coloneqq \int_{(\Lambda^2)^{N_2}} \indic_{I_\bad^2=V_\bad^2 }\prod_{i,j\in V_\bad^2 :i<j}\indic_{\mc{A}_{ij}}\prod_{i\in V_\bad^2}e^{\beta \g_\lambda(x_i-y_i)}\dd x_i \dd y_i.
	\end{equation}
\end{lemma}

\medskip

\begin{proof}
	\
	\paragraph{\bf{Step 1: fixing $J_\bad^3$}}
	Recall  \eqref{def:Z bad second}.
    {\color{black}
    Every block of $\Pi_\mult^U$ contained in $V_\bad^3$ meets $J_\bad^3$ and has at most $p(\beta)$ elements. Hence
    \begin{equation*}
        N_3\leq p(\beta)|J_\bad^3|.
    \end{equation*}
    }
	Therefore, 
	\begin{equation*}
		\Zsf_{\bad,N_1,N_2,N_3}\leq \sum_{\tilde{N}_3=\textcolor{black}{\lceil\frac{1}{p(\beta)}N_3\rceil}}^{N_3}\binom{N_3}{\tilde{N}_3}\max_{\tilde{V}_\bad^3\subset V_\bad^3:|\tilde{V}_\bad^3|=\tilde{N}_3 } I(N_1,N_2,N_3,\tilde{N}_3),
	\end{equation*}
	where 
	\begin{multline*}
		I(N_1,N_2,N_3,\tilde{N}_3)\coloneqq \int_{(\Lambda^2)^{|V_\bad|} } \indic_{ I_\bad^1=V_\bad^1, I_\bad^2=V_\bad^2 }\indic_{J_\bad^3 =\tilde{V}_\bad^3}\indic_{I_\bad^3=V_\bad^3} \prod_{i,j\in V_\bad:i<j}\indic_{\mc{A}_{ij}}\\ \times \prod_{i\in V_\bad}\frac{1}{(\tau_i^{+,Z})^{\frac{\beta}{2}}} \frac{1}{(\tau_i^{-,Z})^{\frac{\beta}{2}}}\prod_{i\in V_\bad}\dd x_i \dd y_i.
	\end{multline*}
	\textcolor{black}{By Definition~\ref{def:continuous clustering}, almost everywhere, each $i\in V_\bad^2\cup V_\bad^3$ is removed as a singleton at $r_i/4$. Hence,} 
	\begin{equation*}
		\tau_i^{+,Z}=\tau_i^{-,Z}=\max(\tfrac{1}{4}r_i,\lambda).
	\end{equation*}
	Therefore, there exists a constant $C>0$ depending on $\beta$ such that 
	\begin{multline*}
		I(N_1,N_2,N_3,\tilde{N}_3)\leq C^{N_2+N_3}\int_{(\Lambda^2)^{|V_\bad|} } \indic_{ I_\bad^1=V_\bad^1, I_\bad^2=V_\bad^2 }\indic_{J_\bad^3 =\tilde{V}_\bad^3}\indic_{I_\bad^3=V_\bad^3} \prod_{i,j\in V_\bad:i<j}\indic_{\mc{A}_{ij}}\\ \times \prod_{i\in  V_\bad^2\cup V_\bad^3}e^{\beta \g_\lambda(x_i-y_i)} \prod_{i\in V_\bad^1}\frac{1}{(\tau_i^{+,Z})^{\frac{\beta}{2}}} \frac{1}{(\tau_i^{-,Z})^{\frac{\beta}{2}}}\prod_{i\in V_\bad}\dd x_i \dd y_i.
	\end{multline*}

	\paragraph{\bf{Step 2: integrating variables in $V_\bad^3\setminus J_\bad^3$}}
    {\color{black}
    Write $B=V_\bad^1\cup V_\bad^2$ and $m=N_3-\tilde N_3$. Every type-$3$ dipole satisfies $r_i\leq\ve_0\Cut$. Deleting all type-$3$ dipoles preserves the bad classes and radii on $B$, since all overcrowded clusters lie in $V_\bad^1$, and the blocks of $\Pi_\mult^U$ contained in $V_\bad^2$ and the nonempty sets $I_\bad^{2,2,k}$ remain. We denote these classes and radii, computed from $B$ alone, by $I_\bad^{1,B},I_\bad^{2,B}$ and $\tau_i^{\pm,B}$, with $N$ kept in the thresholds.

    Let
    \begin{multline}\label{eq:type3 retained integral}
        J(N_1,N_2,\tilde N_3)\coloneqq
        \int_{(\Lambda^2)^{N_1+N_2+\tilde N_3}}
        \indic_{\{I_\bad^{1,B}=V_\bad^1,\ I_\bad^{2,B}=V_\bad^2\}}
        \prod_{i\in\tilde V_\bad^3}
        \indic_{\{r_i\leq\ve_0\Cut,\ \min_{j\in B}d_{ij}\leq M\max(r_i,\lambda)\}}\\
        \times\prod_{i,j\in B\cup\tilde V_\bad^3:i<j}\indic_{\mc A_{ij}}
        \prod_{i\in V_\bad^2\cup\tilde V_\bad^3}e^{\beta\g_\lambda(x_i-y_i)}
        \prod_{i\in V_\bad^1}
        \frac{1}{(\tau_i^{+,B}\tau_i^{-,B})^{\beta/2}}
        \prod_{i\in B\cup\tilde V_\bad^3}\dd x_i\dd y_i.
    \end{multline}

    The graph of $\leftrightarrow$ on $V_\bad^3$ contains a forest $F$ rooted at $\tilde V_\bad^3$, in which every other vertex has one parent. There are at most $N_3^m$ such forests. At fixed displacements $\vec r_i=y_i-x_i$, the point $x_i$ of a leaf $i$ lies in four disks of radius at most $M\max(r_i,\lambda)$ around translates of its parent's endpoints. Integrating out the leaves one at a time gives
    \begin{equation*}
        \int\prod_{ij\in F}\indic_{\mc B_{ij}}
        \prod_{i\in V_\bad^3\setminus\tilde V_\bad^3}\dd x_i
        \leq (4\pi M^2)^m
        \prod_{i\in V_\bad^3\setminus\tilde V_\bad^3}\max(r_i,\lambda)^2.
    \end{equation*}
    Dropping the remaining constraints on $V_\bad^3\setminus\tilde V_\bad^3$ yields
    \begin{equation*}
        I(N_1,N_2,N_3,\tilde N_3)
        \leq C^{N_2+N_3}
        \left(N_3\int_{|\vec r|\leq\ve_0\Cut}
        \max(r,\lambda)^2e^{\beta\g_\lambda(r)}\dd\vec r\right)^m
        J(N_1,N_2,\tilde N_3).
    \end{equation*}
    Here $r=|\vec r|$.
    }

    $\bullet$ Suppose that $\beta\in (2,4)$. Then, there exists $C>0$ depending on $\beta$ and $M$ such that 
	\begin{equation*}
		I(N_1,N_2,N_3,\tilde{N}_3)
		\leq C^{N_2+N_3}(R_{\beta,\lambda}^{4-\beta}\textcolor{black}{N_3})^{N_3-\tilde{N}_3}    J(N_1,N_2,\tilde{N}_3).
	\end{equation*}
	Hence, using \eqref{eq:Clambda bound} and the fact that $R_{\beta,\lambda}^{4-\beta}=\lambda^{2-\beta}$, we obtain the existence of $C>0$ depending on $\beta$ and $M$ such that
	\begin{equation*}
		\frac{1}{(NC_{\beta,\lambda,\ve_0})^{N_3-\tilde{N}_3}} I(N_1,N_2,N_3,\tilde{N}_3)\leq C^{N_2+N_3}\Bigl(\frac{N_3}{N}\Bigr)^{N_3-\tilde{N}_3}J(N_1,N_2,\tilde{N}_3).
	\end{equation*}
	
	$\bullet$ Suppose that $\beta>4$. Then, there exists $C>0$ depending on $\beta$ and $M$ such that
	\begin{equation*}
		I(N_1,N_2,N_3,\tilde{N}_3)
		\leq C^{N_2+N_3}(\lambda^{4-\beta}\textcolor{black}{N_3})^{N_3-\tilde{N}_3}    J(N_1,N_2,\tilde{N}_3).
	\end{equation*}
	Therefore, there exists $C>0$ depending on $\beta$ and $M$ such that
	\begin{equation*}
		\frac{1}{(NC_{\beta,\lambda,\ve_0})^{N_3-\tilde{N}_3}} I(N_1,N_2,N_3,\tilde{N}_3)\leq C^{N_2+N_3}\Bigl(\frac{N_3}{N}\Bigr)^{N_3-\tilde{N}_3}\lambda^{2(N_3-\tilde{N}_3)} J(N_1,N_2,\tilde{N}_3).
	\end{equation*}
	
	$\bullet$ Suppose that $\beta=4$. Then, there exists a constant $C>0$ depending on $p(\beta)$ and $M$ such that
	\begin{equation*}
		I(N_1,N_2,N_3,\tilde{N}_3)
		\leq C^{N_2+N_3}\textcolor{black}{(N_3|\log \lambda|)^{N_3-\tilde{N}_3}}    J(N_1,N_2,\tilde{N}_3).
	\end{equation*}
	Therefore, there exists $C>0$ depending on $p(\beta)$ \textcolor{black}{and $M$} such that
	\begin{equation*}
		\frac{1}{(NC_{\beta,\lambda,\ve_0})^{N_3-\tilde{N}_3}} I(N_1,N_2,N_3,\tilde{N}_3)\leq C^{N_2+N_3}\Bigl(\frac{N_3}{N}\Bigr)^{N_3-\tilde{N}_3}(\lambda^{2}|\log\lambda|)^{N_3-\tilde{N}_3} J(N_1,N_2,\tilde{N}_3).
	\end{equation*}
	
	Therefore, for every $\beta>2$, there exists $C>0$ depending on $\beta, M$ and $p(\beta)$ such that
	\begin{equation}\label{eq:I4out}
		\frac{1}{(NC_{\beta,\lambda,\ve_0})^{N_3-\tilde{N}_3}} I(N_1,N_2,N_3,\tilde{N}_3)\leq C^{N_2+N_3}\Bigl(\frac{N_3}{N}\Bigr)^{N_3-\tilde{N}_3}J(N_1,N_2,\tilde{N}_3).
	\end{equation}

	\paragraph{\bf{Step 3: integrating variables in $J_\bad^3$}}
	
    {\color{black}
    By \eqref{eq:type3 retained integral}, every $i\in\tilde V_\bad^3$ satisfies $r_i\leq\ve_0\Cut$ and $\min_{j\in B}d_{ij}\leq M\max(r_i,\lambda)$. Hence
    \begin{equation*}
        x_i\in\bigcup_{j\in B}
        \left(B(x_j,2M\max(r_i,\lambda))\cup B(y_j,2M\max(r_i,\lambda))\right).
    \end{equation*}
    }
	Performing the change of variables $(x_i,y_i)\mapsto (x_i,\vec{r}_i)=(x_i,y_i-x_i)$ and integrating out the variables $x_i$ for $i\in \tilde{V}_\bad^3$ thus gives the existence of $C>0$ depending on $\beta$ and $M$ such that 
	\begin{multline*}
		J(N_1,N_2,\tilde{N}_3)
		\leq C^{\tilde{N}_3}\int \textcolor{black}{\indic_{I_\bad^{1,B}=V_\bad^1,\ I_\bad^{2,B}=V_\bad^2}}\prod_{i\in \tilde{V}_\bad^3}\textcolor{black}{\Bigl(\max(r_i,\lambda)^{2-\beta }(N_1+N_2)\indic_{r_i\leq\ve_0\Cut}\Bigr)}\\
        \times\prod_{i,j\in\textcolor{black}{B}:i<j}\indic_{\mc{A}_{ij}} \left(\prod_{i\in V_\bad^2 }e^{\beta \g_\lambda(x_i-y_i)}  \right)\prod_{i\in V_\bad^1}\textcolor{black}{\frac{1}{(\tau_i^{+,B}\tau_i^{-,B})^{\beta/2}}}\prod_{i\in\textcolor{black}{B}}\dd x_i \dd y_i\prod_{i\in \tilde{V}_\bad^3}\dd \vr_i.
	\end{multline*}
	
{\color{black}
The restriction to $V_\bad^1$ preserves its clustering and radii, and the type-$2$ components and count thresholds are computed on $V_\bad^2$. Dropping the indicators $\indic_{\mc A_{ij}}$ between the two sets therefore bounds the remaining integral by the product of \eqref{def:ZZbad1} and \eqref{def:ZZbad2}.
}

	$\bullet$ Suppose that $\beta\in (2,4)$. Then, 
	\begin{equation*}
		J(N_1,N_2,\tilde{N}_3)
		\leq C^{\tilde{N}_3}  \Bigl(R_{\beta,\lambda}^{4-\beta }(N_1+N_2)\Bigr)^{\tilde{N}_3} \Zsf_{\bad,N_1}^{(1)}\Zsf_{\bad,N_2}^{(2)},
	\end{equation*}
	where the partition functions $\Zsf_{\bad}$ are as in \eqref{def:ZZbad1} and \eqref{def:ZZbad2}.

	Thus, since $C_{\beta,\lambda,\ve_0}\geq \frac{1}{C}\lambda^{2-\beta}=\frac{1}{C}R_{\beta,\lambda}^{4-\beta}$, we get that
	\begin{equation*}
		\frac{1}{(NC_{\beta,\lambda,\ve_0})^{\tilde{N}_3}}  J(N_1,N_2,\tilde{N}_3) 
		\leq \Bigl(C\frac{N_1+N_2}{N}\Bigr)^{\tilde{N}_3}\Zsf_{\bad,N_1}^{(1)}\Zsf_{\bad,N_2}^{(2)}.
	\end{equation*}
	
	$\bullet$ Suppose that $\beta> 4$. Then, there exists $C>0$ depending on $\beta$ such that 
	\begin{equation*}
		\frac{1}{(NC_{\beta,\lambda,\ve_0})^{\tilde{N}_3}}  J(N_1,N_2,\tilde{N}_3) 
		\leq \Bigl(C\frac{N_1+N_2}{N}\Bigr)^{\tilde{N}_3}\lambda^{2\tilde{N}_3} \Zsf_{\bad,N_1}^{(1)}\Zsf_{\bad,N_2}^{(2)}.
	\end{equation*}
	
	$\bullet$ Suppose that $\beta=4$. Then, there exists $C>0$ depending on $p(\beta)$ such that
	\begin{equation*}
		\frac{1}{(NC_{\beta,\lambda,\ve_0})^{\tilde{N}_3}}  J(N_1,N_2,\tilde{N}_3) 
		\leq \Bigl(C\frac{N_1+N_2}{N}\Bigr)^{\tilde{N}_3}(\lambda^2|\log\lambda|)^{\tilde{N}_3} \Zsf_{\bad,N_1}^{(1)}\Zsf_{\bad,N_2}^{(2)}.
	\end{equation*}

	Thus, for every $\beta>2$, we deduce that there exists $C>0$ depending on $\beta, M$ and $p(\beta)$ such that 
	\begin{equation}\label{eq:J4out}
		\frac{1}{(NC_{\beta,\lambda,\ve_0})^{\tilde{N}_3}}  J(N_1,N_2,\tilde{N}_3) 
		\leq \Bigl(C\frac{N_1+N_2}{N}\Bigr)^{\tilde{N}_3}\Zsf_{\bad,N_1}^{(1)}\Zsf_{\bad,N_2}^{(2)}.
	\end{equation}

	\paragraph{\bf{Step 4: summing over $\tilde{N}_3$}}
{\color{black}
If $N_3=0$, then the factorization of Step 3 gives $\Zsf_{\bad,N_1,N_2,0}\leq C^{N_2}\Zsf_{\bad,N_1}^{(1)}\Zsf_{\bad,N_2}^{(2)}$, which proves the claim. We may thus assume $N_3\geq1$.
}
	
	Combining \eqref{eq:I4out} and \eqref{eq:J4out}, we get that there exists $C>0$ depending on $\beta, M$ and $p(\beta)$ such that
	\begin{align*}
		&\Bigl(\frac{N}{N_3}\Bigr)^{N_3}  \frac{1}{(NC_{\beta,\lambda,\ve_0})^{N_3}} \Zsf_{\bad,N_1,N_2,N_3}\\ &\leq e^{C(N_1+N_2+N_3)} \sum_{\tilde{N}_3=\textcolor{black}{\lceil\frac{1}{p(\beta)}N_3\rceil}}^{N_3}\binom{N_3}{\tilde{N}_3}\Bigl(\frac{N_1+N_2}{N}\Bigr)^{\tilde{N}_3}\Bigl(\frac{N_3}{N}\Bigr)^{N_3-\tilde{N}_3 }\Bigl(\frac{N}{N_3}\Bigr)^{N_3} \Zsf_{\bad,N_1}^{(1)}\Zsf_{\bad,N_2}^{(2)}
		\\ &=e^{C(N_1+N_2+N_3)} \sum_{\tilde{N}_3=\textcolor{black}{\lceil\frac{1}{p(\beta)}N_3\rceil}}^{N_3}\binom{N_3}{\tilde{N}_3}\Bigl(\frac{N_1+N_2}{N}\Bigr)^{\tilde{N}_3}\Bigl(\frac{N}{N_3}\Bigr)^{\tilde{N}_3 }\Zsf_{\bad,N_1}^{(1)}\Zsf_{\bad,N_2}^{(2)}.
	\end{align*}
	Since $\tilde{N}_3\leq N_3\leq p(\beta) \tilde{N}_3$, there exists a constant $C>0$ depending on $p(\beta)$ such that
	\begin{equation*}
		\binom{N_3}{\tilde{N}_3}\leq e^{CN_3} \le e^{Cp(\beta)\tilde N_3}.
	\end{equation*}
	Therefore, there exists $C>0$ depending on $\beta, M$ and $p(\beta)$ such that, for all $t\geq 1$,
	\begin{equation*}
		\begin{aligned}
			&\Bigl(\frac{N}{N_3}\Bigr)^{N_3}  \frac{1}{(NC_{\beta,\lambda,\ve_0})^{N_3}} \Zsf_{\bad,N_1,N_2,N_3}e^{tN_3}\\
			&\qquad\leq e^{C(N_1+N_2)} \sum_{\tilde{N}_3=1}^{N}e^{Ct\tilde{N}_3}\Bigl(\frac{N_1+N_2}{N}\Bigr)^{\tilde{N}_3}\Bigl(\frac{N}{\tilde{N}_3}\Bigr)^{\tilde{N}_3 }\Zsf_{\bad,N_1}^{(1)}\Zsf_{\bad,N_2}^{(2)}\\
			&\qquad\leq e^{C(N_1+N_2)} \sum_{\tilde{N}_3=1}^{N}\frac{e^{Ct\tilde{N}_3}}{\tilde N_3!}\Bigl(\frac{N_1+N_2}{N}\Bigr)^{\tilde{N}_3}N^{\tilde{N}_3} \Zsf_{\bad,N_1}^{(1)}\Zsf_{\bad,N_2}^{(2)}.
		\end{aligned}
	\end{equation*}
	Thus, summing over $\tilde{N}_3$, we recognize the exponential function, and  deduce that there exists $C>0$ depending on $\beta, M$ and $p(\beta)$ such that
	\begin{equation*}
		\max_{N_3}\left(\Bigl(\frac{N}{N_3}\Bigr)^{N_3}  \frac{1}{(NC_{\beta,\lambda,\ve_0})^{N_3}} \Zsf_{\bad,N_1,N_2,N_3}e^{tN_3}\right)\leq e^{e^{Ct}(N_1+N_2)} \Zsf_{\bad,N_1}^{(1)}\Zsf_{\bad,N_2}^{(2)},
	\end{equation*}
	as claimed.
\end{proof}

\subsection{Control of bad points of type 1}
\label{sub:type1}

Next, we control the number of bad points of type 1.

\begin{lemma}\label{lemma:type1}
	Let $\beta\in (2,\infty)$ and let $q(\beta)$ be as in \eqref{def:qbeta}. Let $\Zsf_{\bad,N_1}^{(1)}$ be as in \eqref{def:ZZbad1}. Then, there exists a constant $C>0$ depending on $\beta$, $\ve_0$ and $p(\beta)$ such that, for all $t\geq 1$,
	\begin{equation*}
		\max_{N_1}\left(\Bigl(\frac{N}{N_1}\Bigr)^{N_1}  \frac{1}{(NC_{\beta,\lambda,\ve_0})^{N_1} }\Zsf_{\bad,N_1}^{(1)}e^{t N_1}\right) \leq e^{e^{Ct} N\delta_{\beta,\lambda} }.
	\end{equation*}
\end{lemma}

\medskip

\begin{proof}\
	
	\paragraph{\bf{Step 1: removing the pairing}}
	First, by Lemma~\ref{lemma:stable},
	\begin{equation*}
		\prod_{i,j\in V_\bad^1:i\neq j}\indic_{\mc{A}_{ij}}=\indic_{\{\sigma_{N_1}=\Id\}},
	\end{equation*}
	where $\sigma_{N_1}$ is the stable matching over $V_\bad^1$ introduced in Definition~\ref{def:stable match}. Hence, since there are $N_1!$ permutations of $V_\bad^1$, and since particles are indistinguishable,
	\begin{equation}\label{eq:dematch0}
		\Zsf_{\bad,N_1}^{(1)} =\frac{1}{N_1!}\int \indic_{I_\bad^1=V_\bad^1}\prod_{i\in V_\bad^1}\frac{1}{(\tau_i^{+,Z})^{\frac{\beta}{2}} }\prod_{i\in V_\bad^1}\frac{1}{(\tau_i^{-,Z})^{\frac{\beta}{2}} }\prod_{i\in V_\bad^1}\dd x_i \dd y_i.
	\end{equation}
{\color{black}
In this integral, the radii $r_i=|x_i-y_{\sigma_{N_1}(i)}|$ and the bad sets use the stable matching $\sigma_{N_1}$ of the configuration. The sets $J_\bad^{1,+}$ and $J_\bad^{1,-}$ index positive and negative charges, respectively, and the length cutoff for $y_j$ uses $r_{\sigma_{N_1}^{-1}(j)}$.
}

	\paragraph{\bf{Step 2: re-clustering}}
	Recall the random subpartition $\Clus_{N_1}^+$ of $[N]$ defined in the clustering algorithm of Definition~\ref{def:continuous clustering}. With an abuse of notation, we view $\Clus_{N_1}^+$ as a partition of $V_\bad^1$. Recall $J_{\bad}^{1,+}$ and $J_\bad^{1,-}$ from Definition~\ref{def:type1}, with the abuse of notation referred to above. Note that we can restrict to the event of full measure where the distances $|x_i-x_j|$ and $|y_i-y_j|$ are all distinct, which ensures that at the disappearance parameter, clusters are either of cardinality $1$ or of cardinality $q(\beta)+1$.
	Given $S\in \Clus_{N_1}^+$, we thus  have three cases:
	\begin{enumerate}
		\item $|S|=q(\beta)+1$, which implies that $S\subset J_\bad^{1,+}$. We call such a cluster overcrowded.
		\item \textcolor{black}{$S=\{i\}$ with $r_i>\ve_0\Cut$, so $S\subset J_\bad^{1,+}$ and $\tau_i^{+,Z}\geq\ve_0\Cut/4$.}
		\item \textcolor{black}{$S=\{i\}$ with $i\notin J_\bad^{1,+}$ and $\sigma_{N_1}(i)\in J_\bad^{1,-}$, so $\tau_i^{+,Z}=\max(\tfrac14r_i,\lambda)$.} We say that $x_i$ is ``good''.
	\end{enumerate}
	The same holds for negative clusters. 
	
	\textcolor{black}{For $\lambda<\ve_0\Cut/4$, a singleton of the second kind has a bad partner:}
	\begin{equation}\label{eq:implies}
		\textcolor{black}{r_i>\ve_0\Cut\quad\Longrightarrow\quad\sigma_{N_1}(i)\in J_\bad^{1,-}.} 
	\end{equation}
	Moreover, $x_i$ and its partner $y_{\sigma_{N_1}(i)}$ cannot both be good, since otherwise the dipole would not lie in $I_{\bad}^{1}$:
	\begin{equation}\label{eq:implies2}
		\{i\in I_\bad^1\setminus J_\bad^{1,+}: \sigma_{N_1}(i)\in I_\bad^{1}\setminus J_\bad^{1,-}\}=\emptyset. 
	\end{equation}

	Now, we enlarge each overcrowded cluster of same-sign particles by adding partners that are good. For every $S\in \Clus_{N_1}^+$ such that $|S|=q(\beta)+1$, we let $\Good^-(S)$ denote the set of good negative partners,
	\begin{equation*}
		\Good^-(S)\coloneqq \left\{i\in S:\sigma_{N_1}(i)\in I_\bad^1\setminus J_\bad^{1,-}  \right\}.
	\end{equation*}
	Similarly, for every $S\in \Clus_{N_1}^-$, we let $\Good^+(S)$ denote the set of good positive partners,
	\begin{equation*}
		\Good^+(S)\coloneqq \left\{i\in S:\sigma_{N_1}^{-1}(i)\in I_\bad^1\setminus J_\bad^{1,+} \right\}.
	\end{equation*}
	
	In view of \eqref{eq:implies} and \eqref{eq:implies2}, we can partition the points into  
	\begin{multline}\label{eq:recluster}
		\bigcup_{i\in I_\bad^1}  \{x_i,y_i\}=\bigcup_{S\in \Clus^+_{N_1}:|S|=q(\beta)+1}\textcolor{black}{\left(\{x_i:i\in S\}\cup \{y_{\sigma_{N_1}(i)}:i\in \Good^-(S)\} \right)}\\
        \cup \bigcup_{S\in \Clus^-_{N_1}:|S|=q(\beta)+1}\textcolor{black}{\left( \{y_i:i\in S\}\cup \{x_{\sigma_{N_1}^{-1}(i)}:i\in \Good^+(S)\} \right)}\\
        \cup \bigcup_{\textcolor{black}{S\in \Clus^+_{N_1}:|S|=1,\ S\subset J_\bad^{1,+}}} \{x_i:i\in S\}
        \cup \bigcup_{\textcolor{black}{S\in \Clus^-_{N_1}:|S|=1,\ S\subset J_\bad^{1,-}}} \{y_i:i\in S\}.
	\end{multline}

	Denote by $\tilde{N}_1^+$ the cardinality of $J_\bad^{1,+}$ and by $\tilde{N}_1^-$ the cardinality of $J_\bad^{1,-}$, and split the partition function according to the values of $\tilde{N}_1^+$ and $\tilde{N}_1^-$: by \eqref{eq:dematch0}, 
	\begin{equation}\label{eq:z dematch}
		\Zsf_{\bad,N_1}^{(1)}=\frac{1}{N_1!}\sum_{\tilde{N}_1^+\leq N_1,\tilde{N}_1^-\leq N_1}\binom{N_1}{\tilde{N}_1^+}\binom{N_1}{\tilde{N}_1^-}I(\tilde{N}_1^+,\tilde{N}_1^-),
	\end{equation}
	where 
	\begin{equation*}
		I(\tilde{N}_1^+,\tilde{N}_1^-)\coloneqq \int_{(\Lambda^2)^{N_1} } \indic_{J_\bad^{1,+}=\tilde{V}_\bad^{1,+}} \indic_{J_\bad^{1,-}=\tilde{V}_\bad^{1,-}}\indic_{I_\bad^1=V_\bad^1}\prod_{i\in V_\bad^1}\frac{1}{(\tau_i^{+,Z})^{\frac{\beta}{2}} }\prod_{i\in V_\bad^1}\frac{1}{(\tau_i^{-,Z})^{\frac{\beta}{2}} }\prod_{i\in V_\bad^1}\dd x_i \dd y_i,
	\end{equation*}
	where $\tilde{V}_\bad^{1,+}\subset V_\bad^1$ and $\tilde{V}_\bad^{1,-}\subset V_\bad^1$ are \textcolor{black}{arbitrary sets} such that $|\tilde{V}_\bad^{1,+}|=\tilde{N}_1^+$ and $|\tilde{V}_\bad^{1,-}|=\tilde{N}_1^-$.
	
	\paragraph{\bf{Step 3: assigning a partner to good charges}} We now restrict to the event 
	$I_\bad^1=V_\bad^1, J_\bad^{1,+}=\tilde{V}_\bad^{1,+}, J_\bad^{1,-}=\tilde{V}_\bad^{1,-}$.
	Recall that, by \eqref{eq:implies2}, if $i\in V_\bad^1\setminus \tilde{V}_\bad^{1,+}$, then $\sigma_{N_1}(i)\in \tilde{V}_\bad^{1,-}$. Moreover, there are clearly fewer than $N_1$ choices for the value of $\sigma_{N_1}(i)$. Similarly, if $i\in V_\bad^1\setminus \tilde{V}_\bad^{1,-}$, then $\sigma^{-1}_{N_1}(i)\in \tilde{V}_\bad^{1,+}$, with fewer than $N_1$ choices for $\sigma^{-1}_{N_1}(i)$. Thus,
	\begin{equation}\label{eq:z match}
		I(\tilde{N}_1^+,\tilde{N}_1^-)\leq N_1^{N_1-\tilde{N}_1^+}N_1^{N_1-\tilde{N}_1^-} \max_{f^+,f^- }I(\tilde{N}_1^+,\tilde{N}_1^-,f^+,f^-),
	\end{equation}
	where the maximum is over injective maps 
	$f^+:V_\bad^1\setminus \tilde{V}_\bad^{1,-}\to \tilde{V}_\bad^{1,+}$ and $f^-:V_\bad^1\setminus \tilde{V}_\bad^{1,+}\to \tilde{V}_\bad^{1,-}$, and where  
	\begin{multline*}
		I(\tilde{N}_1^+,\tilde{N}_1^-,f^+,f^-)\coloneqq \int_{(\Lambda^2)^{N_1}} \indic_{J_\bad^{1,+}=\tilde{V}_\bad^{1,+}} \indic_{J_\bad^{1,-}=\tilde{V}_\bad^{1,-}}\indic_{I_\bad^1=V_\bad^1}\prod_{i\in V_\bad^1\setminus \tilde{V}_\bad^{1,+}}\indic_{\sigma_{N_1}(i)=f^-(i)} \\ \times \prod_{i\in V_\bad^1\setminus \tilde{V}_\bad^{1,-}}\indic_{\sigma_{N_1}^{-1}(i)=f^+(i)} \prod_{i\in V_\bad^1}\frac{1}{(\tau_i^{+,Z})^{\frac{\beta}{2}} }\prod_{i\in V_\bad^1}\frac{1}{(\tau_i^{-,Z})^{\frac{\beta}{2}} }\prod_{i\in V_\bad^1}\dd x_i \dd y_i.
	\end{multline*}
	{\color{black}
Fix $f^+,f^-$. Then $\Good^-(S)=S\cap\operatorname{im}f^+$ and $\Good^+(S)=S\cap\operatorname{im}f^-$.
}

	\paragraph{\bf{Step 4: splitting according to the partition of $\tilde{V}_\bad^{1,+}$ and $\tilde{V}_\bad^{1,-}$ }}
	We set 
	\begin{equation*}
		\widetilde{\Clus}^+\coloneqq \{S\in \Clus_{N_1}^+: S\subset \tilde{V}_\bad^{1,+}\}  \quad \mathrm{and}\quad   \widetilde{\Clus}^-\coloneqq \{S\in \Clus_{N_1}^-: S\subset \tilde{V}_\bad^{1,-}\}.
	\end{equation*}

	We will denote by $k_0^+$ the cardinality of $\widetilde{\Clus}^+$ and by $k_0^-$ the cardinality of $\widetilde{\Clus}^-$.

	For every $V\subset [N]$ and \textcolor{black}{$k_0\geq 0$}, let $\mathbf{\Pi}^{k_0}(V)$ be the set of partitions $\pi$ of $V$ with $k_0$ components such that, for every $S\in \pi$, $|S|=q(\beta)+1$  or $|S|=1$. We have 
	\begin{equation}\label{eq:z sum}
		I(\tilde{N}_1^+,\tilde{N}_1^-,f^+,f^-)=\sum_{k_0^+,k_0^-} \sum_{\pi^+\in \mathbf{\Pi}^{k_0^+}(\tilde{V}_\bad^{1,+}) }\sum_{\pi^-\in \mathbf{\Pi}^{k_0^-}(\tilde{V}_\bad^{1,-}) }I(\tilde{N}_1^+,\tilde{N}_1^-,f^+,f^-,\pi^+,\pi^-),
	\end{equation}
	where \begin{multline*} 
		I(\tilde{N}_1^+,\tilde{N}_1^-,f^+,f^-,\pi^+,\pi^-)\coloneqq   \int_{(\Lambda^2)^{N_1}} \indic_{J_\bad^{1,+}=\tilde{V}_\bad^{1,+}}\indic_{J_\bad^{1,-}=\tilde{V}_\bad^{1,-}} \indic_{I_\bad^1=V_\bad^1}\indic_{\widetilde{\Clus}^+=\pi^+ } \indic_{\widetilde{\Clus}^-=\pi^- }\\ \times \prod_{i\in V_\bad^1\setminus \tilde{V}_\bad^{1,+}}\indic_{\sigma_{N_1}(i)=f^-(i)}\prod_{i\in V_\bad^1\setminus \tilde{V}_\bad^{1,-}}\indic_{\sigma_{N_1}^{-1}(i)=f^+(i)} \prod_{i\in V_\bad^1}\frac{1}{(\tau_i^{+,Z})^{\frac{\beta}{2}} }\prod_{i\in V_\bad^1}\frac{1}{(\tau_i^{-,Z})^{\frac{\beta}{2}} }\prod_{i\in V_\bad^1}\dd x_i \dd y_i.
	\end{multline*}
	
	\paragraph{\bf{Step 5: integration}}
    \textcolor{black}{We may assume that $I(\tilde{N}_1^+,\tilde{N}_1^-,f^+,f^-,\pi^+,\pi^-)>0$.}
	Denote by $D$ the set of dyadic scales from $\lambda$ to $\ve_0\Cut$:
	\begin{equation*}
		D=\Bigr\{\lambda 2^k:0\leq k\leq 1+\Bigr\lfloor\log_2\Bigl(\frac{\ve_0\Cut}{\lambda}\Bigr)\Bigr\rfloor\Bigr\}.
	\end{equation*}
	
	For every $S\in \widetilde{\Clus}^+$ \textcolor{black}{with $|S|=q(\beta)+1$}, fix some $\iota_S^+\in S$ 
	and, for every $S\in \widetilde{\Clus}^-$ \textcolor{black}{with $|S|=q(\beta)+1$}, fix some $\iota_S^-\in S$. For every sequence of scales $(d_S^+)_{\textcolor{black}{S\in\widetilde{\Clus}^+:|S|=q(\beta)+1}}$ and $(d_S^-)_{\textcolor{black}{S\in\widetilde{\Clus}^-:|S|=q(\beta)+1}}$ taking values in $D$, set 
{\color{black}
\begin{multline*}
    J(\tilde N_1^+,\tilde N_1^-,f^+,f^-,\pi^+,\pi^-,(d_S^+),(d_S^-))
    \coloneqq\int_{(\Lambda^2)^{N_1}}
    \prod_{\substack{S\in\widetilde\Clus^+\\|S|=q(\beta)+1}}
    \frac{\prod_{i\in S}\indic_{|x_i-x_{\iota_S^+}|\leq8d_S^+}}{(d_S^+)^{\beta|S|/2}}\\
    \times\prod_{\substack{S\in\widetilde\Clus^-\\|S|=q(\beta)+1}}
    \frac{\prod_{i\in S}\indic_{|y_i-y_{\iota_S^-}|\leq8d_S^-}}{(d_S^-)^{\beta|S|/2}}
    \prod_{\substack{S\in\widetilde\Clus^+\\|S|=q(\beta)+1}}\prod_{i\in\Good^-(S)}
    \frac{\indic_{|x_i-y_{(f^+)^{-1}(i)}|\leq\ve_0\Cut}}{\max(|x_i-y_{(f^+)^{-1}(i)}|,\lambda)^{\beta/2}}\\
    \times\prod_{\substack{S\in\widetilde\Clus^-\\|S|=q(\beta)+1}}\prod_{i\in\Good^+(S)}
    \frac{\indic_{|y_i-x_{(f^-)^{-1}(i)}|\leq\ve_0\Cut}}{\max(|y_i-x_{(f^-)^{-1}(i)}|,\lambda)^{\beta/2}}\\
    \times\prod_{\substack{S\in\widetilde\Clus^+\\|S|=1}}(\ve_0\Cut/4)^{-\beta/2}
    \prod_{\substack{S\in\widetilde\Clus^-\\|S|=1}}(\ve_0\Cut/4)^{-\beta/2}
    \prod_{i\in V_\bad^1}\dd x_i\dd y_i.
\end{multline*}
}

	Observe that, for every $S\in \widetilde{\Clus}^+$ \textcolor{black}{with $|S|=q(\beta)+1$},
	\begin{equation}\label{eq:ww1}
		\indic_{\tau_{\textcolor{black}{\iota_S^+}}^{+,Z} \in [d_S^+,2d_S^+)}\leq \prod_{i\in S}\indic_{|x_i-x_{\textcolor{black}{\iota_S^+}}|\leq \textcolor{black}{8}d_S^+}.
	\end{equation}
	Moreover, by the construction of the clustering algorithm, the radii $\tau_i^{+,Z}$ are equal within each cluster, and the same holds for the $\tau_i^{-,Z}$. Thus, splitting the phase space according to the dyadic scale of $\tau_{\iota_S^+}^{+,Z}$ and $\tau_{\iota_S^-}^{-,Z}$ in each cluster and using \eqref{eq:ww1}, we get
	\begin{multline*}
		I(\tilde{N}_1^+,\tilde{N}_1^-,f^+,f^-,\pi^+,\pi^-)\\
        \leq \textcolor{black}{C^{N_1}}\prod_{\textcolor{black}{S\in \widetilde{\Clus}^+:|S|=q(\beta)+1}} \sum_{d_S^+\in D } \prod_{\textcolor{black}{S\in \widetilde{\Clus}^-:|S|=q(\beta)+1}} \sum_{d_S^-\in D } J(\tilde{N}_1^+,\tilde{N}_1^-,f^+,f^-,\pi^+,\pi^-,(d_S^+), (d_S^-)).
	\end{multline*}

	\textcolor{black}{Fix the scales $(d_S^+),(d_S^-)$.}
	
	$\bullet$ Suppose that $\beta\in (2,4)$. Integrating for each cluster over all indices except one (which is either $x_{\textcolor{black}{\iota_S^+}}$ or $y_{\textcolor{black}{\iota_S^-}}$) and then integrating out that last variable over the whole domain $\Lambda$, one obtains that there exists $C>0$ depending on $\beta$ and $q(\beta)$ such that
	\begin{align*}
		&J(\tilde{N}_1^+,\tilde{N}_1^-,f^+,f^-,\pi^+,\pi^-,(d_S^+), (d_S^-))\\ &\leq (CN)^{k_0^++k_0^-} \prod_{\substack{S\in \widetilde{\Clus}^+:\\|S|=q(\beta)+1}}(d_S^+)^{(2-\frac{\beta}{2})|S|-2}(\ve_0 R_{\beta,\lambda})^{(2-\frac{\beta}{2})|\Good^-(S)|} \\ &\times\prod_{\substack{S\in \widetilde{\Clus}^-:\\|S|=q(\beta)+1}}(d_S^-)^{(2-\frac{\beta}{2})|S|-2}(\ve_0 R_{\beta,\lambda})^{(2-\frac{\beta}{2})|\Good^+(S)|} \prod_{\substack{S\in \widetilde{\Clus}^+:\\|S|=1}}(\ve_0 R_{\beta,\lambda})^{-\frac{\beta}{2}} \prod_{\substack{S\in \widetilde{\Clus}^-:\\|S|=1}}(\ve_0 R_{\beta,\lambda})^{-\frac{\beta}{2}}.
	\end{align*}
	By definition \eqref{def:qbeta}, we have $(2-\frac{\beta}{2})(q(\beta)+1)-2\geq 1>0$, and therefore, summing over $d_S^+$ and $d_S^-$ in $D$  gives 
	\begin{multline*}
		I(\tilde{N}_1^+,\tilde{N}_1^-,f^+,f^-,\pi^+,\pi^-)\leq (CN)^{k_0^++k_0^-} \prod_{\substack{S\in \widetilde{\Clus}^+:\\|S|=q(\beta)+1}}(\ve_0R_{\beta,\lambda})^{(2-\frac{\beta}{2})|S|-2}(\ve_0 R_{\beta,\lambda})^{(2-\frac{\beta}{2})|\Good^-(S)|} \\ \times\prod_{\substack{S\in \widetilde{\Clus}^-:\\|S|=q(\beta)+1}}(\ve_0R_{\beta,\lambda})^{(2-\frac{\beta}{2})|S|-2}(\ve_0 R_{\beta,\lambda})^{(2-\frac{\beta}{2})|\Good^+(S)|} \prod_{\substack{S\in \widetilde{\Clus}^+:\\|S|=1}}(\ve_0 R_{\beta,\lambda})^{(2-\frac{\beta}{2})-2} \prod_{\substack{S\in \widetilde{\Clus}^-:\\|S|=1}}(\ve_0 R_{\beta,\lambda})^{(2-\frac{\beta}{2})-2}.
	\end{multline*}
	\textcolor{black}{By \eqref{eq:recluster}, the cluster sizes and the numbers of good partners sum to $2N_1$. Since each cluster also contributes the exponent $-2$, the total exponent of $\ve_0R_{\beta,\lambda}$ is $(4-\beta)N_1-2(k_0^++k_0^-)$. Hence,} 
	\begin{equation*}
		I(\tilde{N}_1^+,\tilde{N}_1^-,f^+,f^-,\pi^+,\pi^-)\leq (CN)^{k_0^++k_0^-} (\ve_0 R_{\beta,\lambda})^{(4-\beta)N_1}(\ve_0R_{\beta,\lambda})^{-2(k_0^++k_0^-)}.
	\end{equation*}
	Recalling that $R_{\beta,\lambda}^{4-\beta}=\lambda^{2-\beta}$, we find that there exists $C>0$ depending on $\beta, q(\beta)$ and $\ve_0$ such that
	\begin{equation}\label{eq:I4-}
		I(\tilde{N}_1^+,\tilde{N}_1^-,f^+,f^-,\pi^+,\pi^-)\leq (CN)^{k_0^++k_0^-} \lambda^{(2-\beta)N_1}R_{\beta,\lambda}^{-2(k_0^++k_0^-)}.
	\end{equation}
	
{\color{black}
$\bullet$ Suppose that $\beta\geq4$. Let $c$ be the number of overcrowded clusters, $h$ the number of singletons in $\widetilde\Clus^+\cup\widetilde\Clus^-$, and $g$ the number of good partners. Then
\[
    2N_1=(q(\beta)+1)c+h+g,\qquad k_0^++k_0^-=c+h.
\]
The integral over each good partner is bounded by $C\lambda^{2-\beta/2}$ for $\beta>4$, and by $C(1+|\log\lambda|)$ for $\beta=4$. The dyadic sums over overcrowded clusters give
\begin{equation*}\frac{I(\tilde N_1^+,\tilde N_1^-,f^+,f^-,\pi^+,\pi^-)} {\lambda^{(2-\beta)N_1}} \leq C^{N_1}N^{c+h} \lambda^{(q(\beta)-1)c+g+(\beta p_0+(\beta-2)/2)h} \times(1+|\log\lambda|\indic_{\beta=4})^g.
\end{equation*}
Since $q(\beta)-1=2p_0$, $\lambda(1+|\log\lambda|)\leq C$ and $N_1\leq(q(\beta)+1)(c+h)$, we obtain
\begin{equation}\label{eq:I4+}
    I(\tilde N_1^+,\tilde N_1^-,f^+,f^-,\pi^+,\pi^-)
    \leq (CN)^{k_0^++k_0^-}
    \lambda^{(2-\beta)N_1}\lambda^{2p_0(k_0^++k_0^-)}.
\end{equation}
}

	Assembling \eqref{eq:I4-} and \eqref{eq:I4+} gives the existence of $C>0$ depending on $\beta$, $q(\beta)$ and $\ve_0$ such that
	\begin{equation}\label{eq:Igen}
		I(\tilde{N}_1^+,\tilde{N}_1^-,f^+,f^-,\pi^+,\pi^-)\leq (CN)^{k_0^++k_0^-} \lambda^{(2-\beta)N_1} \delta_{\beta,\lambda}^{k_0^++k_0^-}, 
	\end{equation}
	where $\delta_{\beta, \lambda}$ is as in \eqref{defdelta}.

	\paragraph{\bf{Step 6: conclusion}}
{\color{black}
We may assume $N_1\geq1$. The number of partitions in $\mathbf{\Pi}^{k_0^\pm}(\tilde V_\bad^{1,\pm})$ is at most
$C^{k_0^\pm}(\tilde N_1^\pm)^{\tilde N_1^\pm}/k_0^\pm!$.
Combining \eqref{eq:z dematch}, \eqref{eq:z match}, \eqref{eq:z sum} and \eqref{eq:Igen}, and using Stirling's formula, we obtain
\begin{equation*}
    \Zsf_{\bad,N_1}^{(1)}
    \leq(CN_1)^{N_1}\lambda^{(2-\beta)N_1}
    \sum_{\substack{k_0^+,k_0^-\geq0\\N_1\leq(q(\beta)+1)(k_0^++k_0^-)}}
    \frac{(CN\delta_{\beta,\lambda})^{k_0^++k_0^-}}{k_0^+!k_0^-!}.
\end{equation*}
Here the sums over $\tilde N_1^\pm$ cost at most $C^{N_1}$, and the constraint on $k_0^++k_0^-$ holds since every dipole has a bad charge. By \eqref{eq:Clambda bound}, for $t\geq1$,
\begin{equation*}
    \Bigl(\frac{N}{N_1}\Bigr)^{N_1}
    \frac{\Zsf_{\bad,N_1}^{(1)}e^{tN_1}}{(NC_{\beta,\lambda,\ve_0})^{N_1}}
    \leq\sum_{k_0^+,k_0^-\geq0}
    \frac{(e^{Ct}N\delta_{\beta,\lambda})^{k_0^++k_0^-}}{k_0^+!k_0^-!}
    =e^{2e^{Ct}N\delta_{\beta,\lambda}}.
\end{equation*}
Increasing $C$ proves the lemma.
}
\end{proof}

\subsection{Control of bad points of type 2}\label{sub:type2}

Recalling Definition~\ref{def:type2}, we write 
{\color{black}
\begin{equation}\label{eq:part type2}
    \Zsf_{\bad,N_2}^{(2)}
    \leq(p(\beta)+1)^{N_2}
    \max_{\substack{N_2'+\sum_{k=1}^{p(\beta)}N_k''=N_2}}
    \frac{N_2!}{N_2'!\prod_{k=1}^{p(\beta)}N_k''!}
    \Zsf_{\bad,N_2'}^{(2,1)}
    \prod_{k=1}^{p(\beta)}\Zsf_{\bad,N_k''}^{(2,2,k)},
\end{equation}
}
where 
\begin{equation}\label{def:ZZbad2'}
	\Zsf_{\bad,N_2'}^{(2,1)}\coloneqq \int_{(\Lambda^2)^{N_2'}} \indic_{I_\bad^{2,1}=V_\bad^{2,1} }\prod_{i,j\in V_\bad^{2,1} :i<j}\indic_{\mc{A}_{ij}}\prod_{i\in V_\bad^{2,1}}e^{\beta \g_\lambda(x_i-y_i)}\dd x_i \dd y_i
\end{equation}
and, for every $k\in [p(\beta)]$,
\begin{equation}\label{def:Zbadk}
	\Zsf_{\bad,N_k''}^{(2,2,k)}\coloneqq \int_{(\Lambda^2)^{N_k''}} \indic_{I_\bad^{2,2,k}=V_\bad^{2,2,k} }\prod_{i,j\in V_\bad^{2,2,k} :i<j}\indic_{\mc{A}_{ij}}\prod_{i\in V_\bad^{2,2,k}}e^{\beta \g_\lambda(x_i-y_i)}\dd x_i \dd y_i,
\end{equation}
for some arbitrary disjoint sets $V_\bad^{2,1}$ and $V_\bad^{2,2,1},\ldots,V_{\bad}^{2,2,p(\beta)}$ of respective cardinalities $N_2'$ and $N_1'',\ldots,N_{p(\beta)}''$.

We begin with bad points of type $(2,1)$, which are those in multipoles of cardinality larger than $p(\beta)$.

\begin{lemma}\label{lemma:type2}
	Let $\beta\in (2,\infty)$ and let $p(\beta)$ be as in  Definition~\ref{def:pbeta}. Let $\Zsf_{\bad,N_2}^{(2,1)}$ be as in \eqref{def:ZZbad2'}. Let $t\geq 1$.
	
	\begin{enumerate}
		\item Suppose that $\beta\in (2,4)$. There exists $\tilde{\ve}_0$ depending on $\beta, M$ and $t$ such that, for all $\ve_0\in (0,\tilde{\ve}_0)$ \textcolor{black}{and $\lambda$ small enough depending on $\beta,M,\ve_0,t$,} there exists a constant $C>0$ depending on \textcolor{black}{$\beta,M,\ve_0$ and $t$} such that
		\begin{equation*}
			\max_{N_2}\left(\Bigl(\frac{N}{N_2}\Bigr)^{N_2}\frac{ \Zsf_{\bad,N_2}^{(2,1)}}{(NC_{\beta,\lambda,\ve_0})^{N_2}}e^{tN_2}\right) \leq \textcolor{black}{e^{CN R_{\beta,\lambda}^{-2}{(1+|\log\lambda|\indic_{\beta=\beta_{p(\beta)+1}})}}}.
		\end{equation*}   
		\item Suppose that $\beta\geq 4$. There exists $\tilde{\lambda}$ depending on \textcolor{black}{$\beta,M,p_0,\ve_0$ and $t$} such that, for all $\lambda\in (0,\tilde{\lambda})$, there exists $C>0$ depending on \textcolor{black}{$\beta,M,p_0$ and $t$} such that 
		\begin{equation*}
			\max_{N_2}\left(\Bigl(\frac{N}{N_2}\Bigr)^{N_2}\frac{ \Zsf_{\bad,N_2}^{(2,1)}}{(NC_{\beta,\lambda,\ve_0})^{N_2}}e^{tN_2}\right) \leq \textcolor{black}{e^{CN \lambda^{2p_0}}}.
		\end{equation*}  
	\end{enumerate}
\end{lemma}

\medskip

\begin{proof}
	Let $\mathbf{\Pi}'(V_\bad^{2,1})$ be the set of partitions $\pi$ of $V_\bad^{2,1}$ such that, for every $S\in \pi$, we have $|S|>p(\beta)$. One can write 
	\begin{equation}\label{eq:sumJpi}
		\Zsf_{\bad,N_2}^{(2,1)}=\sum_{\pi \in \mathbf{\Pi}'(V_\bad^{2,1})}\textcolor{black}{J}(\pi),
	\end{equation}
	where 
	\begin{equation*}
		\textcolor{black}{J}(\pi)\coloneqq \int_{(\Lambda^2)^{N_2} } \indic_{\{ I_\bad^{2,1}=V_\bad^{2,1}\}}\indic_{\Pi_\mult=\pi} \prod_{i,j\in V_\bad^{2,1} :i<j}\indic_{\mc{A}_{ij}} \prod_{i\in V_\bad^{2,1}}e^{\beta \g_\lambda(x_i-y_i)}\dd x_i \dd y_i.
	\end{equation*}
	Fix $\pi\in \mathbf{\Pi}'(V_\bad^{2,1})$. Recall that $\mc{T}_c(S)$ stands for the set of collections of edges $E$ on $S$ such that $(S,E)$ is a tree. Also, by a union bound,
	\begin{equation*}
		\indic_{\mc{B}_S}\leq \sum_{T_S\in \mc{T}_c(S)}\prod_{ij\in T_S}\indic_{\mc{B}_{ij}}\leq \left(\max_{T_S\in \mc{T}_c(S)}\prod_{ij\in T_S}\indic_{\mc{B}_{ij}}\right)|S|^{|S|-2},
	\end{equation*}
	where we have used that, by Cayley's formula, the number of trees on $S$ equals $|S|^{|S|-2}$. 
	
	\textcolor{black}{We drop the indicators $\indic_{\mc A_{ij}}$ and sum over the trees $T_S$ and the indices $\iota_S\in S$, keeping $r_i\leq r_{\iota_S}$ for $i\in S$. Counting these choices gives}
	\begingroup\color{black}
	\begin{equation}\label{eq:bJpi}
		J(\pi)\leq \prod_{S\in \pi}|S|^{|S|-2}|S|\max_{(T_S),(\iota_S)}J'(\pi,(T_S),(\iota_S)),
	\end{equation}
	\endgroup
	where 
	\begin{equation*}
		\textcolor{black}{J}'(\pi,(T_S),(\iota_S))= \int_{(\Lambda^2)^{N_2} } \prod_{S\in \pi}\left(\prod_{ij\in T_S}\indic_{\mc{B}_{ij}} \prod_{i\in S}\indic_{|x_i-y_i|\leq r_{\iota_S}\leq \ve_0\Cut} e^{\beta \g_\lambda(x_i-y_i)}\right)\prod_{i\in V_\bad^{2,1}}\dd x_i \dd y_i.
	\end{equation*}For every $S\in \pi$, fix a tree $T_S$ on $S$ and an index $\iota_S$ in $S$.
	Recall that $\mc{B}_{ij}=\{d_{ij}\leq M\min(\max(r_i,\lambda),\max(r_j,\lambda))\}$ and $d_{ij}=\dist(\{x_i,y_i\},\{x_j,y_j\})$. For every $ij\in \cup_S T_S$, we split the phase space according to the variables achieving $d_{ij}$. We let
	\begin{equation*}
		\begin{aligned}
			 & \mc{D}_{ij}^1 =\{(x_q,y_q)_{q\in S}: d_{ij}= |x_i-x_j|\},\qquad \mc{D}_{ij}^2 =\{(x_q,y_q)_{q\in S}: d_{ij}= |x_i-y_j|\},\\
			 & \mc{D}_{ij}^3 =\{ (x_q,y_q)_{q\in S}:d_{ij}= |y_i-x_j|\},\qquad \mc{D}_{ij}^4 =\{ (x_q,y_q)_{q\in S}: d_{ij}= |y_i-y_j|\}.
		\end{aligned}
	\end{equation*}
	Define
	\begin{multline*}
		\textcolor{black}{J}''(\pi,(T_S^l)_{l\in [4],S\in \pi},(\iota_S))\\ \coloneqq   \int_{(\Lambda^2)^{N_2} } \prod_{S\in \pi}\left( \prod_{l=1}^4\prod_{ij\in T_S^l}\indic_{\mc{B}_{ij}} \indic_{\mc{D}_{ij}^l}
		\prod_{i\in S}\indic_{|x_i-y_i|\leq r_{\iota_S}\leq \ve_0\Cut} e^{\beta \g_\lambda(x_i-y_i)}\right)\prod_{i\in V_\bad^{2,1}}\dd x_i \dd y_i.
	\end{multline*}
	Notice that
	\begin{equation}\label{eq:j'split}
		\textcolor{black}{J}'(\pi,(T_S),(\iota_S))=\sum_{(T_S^1),(T_S^2),(T_S^3),(T_S^4)} \textcolor{black}{J}''(\pi,(T_S^l)_{l\in [4],S\in \pi},(\iota_S)), 
	\end{equation}
	where the sum is over sequences of \textcolor{black}{forests} such that, for every $S\in \pi$, $T_S^1,T_S^2,T_S^3,T_S^4$ are disjoint and $T_S^1\cup T_S^2\cup T_S^3\cup T_S^4=T_S$. Let us fix $T_S^1, T_S^2, T_S^3,T_S^4$ accordingly. Define the variables 
	\begin{equation*}
		z_{ij}=\begin{cases}
			x_{i}-x_j & \text{if $ij\in T_S^1$}\\
			x_i-y_j & \text{if $ij\in  T_S^2 $}\\
			y_i - x_j & \text{if $ij\in T_S^3$}\\
			y_i-y_j & \text{if $ij\in T_S^4$}.
		\end{cases}
	\end{equation*}
	
	Then, perform the change of variables
	\begin{equation*}
		(x_i,y_i)_{i\in V_\bad^{2,1}}\mapsto ((x_{\iota_S})_{S\in \pi},(z_{ij})_{ij\in \cup_{S\in \pi}T_S},(\vec{r}_i)_{i\in V_\bad^{2,1}}).
	\end{equation*}
	Integrating out the variables $(x_{\iota_S})_{S\in \pi}$ yields 
	\begin{multline*}
		\textcolor{black}{J}''(\pi,(T_S^l)_{l\in [4],S\in \pi},(\iota_S)) \leq N^{|\pi|}\prod_{S\in \pi}\int_{\textcolor{black}{(\dR^2)^{2|S|-1}}} \prod_{ij\in T_S}\indic_{|z_{ij}|\leq M\min(\max(r_i,\lambda),\max(r_j,\lambda)) }\\ \times \left(\prod_{i\in S}\indic_{r_i\leq r_{\iota_S}\leq \ve_0\Cut}e^{\beta \g_\lambda(r_i)}\right)\prod_{ij\in T_S}\dd z_{ij}\prod_{i\in S}\dd \vec{r}_i. 
	\end{multline*}
	Integrating out the variables $z_{ij}$ yields
	\begin{multline*}
		\textcolor{black}{J}''(\pi,(T_S^l)_{l\in [4],S\in \pi},(\iota_S))\leq (CM^2)^{N_2}N^{|\pi|}\prod_{S\in \pi}\int_{\textcolor{black}{(\dR^2)^{|S|}}} \prod_{ij\in T_S}\min(\max(r_i,\lambda),\max(r_j,\lambda))^2\\ \times \left(\prod_{i\in S}\indic_{r_i\leq r_{\iota_S}\leq \ve_0\Cut}e^{\beta \g_\lambda(r_i)}\right)\prod_{i\in S}\dd \vec{r}_i. 
	\end{multline*}
	By Lemma~\ref{lemma:ordering tree}, on the event where $r_{\iota_S}=\max_{i\in S}r_i$,
	\begin{equation*}
		\prod_{ij\in T_S}\min(\max(r_i,\lambda),\max(r_j,\lambda))^2\leq \prod_{i\in S:i\neq \iota_S}\max(r_i,\lambda)^2.
	\end{equation*}
	Therefore, integrating out the variables $\vec{r}_i$, $i\neq \iota_S$, and performing a polar change of coordinates gives 
	{\color{black}
	\begin{multline}\label{eq:j'' last}
		\textcolor{black}{J}''(\pi,(T_S^l)_{l\in [4],S\in \pi},(\iota_S))\leq (CM^2)^{N_2}N^{|\pi|}\\
		\times\prod_{S\in \pi}\int_0^{\ve_0\Cut} r\max(r,\lambda)^{-\beta}
		\left(\int_0^r u\max(u,\lambda)^{2-\beta}\dd u\right)^{|S|-1}\dd r.
	\end{multline}
	}
	
	$\bullet$ Suppose $\beta\in (2,4)$. In this case, recall that $\Cut=R_{\beta,\lambda}$. \textcolor{black}{The inner integral is bounded by $C\max(r,\lambda)^{4-\beta}$.} By assumption, for every $S\in \pi$, $|S|>p(\beta)$ with $p(\beta)=p^*(\beta)$. Hence, for $\beta\in (\beta_{p(\beta)},\beta_{p(\beta)+1}]$ (recall \eqref{defbetap}), we have $(4-\beta)|S|-3\textcolor{black}{\geq}-1$ and therefore 
	\begin{equation*}
		\textcolor{black}{J}''(\pi,(T_S^l)_{l\in [4],S\in \pi},(\iota_S))\leq (CM^2)^{N_2}N^{|\pi|} \prod_{S\in \pi}(\ve_0R_{\beta,\lambda})^{(4-\beta)|S|-2}{(1+|\log\lambda|\indic_{\beta=\beta_{p(\beta)+1}})}.  
	\end{equation*}
	Therefore, recalling that $R_{\beta,\lambda}^{4-\beta}=\lambda^{2-\beta}$, we get 
	\begin{equation}\label{eq:J'' 1}
		\textcolor{black}{J}''(\pi,(T_S^l)_{l\in [4],S\in \pi},(\iota_S))\leq (CM^2)^{N_2}N^{|\pi|} \prod_{S\in \pi}\lambda^{(2-\beta)|S|}\ve_0^{(4-\beta)|S|}(\ve_0R_{\beta,\lambda})^{-2}{(1+|\log\lambda|\indic_{\beta=\beta_{p(\beta)+1}})}.  
	\end{equation}
	
	{\color{black}
	$\bullet$ Suppose that $\beta\geq4$. Set
	\[
		\eta_{\beta,\lambda}=\lambda^2(1+|\log\lambda|\indic_{\beta=4}).
	\]
	In each $T_S$, we drop the ordering constraint on the radii and integrate out leaves one at a time until $p_0+1$ vertices remain. Each removed leaf contributes at most
	\[
		C\int_0^{\ve_0\Cut}r\max(r,\lambda)^{2-\beta}\dd r
		\leq C\lambda^{2-\beta}\eta_{\beta,\lambda}.
	\]
	For the remaining tree, summing over its largest radius and arguing as in \eqref{eq:j'' last} gives
	\[
		C\int_0^\infty r\max(r,\lambda)^{-\beta}
		\left(\int_0^r u\max(u,\lambda)^{2-\beta}\dd u\right)^{p_0}\dd r
		\leq C\lambda^{(4-\beta)(p_0+1)-2}.
	\]
	Here $C$ depends on $\beta,p_0$. At $\beta=4$, the rescaled integrand is bounded by $Cr^{-3}(1+\log r)^{p_0}$ for $r\geq1$. Hence
	\begin{equation}\label{eq:J'' 2}
		\textcolor{black}{J}''(\pi,(T_S^l)_{l\in [4],S\in \pi},(\iota_S))\leq (CM^2)^{N_2}N^{|\pi|}\prod_{S\in \pi}
		\lambda^{(2-\beta)|S|}\lambda^{2p_0}\eta_{\beta,\lambda}^{|S|-p_0-1}.
	\end{equation}
	}
	
	For every $\textcolor{black}{m>p(\beta)}$, let us define
	\begin{equation*}
		f_\beta(m)=\begin{cases}
			\ve_0^{(4-\beta)m}(\ve_0R_{\beta,\lambda})^{-2}{(1+|\log\lambda|\indic_{\beta=\beta_{p(\beta)+1}})} & \text{if $\beta\in (2,4)$}\\
			\textcolor{black}{\lambda^{2p_0}\eta_{\beta,\lambda}^{m-p_0-1}} & \text{if $\beta\geq 4$}.
		\end{cases}
	\end{equation*}
	One can combine the bounds \eqref{eq:J'' 1} and \eqref{eq:J'' 2} into 
	\begin{equation*}
		\textcolor{black}{J}''(\pi,(T_S^l)_{l\in [4],S\in \pi},(\iota_S))\leq (CM^2)^{N_2}N^{|\pi|}\prod_{S\in \pi} \lambda^{(2-\beta)|S|}f_\beta(|S|).
	\end{equation*}
	Since $C_{\beta,\lambda,\ve_0}\geq \frac{1}{C}\lambda^{2-\beta}$, we deduce, using \eqref{eq:bJpi} and \eqref{eq:j'split}, that
	\begin{equation*}
		\frac{1}{C_{\beta,\lambda,\ve_0}^{N_2}}\textcolor{black}{J}(\pi)\leq (CM^2)^{N_2} N^{|\pi|}\prod_{S\in \pi} |S|!|S| f_\beta(|S|).
	\end{equation*}

	\textcolor{black}{For $N_2>0$,} we now sum this over $\pi\in  \mathbf{\Pi}'(V_\bad^{2,1})$. We first sum according to the number of components $k_0$ of the partition, then divide by $k_0!$, since the blocks are indistinguishable, and then sum according to the sizes of the blocks. Given $k_0$ labeled blocks, the number of ways to assign $m_i$ elements to block $i$ for every $i=1,\ldots,k_0$ is given by $\frac{N_2!}{m_1!\ldots m_{k_0}!}$. Thus,
	\begin{multline}\label{eq:sumpiJ}
		\sum_{\pi\in \mathbf{\Pi}'(V_\bad^{2,1})} \frac{1}{C_{\beta,\lambda,\ve_0}^{N_2}}\textcolor{black}{J}(\pi)=\sum_{k_0=1}^{N_2} \sum_{\pi\in \mathbf{\Pi}'(V_\bad^{2,1}):|\pi|=k_0} \frac{1}{C_{\beta,\lambda,\ve_0}^{N_2}}\textcolor{black}{J}(\pi) \\
		\leq (CM^2 )^{N_2}\sum_{k_0=1}^{N_2}\frac{1}{k_0!}\sum_{m_1+\cdots+m_{k_0}=N_2:m_i>p(\beta)}
		\frac{N_2!}{m_1!\cdots m_{k_0}!}
		N^{k_0}\prod_{i=1}^{k_0}m_i! m_i f_\beta(m_i).
	\end{multline}
	Therefore
	\begin{multline*}
		\Bigl(\frac{N}{N_2}\Bigr)^{N_2} \frac{ \Zsf_{\bad,N_2}^{(2,1)} }{(NC_{\beta,\lambda,\ve_0})^{N_2}}\\ \leq (CM^2)^{N_2} \Bigl(\frac{N}{N_2}\Bigr)^{N_2}\frac{1}{N^{N_2}}\sum_{k_0=1}^{N_2}\frac{N_2!}{k_0!}N^{k_0}\sum_{m_1+\cdots+m_{k_0}=N_2:m_i>p(\beta)}\prod_{i=1}^{k_0}\Bigl(m_i f_\beta(m_i)\Bigr).
	\end{multline*}
	Applying Stirling's formula to $N_2!$, we get
	\begin{equation*}
		\Bigl(\frac{N}{N_2}\Bigr)^{N_2}\frac{ \Zsf_{\bad,N_2}^{(2,1)} }{(NC_{\beta,\lambda,\ve_0})^{N_2}} \leq (CM^2)^{N_2}\sum_{k_0=1}^{N_2}\frac{N^{k_0}}{k_0!}\sum_{m_1+\cdots+m_{k_0}=N_2:m_i>p(\beta)}\prod_{i=1}^{k_0}m_i f_\beta(m_i).
	\end{equation*}
	{\color{black}
	Summing over $N_2\geq0$ gives, for $C$ depending on $\beta,M,p(\beta)$,
	\begin{equation*}
		\max_{N_2}\Bigl(\frac{N}{N_2}\Bigr)^{N_2}\frac{\Zsf_{\bad,N_2}^{(2,1)} }{(NC_{\beta,\lambda,\ve_0})^{N_2}}e^{tN_2}
		\leq\exp\left(N\sum_{m>p(\beta)}m f_\beta(m)C^me^{tm}\right).
	\end{equation*}
	}

	$\bullet$ Suppose that $\beta\in (2,4)$. There exists $\tilde{\ve}_0$ depending on $\beta, M$ and $t$ such that, for all $\ve_0\in (0,\tilde{\ve}_0)$, there exists $C'>0$ depending on $\beta$, $M$, $\ve_0$ and $t$ such that
	\begin{equation*}
		\sum_{m>p(\beta)}m f_\beta(m)C^{m}e^{tm}\leq C'R_{\beta,\lambda}^{-2}{(1+|\log\lambda|\indic_{\beta=\beta_{p(\beta)+1}})}.
	\end{equation*}
	Therefore, for $\ve_0\in (0,\tilde{\ve}_0)$, there exists a constant $C>0$ depending on $\beta$, $M$, $\ve_0$ and $t$ such that
	\begin{equation*}
		\max_{N_2}\Bigl(\frac{N}{N_2}\Bigr)^{N_2}\frac{\Zsf_{\bad,N_2}^{(2,1)} }{(NC_{\beta,\lambda,\ve_0})^{N_2}}e^{tN_2}\leq \textcolor{black}{\sum_{k_0=0}^{\infty}}\frac{N^{k_0}}{k_0!}(CR_{\beta,\lambda}^{-2}{(1+|\log\lambda|\indic_{\beta=\beta_{p(\beta)+1}})})^{k_0}\leq e^{CNR_{\beta,\lambda}^{-2}{(1+|\log\lambda|\indic_{\beta=\beta_{p(\beta)+1}})}},
	\end{equation*}
	which concludes the proof in the case $\beta\in (2,4)$.
	
	{\color{black}
	$\bullet$ Suppose that $\beta\geq4$. For $\lambda$ small enough that $Ce^t\eta_{\beta,\lambda}\leq\frac12$, the series above is bounded by $C\lambda^{2p_0}$, where $C$ depends on $\beta,M,p_0,t$. Hence
	\begin{equation*}
		\max_{N_2}\Bigl(\frac{N}{N_2}\Bigr)^{N_2}\frac{\Zsf_{\bad,N_2}^{(2,1)} }{(NC_{\beta,\lambda,\ve_0})^{N_2}}e^{tN_2}\leq e^{CN\lambda^{2p_0}}.
	\end{equation*}
	}
	
\end{proof}

Next, we give a bound on bad points of type $(2,2)$, which are those in multipoles of cardinality $k\in \{1,\ldots,p(\beta)\}$ such that $\mc{N}_k$ is above the threshold of Definition~\ref{def:type2}.

\begin{lemma}\label{lemma:type22}
	Let $k\in \{1,\ldots,p(\beta)\}$ and let $\Zsf_{\bad,N''_k}^{(2,2,k)}$ be as in \eqref{def:Zbadk}. Then, there exists a constant $C>0$ depending on $\beta,p(\beta)$ and $M$ such that
	\begin{equation*}
		\Bigl(\frac{N}{N''_k}\Bigr)^{N_k''}\frac{ \Zsf_{\bad,N_k''}^{(2,2,k)}}{(NC_{\beta,\lambda,\ve_0})^{N_k''}}\leq C^{N_k''}\Bigl(\frac{N\lambda^{2(k-1)}}{N_k''}\Bigr)^{\frac{1}{k}N_k''}.
	\end{equation*}
\end{lemma}

\begin{proof}
\textcolor{black}{We may assume that $N_k''>0$ and that $k$ divides $N_k''$, since otherwise the integral is $1$ or $0$.} The argument is similar to that of the preceding lemma, so we record only the estimates. Partition the index set \([N_k'']\) into \(\tfrac{1}{k}N_k''\) blocks of size \(k\). The number of such partitions is
\[
\frac{(N_k'')!}{(k!)^{\tfrac{1}{k}N_k''}(\tfrac{1}{k}N_k'')!}\leq C^{N_k''}(N_k'')^{(1-\tfrac{1}{k}) N_k''}.
\]
Each \(2k\)-pole contributes, after integration, a factor \(N\,\lambda^{(4-\beta)k-2}\). Normalizing by \((N C_{\beta,\lambda,\varepsilon_0})^{N_k''}\) therefore yields
\[
\frac{\Zsf_{\bad,N_k''}^{(2,2,k)} }{(N C_{\beta,\lambda,\varepsilon_0})^{N_k''}}\le C^{N_k''}\textcolor{black}{\Bigl(\frac{N_k''}{N}\Bigr)^{(1-\frac1k)N_k''}}\,\lambda^{2(k-1)\frac{1}{k}N_k''},
\]
which is the desired bound.
\end{proof}

{\color{black}
For the $p(\beta)+3$ class counts $n_j$ with sum $m$ in \eqref{eq:non interacting} and \eqref{eq:part type2}, we have
\[
    \binom Nm\frac{m!}{\prod_jn_j!}
    \leq e^m\prod_j\Bigl(\frac N{n_j}\Bigr)^{n_j}.
\]
We absorb the counting factors by applying Lemma~\ref{lemma:type3} with $u=t+1+\log(3(p(\beta)+1))$ in place of $t$, and then the other bounds with $s=u+e^{Cu}$.
For $N_k''>0$, the threshold of type $(2,2)$ and Lemma~\ref{lemma:type22} give
\[
    e^{sN_k''}C^{N_k''}
    \Bigl(\frac{N\lambda^{2(k-1)}}{N_k''}\Bigr)^{N_k''/k}
    \leq\bigl(Ce^s\ve_0^{\alpha(\beta)/k}\bigr)^{N_k''}\leq1
\]
for $\ve_0$ small enough. The class $k=1$ is empty, since $\ve_0^{-\alpha(\beta)}N>N$. Lemmas~\ref{lemma:type1} and~\ref{lemma:type2} then give $e^{CN\delta_{\beta,\lambda}}$, which proves Proposition~\ref{prop:bad points}.
}

\subsection{Proof of the upper bound}

\begin{prop}[Upper bound]\label{prop:upper bound}
	Let $\beta\in (2,+\infty)$. \textcolor{black}{Let $p_0\in\mathbb N^*$ and $p(\beta)$ be as in Definition~\ref{def:pbeta}.}
	
	Let $\mc{N}_1,\ldots,\mc{N}_{p(\beta)}$ be the numbers of multipoles of cardinalities $1,\ldots,p(\beta)$ in $I_\good$.

	Let $n_1,\ldots,n_{p(\beta)}\in \mathbb{N}$ \textcolor{black}{satisfy}
\[
\textcolor{black}{\sum_{k=1}^{p(\beta)}kn_k\leq N.}
\]
Define
	\begin{equation*}
		\mc{A}\coloneqq \{\mc{N}_1=n_1,\ldots,\mc{N}_{p(\beta)}=n_{p(\beta)}\}. 
	\end{equation*}
	Let $\mathcal{I}_{\beta,p(\beta),\lambda}$ be as in Definition~\ref{def:rate function}. Suppose that, for every $k\in \{2,\ldots,p(\beta)\}$,  \eqref{bornesnkbs} holds.

	Then, \textcolor{black}{for every $t\geq1$ and for $\ve_0,\lambda$ small enough depending on $\beta,M,p(\beta),t$}, we have
	\begin{multline*}
		\log \int_{\mc{A}}e^{t|I_\bad|}e^{-\beta \F_\lambda}\dd \XN\dd \YN\leq N\log N+\log (N!)+N(2-\beta)\log\lambda+N\log \mc{Z}_\beta -N \mathcal{I}_{\beta,p(\beta),\lambda}(\tfrac{n_1}{N},\ldots,\tfrac{n_{p(\beta)}}{N})\\+O_{\beta,M,p,\ve_0,t}(N\delta_{\beta,\lambda}).
	\end{multline*}
\end{prop}

\medskip

\begin{proof}
	Let $p=p(\beta)$ \textcolor{black}{and $\gamma_i=n_i/N$ for $i\in[p]$.} Let $N'\coloneqq n_1+2n_2+\cdots+p n_p$ and let $V_\good\subset [N]$ be an arbitrary set of cardinality $N'$. Set $V_\bad\coloneqq [N]\setminus V_\good$. Let $\pi$ be a partition of $V_\good$ with $n_i$ parts of cardinality $i$ for every $i\in [p]$.
	Recall  Lemma~\ref{lemma:start upper}.
	\paragraph{\bf{Step 1: expanding the multipole terms}}
	Combining Proposition~\ref{prop:expansion +} and Lemma~\ref{lemma:limiting}, we obtain that, for every \textcolor{black}{$S\in\pi$},
	\begin{equation*}
		\textcolor{black}{\widehat{\Msf}_{\ve_0}^{+,Z}}(S)=\frac{1}{N^{|S|-1}}\Bigl(\msf_{\beta,\lambda}(|S|)+\textcolor{black}{O_{\beta,M,p,\ve_0}(\Cut^{-2})+O_{\beta,M,p}\Bigl(\ve_0^{\alpha(\beta)|S|}\frac{|V_\bad|}{N}\Bigr)}\Bigr).
	\end{equation*}
	Therefore, \textcolor{black}{since $\msf_{\beta,\lambda}(1)=1$ and $1+x\leq e^x$ for every real $x$, we have, keeping the singleton factors,}
	\begin{equation*}
		\prod_{S\in \pi}\textcolor{black}{\widehat{\Msf}_{\ve_0}^{+,Z}}(S)\textcolor{black}{\leq}\frac{N^{n_1+\cdots+n_{p(\beta)}}}{N^{N'}}\prod_{k=2}^{p}\msf_{\beta,\lambda}(k)^{n_k}\textcolor{black}{\prod_{k=1}^p}e^{\frac{n_k}{\msf_{\beta,\lambda}(k) }\textcolor{black}{\Bigl(O_{\beta,M,p,\ve_0}(\Cut^{-2})+O_{\beta,M,p}\Bigl(\ve_0^{\alpha(\beta)k}\frac{|V_\bad|}{N}\Bigr)\Bigr)} } .
	\end{equation*}
	Recall that, by Proposition~\ref{prop:bounded lower}, there exists a constant $C>0$ depending on $\beta$, $M$ and $p$ such that 
	\begin{equation*}
		\msf_{\beta,\lambda}(k)\geq \frac{1}{C}\lambda^{2(k-1)}.
	\end{equation*}
	{\color{black}
Using $n_1\leq N$ and \eqref{bornesnkbs}, we obtain
        \[
            \sum_{k=1}^p\frac{n_k}{\msf_{\beta,\lambda}(k)}\ve_0^{\alpha(\beta)k}\leq CN,
        \]
        with $C$ independent of $\ve_0$. Hence,
}
	\begin{equation*}
		\prod_{S\in \pi}\textcolor{black}{\widehat{\Msf}_{\ve_0}^{+,Z}}(S)\textcolor{black}{\leq}\frac{N^{n_1+\cdots+n_{p}}}{N^{N'}}\left(\prod_{k=2}^{p}\msf_{\beta,\lambda}(k)^{n_k}\right)e^{\textcolor{black}{O_{\beta,M,p,\ve_0}(N\Cut^{-2})+O_{\beta,M,p}(|V_\bad|)} } .
	\end{equation*}
	Inserting this into \eqref{numero539}, we deduce that 
	\begin{equation}\label{eq:A1A2}
		\frac{1}{(NC_{\beta,\lambda,\ve_0})^N} \int_{\mc{A}}\indic_{\sigma_N=\Id} e^{-\beta \F_\lambda}\leq T_1 T_2 e^{\textcolor{black}{O_{\beta,M,p,\ve_0}(N\Cut^{-2})+O_{\beta,M,p}(|V_\bad|)}},
	\end{equation}
	where 
	\begin{equation*}
		T_1\coloneqq \frac{(N')!}{1^{n_1}(2!)^{n_2}\cdots (p!)^{n_{p}}n_1!\cdots n_{p}!} \frac{N^{n_1+\cdots+n_{p}}}{N^{N'}}\prod_{k=2}^{p}\msf_{\beta,\lambda}(k)^{n_k}
	\end{equation*}
	and 
	\begin{multline}\label{def:A2part}
		T_2\coloneqq \binom{N}{N'}\frac{1}{(NC_{\beta,\lambda,\ve_0})^{|V_\bad|}}\int_{(\Lambda^2)^{|V_\bad|} } \prod_{i,j\in V_\bad:i<j}e^{-\beta v^Z_{ij}}\indic_{\mc{A}_{ij}}W(Z) \indic_{I_\bad=V_\bad}\\ \times \prod_{i\in V_\bad}\frac{1}{(\tau_i^{+,Z})^{\frac{\beta}{2}}}\frac{1}{(\tau_i^{-,Z})^{\frac{\beta}{2}}} \dd Z.
	\end{multline}
	By Stirling's formula,
	\begin{equation*}
		\log\left( \frac{(N')!}{n_1!\cdots n_{p}!} \frac{N^{n_1+\cdots+n_{p}}}{N^{N'}}\right)=\log\left(\Bigl(\frac{N'}{N}\Bigr)^{N'}e^{-N'}\prod_{k=1}^{p}\Bigl(\frac{N}{n_k} \Bigr)^{n_k}e^{n_k}\right)+O(\log N).
	\end{equation*}
	Since $N'=N-|V_\bad|$, it follows that
	\begin{equation}\label{eq:expansionA1}
		\begin{split}
			T_1&=\Bigl(\frac{N'}{N}\Bigr)^{N'}e^{-N'}\left(\prod_{k=1}^{p}\Bigl(\frac{N}{n_k k!} \Bigr)^{n_k}e^{n_k} \prod_{k=2}^p \msf_{\beta,\lambda}(k)^{n_k} \right)e^{O(\log N) }\\
			&=e^{-N}\left(\prod_{k=1}^{p}\Bigl(\frac{N}{n_k k!} \Bigr)^{n_k}e^{n_k} \prod_{k=2}^p \msf_{\beta,\lambda}(k)^{n_k} \right)e^{O(\log N+|V_\bad|)}.
		\end{split}
	\end{equation}
	
	\paragraph{\bf{Step 2: cluster expansion}}
	
	By Proposition~\ref{prop:absolute upper} and Lemma~\ref{lemma:resum2}, for $\ve_0$ small enough, we have
	\begin{equation*}
		\log W(Z)=\sum_{n=1}^{\infty}\frac{1}{n!}\sum_{\substack{ X_1,\ldots,X_n\in \mc{P}(\pi)\\ \mathrm{connected} }}\left(\prod_{i=1}^n\indic_{|V_{X_i}|\leq p} \textcolor{black}{\widehat{\Ksf}_{\ve_0}^{+,Z}}(X_i)\right)\mathrm{I}(G(X_1,\ldots,X_n))+\textcolor{black}{O_{\beta,M,p,\ve_0}(N\delta_{\beta,\lambda})}.
	\end{equation*}
	{\color{black}
For the retained activities, set
    \[
        D=\sum_{\substack{X\subset\pi\\|V_X|\leq p}}
        |\widehat{\Ksf}_{\ve_0}^{+,Z}(X)-\Ksf_{\ve_0}^0(X)|.
    \]
    Summed over the sets $X$ with $r$ blocks and $|V_X|=k$, counted with \eqref{bornesnkbs}, the term of Proposition~\ref{prop:expansion +} proportional to $|V_\bad|$ is at most $C|V_\bad|\ve_0^{\alpha(\beta)(k-r)}$. Summed in the same way, the term with $C_{\ve_0}\Cut^{-2}$ is at most $C_{\ve_0}N\Cut^{-2}\leq C_{\ve_0}N\delta_{\beta,\lambda}$. Since $r\leq k\leq p$, we obtain $D\leq C_{\ve_0}N\delta_{\beta,\lambda}+C|V_\bad|$.
    By Propositions~\ref{prop:bounded upper} and~\ref{prop:bounded lower}, together with \eqref{bornesnkbs}, the retained activities satisfy
    \[
        q=\max_{S\in\pi}\sum_{\substack{X\subset\pi:S\in X\\|V_X|\leq p}}
        \bigl(|\widehat{\Ksf}_{\ve_0}^{+,Z}(X)|+|\Ksf_{\ve_0}^0(X)|\bigr)
        \leq C_{\ve_0}\lambda^2.
    \]
    We telescope the products and root each spanning tree at the factor that changes. By \eqref{eq:rota}, summing over the leaves one at a time bounds the difference of the two retained series by
    \[
        D\sum_{n\geq1}\frac{n^{n-1}}{n!}(pq)^{n-1}\leq CD,
    \]
    once $epq\leq1/2$, with $C$ independent of $\ve_0$.
}

	\textcolor{black}{Combining these estimates gives} 
	\begin{multline*}
		\log W(Z)= \sum_{n=1}^{\infty}\frac{1}{n!}\sum_{\substack{ X_1,\ldots,X_n\in \mc{P}(\pi)\\ \mathrm{connected} }}\left(\prod_{i=1}^n\indic_{|V_{X_i}|\leq p} \Ksf_{\ve_0}^{0}(X_i)\right)\mathrm{I}(G(X_1,\ldots,X_n))\\
        +\textcolor{black}{O_{\beta,M,p,\ve_0}(N\delta_{\beta,\lambda})+O_{\beta,M,p}(|V_\bad|)}.
	\end{multline*}
	
\begingroup\color{black}
We define $D_0$ and $q_0$ as above for the activities $\Ksf_{\ve_0}^0$ and $\Ksf_\infty^0$. By \eqref{eq:K-diff0} of Proposition~\ref{prop:expansion -} and the same counting, $D_0\leq C_{\ve_0}N\Cut^{-2}$ and $q_0\leq C_{\ve_0}\lambda^2$. With the same telescoping bound, we therefore replace $\Ksf_{\ve_0}^0$ by $\Ksf_\infty^0$, with an error at most $C_{\ve_0}N\Cut^{-2}\leq C_{\ve_0}N\delta_{\beta,\lambda}$.
\endgroup
 By Corollary~\ref{coro:equality mult}, $\Ksf_{\infty}^{0}$ can be replaced by $\Ksf_{\beta,\lambda}^{\mult}$. Hence, applying \eqref{eq:asin}, we then get that, for $\ve_0$ small enough,
	\begin{equation}\label{eq:tWZ}
		\log W(Z)=S_N+\textcolor{black}{O_{\beta,M,p,\ve_0}(N\delta_{\beta,\lambda})+O_{\beta,M,p}(|V_\bad|)} +o(N),
	\end{equation}
	where
	\begin{equation*}
		S_N=-N\sum_{n=1}^{\infty}\frac{(-1)^n}{n!}\sum_{m_1,\ldots,m_p\in \mathbb{N} }\prod_{i=1}^p \frac{\gamma_i^{m_i}}{m_i!}\sum_{\substack{(X_1,\ldots,X_n)\\\in \Htrees_n(Y(m_1,\ldots,m_p))\\ X_1\cup \cdots \cup X_n=Y(m_1,\ldots,m_p)}}\textcolor{black}{\omega(\mathbf X)}\prod_{i=1}^n \ksf_{\beta,\lambda}^\mult((\#_kX_i)_k)\indic_{|V_{X_i}|\leq p},
	\end{equation*}
	where we recall that $Y(m_1,\ldots,m_p)$ is as in Definition~\ref{def:canonical}. Inserting this into \eqref{def:A2part} and using \eqref{eq:non interacting} gives 
	\begin{multline}\label{eq:expansionA2}
		T_2= e^{S_N+\textcolor{black}{O_{\beta,M,p,\ve_0}(N\delta_{\beta,\lambda})+O_{\beta,M,p}(|V_\bad|)} }\\ \times \binom{N}{N'}\frac{1}{(NC_{\beta,\lambda,\ve_0})^{|V_\bad|}}\int_{(\Lambda^2)^{|V_\bad|} } \prod_{i,j\in V_\bad:i<j}\indic_{\mc{A}_{ij}} \indic_{I_\bad=V_\bad}e^{-\beta v_{ij}^Z}\prod_{i\in V_\bad}\frac{1}{(\tau_i^{+,Z})^{\frac{\beta}{2}}}\frac{1}{(\tau_i^{-,Z})^{\frac{\beta}{2}}} \dd x_i \dd y_i.
	\end{multline}
	
	\paragraph{\bf{Step 3: conclusion}}
	Combining \eqref{eq:A1A2}, \eqref{eq:expansionA1}, \eqref{eq:expansionA2} and Definition~\ref{def:rate function}, we obtain
	\begin{multline*}
		\frac{1}{(NC_{\beta,\lambda,\ve_0})^N} \int_{\mc{A}}\indic_{\sigma_N=\Id} e^{-\beta \F_\lambda}\leq e^{-N\mathcal{I}_{\beta,p,\lambda}(\frac{n_1}{N},\ldots,\frac{n_p}{N})+\textcolor{black}{O_{\beta,M,p,\ve_0}(N\delta_{\beta,\lambda})+O_{\beta,M,p}(|V_\bad|)} } \\ \times \binom{N}{N'}\frac{1}{(NC_{\beta,\lambda,\ve_0})^{|V_\bad|}}\int_{(\Lambda^2)^{|V_\bad|} } \prod_{i,j\in V_\bad:i<j}\indic_{\mc{A}_{ij}} \indic_{I_\bad=V_\bad}e^{-\beta v_{ij}^Z}\prod_{i\in V_\bad}\frac{1}{(\tau_i^{+,Z})^{\frac{\beta}{2}}}\frac{1}{(\tau_i^{-,Z})^{\frac{\beta}{2}}} \dd x_i \dd y_i.
	\end{multline*}
	\textcolor{black}{Hence, for every $t\geq1$, there exists $C>0$ depending only on $\beta,M,p$ such that}
	\begin{multline*}
		\frac{1}{(NC_{\beta,\lambda,\ve_0})^N} \int_{\mc{A}}\indic_{\sigma_N=\Id}e^{t|I_\bad|} e^{-\beta \F_\lambda}\leq e^{-N\mc{I}_{\beta,p,\lambda}(\frac{n_1}{N},\ldots,\frac{n_p}{N})}e^{O_{\beta,M,p,\ve_0} (N\delta_{\beta,\lambda}) } \\ \times \binom{N}{N'}\frac{e^{(t+C)|V_\bad|}}{(NC_{\beta,\lambda,\ve_0})^{|V_\bad|}}\int_{(\Lambda^2)^{|V_\bad|} } \prod_{i,j\in V_\bad:i<j}\indic_{\mc{A}_{ij}} \indic_{I_\bad=V_\bad}e^{-\beta v_{ij}^Z}\prod_{i\in V_\bad}\frac{1}{(\tau_i^{+,Z})^{\frac{\beta}{2}}}\frac{1}{(\tau_i^{-,Z})^{\frac{\beta}{2}}} \dd x_i \dd y_i.
	\end{multline*}
	\textcolor{black}{We apply Proposition~\ref{prop:bad points} with $t+C$ in place of $t$, with $\ve_0$ as in this proposition, and with $\lambda$ small enough for the preceding estimates. In view of \eqref{eq:580} and \eqref{def:Z bad}, we obtain} 
	\begin{equation}\label{eq:onlyp}
		\frac{1}{(NC_{\beta,\lambda,\ve_0})^N} \int_{\mc{A}}\indic_{\sigma_N=\Id}e^{t|I_\bad|} e^{-\beta \F_\lambda}\leq e^{-N\mathcal{I}_{\beta,p,\lambda}(\frac{n_1}{N},\ldots,\frac{n_p}{N})+O_{\beta,M,p,\ve_0,t}\left(N\delta_{\beta,\lambda} \right)} .
	\end{equation}
\begingroup\color{black}
By the relabeling identity \eqref{eq:relabel}, removing $\{\sigma_N=\Id\}$ multiplies the integral by $N!$. Moreover, \eqref{devtCbeta} gives
\[
N\log C_{\beta,\lambda,\ve_0}
=N\log(\lambda^{2-\beta}\mc Z_\beta)+O_{\beta,\ve_0}(N\Cut^{-2}).
\]
Taking logarithms proves the proposition.
\endgroup
	
\end{proof}

\subsection{Proof of the multipole distribution theorem}

\begin{proof}[Proof of Theorem~\ref{theorem:LDP}]
	Let $p\coloneqq p(\beta)$. Recall the notation \eqref{defdelta}.
	Denote by $\mc{N}_1^\good,\ldots,\mc{N}_p^\good$ the numbers of multipoles of cardinalities $1,\ldots,p$ in $I_\good$. \textcolor{black}{We have $|I_\good|=\sum_{i=1}^p i\mc N_i^\good$. By Remark~\ref{remarkgoodbad}(2), profiles that violate \eqref{bornesnkbs} have probability zero.} 
	The proof of item~\eqref{eq:ksf thm} is straightforward. By Lemma~\ref{lemma:limiting}, under the assumption of \eqref{eq:ksf thm}, the limits $\msf_{\beta,\lambda}(k)$ and $\ksf_{\beta,\lambda}^{\mult}(\{ \#_m \textcolor{black}{X}\}_m)$ exist. Moreover, inserting the estimates of Proposition~\ref{prop:bounded lower} to control the numerator and the denominator of the quotients defining both limiting activities yields \eqref{eq:msf thm} and \eqref{eq:ksf thm}.
	
	\paragraph{\bf{Step 1: control of the number of bad points}}	
	Let us write
	\begin{equation*}
		\P(|I_\bad|=k)\leq \sum_{\substack{n_1,\ldots,n_p:\\ n_1+2n_2+\cdots+pn_p=N-k}}\P(|I_\bad|=k,\forall i\in [p], \mc{N}_i^\good=n_i).
	\end{equation*}
	Hence,
	\begin{equation}\label{eq:max ni}
		\P(|I_\bad|=k)\leq N^p\max_{\substack{n_1,\ldots,n_p:\\ n_1+2n_2+\cdots+pn_p=N-k}}\P(|I_\bad|=k,\forall i\in [p], \mc{N}_i^\good=n_i).
	\end{equation}
	By Markov's inequality,
	\begin{flalign}
		&& \P(|I_\bad|=k,\forall i\in [p], \mc{N}_i^\good=n_i) \leq e^{-k}\dE_{\P}\left[e^{|I_\bad|}\prod_{i=1}^p\indic_{\mc{N}_i^\good=n_i}\right] \notag\\
		&& =e^{-k}\frac{1}{Z_{N,\beta}^\lambda}\int e^{|I_\bad|}e^{-\beta \F_\lambda}\prod_{i=1}^p\indic_{\mc{N}_i^\good=n_i}\dd \vec{X}_N\dd \vec{Y}_N.\label{eq:Markov}
	\end{flalign}
	
\begingroup\color{black}
Write $\tilde m_i=\msf_{\beta,\lambda}(i)/i!$ and let $h$ be the Mayer series in $\mathcal I_{\beta,p,\lambda}$. If $p=1$, then $h=0$. If $p\geq2$, we write $m=\sum_i m_i$. Every nonzero term of $h$ has $n<m$ and
\[
\sum_{l=1}^n(|X_l|-1)=m-1,
\]
since activities with $|X_l|\leq1$ vanish and the incidence graph is a tree. By \eqref{eq:ksf thm}, the product of its activities is at most $C^n\lambda^{2(m-1)}$. For $|\gamma_i|\leq2$, \eqref{eq:hypertree-count} therefore bounds the sum of the absolute values of the terms of degree $m$ by
\[
\frac{(2p)^m\lambda^{2(m-1)}}{m!}
\sum_{n=1}^{m-1}\frac{C^n n^m(m+n)^n}{n!}
\leq C_1^m\lambda^{2(m-1)}.
\]
Here $n^m/m!\leq e^m$ and $(m+n)^n/n!\leq e^{m+n}$, and the constants depend only on $\beta,M,p$. First and second derivatives cost at most $m^2$, so the series and both derivatives converge uniformly on $[-2,2]^p$. In particular,
\[
\|h\|_\infty+\|\nabla h\|_\infty
\leq\sum_{m\geq2}m^2C_1^m\lambda^{2(m-1)}\leq C\lambda^2.
\]
The rate function has a minimizer $\gamma^*$ on the compact simplex $\triangle_{p,0}$. If $\gamma_j^*=0$ and $\gamma_l^*>0$, then replacing $\gamma^*$ by $\gamma^*+te_j-(j/l)te_l$ changes the rate by $t\log t+O_\lambda(t)<0$ for small $t>0$. Thus all coordinates are positive, and stationarity gives, for a Lagrange multiplier $\mu$,
\[
\gamma_k^*=\tilde m_k e^{-k\mu+\alpha_k},\qquad
\alpha_k=-\partial_kh(\gamma^*),\qquad |\alpha_k|\leq C\lambda^2.
\]
Since $\gamma_1^*\leq1$ and $\tilde m_1=1$, we have $e^{-\mu}\leq e^{C\lambda^2}$. By \eqref{eq:msf thm} and $\sum_k k\gamma_k^*=1$,
\[
\sum_{k\geq2}k\gamma_k^*=O(\lambda^2),\qquad
\gamma_1^*=1-O(\lambda^2),\qquad \mu=O(\lambda^2).
\]
In particular, $C^{-1}\lambda^{2(k-1)}\leq\gamma_k^*\leq C\lambda^{2(k-1)}$. Let $\ve_0$ be small enough for the preceding estimates, with $\ve_0^{-\alpha(\beta)}\geq2C$, and set
\[
n_i^*=\lfloor N\gamma_i^*\rfloor\quad(i\geq2),\qquad n_1^*=N-\sum_{i=2}^p in_i^*.
\]
This profile satisfies \eqref{eq:assnk}. We compare the entropy terms directly, and the values of $h$ with the bound on $\|\nabla h\|_\infty$ above. This gives
\[
N|\mathcal I_{\beta,p,\lambda}(n^*/N)-\mathcal I_{\beta,p,\lambda}(\gamma^*)|
\leq C(\log N+|\log\lambda|).
\]
Applying Proposition~\ref{prop:lower bound} to this profile and absorbing this error for $N$ large enough at fixed $\lambda$ gives
\endgroup
	\begin{equation}\label{eq:expansion Z lower}
		\log Z_{N,\beta}^\lambda\geq \log N!+N\log N+N\log(\lambda^{2-\beta}\mc{Z}_\beta)-N\inf_{\triangle_{p,0}}  \mathcal{I}_{\beta,p,\lambda}+O_{\beta,M,p,\ve_0}(N\delta_{\beta,\lambda}).
	\end{equation}
	Moreover, by Proposition~\ref{prop:upper bound}, 
	\begin{multline*}
		\log	\int e^{|I_\bad|}e^{-\beta \F_\lambda}\prod_{i=1}^p\indic_{\mc{N}_i^\good=n_i}\dd \vec{X}_N\dd \vec{Y}_N\textcolor{black}{\leq}\log N!+N\log N+N\log(\lambda^{2-\beta}\mc{Z}_\beta)-N\mathcal{I}_{\beta,p,\lambda}(\tfrac{n_1}{N},\ldots,\tfrac{n_p}{N})\\+O_{\beta,M,p,\ve_0}(N\delta_{\beta,\lambda}).
	\end{multline*}
	Inserting this into \eqref{eq:Markov} and using \eqref{eq:max ni}, we deduce that 
	\begin{equation}\label{eq:boundPPPProb}
		\log \P(|I_\bad|=k)\leq p\log N-k-N \left(\inf_{\triangle_{p,\frac{k}{N}}} \mathcal{I}_{\beta,p,\lambda}-\inf_{\triangle_{p,0}}\mathcal{I}_{\beta,p,\lambda} \right)+O_{\beta,M,p,\ve_0}(N\delta_{\beta,\lambda}).
	\end{equation}

	Let us study how $\inf_{\triangle_{p,x}}  \mathcal{I}_{\beta,p,\lambda}$ depends on $x$. Recall that $\tilde m_i=\msf_{\beta,\lambda}(i)/i!$, and that $h$ is the Mayer series in $\mc I_{\beta,p,\lambda}$. Thus $\mc{I}_{\beta,p,\lambda}=f+h$, with
	\begin{equation*}
		f:(\gamma_1,\ldots,\gamma_p)\in \cup_{x\in [0,1]}\triangle_{p,x}\mapsto\sum_{i=1}^p \gamma_i(\log \gamma_i-\log \tilde{m}_i) +1-\sum_{i=1}^p \gamma_i
	\end{equation*}
	and 
	\begin{multline*}
		h:(\gamma_1,\ldots,\gamma_p)\in \cup_{x\in [0,1]}\triangle_{p,x}\\ \mapsto \sum_{n=1}^{\infty}\frac{(-1)^n}{n!}\sum_{m_1,\ldots,m_p\in \mathbb{N} }\prod_{i=1}^p \frac{\gamma_i^{m_i}}{m_i!}\sum_{\substack{(X_1,\ldots,X_n)\\\in \Htrees_n(Y(m_1,\ldots,m_p))\\ X_1\cup \cdots \cup X_n=Y(m_1,\ldots,m_p)}}\textcolor{black}{\omega(\mathbf X)}\prod_{i=1}^n \ksf_{\beta,\lambda}^\mult((\#_kX_i)_k)\indic_{|V_{X_i}|\leq p}. 
	\end{multline*}
	\begingroup\color{black}
The bounds above on the terms of degree $m$ give
\begin{equation}\label{eq:h derivatives}
\sup_{\gamma\in[0,1]^p}\bigl(|h(\gamma)|+|\nabla h(\gamma)|+\|\nabla^2h(\gamma)\|\bigr)\leq C\lambda^2,
\end{equation}
where $C$ depends on $\beta,M,p$. For $x\in[0,1]$ and $\gamma\in\triangle_{p,x}$, set $\eta=(\gamma_1+x,\gamma_2,\ldots,\gamma_p)\in\triangle_{p,0}$. Since $\tilde{m}_1=1$ and $\gamma_1+x\le1$, \eqref{eq:h derivatives} gives
\[
\mathcal I_{\beta,p,\lambda}(\eta)-\mathcal I_{\beta,p,\lambda}(\gamma)
=\int_{\gamma_1}^{\gamma_1+x}\log t\dd t+h(\eta)-h(\gamma)\le Cx\lambda^2.
\]
Taking the infimum over $\gamma\in\triangle_{p,x}$ yields
\begin{equation}\label{eq:comparison inf}
\inf_{\triangle_{p,x}}\mathcal I_{\beta,p,\lambda}
\geq\inf_{\triangle_{p,0}}\mathcal I_{\beta,p,\lambda}-Cx\lambda^2.
\end{equation}
\endgroup
	
	Inserting \eqref{eq:comparison inf} into \eqref{eq:boundPPPProb} gives
	\begin{equation*}
		\log \P(|I_\bad|=k )\leq -k(1+O_{\beta,M,p}(\lambda^2)) +O_{\beta,M,p,\ve_0}(N\delta_{\beta,\lambda}).
	\end{equation*}
	Hence,
	\begin{equation}\label{eq:probaIbad}
		\log \P(|I_\bad|=k )\leq -\frac{k}{2} +O_{\beta,M,p,\ve_0}(N\delta_{\beta,\lambda}).
	\end{equation}
	Since $0\leq N-(\mc{N}_1+2\mc{N}_2+\cdots+p\mc{N}_p)\leq |I_\bad|$ and $|\{i\in [N]: |x_i-y_{\sigma_N(i)}|\geq R_{\beta,\lambda}\}|\leq |I_\bad|$, this concludes the proof of \eqref{eq:thm bad} and \eqref{eq:thm bad2}.
	\begingroup\color{black}
	Moreover, for $i\in\{2,\ldots,p\}$, Definition~\ref{def:type2} and Remark~\ref{remarkgoodbad} give
	\[
		i\mc N_i^\good\leq\ve_0^{-\alpha(\beta)}N\lambda^{2(i-1)},
		\qquad 0\leq i(\mc N_i-\mc N_i^\good)\leq |I_\bad|.
	\]
	Hence $\mc N_i\indic_{\mc N_i>\ve_0^{-\alpha(\beta)}N\lambda^{2(i-1)}}\leq |I_\bad|$. Summing the bound before \eqref{eq:probaIbad} gives
	\begin{equation}\label{eq:tail estimate}
		\P(\mc{N}_i>\ve_0^{-\alpha(\beta)}N\lambda^{2(i-1)})
		\leq e^{-N\ve_0^{-\alpha(\beta)}\lambda^{2(i-1)}(1+O_{\beta,M,p,\ve_0}(\lambda^2))+O_{\beta,M,p,\ve_0}(N\delta_{\beta,\lambda})}.
	\end{equation}
	\endgroup
	
	\paragraph{\bf{Step 2: expansion of the partition function}}
	We write 
	\begin{equation*}
		Z_{N,\beta}^\lambda=\sum_{k=0}^N \sum_{\substack{n_1,\ldots,n_p:\\ n_1+2n_2+\cdots+pn_p=N-k}} \int e^{-\beta \F_\lambda}\indic_{|I_\bad|=k}\prod_{i=1}^p \indic_{\mc{N}_i^\good=n_i}\dd \vec{X}_N \dd \vec{Y}_N.
	\end{equation*}
	Proceeding as in Step 1, we get 
	\begin{multline*}
		\log Z_{N,\beta}^\lambda\leq (p+1)\log N+\log N!+N\log N+N\log (\lambda^{2-\beta }\mc{Z}_\beta)+ \max_{x\in [0,1]}(-N x+O_{\beta,M,p,\ve_0}(Nx\lambda^2))\\-N \inf_{\triangle_{p,0}}\mathcal{I}_{\beta,p,\lambda}+O_{\beta,M,p,\ve_0}(N\delta_{\beta,\lambda}).
	\end{multline*}
	Thus, 
	\begin{equation*}
		\log Z_{N,\beta}^\lambda\leq \log N!+N\log N+N\log (\lambda^{2-\beta }\mc{Z}_\beta)-N \inf_{\triangle_{p,0}}\mathcal{I}_{\beta,p,\lambda}+O_{\beta,M,p,\ve_0}(N\delta_{\beta,\lambda}).
	\end{equation*}
	Combining this with \eqref{eq:expansion Z lower} gives  equality, i.e., proves \eqref{eq:expansion Z}.
	
	\paragraph{\bf{Step 3: proof of the large-deviation bounds}}
	
	For every $\gamma\in \triangle_{p,0}$ and $r>0$, let 
	\begin{equation*}
		\bar{B}_{\infty}(\gamma,r)\coloneqq \{(\gamma'_1,\ldots,\gamma'_p)\in (\dR_+)^p:\forall i\in [p], |\gamma_i'-\gamma_i|\leq r\}.
	\end{equation*}
	\begingroup\color{black}
Let $\gamma\in\triangle_{p,0}$ satisfy $\gamma_i\leq C_0\lambda^{2(i-1)}$ for $2\leq i\leq p$. Set $\rho=C_0\delta_{\beta,\lambda}$ and
\[
J_\gamma=\inf_{\bar B_\infty(\gamma,\rho)\cap\triangle_{p,0}}\mathcal I_{\beta,p,\lambda}.
\]
Since $\delta_{\beta,\lambda}=o(\lambda^{2(p-1)})$, we have $\gamma_i+\rho\leq2C_0\lambda^{2(i-1)}$ for $i\geq2$ and $\lambda$ small enough. By \eqref{eq:msf thm} and \eqref{eq:h derivatives}, for $0\leq u\leq\eta\in\bar B_\infty(\gamma,\rho)\cap\triangle_{p,0}$,
\[
\mathcal I_{\beta,p,\lambda}(\eta)-\mathcal I_{\beta,p,\lambda}(u)
\leq C\sum_{i=1}^p(\eta_i-u_i),
\]
since $\log(s/\tilde m_i)\leq C$ for $0<s\leq\gamma_i+\rho$.

For the upper bound, on the event in \eqref{eq:LD bounds}, write $z_i=\mc N_i/N$, $u_i=\mc N_i^\good/N$ and $x=|I_\bad|/N$. Then $u\leq z$, $\sum_i iz_i\leq1$ and $\sum_i iu_i=1-x$. The segment joining $z$ to $(\max(z_i,\gamma_i))_{i=1}^p$ contains $\eta\in\bar B_\infty(\gamma,\rho)\cap\triangle_{p,0}$ with $\eta\geq u$. Hence,
\[
\mathcal I_{\beta,p,\lambda}(u)\geq J_\gamma-Cx.
\]
Fix $t>C+1$, and let $\ve_0$ be small enough for Proposition~\ref{prop:upper bound}, with $\ve_0^{-\alpha(\beta)}\geq2C_0$. Summing that estimate over the at most $(N+1)^p$ good profiles, with the factor $e^{-t|I_\bad|}$, and using \eqref{eq:expansion Z}, we obtain
\begin{equation}\label{eq:probaupper}
\log\P\left(\max_{i\in[p]}|\mc N_i-N\gamma_i|\leq N\rho\right)
\leq-N\left(J_\gamma-\inf_{\triangle_{p,0}}\mathcal I_{\beta,p,\lambda}\right)
+O_{\beta,M,p,C_0,\ve_0}(N\delta_{\beta,\lambda}),
\end{equation}
where $O(\log N)$ is absorbed for $N$ large enough.

For the lower bound, let $\eta$ attain $J_\gamma$. Set $D_p=1+\sum_{i=2}^p i$ and, for $N\rho\geq D_p$,
\[
\eta'=\left(1-\frac{D_p}{N\rho}\right)\eta+\frac{D_p}{N\rho}\gamma,
\qquad n_i=\lfloor N\eta'_i\rfloor\quad(i\geq2),
\qquad n_1=N-\sum_{i=2}^p in_i.
\]
Then $n/N\in\bar B_\infty(\gamma,\rho)\cap\triangle_{p,0}$ and $\|n/N-\eta\|_\infty\leq2D_p/N$. We compare the entropy terms directly, and the values of $h$ with \eqref{eq:h derivatives}. This gives
\[
N\left|\mathcal I_{\beta,p,\lambda}(n/N)-J_\gamma\right|
\leq C(\log N+|\log\lambda|).
\]
We apply Proposition~\ref{prop:lower bound} to this profile and use \eqref{eq:expansion Z}. This gives
\[
\log\P\left(\max_{i\in[p]}|\mc N_i-N\gamma_i|\leq N\rho\right)
\geq-N\left(J_\gamma-\inf_{\triangle_{p,0}}\mathcal I_{\beta,p,\lambda}\right)
+O_{\beta,M,p,C_0,\ve_0}(N\delta_{\beta,\lambda}),
\]
for $N$ large enough at fixed $\lambda$.
\endgroup
	Combining this with \eqref{eq:probaupper} concludes the proof of \eqref{eq:LD bounds}.
	
	\paragraph{\bf{Step 4: study of the minimizer}}

	Let us next study the rate function $\mathcal{I}_{\beta,p,\lambda}$
	from  Definition~\ref{def:rate function}. 
	\begingroup\color{black}By \eqref{eq:h derivatives}, for every $i,j\in[p]$, at points of its domain with positive coordinates,\endgroup 
	\begin{equation*}
		\Bigr|\partial^2_{\gamma_i\gamma_j}\mathcal{I}_{\beta,p,\lambda}(\gamma_1,\ldots,\gamma_p)-\frac{\delta_{i,j}}{\gamma_i}\Bigr|\leq C\lambda^2.
	\end{equation*}
	\textcolor{black}{Thus, for $\lambda$ small enough, at every point with positive coordinates,} 
	\begin{equation}\label{eq:strict convexity Ibb}
		\nabla^2 \mc{I}_{\beta,p,\lambda}\geq \frac{1}{2}\mathrm{diag}\Bigl(\frac{1}{\gamma_1},\ldots,\frac{1}{\gamma_p}\Bigr).
	\end{equation}
	\textcolor{black}{Hence $\mathcal I_{\beta,p,\lambda}$ is strictly convex there. Since every minimizer has positive coordinates by Step 1, the minimizer $\gamma^*$ is unique.} \textcolor{black}{Step 1 also shows that it satisfies $C^{-1}\lambda^{2(i-1)}\leq\gamma_i^*\leq C\lambda^{2(i-1)}$ for $i\in\{2,\ldots,p\}$.}
	This proves \eqref{eq:minimizer}.
	
	\paragraph{\bf{Step 5: proof of \eqref{eq:multipoles number}}}
	\leavevmode\begingroup\color{black}
	Fix $i\in\{2,\ldots,p\}$ and $\ve_0$ small enough depending on $\beta,M,p$ for Proposition~\ref{prop:upper bound} with $t=1$.
	For a good profile $n=(n_1,\ldots,n_p)$, set
	\[
		u=n/N,\qquad b=N-\sum_{j=1}^p jn_j,
		\qquad \eta=(u_1+b/N,u_2,\ldots,u_p)\in\triangle_{p,0}.
	\]
	As in Step 1, $\mathcal I_{\beta,p,\lambda}(u)\geq\mathcal I_{\beta,p,\lambda}(\eta)-C\lambda^2b/N$.
	Proposition~\ref{prop:upper bound} and \eqref{eq:expansion Z lower} thus give, for $\lambda$ small enough,
	\[
		\log\P(\forall j\in[p],\ \mc N_j^\good=n_j)
		\leq -N\bigl(\mathcal I_{\beta,p,\lambda}(\eta)-\mathcal I_{\beta,p,\lambda}(\gamma^*)\bigr)
		-\frac b2+O_{\beta,M,p,\ve_0}(N\delta_{\beta,\lambda}).
	\]
	Since $\gamma^*$ minimizes $\mathcal I_{\beta,p,\lambda}$ on $\triangle_{p,0}$, we integrate \eqref{eq:strict convexity Ibb} along the segment from $\gamma^*$ to $(1-s)\gamma^*+s\eta$ and let $s\uparrow1$. Continuity gives
	\[
		\mathcal I_{\beta,p,\lambda}(\eta)-\mathcal I_{\beta,p,\lambda}(\gamma^*)
		\geq\frac{(\eta_i-\gamma_i^*)^2}{4\max(\eta_i,\gamma_i^*)}.
	\]
	On the event $\{\mc N_i\leq NC_0^{-1}\lambda^{2(i-1)}\}$, we have $\eta_i=u_i\leq C_0^{-1}\lambda^{2(i-1)}$. By \eqref{eq:minimizer}, the last display is then at least $C\lambda^{2(i-1)}$ for $C_0$ large enough. Summing over at most $(N+1)^p$ profiles and absorbing $O(\log N)$ for $N$ large enough at fixed $\lambda$ yields
	\begin{equation}\label{eq:lower bound number}
		\log\P\left(\mc N_i\leq\frac{N}{C_0}\lambda^{2(i-1)}\right)
		\leq -CN\lambda^{2(i-1)}+O_{\beta,M,p,\ve_0}(N\delta_{\beta,\lambda}).
	\end{equation}
	For $C_0>\ve_0^{-\alpha(\beta)}$, \eqref{eq:tail estimate} gives the upper tail at the same rate. Together with \eqref{eq:lower bound number}, this proves \eqref{eq:multipoles number}.
	\endgroup
\end{proof}

\subsection{Reduction to small clusters in the multipole expansion}

We now turn to the proof of Theorem~\ref{theorem:expansion}. 

\begingroup\color{black}
Our aim is to prove that
\begin{multline}\label{eq:goal}
\log\frac{1}{(NC_{\beta,\lambda})^N}
\int_{\sigma_N=\Id}e^{-\beta\F_\lambda(\vec X_N,\vec Y_N)}\dd\vec X_N\dd\vec Y_N
=\log\sum_{n=0}^\infty\frac1{n!}
\sum_{\substack{V_1,\ldots,V_n\subset[N]\ \mathrm{disjoint}\\
\forall i,\ 2\leq|V_i|\leq p(\beta)}}
\prod_{i=1}^n\Ksf_{\beta,\lambda}^\dip(V_i)
\\+O(N\delta_{\beta,\lambda}).
\end{multline}
We first truncate the multipole activities at total size $p(\beta)$ and then sum over partitions.

\begin{definition}
Let $p\in\mathbb N^*$. We denote by $\mathbf\Pi_p([N])$ the set of partitions of $[N]$ whose blocks have cardinality at most $p$.
\end{definition}

For $p=p(\beta)$ and every $\pi\in\mathbf\Pi_p([N])$, define
\begin{equation}\label{def:first Zpi}
\widehat Z(\pi)\coloneqq (NC_{\beta,\lambda})^N\Msf_\infty^0(\pi)
\sum_{n=0}^\infty\frac1{n!}
\sum_{\substack{X_1,\ldots,X_n\subset\pi\ \mathrm{disjoint}\\
\forall i,\ 2\leq|X_i|,\ |V_{X_i}|\leq p}}
\prod_{i=1}^n\Ksf_{\beta,\lambda}^\mult(X_i).
\end{equation}

\begin{prop}\label{prop:rewrite1}
Let $\beta\in(2,\infty)$, let $p=p(\beta)$ be as in Definition~\ref{def:pbeta}, and let $\widehat Z$ be as in \eqref{def:first Zpi}. For $M$ large enough depending on $\beta$ and $p$, for $\lambda$ small enough, and for $N$ large enough depending on $\lambda$, we have
\begin{equation}\label{eq:reduction p}
\log\int_{\sigma_N=\Id}e^{-\beta\F_\lambda(\vec X_N,\vec Y_N)}\dd\vec X_N\dd\vec Y_N
=\log\sum_{\pi\in\mathbf\Pi_p([N])}\widehat Z(\pi)
+O_{\beta,p,M}(N\delta_{\beta,\lambda}).
\end{equation}
\end{prop}

\medskip
\begin{proof}
By \eqref{eq:expansion Z}, \eqref{eq:relabel} and \eqref{limCbl},
\begin{equation}\label{eq:first start} \log\int_{\sigma_N=\Id}e^{-\beta\F_\lambda(\vec X_N,\vec Y_N)}\dd\vec X_N\dd\vec Y_N =N\log(NC_{\beta,\lambda})-N\inf_{\triangle_{p,0}}\mc I_{\beta,p,\lambda} +O_{\beta,p,M}(N\delta_{\beta,\lambda}).
\end{equation}
For $p=1$, the sum has one term $(NC_{\beta,\lambda})^N$ and $\inf_{\triangle_{1,0}}\mc I_{\beta,1,\lambda}=0$. We may thus assume $p\geq2$.
For $2\leq|X|$ and $|V_X|\leq p$, Lemma~\ref{lemma:limiting} and \eqref{eq:ksf thm} give, for $N$ large enough,
\[
|\Ksf_{\beta,\lambda}^\mult(X)|\leq C(\lambda^2/N)^{|X|-1}.
\]
There are finitely many block profiles with $|V_X|\leq p$, so this bound is uniform in $X$ and $\pi$. By Lemmas~\ref{lemma:bounded cluster tail} and~\ref{lemma:resum2}, the sum over $n$ in \eqref{def:first Zpi} is the exponential of an absolutely convergent connected series, hence positive. Since $\Msf_\infty^0(\pi)>0$, we obtain $\widehat Z(\pi)>0$.

Fix a partition $\pi$ with $n_k$ blocks of cardinality $k$, and set $\gamma_k=n_k/N$. By \eqref{eq:choicespi}, \eqref{def:little m}, Stirling's formula and the hypertree calculation \eqref{eq:asin},
\begin{equation*}\log\left(\frac{N!}{\prod_{k=1}^p(k!)^{n_k}n_k!}\widehat Z(\pi)\right) =N\log(NC_{\beta,\lambda})-N\mc I_{\beta,p,\lambda}(\gamma_1,\ldots,\gamma_p)+o_\lambda(N),
\end{equation*}
uniformly over $n_1+2n_2+\cdots+pn_p=N$. Indeed, $\pi$ has $\prod_i\binom{n_i}{m_i}$ subfamilies with $m_i$ blocks of cardinality $i$, and \eqref{eq:asin} is uniform in $(n_1,\ldots,n_p)$. By the bound on $\Ksf_{\beta,\lambda}^\mult$ and Lemma~\ref{lemma:bounded cluster tail}, we may sum over $n$. Summing over the at most $(N+1)^p$ profiles and using continuity of $\mc I_{\beta,p,\lambda}$ on $\triangle_{p,0}$ gives
\begin{equation}\label{eq:goal1}
\log\sum_{\pi\in\mathbf\Pi_p([N])}\widehat Z(\pi)
=N\log(NC_{\beta,\lambda})-N\inf_{\triangle_{p,0}}\mc I_{\beta,p,\lambda}+o_\lambda(N).
\end{equation}
For fixed $\lambda$, the error $o_\lambda(N)$ is absorbed in $O(N\delta_{\beta,\lambda})$. Comparing \eqref{eq:goal1} with \eqref{eq:first start} proves \eqref{eq:reduction p}.
\end{proof}
\endgroup

\subsection{Switching to a dipole expansion}

\begingroup\color{black}
We now sum over the multipole partitions by grouping connected components.

\begin{lemma}\label{lemma:auxiliary}
Let $\beta\in(2,\infty)$, let $p=p(\beta)$ be as in Definition~\ref{def:pbeta}, and let $\widehat Z$ be as in \eqref{def:first Zpi}. Then,
\begin{equation}\label{eq:switch1}
\frac{1}{(NC_{\beta,\lambda})^N}
\sum_{\pi\in\mathbf\Pi_p([N])}\widehat Z(\pi)
=\sum_{n=0}^\infty\frac1{n!}
\sum_{\substack{V_1,\ldots,V_n\subset[N]\ \mathrm{disjoint}\\
\forall i,\ 2\leq|V_i|\leq p}}
\prod_{i=1}^n\Ksf_{\beta,\lambda}^\dip(V_i).
\end{equation}
\end{lemma}

\medskip
\begin{proof}
For $V\subset[N]$ with $2\leq|V|\leq p$, we first show that
\begin{equation}\label{eq:dipole composition}
\Ksf_{\beta,\lambda}^\dip(V)
=\Msf_\infty^0(V)
+\sum_{\substack{\pi\in\mathbf\Pi(V)\\|\pi|\geq2}}
\Msf_\infty^0(\pi)\Ksf_{\beta,\lambda}^\mult(\pi).
\end{equation}
We decompose the integral of $\prod_{i<j\in V}e^{-\beta v_{ij}}\indic_{\mc A_{ij}}$ under $\mu_{\beta,\lambda}^{\otimes|V|}$ by the connected components of the graph of $\leftrightarrow$ on $V$. Since $\mc B_{ij}^c\subset\mc A_{ij}$, we expand the factors between components as
\[
e^{-\beta v_{ij}}\indic_{\mc B_{ij}^c}
=1+\bigl(e^{-\beta v_{ij}}\indic_{\mc B_{ij}^c}-1\bigr)
\]
and group the components linked by the second terms. In the resulting sum over partitions of $V$, each block has the weight given by the right-hand side of \eqref{eq:dipole composition}. The term $\Msf_\infty^0(V)$ accounts for a single multipole. Expanding the same integral with $e^{-\beta v_{ij}}\indic_{\mc A_{ij}}=1+f_{ij}^v$ gives instead the connected weight $\Ksf_{\beta,\lambda}^\dip$. M\"obius inversion (Lemma~\ref{lemma:mobius}) identifies the two weights.

Finally, we group each family in \eqref{def:first Zpi} by the sets $V_{X_i}$ and the blocks of $\pi$ outside them. These form a partition of $[N]$ into sets of cardinality at most $p$. By \eqref{eq:dipole composition}, each set of cardinality at least two has weight $\Ksf_{\beta,\lambda}^\dip$, while a singleton has weight $1$. This proves \eqref{eq:switch1}.
\end{proof}
\endgroup

\subsection{Proof of the free-energy expansion}\label{sub:free energy proof}

\begin{proof}[Proof of Theorem~\ref{theorem:expansion}]
	Let $p\coloneqq p(\beta)$. \textcolor{black}{For fixed $\lambda>0$, $\g_\lambda$ is bounded on $\Lambda-\Lambda$, so $0<Z_{N,\beta}^\lambda<\infty$.} First, notice that \eqref{eq:th1 ksf} follows from the estimate \eqref{eq:Kdip bound} in Proposition~\ref{prop:bounded lower}.

\begingroup\color{black}
We first prove that \eqref{def:fbeta} converges absolutely. For $p=1$, the series is zero. For $p\geq2$, \eqref{eq:th1 ksf} gives $C_0=C_0(\beta,p)>0$ such that $|\ksf_{\beta,\lambda}^{\dip}(j)|\leq C_0\lambda^{2(j-1)}$ for $2\leq j\leq p$. Since $\sum_i(|V_i|-1)=k-1\geq n$, \eqref{eq:hypertree-count} bounds the absolute value of the $(n,k)$ summand by
\[
\frac{C_0^n\lambda^{2(k-1)}n^k(k+n)^n}{n!k!}
\leq e(e^2C_0\lambda)^n(e\lambda)^{k-1},
\]
which is summable over $n,k\geq1$ for $\lambda$ small enough depending on $\beta,p$.
\endgroup
	Let us now prove the expansions \eqref{eq:th1 exp1} and \eqref{eq:th1 exp2}.
	
\begingroup\color{black}
Fix $M$ large enough depending on $\beta,p$. By \eqref{eq:Kdip bound} and Lemma~\ref{lemma:bounded cluster tail}, the logarithm of the right-hand side of \eqref{eq:switch1} has an absolutely convergent connected expansion for $\lambda$ small enough. Combining Proposition~\ref{prop:rewrite1}, \eqref{eq:switch1}, Lemma~\ref{lemma:resum2} and \eqref{limCbl}, we get, for $N$ large enough at fixed $\lambda$,
\begin{multline}\label{eq:poi}
\log Z_{N,\beta}^\lambda=\log(N!)+N\log N+N\log(\lambda^{2-\beta}\mc Z_\beta)
\\+\sum_{n=1}^\infty\frac1{n!}
\sum_{\substack{V_1,\ldots,V_n\subset[N]\ \mathrm{connected}\\
\forall i,\ 2\leq|V_i|\leq p}}
\prod_{i=1}^n\Ksf_{\beta,\lambda}^\dip(V_i)\mathrm I(G(V_1,\ldots,V_n))
+O_{\beta,p}(N\delta_{\beta,\lambda}).
\end{multline}
\endgroup

	Let us expand the right-hand side of \eqref{eq:poi}. Let $n\geq 1$. \textcolor{black}{Below, terms with $|V_i|\leq1$ vanish.} First, we show that only hypertrees contribute to the order-$N$ term. Recall $\Htrees_n$ from \eqref{def:Htrees}. We first sum according to the cardinality of the union of $V_1,\ldots,V_n$:
	\begin{multline}\label{eq:sumk}
		\sum_{\substack{V_1,\ldots,V_n \subset [N]\\ \mathrm{connected}, \forall i,|V_i|\leq p }}\Ksf_{\beta,\lambda}^\dip(V_1)\cdots \Ksf_{\beta,\lambda}^\dip(V_n)\mathrm{I}(G(V_1,\ldots,V_n))\\=\sum_{k=0}^\infty \binom{N}{k}\sum_{\substack{V_1,\ldots,V_n \subset [N]\\ \mathrm{connected}, \forall i,|V_i|\leq p\\
				V_1\cup \cdots \cup V_n=\{1,\ldots,k\}}}\Ksf_{\beta,\lambda}^\dip(V_1)\cdots \Ksf_{\beta,\lambda}^\dip(V_n)\mathrm{I}(G(V_1,\ldots,V_n)).
	\end{multline}
	By Lemma~\ref{lemma:limiting}, we have 
	\begin{equation*}
		\Ksf_{\beta,\lambda}^\dip(V_1)\cdots \Ksf_{\beta,\lambda}^\dip(V_n)=\prod_{i=1}^n N^{1-|V_i|}\left(\ksf_{\beta,\lambda}^\dip(|V_i|)+o_N(1)\right).  
	\end{equation*}
	Using
	\begin{equation*}
		\textcolor{black}{\binom{N}{k}=\frac{N^k}{k!}+O_k(N^{k-1})},
	\end{equation*}
	we see that the only terms contributing in the limit (after dividing by $N$) are the ones such that 
	\begin{equation*}
		n-\sum_{i=1}^n |V_i| + |V_1\cup\cdots\cup V_n|=1.
	\end{equation*}
	\begingroup\color{black}
By \eqref{def:Htrees}, this is equivalent to $(V_1,\ldots,V_n)\in\Htrees_n([k])$. Thus, by \eqref{eq:hypertree weight},
\begin{multline}\label{eq:factKI}
\sum_{\substack{V_1,\ldots,V_n\subset[N]\\\mathrm{connected},\ \forall i,|V_i|\leq p}}
\Ksf_{\beta,\lambda}^{\dip}(V_1)\cdots\Ksf_{\beta,\lambda}^{\dip}(V_n)\mathrm I(G(V_1,\ldots,V_n))\\
=N(-1)^{n-1}\sum_{k=1}^{\infty}\frac1{k!}
\sum_{\substack{\mathbf V\in\Htrees_n([k])\\V_1\cup\cdots\cup V_n=[k]}}
\omega(\mathbf V)\prod_{i=1}^n\ksf_{\beta,\lambda}^{\dip}(|V_i|)\indic_{|V_i|\leq p}+o(N).
\end{multline}
By \eqref{eq:Kdip bound} and Lemma~\ref{lemma:bounded cluster tail}, we may sum over $n$ with weight $1/n!$. The sum over $n$ in \eqref{eq:poi} is therefore $-N\mathrm{Mayer}_{\beta,p,\lambda}+o(N)$. For fixed $\lambda$, the error is absorbed in $O(N\delta_{\beta,\lambda})$ for $N$ large enough.
\endgroup
Inserting this into \eqref{eq:poi} gives
	\begin{equation*}
		\log Z_{N,\beta}^\lambda=\log N!+ N\log N+N\log (\lambda^{2-\beta}\mc{Z}_\beta)-N\mathrm{Mayer}_{\beta,p,\lambda}+O_{\beta,p(\beta)}(N\delta_{\beta,\lambda}).
	\end{equation*}
	\textcolor{black}{Together with \eqref{eq:factKI} and Stirling's formula $\log N!=N\log N-N+O(\log N)$, this proves \eqref{eq:th1 exp1} and \eqref{eq:th1 exp2}.}
\end{proof}

\appendix

\section{Auxiliary estimates}

\textcolor{black}{As in Section~\ref{section:upper}, we take $M>\max(20,3^{p(\beta)})$. Let $V_\bad\subset[N]$. The frozen configurations $Z=(x_i,y_i)_{i\in V_\bad}$ satisfy $I_\bad^{V_\bad}=V_\bad$ and lie in $\bigcap_{i,j\in V_\bad:i<j}\mc A_{ij}$. The bad sets and radii are computed from $Z$, with $N$ kept in the thresholds of Definition~\ref{def:type2}. For $i\in[N]\setminus V_\bad$, the quantities $d_{i,\bad}$ and $\tau_i^Z$ are given by \eqref{eq:dibad} and \eqref{deftauz}.}

\subsection{Control of the multipole partition function}

We start with a couple of preliminary results.

\begin{lemma}\label{lemma:ordering tree}
	Let $T=(V,E)$ be a tree on $V\coloneqq \{1,\ldots,k\}$. Let $f:\dR\to\dR_+$ be an increasing function. Let $r_1,\ldots,r_k\geq 0$ be a family of numbers. Consider an ordering of the $r_i$'s such that $r_{\phi(1)}\geq \cdots \geq r_{\phi(k)}$. We have
	\begin{equation*}
		\prod_{ij\in E}f(\min(r_i,r_j))\leq f(r_{\phi(2)})\cdots f(r_{\phi(k)}).
	\end{equation*}
\end{lemma}

\begin{proof}
	The proof is by induction. Assume that the property holds for every tree of size $k$. Let $T=(V,E)$ be a tree of size $k+1$. Let $v$ be a  leaf of $T$ and $v'$ a neighbor of $v$. We remove the edge $vv'$ from $T$, thus defining a connected subtree $T'=(V\setminus\{v\},E\setminus \{vv'\})$ of size $k$. 
	
	Let $i_0$ be the index of the largest $r_i$ for $i\in V\setminus \{v\}$. Applying the induction hypothesis to $T'$ allows one to write
	\begin{equation}\label{eq:t'}
		\prod_{ij\in E(T')}f(\min(r_i,r_j))\leq \prod_{i\in V\setminus\{v,i_0\}}f(r_i).
	\end{equation}
	
	Assume first that $v=\phi(1)$. Since $f$ is increasing,
	\begin{equation*}
		f(\min(r_v,r_{v'}))\leq f(r_{v'})\leq f(r_{i_0}).
	\end{equation*}
	Combining this with \eqref{eq:t'}, we find that the final product contains all terms but $f(r_{\phi(1)})$, which concludes the proof.
	
	Assume next that $v\neq \phi(1)$. Then $i_0=\phi(1)$. One may write
	\begin{equation*}
		f(\min(r_v,r_{v'}))\leq f(r_v).
	\end{equation*}
	Combined with \eqref{eq:t'}, this proves the result since all terms but $f(r_{i_0})$ are present in the final product.
\end{proof}

\begin{lemma}
	\label{lemsupplapp}
	Let $k\in \mathbb{N}^*$. \textcolor{black}{For $\beta>2$ and $\lambda$ small enough,} there exists $C>0$ depending only on $\beta$, $k$ and $\ve_0$ such that 
	\begin{equation}\label{eq:mm4}
		\frac{1}{C}\lambda^{(2-\beta)k}\gamma_{\beta,\lambda,k}\leq	\int_0^\infty \max(r,\lambda)^{(4-\beta)k-3}\indic_{r\leq \ve_0\Cut}\dd r \le C \lambda^{(2-\beta)k}\gamma_{\beta,\lambda,k},
	\end{equation}
	where $\gamma_{\beta,\lambda,k}$ is as in Definition~\ref{def:gamma beta}. If, moreover, $k\le p^*(\beta)$, there exists $C>0$ depending only on $\beta$, $k$ and $\ve_0$ such that 
	\begin{equation}\label{eq:mm5}
		\int_0^\infty \max(r,\lambda)^{(4-\beta)k-3}\indic_{r\ge \ve_0\Cut}\dd r \le C \lambda^{(2-\beta)k}\Cut^{-2} .
	\end{equation}
\end{lemma}

\begin{proof}
	Let us prove the upper bound in \eqref{eq:mm4}.
	
	Suppose first that $k\leq p^*(\beta)$. Then we have $(4-\beta)k-3<-1$. Thus, there exists a constant $C>0$ depending on $\beta$ and $k$ such that
	\begin{equation*}
		\int \max(r,\lambda)^{(4-\beta)k-3}\indic_{r\leq \ve_0\Cut}\dd r\leq C \lambda^{(4-\beta)k-2}. 
	\end{equation*}

	Next, let us turn to the case where $k\geq p^*(\beta)+1$ (which implies $\beta\in (2,4)$). Suppose that $\beta\in (\beta_{p^*(\beta)},\beta_{p^*(\beta)+1})$, or $\beta=\beta_{p^*(\beta)+1}$ and $k>p^*(\beta)+1$. Then, $(4-\beta)k-3>-1$ and therefore there exists a constant $C>0$ depending on $\beta$ and $k$ such that 
	\begin{equation*}
		\int \max(r,\lambda)^{(4-\beta)k-3}\indic_{r\leq \ve_0\Cut}\dd r\leq C(\ve_0 R_{\beta,\lambda})^{(4-\beta)k-2}=C\ve_0^{(4-\beta)k-2}\lambda^{(2-\beta)k}R_{\beta,\lambda}^{-2},\end{equation*}
	since $R_{\beta,\lambda}^{4-\beta}=\lambda^{2-\beta}$.
	
	Suppose finally that $\beta=\beta_{p^*(\beta)+1}$ and $k=p^*(\beta)+1$ (which implies $\beta\in (2,4)$). Then $(4-\beta)k-3=\frac{2}{p^*(\beta)+1}k-3=-1$. Hence, there exists $C>0$ depending on $\beta$ and $\ve_0$ such that
	\begin{equation*}
		\int \max(r,\lambda)^{(4-\beta)k-3}\indic_{r\leq \ve_0\Cut}\dd r\leq  C|\log\lambda|.
	\end{equation*}
	One can notice that
	\begin{equation*}
		(\beta-2)k=(2-\tfrac{2}{p^*(\beta)+1})(p^*(\beta)+1)=2p^*(\beta), \quad \Cut^{-2}=R_{\beta,\lambda}^{-2}= \lambda^{2p^*(\beta)}
		.
	\end{equation*}
	Thus, we conclude that there exists a constant $C>0$ depending on $\beta$, $k$ and $\ve_0$ such that 
	\begin{equation*}
		\int \max(r,\lambda)^{(4-\beta)k-3}\indic_{r\leq \ve_0\Cut}\dd r\leq C \Cut^{-2}|\log \lambda|\lambda^{(2-\beta)k}.  
	\end{equation*}
	\begingroup\color{black}
This proves the upper bound in \eqref{eq:mm4}. Set $h=\ve_0\Cut\geq2\lambda$. For the lower bound, we restrict the integral to $[0,\lambda]$ when $(4-\beta)k-3<-1$, to $[h/2,h]$ when $(4-\beta)k-3>-1$, and to $[\lambda,h]$ at equality. Finally, \eqref{eq:mm5} follows from
\[
\int_h^\infty r^{(4-\beta)k-3}\dd r
=\frac{h^{(4-\beta)k-2}}{2-(4-\beta)k}
\leq C\lambda^{(2-\beta)k}\Cut^{-2},
\]
by the definition of $\Cut$.
\endgroup
	
\end{proof}

\begingroup\color{black}
For nonempty $S\subset[N]\setminus V_\bad$, $L\in\dR$, $\ve_0\in(0,1)$ and $(w,\tau)\in\{(\tilde v,\tilde\tau),(v^Z,\tau^Z)\}$, where $\tilde\tau_i=\tau_0$ with $\tau_0$ as in Definition~\ref{def:vijtilde}, we define
\endgroup
\begin{multline}\label{def:I1wL}
		\textcolor{black}{I_{\mathrm{mult},\ve_0}(S;w,\tau,L,Z)}\coloneqq \int_{(\Lambda^2)^{|S|} }\prod_{i\in S} \indic_{i\in I_\good^Z} \indic_{\mc{B}_S}\prod_{ i,j\in S:i<j}e^{-\beta w_{ij}}\indic_{\mc{A}_{ij}}\prod_{i\in S }e^{L\frac{|x_i-y_i|^2}{\tau_i^2}} \\ \times\prod_{i\in S} e^{\beta \g_\lambda(x_i-y_i)}\indic_{|x_i-y_i|\leq \ve_0\Cut}\dd x_i \dd y_i.
	\end{multline}
\begingroup\color{black}
The same formula, with $I_\good^\emptyset=[N]$ and $r_i^2/\infty^2=0$, applies to $(w,\tau,L,Z)=(v,\infty,0,\emptyset)$, also with $\ve_0=\infty$. With $L$ the constant of the corresponding model,
\begin{gather*}
		I_{\mathrm{mult},\ve_0}^{-}(S) =I_{\mathrm{mult},\ve_0}(S;\tilde v,\tilde\tau,-64L,\emptyset),\qquad I_{\mathrm{mult},\ve_0}^{+,Z}(S) =I_{\mathrm{mult},\ve_0}(S;v^Z,\tau^Z,L,Z),\\
		I_{\mathrm{mult},\ve_0}^{0}(S) =I_{\mathrm{mult},\ve_0}(S;v,\infty,0,\emptyset).
	\end{gather*}
Thus, for each model and its allowed cutoffs,
\[
I_{\mathrm{mult},\ve_0}^{\star}(S)=(NC_{\beta,\lambda,\ve_0})^{|S|}\Msf_{\ve_0}^{\star}(S),
\qquad C_{\beta,\lambda,\infty}=C_{\beta,\lambda}.
\]
\endgroup

\begin{lemma}\label{lemma:lowerM+}
	Let $\beta\in (2,+\infty)$ and let $\gamma_{\beta,\lambda,k}$ be as in Definition~\ref{def:gamma beta}. Let $S\subset [N]$ be such that $1<|S|\leq p^*(\beta)$ and $V_\bad\subset [N]$ be such that $V_\bad\cap S=\emptyset$. Let $Z=(x_i,y_i)_{i\in V_\bad}$. Recall $I_\good^Z$ from \eqref{def:IgoodZ}. Let $\tau^Z=(\tau_i^Z)_{i\in\textcolor{black}{[N]\setminus V_\bad}}$ be as in Definition~\ref{def:vijZ}. Recall $\tau_0$ from Definition~\ref{def:vijtilde} and set $\tilde{\tau}_i\coloneqq \tau_0$ for every $i\in [N]$. \textcolor{black}{Fix $L\in\dR$. Let either $(w,\tau)=(\tilde v,\tilde\tau)$ and $V_\bad=\emptyset$, or $(w,\tau)=(v^Z,\tau^Z)$. Recall \eqref{def:I1wL}.}
	
	Then, for $\lambda$ small enough, there exists a constant $C>0$ depending on $\beta$, $M$, \textcolor{black}{$L$} and $|S|$ such that
	\begin{equation}\label{eq:int M lower}
		\frac{1}{C} N\lambda^{(2-\beta)|S|}\lambda^{2(|S|-1)}\le \textcolor{black}{I_{\mathrm{mult},\ve_0}(S;w,\tau,L,Z)} \leq CN\lambda^{(2-\beta)|S|}\lambda^{2(|S|-1)}.
	\end{equation}
\end{lemma}

\medskip

\begin{proof}
\begingroup\color{black}
On the integration support, we have $\tau_i^Z\geq r_i/4$ by the geometric argument at the beginning of the proof of Proposition~\ref{prop:bounded upper}, and $\tilde\tau_i=8\ve_0\Cut\geq r_i$. Hence
\[
e^{-16|L||S|}\leq\prod_{i\in S}e^{Lr_i^2/\tau_i^2}\leq e^{16|L||S|}.
\]
It suffices to prove the bounds for $L=0$.
\endgroup
	We first prove the upper bound. Note that on the event $\mc{A}_{ij}$, the interaction $w_{ij}$ is bounded from below (though not from above)\edgreen{, by \eqref{eq:bond-bound} and $|w_{ij}-v_{ij}|\leq C$ (Lemmas~\ref{lemma:error'} and~\ref{lemma:errortilde})}.  Therefore, using also that $\frac{|x_i-y_i|}{\tau_i}\le C$, we find a constant $C>0$ depending on $\beta$, $M$ and $|S|$ such that
	\begin{equation*}
		\textcolor{black}{I_{\mathrm{mult},\ve_0}(S;w,\tau,L,Z)} \leq C  \int_{(\Lambda^2)^{|S|}} \indic_{\mc{B}_S}\prod_{i\in S}e^{\beta \g_\lambda(x_i-y_i)}\indic_{|x_i-y_i|\leq \ve_0\Cut}\dd x_i\dd y_i.
	\end{equation*}
	Recalling the definition of $\mc{B}_S$ from Definition~\ref{def:multipoles mc}, we get 
	\begin{equation*}
		\textcolor{black}{I_{\mathrm{mult},\ve_0}(S;w,\tau,L,Z)}\leq C\sum_{T\in \mc{T}_c(S)} \int_{(\Lambda^2)^{|S|}} \prod_{ij\in T }\indic_{\mc{B}_{ij}}\prod_{i\in S}e^{\beta \g_\lambda(x_i-y_i)}\indic_{|x_i-y_i|\leq \ve_0\Cut}\dd x_i\dd y_i,
	\end{equation*}
	where we recall that $\mc{T}_c(S)$ stands for the set of collections of edges $T$ on $S$ such that $(S,T)$ is a tree. Hence,
	\begin{equation}\label{eq:g1}
		\textcolor{black}{I_{\mathrm{mult},\ve_0}(S;w,\tau,L,Z)}\leq C\max_{T\in \mc{T}_c(S)} \int_{(\Lambda^2)^{|S|}} \prod_{ij\in T}\indic_{\mc{B}_{ij}}\prod_{i\in S}e^{\beta \g_\lambda(x_i-y_i)}\indic_{|x_i-y_i|\leq \ve_0\Cut}\dd x_i\dd y_i,
	\end{equation}
	where $C$ depends on $|S|$.
	Fix $T\in \mc{T}_c(S)$ on $S$. As in the proof of  Lemma~\ref{lemma:type2}, we now split the phase space according to which points attain the distance $\dist(\{x_i,y_i\},\{x_j,y_j\})$ and let
	\begin{equation}\label{eq:D premier}
		\begin{aligned}
			 & \mc{D}_{ij}^1 =\{(x_q,y_q)_{q\in S}: d_{ij}= |x_i-x_j|\},\qquad \mc{D}_{ij}^2 =\{(x_q,y_q)_{q\in S}: d_{ij}= |x_i-y_j|\},\\
			 & \mc{D}_{ij}^3 =\{ (x_q,y_q)_{q\in S}:d_{ij}= |y_i-x_j|\},\qquad \mc{D}_{ij}^4 =\{ (x_q,y_q)_{q\in S}: d_{ij}= |y_i-y_j|\}.
		\end{aligned}
	\end{equation}
	One can write 
	\begin{multline}\label{eq:g2}
		\int_{(\Lambda^2)^{|S|}} \prod_{ij\in T}\indic_{\mc{B}_{ij}}\prod_{i\in S}e^{\beta \g_\lambda(x_i-y_i)}\indic_{|x_i-y_i|\leq \ve_0\Cut}\dd x_i\dd y_i\\ \leq \sum_{T^1,T^2,T^3,T^4}\int_{(\Lambda^2)^{|S|}} \prod_{ij\in T}\indic_{\mc{B}_{ij}}\left(\prod_{l=1}^4 \prod_{ij\in T^l}\indic_{\mc{D}_{ij}^l}\right) \prod_{i\in S}e^{\beta \g_\lambda(x_i-y_i)}\indic_{|x_i-y_i|\leq \ve_0\Cut}\dd x_i\dd y_i,
	\end{multline}
	where the sum is over \textcolor{black}{possibly empty forests} $T^1,T^2,T^3,T^4$ such that $T^1\sqcup T^2\sqcup T^3\sqcup T^4=T$. Fix some $T^1,T^2,T^3,T^4$ accordingly. We define the variables
	\begin{equation}\label{eq:g4}
		z_{ij}=\begin{cases}
			x_{i}-x_j & \text{if $ij\in T^1$}\\
			x_i-y_j & \text{if $ij\in  T^2 $}\\
			y_i - x_j & \text{if $ij\in T^3$}\\
			y_i-y_j & \text{if $ij\in T^4$}.
		\end{cases}
	\end{equation}
	Fix some arbitrary $i_0\in S$ and perform the following change of variables:
	\begin{equation}\label{eq:g5}
		(x_i,y_i)_{i\in S}\mapsto (x_{i_0},(z_{ij})_{ij\in T},(\vec{r}_i)_{i\in S}),
	\end{equation}
	where $\vr_i\coloneqq y_i-x_i$.
	Integrating over the variables $(z_{ij})_{ij\in T}$ and $x_{i_0}$ and recalling that $d_{ij}\le M \max(\min(r_i, r_j), \lambda)$ on $\mc{B}_{ij}$, we find $C>0$ depending on $M$ and $|S|$ such that
	\begin{multline}\label{eq:g3}
		\int_{(\Lambda^2)^{|S|}} \prod_{ij\in T}\indic_{\mc{B}_{ij}}\left(\prod_{l=1}^4 \prod_{ij\in T^l}\indic_{\mc{D}_{ij}^l}\right) \prod_{i\in S}e^{\beta \g_\lambda(x_i-y_i)}\indic_{|x_i-y_i|\leq \ve_0\Cut}\dd x_i\dd y_i\\
		\leq C N\int \prod_{ij\in T}\max(\min(r_i,r_j), \lambda)^2 \prod_{i\in S}e^{\beta \g_\lambda(r_i)}\indic_{r_i\leq \ve_0\Cut}\dd \vr_i.
	\end{multline}
	Next, we split the phase space according to the index of the largest $r_i$ for $i$ in $S$:
	\begin{multline}\label{st1}
		\int \prod_{ij\in T}\max(\min(r_i,r_j),\lambda)^2 \prod_{i\in S}e^{\beta \g_\lambda(r_i)}\indic_{r_i\leq \ve_0\Cut}\dd \vr_i\\ \leq |S|\max_{j_0\in S}\int \prod_{ij\in T}\max(\min(r_i,r_j),\lambda)^2 \prod_{i\in S}e^{\beta \g_\lambda(r_i)}\indic_{r_i\leq r_{j_0}\leq \ve_0\Cut}\dd \vr_i.
	\end{multline}
	By Lemma~\ref{lemma:ordering tree}, 
	\begin{equation*}
		\prod_{ij\in T}\max(\min(r_i,r_j),\lambda)^2\leq \prod_{i\in S:i\neq j_0}\max(r_i,\lambda)^2. 
	\end{equation*}
	Using this and integrating over the variables $\vec{r}_i$ in polar coordinates for $i\neq j_0$, we get that there exists $C>0$ depending on $\beta$ and $|S|$ such that
	\begin{multline}\label{st2}
		\int \prod_{ij\in T}\max(\min(r_i,r_j),\lambda)^2 \prod_{i\in S}e^{\beta \g_\lambda(r_i)}\indic_{r_i\leq \ve_0\Cut}\dd \vr_i\\
	\leq \textcolor{black}{C\int_0^{\ve_0\Cut}r(r\vee\lambda)^{-\beta}\left(\int_0^r t(t\vee\lambda)^{2-\beta}\dd t\right)^{|S|-1}\dd r.} 
	\end{multline}
	\textcolor{black}{For $\beta<4$, the inner integral is at most $C(r\vee\lambda)^{4-\beta}$, so \eqref{eq:mm4} applies. For $\beta\geq4$, a rescaling by $\lambda$ gives the factor $\lambda^{(4-\beta)|S|-2}$ times a bounded integral, since
\[
\int_1^\infty u^{1-\beta}\left(1+\int_1^u t^{3-\beta}\dd t\right)^{|S|-1}\dd u<\infty.
\]
When $|S|\leq p^*(\beta)$, both arguments also apply with the cutoff $\infty$.} Thus there exists $C>0$ depending on $\beta$ and $|S|$ such that
	\begin{equation}\label{eq:g6}
		\int \prod_{ij\in T}\max(\min(r_i,r_j),\lambda)^2 \prod_{i\in S}e^{\beta \g_\lambda(r_i)}\indic_{r_i\leq \ve_0\Cut}\dd \vr_i\leq C \lambda^{(2-\beta)|S|}\gamma_{\beta,\lambda,|S|}.
	\end{equation}
	Combining this with \eqref{eq:g1}, \eqref{eq:g2} and \eqref{eq:g3} concludes the proof of the upper bound.

	We now turn to the proof of  the lower bound. \textcolor{black}{We may assume that $S=[k]$, and we let $T$ be the path $1,\ldots,k$.} Define the good event 
	\begin{equation*}
		\mc{C}\coloneqq \bigcap_{i\in S}\{r_i\in (\lambda,2\lambda)\}\cap \bigcap_{ij\in T}\{d_{ij}\in (2\lambda,4\lambda)\}.  
	\end{equation*}
	We begin by writing
	\begin{multline*}
		\int_{(\Lambda^2)^{|S|}}\prod_{i\in S} \indic_{i\in I_\good^Z}\indic_{\mc{B}_S} \prod_{i,j\in S:i<j}e^{-\beta w_{ij}}\indic_{\mc{A}_{ij}}\prod_{i\in S}e^{\beta \g_\lambda(x_i-y_i)}\indic_{|x_i-y_i|\leq \ve_0\Cut}\dd x_i \dd y_i\\
		\geq  \int_{(\Lambda^2)^{|S|}}\indic_{\mc{C}} \prod_{i\in S} \indic_{i\in I_\good^Z} \indic_{\mc{B}_S}\prod_{i,j\in S:i<j}e^{-\beta w_{ij}}\indic_{\mc{A}_{ij}}\prod_{i\in S}e^{\beta \g_\lambda(x_i-y_i)}\indic_{|x_i-y_i|\leq \ve_0\Cut}\dd x_i \dd y_i.
	\end{multline*}
	\textcolor{black}{On $\mc{C}\cap\bigcap_{i<j}\mc A_{ij}$, we have $|v_{ij}|\leq\|D^2\g_\lambda\|_{L^\infty}r_ir_j\leq C$. Since $\tau_i\geq\lambda\geq r_i/2$, Lemmas~\ref{lemma:error'} and~\ref{lemma:errortilde} also give $|w_{ij}-v_{ij}|\leq C$. Hence, with $C$ depending on $M$ and $|S|$,}
	\begin{equation*}
		\left|\sum_{i,j\in S:i<j}w_{ij}\right|\leq C.
	\end{equation*}
	It follows that there exists $C>0$ depending on $\beta, M$ and $|S|$ such that 
	\begin{multline*}
		\int_{(\Lambda^2)^{|S|}}\prod_{i\in S} \indic_{i\in I_\good^Z}\indic_{\mc{B}_S} \prod_{i,j\in S:i<j}e^{-\beta w_{ij}}\indic_{\mc{A}_{ij}}\prod_{i\in S}e^{\beta \g_\lambda(x_i-y_i)}\indic_{|x_i-y_i|\leq \ve_0\Cut}\dd x_i \dd y_i \\ 
		\geq  \frac{1}{C}\int_{(\Lambda^2)^{|S|}}\indic_{\mc{C}} \prod_{i\in S} \indic_{i\in I_\good^Z} \indic_{\mc{B}_S}\prod_{i,j\in S:i<j}\indic_{\mc{A}_{ij}}\prod_{i\in S}e^{\beta \g_\lambda(x_i-y_i)}\indic_{|x_i-y_i|\leq \ve_0\Cut}\dd x_i \dd y_i.
	\end{multline*}
	\begingroup\color{black}
We restrict to the configurations where $x_1$ is at distance at least $5k\lambda$ from $\partial\Lambda$, where, for every $i>1$, the point $x_i$ lies in the ball of radius $\lambda/16$ about $x_1+\frac92(i-1)\lambda(1,0)$, and where, for every $i$, the point $y_i$ lies in the ball of radius $\lambda/16$ about $x_1+(\frac92(i-1)+\frac32)\lambda(1,0)$. These configurations belong to $\mc C$, and every distance between two dipoles exceeds every $r_i$, so all $\mc A_{ij}$ hold and the path edges satisfy $\mc B_{ij}$. For $\lambda$ small enough, their volume is at least $c_kN\lambda^{4k-2}$, and $\prod_ie^{\beta\g_\lambda(x_i-y_i)}\geq c_{\beta,k}\lambda^{-\beta k}$ on them. Hence,
\endgroup
	\begin{equation}\label{eq:o1}
		\int_{(\Lambda^2)^{|S|}}\indic_{\mc{C}}  \indic_{\mc{B}_S}\prod_{i,j\in S:i<j}\indic_{\mc{A}_{ij}}\prod_{i\in S}e^{\beta \g_\lambda(x_i-y_i)}\indic_{|x_i-y_i|\leq \ve_0\Cut}\dd x_i \dd y_i\geq \frac{N}{C}\lambda^{(4-\beta)|S|-2}.
	\end{equation}
	On the other hand,
	\begin{multline*}
		\int_{(\Lambda^2)^{|S|}}\indic_{\mc{C}}\Bigl(1-\prod_{i\in S}\indic_{i\in I_\good^Z}\Bigr)   \indic_{\mc{B}_S}\prod_{i,j\in S:i<j}\indic_{\mc{A}_{ij}}\prod_{i\in S}e^{\beta \g_\lambda(x_i-y_i)}\indic_{|x_i-y_i|\leq \ve_0\Cut}\dd x_i \dd y_i \\
		\leq \sum_{i_0\in S}\int_{(\Lambda^2)^{|S|}}\indic_{\mc{C}}\indic_{i_0\notin I_\good^Z}\indic_{\mc{B}_S}\prod_{i\in S}e^{\beta \g_\lambda(x_i-y_i)}\indic_{|x_i-y_i|\leq \ve_0\Cut}\dd x_i \dd y_i.
	\end{multline*}
	By \eqref{def:IgoodZ}, if $i_0\notin I_\good^Z$, then there exists $j\in V_\bad$ such that $d_{i_0j}\leq M\max(r_{i_0},\lambda)$.
	\textcolor{black}{We use the tree coordinates \eqref{eq:g5}, rooted at $i_0$. On $\mc C$, the root lies in one of at most $2|V_\bad|$ balls of radius $(2M+2)\lambda$, so its integral is at most $C_M|V_\bad|\lambda^2$. The other coordinates contribute at most $C_{\beta,M,k}\lambda^{(4-\beta)k-2}$. Thus,}
	\begin{equation}\label{eq:o2}
		\int_{(\Lambda^2)^{|S|}}\indic_{\mc{C}}\Bigl(1-\prod_{i\in S}\indic_{i\in I_\good^Z}\Bigr)   \indic_{\mc{B}_S}\prod_{i,j\in S:i<j}\indic_{\mc{A}_{ij}}\prod_{i\in S}e^{\beta \g_\lambda(x_i-y_i)}\indic_{|x_i-y_i|\leq \ve_0\Cut}\dd x_i \dd y_i \leq C|V_\bad|\lambda^2 \times\lambda^{(4-\beta)|S|-2}.
	\end{equation}
	Since $|V_\bad|\leq N$, assembling \eqref{eq:o1} and \eqref{eq:o2} concludes the proof of \eqref{eq:int M lower}.
\end{proof}

\subsection{Integrals for bounded clusters}

\begin{lemma}\label{lemma:integral small -}
	Let $\beta\in (2,+\infty)$. Recall $\gamma_{\beta,\lambda,k}$ from Definition~\ref{def:gamma beta}. Let $V_\bad \subset [N]$, $V_\good =[N]\backslash V_\bad$, and 
	let $X$ be a subpartition of $V_\good$ such that $V_\bad\cap V_X=\emptyset$ \textcolor{black}{and $|X|\geq2$}. Let $X_1,\ldots,X_n\subset X$ be disjoint. Let $T^b\in \textcolor{black}{\mathsf{T}_0}^{\Coarse_X(X_1,\ldots,X_n)}$, and, for every $l\in [n]$, let $T_l^a\in \mathsf{T}^{X_l}$. Set $T^a\coloneqq \cup_{l\in [n]}T_l^a$. Let $Z=(x_i,y_i)_{i\in V_\bad}$ and $d_{i,\bad}$ be as in Definition~\ref{def:dist bad}. 
	
	\textcolor{black}{For every $l\in[n]$, let $\hat S_l\in X_l$ be such that $r_{\hat S_l}=\max_{S\in X_l}r_S$.}

	Define
	\begin{multline}\label{def:J1}
		J\coloneqq \int_{(\Lambda^2)^{|V_X|} } \min\Bigl(\frac{\max_{i\in V_X} r_i}{\max_{e\in T^a}d_{e} },1\Bigr)^2 \left(\prod_{l=1}^n \prod_{S\in X_l:S\neq \hat{S}_l}r_S^2 \right)\left(\prod_{ij\in T^a}\frac{1}{d_{ij}^2}\indic_{ d_{ij}\leq \textcolor{black}{17}\ve_0\Cut }\indic_{\mc{B}_{ij}^c} \right)\prod_{ij\in T^b}b_{ij}^{\abs} \\ \times \prod_{S\in X}\indic_{\mc{B}_S}\prod_{i\in V_X}e^{\beta \g_\lambda(x_i-y_i)}\indic_{|x_i-y_i|\leq \ve_0\Cut}\dd x_i \dd y_i,
	\end{multline}
	with the convention that $\max_{e\in T^a}d_e=0^+$ if $T^a=\emptyset$, and 
	\begingroup\color{black}
\begin{multline}\label{def:J1'}
		\textcolor{black}{J_{\mathrm{ext}}}\coloneqq \int_{(\Lambda^2)^{|V_X|} } \min\left(\frac{\max_{i\in V_X}r_i}{\min_{i\in V_X}d_{i,\bad}},1 \right)^2 
		\left(\prod_{l=1}^n \prod_{S\in X_l:S\neq \hat{S}_l}r_S^2\right)
		\\ \times \prod_{ij\in T^b}b_{ij}^{\abs}
		\prod_{ij\in T^a }\frac{1}{d_{ij}^2}\indic_{ d_{ij}\leq \textcolor{black}{17}\ve_0\Cut }
		\indic_{\mc{B}_{ij}^c}\prod_{S\in X}\indic_{\mc{B}_S} \prod_{i\in V_X}e^{\beta \g_\lambda(x_i-y_i)}\indic_{|x_i-y_i|\leq \ve_0\Cut}\dd x_i \dd y_i.
\end{multline}
\endgroup
	\textcolor{black}{For $\lambda$ small enough, there} exists a constant $C>0$ depending on $\beta,M,\ve_0$ and $|V_X|$ such that
	\begin{equation}\label{eq:bound J1}
		J\leq CN\lambda^{(2-\beta)|V_X|}\gamma_{\beta,\lambda,|V_X|}.
	\end{equation}
\begingroup\color{black}
Let $\widehat J$ be obtained from $J$ by removing the spatial cutoffs and replacing $b_{ij}^{\abs}$ by
\[
\widehat b_{ij}=\widehat a_{ij}^{\,2}+\indic_{\mc B_{ij}},
\qquad
\widehat a_{ij}=\frac{r_ir_j}{d_{ij}\max(d_{ij},r_i,r_j)}\indic_{\mc B_{ij}^c}.
\]
Since $a_{ij}^{\abs}\leq\widehat a_{ij}$ and $b_{ij}^{\abs}\leq\widehat b_{ij}$, we have $J\leq\widehat J$. The bound \eqref{eq:bound J1} also holds for $\widehat J$, which keeps the radial cutoff. If $|V_X|\leq p^*(\beta)$, it also holds for $\widehat J_\infty$, obtained by removing this cutoff.
\endgroup

	{\color{black}
Moreover, for $\lambda$ small enough, there exist $C>0$ depending only on $\beta,M,|V_X|$ and $C_{\ve_0}>0$ also depending on $\ve_0$ such that
    \begin{equation}\label{eq:bound J2}
        \textcolor{black}{J_{\mathrm{ext}}}\leq\lambda^{(2-\beta)|V_X|}
        \bigl(C_{\ve_0}N\Cut^{-2}+C\ve_0^{\alpha(\beta)|V_X|}|V_\bad|\bigr).
    \end{equation}
}

\end{lemma}

\medskip

\begin{proof}
	Let $k\coloneqq |V_X|$.
	Let us first prove \eqref{eq:bound J1}. \textcolor{black}{It suffices to bound $\widehat J$. Set $h=\ve_0\Cut$, $\rho_i=r_i\vee\lambda$, $\rho=\max_i\rho_i$ and $\rho_*=\min_i\rho_i$.} Recall from Definition~\ref{def:multipoles mc} that 
	\begin{equation}\label{st3}
		\indic_{\mc{B}_S}\leq \sum_{R_S\in \mc{T}_c(S) }\prod_{ij\in R_S}\indic_{\mc{B}_{ij}}.
	\end{equation}
	By Cayley's formula, the number of such trees is $|S|^{|S|-2}$. Therefore, there exists a constant $C>0$ depending on $k$ such that 
	\begin{equation}\label{eq:mm0}
		\textcolor{black}{\widehat J}\leq C^{|X|}\max_{(R_S)}J_1((R_S)),
	\end{equation}
	where 
	\begin{multline*}
		J_1((R_S))\coloneqq \int_{(\Lambda^2)^{k} } \min\Bigl(\frac{\max r_i}{\max_{e\in T^a} d_{e} },1\Bigr)^2 \left(\prod_{l=1}^n \prod_{S\in X_l:S\neq \hat{S}_l}r_S^2 \right)\prod_{ij\in T^b}\textcolor{black}{\widehat b_{ij}}\prod_{ij\in T^a }\frac{1}{d_{ij}^2}\textcolor{black}{\indic_{\mc B_{ij}^c}}\\ \times \left(\prod_{S\in X}\prod_{ij\in R_S}\indic_{\mc{B}_{ij}}\right)\prod_{i\in V_X}e^{\beta \g_\lambda(x_i-y_i)}\indic_{|x_i-y_i|\leq \ve_0\Cut}\dd x_i \dd y_i.
	\end{multline*}
	For every $S\in X$, fix some $R_S\in \mc{T}_c(S)$.
	
	Let
	\begin{equation*}
		\begin{aligned}
			 & \mc{D}_{ij}^1 =\{(x_q,y_q)_{q\in V_X}: d_{ij}= |x_i-x_j|\},\qquad \mc{D}_{ij}^2 =\{(x_q,y_q)_{q\in V_X}: d_{ij}= |x_i-y_j|\},\\
			 & \mc{D}_{ij}^3 =\{ (x_q,y_q)_{q\in V_X}:d_{ij}= |y_i-x_j|\},\qquad \mc{D}_{ij}^4 =\{ (x_q,y_q)_{q\in V_X}: d_{ij}= |y_i-y_j|\}.
		\end{aligned}
	\end{equation*}
	Define 
	\begin{multline*}
		J_2((R_S^m),(T^{a,m}),(T^{b,m}))\\\coloneqq \int_{(\Lambda^2)^{k} } \min\Bigl(\frac{\max_{i\in V_X} r_i}{\max_{e\in T^a} d_{e} },1\Bigr)^2 \left(\prod_{l=1}^n \prod_{S\in X_l:S\neq \hat{S}_l}r_S^2 \right)\prod_{ij\in T^b}\textcolor{black}{\widehat b_{ij}}\prod_{ij\in T^a}\frac{1}{d_{ij}^2}\textcolor{black}{\indic_{\mc B_{ij}^c}}\\ \times \left(\prod_{m=1}^4\prod_{ij\in \cup_{S\in X}R_S^m\cup T^{a,m}\cup T^{b,m}}\indic_{\mc{D}_{ij}^m} \right) \left(\prod_{S\in X}\prod_{ij\in R_S}\indic_{\mc{B}_{ij}}\right)\prod_{i\in V_X}e^{\beta \g_\lambda(x_i-y_i)}\indic_{|x_i-y_i|\leq \ve_0\Cut}\dd x_i \dd y_i. 
	\end{multline*}
	There exists $C>0$ depending on $k$ such that
	\begin{equation}\label{eq:mm1}
		J_1((R_S))\leq C \max_{(R_S^m),(T^{a,m}),(T^{b,m})} J_2((R_S^m),(T^{a,m}),(T^{b,m})),
	\end{equation}
	where the maximum is taken over \textcolor{black}{possibly empty forests} $(R_S^m), (T^{a,m}),(T^{b,m})$ such that $\sqcup_{m=1}^4 R_S^m=R_S$, $\sqcup_{m=1}^4T^{a,m}=T^a$ and $\sqcup_{m=1}^4T^{b,m}=T^b$. Fix $(R_S^m), (T^{a,m}),(T^{b,m})$ accordingly. 
	
	Let
	\begin{equation}\label{def:zij' b}
		z_{ij}=\begin{cases}
			x_{i}-x_j & \text{if $ij\in T^{a,1}\cup T^{b,1}\cup \textcolor{black}{\bigcup_{S\in X}R_S^1}$}\\
			x_i-y_j & \text{if $ij\in  T^{a,2}\cup T^{b,2}\cup \textcolor{black}{\bigcup_{S\in X}R_S^2}  $}\\
			y_i - x_j & \text{if $ij\in T^{a,3}\cup T^{b,3}\cup \textcolor{black}{\bigcup_{S\in X}R_S^3}$}\\
			y_i-y_j & \text{if $ij\in T^{a,4}\cup T^{b,4}\cup \textcolor{black}{\bigcup_{S\in X}R_S^4}$}.
		\end{cases}
	\end{equation}
	Fix some arbitrary $i_0\in V_X$. We perform the change of variables
	\begin{equation*}
		(x_i,y_i)_{i\in V_X}\mapsto (x_{i_0},(z_{ij})_{ij\in \textcolor{black}{\bigcup_{S\in X}R_S}\cup T^a\cup T^b}, (\vr_i)_{i\in V_X}).
	\end{equation*}
	\begingroup\color{black}
This change of variables has Jacobian $1$, and the root integral is at most $N$. An edge of some $R_S$ contributes at most $C\min(\rho_i,\rho_j)^2$. For an edge of $T^b$, with $s=\min(\rho_i,\rho_j)$ and $t=\max(\rho_i,\rho_j)$,
\[
\int_{\dR^2}\widehat b_{ij}\dd z_{ij}\leq Cs^2(1+\log(t/s)).
\]
Since $d_{ij}\geq\rho_*$ on $\mc B_{ij}^c$, if $q=|T^a|\geq1$, then the integrals over the odd edges are bounded by
\[
\int_{[\rho_*,\infty)^q}\min\left(1,\frac{\rho^2}{\max_jt_j^2}\right)\prod_{j=1}^q\frac{\dd t_j}{t_j}
\leq C_q(1+\log(\rho/\rho_*))^q.
\]
There are at most $q+|T^b|\leq k-1$ logarithmic factors. Thus
\endgroup
	\begin{multline*}
		J_2((R_S^m),(T^{a,m}),(T^{b,m}))\leq CN\max_{S_0\in X, (\iota_S)}\\ \int \textcolor{black}{(1+\log(\rho/\rho_*))^{k-1}} \prod_{l=1}^n \prod_{S\in X_l:S\neq \hat{S}_l} r_S^2 \prod_{S\in X}\prod_{ij\in R_S}\max(\min(r_i,r_j),\lambda)^2\prod_{ij\in T^b}\max(\min(r_i,r_j),\lambda)^2 \\ \times \prod_{S\in X}\textcolor{black}{\prod_{i\in S}}\indic_{r_i\leq r_{\iota_S}\leq r_{\iota_{S_0}}\leq \ve_0\Cut}\prod_{i\in V_X}e^{\beta \g_\lambda(r_i)}\indic_{r_i\leq \ve_0\Cut}\dd \vr_i.
	\end{multline*}
	By Lemma~\ref{lemma:ordering tree},  since $T^b\in \textcolor{black}{\mathsf{T}_0}^{\Coarse_X(X_1,\ldots,X_n)}$, 
	\begin{equation*}
		\prod_{ij\in T^b}\max(\min(r_i,r_j),\lambda)^2\leq \prod_{S\in \Coarse_X(X_1,\ldots,X_n):S\neq S_0'}\max(r_S,\lambda)^2,
	\end{equation*}
	where $S_0'$ is the index $S\in \Coarse_X(X_1,\ldots,X_n)$ such that $r_S^2$ is maximal. Clearly, recalling that $\textcolor{black}{S_0}$ stands for the $S\in X$ such that $r_{S}$ is maximal, we have $r_{S_0'}=r_{S_0}$. Therefore,
	\begin{equation*}
		\prod_{ij\in T^b}\max(\min(r_i,r_j),\lambda)^2\leq \frac{1}{\max(r_{S_0},\lambda)^2}\prod_{S\in \Coarse_X(X_1,\ldots,X_n)}\max(r_S,\lambda)^2.
	\end{equation*}
	Thus,
	\begin{equation*}
		\prod_{l=1}^n \prod_{S\in X_l:S\neq \hat{S}_l} r_S^2 \prod_{ij\in T^b}\max(\min(r_i,r_j),\lambda)^2\leq \prod_{S\in X:S\neq S_0}\max(r_S,\lambda)^2.
	\end{equation*}
	Applying Lemma~\ref{lemma:ordering tree} to every $S\in X$, we get 
	\begin{equation*}
		\prod_{ij\in R_S}\max(\min(r_i,r_j),\lambda)^2\leq \prod_{i\in S:i\neq \iota_S}\max(r_i,\lambda)^2.
	\end{equation*}
	Therefore, recalling that ${\iota_{S_0}}$ is the index $i\in V_X$ such that $r_i$ is maximal, we get 
	\begin{equation}\label{eq:RS}
		\prod_{l=1}^n \prod_{S\in X_l:S\neq \hat{S}_l} r_S^2 \prod_{ij\in T^b}\max(\min(r_i,r_j),\lambda)^2  \prod_{ij\in \cup_{S\in X}R_S}\max(\min(r_i,r_j),\lambda)^2\leq \prod_{i\in V_X:i\neq \iota_{S_0}}\max(r_i,\lambda)^2. 
	\end{equation}
	It follows that
	\begin{multline}\label{eq:mm2}
		J_2((R_S^m),(T^{a,m}),(T^{b,m}))\\ \leq CN\max_{i_0\in V_X} \int \textcolor{black}{(1+\log(\rho/\rho_*))^{k-1}} \prod_{i\in V_X:i\neq i_0}\max(r_i,\lambda)^2 \indic_{r_i\leq r_{i_0}} \prod_{i\in V_X}e^{\beta \g_\lambda(r_i)}\indic_{r_i\leq \ve_0\Cut}\dd \vr_i.
	\end{multline}
	\begingroup\color{black}
For $\beta<4$ and $u\geq\lambda$,
\[
\int_0^u t(t\vee\lambda)^{2-\beta}(1+\log(u/(t\vee\lambda)))^{k-1}\dd t
\leq C_{\beta,k}u^{4-\beta}.
\]
We apply this bound with $\rho_i=\rho_*$ for an index $i\neq i_0$, and we bound the right-hand side of \eqref{eq:mm2} by
\[
CN\left(\lambda^{(4-\beta)k-2}+\int_\lambda^h u^{(4-\beta)k-3}\dd u\right).
\]
For $\beta\geq4$, a rescaling by $\lambda$ bounds it by
\[
CN\lambda^{(4-\beta)k-2}\left[1+\int_1^\infty u^{1-\beta}(\log(2u))^k
\left(1+\int_1^u t^{3-\beta}\dd t\right)^{k-1}\dd u\right]
\leq CN\lambda^{(4-\beta)k-2}.
\]
Thus, by Lemma~\ref{lemsupplapp},
\begin{equation}\label{eq:mm3}
J_2((R_S^m),(T^{a,m}),(T^{b,m}))\leq CN\lambda^{(2-\beta)k}\gamma_{\beta,\lambda,k}.
\end{equation}
For $k\leq p^*(\beta)$, these integrals converge with $h=\infty$.
\endgroup

	Combining \eqref{eq:mm0}, \eqref{eq:mm1}, \eqref{eq:mm2} and \eqref{eq:mm3} proves \eqref{eq:bound J1}.
	
	{\color{black}
For $\textcolor{black}{J_{\mathrm{ext}}}$, let $R=\max_{i\in V_X}(r_i\vee\lambda)$ and $R_*=\min_i(r_i\vee\lambda)$, and note that $\lambda\leq h$ for $\lambda$ small. At fixed displacements and tree coordinates, the integral in the common translation $x_{i_0}$ satisfies
    \[
        \int_\Lambda\min\left(1,\frac{R^2}{\min_i d_{i,\bad}^2}\right)\dd x_{i_0}
        \leq\frac{CNR^2}{h^2}
        +Ck|V_\bad|R^2\left(1+\log\frac hR\right).
    \]
    The first term comes from the cap in \eqref{eq:dibad}, and each of the at most $4k|V_\bad|$ translated bad endpoints contributes at most
\[
2\pi\left(\int_0^R t\dd t+R^2\int_R^h\frac{\dd t}{t}\right)
\leq CR^2(1+\log(h/R)).
\] With the indicators of $\mc B_{ij}^c$ and the cutoffs $d_{ij}\leq\textcolor{black}{17}h$, the tree integrals contribute at most $C(1+\log(\textcolor{black}{17}h/R_*))^{k-1}$. By \eqref{eq:RS}, the radius factors are at most $\prod_i(r_i\vee\lambda)^2/R^2$, whose denominator cancels the $R^2$ above. Then \eqref{eq:radial cutoff gain}, with the extra logarithm from $x_{i_0}$, gives \eqref{eq:bound J2}.
}

\end{proof}

\subsection{Integrals for dipole clusters}

\textcolor{black}{For $V\subset[N]$, $E\subset\binom V2$ and $\ve_0\in(0,1)\cup\{\infty\}$, we define}
\begin{equation}\label{def:IkE}
		\textcolor{black}{I_{\mathrm{dip},\ve_0}(V;E)}\coloneqq \int_{(\Lambda^{2})^{|V|}} \prod_{ij\in E }  f_{ij}^v \prod_{i\in V}e^{\beta \g_\lambda(x_i-y_i)}\textcolor{black}{\indic_{r_i\leq\ve_0\Cut}}\dd x_i \dd y_i.
	\end{equation}

\begin{lemma}\label{lemma:int dipole}
	Let $\beta\in (2,+\infty)$. Let $p^*(\beta)$ be as in \eqref{defpstar}. Let $V\subset [N]$ be such that $|V|\leq p^*(\beta)$. Let $E\in \mc{G}_c(V)$, where $\mc{G}_c(V)$ is as in Definition~\ref{def:dipole activity}. \textcolor{black}{Recall \eqref{def:IkE}.}
	
	\textcolor{black}{There exists a constant $C>0$ depending on $\beta$ and $|V|$ such that, for every $\lambda>0$,} 
	\begin{equation*}
		|\textcolor{black}{I_{\mathrm{dip},\infty}(V;E)}|\leq CN\lambda^{(2-\beta)|V|+2(|V|-1)}.
	\end{equation*}
\end{lemma}
\medskip

\begingroup\color{black}
\begin{proof}
	Since $\textcolor{black}{I_{\mathrm{dip},\infty}(V;E)}$ does not depend on $M$, we may take $M=30$ in this proof. Let $k=|V|$.
	We split $E$ into far edges $E'$ and near edges $E\setminus E'$. For each $E'\subset E$, let $X$ be the collection of connected components of $(V,E\setminus E')$ and let $D=E\cap\mc{E}^{\inter}(X)$. For $S\in X$, let $F_S$ be the product of all $E$-bonds internal to $S$, with their near/far indicators. Each $F_S$ is invariant under simultaneous reversal of the dipoles in $S$. Since $\mc{B}_{ij}^c\subset\mc{A}_{ij}$, the parity argument of Lemma~\ref{lemma:cancellation odd}, applied only to $D$, gives
	\begin{multline}\label{eq:xq2}
		\textcolor{black}{I_{\mathrm{dip},\infty}(V;E)}=\sum_{E'\subset E}\sum_{\substack{E^{\odd}\subset D\\ E^{\odd}\in\Eul^X}}
		\int_{(\Lambda^2)^k}\prod_{S\in X}F_S
		\prod_{ij\in E^{\odd}}a_{ij}^v
		\prod_{ij\in D\setminus E^{\odd}}b_{ij}^v\\
		\times\prod_{ij\in D}\indic_{\mc{B}_{ij}^c}
		\prod_{i\in V}e^{\beta\g_\lambda(x_i-y_i)}\dd x_i\dd y_i.
	\end{multline}
	Set
	\begin{equation*}
		\widehat a_{ij}=\frac{r_ir_j}{d_{ij}\max(d_{ij},r_i,r_j)}\indic_{\mc{B}_{ij}^c},
		\qquad \widehat b_{ij}=\widehat a_{ij}^{\,2}+\indic_{\mc{B}_{ij}}.
	\end{equation*}
	Both weights are at most $1$. By Lemma~\ref{lemma:vij},
	\begin{equation*}
		|a_{ij}^v|\indic_{\mc{B}_{ij}^c}\le C_\beta\widehat a_{ij},
		\qquad |b_{ij}^v|\indic_{\mc{B}_{ij}^c}\le C_\beta\widehat b_{ij}.
	\end{equation*}
	On $\mc{A}_{ij}$, we have $v_{ij}\ge-C$\edgreen{ by \eqref{eq:bond-bound}}, so $|F_S|\le C_{\beta,k}\indic_{\mc{B}_S}$.

	Let $X_1,\ldots,X_n$ be the components with at least two vertices of the odd quotient multigraph on $X$, and let $E_l=E^{\odd}\cap\mc{E}^{\inter}(X_l)$. Set $R_i=\max(r_i,\lambda)$ and $R_S=\max_{i\in S}R_i$ for every nonempty $S\subset V$. For each $l$, write $\mathcal{R}_l=\max_{S\in X_l}R_S$. We claim that
	\begin{equation}\label{eq:prodE1l'}
		\prod_{ij\in E_l}\widehat a_{ij}
		\le C_k\frac{\prod_{S\in X_l}R_S^2}{\mathcal{R}_l^2}
		\sum_{\substack{T_l^a\in\mathsf{T}^{X_l}\\ T_l^a\subset E_l}}
		\min\left(1,\frac{\mathcal{R}_l^2}{\max_{e\in T_l^a}d_e^2}\right)
		\prod_{ij\in T_l^a}\frac{\indic_{\mc{B}_{ij}^c}}{d_{ij}^2}.
	\end{equation}
	Indeed, since the odd quotient is bridgeless, each uncovered component has at least two edges to the covered blocks. We can therefore start with a cycle through a block attaining $\mathcal{R}_l$ and add ears with new interior blocks until all blocks are covered, counting the initial cycle as an ear and allowing closed ears and two-edge cycles. We discard the unused factors $\widehat a_{ij}$ and then enlarge the radii to the block maxima, since $st/[d\max(d,s,t)]$ is nondecreasing in $s,t$.
	On an ear $(S_0,\ldots,S_h)$, let $s_j=R_{S_j}$ and let $d_j$ be the distance of its $j$-th edge. For $h\ge2$,
	\begin{equation*}
		\prod_{j=1}^h\frac{s_{j-1}s_j}{d_j\max(d_j,s_{j-1},s_j)}
		\le\frac12\prod_{j=1}^{h-1}s_j^2
		\left(\prod_{j=2}^h d_j^{-2}+\prod_{j=1}^{h-1}d_j^{-2}\right).
	\end{equation*}
	To see this, we bound the first and last maxima below by $s_0$ and $s_h$, the other ones by $d_j$, and use $2/(d_1d_h)\leq d_1^{-2}+d_h^{-2}$. Let $e_0$ maximize $d_e$ over these ears. If $d_{e_0}\ge \mathcal{R}_l$, then we delete $e_0$ from its ear instead and bound $s_0s_h\le \mathcal{R}_l^2$, which leaves the factor $\mathcal{R}_l^2/d_{e_0}^2$. Deleting one edge from each ear gives a spanning tree whose maximal distance is at most $d_{e_0}$. If $d_{e_0}<\mathcal{R}_l$, the minimum in \eqref{eq:prodE1l'} is therefore $1$. This proves~\eqref{eq:prodE1l'}.

	Using~\eqref{st3} and $\widehat b_{ij}\le1$, we reduce~\eqref{eq:xq2} to near trees $T_S$ in each $S\in X$, odd trees $T_l^a$, and an even tree $T^b$ on $\Coarse_X(X_1,\ldots,X_n)$. Set $T^a=\bigcup_lT_l^a$. Their union $\bigcup_ST_S\cup T^a\cup T^b$ is a tree on $V$. Write $R=\max_iR_i$ and $m=\min_iR_i$. The product of the factors $\min(1,\mathcal R_l^2/\max_{e\in T_l^a}d_e^2)$ is at most
	\[
		\min\left(1,\frac{R^2}{\max_{e\in T^a}d_e^2}\right)\,.
	\]
	We integrate in the variables \eqref{def:zij' b} of the proof of Lemma~\ref{lemma:integral small -}. As there, the root contributes at most $N$ and an edge of some $T_S$ at most $C\min(R_i,R_j)^2$. For an edge of $T^b$, with $s=\min(R_i,R_j)$ and $t=\max(R_i,R_j)$,
	\[
		\int_{\dR^2}\widehat b_{ij}\dd z_{ij}\le Cs^2\left(1+\log\frac ts\right).
	\]
	If $q=|T^a|\ge1$, then the integrals over the odd edges are bounded by
	\[
		\int_{[m,\infty)^q}\min\left(1,\frac{R^2}{\max_jt_j^2}\right)
		\prod_{j=1}^q\frac{\dd t_j}{t_j}
		\le C_q\left(1+\log\frac Rm\right)^q.
	\]
	By Lemma~\ref{lemma:ordering tree}, the factors of the radii are bounded by $\prod_iR_i^2/R^2$, and there are at most $q+|T^b|\le k-1$ logarithmic factors.
	Since $e^{\beta\g_\lambda(r)}\le C_\beta\max(r,\lambda)^{-\beta}$, we scale the dipole vectors by $\lambda$ and integrate the radii below the largest one, which gives
\begin{equation*}|\textcolor{black}{I_{\mathrm{dip},\infty}(V;E)}|\le C_{\beta,k}N\lambda^{(4-\beta)k-2} \times\left[1+\int_1^\infty t^{1-\beta}(\log(2t))^k \left(1+\int_1^t s^{3-\beta}\dd s\right)^{k-1}\dd t\right].
\end{equation*}
	For $\beta<4$, the integral converges, since $(4-\beta)k-3<-1$ for $k\le p^*(\beta)$. For $\beta=4$, the integrand is $t^{-3}$ times logarithms, and, for $\beta>4$, the inner integral is bounded. This proves the result.
\end{proof}
\endgroup

\subsection{Integrals for activity errors}\label{sub:int error}

\begin{lemma}\label{lemma:general exp M}
	Let $\beta\in (2,\infty)$ and let $p^*(\beta)$ be as in \eqref{defpstar}. Let $S\subset \textcolor{black}{[N]\setminus V_\bad}$ be such that \textcolor{black}{$1\leq|S|\leq p^*(\beta)$.} Let $\tau^Z=(\tau_i^Z)_{i\in\textcolor{black}{[N]\setminus V_\bad}}$ be as in Definition~\ref{def:vijZ}. Recall $\tau_0$ from Definition~\ref{def:vijtilde} and set $\textcolor{black}{\tilde{\tau}_i\coloneqq \tau_0}$ for every $i\in [N]$. \textcolor{black}{Let either $(w,\tau)=(\tilde v,\tilde\tau)$ and $V_\bad=\emptyset$, or $(w,\tau)=(v^Z,\tau^Z)$.} \textcolor{black}{Fix $L\in\dR$ and recall $\textcolor{black}{I_{\mathrm{mult},\ve_0}(S;w,\tau,L,Z)}$ from \eqref{def:I1wL}.}
	
	\textcolor{black}{Then, for $\lambda$ small enough, there exist $C>0$ depending only on $\beta,M,L,|S|$ and $C_{\ve_0}>0$ also depending on $\ve_0$ such that}
	\begin{equation*}
		|  \textcolor{black}{I_{\mathrm{mult},\ve_0}(S;w,\tau,L,Z)}-\textcolor{black}{I_{\mathrm{mult},\infty}^{0}}(S)|\leq \textcolor{black}{\lambda^{(2-\beta)|S|}\bigl(C_{\ve_0}N\Cut^{-2}+C\ve_0^{\alpha(\beta)|S|}|V_\bad|\bigr).} 
	\end{equation*}
\end{lemma}

\medskip

\begin{proof}
\begingroup\color{black}
By the geometric argument at the beginning of the proof of Proposition~\ref{prop:bounded upper}, we have $\tau_i^Z\geq r_i/4$ on the integration support, and $\tilde\tau_i=8\ve_0\Cut\geq r_i$. For the expansion below, we extend $\tau_i^Z$, and hence $v^Z$, outside the event $\{i\in I_\good^Z\}$ by
\[
\tau_i^Z=\max\left(\tfrac14d_{i,\bad},\lambda,\tfrac14r_i\right).
\]
On the event $\{i\in I_\good^Z\}$, the right-hand side is the original $\tau_i^Z$ of \eqref{deftauz}, so the integrals are unchanged. In the whole expansion, Lemmas~\ref{lemma:error'} and~\ref{lemma:errortilde} then give
\[
|w_{ij}-v_{ij}|\leq C\frac{r_ir_j}{\max(\tau_i,\tau_j)^2}\leq C,
\]
since $\tau_i\geq r_i/4$. Thus $|e^{\beta(v_{ij}-w_{ij})}-1|\leq C|v_{ij}-w_{ij}|$, and the exponential factors of the expansion are bounded.
\endgroup

	\paragraph{\bf{Step 1: starting point}}
	We have
	\begin{multline*}
		\prod_{ i,j\in S:i<j}e^{-\beta w_{ij}}\indic_{\mc{A}_{ij}}\prod_{i\in S}e^{L\frac{|x_i-y_i|^2}{\tau_i^2 }} \prod_{i\in S}\indic_{i\in I_\good^Z} =    \prod_{ i,j\in S:i<j}e^{-\beta v_{ij}}\indic_{\mc{A}_{ij}}( e^{\beta (v_{ij}-w_{ij})}-1+1)\\ \times \prod_{i\in S}\Bigl(e^{L\frac{|x_i-y_i|^2}{\tau_i^2}}-1+1\Bigr) \prod_{i\in S}(1-\indic_{i\notin I_\good^Z}).
	\end{multline*}
	Hence,
	\begin{equation}\label{eq:Idiff1}
		\textcolor{black}{I_{\mathrm{mult},\ve_0}(S;w,\tau,L,Z)}-\textcolor{black}{I_{\mathrm{mult},\ve_0}^{0}}(S)=\sum_{(F,U_1,U_2)\neq (\emptyset,\emptyset,\emptyset)} \textcolor{black}{D}(F,U_1,U_2),
	\end{equation}
	where 
	\begin{multline*}
		\textcolor{black}{D}(F,U_1,U_2)= \int_{(\Lambda^2)^{|S|}}\indic_{\mc{B}_S}\textcolor{black}{\prod_{\substack{i,j\in S\\i<j}}}e^{-\beta v_{ij}}\indic_{\mc{A}_{ij}} \prod_{ij\in F}(e^{\beta (v_{ij}-w_{ij})}-1)\indic_{\mc{A}_{ij}}\prod_{i\in U_1} \Bigl(e^{L\frac{|x_i-y_i|^2}{\tau_i^2}}-1\Bigr)\\ \times  \prod_{i\in U_2}(-1)\indic_{i\notin I_\good^Z} \prod_{i\in S} e^{\beta \g_\lambda(x_i-y_i)}\indic_{|x_i-y_i|\leq \ve_0\Cut} \dd x_i \dd y_i.
	\end{multline*}
	Fix $(F,U_1,U_2)\neq (\emptyset,\emptyset,\emptyset)$. Suppose first that $U_1\neq \emptyset$ and let $i_0\in U_1$. Then, there exists a constant $C>0$ depending on \textcolor{black}{$\beta,M,L$} and $|S|$ such that 
	\begin{flalign}
		&& \textcolor{black}{|\textcolor{black}{D}(F,U_1,U_2)|} \leq C\int_{(\Lambda^2)^{|S|}}\indic_{\mc{B}_S}\prod_{ i,j\in S:i<j}\indic_{\mc{A}_{ij}} \frac{|x_{i_0}-y_{i_0}|^2}{\tau_{i_0}^2} \prod_{i\in S} e^{\beta \g_\lambda(x_i-y_i)}\indic_{|x_i-y_i|\leq \ve_0\Cut}\dd x_i \dd y_i \notag\\
		&& \leq C\int_{(\Lambda^2)^{|S|}}\indic_{\mc{B}_S}\prod_{ i,j\in S:i<j}\indic_{\mc{A}_{ij}} \frac{\max_{i\in S}|x_{i}-y_{i}|^2}{\tau_{i_0}^2} \prod_{i\in S} e^{\beta \g_\lambda(x_i-y_i)}\indic_{|x_i-y_i|\leq \ve_0\Cut}\dd x_i \dd y_i.\label{eq:Idiff2}
	\end{flalign}
	Suppose that $U_2\neq \emptyset$ and let $i_0\in U_2$. Then, there exists a constant $C>0$ depending on \textcolor{black}{$\beta,M,L$} and $|S|$ such that 
	\begin{equation}\label{eq:Idiff2 bis} 
		\textcolor{black}{|\textcolor{black}{D}(F,U_1,U_2)|}\leq C\int_{(\Lambda^2)^{|S|}}\indic_{\mc{B}_S}\prod_{ i,j\in S:i<j}\indic_{\mc{A}_{ij}}\indic_{i_0\notin I_\good^Z}\prod_{i\in S} e^{\beta \g_\lambda(x_i-y_i)}\indic_{|x_i-y_i|\leq \ve_0\Cut}\dd x_i \dd y_i.
	\end{equation}
	Suppose that $F\neq \emptyset$ and let $i_0 j_0\in F$. Then
	\begin{equation}
		|e^{\beta (v_{i_0j_0}-w_{i_0j_0})}-1|\leq C|v_{i_0j_0}-w_{i_0j_0}|.
	\end{equation}
	\textcolor{black}{By the last assertions of Lemmas~\ref{lemma:error'} and~\ref{lemma:errortilde},} 
	\begin{equation}\label{eq:UU}
		|w_{ij}-v_{ij}|\leq C\frac{r_{i}r_{j}}{\max(\tau_{i}, \tau_{j})^2}.
	\end{equation}

	We conclude from \eqref{eq:Idiff1}--\eqref{eq:UU}  and $\tau_i\ge \frac{r_i}{4}$ that there exists a constant $C>0$ depending on \textcolor{black}{$\beta,M,L$} and $|S|$ such that
	\begin{multline}\label{eq:a intstep2}
		|  \textcolor{black}{I_{\mathrm{mult},\ve_0}(S;w,\tau,L,Z)}-\textcolor{black}{I_{\mathrm{mult},\ve_0}^{0}}(S)|\\\leq C\max_{i_0\in S}\int_{(\Lambda^2)^{|S|}}\indic_{\mc{B}_S}\prod_{ i,j\in S:i<j}\indic_{\mc{A}_{ij}} \left(\frac{|x_{i_0}-y_{i_0}|^2}{\tau_{i_0}^2}+\indic_{i_0\notin I_\good^Z}\right)\\ \times \prod_{i\in S} e^{\beta \g_\lambda(x_i-y_i)}\indic_{|x_i-y_i|\leq \ve_0\Cut}\dd x_i \dd y_i.
	\end{multline}

	On the other hand, there exists a constant $C>0$ depending on $\beta$ and $|S|$ such that
	\begin{equation}\label{eq:deuxieme bout}
		| \textcolor{black}{I_{\mathrm{mult},\ve_0}^{0}}(S)-\textcolor{black}{I_{\mathrm{mult},\infty}^{0}}(S)| \leq C\max_{i_0\in S}\int_{(\Lambda^2)^{|S|}} \indic_{\mc{B}_S}\prod_{i\in S}e^{\beta \g_\lambda(x_i-y_i)}\indic_{|x_{i_0}-y_{i_0}|\geq \ve_0\Cut}\prod_{i\in S}\dd x_i \dd y_i.
	\end{equation}

	\paragraph{\bf{Step 2: integration of \eqref{eq:a intstep2}}}
	We then split the integral on the right-hand side of \eqref{eq:a intstep2} according to the label of the largest $r_i$, $i\in S$. Again, by a union bound, $\indic_{\mc{B}_S}\leq \sum_{T\in \mc{T}_c(S)}\prod_{ij\in T}\indic_{\mc{B}_{ij}}$. 
	We deduce from this that there exists a constant $C>0$ depending on $|S|$ such that
	\begin{multline}\label{eq:qq5}
		\int_{(\Lambda^2)^{|S|}}\indic_{\mc{B}_S}\prod_{ i,j\in S:i<j}\indic_{\mc{A}_{ij}} \left(\frac{|x_{i_0}-y_{i_0}|^2}{\tau_{i_0}^2}+\indic_{i_0\notin I_\good^Z}\right) \prod_{i\in S} e^{\beta \g_\lambda(x_i-y_i)}\indic_{|x_i-y_i|\leq \ve_0\Cut}\dd x_i \dd y_i\\ \leq C \max_{T\in \mc{T}_c(S)} \max_{\iota_S \in S} J(\iota_S,T,i_0),  
	\end{multline}
	where 
	\begin{multline*}
		J(\iota_S,T,i_0)\coloneqq  \int_{(\Lambda^2)^{|S|}}\prod_{ij\in T}\indic_{\mc{B}_{ij}}\prod_{ i,j\in S:i<j}\indic_{\mc{A}_{ij}} \left(\frac{|x_{i_0}-y_{i_0}|^2}{\tau_{i_0}^2}+\indic_{i_0\notin I_\good^Z}\right)\\ \times \prod_{i\in S} e^{\beta \g_\lambda(x_i-y_i)}\indic_{|x_i-y_i|\leq |x_{\iota_S}-y_{\iota_S} |\leq \ve_0\Cut}\dd x_i \dd y_i.
	\end{multline*}
	Define
	\begin{equation*}
		\begin{aligned}
			 & \mc{D}_{ij}^1 =\{(x_q,y_q)_{q\in S}: d_{ij}= |x_i-x_j|\},\qquad \mc{D}_{ij}^2 =\{(x_q,y_q)_{q\in S}: d_{ij}= |x_i-y_j|\},\\
			 & \mc{D}_{ij}^3 =\{ (x_q,y_q)_{q\in S}:d_{ij}= |y_i-x_j|\},\qquad \mc{D}_{ij}^4 =\{ (x_q,y_q)_{q\in S}: d_{ij}= |y_i-y_j|\}.
		\end{aligned}
	\end{equation*}
	We have
	\begin{equation}\label{eq:qq4}
		J(\iota_S,T,i_0)\leq C\max_{T^1,\ldots,T^4:T^1\sqcup \cdots \sqcup T^4=T}\textcolor{black}{J'}(\iota_S,(T^l)_{1\leq l\leq 4},i_0),
	\end{equation}
	where 
	\begin{multline*}
		\textcolor{black}{J'}(\iota_S,(T^l)_{1\leq l\leq 4},i_0)\\=  \int_{(\Lambda^2)^{|S|}}\left(\prod_{l=1}^4 \prod_{ij\in T^l}\indic_{d_{ij} \le M \max(\min(r_i,r_j),\lambda) }\indic_{\mc{D}_{ij}^l}\right) \left(\frac{|x_{\iota_S}-y_{\iota_S}|^2}{\tau_{i_0}^2}+\indic_{i_0\notin I_\good^Z}\right) \\ \times \prod_{i\in S} e^{\beta \g_\lambda(x_i-y_i)}\indic_{|x_i-y_i| \leq |x_{\iota_S}-y_{\iota_S} |\leq \ve_0\Cut}\dd x_i \dd y_i. 
	\end{multline*}
	We then let
	\begin{equation*}
		z_{ij}=\begin{cases}
			x_{i}-x_j & \text{if $ij\in T^1$}\\
			x_i-y_j & \text{if $ij\in  T^2 $}\\
			y_i - x_j & \text{if $ij\in T^3$}\\
			y_i-y_j & \text{if $ij\in T^4$}
		\end{cases}
	\end{equation*}
	and, as before, perform the change of variables
	$(x_i,y_i)_{i\in S}\mapsto (x_{i_0},(z_{ij})_{ij\in T},(\vec{r}_i)_{i\in S}).$
	Integrating over the variables $(z_{ij})_{ij\in T}$, we find $C>0$ depending on $\beta,M$ and $|S|$ such that
	\begin{align*}
		\textcolor{black}{J'(\iota_S,(T^l)_{1\leq l\leq 4},i_0)}  &\leq C\int \frac{1}{\tau_{i_0}^2} \prod_{i\in S}\max(r_i,\lambda)^2 \prod_{i\in S}e^{\beta \g_\lambda(r_i)}\indic_{r_i\leq \ve_0\Cut}\dd \vec{r}_i\dd x_{i_0}\\ &+C \int \indic_{i_0\notin I_\good^Z}\prod_{i\in S:i\neq \iota_S}\max(r_i,\lambda)^2 \prod_{i\in S}e^{\beta \g_\lambda(r_i)}\indic_{r_i\leq r_{\iota_S}\leq \ve_0\Cut}\dd \vr_i \dd x_{i_0}.
	\end{align*}
	Let us bound the second integral. First, if $i_0\notin I_\good^Z$, then by definition \eqref{def:IgoodZ}, 
	\begin{equation}\label{eq:st9}
		x_{i_0}\in \bigcup_{j\in V_\bad}\Bigl(B(x_j,2M\max(r_{i_0},\lambda) )\cup B(y_j,2M\max(r_{i_0},\lambda))\Bigr).
	\end{equation}
	{\color{black}
Set $h=\ve_0\Cut$, $k=|S|$ and $R_i=r_i\vee\lambda$. For $\lambda$ small, we have $\lambda\leq h$. By \eqref{eq:st9} and the definition \eqref{deftauz} of $\tau_{i_0}$,
\begin{equation}\label{eq:root bad integral}
\int_\Lambda\indic_{i_0\notin I_\good^Z}\dd x_{i_0} \leq C|V_\bad|R_{i_0}^2,\qquad \int_\Lambda\frac{\dd x_{i_0}}{\tau_{i_0}^2} \leq\frac{CN}{h^2} +C|V_\bad|\left(1+\log\frac h{R_{i_0}}\right).
\end{equation}
    For fixed $k\geq1$ and $q\geq0$, integrating the radii gives
    \begin{multline}\label{eq:radial cutoff gain}
        \int_{[0,h]^k}\left(1+\log\frac{\textcolor{black}{17}h}{\min_i(r_i\vee\lambda)}\right)^q
        \prod_{i=1}^k r_i(r_i\vee\lambda)^{2-\beta}\dd r_i\\
        \leq C
        \begin{cases}
            h^{(4-\beta)k},&2<\beta<4,\\
            (1+\log(h/\lambda))^{k+q},&\beta=4,\\
            \lambda^{(4-\beta)k}(1+\log(h/\lambda))^q,&\beta>4
        \end{cases}
        \leq C\ve_0^{\alpha(\beta)k}\lambda^{(2-\beta)k},
    \end{multline}
    for $\lambda$ small enough, with $C$ independent of $\ve_0$. For $\beta<4$, we scale by $h$ and use $h^{(4-\beta)k}=\ve_0^{2\alpha(\beta)k}\lambda^{(2-\beta)k}$. For $\beta\geq4$, we split at $\lambda$ and use $\log(h/\lambda)\leq(2p_0+1)|\log\lambda|$.

    Since $R_{i_0}\leq R_{\iota_S}$, these estimates give
\begin{equation}\label{eq:qq1} \int\indic_{i_0\notin I_\good^Z} \prod_{i\in S:i\ne\iota_S}R_i^2 \prod_{i\in S}e^{\beta\g_\lambda(r_i)} \indic_{r_i\leq r_{\iota_S}\leq h}\prod_{i\in S}\dd\vr_i\dd x_{i_0} \leq C\ve_0^{\alpha(\beta)k}|V_\bad|\lambda^{(2-\beta)k},
\end{equation}
    and
\begin{equation}\label{eq:qq3} \int\frac{1}{\tau_{i_0}^2} \prod_{i\in S}\bigl(R_i^2e^{\beta\g_\lambda(r_i)}\indic_{r_i\leq h}\dd\vr_i\bigr)\dd x_{i_0} \leq\lambda^{(2-\beta)k} \bigl(C_{\ve_0}N\Cut^{-2}+C\ve_0^{\alpha(\beta)k}|V_\bad|\bigr).
\end{equation}
    Combining \eqref{eq:qq1}, \eqref{eq:qq3}, \eqref{eq:qq4} and \eqref{eq:qq5} with \eqref{eq:a intstep2}, we obtain
\begin{equation}\label{eq:CC1} |\textcolor{black}{I_{\mathrm{mult},\ve_0}(S;w,\tau,L,Z)}-\textcolor{black}{I_{\mathrm{mult},\ve_0}^{0}}(S)| \leq\lambda^{(2-\beta)k} \bigl(C_{\ve_0}N\Cut^{-2}+C\ve_0^{\alpha(\beta)k}|V_\bad|\bigr).
\end{equation}
}

	\paragraph{\bf{Step 3: integration of \eqref{eq:deuxieme bout}}}
	
	Splitting the phase space as in Step 2 shows that 
	\begin{align*}
		& \int_{(\Lambda^2)^{|S|}} \indic_{\mc{B}_S}\prod_{i\in S}e^{\beta \g_\lambda(x_i-y_i)}\indic_{|x_{i_0}-y_{i_0}|\geq \ve_0\Cut}\prod_{i\in S}\dd x_i \dd y_i\\	&\qquad\leq \textcolor{black}{CN\max_{j_0\in S}\int \prod_{i\in S:i\neq j_0}\max(r_i,\lambda)^2 \indic_{r_{j_0}\geq \ve_0\Cut} \prod_{i\in S}e^{\beta \g_\lambda(r_i)}\indic_{r_i\leq r_{j_0}} \dd \vr_i}.
	\end{align*}
    {\color{black}
We integrate the inner radii and bound the preceding integral by
    \begin{equation*}
		CN\int_h^\infty r^{1-\beta}
        \left(\int_0^r s(s\vee\lambda)^{2-\beta}\dd s\right)^{k-1}\dd r \qquad\leq CN\begin{cases}
            h^{(4-\beta)k-2},&2<\beta<4,\\
            h^{-2}(1+\log(h/\lambda))^{k-1},&\beta=4,\\
            h^{2-\beta}\lambda^{(4-\beta)(k-1)},&\beta>4.
        \end{cases}
	\end{equation*}
    Each bound is at most $C_{\ve_0}N\lambda^{(2-\beta)k}\Cut^{-2}$. Hence,
    }
    \begin{equation}\label{eq:one large}
		\int_{(\Lambda^2)^{|S|}} \indic_{\mc{B}_S}\prod_{i\in S}e^{\beta \g_\lambda(x_i-y_i)}\indic_{|x_{i_0}-y_{i_0}|\geq \ve_0\Cut}\prod_{i\in S}\dd x_i \dd y_i \leq CN \lambda^{(2-\beta)|S|}\Cut^{-2}.
	\end{equation}
	Therefore, by \eqref{eq:deuxieme bout}, there exists $C>0$ depending on $\beta, M$, $\ve_0$, and $|S|$ such that
	\begin{equation}\label{eq:CC2}
		| \textcolor{black}{I_{\mathrm{mult},\ve_0}^{0}}(S)-\textcolor{black}{I_{\mathrm{mult},\infty}^{0}}(S)| \leq CN\lambda^{(2-\beta)|S|}\Cut^{-2}.
	\end{equation}
	Assembling \eqref{eq:CC1} and \eqref{eq:CC2} proves the lemma.
\end{proof}

\medskip
\medskip

\begingroup\color{black}
Let $X$ be a subpartition of $[N]\setminus V_\bad$ with $k=|V_X|$. Let $\mathbf X=(X_l)_{l=1}^n$ be a family of disjoint nonempty subcollections of $X$, let $\mathbf E=(E_l)_{l=1}^n$ with $E_l\in\mathsf E^{X_l}$, and let $F\subset\mc E^\inter(\Coarse_X(\mathbf X))$. For the parameters of \eqref{def:I1wL}, we define
\endgroup
\begin{flalign}
		&& \textcolor{black}{I_{\mathrm{cl},\ve_0}(X;\mathbf X,\mathbf E,F;w,\tau,L,Z)} \coloneqq  \int_{(\Lambda^2)^{k} } \prod_{i\in V_X}\indic_{i\in I_\good^Z} \prod_{ij\in \cup_l E_l}f_{ij}^{w}\prod_{ij\in F}(-\indic_{\mc{B}_{ij}})\prod_{ij\in \cup_{l=1}^n \mc{E}^\inter(X_l)}\indic_{\mc{B}_{ij}^c} \notag\\
		&& \times \left(\prod_{S\in X}\indic_{\mc{B}_S} \textcolor{black}{\prod_{\substack{i,j\in S\\i<j}}}e^{-\beta w_{ij}}\indic_{\mc{A}_{ij}} \right)\prod_{i\in V_X}e^{L\frac{|x_i-y_i|^2}{\tau_i^2}}\prod_{i\in V_X}e^{\beta \g_\lambda(x_i-y_i)}\indic_{|x_i-y_i|\leq \ve_0\Cut}\dd x_i \dd y_i.\label{def:I2wL}
	\end{flalign}
\begingroup\color{black}
We define $I_{\mathrm{cl},\ve_0}^{-}$, $I_{\mathrm{cl},\ve_0}^{+,Z}$ and $I_{\mathrm{cl},\ve_0}^{0}$ by the same choices of $(w,\tau,L,Z)$ as for $I_{\mathrm{mult},\ve_0}^{-}$, $I_{\mathrm{mult},\ve_0}^{+,Z}$ and $I_{\mathrm{mult},\ve_0}^{0}$, also for $\ve_0=\infty$ in the last case, and we omit the fixed arguments $(X;\mathbf X,\mathbf E,F)$ in the notation.
\endgroup

\begin{lemma}\label{lemma:general exp K}
	Let $\beta\in (2,\infty)$ and let $p^*(\beta)$ be as in \eqref{defpstar}. Let $V_\bad \subset [N]$ and $V_\good = [N]\backslash V_\bad$. Let $X$ be a subpartition of $V_\good$ such that $|V_X|\leq p^*(\beta)$ \textcolor{black}{and $|X|\geq2$}. Set $k\coloneqq |V_X|$. Let $X_1,\ldots,X_n\subset X$ be disjoint, $E_1\in \mathsf{E}^{X_1},\ldots,E_n\in \mathsf{E}^{X_n}$ and $F\in \textcolor{black}{\mathsf{E}_0}^{\Coarse_X(X_1,\ldots,X_n)}$. Let $\tau^Z=(\tau_i^Z)_{i\in\textcolor{black}{[N]\setminus V_\bad}}$ be as in Definition~\ref{def:vijZ}. Recall $\tau_0$ from Definition~\ref{def:vijtilde} and set $\tilde{\tau}_i\coloneqq \tau_0$ for every $i\in [N]$. \textcolor{black}{Let either $(w,\tau)=(\tilde v,\tilde\tau)$ and $V_\bad=\emptyset$, or $(w,\tau)=(v^Z,\tau^Z)$.} Recall $I_\good^Z$ from \eqref{def:IgoodZ}. Let $L\in \dR$. \textcolor{black}{Recall \eqref{def:I2wL}.}

	\textcolor{black}{Then, for $\lambda$ small enough, there exist $C>0$ depending only on $\beta,M,L,k$ and $C_{\ve_0}>0$ also depending on $\ve_0$ such that}
	\begin{equation}\label{eq:estimateI2}
		|\textcolor{black}{I_{\mathrm{cl},\ve_0}(w,\tau,L,Z)}-\textcolor{black}{I_{\mathrm{cl},\ve_0}^{0}}|\leq \textcolor{black}{\lambda^{(2-\beta)k}\bigl(C_{\ve_0}N\Cut^{-2}+C\ve_0^{\alpha(\beta)k}|V_\bad|\bigr)},
	\end{equation}
	and 
	\begin{equation}\label{def:estimateI2 bis}
		|\textcolor{black}{I_{\mathrm{cl},\ve_0}^{0}}-\textcolor{black}{I_{\mathrm{cl},\infty}^{0}}|\leq \textcolor{black}{C_{\ve_0}N\Cut^{-2}\lambda^{(2-\beta)k}.}
	\end{equation}
\end{lemma}

\medskip

\begin{proof}[Proof of Lemma~\ref{lemma:general exp K}]
	As at the beginning of the proof of Lemma~\ref{lemma:general exp M}, we may assume that  $\tau_i\geq \frac{r_i}{4}$. 
	
	{\color{black}
Let $E=\bigcup_lE_l$, $h=\ve_0\Cut$, $R_i=r_i\vee\lambda$, $R=\max_iR_i$ and $R_*=\min_iR_i$. For \eqref{eq:estimateI2}, we restrict both integrals to
    \[
        \mc D=\bigcap_{ij\in E}\{d_{ij}\leq\textcolor{black}{17}h\}.
    \]
    The integrand of $I_{\mathrm{cl},\ve_0}(w,\tau,L,Z)$ vanishes on $\mc D^c$, and the remaining integral over $\mc D^c$ is bounded at the end of Step 3.
}

	\paragraph{\bf{Step 1: starting point}}
	Splitting $f_{ij}^w$ into $a_{ij}^w+b_{ij}^w$ with $a_{ij}^w$ and $b_{ij}^w$ as in Definition~\ref{def:awbw}, and using the parity argument of Lemma~\ref{lemma:cancellation odd}, we get
	\begin{equation*}
		\textcolor{black}{I_{\mathrm{cl},\ve_0}(w,\tau,L,Z)}=\left(\prod_{l=1}^n\sum_{E_{1,l}\subset E_l:E_{1,l}\in \Eul^{X_l}}\right) \textcolor{black}{I_{\mathrm{cl},\ve_0}^{ab}(\mathbf E_1;w,\tau,L,Z)}, 
	\end{equation*}
	where \textcolor{black}{$\mathbf E_1=(E_{1,l})_{l=1}^n$ and}
	\begin{multline*}
		\textcolor{black}{I_{\mathrm{cl},\ve_0}^{ab}(\mathbf E_1;w,\tau,L,Z)}\\\coloneqq \int_{(\Lambda^2)^{|V_X|} } \prod_{i\in V_X}\indic_{i\in I_\good^Z}\prod_{ij\in \cup_l E_{1,l}}a_{ij}^{w}\prod_{ij\in \textcolor{black}{\cup_l(E_l\setminus E_{1,l})}}b_{ij}^{w} \prod_{ij\in F}(-\indic_{\mc{B}_{ij}})\prod_{ij\in \cup_l\mc{E}^\inter(X_l)}\indic_{\mc{B}_{ij}^c} 
		\\ \times \left(\prod_{S\in X}\indic_{\mc{B}_S} \textcolor{black}{\prod_{\substack{i,j\in S\\i<j}}}e^{-\beta w_{ij}}\indic_{\mc{A}_{ij}}\right)
		\prod_{i\in V_X}e^{L\frac{|x_i-y_i|^2}{\tau_i^2}} \prod_{i\in V_X}e^{\beta \g_\lambda(x_i-y_i)}\indic_{|x_i-y_i|\leq \ve_0\Cut}\dd x_i \dd y_i.
	\end{multline*}
	For every $l\in [n]$, fix $E_{1,l}\subset E_l$ such that $E_{1,l}\in \Eul^{X_l}$. Set $E^{\odd}\coloneqq \cup_l E_{1,l}$ and $E^{\even}\coloneqq \cup_l E_l \setminus E^{\odd}$. Decomposing $a_{ij}^w= a_{ij}^v+(a_{ij}^w-a_{ij}^v)$, $e^{-\beta w_{ij}}= e^{-\beta v_{ij}} (e^{\beta (v_{ij}-w_{ij})}-1+1) $, etc., and expanding, we may write 
	\begin{equation*}
		\textcolor{black}{I_{\mathrm{cl},\ve_0}^{ab}(\mathbf E_1;w,\tau,L,Z)} =\sum_{\substack{F_1\subset E^{\odd}, F_2\subset E^{\even}\\ F_3\subset \mc{E}^\intra(X) \\ U_1\subset V_X, U_2\subset V_X}}\textcolor{black}{D}(F_1,F_2,F_3,U_1,U_2),
	\end{equation*}
	where
	\begin{multline*}
		\textcolor{black}{D}(F_1,F_2,F_3,U_1,U_2)=\int_{(\Lambda^2)^{|V_X|} } \textcolor{black}{\indic_{\mc D}}\prod_{ij\in F_1}(a_{ij}^{w}-a_{ij}^v)\prod_{ij\in E^{\odd}\setminus F_1} a_{ij}^v\prod_{ij\in F_2}(b_{ij}^{w}-b_{ij}^v) \prod_{ij\in E^{\even}\setminus F_2}b_{ij}^{v}\\
		\prod_{S\in X}\indic_{\mc{B}_S} \prod_{ ij\in \mc{E}^\intra(X)}e^{-\beta v_{ij}}\indic_{\mc{A}_{ij}} \prod_{ij\in F_3}(e^{\beta (v_{ij}-w_{ij})}-1)\prod_{ij\in \cup_l\mc{E}^\inter(X_l)}\indic_{\mc{B}_{ij}^c}\prod_{ij\in F}(-\indic_{\mc{B}_{ij}})
		\\ \times\prod_{i\in U_1} \Bigl(e^{L\frac{|x_i-y_i|^2}{\tau_i^2}}-1\Bigr)\prod_{i\in U_2}(-1)\indic_{i\notin I_\good^Z}\prod_{i\in V_X}e^{\beta \g_\lambda(x_i-y_i)}\indic_{|x_i-y_i|\leq \ve_0\Cut}\dd x_i \dd y_i.
	\end{multline*}
	Thus,
	\begin{multline}\label{eq:Iexp b}
		\textcolor{black}{I_{\mathrm{cl},\ve_0}^{ab}(\mathbf E_1;w,\tau,L,Z)}- \textcolor{black}{I_{\mathrm{cl},\ve_0}^{ab}(\mathbf E_1;v,\infty,0,\emptyset)}\\=
        \sum_{\substack{(F_1,F_2,F_3,U_1,U_2)\\ \neq (\emptyset,\emptyset,\emptyset,\emptyset,\emptyset)}}\textcolor{black}{D}(F_1,F_2,F_3,U_1,U_2)\textcolor{black}{+O_{\beta,M,k,\ve_0}(N\Cut^{-2}\lambda^{(2-\beta)k})}.
	\end{multline}
	On the other hand,
	\begin{equation}\label{eq:second bout 2}
		\textcolor{black}{\textcolor{black}{I_{\mathrm{cl},\ve_0}^{0}}-\textcolor{black}{I_{\mathrm{cl},\infty}^{0}}} =\sum_{\emptyset \neq U_0\subset V_X}\sum_{\substack{E^{\odd}\subset E:\\E^{\odd}\in \Eul^X}}\textcolor{black}{D_{\mathrm{tail}}}(E^{\odd},U_0),  
	\end{equation}
	where
	\begin{multline}\label{def:JE1U0}
		\textcolor{black}{D_{\mathrm{tail}}}(E^{\odd},U_0)\\\coloneqq \int_{(\Lambda^2)^{|V_X|} } \prod_{ij\in E^{\odd}}a_{ij}^{v}\prod_{ij\in E\setminus E^{\odd}}b_{ij}^{v}\left(\prod_{S\in X}\indic_{\mc{B}_S}\textcolor{black}{\prod_{\substack{i,j\in S\\i<j}}}e^{-\beta v_{ij}}\indic_{\mc{A}_{ij}}\right) \prod_{i\in U_0}(-\indic_{|x_{i}-y_{i}|\geq \ve_0\Cut})\\ \times \prod_{ij\in \cup_l\mc{E}^\inter(X_l)}\indic_{\mc{B}_{ij}^c}\prod_{ij\in F}(-\indic_{\mc{B}_{ij}}) \prod_{i\in V_X}e^{\beta \g_\lambda(x_i-y_i)}\dd x_i \dd y_i.    
	\end{multline}

	\paragraph{\bf{Step 2: bounding $\textcolor{black}{D}$}}
	Let $(Y_1,\ldots,Y_m)$ be the connected components (relative to $X$), with at least two multipoles, of the graph $(V_X,E^{\odd})$. \textcolor{black}{The quotient of $E\cup F$ by $X$ is connected. Contracting the odd components in it leaves a connected graph with edges in $E^{\even}\cup F$. Let $T^b\in\mathsf T_0^{\Coarse_X(Y_1,\ldots,Y_m)}$ be a spanning tree of this graph.}
	
	For every $l\in [m]$,  let $E_{1,l}'$ be the subset of edges in $E^{\odd}$ incident to some vertex in $V_{Y_l}$ and let $E_{1,l}''$ be \textcolor{black}{a minimal spanning $2$-edge-connected (relative to $Y_l$) subgraph of $(V_{Y_l},E_{1,l}')$}. Set $E_1''=\cup_{l\in [m]}E_{1,l}''$.
	
	\textcolor{black}{On $\mc D\cap\mc B_{ij}^c$, Lemmas~\ref{lemma:vij}, \ref{lemma:error'} and~\ref{lemma:errortilde} give $|v_{ij}|+|w_{ij}|\leq Ca_{ij}^{\abs}\leq C$. Bounding the series of Definition~\ref{def:awbw} term by term, we obtain}
\begin{equation}\label{eq:boundsaaaa}
|a_{ij}^{v}| \leq C a^{\abs}_{ij},\quad |a_{ij}^{w}-a_{ij}^v| \leq C a^{\abs}_{ij},\qquad |b_{ij}^{v}| \leq C b^{\abs}_{ij},\quad |b_{ij}^{w}-b_{ij}^v| \leq C b^{\abs}_{ij}.
\end{equation}
	
	$\bullet$ Suppose that $U_1\neq \emptyset$, or $U_2\neq \emptyset$, or $F_3\neq \emptyset$ (case A). Using the last display and proceeding as in the proof of Lemma~\ref{lemma:general exp M} (see \eqref{eq:a intstep2}), we get that there exists $C>0$ depending on \textcolor{black}{$\beta,M,L$} and $k$ such that
	\begin{multline}\label{eq:xx1}
		|\textcolor{black}{D}(F_1,F_2,F_3,U_1,U_2)|\\\leq C\max_{i_0\in V_X}\int_{(\Lambda^2)^{|V_X|} } \left(\frac{\max_{i\in V_X}r_i^2}{\tau_{i_0}^2}+\indic_{i_0\notin I_\good^Z}\right) \prod_{ij\in E_1''}a^{\abs}_{ij}\prod_{ij\in T^b}b^{\abs}_{ij} \prod_{S\in X}\indic_{\mc{B}_S}\\ \times \prod_{i\in V_X}e^{\beta \g_\lambda(x_i-y_i)}\indic_{|x_i-y_i|\leq \ve_0\Cut}\dd x_i \dd y_i.
	\end{multline}

	$\bullet$ Suppose that $\textcolor{black}{E^{\odd}}\setminus E_{1}''$ intersects $F_1$ (case B) and let $i_0j_0$ be in this intersection. Using \eqref{eq:UU}, we have 
	\begin{equation}\label{eq:st8}
		|a_{i_0j_0}^{w}-a_{i_0j_0}^{v}|\leq C|w_{i_0j_0}-v_{i_0j_0}|\leq C\frac{r_{i_0}r_{j_0}}{\min(\tau_{i_0},\tau_{j_0})^2}\le
		C\frac{\max_{i\in V_X}r_i^2}{\min(\tau_{i_0},\tau_{j_0})^2}.
	\end{equation}
	Therefore, by \eqref{eq:boundsaaaa}, there exists $C>0$ depending on \textcolor{black}{$\beta,M,L$} and $k$ such that
	\begin{multline}\label{eq:xx2}
		|\textcolor{black}{D}(F_1,F_2,F_3,U_1,U_2)|\leq C\max_{i_0\in V_X} \int_{(\Lambda^2)^{|V_X|} } \frac{\max_{i\in V_X}r_i^2}{\tau_{i_0}^2}\prod_{ij\in E_1''} a^{\abs}_{ij}\prod_{ij\in T^b}b^{\abs}_{ij}
		\prod_{S\in X}\indic_{\mc{B}_S} \\ \times \prod_{i\in V_X}e^{\beta \g_\lambda(x_i-y_i)}\indic_{|x_i-y_i|\leq \ve_0\Cut}\dd x_i \dd y_i.
	\end{multline}
	
	$\bullet$ Suppose that $E^{\even}\setminus T^b$ intersects $F_2$ (case C) and let $i_0j_0$ be in this intersection. We use the crude bound, again thanks to \eqref{eq:UU}, 
	\begin{equation*}
		|b_{i_0j_0}^{w}-b_{i_0j_0}^{v}|\leq C|w_{i_0j_0}-v_{i_0j_0}|\leq C\frac{\max_{i\in V_X}r_i^2}{\min(\tau_{i_0},\tau_{j_0})^2}
	\end{equation*}
	and conclude as above that 
	\begin{multline}\label{eq:xx3}
		|\textcolor{black}{D}(F_1,F_2,F_3,U_1,U_2)|\leq C\max_{i_0\in V_X} \int_{(\Lambda^2)^{|V_X|} } \frac{\max_{i\in V_X}r_i^2}{\tau_{i_0}^2}\prod_{ij\in E_1''} a^{\abs}_{ij}\prod_{ij\in T^b}b^{\abs}_{ij}
		\prod_{S\in X}\indic_{\mc{B}_S} \\ \times \prod_{i\in V_X}e^{\beta \g_\lambda(x_i-y_i)}\indic_{|x_i-y_i|\leq \ve_0\Cut}\dd x_i \dd y_i.
	\end{multline}

\begingroup\color{black}
$\bullet$ Suppose that there exists $l_0\in[m]$  such that $E_{1,l_0}''\cap F_1\neq \emptyset$ (case D) and let $i_0j_0\in E_{1,l_0}''\cap F_1$. Then, we select a strict ear decomposition $(P_1,\ldots,P_K)$ of $(V_{Y_{l_0}},E_{1,l_0}'')$ as in Definition~\ref{def:strict ear relative} such that the base cycle $P_1$ contains the edge $i_0 j_0$ (which is possible by Lemma~\ref{lemma:strict ear}). We delete $i_0j_0$ from $P_1$, apply the deletion rule of item (2) of Definition~\ref{def:peeling mini} to every later ear, and let $T_{l_0}^a$ be the resulting tree. In particular, $i_0j_0\notin T_{l_0}^a$. The proof of Lemma~\ref{lemma:along cycles}, in which we drop the factors $g_{ij}\leq1$, gives 
	\begin{equation}\label{eq:ears second}
		\prod_{ij\in P_2\cup \cdots \cup P_K}a^{\abs}_{ij}\leq C\prod_{S\in Y_{l_0}\setminus \Res(Y_{l_0},P_1)}r_S^2\prod_{ij\in T_{l_0}^a\setminus P_1}\frac{1}{d_{ij}^2}\indic_{d_{ij}\leq \textcolor{black}{17} \ve_0\Cut}\indic_{\mc{B}_{ij}^c}.
	\end{equation}
	It remains to bound the product over the base cycle $P_1$. We have
	\begin{multline*}
		|a_{i_0j_0}^{w}-a_{i_0j_0}^v| \prod_{ij\in P_1\setminus\{i_0j_0\}}a_{ij}^{\abs}\leq C |a_{i_0j_0}^{w}-a_{i_0j_0}^v| \textcolor{black}{r_{[i_0]^X}r_{[j_0]^X}}\prod_{S\in \Res(Y_{l_0},P_1):S\neq [i_0]^X,[j_0]^X}r_S^2\\ \times \prod_{ij\in P_1\setminus\{i_0j_0\}}\frac{1}{d_{ij}^2}\indic_{d_{ij}\leq \textcolor{black}{17} \ve_0\Cut}\indic_{\mc{B}_{ij}^c}.
	\end{multline*}
	Inserting \eqref{eq:st8},  we get 
	\begin{equation*}
		|a_{i_0j_0}^{w}-a_{i_0j_0}^v| \prod_{ij\in P_1\setminus\{i_0j_0\}}a^{\abs}_{ij}\leq  \frac{C}{\min(\tau_{i_0},\tau_{j_0})^2 } \prod_{S\in \Res(Y_{l_0},P_1)}r_S^2\prod_{ij\in P_1\setminus\{i_0j_0\}}\frac{1}{d_{ij}^2}\indic_{d_{ij}\leq \textcolor{black}{17}  \ve_0\Cut}\indic_{\mc{B}_{ij}^c}.
	\end{equation*}
	Combining this with \eqref{eq:ears second} gives 
	\begin{equation}\label{eq:hell}
		|a_{i_0j_0}^{w}-a_{i_0j_0}^v| \prod_{ij\in E_{1,{l_0}}''\setminus\{i_0j_0\}}a_{ij}^{\abs}\leq \frac{C}{\min(\tau_{i_0},\tau_{j_0})^2 } \prod_{S\in Y_{l_0}}r_S^2\prod_{ij\in T_{l_0}^a}\frac{1}{d_{ij}^2}\indic_{d_{ij}\leq\textcolor{black}{17} \ve_0\Cut}\indic_{\mc{B}_{ij}^c}.
	\end{equation}
	
	For $l\in [m]\setminus\{l_0\}$,  we let $T_l^a\coloneqq \mc{T}^{\textcolor{black}{Y_l}}(\cdot,E_{1,l}'')$ be as in Definition~\ref{def:peeling mini}. For each $l$, the family $B((x_i,y_i)_{i\in V_{Y_l}},T_l^a)$ is taken relative to $Y_l$. \textcolor{black}{Let $\hat S_l\in Y_l$ have the largest radius.} By Corollary~\ref{coro:prod a} and \eqref{eq:hell},
	\begin{multline*}
		|a_{i_0j_0}^{w}-a_{i_0j_0}^v| \prod_{l=1}^m\prod_{ij\in E_{1,l}''\setminus\{i_0j_0\}}a_{ij}^\abs\\ \leq \frac{Cr_{\hat{S}_{l_0}}^2}{\min(\tau_{i_0},\tau_{j_0})^2 }\prod_{S\in \cup_l (\textcolor{black}{Y_l}\setminus \{\hat{S}_l\})}r_S^2 \max_{\substack{T'_l\in B((x_i,y_i)_{i\in V_{Y_l}},T_l^a)\\ l\in[m]}}\prod_{ij\in\cup_lT'_l}\frac{1}{d_{ij}^2}\indic_{d_{ij}\leq \textcolor{black}{17}\ve_0\Cut}\indic_{\mc{B}_{ij}^c}.
	\end{multline*}
\endgroup
	We conclude that in this case 
	\begin{multline*}
		|\textcolor{black}{D}(F_1,F_2,F_3,U_1,U_2)|\\\leq C\max_{i_0\in V_X}\max_{\substack{T_1^a\in \mathsf{T}^{\textcolor{black}{Y_1}}\\ \ldots\\  \textcolor{black}{T_m^a\in \mathsf{T}^{Y_m}} } } \int_{(\Lambda^2)^{|V_X|}} \frac{\max_{i\in V_X}r_i^2}{\tau_{i_0}^2}\prod_{S\in \cup_l (\textcolor{black}{Y_l}\setminus \{\hat{S}_l\})}r_S^2 \prod_{ij\in \cup_lT_l^a}\frac{1}{d_{ij}^2}\indic_{d_{ij}\leq \textcolor{black}{17} \ve_0\Cut}\indic_{\mc{B}_{ij}^c}\\ \times \prod_{ij\in T^b}b_{ij}^{\abs}\prod_{S\in X}\indic_{\mc{B}_S}\prod_{i\in V_X}e^{\beta \g_\lambda(x_i-y_i)}\indic_{|x_i-y_i|\leq \ve_0\Cut}\dd x_i \dd y_i.
	\end{multline*}

	$\bullet$ Suppose that $T^b\cap F_2\neq \emptyset$ (case E) and let $i_0j_0\in T^b \cap F_2$. 
	Using \eqref{eq:boundsaaaa}, we have
	\begin{equation*}
		|b_{i_0j_0}^{w}-b_{i_0j_0}^v|\leq \textcolor{black}{C|w_{i_0j_0}-v_{i_0j_0}|\bigl(|w_{i_0j_0}|+|v_{i_0j_0}|\bigr).} 
	\end{equation*}
	Using Lemma~\ref{lemma:error'} or~\ref{lemma:errortilde}, we find that on the event \textcolor{black}{$\mc B_{i_0j_0}^c$,} 
	\begin{equation}\label{eq:bbbbb}
		|b_{i_0j_0}^{w}-b_{i_0j_0}^v| \leq C\frac{r_{i_0}^2r_{j_0}^2}{\max(\tau_{i_0},\tau_{j_0},d_{i_0j_0})^2d_{i_0j_0}\max(d_{i_0j_0},r_{i_0},r_{j_0})}.
	\end{equation}
	{\color{black}
In \eqref{eq:bbbbb}, we bound $\max(\tau_{i_0},\tau_{j_0},d_{i_0j_0})^{-2}$ by $\tau_{i_0}^{-2}$ and enlarge $r_{i_0},r_{j_0}$ to $R_{i_0},R_{j_0}$. With root $x_{i_0}$, the integral in the variable $z$ of the edge $i_0j_0$ satisfies
    \begin{equation}\label{eq:xx4}
        \int_{|z|\leq\textcolor{black}{17}h}
        \frac{R_{i_0}^2R_{j_0}^2}{|z|\max(|z|,R_{i_0},R_{j_0})}\dd z
        \leq CR_{i_0}^2R_{j_0}^2
        \left(1+\log\frac{\textcolor{black}{17}h}{\max(R_{i_0},R_{j_0})}\right).
    \end{equation}
    Removing $i_0j_0$ splits the spanning tree in two, and Lemma~\ref{lemma:ordering tree}, applied to both parts, gives the full product $\prod_iR_i^2$, together with the factor $\tau_{i_0}^{-2}$.
}

	\paragraph{\bf{Step 3: integration in cases A, B, C, D, E}}
{\color{black}
In cases A, B, C and E, Corollary~\ref{coro:prod a} reduces each $E_{1,l}''$ to a tree relative to $Y_l$, and case D was treated above. We choose trees inside the multipoles as in \eqref{st3}, which form a tree on $V_X$ together with the odd trees and $T^b$. We integrate in the variables \eqref{def:zij' b}, with root $x_{i_0}\in\Lambda$. Since we keep $\mc B_{ij}^c$, we have $d_{ij}>M\min(R_i,R_j)$ on the odd edges. With \eqref{eq:xx4} in case E, this gives
    \begin{multline*}
        |\textcolor{black}{D}(F_1,F_2,F_3,U_1,U_2)|
        \leq C\max_{i_0\in V_X}\int
        \left(\frac{1}{\tau_{i_0}^2}+\frac{\indic_{i_0\notin I_\good^Z}}{R^2}\right)
        \left(1+\log\frac{\textcolor{black}{17}h}{R_*}\right)^k\\
        \times\prod_{i\in V_X}
        \bigl(R_i^2e^{\beta\g_\lambda(r_i)}\indic_{r_i\leq h}\dd\vr_i\bigr)\dd x_{i_0}.
    \end{multline*}
    By \eqref{eq:root bad integral} and \eqref{eq:radial cutoff gain},
    \begin{equation}\label{eq:I61}
        |\textcolor{black}{D}(F_1,F_2,F_3,U_1,U_2)|
        \leq\lambda^{(2-\beta)k}
        \bigl(C_{\ve_0}N\Cut^{-2}+C\ve_0^{\alpha(\beta)k}|V_\bad|\bigr).
    \end{equation}

    On $\mc D^c$, only the integrand of $I_{\mathrm{cl},\ve_0}^{0}$ remains. We bound its bonds by the weights $\widehat a,\widehat b$ of the proof of Lemma~\ref{lemma:int dipole} and fix an edge $e\in E$ with $d_e>\textcolor{black}{17}h$. If $e\notin E_1''\cup T^b$, then its weight is at most $CR^2/h^2$. If $e\in E_1''$, then opening the longest edge as in \eqref{eq:prodE1l'} leaves the factor
    \[
        \min\left(\frac{R^2}{(\textcolor{black}{17}h)^2},\frac{R^2}{\max_{e\in T^a}d_e^2}\right).
    \]
    Its tree integral is at most $CR^2h^{-2}(1+\log(\textcolor{black}{17}h/R_*))^k$. If $e\in T^b$, then
    \[
        \int_{|z|>\textcolor{black}{17}h}\frac{r_i^2r_j^2}{|z|^4}\dd z
        \leq Ch^{-2}\min(R_i,R_j)^2R^2.
    \]
    This $R^2$ compensates for the division by $R^2$ in Lemma~\ref{lemma:ordering tree}. The integral over $\mc D^c$ is thus at most $CN/h^2$ times the radial integral in \eqref{eq:radial cutoff gain}, hence at most $C_{\ve_0}N\Cut^{-2}\lambda^{(2-\beta)k}$.
}

	\paragraph{\bf{Step 4: integration in case F}}
{\color{black}
As in the proofs of Lemmas~\ref{lemma:int dipole} and~\ref{lemma:limiting}, the integral \eqref{def:JE1U0} is bounded by a radial integral, now with $\max_i r_i\geq h$. For $2<\beta<4$, integrating the inner radii, with their logarithmic moments, gives $Ch^{(4-\beta)k-2}$. For $\beta\geq4$, we split each inner integral at $\lambda$. Consequently,
\[
|D_{\mathrm{tail}}|\leq CN
\begin{cases}
h^{(4-\beta)k-2},&2<\beta<4,\\
h^{-2}(1+\log(h/\lambda))^{2k},&\beta=4,\\
h^{2-\beta}\lambda^{(4-\beta)(k-1)}(1+\log(h/\lambda))^k,&\beta>4.
\end{cases}
\]
By the definition of $\Cut$, this gives, for $\lambda$ small enough,
    \begin{equation}\label{eq:caseF}
        |\textcolor{black}{D_{\mathrm{tail}}}(E^{\odd},U_0)|\leq C_{\ve_0}N\Cut^{-2}\lambda^{(2-\beta)k}.
    \end{equation}
}

	\paragraph{\bf{Step 5: conclusion}}
{\color{black}
Combining \eqref{eq:I61} with \eqref{eq:Iexp b} proves \eqref{eq:estimateI2}. Combining \eqref{eq:caseF} with \eqref{eq:second bout 2} proves \eqref{def:estimateI2 bis}.
}
\end{proof}

\begin{lemma}\label{lemma:I3}
	Let $\beta\in (2,\infty)$. Let $\ve_0\in (0,1)$. Let $V\subset [N]$ be such that $|V|\leq p^*(\beta)$.
	\textcolor{black}{For $\lambda$ small enough,} there exists $C>0$ depending on $\beta$, $|V|$ and $\ve_0$ such that
	\begin{equation*}
		\textcolor{black}{\left|\sum_{E\in\mc G_c(V)}\bigl(I_{\mathrm{dip},\ve_0}(V;E)-I_{\mathrm{dip},\infty}(V;E)\bigr)\right|}\leq CN\lambda^{(2-\beta)|V|}\Cut^{-2}.
	\end{equation*}
\end{lemma}

\begingroup\color{black}
\begin{proof}
Set $k=|V|$ and $h=\ve_0\Cut$. We may assume that $k\geq2$. For each $E\in\mc G_c(V)$,
\[
I_{\mathrm{dip},\ve_0}(V;E)-I_{\mathrm{dip},\infty}(V;E)
=-\int_{(\Lambda^2)^k}\prod_{ij\in E}f_{ij}^v
\left(1-\prod_{i\in V}\indic_{r_i\leq h}\right)
\prod_{i\in V}e^{\beta\g_\lambda(x_i-y_i)}\dd x_i\dd y_i.
\]
The factor $1-\prod_{i\in V}\indic_{r_i\leq h}$ is invariant under every reversal and vanishes unless $\max_i r_i>h$. Therefore, the splitting into near and far edges and the average over the reversals of the merged parts, as in the proof of Lemma~\ref{lemma:int dipole}, apply to this integral, and the proof of \eqref{eq:caseF} bounds it by $C_{\beta,k,\ve_0}N\lambda^{(2-\beta)k}\Cut^{-2}$. Summing over the finitely many $E$ proves the claim.
\end{proof}
\endgroup

\subsection{Integrals for unbounded clusters}\label{sub:app unbounded}

For the proof of Lemma~\ref{lemma:integration summation}, we will  need the following summation result on trees.

\begin{lemma}\label{lemma:Cayley}
	Let $S=\{1,\ldots,p\}$ with $p\geq 2$. Let $(x_i)_{i\in S}\in \dR^{|S|}$. We have
	\begin{equation*}
		\sum_{\substack{T\text{ tree}\\\text{on }S}}\prod_{i\in S}x_i^{\deg_T(i)}=|S|^{|S|-2}\Bigl(\prod_{i\in S}x_i\Bigr)\Bigl(\frac1{|S|}\sum_{i\in S}x_i\Bigr)^{|S|-2}.
	\end{equation*}
\end{lemma}
\begin{proof}
	For every tree $T$ on $S$, we have $\sum_{i\in S} \deg_T(i)=2(|S|-1)$. 
	Let $\mc{D}$ be the set of sequences $(d_i)_{i\in S}$ such that $\sum_{i\in S} d_i=2(|S|-1)$ and, for every $i\in S$, $d_i\in \mathbb{N}^*$. One may write 
	\begin{equation*}
		\sum_{\substack{T\text{ tree}\\\text{on }S}}\prod_{i\in S} x_i^{\deg_{T}(i)}=\sum_{d\in \mc{D}}\sum_{T\in \mc{T}(d)}\prod_{i\in S}x_i^{d_i}, 
	\end{equation*}
	where, for every $d\in\mc{D}$, $\mc{T}(d)$ stands for the set of trees on $S$ with degree sequence $d$. By Prüfer's bijection \cite{prufer}, for every $d\in \mc{D}$, we have 
	\begin{equation*}
		|\mc{T}(\textcolor{black}{d})|=\frac{(|S|-2)!}{\prod_{i\in S}(d_i-1)!}.
	\end{equation*}
	Therefore, \begin{equation*}
		\sum_{T\in \mc{T}(d)}\prod_{i\in S}x_i^{\deg_{T}(i)}=\frac{(|S|-2)!}{\prod_{i\in S}(d_i-1)!}\prod_{i\in S}x_i^{d_i}.  
	\end{equation*}
	Let us now sum this over $d\in \mc{D}$. By  the change of variables $d_i'=d_i-1$, one can write 
	\begin{equation*}
		\sum_{d\in \mc{D}}\prod_{i\in S}\frac{x_i^{d_i}}{(d_i-1)!}=\sum_{d\in \mc{D}'}\prod_{i\in S}\frac{x_i^{d_i+1}}{d_i!}=\prod_{i\in S}x_i \sum_{d\in \mc{D}'}\prod_{i\in S}\frac{x_i^{d_i}}{d_i!},
	\end{equation*}
	where $\mc{D}'$ is the set of sequences $(d_i)_{i\in S}$ such that $\sum_{i\in S}d_i=|S|-2$ and, for every $i\in S$, $d_i\in \mathbb{N}$. One can notice that 
	\begin{equation*}
		\sum_{d\in \mc{D}'}\prod_{i\in S}\frac{x_i^{d_i}}{d_i!}=\frac{1}{(|S|-2)!}\Bigl(\sum_{i\in S}x_i\Bigr)^{|S|-2}.
	\end{equation*}
	Indeed, the right-hand side of the last display can be viewed as the coefficient of $t^{|S|-2}$ in the expansion of $\prod_{i\in S}e^{x_i t}$. Combining the last relations concludes the proof.
\end{proof}

\begin{proof}[Proof of Lemma~\ref{lemma:integration summation}]
	Let $Y\coloneqq \Coarse_X(X_1,\ldots,X_n)$ and $k\coloneqq |V_X|$.

	Recall \eqref{def:IT'}.  
	Using $\indic_{\mc{B}_S}\le \sum_{T_S\in \mc{T}_c(S)} \prod_{ij\in T_S} \indic_{\mc{B}_{ij}}$ and  Cayley's theorem, we obtain
	\begin{equation*}
		\textcolor{black}{J_{\mathrm{large},C_0}}(\cup_{l=1}^n T_l^a,T^b)\leq \Bigl(\prod_{S\in X} |S|^{|S|-2}\max_{T_S\in \mc{T}_c(S)}\Bigr)\textcolor{black}{J_1}(\cup_{l=1}^n T_l^a,T^b, (T_S)),
	\end{equation*}
	where
	\begin{multline*}
		\textcolor{black}{J_1}(\cup_{l=1}^n T_l^a,T^b, (T_S))\coloneqq \int_{(\Lambda^2)^{|V_X|}} \prod_{ij\in T^b}b_{ij}^{\abs}\prod_{l=1}^n\Bigl(\prod_{S\in X_l, S\neq \hat{S}_l}r_{S}^2\Bigr)
		\prod_{ij\in \cup_{l=1}^nT_l^a }\Bigl(\frac{1}{d_{ij}^2}\indic_{ d_{ij}\leq \textcolor{black}{17}\ve_0\Cut}\indic_{\mc{B}_{ij}^c} \Bigr) \\ \times \prod_{S\in X}\prod_{ij\in T_S}\indic_{\mc{B}_{ij}} \prod_{i\in V_X}\Bigl(\Bigl(\frac{\ve_0\Cut}{r_i}\Bigr)^{\frac{C_0}{M}}e^{\beta \g_\lambda(x_i-y_i)}\indic_{|x_i-y_i|\leq \ve_0\Cut} \dd x_i \dd y_i\Bigr).
	\end{multline*}
	Splitting the phase space according to the variables that attain $d_{ij}$, for every $ij\in \cup_S T_S\cup T^b\cup \cup_l T_l^a$ as in \eqref{eq:g2}, using the variables \eqref{eq:g4} and Lemma~\ref{lemma:ordering tree}, and recalling that $d_{ij}\ge M \max(\min(r_i,r_j),\lambda)$ on $\mc{B}_{ij}^c$, we get that there exists $C>0$ depending on $\beta$ and $M$ such that
	\begin{multline*}
		\textcolor{black}{J_1}(\cup_{l=1}^n T_l^a,T^b, (T_S))\leq NC^k \int   \prod_{ij\in T^b\cup_{S\in X}T_S }\max(\min(r_i,r_j),\lambda)^2 \prod_{l=1}^n\Bigl(\prod_{S\in X_l, S\neq \hat{S}_l}r_{S}^2\Bigr)\\
		\times  \prod_{ij\in T^a\cup T^b}\log \frac{\textcolor{black}{17}\ve_0 \Cut}{\min(r_i,r_j)} \prod_{i\in V_X}\Bigl(\frac{\ve_0\Cut}{r_i}\Bigr)^{\frac{C_0}{M}}e^{\beta \g_\lambda(r_i)}\indic_{r_i\leq \ve_0\Cut} \dd \vr_i.
	\end{multline*}
\begingroup\color{black}
Set $h=\ve_0\Cut$, $a=C_0/M$ and $\ell_i=\log(\textcolor{black}{17}h/r_i)\geq1$. For a partition $\mathcal Q$ of a nonempty label set $W$, write $q=|\mathcal Q|$, $L_S=\sum_{i\in S}\ell_i$ and $L=\sum_{i\in W}\ell_i$. Since $\ell_i\geq1$, we have $\max(\ell_i,\ell_j)\leq\ell_i\ell_j$. We sum this bound over the endpoints of each edge in the two blocks it joins, and Lemma~\ref{lemma:Cayley} then gives, for $q\geq2$,
\begin{equation*}
\sum_{T\in\mathsf T^{\mathcal Q}}\prod_{ij\in T}\log\frac{\textcolor{black}{17}h}{\min(r_i,r_j)} \leq L^{q-2}\prod_{S\in\mathcal Q}L_S \leq q^{-q}L^{2q-2}\leq C_a^{|W|}q^q e^{aL}.
\end{equation*}
The last inequality follows by maximizing $L^{2q-2}e^{-aL}$, with $C_a$ depending only on $a$. Consequently,
\begin{equation}\label{eq:C0M}
\sum_{T_l^a\in\mathsf T^{X_l}}\prod_{ij\in T_l^a}\log\frac{\textcolor{black}{17}h}{\min(r_i,r_j)}
\leq C^{|V_{X_l}|}|X_l|^{|X_l|}\prod_{i\in V_{X_l}}\left(\frac{\textcolor{black}{17}h}{r_i}\right)^a,
\end{equation}
and, with the blocks of $Y$ in place of those of $X_l$,
\begin{equation}\label{eq:C0M2}
\sum_{T^b\in\mathsf T_0^Y}\prod_{ij\in T^b}\log\frac{\textcolor{black}{17}h}{\min(r_i,r_j)}
\leq C^k|Y|^{|Y|}\prod_{i\in V_X}\left(\frac{\textcolor{black}{17}h}{r_i}\right)^a.
\end{equation}
\endgroup
	By \eqref{eq:RS}, 
	\begin{equation}\label{eq:RS3}
		\prod_{l=1}^n \prod_{S\in X_l:S\neq \hat{S}_l} r_S^2 \prod_{ij\in T^b}\max(\min(r_i,r_j),\lambda)^2  \prod_{ij\in \cup_{S\in X}T_S}\max(\min(r_i,r_j),\lambda)^2\leq \frac{\prod_{i\in V_X}\max(r_i,\lambda)^2}{\max_{i\in V_X}\max(r_i,\lambda)^2}. 
	\end{equation}
	\textcolor{black}{The disjoint sets $X_l$ contribute a total constant $C^k$. The factor $(\ve_0\Cut/r_i)^{C_0/M}$ and the two tree sums give at most the power $3C_0/M$ per radius. We absorb the factors $\textcolor{black}{17}^{2C_0k/M}$ into $C^k$ and, since $r_i\leq\ve_0\Cut$, increase this power to $4C_0/M$.} Therefore, combining \eqref{eq:C0M}, \eqref{eq:C0M2} and \eqref{eq:RS3}, we obtain that there exists $C>0$ depending on $\beta,p(\beta)$ and $M$ such that
	\begin{multline}\label{eq:lastII}
		\Bigl(\prod_{l=1}^n\sum_{T_l^a\in \mathsf{T}^{X_l}}\Bigr)\sum_{T^b\in \textcolor{black}{\mathsf{T}_0}^Y } \textcolor{black}{J_1}(\cup_{l=1}^n T_l^a,T^b, (T_S)) \\
		\leq NC^k\prod_{l=1}^n|X_l|^{|X_l|}|Y|^{|Y|} \int \frac{\prod_{i\in V_X}\max(r_i,\lambda)^2}{\max_{i\in V_X}\max(r_i,\lambda)^2} \prod_{i\in V_X}\Bigl(\frac{\ve_0\Cut}{r_i}\Bigr)^{\frac{4C_0}{M}}e^{\beta \g_\lambda(r_i)}\indic_{r_i\leq \ve_0\Cut} \dd \vr_i.  
	\end{multline}
	\begingroup\color{black}
$\bullet$ Suppose that $\beta\in(2,4)$. Set $a=4C_0/M$ and $\rho_i=r_i\vee\lambda$. We take $M$ large enough that $a\leq(4-\beta)/2$. Since $k>4/(4-\beta)$,
\[
\frac{\prod_i\rho_i^2}{\max_i\rho_i^2}\leq\prod_i\rho_i^{2-2/k},
\qquad 4-\beta-a-2/k>0.
\]
Thus the radial integral in \eqref{eq:lastII} is at most
\[
C^k\left(h^a\int_0^h r^{3-\beta-a-2/k}\dd r\right)^k
=C^k\frac{h^{(4-\beta)k-2}}{(4-\beta-a-2/k)^k}.
\]
The denominator is bounded below by a positive constant depending only on $\beta,M$, uniformly in the integers $k>4/(4-\beta)$. Hence \eqref{eq:lastII} is bounded by
\[
NC^k\prod_{l=1}^n|X_l|^{|X_l|}|Y|^{|Y|}(\ve_0 R_{\beta,\lambda})^{(4-\beta)k-2}.
\]
Since $R_{\beta,\lambda}^{4-\beta}=\lambda^{2-\beta}$, this proves the claim for $\beta\in(2,4)$.
\endgroup
	
	\begingroup\color{black}
$\bullet$ Suppose that $\beta\geq4$, and set $a=4C_0/M\in(0,2)$. With the same factorization and a split at $\lambda\leq h$, each radial integral satisfies
\[
h^a\int_0^h r^{1-a}(r\vee\lambda)^{2-2/k-\beta}\dd r
\leq h^a\lambda^{4-\beta-2/k-a}
\left(\frac1{2-a}+\frac1{\beta-4+2/k+a}\right).
\]
The parenthesis is at most $(2-a)^{-1}+a^{-1}$, uniformly in $k$. Thus the radial integral in \eqref{eq:lastII} is at most $C^k\lambda^{(4-\beta)k-2}(h/\lambda)^{ak}$. In particular, for $|\log\lambda|\geq1$, we obtain
\endgroup
	\begin{multline}\label{eq:xyt}
		\Bigl(\prod_{l=1}^n\sum_{T_l^a\in \mathsf{T}^{X_l}}\Bigr)\sum_{T^b\in \textcolor{black}{\mathsf{T}_0}^Y } \textcolor{black}{J_1}(\cup_{l=1}^n T_l^a,T^b, (T_S))    \leq NC^k\prod_{l=1}^n |X_l|^{|X_l|}|Y|^{|Y|} \\ \times\Bigl(\frac{\ve_0\Cut}{\lambda}\Bigr)^{\frac{4C_0}{M}k}\Bigl( \lambda^{(4-\beta)k-2}\indic_{\beta>4}+\lambda^{-2}|\log\lambda|\indic_{\beta=4}\Bigr). 
	\end{multline}
	Notice that 
	\begin{equation*}
		\Bigl( \lambda^{(4-\beta)k-2}\indic_{\beta>4}+\lambda^{-2}|\log\lambda|\indic_{\beta=4}\Bigr)\Bigl(\frac{\ve_0\Cut}{\lambda}\Bigr)^{\frac{4C_0}{M}k}=\lambda^{(2-\beta)k} \lambda^{2(k-1)} (\ve_0 \lambda^{-2p_0-1})^{\frac{4C_0}{M}k}(\indic_{\beta>4}+|\log\lambda|\indic_{\beta=4}).
	\end{equation*}
	\begingroup\color{black}
We take $M$ large enough that $a(2p_0+1)\leq1/2$. Since $k>2p_0\geq2$,
\[
(2-a(2p_0+1))k-2-2p_0\geq k/6.
\]
For $\lambda^{1/12}\leq\ve_0$ and $\lambda$ small enough, we have $\lambda^{k/12}|\log\lambda|\leq1$ uniformly in $k\geq3$. Hence,
\endgroup 
	\begin{equation*}
		\Bigl(\frac{\ve_0\Cut}{\lambda}\Bigr)^{\frac{4C_0}{M}k}\Bigl( \lambda^{(4-\beta)k-2}\indic_{\beta>4}+\lambda^{-2}|\log\lambda|\indic_{\beta=4}\Bigr)\le 
		\ve_0^{2\alpha(\beta)k-2} \lambda^{(2-\beta)k}\lambda^{2p_0},
	\end{equation*} when $\lambda$ is small enough with respect to $\ve_0$. 
	Inserting this into \eqref{eq:xyt} concludes the proof in the case $\beta\geq 4$. 
	
\end{proof}

\section{Sine--Gordon correspondence}\label{app:sine-gordon}

\subsection{Grand-canonical representation}

We first identify, at fixed lattice spacing, the sine--Gordon partition
function with that of a grand-canonical two-component Coulomb gas. The
next subsection relates this ensemble to the canonical setting studied
here.

Fix $\lambda,L>0$ such that $\sqrt{L}/\lambda\in\dN$, and let
\[
  \Omega_L:=[0,\sqrt{L})^2,
  \qquad
  \Lambda_L^\lambda:=(\lambda\dZ)^2\cap\Omega_L,
  \qquad
  |\Lambda_L^\lambda|=L/\lambda^2.
\]
Thus $L$ is the area of the torus. We identify opposite sides of
$\Omega_L$ and write
$\Lambda:=\Lambda_L^\lambda$. Each lattice point has mass $\lambda^2$.
Let
$G_\lambda\colon\Lambda\times\Lambda\to\dR$ be a symmetric
positive-definite kernel, and let $\mu_{G_\lambda}$ be the centered
Gaussian measure on $\dR^\Lambda$ with covariance $G_\lambda$. For the
massive sine--Gordon model, \textcolor{black}{we take the kernel with respect to $\lambda^2$ times the counting measure}:
{\color{black}
\[
 [(-\Delta_\lambda+m^2)^{-1}f](x)
 =\lambda^2\sum_{y\in\Lambda}H_\lambda(x,y)f(y),
 \qquad G_\lambda(x,y)=2\pi H_\lambda(x,y),\qquad m>0,
\]
\[
 (\Delta_\lambda f)(x)
 =\lambda^{-2}\sum_{a=1}^2
 \bigl(f(x+\lambda e_a)+f(x-\lambda e_a)-2f(x)\bigr).
\]
Here $e_1,e_2$ are the coordinate unit vectors and the arguments are taken modulo the torus.
} The
factor $2\pi$ matches the convention $-\Delta\g=2\pi\delta_0$ used in
this paper.

For $\beta>0$ and $\sigma\in\{-1,1\}$, define the Wick-ordered
exponential by
\[
  \nord{e^{i\sigma\sqrt{\beta}\,\phi_x}}
  :=e^{\frac{\beta}{2}G_\lambda(x,x)}
  e^{i\sigma\sqrt{\beta}\,\phi_x}.
\]
This normalization gives
$\int_{\dR^\Lambda}\nord{e^{i\sigma\sqrt{\beta}\,\phi_x}}
\,\dd\mu_{G_\lambda}(\phi)=1$.
The Wick-ordered cosine is therefore
\[
  \nord{\cos(\sqrt{\beta}\,\phi_x)}
  :=\frac12\sum_{\sigma\in\{-1,1\}}
  \nord{e^{i\sigma\sqrt{\beta}\,\phi_x}}
  =e^{\frac{\beta}{2}G_\lambda(x,x)}
  \cos(\sqrt{\beta}\,\phi_x).
\]
For a charge fugacity $z>0$, define the Wick-ordered sine--Gordon
partition function by
\begin{equation}\label{eq:Z-SG appendix}
  Z^{\SG}_\Lambda(\beta,z)
  :=\int_{\dR^\Lambda}
  \exp\!\left(2z\lambda^2\sum_{x\in\Lambda}
  \nord{\cos(\sqrt{\beta}\,\phi_x)}\right)
  \,\dd\mu_{G_\lambda}(\phi).
\end{equation}

The grand-canonical two-component lattice Coulomb gas partition function
is
\begin{equation}\label{eq:Z-Coul appendix}
  Z^{\mathrm{Coul}}_\Lambda(\beta,z)
  :=\sum_{M\geq0}\frac{(z\lambda^2)^M}{M!}
  \sum_{\boldsymbol{x}\in\Lambda^M}
  \sum_{\boldsymbol{\sigma}\in\{-1,1\}^M}
  \exp\!\left(-\beta\sum_{1\leq j<k\leq M}
  \sigma_j\sigma_kG_\lambda(x_j,x_k)\right).
\end{equation}

\begin{lemma}
For every $\beta,z>0$,
\[
  Z^{\SG}_\Lambda(\beta,z)=Z^{\mathrm{Coul}}_\Lambda(\beta,z).
\]
\end{lemma}

\begin{proof}
We expand the exponential in \eqref{eq:Z-SG appendix} and then evaluate
each Gaussian expectation. Since
\[
  2\nord{\cos(\sqrt{\beta}\,\phi_x)}
  =\sum_{\sigma\in\{-1,1\}}
  \nord{e^{i\sigma\sqrt{\beta}\,\phi_x}},
\]
we obtain
\begin{equation*}
Z^{\SG}_\Lambda(\beta,z) ={} \sum_{M\geq0}\frac{(z\lambda^2)^M}{M!} \sum_{\boldsymbol{x}\in\Lambda^M} \sum_{\boldsymbol{\sigma}\in\{-1,1\}^M} \int_{\dR^\Lambda} \prod_{j=1}^M\nord{e^{i\sigma_j\sqrt{\beta}\,\phi_{x_j}}} \,\dd\mu_{G_\lambda}(\phi).
\end{equation*}
Termwise integration is justified by the uniform bound
\[
  \exp\!\left(2z\lambda^2\sum_{x\in\Lambda}
  e^{\frac{\beta}{2}G_\lambda(x,x)}\right)<\infty.
\]

For fixed $\boldsymbol{x}$ and $\boldsymbol{\sigma}$, the Gaussian
characteristic function gives
\begin{align*}
		\int_{\dR^\Lambda}
  \prod_{j=1}^M\nord{e^{i\sigma_j\sqrt{\beta}\,\phi_{x_j}}}
  \,\dd\mu_{G_\lambda}(\phi) &\quad=\exp\!\left(
  \frac{\beta}{2}\sum_{j=1}^MG_\lambda(x_j,x_j)
  -\frac{\beta}{2}\sum_{j,k=1}^M
  \sigma_j\sigma_kG_\lambda(x_j,x_k)\right)\\
		&\quad=\exp\!\left(-\beta\sum_{1\leq j<k\leq M}
  \sigma_j\sigma_kG_\lambda(x_j,x_k)\right).
	\end{align*}
The last equality uses $\sigma_j^2=1$ and the symmetry of $G_\lambda$.
Substitution in the preceding expansion proves the result.
\end{proof}

\subsection{From grand-canonical to canonical}

In the grand-canonical Coulomb gas \eqref{eq:Z-Coul appendix}, the
numbers of positive and negative charges fluctuate. Our canonical model
has $N$ charges of each sign in a box of area $N$. To compare the two
ensembles, \textcolor{black}{we match their mean numbers of close pairs of opposite charges}.

{\color{black}
For the massive kernel, let $d_\Lambda$ be the torus distance. The weight of the pairs of opposite charges at distance at most $\lambda$ satisfies
\[
 L\lambda^2e^{\beta G_\lambda(0,0)}
 \leq \lambda^4\sum_{\substack{x,y\in\Lambda\\d_\Lambda(x,y)\leq\lambda}}
 e^{\beta G_\lambda(x,y)}
 \leq 5L\lambda^2e^{\beta G_\lambda(0,0)}.
\]
The diagonal gives the lower bound. For the upper bound, each $x$ has at most five such sites $y$, and $G_\lambda(x,y)\leq G_\lambda(0,0)$.
If we treat the pairs as independent, their mean number is therefore of order
$z^2L\lambda^2e^{\beta G_\lambda(0,0)}$. Taking $L=N$ and setting this mean equal to $N$ gives
\[
 z^2\lambda^2e^{\beta G_\lambda(0,0)}\asymp1.
\]
For fixed $m>0$ and $L\geq m^{-2}$, we have
$G_\lambda(0,0)=|\log\lambda|+O_m(1)$ as $\lambda\to0$, uniformly in $L$, so that $z\asymp_{\beta,m}\lambda^{(\beta-2)/2}$.
With the Wick factor $e^{\frac{\beta}{2}G_\lambda(0,0)}$, the dimensionless initial fugacity in \eqref{eq:ODE} is then
\[
 z(0)=z\lambda^2e^{\frac{\beta}{2}G_\lambda(0,0)}\asymp\lambda.
\]
}

\subsection{Renormalization-group interpretation of the critical scale}

\textcolor{black}{We now consider the formal massless flow with $z(0)\asymp\lambda$, which can describe the massive model on the torus only at lengths much smaller than $\min(m^{-1},\sqrt L)$.} Linearizing
\eqref{eq:ODE} at $z=0$ keeps $K(l)$ at its initial value $\beta/(2\pi)$.
For $\beta\in(2,4)$, the second equation then gives
\[
  z(l)\asymp \lambda e^{(2-\beta/2)l}.
\]
This reaches order one when
\[
  l=\frac{2|\log\lambda|}{4-\beta}+O_\beta(1).
\]
The physical length at RG time $l$ is $\lambda e^l$. Hence
\[
  \lambda e^l\asymp_\beta
  \lambda e^{2|\log\lambda|/(4-\beta)}
  =\lambda^{-\frac{\beta-2}{4-\beta}}
  =R_{\beta,\lambda}.
\]
The linearized flow predicts that the fugacity becomes of order one at
length $R_{\beta,\lambda}$. \textcolor{black}{This prediction can be compared with the massive model on the torus only when $mR_{\beta,\lambda}\ll1$ and $R_{\beta,\lambda}\ll\sqrt L$, with the fugacity $z$ chosen as in the previous subsection.} For $\beta>4$, the exponent $2-\beta/2$ is
negative. At $\beta=4$, it vanishes. In these cases, the linearized flow
produces no algebraic scale in $\lambda$, in agreement with the convention
$R_{\beta,\lambda}=+\infty$ for $\beta\geq4$.

\bibliographystyle{alpha}
\bibliography{bib.bib}

\end{document}